\documentclass[final,notitlepage,11pt]{article}
\usepackage{amsfonts}
\usepackage{amssymb}
\usepackage{booktabs}
\usepackage{enumitem}
\usepackage{multirow}
\usepackage{graphicx}
\usepackage{amsmath}
\usepackage{latexsym}
\usepackage{rotating}
\usepackage{mathrsfs}
\usepackage{hyperref}
\usepackage[svgnames]{xcolor}
\usepackage{caption}
\usepackage{subcaption}

\hypersetup{
colorlinks=true,
linkcolor=DarkRed,
citecolor=NavyBlue}
\allowdisplaybreaks
\RequirePackage[round,authoryear,longnamesfirst]{natbib}
\newtheorem{theorem}{Theorem}[section]
\newtheorem{ass}{Assumption}

\newtheorem{lemma}[theorem]{Lemma}

\newtheorem{proposition}[theorem]{Proposition}

\DeclareMathOperator*{\argmin}{arg\,min}

\newenvironment{proof}[1][Proof]{\textbf{#1.} }{\  \rule{0.5em}{0.5em}}
\numberwithin{equation}{section}
\newcommand{\RNum}[1]{\uppercase\expandafter{\romannumeral #1\relax}}

\begin{document}

\title{Low-rank Panel Quantile Regression: Estimation and Inference\thanks{%
Su gratefully acknowledges the support from the National Natural Science
Foundation of China under Grant No. 72133002. Zhang acknowledges the financial support from a Lee Kong Chian fellowship. Any and all errors are our own. }}
\author{Yiren Wang$^{a}$, Liangjun Su$^{b}$ and Yichong Zhang$^{a}$ \\
$^{a}$School of Economics, Singapore Management University, Singapore\\
$^{b}$School of Economics and Management, Tsinghua University, China}
\maketitle

\begin{abstract}
In this paper, we propose a class of low-rank panel quantile regression
models which allow for unobserved slope heterogeneity over both individuals
and time. We estimate the heterogeneous intercept and slope matrices via
nuclear norm regularization followed by sample splitting, row- and
column-wise quantile regressions and debiasing. We show that the estimators
of the factors and factor loadings associated with the intercept and slope
matrices are asymptotically normally distributed. In addition, we develop
two specification tests: one for the null hypothesis that the slope
coefficient is a constant over time and/or individuals under the case that
true rank of slope matrix equals one, and the other for the null hypothesis
that the slope coefficient exhibits an additive structure under the case
that the true rank of slope matrix equals two. We illustrate the finite
sample performance of estimation and inference via Monte Carlo simulations
and real datasets.\medskip

 \medskip

\medskip\ 

\ \ \ \ \ \ \ \ \ \ \ 

\ \ \ \ \ \ \ \ \ \ \ 
\end{abstract}

\ \ \ \ \pagebreak

\section{Introduction}

Panel quantile regressions are widely used to estimate the conditional
quantiles, which can capture the heterogeneous effects that may vary across
the distribution of the outcomes. Such effects are usually assumed to be
homogeneous across individuals and over time periods. However, in empirical
analyses, it is usually unknown whether the slope coefficients are
homogeneous across individuals and/or time. Mistakenly forcing slopes to be
homogeneous across time and individuals may lead to inconsistent estimation
and misleading inferences. This prompts two questions to be answered: how
can we estimate the true model at different quantiles when we allow for
heterogeneous slopes across individuals and time at the same time? How to
conduct specification tests for homogeneous effects over individuals or time
and tests for the additive structure of the slope coefficients?

To answer the first question, we propose an estimation procedure for
heterogeneous panel quantile regression models where we allow the fixed
effects to be either additive or interactive, and the slope coefficients to
be heterogeneous over both individuals and time. We impose a low-rank
structure for both the intercept and slope coefficient matrices and estimate
them via nuclear norm regularization (NNR) followed by the sample splitting,
row- and column-wise quantile regressions and debiasing steps. The
estimation algorithm is inspired by \cite{chernozhukov2019inference}, where
the main difference is that we split the full sample into three subsamples
rather than two because we need certain uniform results which require
independence of regressors and regressand used in the debiasing step, and we
do not have the closed form for the quantile regression estimates. At last,
we derive the asymptotic distributions for the estimators of the factors and
factor loadings associated with slope coefficient matrices.

To answer the second question, under the case when the rank of slope
coefficient matrix equals one, we conduct sup-type specification tests for
homogeneous effects over individuals or time following the lead of \cite%
{castagnetti2015inference} and \cite{lu2021uniform}. We show that our
sup-test statistics follow the Gumbel distribution under the null, and the
tests have non-trivial power against certain classes of local alternatives.
Under the case when the rank of slope matrix equals two, our sup-type test
statistic is also shown to follow the Gumbel distribution under the null
that the slope coefficient exhibits an additive structure.

This paper relates to three bunches of literature. First, we contribute to
the large literature on panel quantile regressions (PQRs). Since \cite%
{koenker2004quantile} studied the PQRs with individual fixed effects, there
has been an increasing number of papers on PQRs. \cite{galvao2010penalized}, 
\cite{kato2012asymptotics}, \cite{galvao2015efficient}, \cite%
{galvao2016smoothed}, \cite{machado2019quantiles}, and \cite%
{galvao2020unbiased} study the asymptotics for PQRs with individual fixed
effects. \cite{chen2021quantile} study quantile factor models and \cite%
{chen2019two} considers PQRs with interactive fixed effects (IFEs). We
complement the literature by allowing for unobserved heterogeneity in the
slope coefficients of PQRs.

Second, our paper also pertains to slope heterogeneity in panel data models.
Latent group structures across individuals and structural changes over time
are two common types of slope heterogeneity that have received vast
attention in the literature. To recover the unobserved group structures,
various methods have been proposed. For example, \cite{lin2012estimation}, 
\cite{bonhomme2015grouped} and \cite{ando2016panel} use the K-means
algorithm; \cite{su2016identifying} propose the C-lasso algorithm which is
further studied and extended by \cite{su2018identifying}, \cite{su2019sieve}
and \cite{wang2019heterogeneous}; \cite{wang2018homogeneity} propose an
clustering algorithm in regression via data-driven segmentation called
CARDS; \cite{wang2021identifying} propose a sequential binary segmentation
algorithm to identify the latent group structures in nonlinear panels.
Recent literature on the estimation with structural changes in panel data
models includes, but is not limited to, \cite{chen2015estimating}, \cite%
{cheng2016shrinkage}, \cite{ma2018estimation}, \cite{baltagi2021estimating}.
In addition, \cite{galvao2018testing} and \cite{zhang2019quantile} consider
individual heterogeneity in PQRs while they assume homogeneity across time.
To allow for both latent groups and structural breaks, \cite%
{okui2021heterogeneous} study a linear panel data model with individual
fixed effects where each latent group has common breaks and the breaking
points can be different across different groups, and they propose a grouped
adaptive group fused lasso (GAGFL) approach to estimate slope coefficients. 
\cite{lumsdaine2021estimation} consider a linear panel data model with a
grouped pattern of heterogeneity where the latent group membership structure
and/or the values of slope coefficients can change at a breaking point, and
they propose a K-means-type estimation algorithm and establish the
asymptotic properties of the resulting estimators. Compared with the models
studied above, our model combines both individual and time heterogeneity and
only requires certain low-rank structure in the slope coefficient matrix. So
the unobserved heterogeneity takes a more flexible form in our model than
those in the literature such as \cite{okui2021heterogeneous} and \cite%
{lumsdaine2021estimation}.


Last, our paper also connects with the burgeoning literature on nuclear norm
regularization. Such a method has been widely adopted to study panel and
network models. See, \cite{alidaee2020recovering}, \cite{athey2021matrix}, 
\cite{bai2019rank}, \cite{belloni2019high}, \cite{chen2020noisy}, \cite%
{chernozhukov2019inference}, \cite{feng2019regularized}, \cite%
{Hong_Su_Jiang2022}, \cite{Miao_Phillips_Su2022}, among others. In the least
squares panel framework, \cite{moon2018nuclear} consider a homogeneous panel
with IFEs by using NNR-based estimator as an initial estimator to construct
iterative estimators that are asymptotically equivalent to the least squares
estimators; \cite{chernozhukov2019inference} study a heterogenous panel
where both the intercept and slope coefficient matrices exhibit a low-rank
structure and establish the asymptotic distribution theory based on NNR. In
the presence of endogeneity, \cite{Hong_Su_Jiang2022} proposes a profile GMM
method to estimate panel data models with IFEs. In the panel quantile
regression setting, \cite{feng2019regularized} develops error bounds for the
low-rank estimates in terms of Frobenius norms under independence
assumption; \cite{belloni2019high} relaxes the independence assumption to
the $\beta $-mixing condition along the time dimension. Our paper extends 
\cite{chernozhukov2019inference} from the least squares framework to the PQR
framework, derives the asymptotic distribution theory and develops various
specification tests under some strong mixing conditions along the time
dimension that is weaker than the $\beta $-mixing condition. We also rely on
the sequential symmetrization technique developed by \cite%
{rakhlin2015sequential} to obtain the convergence rates of the nuclear norm
regularized estimators.


The rest of the paper is organized as follows. We first introduce the
low-rank structure PQR model and the estimation algorithm in Section 2. We
study the asymptotic properties of our estimators in Section 3. In Section
4, we propose two specification tests: one for the no-factor structure and
one for the additive structure, and study the asymptotic properties of the
test statistics. In Section 5, we show the finite sample performance of our
method via Monte Carlo simulations. In Section 6, we apply our method to two
datasets: one is to study how Tobin's q and cash flows affect corporate
investment and whether firm's external investment to its internal financing
exhibits heterogeneity structure, and the other is to study the relationship
between economics growth, foreign direct investment and unemployment.
Section 7 concludes. All proofs are related to the online supplement.

\textit{Notation.} $\left\Vert \cdot \right\Vert _{1}$, $\left\Vert \cdot
\right\Vert _{op}$, $\left\Vert \cdot \right\Vert _{\infty }$, $\left\Vert
\cdot \right\Vert _{\max }$ $\left\Vert \cdot \right\Vert _{2}$, $\left\Vert
\cdot \right\Vert _{F}$, $\left\Vert \cdot \right\Vert _{\ast }$ denote the
matrix norm induced by 1-norms, the matrix norm induced by 2-norms, the
matrix norm induced by $\infty $-norms, the maximum norm, the Euclidean
norm, the Frobenius norm and the nuclear norm. $\odot $ is the element-wise
product. $\lfloor \cdot \rfloor $ and $\lceil \cdot \rceil $ denote the
floor and ceiling functions, respectively. $a\vee b$ and $a\wedge b$ return
the max and the min of $a$ and $b,$ respectively. The symbol $\lesssim $
means \textquotedblleft the left is bounded by a positive constant times the
right\textquotedblright . Let $A=\{A_{it}\}_{i \in [n], t\in [T]}$ be a matrix with its $(i,t)$-th
entry denoted as $A_{it}$, where $[n]$ to denote the set $%
\{1,\cdots ,n\}$ for any positive integer $n$. Let $\{A_{j}\}_{j=0}^{p}$ denote the collection
of matrices $A_{j}$ for all $j\in \{0,\cdots ,p\}$. When $A$ is symmetric, $%
\lambda _{\max }(A)$ and $\lambda _{\min }(A)$ denote its largest and
smallest eigenvalues, respectively. The operators $\rightsquigarrow $ and $%
\operatornamewithlimits{\to}\limits^{p}$ denote convergence in distribution
and in probability, respectively. Besides, we use w.p.a.1 and
a.s. to abbreviate \textquotedblleft with probability approaching
1\textquotedblright\ and \textquotedblleft almost surely\textquotedblright ,
respectively.

\section{Model and Estimation}

In this section, we introduce the PQR model and estimation algorithm.

\subsection{Model}

Consider the PQR model 
\begin{equation}
\mathscr{Q}_{\tau }\left( Y_{it}\bigg|\left\{ X_{j,it}\right\} _{j\in
\lbrack p],t\in \lbrack T]},\left\{ \Theta _{j,it}^{0}\left( \tau \right)
\right\} _{j\in \lbrack p]\cup \{0\},t\in \lbrack T]}\right) =\Theta
_{0,it}^{0}(\tau )+\sum_{j=1}^{p}X_{j,it}\Theta _{j,it}^{0}(\tau ),
\label{eq:model}
\end{equation}%
where $i\in \left[ N\right] ,$ $t\in \left[ T\right] ,$ $\tau \in (0,1)$ is
the quantile index, $Y_{it}$ is the dependent variable, $X_{j,it}$ is the $j$%
-th regressor for individual $i$ at time $t$, $\{\Theta _{j,it}^{0}\}_{j\in
\lbrack p]}$ is the corresponding slope coefficient, $\Theta _{0,it}^{0}$ is
the intercept, and $\mathscr{Q}_{\tau }\left( Y_{it}\bigg|\left\{
X_{j,it}\right\} _{j\in \lbrack p],t\in \lbrack T]},\left\{ \Theta
_{j,it}^{0}\left( \tau \right) \right\} _{j\in \lbrack p]\cup \{0\},t\in
\lbrack T]}\right) $ denotes the conditional $\tau $-quantile of $Y_{it}$
given the regressors $\left\{ X_{j,it}\right\} _{j\in \lbrack p],t\in
\lbrack T]}$ and\ the parameters $\left\{ \Theta _{j,it}^{0}\left( \tau
\right) \right\} _{j\in \lbrack p]\cup {0},t\in \lbrack T]}$.\footnote{%
We will assume that both the intercept term $\Theta _{0,it}^{0}$ and the
slope coefficients $\{\Theta _{j,it}^{0}\}_{j\in \lbrack p]}$ have low-rank
structures, and follow the convention in the panel data literature by
treating the factors to be random. Therefore, $\{\Theta _{j,it}^{0}\}_{j\in
\lbrack p]\cup \{0\}}$ are random as well.} Alternatively, we can rewrite
the above model as 
\begin{align}
& Y=\Theta _{0}^{0}(\tau )+\sum_{j=1}^{p}X_{j}\odot \Theta _{j}^{0}(\tau
)+\epsilon (\tau)\quad \text{and}  \notag \\
& \mathscr{Q}_{\tau }\left( \epsilon _{it} (\tau)\bigg|\left\{ X_{j,it}\right\}
_{j\in \lbrack p],t\in \lbrack T]},\left\{ \Theta _{j,it}^{0}\left( \tau
\right) \right\} _{j\in \lbrack p]\cup \{0\},t\in \lbrack T]}\right) =0,
\end{align}%
where $\epsilon (\tau)$ is the idiosyncratic error matrix with the $(i,t)$%
-th entry being $\epsilon _{it} (\tau)$. Similarly, $X_{j}$, $\Theta
_{j}\left( \tau \right) $, and $Y$ are matrices with the $(i,t)$-th entry
being $X_{j,it}$, $\Theta _{j,it}\left( \tau \right) $, and $Y_{it}$,
respectively. In this model, we assume $p$, the number of regressors, is
fixed and both $N$ and $T$ pass to infinity. In Assumption \ref{ass:1}
below, we characterize the dependence of the data, under which %
\eqref{eq:model} holds.

In the paper, we focus on the panel quantile regression for a fixed $\tau $
and thus suppress the dependence of $\Theta _{j}^{0}(\tau )$ and $\epsilon
(\tau )$ on $\tau $ for notation simplicity. In addition, we impose low-rank
structures for the intercept and slope matrices, i.e., $\text{rank}(\Theta
_{j}^{0})=K_{j}$ for some positive constant $K_{j}$ and for each $j\in
\{0,\cdots ,p\}$. By the singular value decomposition (SVD), we have 
\begin{equation*}
\Theta _{j}^{0}=\sqrt{NT}\mathcal{U}_{j}^{0}\Sigma _{j}^{0}\mathcal{V}%
_{j}^{0\prime }=U_{j}^{0}V_{j}^{0\prime }\text{ }\,\forall \,j=0,\cdots ,p,
\end{equation*}%
where $\mathcal{U}_{j}^{0}\in \mathbb{R}^{N\times K_{j}}$, $\mathcal{V}%
_{j}^{0}\in \mathbb{R}^{T\times K_{j}}$, $\Sigma _{j}^{0}=\text{diag}(\sigma
_{1,j},\cdots ,\sigma _{K_{j},j})$, $U_{j}^{0}=\sqrt{N}\mathcal{U}%
_{j}^{0}\Sigma _{j}^{0}$ with each row being $u_{i,j}^{0\prime }$, and $%
V_{j}^{0}=\sqrt{T}\mathcal{V}_{j}^{0}$ with each row being $v_{t,j}^{0\prime
}$.

The low-rank structure assumption includes several popular cases. For the
intercept term, one commonly assumes that $\Theta _{0,it}^{0}$ to take the
forms $\alpha _{i}^{0},$ $\mu _{t}^{0},$ or $\alpha _{i}^{0}+\mu _{t}^{0}$
in classical PQRs. Then the matrix $\Theta _{0}^{0}$ has rank 1, 1, and 2,
respectively. It is also possible to assume $\Theta _{0,it}^{0}$ to take an
interactive form, say, $\Theta _{0,it}^{0}=\lambda _{0,i}^{0\prime
}f_{0,t}^{0},$ where both $\lambda _{0,i}^{0}$ and $f_{0,t}^{0}$ are $K_{0}$%
-vectors. For the slope matrix $\Theta _{j}^{0}$, $j\in \left[ p\right] ,$
the early PQR models frequently assume that $\Theta _{j,it}^{0}$ is a
constant across $\left( i,t\right) \ $ to yield a homogenous PQR model.
Obviously, such a model is very restrictive by assuming homogenous slope
coefficients. It is possible to allow the slope coefficients to change over
either $i$, or $t$, or both. See the following examples for different
low-rank structures.\bigskip

\noindent \textbf{Example 1.} When $\Theta _{j,it}^{0}=\Theta _{j,i}^{0}$ $%
\forall t\in \left[ T\right] ,$ or $\Theta _{j,it}^{0}=\Theta _{j,t}^{0}$ $%
\forall i\in \left[ N\right] ,$ or $\Theta _{j,it}^{0}=\Theta _{j}^{0}$ $%
\forall \left( i\text{,}t\right) \in \left[ N\right] \times \left[ T\right] $%
, and this holds for all $j\in \left[ p\right] ,$ we have the PQR models
with only individual heterogeneity, with only time heterogeneity, and with
homogeneity, respectively. We observe that $K_{j}=1$ for these three
cases.\bigskip

\noindent \textbf{Example 2.} When $\Theta _{j,it}^{0}=\lambda
_{j,i}^{0}+f_{j,t}^{0}$, we notice that 
\begin{equation*}
\frac{\Theta _{j}^{0}}{\sqrt{NT}}=%
\begin{bmatrix}
\frac{1}{\sqrt{N}} & \frac{\lambda _{j,1}^{0}}{\sqrt{N}} \\ 
\vdots  & \vdots  \\ 
\frac{1}{\sqrt{N}} & \frac{\lambda _{j,N}^{0}}{\sqrt{N}}%
\end{bmatrix}%
\begin{bmatrix}
\frac{f_{j,1}^{0}}{\sqrt{T}} & \cdots  & \frac{f_{j,T}^{0}}{\sqrt{T}} \\ 
\frac{1}{\sqrt{T}} & \cdots  & \frac{1}{\sqrt{T}}%
\end{bmatrix}%
:=A_{j}B_{j}^{\prime }.
\end{equation*}%
Let $\Sigma _{A,j}:=A_{j}^{\prime }A_{j}$ and $\Sigma _{B,j}:=B_{j}^{\prime
}B_{j}$. Let $\Sigma _{A,j}^{\frac{1}{2}}$ (resp. $\Sigma _{B,j}^{\frac{1}{2}%
}$) be the symmetric square root of $\Sigma _{A,j}$ (resp. $\Sigma _{B,j}$).
By eigendecomposition, we have $\Sigma _{A,j}^{\frac{1}{2}%
}=P_{j,1}S_{j,1}P_{j,1}^{\prime }$ and $\Sigma _{B,j}^{\frac{1}{2}%
}=P_{j,2}S_{j,2}P_{j,2}^{\prime }$. Besides, we apply singular value
decomposition to matrix $S_{j,1}P_{j,1}^{\prime }P_{j,2}S_{j,2}$: $%
S_{j,1}P_{j,1}^{\prime }P_{j,2}S_{j,2}=Q_{j,1}R_{j}Q_{j,2}^{\prime }$. Then
it follows that 
\begin{align*}
\frac{\Theta _{j}^{0}}{\sqrt{NT}}& =A_{j}B_{j}^{\prime }=A_{j}\Sigma
_{A,j}^{-\frac{1}{2}}P_{j,1}S_{j,1}P_{j,1}^{\prime
}P_{j,2}S_{j,2}P_{j,2}^{\prime }\Sigma _{B,j}^{-\frac{1}{2}}B_{j}^{\prime }
\\
& =A_{j}\Sigma _{A,j}^{-\frac{1}{2}}P_{j,1}Q_{j,1}R_{j}Q_{j,2}^{\prime
}P_{j,2}^{\prime }\Sigma _{B,j}^{-\frac{1}{2}}B_{j}^{\prime }:=\mathcal{U}%
_{j}^{0}\Sigma _{j}^{0}\mathcal{V}_{j}^{0\prime },
\end{align*}%
where $\mathcal{U}_{j}^{0}=A_{j}\Sigma _{A,j}^{-\frac{1}{2}}P_{j,1}Q_{j,1}$, 
$\Sigma _{j}^{0}=R_{j}$ and $\mathcal{V}_{j}^{0}=B_{j}\Sigma _{B,j}^{-\frac{1%
}{2}}P_{j,2}Q_{j,2}$. Given $P_{j,1}$, $P_{j,2}$, $Q_{j,1}$ and $Q_{j,2}$
are orthonormal matrices, it's easy to  that $\mathcal{U}_{j}^{0}$ and $%
\mathcal{V}_{j}^{0}$ are also orthonormal so that $\mathcal{U}_{j}^{0\prime }%
\mathcal{U}_{j}^{0}=\mathcal{V}_{j}^{0\prime }\mathcal{V}_{j}^{0}=I_{2}$.
When $j=0$, $\left\{ \lambda _{0,i}^{0}\right\} _{i=1}^{N}$ and $\left\{
f_{0,t}^{0}\right\} _{t=1}^{T}$ are usually referred to as the individual
and time fixed effects, respectively, so that the intercept term exhibits an
additive fixed effects structure.\bigskip 

\noindent \textbf{Example 3.} Let $\Theta _{j,it}^{0}=\sum_{k\in \lbrack
K_{j,t}]}\alpha _{j,kt}\mathbf{1}\{i\in G_{j,kt}\}$, where $\left\{
G_{j,kt}\right\} $ forms a partition of $[N]$ for each specific time $t$ and 
$K_{j,t}$ is the number of groups at time $t$. Moreover, let 
\begin{equation*}
\alpha _{j,kt}=\left\{ \begin{aligned} &\alpha_{j,k}^{(1)},
\quad\text{for}\quad t=1,\dots,T_{b},\\ &
\alpha_{j,k}^{(2)},\quad\text{for}\quad t=T_{b}+1,\dots,T,\\ \end{aligned}%
\right.
\end{equation*}%
\begin{equation*}
G_{j,kt}=\left\{ \begin{aligned} &G_{j,k}^{(1)},\quad \text{for}\quad
t=1,\dots,T_{b}, k=1,\dots,K_{j}^{(1)},\\ & G_{j,k}^{(2)},\quad
\text{for}\quad t=T_{b}+1,...,T, k=1,...,K_{j}^{(2)},\\ \end{aligned}\right.
\end{equation*}%
where $K_{j}^{(1)}$ and $K_{j}^{(2)}$ are the number of groups before and
after the break point $T_{b}$. If $K_{j}^{(1)}=K_{j}^{(2)}$, it is clear
that $rank(\Theta _{j}^{0})=1$. If the group structure does not change after
the break but $\alpha _{j,k}^{(1)}=c\alpha _{j,k}^{(2)}$ for some constant $%
c $, we also have $rank(\Theta _{j}^{0})=1$. Except for these two cases, we
can show that 
\begin{equation*}
\Theta _{j}^{0}=%
\begin{bmatrix}
\operatornamewithlimits{\sum}\limits_{k\in \lbrack K^{(1)}]}\alpha
_{j,k}^{(1)}\mathbf{1}\left\{ 1\in G_{j,k}^{(1)}\right\} , & %
\operatornamewithlimits{\sum}\limits_{k\in \lbrack K^{(2)}]}\alpha
_{j,k}^{(2)}\mathbf{1}\left\{ 1\in G_{j,k}^{(2)}\right\} \\ 
\vdots & \vdots \\ 
\operatornamewithlimits{\sum}\limits_{k\in \lbrack K^{(1)}]}\alpha
_{j,k}^{(1)}\mathbf{1}\left\{ i\in G_{j,k}^{(1)}\right\} , & %
\operatornamewithlimits{\sum}\limits_{k\in \lbrack K^{(2)}]}\alpha
_{j,k}^{(2)}\mathbf{1}\left\{ i\in G_{j,k}^{(2)}\right\} \\ 
\vdots & \vdots \\ 
\operatornamewithlimits{\sum}\limits_{k\in \lbrack K^{(1)}]}\alpha
_{j,k}^{(1)}\mathbf{1}\left\{ N\in G_{j,k}^{(1)}\right\} , & %
\operatornamewithlimits{\sum}\limits_{k\in \lbrack K^{(2)}]}\alpha
_{j,k}^{(2)}\mathbf{1}\left\{ N\in G_{j,k}^{(2)}\right\}%
\end{bmatrix}%
\begin{bmatrix}
\iota _{T_{b}} &  & \mathbf{0}_{T_{b}} \\ 
\mathbf{0}_{T-T_{b}} &  & \iota _{T-T_{b}}%
\end{bmatrix}%
^{\prime }
\end{equation*}%
where $\iota _{T_{b}}$ is a $T_{b}\times 1$ vector of ones and $\mathbf{0}%
_{T_{b}}$ is a $T_{b}\times 1$ vector of zeros. In this case, we notice that 
$rank(\Theta _{j}^{0})=2$.\bigskip

\noindent \textbf{Example 4.} When $\Theta _{j,it}^{0}=\lambda
_{j,i}^{0\prime }f_{j,t}^{0}$ with $\lambda _{j,i}^{0}$ and $f_{j,t}^{0}$
being two $K_{j}$-vectors, we have the IFEs structure. This is the most
general example without further restrictions.

Like \cite{chernozhukov2019inference}, we assume that for each $j\in \left[ p%
\right] ,$ $X_{j,it}$ exhibits a factor structure: $X_{j,it}=\mu
_{j,it}+e_{j,it}=l_{j,i}^{0\prime }w_{j,t}^{0}+e_{j,it},$ where $w_{j,t}^{0}$
and $l_{j,i}^{0}$ are the factors and factor loadings of dimension $r_{j}.$

\subsection{Estimation Algorithm}

In this subsection we provide the estimation algorithm by assuming that $%
K_{j}$ are all known for all $j$. In the next subsection, we will introduce
a rank estimation method to estimate $K_{j}$ consistently.

Define the check function $\rho _{\tau }(u)=u\left( \tau -\mathbf{1}\{u\leq
0\}\right) $. The estimation procedure goes as follows:

\begin{itemize}[leftmargin=40pt]
\item[Step 1:] \textbf{Sample Splitting and Nuclear Norm Regularization.} Along the cross-section span, randomly split the sample into three
subsets denoted as $I_{1}$, $I_{2}$ and $I_{3}$, where $I_{\ell }$ has $%
N_{\ell }$ individuals such that $N_{1}\approx N_{2}\approx N_{3}\approx N/3$%
. Using the data with $(i,t)\in I_{1}\times \lbrack T]$, we run the nuclear
norm regularized quantile regression (QR) and obtain $\{\tilde{\Theta}%
_{j}^{(1)}\}_{j\in \{0,\cdots ,p\}}$, i.e., 
\begin{equation}
\{\tilde{\Theta}_{j}^{(1)}\}_{j=0}^{p}=\argmin\limits_{\left\{ \Theta
_{j}\right\} _{j=0}^{p}}\frac{1}{N_{1}T}\sum_{i\in I_{1}}\sum_{t=1}^{T}\rho
_{\tau }(Y_{it}-\sum_{j=1}^{p}X_{j,it}\Theta _{j,it}-\Theta
_{0,it})+\sum_{j=0}^{p}\nu _{j}\left\Vert \Theta _{j}\right\Vert _{\ast },
\label{obj1}
\end{equation}%
where $\nu _{j}$ is a tuning parameter. For each $j$, conduct the SVD: $%
\frac{1}{\sqrt{N_{1}T}}\tilde{\Theta}_{j}^{(1)}=\hat{\tilde{\mathcal{U}}}%
_{j}^{(1)}\hat{\tilde{\Sigma}}_{j}^{(1)}\hat{\tilde{\mathcal{V}}}%
_{j}^{(1)\prime }$, where $\hat{\tilde{\Sigma}}_{j}^{(1)}$ is the diagonal
matrix with the diagonal elements being the descending singular values of $%
\tilde{\Theta}_{j}^{(1)}$. Let $\tilde{\mathcal{V}}_{j}^{(1)}$ consists the
first $K_{j}$ columns of $\hat{\tilde{\mathcal{V}}}_{j}^{(1)}.$ Let $\tilde{V%
}_{j}^{(1)}=\sqrt{T}\tilde{\mathcal{V}}_{j}^{(1)}$ and $\tilde{v}%
_{t,j}^{(1)\prime }$ be the $t$-th row of $\tilde{V}_{j}^{(1)}$ $\forall
t\in \lbrack T]$.

\item[Step 2:] \textbf{Row- and Column-Wise Quantile Regression.} Using the
data with $(i,t)\in I_{2}\times \lbrack T]$, we first run the row-wise QR
of\ $Y_{it}$ on $\left( \tilde{v}_{t,0}^{(1)},\left\{ \tilde{v}%
_{t,j}^{(1)}X_{j,it}\right\} _{j\in \lbrack p]}\right) $ to obtain $\{\dot{u}%
_{i,j}^{(1)}\}_{j=0}^{p}$ for $i\in I_{2}$, and then run the column-wise QR
of $Y_{it}$ on $(\dot{u}_{i,0}^{(1)},\{\dot{u}_{i,j}^{(1)}X_{j,it}\}_{j\in
\lbrack p]})$ to obtain $\{\dot{v}_{t,j}^{(1)}\}_{j=0}^{p}$ for $t\in
\lbrack T]$. That is, 
\begin{align}
\{\dot{u}_{i,j}^{(1)}\}_{j=0}^{p}& =\argmin\limits_{\{u_{i,j}\}_{j\in
\lbrack p]\cup \{0\}}}\frac{1}{T}\sum_{t\in \lbrack T]}\rho _{\tau }\left(
Y_{it}-u_{i,0}^{\prime }\tilde{v}_{t,0}^{(1)}-\sum_{j=1}^{p}u_{i,j}^{\prime }%
\tilde{v}_{t,j}^{(1)}X_{j,it}\right) ,\forall i\in I_{2},  \label{obj2} \\
\{\dot{v}_{t,j}^{(1)}\}_{j=0}^{p}& =\argmin\limits_{\{v_{t,j}\}_{j\in
\lbrack p]\cup \{0\}}}\frac{1}{N_{2}}\sum_{i\in I_{2}}\rho _{\tau }\left(
Y_{it}-v_{t,0}^{\prime }\dot{u}_{i,0}^{(1)}-\sum_{j=1}^{p}v_{t,j}^{\prime }%
\dot{u}_{i,j}^{(1)}X_{j,it}\right) ,\forall t\in \lbrack T].
\end{align}%
Similarly, we run the row-wise QR of $Y_{it}$ on $(\dot{v}_{t,0}^{(1)},\{%
\dot{v}_{t,j}^{(1)}X_{j,it}\}_{j\in \lbrack p]})$ to obtain $\{\dot{u}%
_{i,j}^{(1)}\}_{j=0}^{p}$ for $i\in I_{3}$, i.e., 
\begin{equation*}
\{\dot{u}_{i,j}^{(1)}\}_{j=0}^{p}=\argmin\limits_{\{u_{i,j}\}_{j\in \lbrack
p]\cup \{0\}}}\frac{1}{T}\sum_{t\in \lbrack T]}\rho _{\tau }\left(
Y_{it}-u_{i,0}^{\prime }\dot{v}_{t,0}^{(1)}-\sum_{j=1}^{p}u_{i,j}^{\prime }%
\dot{v}_{t,j}^{(1)}X_{j,it}\right) ,\forall i\in I_{3}.
\end{equation*}

\item[Step 3:] \textbf{Debiasing.}

\begin{itemize}
\item[Step 3.1:] For each $j\in \lbrack p]$, we conduct the principle
component analysis (PCA) for $X_{j,it}$ with $\left( i,t\right) \in \lbrack
N]\times \lbrack T]$ to obtain the factor and factor loading estimates as 
\begin{equation}
\left\{ \hat{l}_{j,i},\hat{w}_{j,t}\right\} _{i\in \lbrack N],t\in \lbrack
T]}=\operatornamewithlimits{\argmin}\limits_{\left\{ l_{j,i},w_{j,t}\right\}
_{i\in \lbrack N],t\in \lbrack T]}}\frac{1}{NT}\sum_{i\in \lbrack
N]}\sum_{t\in \lbrack T]}\left( X_{j,it}-l_{j,i}^{\prime }w_{j,t}\right)
^{2},  \label{debias_0}
\end{equation}%
subject to the normalizations: $\frac{1}{N}\sum_{i=1}^{N}l_{i,j}l_{i,j}^{%
\prime }=I_{r_{j}}$ and $\frac{1}{T}\sum_{t=1}^{T}w_{j,t}w_{j,t}^{\prime }$
is a diagonal matrix with descending diagonal elements. Then we define $\hat{%
\mu}_{j,it}=\hat{l}_{j,i}^{\prime }\hat{w}_{j,t}$ and $\hat{e}%
_{j,it}=X_{j,it}-\hat{\mu}_{j,it}$. 

\item[Step 3.2:] For $(i,t)\in I_{3}\times [T]$, let $\tilde{Y}_{it}=Y_{it}-%
\operatornamewithlimits{\sum}\limits_{j=1}^{p}\hat{\mu}_{j,it}\dot{u}%
_{i,j}^{(1)\prime}\dot{v}_{t,j}^{(1)}.$ We run the row-wise QR\ of $\tilde{Y}%
_{it}$ on $(\dot{v}_{t,0}^{(1)},\{\dot{v}_{t,j}^{(1)}\hat{e}_{j,it}\}_{j\in[p%
]})$ to obtain the final estimates $\hat{u}_{i,j}^{(3,1)}$, i.e., 
\begin{equation}
\{\hat{u}_{i,j}^{(3,1)}\}_{j=0}^{p}=\argmin\limits_{\{u_{i,j}\}_{j=0}^{p}}%
\frac{1}{T}\sum_{t\in[T]}\rho_{\tau }\left(\tilde{Y}_{it}-u_{i,0}^{\prime}%
\dot{v}_{t,0}^{(1)}-\sum_{j=1}^{p}u_{i,j}^{\prime}\dot{v}_{t,j}^{(1)}\hat{e}%
_{j,it}\right) ,\forall i\in I_{3}.  \label{debias_1}
\end{equation}
Updating $\hat{Y}_{it}=Y_{it}-\sum_{j=1}^{p}\hat{\mu}_{j,it}\hat{u}%
_{i,j}^{(3,1)\prime}\dot{v}_{t,j}^{(1)}$, we run the column-wise QR\ of $%
\hat{Y}_{it}$ on $(\hat{u}_{i,0}^{(3,1)},\{\hat{u}_{i,j}^{(3,1)}\hat{e}%
_{j,it}\}_{j\in[p]})$ to obtain $\hat{v}_{t,j}^{(3,1)}$, i.e., 
\begin{equation}
\{\hat{v}_{t,j}^{(3,1)}\}_{j=0}^{p}=\argmin\limits_{\{v_{t,j}\}_{j=0}^{p}}%
\frac{1}{N_{3}}\sum_{i\in I_{3}}\rho_{\tau }\left(\hat{Y}_{it}-v_{t,0}^{%
\prime}\hat{u}_{i,0}^{(3,1)}-\sum_{j=1}^{p}v_{t,j}^{\prime}\hat{u}%
_{i,j}^{(3,1)}\hat{e}_{j,it}\right) ,\forall t\in[T].  \label{debias_2}
\end{equation}
\end{itemize}
\end{itemize}

In order to obtain the final estimators for the full sample, we propose to
switch the role of each subsample for the low-rank estimation, row- and
column-wise QR and debiasing, then repeat Steps 1-3 to obtain $\left\{ \hat{u%
}_{i,j}^{(a,b)}\right\} _{j=0}^{p}$ and $\left\{ \hat{v}_{t,j}^{(a,b)}\right%
\} _{j=0}^{p}$ for $a\in \lbrack 3]$ and $b\in \lbrack 3]\setminus \{a\}$.
Here $(a,b)$ denotes the final estimates for subsample $I_{a}$ obtained from
the first step NNR estimates with subsample $I_{b}$. Table \ref%
{tab:estimator symbol} shows the final estimators we obtain by using
different combination of subsamples. 
\begin{table}[th]
\caption{Estimators using different subsamples at different steps in
algorithm.}
\label{tab:estimator symbol}\centering  
\begin{tabular}{cccc}
\toprule \toprule Step 1 $(b)$ & Step 2 & Step 3 $(a)$ & estimators $(a,b)$ \\ 
\midrule $I_{1}$ & $I_{2}$ & $I_{3}$ & $\hat{u}_{i,j}^{(3,1)}$, $\hat{v}%
_{t,j}^{(3,1)}$ \\ 
$I_{2}$ & $I_{1}$ & $I_{3}$ & $\hat{u}_{i,j}^{(3,2)}$, $\hat{v}%
_{t,j}^{(3,2)} $ \\ 
$I_{1}$ & $I_{3}$ & $I_{2}$ & $\hat{u}_{i,j}^{(2,1)}$, $\hat{v}%
_{t,j}^{(2,1)} $ \\ 
$I_{3}$ & $I_{1}$ & $I_{2}$ & $\hat{u}_{i,j}^{(2,3)}$, $\hat{v}%
_{t,j}^{(2,3)} $ \\ 
$I_{2}$ & $I_{3}$ & $I_{1}$ & $\hat{u}_{i,j}^{(1,2)}$, $\hat{v}%
_{t,j}^{(1,2)} $ \\ 
$I_{3}$ & $I_{2}$ & $I_{1}$ & $\hat{u}_{i,j}^{(1,3)}$, $\hat{v}%
_{t,j}^{(1,3)} $ \\ 
\bottomrule &  &  & 
\end{tabular}%
\end{table}

Several remarks are in order. First, we randomly split the full sample into
three subsamples, each playing a significant role in the algorithm. We use
the first subsample for the low-rank estimation to obtain the preliminary
NNR estimators of the submatrices of the intercept and slope matrices. But
these estimators are only consistent in terms of Frobenius norm, and one
cannot derive the pointwise or uniform convergence rates for them. With the
low-rank estimates, we use the second subsample to do the row- and
column-wise QRs and can now establish the uniform convergence rates for each
row of factor and factor loading estimators. Then we use the remaining
subsample to debias the second-stage estimator and to obtain the final
estimators that have the desirable asymptotic properties.

Second, to reduce the randomness of sample splitting, one can run the
estimation algorithm several times with different splittings in practice.
Once one obtains factor and factor loading estimates, one can construct
estimators for $\Theta _{j}^{0}$ under different splittings and then choose
the one specific splitting which yields the minimum quantile objective function.

Third, the bias in the second-stage estimator is inherent from the
first-stage NNR estimator. We follow the lead of \cite%
{chernozhukov2019inference} to assume that $X_{j,it}$ has a factor structure
with an additive idiosyncratic term, and remove the bias by a QR with the
demeaned $X_{j,it}$ as regressors. In the least squares panel regression
framework, the objective function is smooth and one has closed-form
solutions in the last stage so that \cite{chernozhukov2019inference} only
need to split the sample into two subsamples. In contrast, in the PQR
framework, the objective function is non-smooth, we do not have closed-form
solutions in any stage. In order to remove the bias from the early stage
estimation and to derive the distributional results, we need to split the
sample into three subsamples.

To save space, we relegate the detailed algorithm for the nuclear norm
regularization to the online supplement.

\subsection{Rank Estimation}

In this subsection we discuss how to estimate the ranks $K_{j}$
consistently. To estimate the ranks, we consider the full sample NNR QR
estimation: 
\begin{equation}
\{\tilde{\Theta}_{j}\}_{j=0}^{p}=\argmin\limits_{\left\{ \Theta _{j}\right\}
_{j=0}^{p}}\,\,\frac{1}{NT}\sum_{i=1}^{N}\sum_{t=1}^{T}\rho _{\tau
}(Y_{it}-\sum_{j=1}^{p}X_{j,it}\Theta _{j,it}-\Theta
_{0,it})+\sum_{j=0}^{p}\nu _{j}\left\Vert \Theta _{j}\right\Vert _{\ast }.
\label{pre rank}
\end{equation}%
For $j\in \left\{ 0,\cdots ,p\right\} $, we estimate $K_{j}$ by the popular
singular value thresholding (SVT) as follows%
\begin{equation*}
\hat{K}_{j}=\sum_{m}\mathbf{1}\left\{ \lambda _{m}\left( \tilde{\Theta}%
_{j}\right) \geq 0.6\left( NT\nu _{j}\left\Vert \tilde{\Theta}%
_{j}\right\Vert _{op}\right) ^{1/2}\right\} .
\end{equation*}%
It is standard to show that $\mathbb{P}(\hat{K}_{j}=K_{j})\rightarrow 1$ as $%
\left( N,T\right) \rightarrow \infty $ under some regularity conditions
given in the next section; see also Proposition D.1 in \cite%
{chernozhukov2019inference} and Theorem 2 in \cite{Hong_Su_Jiang2022}. Since
the ranks can be estimated consistently, we assume that they are known in
the asymptotic theory below.


\section{Asymptotic Theory}

In this section, we study the asymptotic properties of the estimators
introduced in the last section.

\subsection{First Stage Estimator}

Recall that $X_{j,it}=\mu _{j,it}+e_{j,it}=l_{j,i}^{0\prime
}w_{j,t}^{0}+e_{j,it}$ for each $j\in \lbrack p]$. Let $%
X_{it}=(X_{1,it},...,X_{p,it})^{\prime }$ and $%
e_{it}=(e_{1,it},...,e_{p,it})^{\prime }.$ Define $\epsilon _{i}=\left(
\epsilon _{i1},\cdots ,\epsilon _{it}\right) ^{\prime }$, $e_{j,i}=\left(
e_{j,i1},\cdots ,e_{j,iT}\right) ^{\prime }$, $W_{j}^{0}$ as the $T\times
r_{j}$ matrix with each row being $w_{j,t}^{0\prime }$, and $V_{j}^{0}$ as
the $T\times K_{j}$ matrix with each row being $v_{t,j}^{0\prime }$. Further
define $a_{it}=\tau -\mathbf{1}\left\{ \epsilon _{it}\leq 0\right\} $ with $%
a_{i}=\left( a_{i1},\cdots ,a_{iT}\right) ^{\prime }$ and $a=\left(
a_{1},\cdots ,a_{N}\right) ^{\prime }$. Throughout the paper, we treat the
factors $\{v_{t,j}^{0}\}_{t\in \lbrack T],j\in \lbrack p]\cup \{0\}}$ and $%
\{w_{j,t}^{0}\}_{t\in \lbrack T],j\in \lbrack p]}$ as random and their
loadings $\{u_{i,j}^{0}\}_{i\in \lbrack N],j\in \lbrack p]\cup \{0\}}$ and $%
\{l_{j,i}^{0}\}_{i\in \lbrack N],j\in \lbrack p]}$ as deterministic.

Table \ref{sigma fields} defines several $\sigma$-fields. We use $\mathscr{D}
$ to denote the minimal $\sigma$-field generated by $\left\{V_{j}^{0}\right%
\}_{j\in [p]\cup \{0\}}\bigcup \left\{W_{j}^{0}\right\}_{j\in [p]};$ the
superscripts $I_{1}$ and $I_{1}\cup I_{2}$ are associated with the first
subsample and the first two subsamples, respectively.  For example, $%
\mathscr{D}_{e_{i}}^{I_{1}}$ denotes the minimal $\sigma$-field generated by 
$\mathscr{D},$ $\left\{e_{it}\right\}_{t\in \left[T\right] }$ and $\left\{
\epsilon_{it},e_{it}\right\}_{i\in I_{1},t\in \left[T\right] }.$ 
\begin{table}[h]
\caption{Definition of various $\protect\sigma $-fields}
\label{sigma fields}\centering%
\begin{tabular}{cc}
\toprule\toprule Notation & $\sigma$-fields generated by \\ 
\midrule$\mathscr{D}$ & $\left\{ V_{j}^{0}\right\}_{j\in [p]\cup
\{0\}}\bigcup \left\{ W_{j}^{0}\right\}_{j\in [p]}$ \\ 
\midrule$\mathscr{D}_{e_{it}}$ & $\mathscr{D}\bigcup e_{it}$ \\ 
\midrule$\mathscr{D}_{e_{i}}$ & $\mathscr{D}\bigcup \left\{ e_{it}\right\}
_{t\in[T]}$ \\ 
\midrule$\mathscr{D}_{e}$ & $\mathscr{D}\bigcup \left\{ e_{it}\right\} _{i\in%
[N], t\in[T]}$ \\ 
\midrule$\mathscr{D}^{I_{1}\cup I_{2}}$ & $\mathscr{D}\bigcup \left\{
\epsilon_{it},e_{it}\right\}_{i\in I_{1}\cup I_{2},t\in[T]}$ \\ 
\midrule$\mathscr{D}^{I_{1}}_{\{e_{is}\}_{s<t}}$ & $\mathscr{D}\bigcup
\{e_{is}\}_{s<t}\bigcup \left\{
\epsilon_{i^{*}t^{*}},e_{i^{*}t^{*}}\right\}_{i^{*}\in I_{1},t^{*}\in[T]}$
\\ 
\midrule$\mathscr{D}_{e_{i}}^{I_{1}}$ & $\mathscr{D}\bigcup \left\{
e_{it}\right\}_{t\in[T]}\bigcup \left\{ \epsilon
_{i^{*}t^{*}},e_{i^{*}t^{*}}\right\}_{i^{*}\in I_{1},t^{*}\in[T]}$ \\ 
\midrule$\mathscr{D}_{e_{i}}^{I_{1}\cup I_{2}}$ & $\mathscr{D}\bigcup
\left\{ e_{it}\right\}_{t\in[T]}\bigcup \left\{ \epsilon
_{i^{*}t^{*}},e_{i^{*}t^{*}}\right\}_{i^{*}\in I_{1}\cup I_{2},t^{*}\in[T]}$
\\ 
\midrule$\mathscr{D}_{e}^{I_{1}\cup I_{2}}$ & $\mathscr{D}\bigcup \left\{
e_{it}\right\}_{i\in [N,]t\in[T]}\bigcup \left\{ \epsilon
_{it},e_{it}\right\}_{i\in I_{1}\cup I_{2},t\in[T]}$ \\ 
\bottomrule & 
\end{tabular}%
\end{table}

Let $M$ denote a generic bounded constant that may vary across places. Let $%
\mathscr{G}_{i,t-1}$ denote the minimal $\sigma $-field generated by $%
\mathscr{D}\cup \{e_{ls}\}_{l\leq i-1,s\in \lbrack T]}\cup \{e_{is}\}_{s\leq
t}\cup \{\epsilon _{ls}\}_{l\leq i-1,s\in \lbrack T]}\cup \{\epsilon
_{is}\}_{s\leq t-1}$. Let $\mathsf{F}_{it}(\cdot )$ and $\mathsf{f}%
_{it}(\cdot )$ be the conditional cumulative distribution function (CDF) and
probability density function (PDF) of $\epsilon _{it}$ given $\mathscr{G}%
_{i,t-1}$, respectively. Similarly, let $\mathfrak{F}_{it}(\cdot )$ and $%
\mathfrak{f}_{it}(\cdot )$ denote the conditional CDF and PDF of $\epsilon
_{it}$ given $\mathscr{D}_{e_{i}};$ $F_{it}(\cdot )$ and $f_{it}(\cdot )$
denote the conditional CDF and PDF of $\epsilon _{it}$ given $\mathscr{D}%
_{e}.$ Let $\mathsf{f}_{it}^{\prime }\left( \cdot \right) $, $\mathfrak{f}%
_{it}^{\prime }\left( \cdot \right) ,$ and $f_{it}^{\prime }\left( \cdot
\right) $ denotes the first derivative of the density $\mathsf{f}_{it}\left(
\cdot \right) $, $\mathfrak{f}_{it}\left( \cdot \right) ,$ and $f_{it}\left(
\cdot \right) ,$ respectively.

We make the following assumptions.

\begin{ass}
\begin{itemize}
\item[(i)] $\left\{ \epsilon_{it},e_{it}\right\}_{t\in[T]}$ are
conditionally independent across $i$ given $\mathscr{D}$.

\item[(ii)] $\mathbb{E}\left(a_{it}\bigg |\mathscr{D}_{e}\right) =0$.

\item[(iii)] For each $i$, $\left\{ \epsilon_{it},t\geq 1\right\} $ is
strong mixing conditional on $\mathscr{D}_{e_{i}}$, and $\{\left(%
\epsilon_{it},e_{it}\right) ,t\geq 1\}$ is strong mixing conditional on $%
\mathscr{D}$. Both mixing coefficients are upper bounded by $%
\alpha_{i}(\cdot)$ such that $\max_{i\in [N]}\alpha_{i}(z)\leq M\alpha ^{z}$
for some constant $\alpha \in \left(0,1\right) $.

\item[(iv)] $\max_{i\in [N]}\frac{1}{T}\sum_{t\in [ T]}\left\Vert
X_{it}\right\Vert_{2}^{3}\leq M~$a.s., $\max_{t\in[T]}\frac{1}{N_{2}}%
\sum_{i\in I_{2}}\left\Vert X_{it}\right\Vert_{2}^{4}\leq M~$a.s., \newline
$\max_{\left(i,t\right) \in [N]\times [T]}\mathbb{E}\left[ \left\Vert
X_{it}\right\Vert_{2}^{3}\bigg|\mathscr{D}\right] \leq M$a.s., $%
\max_{i\in [N]}\sqrt{\frac{1}{T}\sum_{t\in[T]}\left[ \mathbb{E}%
\left(\epsilon_{it}^{2}\bigg|\mathscr{D}_{e_{i}}\right) \right] ^{2}}$ $\leq
M~a.s.$, and \newline
$\max_{\left(i,t\right) \in [N]\times [T]}\mathbb{E}\left(\left\Vert
X_{it}\right\Vert_{2}^{2}\bigg|\mathscr{D}_{\left\{e_{is}\right\}_{s<t}}
\right)$ $\leq M~a.s.$

\item[(v)] For $j\in [p]$, there exists a positive sequence $\xi_{N}$ such
that $\max_{\left(i,t\right) \in [N]\times [T]}\left\vert
X_{j,it}\right\vert \leq \xi_{N}~a.s.$

\item[(vi)] $\min_{\left(i,t\right) \in [N]\times [T]} \mathsf{f}_{it}(0)\geq \underline{\mathsf{f}}>0$ and $%
\max_{\left(i,t\right) \in [N]\times [T]}\sup_{\epsilon}\left\vert \mathsf{f}%
_{it}^{\prime }\left(\epsilon \right) \right\vert\leq \bar{\mathsf{f}}%
^{\prime }$.

\item[(vii)] $\min_{\left(i,t\right) \in [N]\times [T]}\mathfrak{f}_{it}(0)\geq \underline{\mathfrak{f}}>0$ and $%
\max_{\left(i,t\right) \in [N]\times [T]}\sup_{\epsilon}\left\vert \mathfrak{%
f}_{it}^{\prime }\left(\epsilon \right) \right\vert\leq \bar{\mathfrak{f}}%
^{\prime }$.

\item[(viii)] $\min_{\left(i,t\right) \in [N]\times [T]}f_{it}(0)\geq \underline{f}>0$ and $\max_{\left(i,t\right)\in
[N]\times [T]}\sup_{\epsilon }\left\vert f_{it}^{\prime}\left(\epsilon
\right) \right\vert \leq \bar{f}^{\prime }$.

\item[(ix)] $\frac{\xi_{N}^{4}\log (N\vee T)\sqrt{N\vee T}}{N\wedge T}=o(1)$
and $\frac{\left(\frac{N}{T}\vee 1\right) ^{1/2}}{\left(N\wedge T\right) ^{%
\frac{1}{4+2\vartheta }}}\left(\log (N\vee T)\right) ^{\frac{3+\vartheta }{%
4+2\vartheta }}\xi_{N}^{\frac{5+\vartheta }{2+\vartheta }}=o(1)$ for any $%
\vartheta >0$.
\end{itemize}

\label{ass:1}
\end{ass}

Assumptions \ref{ass:1}(i) imposes conditional independence of the error
terms and covariates $X_{j,it}$ given the fixed effects. Assumptions \ref%
{ass:1}(ii) imposes the moment condition for QR. Assumptions \ref{ass:1}%
(iii) imposes the weak dependence assumption along the time dimension via
the use of the notion of conditional strong mixing. See \cite%
{Prakasa_Rao2009} for the definition of conditional strong mixing and \cite%
{Su_Chen2013} for an application in the panel setup. Assumptions \ref{ass:1}%
(iv)-(v) essentially imposes some conditions on the moments and tail
behavior of the both covariates and errors. Note that we allow $X_{j,it}$ to
have an infinite support. Assumptions \ref{ass:1}(vi)-(viii), which are used
in the proofs of Theorems \ref{Thm1}, \ref{Thm2} and \ref{Thm3},
respectively, specify conditions on the conditional density of $\epsilon
_{it}$ given different $\sigma $-fields. Assumption \ref{ass:1}(ix) imposes
some restrictions on $N$, $T$ and $\xi _{N}$ in order to obtain the error
bound of NNR estimators and to achieve the unbiasedness. It allows not only the case that $N$ and $T$ diverge to infinity at the the same rate, but also the case that $N$ diverges to infinity not too faster than $T$, and vice versa.

\begin{ass}
\label{ass:2} $\Theta _{0}^{0}$ is the fixed effect matrix with fixed rank $%
K_{0}$ and $\left\Vert \Theta _{0}^{0}\right\Vert _{\max }\leq M$. For each $%
j\in \lbrack p],$ $\Theta _{j}^{0}$ is the slope matrix of regressor $j$
with rank being $K_{j}$ such that $\max_{j\in \lbrack p]}\left\Vert \Theta
_{j}\right\Vert _{\max }\leq M$ and $\max_{j\in \lbrack p]}K_{j}\leq \bar{K}$
for some fixed finite $\bar{K}.$
\end{ass}

Assumption \ref{ass:2} is the low-rank assumption for the intercept and
slope matrices, which is the key assumption for the NNR.
The uniform boundedness of elements of these matrices facilitates the
asymptotic analysis,  but can be relaxed at the cost of more lengthy argument.
See \cite{ma2020detecting} for a similar condition.

\begin{ass}
\label{ass:3} There exist some constants $C_{\sigma}$ and $c_{\sigma}$ such
that 
\begin{align*}
\infty>C_{\sigma}\geq \lim\sup_{N,T} \max_{j\in [p]\cup\{0\}}
\sigma_{1,j}\geq \lim\inf_{N,T} \min_{j\in [p]\cup\{0\}}
\sigma_{K_{j},j}\geq c_{\sigma}>0.
\end{align*}
\end{ass}

Assumption \ref{ass:3} imposes some conditions on the singular values of the
coefficient matrices. It implies that we only allow pervasive factors when
these matrices are written as a factor structure. Such an assumption is
common in the literature; see, e.g., Assumption 3 in \cite{ma2020detecting}.

To introduce the next assumption, we need some notation. Let $\Theta
_{j}^{0}=R_{j}\Sigma _{j}S_{j}^{\prime }$ be the SVD for $\Theta _{j}^{0}$.
Further decompose $R_{j}=\left( R_{j,r},R_{j,0}\right) $ with $R_{j,r}$
being the singular vectors corresponding to the nonzero singular values, $%
R_{j,0}$ being the singular vectors corresponding to the zero singular
values. Decompose $S_{j}=\left( S_{j,r},S_{j,0}\right) $ with $S_{j,r}$ and $%
S_{j,0}$ defined analogously. For any matrix $W\in \mathbb{R}^{N\times T}$,
we define 
\begin{equation*}
\mathcal{P}_{j}^{\bot }\left( W\right) =R_{j,0}R_{j,0}^{\prime
}WS_{j,0}S_{j,0}^{\prime },\quad \mathcal{P}_{j}\left( W\right) =W-\mathcal{P%
}_{j}^{\bot }\left( W\right) ,
\end{equation*}%
where $\mathcal{P}_{j}\left( W\right) $ and $\mathcal{P}_{j}^{\bot }\left(
W\right) $ are the linear projection of matrix $W$ onto the low-rank space
and its orthogonal space, respectively. Let $\Delta _{\Theta _{j}}=\Theta
_{j}-\Theta _{j}^{0}$ for any $\Theta _{j}$. With some positive constants $C_{1}$ and $C_{2}$, we define the
following cone-like restricted set: 
\begin{equation*}
\mathcal{R}(C_{1},C_{2}):=\left\{ \left(\{\Delta _{\Theta
_{j}}\}_{j=0}^{p}\right):\sum_{j=0}^{p}\left\Vert \mathcal{P}_{j}^{\bot }(\Delta
_{\Theta _{j}})\right\Vert _{\ast }\leq C_{1}\sum_{j=0}^{p}\left\Vert 
\mathcal{P}_{j}(\Delta _{\Theta _{j}})\right\Vert _{\ast
},\,\sum_{j=0}^{p}\left\Vert \Delta _{\Theta _{j}}\right\Vert _{F}^{2}\geq
C_{2}\sqrt{NT}\right\} .
\end{equation*}


\begin{ass}
\label{ass:4} Let $C_{2}>0$ be a sufficiently large but fixed constant.
There are constants $C_{3},C_{4}$, such that, uniformly over $(\{\Delta
_{\Theta _{j}}\}_{j=0}^{p})\in \mathcal{R}(3,C_{2})$, we have 
\begin{equation*}
\left\Vert \Delta _{\Theta _{0}}+\sum_{j=1}^{p}\Delta _{\Theta _{j}}\odot
X_{j}\right\Vert _{F}^{2}\geq C_{3}\sum_{j=0}^{p}\left\Vert \Delta _{\Theta
_{j}}\right\Vert _{F}^{2}-C_{4}(N+T)~w.p.a.1.
\end{equation*}%
The same condition holds when $\Theta_j^0$ is replaced by $\{\Theta_{j,it}^0\}_{i \in I_a, t \in [T]}$ for $a = 1,2,3.$
\end{ass}

Assumption \ref{ass:4} parallels the restricted strong convexity (RSC)
condition in Assumption 3.1 of \cite{chernozhukov2019inference} who also
provide some sufficient primitive conditions.

For any $j\in \left\{ 0,\cdots ,p\right\} $, define $\tilde{\Delta}_{\Theta
_{j}}=\tilde{\Theta}_{j}-\Theta _{j}^{0}$ and $\tilde{\Delta}_{\Theta
_{j}}^{(1)}=\tilde{\Theta}_{j}^{(1)}-\Theta _{j}^{0,(1)}$, where $\Theta
_{j}^{0,(1)}=\left\{ \Theta _{j,it}^{0}\right\} _{i\in I_{1},t\in \lbrack
T]} $. The following theorem establishes the convergence rates of the NNR
estimators of the coefficient matrices.

\begin{theorem}
\label{Thm1} If Assumptions \ref{ass:1}-\ref{ass:4} hold, for $\forall
j\in\left\{0,\cdots,p\right\}$, we have

\begin{itemize}
\item[(i)] $\frac{1}{\sqrt{NT}}\left\Vert \tilde{\Delta}_{\Theta_{j}}\right%
\Vert_{F}=O_{p}\left(\sqrt{\frac{\log (N\vee T)}{N\wedge T}}%
\xi_{N}^{2}\right) $, $\frac{1}{\sqrt{NT}}\left\Vert \tilde{\Delta}%
_{\Theta_{j}}^{(1)}\right\Vert_{F}=O_{p}\left(\sqrt{\frac{\log (N\vee T)}{%
N\wedge T}}\xi_{N}^{2}\right) $,

\item[(ii)] $\max_{k\in [K_{j}]}\left\vert \tilde{\sigma}_{k,j}-\sigma_{k,j}%
\right\vert =O_{p}\left(\sqrt{\frac{\log (N\vee T)}{N\wedge T}}%
\xi_{N}^{2}\right) $, $\max_{k\in [K_{j}]}\left\vert \tilde{\sigma}%
_{k,j}^{(1)}-\sigma_{k,j}\right\vert =O_{p}\left(\sqrt{\frac{\log (N\vee T)}{%
N\wedge T}}\xi_{N}^{2}\right) $,

\item[(iii)] $\frac{1}{\sqrt{T}}\left\Vert V_{j}^{0}-\tilde{V}%
_{j}O_{j}\right\Vert_{F}=O_{p}\left(\sqrt{\frac{\log (N\vee T)}{N\wedge T}}%
\xi_{N}^{2}\right) $, $\frac{1}{\sqrt{T}}\left\Vert V_{j}^{0}-\tilde{V}%
_{j}^{(1)}O_{j}^{(1)}\right\Vert_{F}=O_{p}\left(\sqrt{\frac{\log \,N\vee T}{%
N\wedge T}}\xi_{N}^{2}\right) $,
\end{itemize}

where $O_{j}$ and $O_{j}^{(1)}$ are some orthogonal rotation matrices
defined in the proof.
\end{theorem}

\textbf{Remark 1.} Theorem \ref{Thm1}(i) reports the \textquotedblleft
rough\textquotedblright\ convergence rates of the NNR estimators of the
coefficient matrices in terms of Frobenius norm for both the full-sample and
sub-sample estimators. Unlike the traditional $\left( N\wedge T\right)
^{-1/2}$-rate in the least squares framework, NNR estimators' convergence rates in the PQR
framework usually have an additional $\sqrt{\log (N\vee T)}$ term due to the use
of some exponential inequalities. The extra term $\xi
_{N}^{2}$\ in our rate is due to the upper bound of $|X_{j,it}|$, and it
disappears in case $X_{j,it}$'s are uniformly bounded. Theorem \ref{Thm1}(ii)-(iii) report the convergence rates for the estimators
of the factors and factor loadings of $\Theta _{j}^{0}$, which are inherited from those in Theorem \ref{Thm1}(i).  To derive these results, we establish the symmetrization inequality and contraction principle for the sequential symmetrization developed by \cite{rakhlin2015sequential}. See Lemmas \ref{Lem:symmetrization} and \ref{Lem:contraction} in the online supplement for more detail.



\subsection{Second Stage Estimator}

To study the asymptotic properties of the second-stage estimators, we add some notation. Define 
\begin{equation*}
\Phi_{i}=\frac{1}{T}\sum_{t=1}^{T}\Phi_{it}^{0}\Phi_{it}^{0\prime}\quad \text{and }\Psi_{t}=\frac{1}{N_{2}}\sum_{i\in I_{2}}\Psi_{it}^{0}\Psi_{it}^{0\prime},
\end{equation*}%
where $\Phi_{it}^{0}=(v_{t,0}^{0\prime},v_{t,1}^{0\prime}X_{1,it},\cdots,v_{t,p}^{0\prime}X_{p,it})^{\prime}$ and $\Psi_{it}^{0}=(u_{i,0}^{0\prime},u_{i,1}^{0\prime}X_{1,it},\cdots,u_{i,p}^{0\prime}X_{p,it})^{\prime}.$ Let $K=\sum_{j=0}^{p}K_{j}.$ Note that $\Phi_{i}$ and $\Psi_{t}$ are $K\times K$ matrices. We add the following two assumptions.

\begin{ass}
\label{ass:5} There exist constants $C_{\phi }$ and $c_{\phi }$ such that
a.s. 
\begin{align*}
\infty & >C_{\psi }\geq \limsup_{T}\max_{t\in[T]}\lambda_{\max}(\Psi_{t})\geq \liminf_{T}\min_{t\in[T]}\lambda_{\min }(\Psi_{t})\geq c_{\psi }>0, \\
\infty & >C_{\phi }\geq \limsup_{N}\max\limits_{i\in I_{2}}\lambda_{\max}(\Phi_{i})\geq \liminf_{N}\min_{i\in I_{2}}\lambda_{\min}(\Phi_{i})\geq c_{\phi }>0.
\end{align*}
\end{ass}

Assumption \ref{ass:5} is similar to Assumption 8 in \cite{ma2020detecting}.
To introduce Theorem \ref{Thm2}, we define 
\begin{align*}
& \dot{\varpi}_{it}=\left(\dot{v}_{t,0}^{(1)\prime},\dot{v}%
_{t,1}^{(1)\prime}X_{1,it},\cdots ,\dot{v}_{t,p}^{(1)\prime}X_{p,it}\right)
^{\prime}, \\
& \varpi_{it}^{0}=\left(\left(O_{0}^{(1)}v_{t,0}^{0}\right)
^{\prime},\left(O_{1}^{(1)}v_{t,1}^{0}\right) ^{\prime}X_{1,it},\cdots
,\left(O_{p}^{(1)}v_{t,p}^{0}\right) ^{\prime}X_{p,it}\right) ^{\prime}, \\
& u_{i}^{0}=\left(u_{i,0}^{0\prime},\cdots ,u_{i,p}^{0\prime}\right)
^{\prime},\quad \dot{\Delta}_{t,j}=O_{j}^{(1)\prime}\dot{v}%
_{t,j}^{(1)}-v_{t,j}^{0},\quad \dot{\Delta}_{t,v}=\left(\dot{\Delta}%
_{t,0}^{\prime},\cdots ,\dot{\Delta}_{t,p}^{\prime}\right) ^{\prime}, \\
& \dot{\Delta}_{i,j}=O_{j}^{(1)\prime}\dot{u}_{i,j}^{(1)}-u_{i,j}^{0},\quad 
\dot{\Delta}_{i,u}=\left(\dot{\Delta}_{i,0}^{\prime},\cdots ,\dot{\Delta}%
_{i,p}^{\prime}\right) ^{\prime}, \\
& D_{i}^{I}=\frac{1}{T}\sum_{t=1}^{T}\mathfrak{f}_{it}(0)\varpi_{it}^{0}%
\varpi_{it}^{0\prime},\quad D_{i}^{II}=\frac{1}{T}\sum_{t=1}^{T}\left[ \tau -%
\mathbf{1}\left\{ \epsilon_{it}\leq 0\right\} \right] \varpi_{it}^{0}, \\
& \mathbb{J}_{i}\left(\left\{ \dot{\Delta}_{t,v}\right\}_{t\in[T]}\right) =%
\frac{1}{T}\sum_{t=1}^{T}\left[ \mathbf{1}\left\{ \epsilon_{it}\leq
0\right\} -\mathbf{1}\left\{ \epsilon_{it}\leq \dot{\Delta}%
_{t,v}^{\prime}\Psi_{it}^{0}\right\} \right] \varpi_{it}^{0}.
\end{align*}%
Theorem \ref{Thm2} below gives the uniform convergence rate and linear
expansion of the factor loading estimators from second stage estimation.

\begin{theorem}
\label{Thm2} Suppose Assumptions \ref{ass:1}-\ref{ass:5} hold. Then for each 
$j\in \left\{ 0,\cdots ,p\right\} $, we have

\begin{itemize}
\item[(i)] $\operatornamewithlimits{\max}\limits_{i\in I_{2}\cup I_{3}}
\left\Vert\dot{u}_{i,j}^{(1)}-O_{j}^{(1)}u_{i,j}^{0}\right\Vert_{2}=O_{p}%
\left(\sqrt{\frac{\log(N \vee T)}{N\wedge T}}\xi_{N}^{2 }\right)$,

\item[(ii)] $\operatornamewithlimits{\max}\limits_{t\in[T]} \left\Vert\dot{v}%
_{t,j}^{(1)}-O_{j}^{(1)}v_{t,j}^{0}\right\Vert_{2}=O_{p}\left(\sqrt{\frac{%
\log(N \vee T)}{N\wedge T}}\xi_{N}^{2 }\right)$,

\item[(iii)] $\dot{\Delta}_{i,u}=\left[ D_{i}^{I}\right] ^{-1}\left[
D_{i}^{II}+\mathbb{J}_{i}\left( \left\{ \dot{\Delta}_{t,v}\right\} _{t\in
\lbrack T]}\right) \right] +o_{p}\left( \left( N\vee T\right) ^{-1/2}\right) 
$ uniformly over $i\in I_{3}$.
\end{itemize}
\end{theorem}

\textbf{Remark 2.} Theorem \ref{Thm2}(i) reports the uniform convergence
rate for the factor loading estimators of $\Theta _{j}^{0}$ for $i\in
I_{2}\cup I_{3};$ Theorem \ref{Thm2}(ii) reports the uniform convergence
rate for the factor estimators of $\Theta _{j}^{0}$ for $t\in \left[ T\right]
$; Theorem \ref{Thm2}(iii) reports the linear expansion for the factor
loading estimators of $\Theta _{j}^{0}$ for $i\in I_{3}$. However, the $\mathbb{J}_{i}\left( \left\{ \dot{\Delta}_{t,v}\right\} _{t\in
\lbrack T]}\right)$ term is not mean-zero and represents the bias induced by the first stage NNR. In the third stage below, we aim to remove such a bias from the linear expansion. 


\subsection{Third Stage Estimator}

In the debiasing stage, we first apply PCA to all independent variables $X_{j,it}$, and then run the row- and column-wise QRs to obtain the final
estimators. Below we give Assumptions \ref{ass:7}-\ref{ass:9} for the PCA
procedure and establish the asymptotic linear expansions of PCA estimates in
the online supplement. Theorem \ref{Thm3} below gives the asymptotic
distribution of our final factor and factor loading estimates.

\begin{ass}
\label{ass:7} For all $j\in[p]$, there exists a constant $M>0$ such that
\begin{itemize}
\item[(i)] $\mathbb{E}\left(e_{j,it}|\mu_{j,it} \right)=0$,

\item[(ii)] $\mathbb{E}\left[\frac{1}{\sqrt{N}}\sum_{i=1}^{N}\left[%
e_{j,it}e_{j,is}-\mathbb{E}(e_{j,it}e_{j,is})\right]\right]^{2}\leq M, $

\item[(iii)] for all $i\in[N]$, $\frac{1}{T}\sum_{t=1}^{T}\sum_{s=1}^{T}%
\left\vert\mathbb{E}(e_{j,it}e_{j,is})\right\vert\leq M$,

\item[(iv)] $\max_{t\in[T]}\frac{1}{N\sqrt{T}}\left\Vert
e_{j,t}^{\prime}E_{j}\right\Vert_{2}=O_{p}\left(\frac{\log N\vee T}{N\wedge T%
}\right)$, and for $\max_{i\in [N]}\frac{1}{T\sqrt{N}}\left\Vert
e_{j,i}^{\prime}E_{j}^{\prime}\right\Vert_{2}=O_{p}\left(\frac{\log N\vee T}{%
N\wedge T}\right)$, where $e_{j,i}=\left(e_{j,i1},...,e_{j,iT}\right)
^{\prime}$, $e_{j,t}=\left(e_{j,1t},...,e_{j,Nt}\right) ^{\prime}$, and $%
E_{j}=\{e_{j,it}\}_{i \in [N],t \in [T]}$.
\end{itemize}
\end{ass}

\begin{ass}
\label{ass:8} For all $j\in[p]$, 

\begin{itemize}
\item[(i)] recall that $L_{j}^{0}=\left(l_{j,1}^{0},\cdots,l_{j,N}^{0}%
\right) ^{\prime}$ and $W_{j}^{0}=\left(w_{j,1}^{0},\cdots,w_{j,T}^{0}%
\right) ^{\prime}$. $\lim_{N\rightarrow \infty }\frac{L_{j}^{0%
\prime}L_{j}^{0}}{N}=\Sigma_{L_{j}}>0$ and $\lim_{T\rightarrow\infty }\frac{%
W_{j}^{0\prime}W_{j}^{0}}{T}=\Sigma_{W_{j}}>0$,

\item[(ii)] the $r_{j}$ eigenvalues of $\Sigma_{L_{j}}\Sigma_{W_{j}}$ are
distinct.
\end{itemize}
\end{ass}

\begin{ass}
\label{ass:9} For all $j\in[p]$, there exists a constant $M>0$ such that

\begin{itemize}
\item[(i)] $\max_{t\in[T]}\mathbb{E}\left\Vert \frac{1}{\sqrt{N}}%
\sum_{i=1}^{N}l_{j,i}^{0}e_{j,it}\right\Vert_{2}^{2}\leq M$ and $\max_{t\in[T%
]}\frac{1}{NT}e_{j,t}^{\prime}E_{j}^{\prime}L_{j}^{0}=O_{p}\left(\frac{\log
\left(N\vee T\right) }{N\wedge T}\right) $,

\item[(ii)] $\max_{i\in [N]}\mathbb{E}\left\Vert \frac{1}{\sqrt{T}}%
\sum_{t=1}^{T}w_{j,t}^{0}e_{j,it}\right\Vert_{2}^{2}\leq M$ and $\max_{i\in[N%
]}\frac{1}{NT}e_{j,t}^{\prime}E_{j}^{\prime}W_{j}^{0}=O_{p}\left(\frac{\log
\left(N\vee T\right) }{N\wedge T}\right) $.
\end{itemize}
\end{ass}

\begin{ass}
\label{ass:10} For $\forall j\in[p]$,

\begin{itemize}
\item[(i)] $\mathbb{E}\left[f_{it}(0)e_{j,it} \bigg|\mathscr{D}\right]=0$,

\item[(ii)] for each $i\in[N]$ and $j\in[p]$, $\left%
\{f_{it}(0),f_{it}(0)e_{j,it}\right\}$ is stationary strong mixing across $t$
conditional on $\mathscr{D}$.
\end{itemize}
\end{ass}

Assumptions \ref{ass:7}-\ref{ass:9} are stronger than those in \cite%
{bai2020simpler} because we strengthen their Assumptions A1(c) and A3 to hold uniformly. Assumption \ref{ass:10} imposes some moment
and mixing conditions. Even though $f_{it}(\cdot )$ (the PDF of $\epsilon _{it}$ given $\mathscr{D}%
_{e}$) is a function of $%
\left\{ e_{j,it}\right\} _{j\in \lbrack p],i\in \lbrack N],t\in \lbrack T]}$%
, we can show that Assumption \ref{ass:10} holds under some reasonable
conditions. For example, we consider the location scale model: 
\begin{equation*}
Y_{it}=\beta _{0,it}+\sum_{j\in \lbrack p]}X_{j,it}\beta _{j,it}+\left(
\gamma _{0,it}+\sum_{j\in \lbrack p]}X_{j,it}\gamma _{j,it}\right)
u_{it},\quad \text{with}\quad X_{j,it}=l_{j,i}^{0\prime
}w_{j,t}^{0}+e_{j,it},
\end{equation*}%
where $u_{it}$ is independent of $\{w_{j,t}^{0},e_{j,it}\}_{j\in \lbrack
p],t\in \lbrack T]}$ and $l_{j,i}^{0}$ and $\beta _{j,it}$ are nonrandom. In
this case, $\Theta _{j,it}^{0}=\beta _{j,it}+\gamma _{j,it}\mathscr{Q}_{\tau
}(u_{it})$, $\epsilon _{it}=\left( \gamma _{0,it}+\sum_{j\in \lbrack
p]}X_{j,it}\gamma _{j,it}\right) \left[ u_{it}-\mathscr{Q}_{\tau }(u_{it})%
\right] $, where $\mathscr{Q}_{\tau }(u_{it})$ is the $\tau $-quantile of $%
u_{it}$. It is clear that $f_{it}(\cdot )$ is the function of $\left\{
e_{j,it}\right\} _{j\in \lbrack p],i\in \lbrack N],t\in \lbrack T]}$ and all
factors. However, if $u_{it}$ is independent of sequence $\left\{
e_{j,it}\right\} _{j\in \lbrack p],t\in \lbrack T]}$, we observe that $%
f_{it}(0)$ is the PDF of $u_{it}-\mathscr{Q}_{\tau }(u_{it})$ evaluated at
zero point, which is independent of $\left\{ e_{j,it}\right\} _{j\in \lbrack
p],t\in \lbrack T]}$. Therefore, Assumption \ref{ass:10}%
(i) holds under mild conditions that $u_{it}$ is independent of the sequence $%
\left\{ e_{it}\right\} _{t\in \lbrack T]}$ and $\mathbb{E}\left( e_{it}\big|%
\mathscr{D}\right) =0$.

Define 
\begin{align*}
& \hat{V}_{u_{j},i}=\frac{1}{T}\sum_{t=1}^{T}\mathbb{E}\left[%
f_{it}(0)e_{j,it}^{2}\big|\mathscr{D}\right] v_{t,j}^{0}v_{t,j}^{0\prime},%
\quad V_{u_{j},i}=\mathbb{E}\left(\hat{V}_{u_{j}}\right) \\
& \Omega_{u_{j},i}=Var\left[ \frac{1}{\sqrt{T}}%
\sum_{t=1}^{T}e_{j,it}v_{t,j}^{0}(\tau -\mathbf{1}\left\{ \epsilon_{it}\leq
0\right\} )\right] , \\
& \hat{V}_{v_{j},t}^{(3)}=\frac{1}{N_{3}}\sum_{i\in
I_{3}}f_{it}(0)e_{j,it}^{2}u_{i,j}^{0}u_{i,j}^{0\prime},\quad V_{v_{j}}=%
\mathbb{E}\left(\hat{V}_{v_{j},t}^{(3)}\right),
\quad\Omega_{v_{j}}^{(3)}=\tau \left(1-\tau \right) \frac{1}{N_{3}}%
\sum_{i\in I_{3}}\mathbb{E}\left(e_{j,it}^{2}u_{i,j}^{0}u_{i,j}^{0\prime}%
\right) .
\end{align*}%
Let $\Sigma_{u_{j},i}=O_{j}^{(1)}V_{u_{j},i}^{-1}%
\Omega_{u_{j},i}V_{u_{j},i}^{-1}O_{j}^{(1)\prime}$, $%
\Sigma_{v_{j}}^{(3)}=O_{j}^{(1)}\left(V_{v_{j}}^{(3)}\right)^{-1}%
\Omega_{v_{j}}\left(V_{v_{j}}^{(3)}\right)^{-1}O_{j}^{(1)\prime}$, $%
b_{j,it}^{0}=e_{j,it}v_{t,j}^{0}(\tau -\mathbf{1}\left\{ \epsilon_{it}\leq
0\right\} )$ and $\xi_{j,it}^{0}=e_{j,it}u_{i,j}^{0}\left(\tau -\mathbf{1}%
\left\{ \epsilon_{it}\leq 0\right\} \right) $. The following theorem establishes
the asymptotic properties of the third-stage estimators.

\begin{theorem}
\label{Thm3} Suppose that Assumptions \ref{ass:1}-\ref{ass:10} hold. Suppose
that Assumption \ref{ass:15} in Appendix B.3 of the online supplement hold.
Let $O_{u,j}^{(1)}$ be the bounded matrix defined in the appendix that is
related to rotation matrix $O_{j}^{(1)}$. Then we have that $\forall j\in
\lbrack p]$,

\begin{itemize}
\item[(i)] $\hat{u}_{i,j}^{(3,1)}-O_{u,j}^{(1)}u_{i,j}^{0}=O_{j}^{(1)}\hat{V}%
_{u_{j},i}^{-1}\frac{1}{T}\sum_{t=1}^{T}b_{j,it}^{0}+\mathcal{R}_{i,u}^{j}$
and $\sqrt{T}\left(\hat{u}_{i,j}^{(3,1)}-O_{u,j}^{(1)}u_{i,j}^{0}\right)
\rightsquigarrow \mathcal{N}\left(0,\Sigma_{u_{j},i}\right) $ $\forall i\in
I_{3}$,

\item[(ii)] $\hat{v}_{t,j}^{(3,1)}-\left( O_{u,j}^{(1)\prime }\right)
^{-1}v_{t,j}^{0}=O_{j}^{(1)}\left( \hat{V}_{v_{j},t}^{(3)}\right) ^{-1}\frac{%
1}{N_{3}}\sum_{i\in I_{3}}\xi _{j,it}^{0}+\mathcal{R}_{t,v}^{j}$ and $\sqrt{%
N_{3}}\left( \hat{v}_{t,j}^{(3,1)}-O_{v,j}^{(1)}v_{t,j}^{0}\right)
\rightsquigarrow \mathcal{N}\left( 0,\Sigma _{v_{j}}^{(3)}\right) $ $\forall
t\in \lbrack T]$,\newline
where $\max_{i\in I_{3}}\left\vert \mathcal{R}_{i,u}^{j}\right\vert
=o_{p}\left( \left( N\vee T\right) ^{-1/2}\right) $, and $\max_{t\in \lbrack
T]}\left\vert \mathcal{R}_{t,v}^{j}\right\vert =o_{p}\left( \left( N\vee
T\right) ^{-\frac{1}{2}}\right) $.
\end{itemize}
\end{theorem}

\textbf{Remark 3.} Theorem \ref{Thm3} reports the linear expansions for the
factor and factor loading estimators for each slope matrix obtained in Step
3. Compared with \cite{chernozhukov2019inference}, Theorem \ref{Thm3}
obtains the uniform convergence rate rather than the point-wise result for
the reminder terms $\mathcal{R}_{i,u}^{j}$ and $\mathcal{R}_{t,v}^{j}.$ In
addition, since the regressors in the debiasing step are obtained from Step
2 instead of Step 1, we don't have independence between the regressors and
error terms, which makes the proof more complex than that in \cite%
{chernozhukov2019inference}. See the proof in the appendix on how to handle
the dependence. Assumption \ref{ass:15} in the online supplement is a regularity condition on the density of $\epsilon_{it}$. 

Following Theorem \ref{Thm3} and estimators defined in Table \ref%
{tab:estimator symbol}, we have that $\forall j\in \left[ p\right] $, $%
\forall i\in \lbrack N]$ and $\forall t\in \left[ T\right] $, 
\begin{align*}
& \hat{u}_{i,j}^{(a,b)}-O_{u,j}^{(b)}u_{i,j}^{0}=O_{j}^{(b)}\hat{V}%
_{u_{j},i}^{-1}\frac{1}{T}\sum_{t=1}^{T}v_{t,j}^{0}b_{j,it}^{0}+\mathcal{R}%
_{i,u}^{j}, \\
& \hat{v}_{t,j}^{(a,b)}-\left( O_{u,j}^{(b)\prime }\right)
^{-1}v_{t,j}^{0}=O_{j}^{(b)}\left( \hat{V}_{v_{j},t}^{(a)}\right) ^{-1}\frac{%
1}{N_{a}}\sum_{i\in I_{a}}u_{i,j}^{0}\xi _{j,it}^{0}+\mathcal{R}_{t,v}^{j},
\end{align*}%
where $\hat{V}_{v_{j},t}^{(a)}=\frac{1}{N_{a}}\sum_{i\in
I_{a}}f_{it}(0)e_{j,it}^{2}u_{i,j}^{0}u_{i,j}^{0\prime }$, $a\in \lbrack 3]$
and $b\in \lbrack 3]\setminus \{a\}$.

Given the above estimates for the factors and factor loadings, we can
estimate $\Theta _{j,it}^{0}$ by 
\begin{equation*}
\hat{\Theta}_{j,it}=\frac{1}{2}\sum_{a\in \lbrack 3]}\sum_{b\in \lbrack
3]\setminus \{a\}}\left\{ \hat{u}_{i,j}^{(a,b)\prime }\hat{v}%
_{t,j}^{(a,b)}\right\} \mathbf{1}_{ia}
\end{equation*}%
where $\mathbf{1}_{ia}=\mathbf{1}\left\{ i\in I_{a}\right\} $ for $i\in
\lbrack N]$. Let $\Xi _{j,it}^{0}=\frac{1}{T}v_{t,j}^{0\prime }\Sigma
_{u_{j},i}v_{t,j}^{0}+\sum_{a=1}^{3}\frac{1}{N_{a}}\mathbf{1}%
_{ia}u_{i,j}^{0\prime }\Sigma _{v_{j}}^{a}u_{i,j}^{0}$. The following
proposition studies the asymptotic properties of $\hat{\Theta}_{j,it}$.

\begin{proposition}
\label{Pro4} Under Assumptions \ref{ass:1}-\ref{ass:10} and Assumption \ref%
{ass:15}, $\forall j\in \lbrack p]$ we have

\begin{itemize}
\item[(i)] $\hat{\Theta}_{j,it}-\Theta_{j,it}^{0}=\sum_{a=1}^{3}u_{i,j}^{0%
\prime}\left(\hat{V}_{v_{j},t}^{(a)}\right) ^{-1}\frac{1}{N_{a}}%
\sum_{i^{*}\in I_{a}}\xi_{j,i^{*}t}\mathbf{1}_{i^{*}a}+v_{t,j}^{0\prime}\hat{%
V}_{u_{j}}^{-1}\frac{1}{T}\sum_{t^{*}=1}^{T}b_{j,it^{*}}^{0}+\mathcal{R}%
_{it}^{j}$, where $\max_{i\in I_{3},t\in[T]}\left\vert \mathcal{R}_{it}^{j}\right\vert
=o_{p}\left(\left(N\vee T\right) ^{-1/2}\right)$, 

\item[(ii)] $\max_{i\in[N],t\in[T]}\left\vert \hat{\Theta}%
_{j,it}-\Theta_{j,it}^{0}\right\vert =O_{p}\left(\sqrt{\frac{\log N\vee T}{%
N\wedge T}}\right)$,

\item[(iii)] $\left(\Xi_{j,it}^{0}\right) ^{-1/2}\left(\hat{\Theta}%
_{j,it}-\Theta_{j,it}^{0}\right) \rightsquigarrow \mathcal{N}%
\left(0,1\right) $.
\end{itemize}

\end{proposition}

\textbf{Remark 4.} Proposition \ref{Pro4} establishes the distribution
theory for the slope estimators. Recall that we remove the principle
component from the independent variables $X_{j,it}$ which is the key point
in the debiasing step and why we don't have the distribution theory
result for the intercept estimates $\hat{\Theta}_{0,it}$ in the current
framework. However, once we have the distribution theory for the slope
estimates, we can follow \cite{chen2021quantile} and obtain a new estimator for  $\Theta_{0,it}^0$ 
from the smoothed quantile regression and establish its distribution theory. We leave this for the further research.

To make inference for $u_{i,j}^{0},$ $v_{t,j}^{0},$ and $\Theta _{j,it}^{0},$
one needs to estimate their asymptotic variances $\Sigma _{u_{j},i},$ $\Sigma
_{v_{j}}$ and $\Xi _{j,it}^{0}$ consistently. Let $k(\cdot )$ be a PDF-type
kernel function and $K(\cdot )$ be its survival function such that $\int
k(u)du=1$ and $K(u):=\int_{u}^{\infty }k(v)dv$. Let $h_{N}$ be the bandwidth
such that $h_{N}\rightarrow 0$ with $N\rightarrow \infty $. Define $%
K_{h_{N}}(\cdot )=K(\frac{\cdot }{h_{N}})$, $k_{h_{N}}(\cdot )=\frac{1}{h_{N}%
}k(\frac{\cdot }{h_{N}})$. Let $\hat{\epsilon}_{it}=Y_{it}-\hat{\Theta}%
_{0,it}-\sum_{j\in \lbrack p]}X_{j,it}\hat{\Theta}_{j,it}$, $\hat{\mathrm{v}}%
_{t,s,j}=\frac{1}{6}\sum_{a\in \lbrack 3]}\sum_{b\in \lbrack 3]\setminus
\{a\}}\hat{v}_{t,j}^{(a,b)}\hat{v}_{s,j}^{(a,b)\prime }$, and $\hat{\mathrm{u%
}}_{i,i,j}=\frac{1}{2}\sum_{a\in \lbrack 3]}\sum_{b\in \lbrack 3]\setminus
\{a\}}\hat{u}_{i,j}^{(a,b)}\hat{u}_{i,j}^{(a,b)\prime }\mathbf{1}_{ia}$.
Define 
\begin{align*}
\hat{\mathbb{V}}_{u_{j}}& =\frac{1}{NT}\sum_{i\in \lbrack N]}\sum_{t\in
\lbrack T]}k_{h_{N}}(\hat{\epsilon}_{it})\hat{e}_{j,it}^{2}\hat{\mathrm{v}}%
_{t,t,j},\quad \hat{\mathbb{V}}_{v_{j}}=\frac{1}{NT}\sum_{i\in \lbrack
N]}\sum_{t\in \lbrack T]}k_{h_{N}}(\hat{\epsilon}_{it})\hat{e}_{j,it}^{2}%
\hat{\mathrm{u}}_{i,i,j}, \\
\hat{\Omega}_{u_{j}}& =\frac{1}{NT}\sum_{i\in \lbrack N]}\left\{ \sum_{t\in
\lbrack T]}\tau (1-\tau )\hat{e}_{j,it}^{2}\hat{\mathrm{v}}%
_{t,t,j}+\sum_{t=1}^{T-T_{1}}\sum_{s=t+1}^{t+T_{1}}S_{j,its}+%
\sum_{t=1+T_{1}}^{T}\sum_{s=t-T_{1}}^{t-1}S_{j,its}\right\} , \\
\hat{\Omega}_{v_{j}}& =\frac{\tau (1-\tau )}{NT}\sum_{i\in \lbrack
N]}\sum_{t\in \lbrack T]}\hat{e}_{j,it}^{2}\hat{\mathrm{u}}_{i,i,j},\quad 
\hat{\Sigma}_{u_{j}}=\hat{\mathbb{V}}_{u_{j}}^{-1}\hat{\Omega}_{u_{j}}\hat{%
\mathbb{V}}_{u_{j}}^{-1},\quad \hat{\Sigma}_{v_{j}}=\hat{\mathbb{V}}%
_{v_{j}}^{-1}\hat{\Omega}_{v_{j}}\hat{\mathbb{V}}_{v_{j}}^{-1},
\end{align*}%
where $S_{j,its}=\hat{e}_{j,it}\hat{e}_{j,is}\hat{\mathrm{v}}_{t,s,j}\left[
\tau -K\left( \frac{\hat{\epsilon}_{it}}{h_{N}}\right) \right] \left[ \tau
-K\left( \frac{\hat{\epsilon}_{is}}{h_{N}}\right) \right] .$ We further
define 
\begin{equation*}
\hat{\Xi}_{j,it}=\frac{1}{2}\sum_{a\in \lbrack 3]}\sum_{b\in \lbrack
3]\setminus \{a\}}\left( \frac{1}{T}\hat{v}_{t,j}^{(a,b)\prime }\hat{\Sigma}%
_{u_{j}}\hat{v}_{t,j}^{(a,b)}+\frac{1}{N_{a}}\mathbf{1}_{ia}\hat{u}%
_{i,j}^{(a,b)\prime }\hat{\Sigma}_{v_{j}}\hat{u}_{i,j}^{(a,b)}\right) .
\end{equation*}%
Let $F_{i,ts}(\cdot ,\cdot )$ and $f_{i,ts}(\cdot ,\cdot )$ denote the joint
CDF and PDF of $(\epsilon _{it},\epsilon _{is})$ given $\mathscr{D}_{e},$
respectively. To justify the consistency of the variance estimators, we add
the following assumption.

\begin{ass}
\label{ass:11}

\begin{itemize}
\item[(i)] $\int_{-\infty}^{+\infty}k(u)du=1$, $\int_{-\infty}^{+%
\infty}k(u)u^{j}du=0$ for $j\in\{1,\cdots,m-1\}$ and $\int_{-\infty}^{+%
\infty}k(u)u^{m}du\ne0$ for $m\geq 1$.

\item[(ii)] $h_{N}\rightarrow 0$ and $\left(\frac{\log (N\vee T)}{N\wedge T}%
\right) ^{1/4}\frac{\xi_{N}^{2}}{h_{N}}\rightarrow 0$.

\item[(iii)] $T_{1}\rightarrow \infty $ and $\sqrt{\frac{\log (N\vee T)}{%
N\wedge T}}\frac{\xi_{N}^{2 }T_{1}}{h_{N}^{2}}\rightarrow 0$.

\item[(iv)] $f_{it}(c)$ is $m$ times continuously differentiable with
respect to $c$ and $f_{i,ts}(c_{1},c_{2})$ is $m$ times continuously
differentiable with respect to $(c_{1},c_{2})$.

\item[(v)] $\forall i\in \lbrack N]$, $V_{u_{j},i}=V_{u_{j}}$ and $\Omega
_{u_{j},i}=\Omega _{u_{j}}$.

\item[(vi)] $\forall a\in \lbrack 3]$, $V_{v_{j}}^{(a)}=\frac{1}{N}%
\operatornamewithlimits{\sum}\limits_{i\in \lbrack N]}\mathbb{E}\left[
f_{it}(0)e_{j,it}^{2}u_{i,j}^{0}u_{i,j}^{0\prime }\right] +o_{p}(1)$ and $%
\Omega _{v_{j}}^{(a)}=\frac{\tau \left( 1-\tau \right) }{N}%
\operatornamewithlimits{\sum}\limits_{i\in \lbrack N]}\mathbb{E}\left(
e_{j,it}^{2}u_{i,j}^{0}u_{i,j}^{0\prime }\right) +o_{p}(1)$.
\end{itemize}
\end{ass}

Assumption \ref{ass:11}(i)-(iv) are standard for consistent estimation of
the asymptotic variance matrix; see, e.g., \cite{chen2019two} and \cite%
{galvao2016smoothed}. Assumption \ref{ass:11}(v) imposes the homogeneity
moment condition across individuals, and Assumption \ref{ass:11}(vi) assumes
the moments calculated from subsamples are close to those from the full
sample given the random splitting. Under Assumption \ref{ass:11}, following
the idea of \cite{chen2019two}, we establish in Lemma \ref{Lem:covhat} of
the online supplement the consistency of $\hat{\Sigma}_{u_{j}}$ and $\hat{%
\Sigma}_{v_{j}}$. Similar conclusions hold for the other estimates.


\section{Specification Tests}

In this section, we consider two specification tests under different rank
conditions.

\subsection{Testing for Homogeneity across Individuals or Time}

When $K_{j}=1$ for some $j\in [p]$, it is interesting to test whether the
matrix $\Theta_{j}^{0}$ is homogeneous across individuals (i.e., row-wise) or
across time (i.e., column-wise). For these two cases, we can write factors and
factor loadings as
\begin{equation*}
u_{i,j}^{0}=u_{j}+c_{i,j}^{u}\quad \text{and}\quad v_{t,j}^{0}=v_{j}+c_{t,j}^{v}, \quad \text{respectively}, 
\end{equation*}%
where $u_{j}=\frac{1}{N}\sum_{i=1}^{N}u_{i,j}^{0}$ and $v_{j}=\frac{1}{T}%
\sum_{t=1}^{T}v_{t,j}^{0}.$ For the homogeneity across individuals, the null
and alternative hypotheses can be written as 
\begin{equation}
H_{0}^{I}:c_{i,j}^{u}=0\quad \forall i\in [N]\quad v.s.\quad
H_{1}^{I}:c_{i,j}^{u}\neq 0\text{ for some }i\in [N].  \label{Hypo1}
\end{equation}%
Similarly, for the homogeneity across time, the null and alternative
hypotheses can be written as 
\begin{equation}
H_{0}^{II}:c_{t,j}^{v}=0\quad \forall t\in[T]\quad v.s.\quad
H_{1}^{II}:c_{t,j}^{v}\neq 0\text{ for some }t\in[T].  \label{Hypo2}
\end{equation}%
Note that we aim to test the two null hypotheses separately. That is, we can
test for homogeneous slope across individuals while allowing for
heterogeneous slopes across time and vice versa. This is different from the
majority of the literature which either tests for slope homogeneity across
individuals while assuming the slopes are homogeneous across time or tests
for structural breaks across time while assuming the slopes are homogeneous
across individuals.


We first consider testing $H_{0}^{I}$. Following the lead of \cite%
{castagnetti2015inference}, we define\footnote{%
Alternatively, we can also define $S_{u_{j}}^{o}=\max \left(
S_{u_{j}}^{(3,2)},S_{u_{j}}^{(2,1)},S_{u_{j}}^{(1,3)}\right) .$ It is easy
to show that this statistic shares the same asymptotic null distribution as $%
S_{u_{j}}.$ But due to the unknown dependence structure between the two, we
cannot take the maximum or the other continuous function of $S_{u_{j}}$ and $%
S_{u_{j}}^{o}$ as a new test statistic.} 
\begin{align}
& S_{u_{j}}^{(a,b)}=\max_{i\in I_{a}}T(\hat{u}_{i,j}^{(a,b)}-\hat{\bar{u}}%
_{j}^{(a,b)})^{\prime }\hat{\Sigma}_{u_{j}}^{-1}(\hat{u}_{i,j}^{(a,b)}-\hat{%
\bar{u}}_{j}^{(a,b)})\quad \text{and}  \notag \\
& S_{u_{j}}=\max \left(
S_{u_{j}}^{(3,1)},S_{u_{j}}^{(2,3)},S_{u_{j}}^{(1,2)}\right) ,
\label{statistic for u}
\end{align}%
where $\hat{\bar{u}}_{j}^{(a,b)}=\frac{1}{N_{a}}\sum_{i\in I_{a}}\hat{u}%
_{i,j}^{(a,b)}$. Similarly, to test for $H_{0}^{II}$, we construct 
\begin{equation*}
S_{v_{j}}=\max \left( \tilde{S}_{v_{j}}^{(3,1)},\tilde{S}_{v_{j}}^{(2,3)},%
\tilde{S}_{v_{j}}^{(1,3)}\right) ,
\end{equation*}%
where 
\begin{equation*}
S_{v_{j}}^{(a,b)}=\max_{t\in \lbrack T]}N(\hat{v}_{t,j}^{(a,b)}-\hat{\bar{v}}%
_{j}^{(a,b)})^{\prime }\hat{\Sigma}_{v_{j}}^{-1}(\hat{v}_{t,j}^{(a,b)}-\hat{%
\bar{v}}_{j}^{(a,b)}),\quad  \tilde{S}_{v_{j}}^{(a,b)}=\frac{1%
}{2}S_{v_{j}}^{(a,b)}-\mathsf{b}(T),
\end{equation*}
$\hat{\bar{v}}_{j}^{(a,b)}=\frac{1}{T}\sum_{t=1}^{T}\hat{v}%
_{t,j}^{(a,b)},$ and $\mathsf{b}(n)=\log n-\frac{1}{2}\log \log n-\log
\Gamma (\frac{1}{2})$ for $n\in \left\{ N,T,NT\right\} .$


To proceed, we introduce some notation. Recall that $%
b_{j,it}^{0}=e_{j,it}v_{t,j}^{0}(\tau -\mathbf{1}\left\{ \epsilon_{it}\leq
0\right\} )$ and $\xi_{j,it}^{0}=e_{j,it}u_{i,j}^{0}\left(\tau -\mathbf{1}%
\left\{ \epsilon_{it}\leq 0\right\} \right) $. Define 
\begin{equation*}
\mathfrak{b}_{j,it}^{(1)}=\hat{V}_{u_{j}}^{-1}b_{j,it}^{0},\quad \mathfrak{b}%
_{j,it}^{(2)}=\left(\hat{V}_{v_{j},t}^{(3)}\right) ^{-1}\xi_{j,it}^{0},\text{
}\mathfrak{b}_{j,it}^{(3)}=\left(\hat{V}_{v_{j},t}^{(2)}\right)^{-1}%
\xi_{j,it}^{0},\text{ and }\mathfrak{b}_{j,it}^{(4)}=\left(\hat{V}%
_{v_{j},t}^{(1)}\right) ^{-1}\xi_{j,it}^{0}.
\end{equation*}%
Let $\mathfrak{B}_{j,t}^{(\ell )}=\left(\mathfrak{b}_{j,1t}^{(\ell
)\prime},\cdots ,\mathfrak{b}_{j,Nt}^{(\ell )\prime }\right) ^{\prime}$ for $%
\ell \in \left[ 4\right] .$ Define 
\begin{align*}
& \Sigma_{\mathfrak{B},j}^{(1)}=\frac{1}{T}\sum_{t=1}^{T}\sum_{s=1}^{T}%
\mathbb{E}\left(\mathfrak{B}_{j,t}^{(1)}\mathfrak{B}_{j,s}^{(1)\prime}%
\right) ,\quad \Sigma_{\mathfrak{B},j}^{(2)}=\frac{1}{N_{3}}\sum_{i\in I_{3}}%
\mathbb{E}\left(\mathfrak{B}_{j,i}^{(2)}\mathfrak{B}_{j,i}^{(2)\prime}%
\right) , \\
& \Sigma_{\mathfrak{B},j}^{(3)}=\frac{1}{N_{2}}\sum_{i\in I_{2}}\mathbb{E}%
\left(\mathfrak{B}_{j,i}^{(3)}\mathfrak{B}_{j,i}^{(3)\prime }\right) ,\quad 
\text{and }\Sigma_{\mathfrak{B},j}^{(4)}=\frac{1}{N_{1}}\sum_{i\in I_{1}}%
\mathbb{E}\left(\mathfrak{B}_{j,i}^{(4)}\mathfrak{B}_{j,i}^{(4)\prime}%
\right) .
\end{align*}%
We add the following two assumptions. \bigskip

\begin{ass}
\label{ass:12} $\forall j\in \lbrack p]$, we assume 
\begin{equation*}
\bar{\lambda}\geq \lambda _{\max }\left( \Sigma _{u_{j}}\right) \geq \lambda
_{\min }\left( \Sigma _{u_{j}}\right) \geq \underline{\lambda }>0,\quad \bar{%
\lambda}\geq \lambda _{\max }\left( \Sigma _{v_{j}}\right) \geq \lambda
_{\min }\left( \Sigma _{v_{j}}\right) \geq \underline{\lambda }>0.
\end{equation*}
\end{ass}

\begin{ass}
\label{ass:13}
\begin{itemize}
\item[(i)] There exists a high dimensional Gaussian vector $\mathbb{Z}_{%
\mathfrak{B}}^{(1)}\sim N\left(0,\Sigma_{\mathfrak{B},j}^{(1)}\right) $ such
that $\left\Vert \frac{1}{\sqrt{T}}\sum_{t=1}^{T}\mathfrak{B}_{j,t}^{(1)}-%
\mathbb{Z}_{\mathfrak{B}}^{(1)}\right\Vert_{\max }=o_{p}(1).$

\item[(ii)] There exists high dimensional Gaussian vectors $\mathbb{Z}_{%
\mathfrak{B}}^{(\ell )}\sim N\left(0,\Sigma_{\mathfrak{B},j}^{(\ell)}\right) 
$ for $\ell =2,3,4$ such that $\left(\mathbb{Z}_{%
\mathfrak{B}}^{(2 )}, \mathbb{Z}_{%
\mathfrak{B}}^{(3 )}, \mathbb{Z}_{%
\mathfrak{B}}^{(4)}\right)$ are independent, 
\begin{align*}
& \left\Vert \frac{1}{\sqrt{N_{3}}}\sum_{i\in I_{3}}\mathfrak{B}_{j,i}^{(2)}-%
\mathbb{Z}_{\mathfrak{B}}^{(2)}\right\Vert_{\max }=o_{p}(1),\quad\left\Vert 
\frac{1}{\sqrt{N_{2}}}\sum_{i\in I_{2}}\mathfrak{B}_{j,i}^{(3)}-\mathbb{Z}_{%
\mathfrak{B}}^{(3)}\right\Vert_{\max }=o_{p}(1),\text{ and} \\
& \left\Vert \frac{1}{\sqrt{N_{1}}}\sum_{i\in I_{1}}\mathfrak{B}_{j,i}^{(4)}-%
\mathbb{Z}_{\mathfrak{B}}^{(4)}\right\Vert_{\max }=o_{p}(1).
\end{align*}
\end{itemize}
\end{ass}

Assumption \ref{ass:12} implies that both $\Sigma _{u_{j}}$ and $%
\Sigma_{v_{j}}$ are well behaved. Assumption \ref{ass:13} imposes that we
can approximate high dimensional vectors $\frac{1}{\sqrt{T}}\sum_{t\in[T]}%
\mathfrak{B}_{j,t}^{(1)}$, $\frac{1}{\sqrt{N_{3}}}\sum_{i\in I_{3}}\mathfrak{%
B}_{j,i}^{(2)}$, $\frac{1}{\sqrt{N_{2}}}\sum_{i\in I_{2}}\mathfrak{B}%
_{j,i}^{(3)}$ and $\frac{1}{\sqrt{N_{1}}}\sum_{i\in I_{1}}\mathfrak{B}%
_{j,i}^{(4)}$ by four Gaussian vectors. Similar conditions have been imposed
in the literature; see, e.g., Assumption SA3 \cite{lu2021uniform}.

The following theorem reports the asymptotic properties of $S_{u_{j}}$ and $%
S_{v_{j}}$ under the respective null and alternative hypotheses.

\begin{theorem}
\label{Thm5} Suppose that Assumptions \ref{ass:1}-\ref{ass:13} and
Assumptions \ref{ass:15} in the online supplement hold and $(N,T)\rightarrow \infty $. Then

\begin{itemize}
\item[(i)] Under $H_{0}^{I}$, we have $\mathbb{P}\left(\frac{1}{2}%
S_{u_{j}}\leq x+\mathsf{b}(N)\right) \rightarrow e^{-e^{-x}}$; and under $%
H_{0}^{II}$, we have $\mathbb{P}\left(S_{v_{j}}\leq x\right) \rightarrow
e^{-3e^{-x}}.$

\item[(ii)] Under $H_{1}^{I}$, if $\frac{T}{\log N}\max_{i\in [N]}\left\Vert
c_{i,j}^{u}\right\Vert_{2}^{2}\rightarrow \infty $, we have $\mathbb{P}%
\left(S_{u_{j}}>c_{\alpha ,1\cdot N}\right) \rightarrow 1$ with $c_{\alpha
,1\cdot N}=2\mathsf{b}(N)-\log \left\vert \log \left(1-\alpha \right)
\right\vert ^{2}$ and $\alpha $ is the significance level. Under $H_{1}^{II}$%
, if $\frac{N}{\log T}\max_{t\in[T]}\left\Vert
c_{t,j}^{v}\right\Vert_{2}^{2}\rightarrow \infty $, we have $\mathbb{P}%
\left(S_{v_{j}}>c_{\alpha,2}\right)\to 1$ with $c_{\alpha,2}=-\log \left(-%
\frac{1}{3}\log \left(1-\alpha\right)\right) $.
\end{itemize}
\end{theorem}

\textbf{Remark 5.} Theorem \ref{Thm5} implies that our test statistics
follow the Gumbel distributions asymptotically under the null, are
consistent under the global alternatives, and have non-trivial power against
the local alternatives. The power function of $S_{u_{j}}$ approaches 1 as
long as $\frac{T}{\log N}\max_{i\in [N]}\left\Vert
c_{i,j}^{u}\right\Vert_{2}^{2}$ diverges to infinity as $\left(N,T\right)
\rightarrow \infty .$


\subsection{Test for an Additive Structure}

When $K_{j}=2$ for some $j\in [p]$, it is interesting to test whether $%
\Theta_{j,it}^{0}$ exhibits the additive structure which is widely assumed in a
two-way fixed effects model. That is, one may test the following null
hypothesis 
\begin{equation}
H_{0}^{III}:\Theta_{j,it}^{0}=\lambda_{j,i}+f_{j,t},~ \forall
\left(i,t\right) \in [N]\times [T],  \label{Hypo3}
\end{equation}%
The alternative hypothesis $H_{1}^{III}$ is the negation of $H_{0}^{III}.$

Let $\bar{\Theta}_{j,i\cdot }=\frac{1}{T}\sum_{t\in \lbrack T]}\Theta
_{j,it}^{0}$, $\bar{\Theta}_{j,\cdot t}^{I_{a}}=\frac{1}{N_{a}}\sum_{i\in
I_{a}}\Theta _{j,it}^{0}$, and $\bar{\Theta}_{j}^{I_{a}}=\frac{1}{N_{a}T}%
\sum_{i\in I_{a}}\sum_{t\in \lbrack T]}\Theta _{j,it}^{0}$ for $a\in \lbrack
3]$. Define 
\begin{equation*}
\Theta _{j,it}^{\ast }=\Theta _{j,it}^{0}-\bar{\Theta}_{j,i\cdot }-\bar{%
\Theta}_{j,\cdot t}^{I_{a}}+\bar{\Theta}_{j}^{I_{a}},\quad \forall i\in
I_{a},t\in \lbrack T],j\in \lbrack p].
\end{equation*}%
Note that $\Theta _{j,it}^{\ast }=0$ $\forall \left( i,t\right) \in \lbrack
N]\times \lbrack T]$ under $H_{0}^{III}.$ So we can propose a test for $%
H_{0}^{III}$ based on estimates of $\Theta _{j,it}^{\ast }.$ Define 
\begin{equation*}
\hat{\bar{\Theta}}_{j,i\cdot }=\frac{1}{T}\sum_{t\in \lbrack T]}\hat{\Theta}%
_{j,it},\text{ }\hat{\bar{\Theta}}_{j,\cdot t}^{I_{a}}=\frac{1}{N_{a}}%
\sum_{i\in I_{a}}\hat{\Theta}_{j,it},\text{ and }\hat{\bar{\Theta}}%
_{j}^{I_{a}}=\frac{1}{N_{a}T}\sum_{i\in I_{a}}\sum_{t\in \lbrack T]}\hat{%
\Theta}_{j,it}
\end{equation*}%
for $a\in \lbrack 3]$. Then, we can define the sample analogue of $\hat{%
\Theta}_{j,it}^{\ast }$ as 
\begin{equation*}
\hat{\Theta}_{j,it}^{\ast a}=\hat{\Theta}_{j,it}-\hat{\bar{\Theta}}%
_{j,i\cdot }-\hat{\bar{\Theta}}_{j,\cdot t}^{I_{a}}+\hat{\bar{\Theta}}%
_{j}^{I_{a}},\quad \forall i\in I_{a},t\in \lbrack T],j\in \lbrack p].
\end{equation*}%
Its corresponding asymptotic variance can be estimated by $\hat{\Sigma}%
_{j,it}^{\ast }$ defined as
\begin{align*}
\hat{\Sigma}_{j,it}^{\ast }& =\frac{1}{2}\sum_{a\in \lbrack 3]}\sum_{b\in
\lbrack 3]\setminus \{a\}}\frac{1}{N_{a}}\left( \hat{u}_{i,j}^{(a,b)}-\hat{%
\bar{u}}_{j}^{(a,b)}\right) ^{\prime }\hat{\Sigma}_{v_{j}}\left( \hat{u}%
_{i,j}^{(a,b)}-\bar{u}_{j}^{(a,b)}\right) \mathbf{1}_{ia} \\
& +\frac{1}{6}\sum_{a\in \lbrack 3]}\sum_{b\in \lbrack 3]\setminus \{a\}}%
\frac{1}{T}\left( \hat{v}_{t,j}^{(a,b)}-\hat{\bar{v}}_{j}^{(a,b)}\right)
^{\prime }\hat{\Sigma}_{u_{j}}\left( \hat{v}_{t,j}^{(a,b)}-\hat{\bar{v}}%
_{j}^{(a,b)}\right) ,
\end{align*}%
where $\hat{\bar{u}}_{j}^{(a,b)}=\frac{1}{N_{a}}\sum_{i\in I_{a}}\hat{u}%
_{i,j}^{(a,b)}$ and $\hat{\bar{v}}_{j}^{(a,b)}=\frac{1}{T}\sum_{t\in \lbrack
T]}\hat{v}_{t,j}^{(a,b)}$. Then, the final test statistic is 
$$S_{NT}=\max_{i\in \lbrack N],t\in \lbrack T]}\left( \hat{\Theta}%
_{j,it}^{\ast }\right) ^{2}/\hat{\Sigma}_{j,it}^{\ast }.$$

The following theorem studies the asymptotic properties of $S_{NT}$ under
the null and alternatives.

\begin{theorem}
\label{Thm6} Suppose Assumptions \ref{ass:1}-\ref{ass:15} hold and $(N,T)\rightarrow \infty $. Under $H_{0}^{III}$, 
\begin{equation*}
\mathbb{P}\left(\frac{1}{2}S_{NT}\leq x+\mathsf{b}(NT)\right) \rightarrow
e^{-e^{-x}};
\end{equation*}%
under $H_{1}^{III}$, if $\frac{N\wedge T}{\log NT}\max_{i\in [N],t\in[T]%
}\left\vert \Theta_{j,it}^{\ast }\right\vert ^{2}\rightarrow\infty $, then
we have $\mathbb{P}\left(S_{NT}>c_{\alpha ,3\cdot NT}\right) \rightarrow 1$
with $c_{\alpha ,3\cdot NT}=2\mathsf{b}(NT)-\log \left\vert \log
\left(1-\alpha \right) \right\vert ^{2}$.
\end{theorem}

Similar remark after Theorem \ref{Thm5} holds here. In particular, $%
S_{NT}$ has the desired asymptotic Gumbel distribution under the null and is
consistent under the global alternative.

\section{Monte Carlo Simulations}

In this section, we conduct a set of Monte Carlo simulations to show the
finite sample performance of our low-rank quantile regression estimates and specification tests. 

\subsection{Data Generating Processes}

Below we will consider the following data generating process (DGP): 
\begin{equation*}
Y_{it}=\Theta _{0,it}+X_{it}^{\prime }\Theta
_{it}+(1+0.1X_{1,it}+0.1X_{2,it})u_{it},
\end{equation*}%
where $X_{it}=(X_{1,it},X_{2,it})^{\prime }$, $\Theta _{it}=(\Theta
_{1,it},\Theta _{2,it})^{\prime }$, $\Theta _{0,it}$ is the intercept term
which will be specified via the IFEs.

First, we consider four DGPs where the rank of each slope matrix is 1:

\begin{itemize}[leftmargin=40pt]
\item[DGP 1:] \textbf{Constant slope with i.i.d. error.} Let $\Theta
_{0,it}=\lambda _{i}f_{t}$, where $\lambda _{i},f_{t}\sim N(2,5)$. Then let $%
\Theta _{1,it}=\Theta _{2,it}=2$ $\forall \left( i,t\right) \in \lbrack
N]\times \lbrack T]$, and $X_{j,it}=l_{j,i}^{0}w_{j,t}^{0}+U(0,1)$ for $j\in
\{1,2\}$ with $l_{1,i}^{0}$, $l_{2,i}^{0}$, $w_{1,t}^{0}$ and $w_{2,t}\sim U(0,1)$. $%
u_{it}\operatornamewithlimits{\sim}\limits^{i.i.d}\frac{t(3)}{\sqrt{3}}$.

\item[DGP 2:] \textbf{Factor slope with rank 1 and i.i.d. error.} Same as DGP 1 except that the slope coefficients follow the factor structure with one
factor rather than homogeneous across both individuals and time, i.e., $%
\Theta _{1,it}=a_{1,i}g_{1,t}$, $\Theta _{2,it}=a_{2,i}g_{2,t}$, where $%
a_{1,i}$, $g_{1,t}$, $a_{2,i}$ and $g_{2,t}\sim N(0,2)$. Except these, all
other settings remain the same as in DGP 1.

\item[DGP 3:] \textbf{Constant slope with serial correlation.} Same as DGP 1 except that we set $u_{it}=0.2u_{i,t-1}+\varepsilon _{it}$, $\varepsilon _{it}%
\operatornamewithlimits{\sim}\limits^{i.i.d}\frac{t(3)}{\sqrt{3}}$ and all
other settings remain the same.

\item[DGP 4:] \textbf{Factor slope with rank 1 and serial correlation.}
Same as DGP 2 except that we set $u_{it}=0.2u_{i,t-1}+\varepsilon _{it}$, $%
\varepsilon _{it}\operatornamewithlimits{\sim}\limits^{i.i.d}\frac{t(3)}{%
\sqrt{3}}$ and all other settings remain the same.
\end{itemize}

For the case that the rank of the slope matrix is 2, we consider two DGPs
which have the additive structure for the slope coefficient of one regressor
and the factor structure with two factors for the slope coefficient of
another regressor. Specifically,

\begin{itemize}[leftmargin=40pt]
\item[DGP 5:] \textbf{Additive and factor slopes with i.i.d. error.} $%
\Theta_{0,it}=\lambda_{i}f_{t}$, $\Theta_{1,it}=a_{1,i}+g_{1,t}$ and $%
\Theta_{2,it}=a_{2,i}^{\prime}g_{2,t}$ such that $%
a_{2,i}=(a_{2,i,1},a_{2,i,2})^{\prime}$, $g_{2,t}=(g_{2,t,1},g_{2,t,2})^{%
\prime}$, $\lambda_{i},f_{t},a_{1,i},g_{1,i}\sim N(2,5)$ and $%
a_{2,i,1},a_{2,i,2},g_{2,i,1},g_{2,i,2}\sim N(0,5)$. Moreover, $%
X_{1,it}=l_{1,i}^{0}w_{1,t}^{0}+U(0,4)$, $%
X_{2,it}=l_{2,i}^{0}w_{2,t}^{0}+Beta(2,5)$ with $l_{1,i}^{0},w_{1,t}^{0}\sim
U(0,4)$ and $l_{2,i}^{0},w_{2,t}^{0}\sim Beta(2,5)$. $u_{it}%
\operatornamewithlimits{\sim}\limits^{i.i.d}\frac{t(3)}{\sqrt{3}}$.

\item[DGP 6:] \textbf{Additive and factor slopes with serial correlation.}
Same as DGP 5 except that the error $u_{it}$ follows AR(1) process like in DGPs 3 and 4.
\end{itemize}



\subsection{Estimation Results}


For $\Theta \in \mathbb{R}^{N\times T}$, define $RMSE(\Theta )=\frac{1}{%
\sqrt{NT}}\left\Vert \Theta -\Theta ^{0}\right\Vert _{F}$. Table \ref%
{tab:RMSE} shows the RMSEs of the full-sample low rank matrix estimates
under different quantiles for each DGP. As Theorem \ref{Thm1}(i) predicts,
the RMSEs decrease as both $N$ and $T\ $increase. Given the fact that $%
N\wedge T=T$ in the simulations, the decrease of the RMSEs is largely driven
by the increase of $T.$

Table \ref{tab:rank} reports the frequency of correct rank estimation by the
singular value thresholding (SVT) approach based on 1000 replications. Note
that the true ranks of the intercept and slope matrices in DGPs 1-4 and
5-6 are 1 and 2, respectively. The results show that the SVT can
accurately determine the correct rank of the coefficient matrices in all
DGPs for all three quantile indices under investigation.

\begin{table}[h]
\caption{RMSEs of low rank estimates in the full sample}
\scriptsize \centering
\label{tab:RMSE}
\scalebox{0.8}{\begin{tabular}{cccccccccccc}
    \toprule
    \toprule
    \multirow{2}[4]{*}{DGP} & \multirow{2}[4]{*}{N} & \multirow{2}[4]{*}{T} & \multicolumn{3}{c}{$\tau=0.25$} & \multicolumn{3}{c}{$\tau=0.50$} & \multicolumn{3}{c}{$\tau=0.75$} \\
\cmidrule{4-12}          &       &       & $\tilde{\Theta}_{0}$ & $\tilde{\Theta}_{1}$ & $\tilde{\Theta}_{2}$ & $\tilde{\Theta}_{0}$ & $\tilde{\Theta}_{1}$ & $\tilde{\Theta}_{2}$ & $\tilde{\Theta}_{0}$ & $\tilde{\Theta}_{1}$ & $\tilde{\Theta}_{2}$ \\
    \midrule
    \multirow{4}[2]{*}{1} & \multirow{2}[1]{*}{75} & 35    & 0.922 & 0.324 & 0.329 & 1.242 & 0.288 & 0.297 & 1.839 & 0.609 & 0.658 \\
          &       & 70    & 0.707 & 0.280 & 0.275 & 0.819 & 0.220 & 0.203 & 1.266 & 0.519 & 0.523 \\
          & \multirow{2}[1]{*}{150} & 35    & 1.012 & 0.337 & 0.340 & 1.099 & 0.258 & 0.262 & 1.932 & 0.661 & 0.623 \\
          &       & 70    & 0.745 & 0.272 & 0.265 & 0.825 & 0.205 & 0.206 & 1.324 & 0.522 & 0.504 \\
    \midrule
    \multirow{4}[2]{*}{2} & \multirow{2}[1]{*}{75} & 35    & 0.871 & 0.521 & 0.505 & 0.881 & 0.704 & 0.680 & 1.278 & 1.055 & 0.970 \\
          &       & 70    & 0.692 & 0.401 & 0.373 & 0.672 & 0.553 & 0.537 & 1.057 & 0.744 & 0.768 \\
          & \multirow{2}[1]{*}{150} & 35    & 0.877 & 0.507 & 0.480 & 1.022 & 0.790 & 0.815 & 1.334 & 1.018 & 1.040 \\
          &       & 70    & 0.703 & 0.374 & 0.373 & 0.689 & 0.531 & 0.538 & 1.059 & 0.829 & 0.787 \\
    \midrule
    \multirow{4}[2]{*}{3} & \multirow{2}[1]{*}{75} & 35    & 0.945 & 0.334 & 0.329 & 1.115 & 0.280 & 0.265 & 1.876 & 0.630 & 0.627 \\
          &       & 70    & 0.682 & 0.286 & 0.279 & 0.809 & 0.230 & 0.214 & 1.244 & 0.486 & 0.492 \\
          & \multirow{2}[1]{*}{150} & 35    & 0.973 & 0.334 & 0.331 & 1.211 & 0.287 & 0.291 & 1.771 & 0.590 & 0.612 \\
          &       & 70    & 0.757 & 0.274 & 0.272 & 0.801 & 0.208 & 0.195 & 1.360 & 0.494 & 0.527 \\
    \midrule
    \multirow{4}[2]{*}{4} & \multirow{2}[1]{*}{75} & 35    & 0.885 & 0.515 & 0.519 & 0.915 & 0.693 & 0.723 & 1.382 & 1.125 & 1.037 \\
          &       & 70    & 0.669 & 0.393 & 0.384 & 0.652 & 0.511 & 0.520 & 1.053 & 0.812 & 0.774 \\
          & \multirow{2}[1]{*}{150} & 35    & 0.889 & 0.513 & 0.483 & 0.905 & 0.761 & 0.686 & 1.409 & 1.118 & 1.133 \\
          &       & 70    & 0.725 & 0.376 & 0.377 & 0.717 & 0.547 & 0.565 & 1.058 & 0.724 & 0.775 \\
    \midrule
    \multirow{4}[2]{*}{5} & \multirow{2}[1]{*}{75} & 35    & 0.218 & 0.268 & 0.450 & 0.307 & 0.308 & 0.606 & 0.844 & 0.466 & 0.936 \\
          &       & 70    & 0.174 & 0.226 & 0.414 & 0.213 & 0.200 & 0.493 & 0.610 & 0.388 & 0.838 \\
          & \multirow{2}[1]{*}{150} & 35    & 0.236 & 0.245 & 0.458 & 0.299 & 0.291 & 0.634 & 1.299 & 0.863 & 1.778 \\
          &       & 70    & 0.174 & 0.214 & 0.423 & 0.216 & 0.203 & 0.450 & 0.629 & 0.377 & 0.679 \\
    \midrule
    \multirow{4}[2]{*}{6} & \multirow{2}[1]{*}{75} & 35    & 0.253 & 0.267 & 0.293 & 0.382 & 0.227 & 0.421 & 1.293 & 0.609 & 0.892 \\
          &       & 70    & 0.207 & 0.239 & 0.278 & 0.261 & 0.192 & 0.366 & 0.576 & 0.287 & 0.415 \\
          & \multirow{2}[1]{*}{150} & 35    & 0.225 & 0.254 & 0.269 & 0.363 & 0.225 & 0.422 & 1.486 & 0.695 & 0.992 \\
          &       & 70    & 0.193 & 0.254 & 0.263 & 0.254 & 0.171 & 0.379 & 0.797 & 0.391 & 0.551 \\
    \bottomrule
    \end{tabular}} 
\end{table}

\begin{table}[h]
\caption{Frequency of correct rank estimation via the SVT approach}
\label{tab:rank}
\scriptsize \centering
\scalebox{0.8}{\begin{tabular}{cccccccccccc}
    \toprule
    \toprule
    \multirow{2}[4]{*}{DGP} & \multirow{2}[4]{*}{N} & \multirow{2}[4]{*}{T} & \multicolumn{3}{c}{$\tau=0.25$} & \multicolumn{3}{c}{$\tau=0.50$} & \multicolumn{3}{c}{$\tau=0.75$} \\
\cmidrule{4-12}          &       &       & $\hat{K}_{0}$ & $\hat{K}_{1}$ & $\hat{K}_{2}$ & $\hat{K}_{0}$ & $\hat{K}_{1}$ & $\hat{K}_{2}$ & $\hat{K}_{0}$ & $\hat{K}_{1}$ & $\hat{K}_{2}$ \\
    \midrule
    \multirow{4}[2]{*}{1} & \multirow{2}[1]{*}{75} & 35    & 1.00  & 0.996 & 0.996 & 1.00  & 0.999 & 1.00  & 1.00  & 0.999 & 0.999 \\
          &       & 70    & 1.00  & 0.994 & 0.996 & 1.00  & 1.00  & 1.00  & 1.00  & 1.00  & 1.00 \\
          & \multirow{2}[1]{*}{150} & 35    & 1.00  & 0.994 & 0.995 & 1.00  & 1.00  & 0.999 & 1.00  & 1.00  & 0.999 \\
          &       & 70    & 1.00  & 0.995 & 0.996 & 1.00  & 0.999 & 1.00  & 1.00  & 0.999 & 1.00 \\
    \midrule
    \multirow{4}[2]{*}{2} & \multirow{2}[1]{*}{75} & 35    & 1.00  & 0.993 & 0.999 & 1.00  & 0.999 & 1.00  & 1.00  & 1.00  & 1.00 \\
          &       & 70    & 1.00  & 0.997 & 0.998 & 1.00  & 1.00  & 1.00  & 1.00  & 1.00  & 1.00 \\
          & \multirow{2}[1]{*}{150} & 35    & 1.00  & 0.995 & 0.996 & 1.00  & 0.998 & 1.00  & 1.00  & 1.00  & 1.00 \\
          &       & 70    & 1.00  & 1.00  & 1.00  & 1.00  & 1.00  & 0.998 & 1.00  & 1.00  & 1.00 \\
    \midrule
    \multirow{4}[2]{*}{3} & \multirow{2}[1]{*}{75} & 35    & 1.00  & 0.990 & 0.997 & 1.00  & 0.997 & 0.997 & 1.00  & 0.999 & 0.999 \\
          &       & 70    & 1.00  & 0.994 & 0.994 & 1.00  & 0.999 & 0.999 & 1.00  & 0.999 & 0.998 \\
          & \multirow{2}[1]{*}{150} & 35    & 1.00  & 0.999 & 0.992 & 1.00  & 1.00  & 1.00  & 1.00  & 1.00  & 1.00 \\
          &       & 70    & 1.00  & 0.996 & 0.994 & 1.00  & 0.998 & 1.00  & 1.00  & 1.00  & 1.00 \\
    \midrule
    \multirow{4}[2]{*}{4} & \multirow{2}[1]{*}{75} & 35    & 1.00  & 0.992 & 0.991 & 1.00  & 0.999 & 0.999 & 1.00  & 0.999 & 1.00 \\
          &       & 70    & 1.00  & 0.995 & 0.995 & 1.00  & 0.999 & 0.999 & 1.00  & 1.00  & 1.00 \\
          & \multirow{2}[1]{*}{150} & 35    & 1.00  & 0.996 & 0.997 & 1.00  & 0.999 & 1.00  & 1.00  & 1.00  & 1.00 \\
          &       & 70    & 1.00  & 0.997 & 0.999 & 1.00  & 0.999 & 1.00  & 1.00  & 1.00  & 1.00 \\
    \midrule
    \multirow{4}[2]{*}{5} & \multirow{2}[1]{*}{75} & 35    & 1.00  & 1.00  & 1.00  & 1.00  & 1.00  & 1.00  & 1.00  & 1.00  & 1.00 \\
          &       & 70    & 1.00  & 1.00  & 1.00  & 1.00  & 1.00  & 1.00  & 1.00  & 1.00  & 1.00 \\
          & \multirow{2}[1]{*}{150} & 35    & 1.00  & 1.00  & 1.00  & 1.00  & 1.00  & 1.00  & 1.00  & 1.00  & 1.00 \\
          &       & 70    & 1.00  & 1.00  & 1.00  & 1.00  & 1.00  & 1.00  & 1.00  & 1.00  & 1.00 \\
    \midrule
    \multirow{4}[2]{*}{6} & \multirow{2}[1]{*}{75} & 35    & 1.00  & 1.00  & 1.00  & 1.00  & 1.00  & 1.00  & 1.00  & 1.00  & 1.00 \\
          &       & 70    & 1.00  & 1.00  & 1.00  & 1.00  & 1.00  & 1.00  & 1.00  & 1.00  & 1.00 \\
          & \multirow{2}[1]{*}{150} & 35    & 1.00  & 1.00  & 0.999 & 1.00  & 1.00  & 1.00  & 1.00  & 1.00  & 1.00 \\
          &       & 70    & 1.00  & 1.00  & 1.00  & 1.00  & 1.00  & 1.00  & 1.00  & 1.00  & 1.00 \\
    \bottomrule
    \end{tabular}} 
\end{table}

\subsection{Test Results}

In Section 4, we define $S_{u_{j}}$\ and $S_{v_{j}}$\ as the sup-type test
statistics. Table \ref{tab:test_homo} reports the empirical size and power
at the 5\% nominal level for the null hypothesis that the slope coefficient
is homogeneous across either $i$ or $t$. The results in DGPs 1 and 3 give
the empirical size, and those in DGPs 2 and 4 give the empirical power. As
the results in Table \ref{tab:test_homo} indicate, our tests have reasonable
size despite the fact that they are slightly conservative like most extreme-value
based sup-tests in the literature. In terms of power, out tests have superb
power in both DGPs across all three quantile indices.

Table \ref{tab:test_additive} shows the empirical size and power of our test
for DGPs 5 and 6. The findings are similar to those in Table \ref%
{tab:test_homo}. In particular, our tests are a bit conservative under the null. The empirical power tends to 1 quickly as $T$ increases.

\begin{table}[h]
\caption{Empirical size and power of testing slope homogeneity across either 
$i$ or $t$ (nominal level: $0.05$)}
\label{tab:test_homo}
\scriptsize \centering
\begin{tabular}{ccccccccccccccc}
\toprule \toprule \multirow{2}[4]{*}{DGP} & \multirow{2}[4]{*}{N} & %
\multirow{2}[4]{*}{T} & \multicolumn{4}{c}{$\tau=0.25$} & \multicolumn{4}{c}{%
$\tau=0.5$} & \multicolumn{4}{c}{$\tau=0.75$} \\ 
\cmidrule{4-15} &  &  & $u_{1}$ & $v_{1}$ & $u_{2}$ & $v_{2}$ & $u_{1}$ & $%
v_{1}$ & $u_{2}$ & $v_{2}$ & $u_{1}$ & $v_{1}$ & $u_{2}$ & $v_{2}$ \\ 
\midrule \multirow{4}[2]{*}{DGP 1} & \multirow{2}[1]{*}{75} & 35 & 0.040 & 
0.051 & 0.049 & 0.032 & 0.024 & 0.054 & 0.034 & 0.054 & 0.036 & 0.047 & 0.036
& 0.048 \\ 
&  & 70 & 0.040 & 0.055 & 0.050 & 0.044 & 0.020 & 0.056 & 0.017 & 0.068 & 
0.025 & 0.037 & 0.029 & 0.029 \\ 
& \multirow{2}[1]{*}{150} & 35 & 0.028 & 0.036 & 0.058 & 0.048 & 0.065 & 
0.054 & 0.052 & 0.055 & 0.074 & 0.030 & 0.076 & 0.024 \\ 
&  & 70 & 0.034 & 0.025 & 0.030 & 0.023 & 0.035 & 0.048 & 0.028 & 0.040 & 
0.035 & 0.025 & 0.039 & 0.025 \\ 
\midrule \multirow{4}[2]{*}{DGP 2} & \multirow{2}[1]{*}{75} & 35 & 1.00 & 
1.00 & 1.00 & 1.00 & 1.00 & 1.00 & 1.00 & 1.00 & 1.00 & 1.00 & 1.00 & 1.00
\\ 
&  & 70 & 1.00 & 1.00 & 1.00 & 1.00 & 1.00 & 1.00 & 1.00 & 1.00 & 1.00 & 1.00
& 1.00 & 1.00 \\ 
& \multirow{2}[1]{*}{150} & 35 & 1.00 & 1.00 & 1.00 & 1.00 & 1.00 & 1.00 & 
1.00 & 1.00 & 1.00 & 1.00 & 1.00 & 1.00 \\ 
&  & 70 & 1.00 & 1.00 & 1.00 & 1.00 & 1.00 & 1.00 & 1.00 & 1.00 & 1.00 & 1.00
& 1.00 & 1.00 \\ 
\midrule \multirow{4}[2]{*}{DGP 3} & \multirow{2}[1]{*}{75} & 35 & 0.045 & 
0.057 & 0.050 & 0.041 & 0.022 & 0.050 & 0.038 & 0.089 & 0.054 & 0.047 & 0.049
& 0.047 \\ 
&  & 70 & 0.048 & 0.031 & 0.031 & 0.033 & 0.028 & 0.086 & 0.023 & 0.069 & 
0.041 & 0.046 & 0.032 & 0.038 \\ 
& \multirow{2}[1]{*}{150} & 35 & 0.065 & 0.054 & 0.058 & 0.034 & 0.064 & 
0.051 & 0.068 & 0.045 & 0.084 & 0.018 & 0.089 & 0.023 \\ 
&  & 70 & 0.046 & 0.030 & 0.044 & 0.025 & 0.022 & 0.037 & 0.037 & 0.030 & 
0.046 & 0.022 & 0.048 & 0.015 \\ 
\midrule \multirow{4}[2]{*}{DGP 4} & \multirow{2}[1]{*}{75} & 35 & 1.00 & 
1.00 & 1.00 & 1.00 & 1.00 & 1.00 & 1.00 & 1.00 & 1.00 & 1.00 & 1.00 & 1.00
\\ 
&  & 70 & 1.00 & 1.00 & 1.00 & 1.00 & 1.00 & 1.00 & 1.00 & 1.00 & 1.00 & 1.00
& 1.00 & 1.00 \\ 
& \multirow{2}[1]{*}{150} & 35 & 1.00 & 1.00 & 1.00 & 1.00 & 1.00 & 1.00 & 
1.00 & 1.00 & 1.00 & 1.00 & 1.00 & 1.00 \\ 
&  & 70 & 1.00 & 1.00 & 1.00 & 1.00 & 1.00 & 1.00 & 1.00 & 1.00 & 1.00 & 1.00
& 1.00 & 1.00 \\ 
\bottomrule &  &  &  &  &  &  &  &  &  &  &  &  &  & 
\end{tabular}
\end{table}

\begin{table}[h]
\caption{Empirical size and power for testing additive slopes (nomial level: 
$0.05$)}
\label{tab:test_additive}
\scriptsize \centering
\begin{tabular}{ccccccccc}
\toprule \toprule \multirow{2}[4]{*}{DGP} & \multirow{2}[4]{*}{N} & %
\multirow{2}[4]{*}{T} & \multicolumn{2}{c}{$\tau=0.25$} & \multicolumn{2}{c}{%
$\tau=0.50$} & \multicolumn{2}{c}{$\tau=0.75$} \\ 
\cmidrule{4-9} &  &  & size & power & size & power & size & power \\ 
\midrule \multirow{4}[2]{*}{DGP 5} & \multirow{2}[1]{*}{75} & 35 & 0.027 & 
0.807 & 0.034 & 1.00 & 0.026 & 0.979 \\ 
&  & 70 & 0.065 & 1.00 & 0.061 & 1.00 & 0.034 & 1.00 \\ 
& \multirow{2}[1]{*}{150} & 35 & 0.018 & 1.00 & 0.026 & 1.00 & 0.011 & 1.00
\\ 
&  & 70 & 0.012 & 1.00 & 0.029 & 1.00 & 0.010 & 1.00 \\ 
\midrule \multirow{4}[2]{*}{DGP 6} & \multirow{2}[1]{*}{75} & 35 & 0.039 & 
0.796 & 0.058 & 1.00 & 0.045 & 0.979 \\ 
&  & 70 & 0.024 & 1.00 & 0.06 & 1.00 & 0.032 & 1.00 \\ 
& \multirow{2}[1]{*}{150} & 35 & 0.019 & 1.00 & 0.025 & 1.00 & 0.022 & 1.00
\\ 
&  & 70 & 0.01 & 1.00 & 0.022 & 1.00 & 0.014 & 1.00 \\ 
\bottomrule &  &  &  &  &  &  &  & 
\end{tabular}
\end{table}

\section{Empirical Study}

In this section we consider two empirical applications: the heterogeneous investment equation and the heterogeneous
quantile effect of foreign direct investment on unemployment.
\subsection{Investment Equation}

In this subsection, we revisit the investment equation. \cite%
{fazzari1988financing} point out that investment may show sensitivity to
movements in cash flow when firms face constraints for external finance.
Since \cite{fazzari1988financing}, there has been a large literature on the effect of cash flow on
the corporate investment; see \cite{devereux199011}, \cite%
{gilchrist1995evidence}, \cite{kaplan1995financing}, \cite%
{cleary1999relationship}, \cite{rauh2006investment}, and \cite%
{almeida2007financial}, among others. Using the panel dataset, we consider
the scaled version of the investment equation as follows: 
\begin{equation*}
\frac{I_{it}}{K_{i,t-1}}=\Theta _{0,it}+\Theta _{1,it}\frac{CF_{it}}{%
K_{i,t-1}}+\Theta _{2,it}q_{i,t-1}+u_{it},
\end{equation*}%
where $I$ is the corporate investment, $CF$ is the cash flow, $q$ is the
Tobin's q, $K$ is the capital stock and $u$ is the innovation. $\Theta
_{0,it}$ refers to the fixed effects (FEs). Rather than the mean estimation, 
\cite{galvao2015efficient} estimate the effects of the firm's cash flow and
Tobin's q on investment at different quantiles. By using the panel quantile
regression with individual FEs, they show that the slope estimates change
across $\tau $. However, they do not allow the slope
coefficients, $\Theta _{1}$ and $\Theta _{2},$ to change either over $i$ or $%
t.$ Inspired by \cite{galvao2015efficient}, we estimate the following model 
\begin{equation}
\label{base}
\mathscr{Q}_{\tau }\left( IK_{it}\big |\left\{ CFK_{it},q_{i,t-1}\right\}
_{t\in \lbrack T]},\left\{ \Theta _{j,it}\right\} _{t\in \lbrack T],j\in
\{0,1,2\}}\right) =\Theta _{0,it}(\tau )+\Theta _{1,it}(\tau
)CFK_{it}+\Theta _{2,it}(\tau )q_{i,t-1},
\end{equation}
where $IK_{it}=\frac{I_{it}}{K_{i,t-1}}$, and $CFK_{it}=\frac{CF_{it}}{%
K_{i,t-1}}$. Here we don't restrict the specific structure on the FEs and
they can be either additive or interactive.

The data are taken from the China Stock Market \& Accounting Research
(CSMAR) Database. We use quarterly data for 195 manufacturing firms in China
from 2003 to 2020. Based on the model (\ref{base}), we define corporate
investment as $I_{it}=LI_{it}-LI_{i,t-1}$, where $LI_{it}$ is the total
value of long-term corporate investment as the sum of long-term equity
investment, long-term bound investment, fixed assets and immaterial assets.
The investment measures the change of firm's total investment compared to
the last period. All these four variables can be easily obtained from the
balance sheet. We directly use Tobin's q from the CSMAR database, where by
definition $q=\frac{MV}{K}$ and $MV$ is the market value of the firm. We obtain a balanced panel dataset with 195 firms and 72
time periods. 
The units of corporate investment, capital and cash flow are measured by billions of
Chinese RMB.

By using the SVT approach, we obtain the estimates of the ranks of $\Theta
_{1}$ and $\Theta _{2}:$ $\hat{r}_{1}=\hat{r}_{2}=1$ for each $\tau
=\{0.25,0.5,0.75\}$. Consequently, we can consider the test that whether $%
\Theta _{j,it}$ is constant over $i$ or constant over $t$ for both $j=1,2$.
Specifically, we want to test whether the effect of cash flow and Tobin's q
on the firm's investment is homogeneous over $i$ or across $t$ with market
imperfection. That is, for $j\in \{1,2\},$ we shall test

\begin{itemize}
\item $H_{0}^{a}$: $\Theta _{j,it}$ is a constant over $i$,

\item $H_{0}^{b}$: $\Theta _{j,it}$ is a constant over $t$.
\end{itemize}

Figure \ref{fig:emp_1} shows the estimation results for the factor and
factor loadings of two slope coefficient matrices under different quantiles.
In each sub-figure, the first and second rows report the results for $\Theta
_{1}$ and $\Theta _{2},$ respectively. Specifically, the first row of Figure %
\ref{fig:emp_1}(a) gives the plot of $\left\{ \hat{u}_{i,1}\right\} _{i\in
\lbrack N]}$ as a catenation of $\{\hat{u}_{i,1}^{(1,2)},\hat{u}%
_{i,1}^{(2,3)},\hat{u}_{i,1}^{(3,1)}\}$ at the left and as a cantenation of $%
\{\hat{u}_{i,1}^{(1,3)},\hat{u}_{i,1}^{(2,1)},\hat{u}_{i,1}^{(3,2)}\}$ at
the right in the first row, and similarly the plot of $\left\{ \hat{u}%
_{i,2}\right\} _{i\in \lbrack N]}\ $in the second row. Similarly, the first
row of Figure \ref{fig:emp_1}(d) shows $\{\hat{v}_{t,1}^{(a,b)}\}_{t\in
\lbrack T]}$ for $a\in \lbrack 3],$ $b\in \lbrack 3]\setminus \{a\}$ in the
first row and $\{\hat{v}_{t,3}^{(a,b)}\}_{t\in \lbrack T]}$ for $a\in
\lbrack 3],$ $b\in \lbrack 3]\setminus \{a\}$ in the second row.

Table \ref{tab:emp_1} reports the test statistics, critical values, and $p$-values. Tobin's q can measure a firm's investment demand. After controlling
the Tobin's q and the intercept FEs, the coefficient of cash flow captures a
firm's potential for external investment with the variation of internal
finance. It is clear that we can reject the homogeneous hypotheses for both $%
i$ and $t$ at the 1\% significance level for each $\tau \in
\{0.25,0.5,0.75\} $. This indicates that with high probability, the slope
coefficient of both $CFK$ and Tobin's q follow the factor structure with one
factor.

The above study shows strong evidence that under imperfect market, the
sensitivity of corporate investment to cash flow exhibits both individual
heterogeneity and time heterogeneity across quantiles. It implies that
neither the usual homogenous panel QR model nor the panel QR model with
either cross-section or time heterogeneity alone in the slope coefficients
fails to fully capture the unobserved heterogeneity in the investment
equation.


\begin{table}[tbph]
\caption{Test results under different quantiles for the investment equation}
\label{tab:emp_1}
\scriptsize \centering
\begin{tabular}{ccccccc}
\toprule \toprule $\tau$ & Test & $S$ & $cv_{\alpha=0.01}$ & $%
cv_{\alpha=0.05}$ & $cv_{\alpha=0.1}$ & $p$-value \\ 
\midrule \multirow{4}[2]{*}{0.25} & $u_{CFK}$ & $1.28\times10^{3}$ & %
\multirow{2}[1]{*}{16.94} & \multirow{2}[1]{*}{13.68} & %
\multirow{2}[1]{*}{12.24} & 0.00 \\ 
& $u_{q}$ & $4.16\times10^{4}$ &  &  &  & 0.00 \\ 
& $v_{CFK}$ & 13.85 & \multirow{2}[1]{*}{5.70} & \multirow{2}[1]{*}{4.07} & %
\multirow{2}[1]{*}{3.35} & 0.00 \\ 
& $v_{q}$ & 870.85 &  &  &  & 0.00 \\ 
\midrule \multirow{4}[2]{*}{0.50} & $u_{CFK}$ & 148.28 & %
\multirow{2}[1]{*}{16.94} & \multirow{2}[1]{*}{13.68} & %
\multirow{2}[1]{*}{12.24} & 0.00 \\ 
& $u_{q}$ & $1.24\times10^{5}$ &  &  &  & 0.00 \\ 
& $v_{CFK}$ & 49.57 & \multirow{2}[1]{*}{5.70} & \multirow{2}[1]{*}{4.07} & %
\multirow{2}[1]{*}{3.35} & 0.00 \\ 
& $v_{q}$ & 138.83 &  &  &  & 0.00 \\ 
\midrule \multirow{4}[2]{*}{0.75} & $u_{CFK}$ & 313.21 & %
\multirow{2}[1]{*}{16.94} & \multirow{2}[1]{*}{13.68} & %
\multirow{2}[1]{*}{12.24} & 0.00 \\ 
& $u_{q}$ & $2.03\times10^{4}$ &  &  &  & 0.00 \\ 
& $v_{CFK}$ & 31.50 & \multirow{2}[1]{*}{5.70} & \multirow{2}[1]{*}{4.07} & %
\multirow{2}[1]{*}{3.35} & 0.00 \\ 
& $v_{q}$ & 58.29 &  &  &  & 0.00 \\ 
\bottomrule &  &  &  &  &  & \\
\multicolumn{7}{p{10cm}}{$Notes$: $S$ is the test statistics for the factor
or factor loadings under different quantiles, $H_{0}^{a}(CFK)$ and $%
H_{0}^{a}(q) $ refer to the hypotheses that the slope of CFK and Tobin'q is
homogeneous across $i$, respectively. $H_{0}^{b}(CFK)$ and $H_{0}^{b}(q)$
refer to the the hypotheses that the slope of CFK and Tobin'q is homogeneous
across $t$, respectively. $\text{cv}_{\alpha=a}$ is the critical value under
the significance level a where a=0.1, 0.05, and 0.01.}%
\end{tabular}
\end{table}

\begin{figure}[tbp]
\centering
\begin{subfigure}[b]{0.3\textwidth}
         \centering
         \includegraphics[width=\textwidth,height=2.5cm]{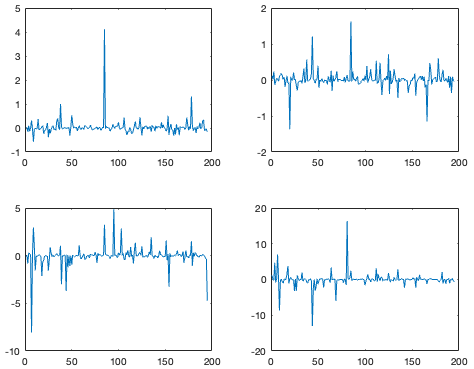}
         \caption{factor loading estimates under $\tau=0.25$}
     \end{subfigure}
\hfill 
\begin{subfigure}[b]{0.3\textwidth}
         \centering
         \includegraphics[width=\textwidth,height=2.5cm]{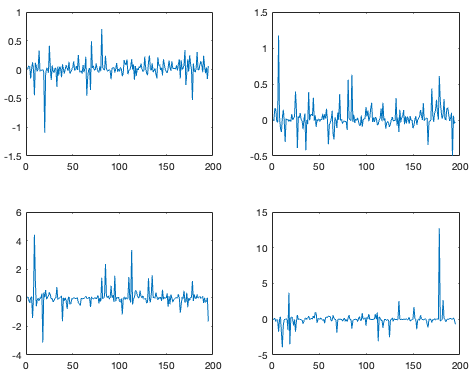}
         \caption{factor loading estimates under $\tau=0.5$}
     \end{subfigure}
\hfill 
\begin{subfigure}[b]{0.3\textwidth}
         \centering
         \includegraphics[width=\textwidth,height=2.5cm]{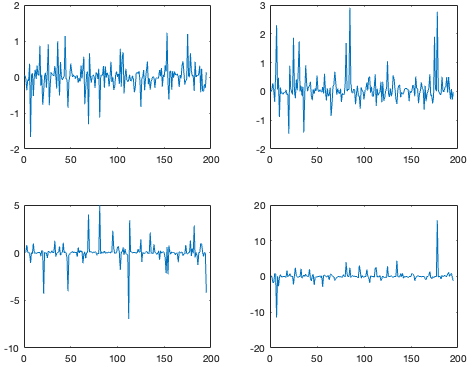}
         \caption{factor loading estimates under $\tau=0.75$}
     \end{subfigure}
\hfill 
\begin{subfigure}[b]{0.3\textwidth}
         \centering
         \includegraphics[width=\textwidth,height=2.5cm]{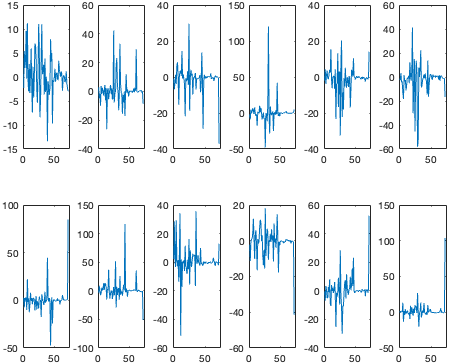}
         \caption{factor estimates under $\tau=0.25$}
     \end{subfigure}
\hfill 
\begin{subfigure}[b]{0.3\textwidth}
         \centering
         \includegraphics[width=\textwidth,height=2.5cm]{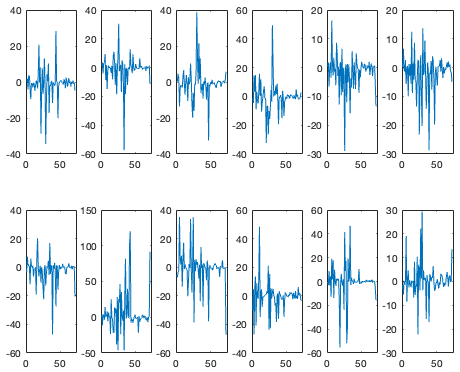}
         \caption{factor estimates under $\tau=0.5$}
     \end{subfigure}
\hfill 
\begin{subfigure}[b]{0.3\textwidth}
         \centering
         \includegraphics[width=\textwidth,height=2.5cm]{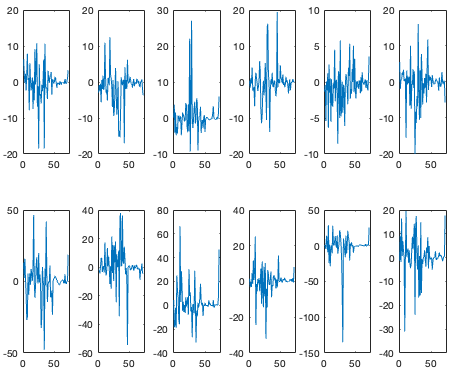}
         \caption{factor estimates under $\tau=0.75$}
     \end{subfigure}
\caption{Factor loading and factor estimates under different quantiles}
\label{fig:emp_1}
\end{figure}

\subsection{Foreign Direct Investment and Unemployment}

Investment is one of the major driving forces for economic growth and
employment. Among the investment, foreign direct investment (FDI) is an
important contributor to the employment. See \cite{craigwell2006foreign}, 
\cite{aktar2009can}, \cite{karlsson2009foreign}, \cite{mucuk2013effect}, and 
\cite{strat2015fdi}, among others. Controversially, \cite{mucuk2013effect}
argue that FDI may have both positive and negative effects on employment. On
the one hand, FDI adds to the net capital and creates jobs through forward and backward linkages and
multiplier effects in local economy. On the other hand, acquisitions may
rely on imports or displacement of existing firms which may result in job
loss.

To study the relationship of FDI, economic growth rate and unemployment at
the country level, we consider the following panel quantile regression
model, 
\begin{equation*}
\mathscr{Q}_{\tau }\left( U_{it}\big |\left\{ G_{i,t-1},FDI_{it}\right\}
_{t\in \lbrack T]},\left\{ \Theta _{j,it}\right\} _{t\in \lbrack T],j\in
\{0,1,2\}}\right) =\Theta _{0,it}(\tau )+\Theta _{1,it}(\tau
)G_{i,t-1}+\Theta _{2,it}(\tau )FDI_{it},
\end{equation*}%
where $U_{it}$ is the unemployment rate of country $i$ at year $t$, $%
G_{i,t-1}$ is the economic growth measured by the growth of real GDP. $%
\Theta _{0,it}$ is the FEs of country $i$ and year $t$, $\Theta _{1,it}$ is
the elasticity of the economic growth in the previous year to the
unemployment this year, and $\Theta _{2,it}$ is the elasticity of FDI to the
unemployment.

We draw the data for 126 countries from 1992-2019. The data for the
unemployment rate are taken from International Labor Organization (ILO) and
GDP growth and FDI are from the World Bank Development Indicators (WDI)
historical database. The rank estimation procedure shows that $\hat{r}_{1}=2$
and $\hat{r}_{2}=1$. Consequently, we can test whether the elasticity of FDI
to the unemployment rate is homogeneous across individual countries and over
years 1992-2009, and whether the elasticity of growth rate to unemployment
follows the additive structure, i.e.,

\begin{itemize}
\item $H_{0}^{c}$: $\Theta_{1,it}=\Theta_{1,i}+\Theta_{1,t}$,

\item $H_{0}^{d}$: $\Theta_{2,it}$ is a constant over $i$,

\item $H_{0}^{e}$: $\Theta_{2,it}$ is a constant over $t$.
\end{itemize}

Table \ref{tab:emp_2} reports the test results under quantiles 0.25, 0.5 and
0.75 for the above three null hypotheses. Figure \ref{fig:emp_2} gives the
estimation results for the factor and factor loading estimates of the slope
coefficient $\Theta _{2}$. As Table \ref{tab:emp_2} suggests, we can reject
all the above three null hypotheses safely at the conventional 5\%
significance level. This means that the effect of FDI on the unemployment
rate is different across both countries and time even though the estimated
rank of $\Theta _{2}$ is one, and the effect of economic growth rate on the
unemployment is heterogeneous across both countries and time and it does not
exhibit an additive structure.


\begin{table}[htbp]
  \caption{Test results under different quantiles}
  \label{tab:emp_2}
   \scriptsize \centering
    \begin{tabular}{ccccccc}
    \toprule
    \toprule
    Test  & $\tau$ & $S$   & $cv_{\alpha=0.01}$ & $cv_{\alpha=0.05}$ & $cv_{\alpha=0.10}$ & $p-$value \\
    \midrule
    \multirow{3}[2]{*}{$H_{0}^{c}$} & 0.25  & 38.92 & 35.55 & 32.29 & 30.85 & 0.00 \\
          & 0.50   & 80.84 & 22.29 & 19.03 & 17.59 & 0.00 \\
          & 0.75  & 66.24 & 35.55 & 32.29 & 30.85 & 0.00 \\
    \midrule
    \multirow{3}[2]{*}{$H_{0}^{d}$} & 0.25  & $1.41\times 10^{6}$ & \multirow{3}[2]{*}{16.15} & \multirow{3}[2]{*}{12.89} & \multirow{3}[2]{*}{11.45} & 0.00 \\
          & 0.50   & $6.39\times 10^{6}$ &       &       &       & 0.00 \\
          & 0.75  & $3.07\times 10^{7}$ &       &       &       & 0.00 \\
    \midrule
    \multirow{3}[2]{*}{$H_{0}^{e}$} & 0.25  & 36.36 & \multirow{3}[2]{*}{5.70} & \multirow{3}[2]{*}{4.07} & \multirow{3}[2]{*}{3.35} & 0.00 \\
          & 0.50   & 164.03 &       &       &       & 0.00 \\
          & 0.75  & 5.44  &       &       &       & 0.013 \\
    \bottomrule
    \end{tabular}%
\end{table}%

\begin{figure}[h]
\centering
\begin{subfigure}[b]{0.3\textwidth}
         \centering
         \includegraphics[width=\textwidth,height=2.5cm]{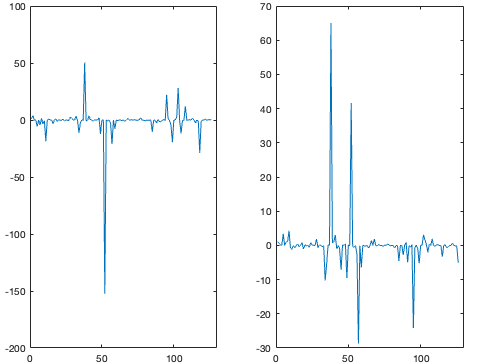}
         \caption{factor loading estimates under $\tau=0.25$}
     \end{subfigure}
\hfill 
\begin{subfigure}[b]{0.3\textwidth}
         \centering
         \includegraphics[width=\textwidth,height=2.5cm]{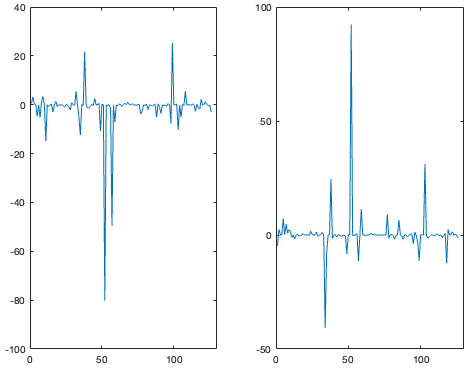}
         \caption{factor loading estimates under $\tau=0.5$}
     \end{subfigure}
\hfill 
\begin{subfigure}[b]{0.3\textwidth}
         \centering
         \includegraphics[width=\textwidth,height=2.5cm]{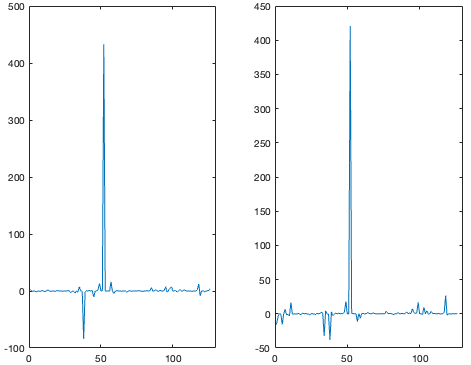}
         \caption{factor loading estimates under $\tau=0.75$}
     \end{subfigure}
\hfill 
\begin{subfigure}[b]{0.3\textwidth}
         \centering
         \includegraphics[width=\textwidth,height=2.5cm]{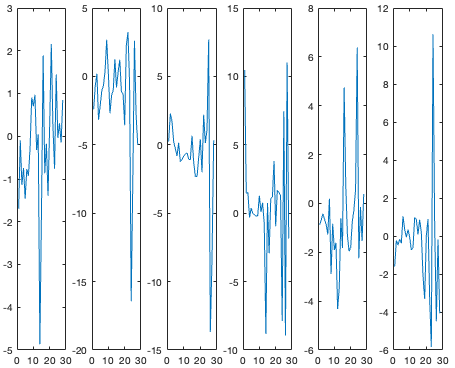}
         \caption{factor estimates under $\tau=0.25$}
     \end{subfigure}
\hfill 
\begin{subfigure}[b]{0.3\textwidth}
         \centering
         \includegraphics[width=\textwidth,height=2.5cm]{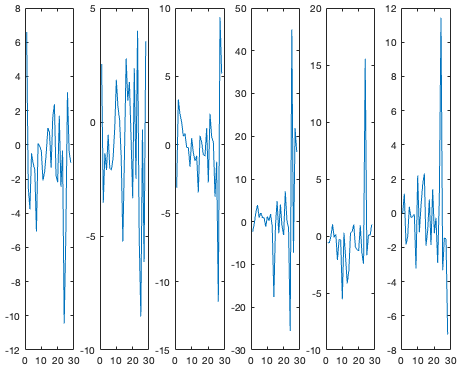}
         \caption{factor estimates under $\tau=0.5$}
     \end{subfigure}
\hfill 
\begin{subfigure}[b]{0.3\textwidth}
         \centering
         \includegraphics[width=\textwidth,height=2.5cm]{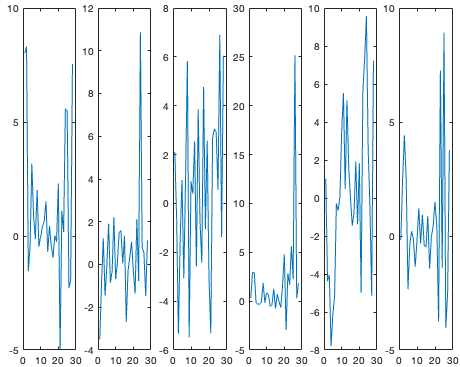}
         \caption{factor estimates under $\tau=0.75$}
     \end{subfigure}
\caption{Factor loading and factor estimates of $\Theta_{2}$ under different
quantiles}
\label{fig:emp_2}
\end{figure}

\section{Conclusion}

This paper considers panel QR model with heterogeneous slopes over both $i$ and $t$. Compared to \cite{chernozhukov2019inference}, to remove the bias from the nuclear norm regularization, we split the full sample into three subsamples. We then use the first subsample to compute initial estimators via NNR, the second sample to refine the convergence rate of the initial estimator, and the last subsample to debias the refined estimator. Our asymptotic theory shows that the factor estimates, factor loading estimates and the slope estimates all follow the normal distributions asymptotically. By constructing the consistent estimator for the asymptotic variance, we also conduct two specification tests: (1) the slope coefficient is constant over time or individuals under the case that true rank of slope matrix equals one and (2) the slope coefficient exhibits the additive structure under the case that true rank of the slope coefficient matrix equals two. Our test statistics are shown to follow the Gumbel distribution asymptotically under the null, consistent under the global alternative and have non-trivial power against local alternatives.  Monte Carlo simulation and empirical studies illustrate the finite sample performance of our algorithm and test statistics.

\bibliography{chapter1}
\newpage 
\appendix%
\linespread{1.2}%
\small%
\newpage

\setcounter{footnote}{0}\setcounter{page}{1}\setlength{\baselineskip}{16pt}{}

\begin{center}
{\small {\Large Online Supplement for} }

{\Large \textquotedblleft Low-rank Panel Quantile Regression: Estimation and
Inference\textquotedblright }

$\medskip $

Yiren Wang$^{a}$, Liangjun Su$^{b}$ and Yichong Zhang$^{a}$

$^{a}$School of Economics, Singapore Management University, Singapore

$^{b}$School of Economics and Management, Tsinghua University, China

{\small \ \ \ \ \ \ \ \ \ }
\end{center}

\noindent This supplement contains three sections. Section A contains the
proofs of the main results by calling upon some technical lemmas in Section
B. Section B states and proves the technical lemmas used in Section A.
Section C provides detail algorithm for the nuclear norm regularized panel
quantile regression.

\section{Proofs of the Main Results}

\subsection{Proof of Theorem \protect\ref{Thm1}}

We focus on the full sample estimators $\tilde{\Delta}_{\Theta_j}$, $\tilde{%
\sigma}_{k,j}$, and $\tilde{V}_j$ in the proof. The results for their
subsample counterparts can be established in the same manner, and we omit
the detail for brevity.

\subsubsection{Proof of Statement (i)}

Recall that 
\begin{equation*}
\mathcal{R}(C_{1},C_{2}):=\left\{
(\{\Delta_{\Theta_{j}}\}_{j=0}^{p}):\sum_{j=0}^{p}\left\Vert \mathcal{P}%
_{j}^{\bot }(\Delta_{\Theta_{j}})\right\Vert_{\ast }\leq
C_{1}\sum_{j=0}^{p}\left\Vert \mathcal{P}_{j}(\Delta_{\Theta_{j}})\right%
\Vert_{\ast},\,\,\sum_{j=0}^{p}\left\Vert
\Delta_{\Theta_{j}}\right\Vert_{F}^{2}\geq C_{2}\sqrt{NT}\right\} .
\end{equation*}%
Define $\mathcal{R}(C_{1}):=\left\{
\{\Delta_{\Theta_{j}}\}_{j=0}^{p}:\sum_{j=0}^{p}\left\Vert \mathcal{P}%
_{j}^{\bot }(\Delta_{\Theta_{j}})\right\Vert_{\ast }\leq
C_{1}\sum_{j=0}^{p}\left\Vert \mathcal{P}_{j}(\Delta_{\Theta_{j}})\right%
\Vert_{\ast }\right\} $. By Lemma \ref{Lem:RS}, $\mathbb{P\{}\{\tilde{\Delta}%
_{\Theta_{j}}(\tau)\}_{j=0}^{p}$ $\in \mathcal{R}(3)\}\rightarrow 1.$ When $%
\{\tilde{\Delta}_{\Theta_{j}}\}_{j=0}^{p}\in \mathcal{R}(C_{1})$ and $\{%
\tilde{\Delta}_{\Theta_{j}}\}_{j=0}^{p}\notin \mathcal{R}(3,C_{2})$, we have 
$\sum_{j=0}^{p}\left\Vert \tilde{\Delta}_{\Theta_{j}}\right%
\Vert_{F}^{2}<C_{2}\sqrt{NT},$ which implies $\frac{1}{\sqrt{NT}}\left\Vert 
\tilde{\Delta}_{\Theta_{j}}\right\Vert_{F}=O_{p}\left((N\wedge
T)^{-1/2}\right), \quad \forall j\in [p]\cup \{0\}.$ It suffices to consider
the case that $\{\tilde{\Delta}_{\Theta_{j}}\}_{j=0}^{p}\in \mathcal{R}%
(3,C_{2})$.

Define 
\begin{eqnarray*}
\mathbb{Q}_{\tau }\left(\left\{ \Theta_{j}\right\}_{j=0}^{p}\right) &=&\frac{%
1}{NT}\sum_{i=1}^{N}\sum_{t=1}^{T}\rho_{\tau
}\left(Y_{it}-\Theta_{0,it}-\sum_{j=1}^{p}X_{j,it}\Theta_{j,it}\right) ,%
\text{ and} \\
Q_{\tau }(\left\{ \Theta_{j}\right\}_{j=0}^{p}) &=&\frac{1}{NT}%
\sum_{i=1}^{N}\sum_{t=1}^{T}\mathbb{E}\left[ \rho_{\tau
}\left(Y_{it}-\Theta_{0,it}-\sum_{j=1}^{p}X_{j,it}\Theta_{j,it}\right) \bigg
|\mathscr{G}_{i,t-1}\right] ,
\end{eqnarray*}%
where $\mathscr{G}_{i,t-1}$ is defined in Assumption \ref{ass:1}. Then we
have 
\begin{align}
0& \geq \mathbb{Q}_{\tau }\left(\left\{ \Theta_{j}^{0}+\tilde{\Delta}%
_{\Theta_{j}}\right\}_{j=0}^{p}\right) -\mathbb{Q}_{\tau
}\left(\left\{\Theta_{j}^{0}\right\}_{j=0}^{p}\right)
+\sum_{j=0}^{p}\nu_{j}\left(\left\Vert \Theta_{j}^{0}+\tilde{\Delta}%
_{\Theta_{j}}\right\Vert_{\ast}-\left\Vert \Theta_{j}^{0}\right\Vert_{\ast
}\right)  \notag  \label{A.1} \\
& =\left\{ \mathbb{Q}_{\tau }\left(\left\{ \Theta_{j}^{0}+\tilde{\Delta}%
_{\Theta_{j}}\right\}_{j=0}^{p}\right) -\mathbb{Q}_{\tau
}\left(\left\{\Theta_{j}^{0}\right\}_{j=0}^{p}\right) -\left[ Q_{\tau
}\left(\left\{\Theta_{j}^{0}+\tilde{\Delta}_{\Theta_{j}}\right\}_{j=0}^{p}%
\right)-Q_{\tau }\left(\left\{ \Theta_{j}^{0}\right\}_{j=0}^{p}\right) %
\right]\right\}  \notag \\
& +\left[ Q_{\tau }\left(\left\{ \Theta_{j}^{0}+\tilde{\Delta}%
_{\Theta_{j}}\right\}_{j=0}^{p}\right) -Q_{\tau }\left(\left\{
\Theta_{j}^{0}\right\}_{j=0}^{p}\right) \right] +\sum_{j=0}^{p}\nu_{j}\left(%
\left\Vert \Theta_{j}^{0}+\tilde{\Delta}_{\Theta_{j}}\right\Vert_{\ast}-%
\left\Vert \Theta_{j}^{0}\right\Vert_{\ast }\right) ,
\end{align}%
where the first inequality holds by the definition of the estimator. Noted
that 
\begin{equation}  \label{A.2}
\nu_{j}\left\vert \left\Vert \Theta_{j}^{0}+\tilde{\Delta}%
_{\Theta_{j}}\right\Vert_{\ast }-\left\Vert
\Theta_{j}^{0}\right\Vert_{\ast}\right\vert \leq \nu_{j}\left\Vert \tilde{%
\Delta}_{\Theta_{j}}\right\Vert_{\ast }\leq
c_{8}\nu_{j}\sum_{j=0}^{p}\left\Vert \tilde{\Delta}_{\Theta_{j}}\right%
\Vert_{F}
\end{equation}
where the first inequality is due to triangle inequality and the second
inequality holds by Lemma \ref{Lem:delta nuclear} with positive constant $%
c_{8}$ defined in the lemma.

Define 
\begin{align}
& \rho_{it}\left(\left\{
\Delta_{\Theta_{j},it},X_{j,it}\right\}_{j=0}^{p},\epsilon_{it}\right)
=\rho_{\tau
}\left(\epsilon_{it}-\Delta_{0,it}-\sum_{j=1}^{p}X_{j,it}\Delta_{%
\Theta_{j},it}\right) -\rho_{\tau}\left(\epsilon_{it}\right) ,
\label{eq:rho} \\
& \bar{\rho}_{it}\left(\left\{
\Delta_{\Theta_{j},it},X_{j,it}\right\}_{j=0}^{p},\epsilon_{it}\right) =%
\mathbb{E}\left[ \rho_{\tau
}\left(\epsilon_{it}-\Delta_{0,it}-\sum_{j=1}^{p}X_{j,it}\Delta_{%
\Theta_{j},it}\right) -\rho_{\tau }\left(\epsilon_{it}\right) \bigg |%
\mathscr{G}_{i,t-1}\right] ,  \notag \\
& \tilde{\rho}_{it}\left(\left\{
\Delta_{\Theta_{j},it},X_{j,it}\right\}_{j=0}^{p},\epsilon_{it}\right)
=\rho_{it}\left(\left\{
\Delta_{\Theta_{j},it},X_{j,it}\right\}_{j=0}^{p},\epsilon_{it}\right) -\bar{%
\rho}_{it}\left(\left\{
\Delta_{\Theta_{j},it},X_{j,it}\right\}_{j=0}^{p},\epsilon_{it}\right) ,
\label{eq:rhotilde} \\
& \mathscr{A}_{1}=\left\{ \sup_{\left\{
\Delta_{\Theta_{j}}\right\}_{j=0}^{p}\in \mathcal{R}(3,C_{2})}\frac{%
\left\vert \frac{1}{NT}\sum_{i=1}^{N}\sum_{t=1}^{T}\tilde{\rho}%
_{it}\left(\left\{
\Delta_{\Theta_{j},it},X_{j,it}\right\}_{j=0}^{p},\epsilon_{it}\right)
\right\vert }{\sum_{j=0}^{p}\left\Vert \Delta_{\Theta_{j}}\right\Vert_{F}}%
\leq C_{5}a_{NT}\right\} ,  \notag
\end{align}%
with $a_{NT}=\frac{\sqrt{\left(N\vee T\right) \log \left(N\vee T\right) }}{NT%
}$ for some positive constant $C_{5}$, and $\mathscr{A}_{1}^{c}$ as the
complement of $\mathscr{A}_{1}$.

On $\mathscr{A}_{1}$, following (\ref{A.1}), we have w.p.a.1 
\begin{align*}
0& \geq \left[ Q_{\tau }\left(\left\{ \Theta_{j}^{0}+\tilde{\Delta}%
_{\Theta_{j}}\right\}_{j=0}^{p}\right)-Q_{\tau }\left(\left\{
\Theta_{j}^{0}\right\}_{j=0}^{p}\right) \right] \\
&-\left\vert \mathbb{Q}_{\tau }\left(\left\{ \Theta_{j}^{0}+\tilde{\Delta}%
_{\Theta_{j}}\right\}_{j=0}^{p}\right) -\mathbb{Q}_{\tau
}\left(\left\{\Theta_{j}^{0}\right\}_{j=0}^{p}\right) -\left[ Q_{\tau
}\left(\left\{\Theta_{j}^{0}+\tilde{\Delta}_{\Theta_{j}}\right\}_{j=0}^{p}%
\right)-Q_{\tau }\left(\left\{ \Theta_{j}^{0}\right\}_{j=0}^{p}\right) %
\right]\right\vert \\
&-\sum_{j=0}^{p}\nu_{j}\left\vert \left\Vert \Theta_{j}^{0}+\tilde{\Delta}%
_{\Theta_{j}}\right\Vert_{\ast }-\left\Vert \Theta_{j}^{0}\right\Vert_{\ast
}\right\vert \\
& \geq \frac{c_{7}C_{3}}{NT\xi_{N}^{2}}\sum_{j=0}^{p}\left\Vert\tilde{\Delta}%
_{\Theta_{j}}\right\Vert_{F}^{2}-\frac{c_{7}C_{4}}{NT\xi_{N}^{2}}%
\left(N+T\right)-\frac{C_{5}\sum_{j=0}^{p}\left\Vert \tilde{\Delta}%
_{\Theta_{j}}\right\Vert_{F}\sqrt{\left(N\vee T\right) \log \left(N\vee
T\right) }}{NT} \\
& -c_{8}\sum_{j=0}^{p}\nu_{j}\sum_{j=0}^{p}\left\Vert\tilde{\Delta}%
_{\Theta_{j}}\right\Vert_{F},
\end{align*}%
where the first inequality is by triangle inequality, the second inequality
holds by (\ref{A.2}) and Lemmas \ref{Lem:exp lower} and \ref{Lem:empirical
process}. It follows that 
\begin{align*}
& \frac{c_{7}C_{3}}{NT\xi_{N}^{2}}\sum_{j=0}^{p}\left\Vert \tilde{\Delta}%
_{\Theta_{j}}\right\Vert_{F}^{2}-\frac{c_{8}\left(p+1\right) c_{0}\sqrt{%
N\vee \left(T\log T\right) }+C_{5}\sqrt{\left(N\vee T\right) \log
\left(N\vee T\right) }}{NT}\sum_{j=0}^{p}\left\Vert \tilde{\Delta}%
_{\Theta_{j}}\right\Vert_{F} \\
& -\frac{c_{7}C_{4}}{NT\xi_{N}^{2}}\left(N+T\right) \leq 0,
\end{align*}%
which implies 
\begin{equation*}
\frac{1}{\sqrt{NT}}\left\Vert \tilde{\Delta}_{\Theta_{j}}\right\Vert_{F}=O%
\left(\frac{\sqrt{\log (N\vee T)}\xi_{N}^{2}}{\sqrt{N\wedge T}}\right)
\end{equation*}%
under the event $\mathscr{A}_{1}$.

By Lemma \ref{Lem:empirical process}, for any $\delta >0$, we can choose a
sufficiently large $C_{5}$ such that $\mathbb{P}\left\{ \mathscr{A}%
_{1}^{c}\right\} \leq \delta $. This implies 
\begin{equation*}
\frac{1}{\sqrt{NT}}\left\Vert \tilde{\Theta}_{j}-\Theta_{j}^{0}\right%
\Vert_{F}=O_{p}\left(\frac{\sqrt{\log (N\vee T)}\xi_{N}^{2}}{\sqrt{N\wedge T}%
}\right) ,\quad \forall j\in \left\{ 0,\cdots ,p\right\} .\quad \blacksquare
\end{equation*}

\subsubsection{Proof of Statement (ii)}

With the statement (i), the second statement holds by the Weyl's inequality. 
$\blacksquare$

\subsubsection{Proof of Statement (iii)}

For $\forall j\in \left\{ 0,\cdots ,p\right\} $, let $\tilde{D}_{j}=\frac{1}{%
NT}\tilde{\Theta}_{j}^{\prime }\tilde{\Theta}_{j}=\hat{\tilde{\mathcal{V}}}%
_{j}\hat{\tilde{\Sigma}}_{j}\hat{\tilde{\mathcal{V}}}_{j}^{\prime }$, and
recall that $D_{j}^{0}=\frac{1}{NT}\Theta _{j}^{0\prime }\Theta _{j}^{0}=%
\mathcal{V}_{j}^{0}\Sigma _{j}^{0}\mathcal{V}_{j}^{0\prime }$. Define the
event $\mathscr{A}_{2}(M)=\left\{ \frac{1}{\sqrt{NT}}\left\Vert \tilde{\Theta%
}_{j}-\Theta _{j}^{0}\right\Vert _{F}\leq M\eta _{N},\forall j\in \left\{
0,\cdots ,p\right\} \right\} $ with $\eta _{N}=\frac{\sqrt{\log (N\vee T)}%
\xi _{N}^{2}}{\sqrt{N\wedge T}}$. On event $\mathscr{A}_{2}(M)$, for some
positive constant $C_{6}$, 
\begin{equation*}
\left\Vert \tilde{D}_{j}-D_{j}^{0}\right\Vert _{F}^{2}\leq C_{6}\eta _{N}.
\end{equation*}%
By Lemma C.1 of \cite{su2020strong} and Davis-Kahan sin$\Theta $ theorem in 
\cite{yu2015useful}, there exists an orthogonal rotation matrix $O_{j}$ such
that 
\begin{equation*}
\left\Vert \mathcal{V}_{j}^{0}-\hat{\tilde{\mathcal{V}}}_{j}O_{j}\right\Vert
_{F}\leq \sqrt{K_{j}}\left\Vert \mathcal{V}_{j}^{0}-\hat{\tilde{\mathcal{V}}}%
_{j}O_{j}\right\Vert _{op}\leq \sqrt{K_{j}}\frac{\sqrt{2}C_{6}\eta _{N}}{%
\Sigma _{K_{j},1}^{2}-C_{6}\eta _{N}}\leq \sqrt{K_{j}}\frac{\sqrt{2}%
C_{6}\eta _{N}}{c_{\sigma }^{2}-C_{6}\eta _{N}}\leq \sqrt{K_{j}}\frac{\sqrt{2%
}C_{6}\eta _{N}}{C_{7}c_{\sigma }^{2}}\leq C_{8}\eta _{N},
\end{equation*}%
for $C_{8}=\frac{\sqrt{2}C_{6}\sqrt{\bar{K}}}{C_{7}c_{\sigma }^{2}}$. The
second last inequality holds with some positive constant $C_{7}$ and the
fact that $\eta _{N}$ is sufficiently small.

Then $\left\Vert V_{j}^{0}-\tilde{V}_{j}O_{j}\right\Vert_{F}\leq C_{8}\sqrt{T%
}\eta_{N}$ by the definition of $\tilde{V}_{j}$ and $V_{j}$. Let $\mathscr{A}%
_{2}^{c}(M)$ be the complement of event $\mathscr{A}_{2}(M)$. Combining the
fact that $\mathbb{P}\left\{ \mathscr{A}_{2}^{c}(M)\right\}\rightarrow 0$,
it implies $\left\Vert V_{j}^{0}-\tilde{V}_{j}O_{j}\right\Vert_{F}=O_{p}%
\left(\sqrt{T}\eta_{N}\right) .\quad\blacksquare $

\subsection{Proof of Theorem \protect\ref{Thm2}}

\subsubsection{Proof of Statement (i)}

We prove that $\max_{i\in I_{2}}\left\Vert O_{j}^{(1)\prime}\dot{u}_{i,j}^{(1)}-u_{i,j}^{0}\right\Vert_{2}=O_{p}(\eta_{N})$ and the $\max_{i\in I_{3}}\left\Vert O_{j}^{(1)\prime}\dot{u}_{i,j}^{(1)}-u_{i,j}^{0}\right\Vert_{2}=O_{p}(\eta_{N})$ can be derived in the same manner once statement (ii) is satisfied. Define 
\begin{align*}
& \tilde{\mathbb{Q}}_{\tau i}\left(\left\{ u_{i,j}\right\}_{j\in [p]\cup
\{0\}}\right) =\frac{1}{T}\sum_{t=1}^{T}\rho_{\tau
}\left(Y_{it}-u_{i,0}^{\prime}\tilde{v}_{t,0}^{(1)}-\sum_{j=1}^{p}u_{i,j}^{%
\prime}\tilde{v}_{t,j}^{(1)}X_{j,it}\right) , \\
& u_{i}^{0}=\left[ u_{i,0}^{0\prime},\cdots ,u_{i,p}^{0\prime}\right]
^{\prime},\quad \dot{\Delta}_{i,j}=O_{j}^{(1)\prime}\dot{u}%
_{i,j}^{(1)}-u_{i,j}^{0},\quad \dot{\Delta}_{i,u}=\left(\dot{\Delta}%
_{i,0}^{\prime},\cdots ,\dot{\Delta}_{i,p}^{\prime}\right) ^{\prime}, \\
& \tilde{\Phi}_{it}^{(1)}=\left[ \left(O_{0}^{(1)\prime}\tilde{v}%
_{t,0}^{(1)}\right) ^{\prime},\left(O_{1}^{(1)\prime}\tilde{v}%
_{t,1}^{(1)}X_{1,it}\right) ^{\prime},\cdots ,\left(O_{p}^{(1)\prime}\tilde{v%
}_{t,p}^{(1)}X_{p,it}\right) ^{\prime}\right] ^{\prime},\quad \tilde{\Phi}%
_{i}^{(1)}=\frac{1}{T}\sum_{t=1}^{T}\tilde{\Phi}_{it}^{(1)}\tilde{\Phi}%
_{it}^{(1)\prime}, \\
& w_{1,it}=Y_{it}-\left(O_{0}^{(1)}u_{i0}^{0}\right) ^{\prime}\tilde{v}%
_{t,0}^{(1)}-\sum_{j=1}^{p}\left(O_{j}^{(1)}u_{i,j}^{0}\right) ^{\prime}%
\tilde{v}_{t,j}^{(1)}X_{j,it}=Y_{it}-u_{i}^{0\prime}\tilde{\Phi}%
_{it}^{(1)}=\epsilon_{it}-u_{i}^{0\prime}(\tilde{\Phi}_{it}^{(1)}-%
\Phi_{it}^{0}),
\end{align*}%
and for $i\in I_{2}$, recall that $\mathscr{D}_{e_{i}}^{I_{1}}$ is the $%
\sigma$-field generated by 
\begin{equation*}
\left\{ \epsilon_{i^{*}t},e_{i^{*}t}\right\}_{i^{*}\in I_{1},t\in[T]%
}\bigcup\left\{ e_{it}\right\}_{t\in[T]}\bigcup \left\{
V_{j}^{0}\right\}_{j\in [p]\cup \{0\}}\bigcup \left\{
W_{j}^{0}\right\}_{j\in [ p]}.
\end{equation*}

By construction, we have 
\begin{align}
0& \geq \dot{Q}_{\tau i,u}\left(\left\{ \dot{u}_{i,j}^{(1)}\right\}_{j\in[p]%
\cup \{0\}}\right) -\dot{Q}_{\tau
i,u}\left(\left\{O_{j}^{(1)}u_{i,j}^{0}\right\}_{j\in [p]\cup \{0\}}\right) 
\notag  \label{A:obj diff} \\
& =\frac{1}{T}\sum_{t=1}^{T}\rho_{\tau }\left(Y_{it}-\dot{u}%
_{i,0}^{(1)\prime}O_{0}^{(1)}O_{0}^{(1)\prime}\tilde{v}_{t,0}^{(1)}-%
\sum_{j=1}^{p}\dot{u}_{i,j}^{(1)\prime}O_{j}^{(1)}O_{j}^{(1)\prime}\tilde{v}%
_{t,j}^{(1)}X_{j,it}\right)  \notag \\
& -\frac{1}{T}\sum_{t=1}^{T}\rho_{\tau
}\left(Y_{it}-u_{i,0}^{0\prime}O_{0}^{(1)\prime}\tilde{v}_{t,0}^{(1)}-%
\sum_{j=1}^{p}u_{i,j}^{0\prime}O_{j}^{(1)\prime}\tilde{v}%
_{t,j}^{(1)}X_{j,it}\right)  \notag \\
& =\frac{1}{T}\sum_{t=1}^{T}\left[ \tilde{\Phi}_{it}^{(1)\prime}\dot{\Delta}%
_{i,u}\left(\tau -\mathbf{1}\left\{ w_{1,it}\leq 0\right\} \right) \right] +%
\frac{1}{T}\sum_{t=1}^{T}\int_{0}^{\tilde{\Phi}_{it}^{(1)\prime}\dot{\Delta}%
_{i,u}}\left(\mathbf{1}\left\{ w_{1,it}\leq s\right\} -\mathbf{1}%
\left\{w_{1,it}\leq 0\right\} \right) ds  \notag \\
& =\frac{1}{T}\sum_{t=1}^{T}\mathbb{E}\left[ \tilde{\Phi}_{it}^{(1)\prime}%
\left(\tau -\mathbf{1}\left\{ w_{1,it}\leq 0\right\} \right) \bigg |%
\mathscr{D}_{e_{i}}^{I_{1}}\right] \dot{\Delta}_{i,u}  \notag \\
& +\left\{ \frac{1}{T}\sum_{t=1}^{T}\left[ \tilde{\Phi}_{it}^{(1)\prime}%
\left(\tau -\mathbf{1}\left\{ w_{1,it}\leq 0\right\} \right) -\mathbb{E}%
\left(\tilde{\Phi}_{it}^{(1)\prime}\left(\tau -\mathbf{1}\left\{w_{1,it}\leq
0\right\} \right) \bigg |\mathscr{D}_{e_{i}}^{I_{1}}\right) \right] \right\} 
\dot{\Delta}_{i,u}  \notag \\
& +\frac{1}{T}\sum_{t=1}^{T}\int_{0}^{\tilde{\Phi}_{it}^{(1)\prime}\dot{%
\Delta}_{i,u}}\left(\mathbf{1}\left\{ \epsilon_{it}\leq s\right\} -\mathbf{1}%
\left\{ \epsilon_{it}\leq 0\right\} \right) -\mathbb{E}\left(\mathbf{1}%
\left\{ \epsilon_{it}\leq s\right\} -\mathbf{1}\left\{ \epsilon_{it}\leq
0\right\} \bigg |\mathscr{D}_{e_{i}}^{I_{1}}\right) ds  \notag \\
& +\frac{1}{T}\sum_{t=1}^{T}\int_{0}^{\tilde{\Phi}_{it}^{(1)\prime}\dot{%
\Delta}_{i,u}}\mathbb{E}\left(\mathbf{1}\left\{ \epsilon_{it}\leq s\right\} -%
\mathbf{1}\left\{ \epsilon_{it}\leq 0\right\} \bigg |\mathscr{D}%
_{e_{i}}^{I_{1}}\right) ds  \notag \\
& +\frac{1}{T}\sum_{t=1}^{T}\int_{0}^{\tilde{\Phi}_{it}^{(1)\prime}\dot{%
\Delta}_{i,u}}\left(\mathbf{1}\left\{ w_{1,it}\leq s\right\} -\mathbf{1}%
\left\{\epsilon_{it}\leq s\right\} \right) -\mathbb{E}\left(\mathbf{1}%
\left\{w_{1,it}\leq s\right\} -\mathbf{1}\left\{ \epsilon_{it}\leq s\right\} %
\bigg|\mathscr{D}_{e_{i}}^{I_{1}}\right) ds  \notag \\
& +\frac{1}{T}\sum_{t=1}^{T}\int_{0}^{\tilde{\Phi}_{it}^{(1)\prime}\dot{%
\Delta}_{i,u}}\mathbb{E}\left(\mathbf{1}\left\{ w_{1,it}\leq s\right\} -%
\mathbf{1}\left\{ \epsilon_{it}\leq s\right\} \bigg |\mathscr{D}%
_{e_{i}}^{I_{1}}\right) ds  \notag \\
& +\frac{1}{T}\sum_{t\in[T]}\left[ \tilde{\Phi}_{it}^{(1)\prime}\dot{\Delta}%
_{i,u}\left(\mathbf{1}\left\{ \epsilon_{it}\leq 0\right\} -\mathbf{1}\left\{
w_{1,it}\leq 0\right\} \right) \right]  \notag \\
& :=\sum_{m=1}^{7}A_{m,i},
\end{align}%
where the first inequality holds by the definition of the estimator and the
second equality holds by Knight's identity in \cite{knight1998limiting}
which states that 
\begin{equation*}
\rho_{\tau }(u-v)-\rho_{\tau }(u)=v(\tau -\mathbf{1}\left\{ u\leq 0\right\}
)+\int_{0}^{v}\left(\mathbf{1}\left\{ u\leq s\right\} -\mathbf{1}\{u\leq
0\}\right) ds.
\end{equation*}%
After simple manipulation, we have 
\begin{align*}
\left\vert A_{4,i}\right\vert &
=A_{4,i}\leq-A_{1,i}-A_{2,i}-A_{3,i}-A_{5,i}-A_{6,i}-A_{7,i} \\
& \leq \left\vert A_{1,i}\right\vert +\left\vert
A_{2,i}\right\vert+\left\vert A_{3,i}\right\vert +\left\vert
A_{5,i}\right\vert +\left\vert A_{6,i}\right\vert +\left\vert
A_{7,i}\right\vert .
\end{align*}

Define, for some constant $M$, an event set 
\begin{equation*}
\mathscr{A}_{3}(M)=\left\{ \max_{i\in I_{2}}\left(|A_{m,i}|/\left\Vert \dot{%
\Delta}_{i,u}\right\Vert_{2}\right) \leq M\eta_{N},\quad
m=1,2,3,5,6,7\right\}
\end{equation*}%
and 
\begin{equation}
q_{i}^{I}=\inf_{\Delta }\frac{\left[ \frac{1}{T}\sum_{t\in[T]}\left(\tilde{%
\Phi}_{it}^{(1)\prime}\Delta \right) ^{2}\right] ^{\frac{3}{2}}}{\frac{1}{T}%
\sum_{t\in[T]}\left\vert \tilde{\Phi}_{it}^{(1)\prime}\Delta\right\vert ^{3}}%
.  \label{eq:qiI}
\end{equation}%
Then, under $\mathscr{A}_{3}(M)$, we have 
\begin{align}
M\eta_{N}\left\Vert \dot{\Delta}_{i,u}\right\Vert_{2}& \geq |A_{4,i}|  \notag
\\
& \geq \min \left(\frac{\left(3c_{11}^{2}\underline{\mathfrak{f}}-c_{11}^{3}%
\bar{\mathfrak{f}}^{\prime }\right) c_{\phi }\max_{i\in I_{2}}\left\Vert 
\dot{\Delta}_{i,u}\right\Vert_{2}^{2}}{12},\frac{\left(3c_{11}^{2}\underline{%
\mathfrak{f}}-c_{11}^{3}\bar{\mathfrak{f}}^{\prime}\right) \sqrt{c_{\phi }}%
q_{i}^{I}\max_{i\in I_{2}}\left\Vert \dot{\Delta}_{i,u}\right\Vert_{2}}{6%
\sqrt{2}}\right) ,  \label{A:less case}
\end{align}%
where $c_{11}<\min (\frac{3\underline{\mathfrak{f}}}{\bar{\mathfrak{f}}%
^{\prime }},1)$ and the second inequality holds by Lemma \ref{Lem:A}. 
In addition, note that 
\begin{align*}
& \max_{i\in I_{2}}\frac{1}{T}\sum_{t\in[T]}\left\Vert \tilde{\Phi}%
_{it}^{(1)}\right\Vert_{2}^{3}=\max_{i\in I_{2}}\frac{1}{T}\sum_{t\in[T]%
}\left(\left\Vert \tilde{v}_{t,0}^{(1)}\right\Vert_{2}^{2}+\sum_{j\in
[p]}\left\Vert \tilde{v}_{t,j}^{(1)}X_{j,it}\right\Vert_{2}^{2}\right) ^{3/2}
\\
& \leq \max_{i\in I_{2}}\frac{1}{T}\sum_{t\in[T]}\left[ \left(\frac{2M}{%
c_{\sigma }}\right) ^{2}\left(1+\sum_{j\in [p]}X_{j,it}^{2}\right) \right]
^{3/2}\leq C_{9},\quad \text{a.s.}
\end{align*}%
with $C_{9}$ being a positive constant, where the first inequality holds by
Lemma \ref{Lem:bounded u&v_tilde}(ii) and the second inequality is by
Assumption \ref{ass:1}(iv). Then we have by Lemma \ref{Lem:phi eig} 
\begin{equation*}
q_{i}^{I}\geq \inf_{\Delta }\frac{\left\Vert \Delta \right\Vert_{2}^{3}\left[
\lambda_{\min }\left(\tilde{\Phi}_{i}^{(1)}\right) \right] ^{2/3}}{%
\left\Vert \Delta \right\Vert_{2}^{3}\frac{1}{T}\sum_{t\in [T]}\left\Vert 
\tilde{\Phi}_{it}^{(1)}\right\Vert_{2}^{3}}\geq \frac{\min_{i\in I_{2}}\left[
\lambda_{\min }\left(\tilde{\Phi}_{i}^{(1)}\right) \right] ^{2/3}}{%
\max_{i\in I_{2}}\frac{1}{T}\sum_{t\in[T]}\left\Vert \tilde{\Phi}%
_{it}^{(1)}\right\Vert_{2}^{3}}>\frac{\left(c_{\phi }/2\right)^{2/3}}{C_{9}},
\end{equation*}%
which implies 
\begin{equation*}
\frac{\left(3c_{11}^{2}\underline{\mathfrak{f}}-c_{11}^{3}\bar{\mathfrak{f}}%
^{\prime }\right) \sqrt{c_{\phi }}q_{i}^{I}\max_{i\in I_{2}}\left\Vert \dot{%
\Delta}_{i,u}\right\Vert_{2}}{6\sqrt{2}}>C_{10}\max_{i\in I_{2}}\left\Vert 
\dot{\Delta}_{i,u}\right\Vert_{2}\eta_{N},
\end{equation*}%
as $C_{10}$ is defined to be the positive constant and $\eta_{N}=o(1)$.
Combining this with (\ref{A:less case}), we have 
\begin{equation*}
\max_{i\in I_{2}}\left\Vert O_{j}^{(1)\prime}\dot{u}_{i,j}^{(1)}-u_{i,j}^{0}%
\right\Vert_{2}\leq \left\Vert \dot{\Delta}_{i,u}\right\Vert_{2}\leq
M^{\prime }\eta_{N}
\end{equation*}%
for some constant $M^{\prime }$ which may depends on $M$. In addition, for
an arbitrary constant $e>0$, we can find a sufficiently large constant $M$
such that $\mathbb{P}(\mathscr{A}_{2}^{c}(M))\leq e$, which implies $%
\max_{i\in I_{2}}\left\Vert O_{j}^{(1)\prime}\dot{u}_{i,j}^{(1)}-u_{i,j}^{0}%
\right\Vert_{2}=O_{p}(\eta_{N}) .\quad \blacksquare $

\subsubsection{Proof of Statement (ii)}

Differently from the proof in the previous subsection, owing to the
dependence of $\dot{u}_{i,j}^{(1)}$ and $\epsilon_{it}$, we can not directly
use conditional exponential inequality. In this subsection, we will show how
to handle this dependence in detail. Recall that $\mathscr{D}_{e_{it}}$ is
the $\sigma$-field generated by $\{e_{j,it}\}_{j\in [p]}\cup \left\{
V_{j}^{0}\right\}_{j\in [p]\cup \{0\}}\cup \left\{W_{j}^{0}\right\}_{j\in
[p]}$ and define 
\begin{align*}
& \dot{Q}_{\tau t,v}\left(\left\{ v_{t,j}\right\}_{j\in [p]\cup\{0\}}\right)
=\frac{1}{N_{2}}\sum_{i\in I_{2}}\rho_{\tau }\left(Y_{it}-v_{t,0}^{\prime}%
\dot{u}_{i,0}^{(1)}-\sum_{j=1}^{p}v_{t,j}^{\prime}\dot{u}%
_{i,j}^{(1)}X_{j,it}\right) , \\
& v_{t}^{0}=\left(v_{t,0}^{0\prime},\cdots ,v_{t,p}^{0\prime}\right)
^{\prime},\quad \dot{\Delta}_{t,j}=O_{j}^{(1)\prime}\dot{v}%
_{t,j}^{(1)}-v_{t,j}^{0},\quad \dot{\Delta}_{t,v}=\left(\dot{\Delta}%
_{t,0}^{\prime},\cdots ,\dot{\Delta}_{t,p}^{\prime}\right) ^{\prime}, \\
& \dot{\Psi}_{it}^{(1)}=\left[ \left(O_{0}^{(1)\prime}\dot{u}%
_{i,0}^{(1)}\right) ^{\prime},\left(O_{1}^{(1)\prime}\dot{u}%
_{i,1}^{(1)}X_{1,it}\right) ^{\prime},...,\left(O_{p}^{(1)\prime}\dot{u}%
_{i,p}^{(1)}X_{p,it}\right) ^{\prime}\right] ^{\prime}, \\
& \dot{\Psi}_{t}^{(1)}=\frac{1}{N_{2}}\sum_{i\in I_{2}}\dot{\Psi}%
_{it}^{(1)}\dot{\Psi}_{it}^{(1)\prime}.
\end{align*}%
As in (\ref{A:obj diff}), we have 
\begin{align}
0& \geq \dot{Q}_{\tau t,v}\left(\left\{ \dot{v}_{t,j}^{(1)}\right\}_{j\in[p]%
\cup \{0\}}\right) -\dot{Q}_{\tau
t,v}\left(\left\{O_{j}^{(1)}v_{t,j}^{0}\right\}_{j\in [p]\cup \{0\}}\right) 
\notag  \label{B:obj diff} \\
& =\frac{1}{N_{2}}\sum_{i\in I_{2}}\rho_{\tau }\left(Y_{it}-\dot{v}%
_{t,0}^{(1)\prime}O_{0}^{(1)}O_{0}^{(1)\prime}\dot{u}_{i,0}^{(1)}-%
\sum_{j=1}^{p}\dot{v}_{i,j}^{(1)\prime}O_{j}^{(1)}O_{j}^{(1)\prime}\dot{u}%
_{t,j}^{(1)}X_{j,it}\right)  \notag \\
& -\frac{1}{N_{2}}\sum_{i\in I_{2}}\rho_{\tau
}\left(Y_{it}-v_{t,0}^{0\prime}O_{0}^{(1)\prime}\dot{u}_{i,0}^{(1)}-%
\sum_{j=1}^{p}v_{t,j}^{0\prime}O_{j}^{(1)\prime}\dot{u}_{t,j}^{(1)}X_{j,it}%
\right)  \notag \\
& =\frac{1}{N_{2}}\sum_{i\in I_{2}}\left[ \dot{\Psi}_{it}^{(1)\prime}\dot{%
\Delta}_{t,v}\left(\tau -\mathbf{1}\left\{ w_{3,it}\leq 0\right\} \right) %
\right] +\frac{1}{N_{2}}\sum_{i\in I_{2}}\int_{0}^{\dot{\Psi}%
_{it}^{(1)\prime}\dot{\Delta}_{t,v}}\left(\mathbf{1}\left\{ w_{3,it}\leq
s\right\} -\mathbf{1}\left\{ w_{3,it}\leq 0\right\} \right) ds  \notag \\
& =\frac{1}{N_{2}}\sum_{i\in I_{2}}\left[ \Psi_{it}^{0\prime}\dot{\Delta}%
_{t,v}\left(\tau -\mathbf{1}\left\{ \epsilon_{it}\leq 0\right\} \right) %
\right] +\frac{1}{N_{2}}\sum_{i\in I_{2}}\left[ \left(\dot{\Psi}%
_{it}^{(1)}-\Psi_{it}^{0}\right) ^{\prime}\dot{\Delta}_{t,v}\left(\tau -%
\mathbf{1}\left\{ \epsilon_{it}\leq 0\right\} \right) \right]  \notag \\
& +\frac{1}{N_{2}}\sum_{i\in I_{2}}\left[ \dot{\Psi}_{it}^{(1)\prime}\dot{%
\Delta}_{t,v}\left(\mathbf{1}\left\{ \epsilon_{it}\leq 0\right\} -\mathbf{1}%
\left\{ w_{3,it}\leq 0\right\} \right) \right]  \notag \\
& +\frac{1}{N_{2}}\sum_{i\in I_{2}}\int_{0}^{\dot{\Psi}_{it}^{(1)\prime}\dot{%
\Delta}_{t,v}}\mathbb{E}\left[ \left(\mathbf{1}\left\{ \epsilon_{it}\leq
s\right\} -\mathbf{1}\left\{ \epsilon_{it}\leq 0\right\} \right) ds\bigg |%
\mathscr{D}_{e_{it}}\right]  \notag \\
& +\frac{1}{N_{2}}\sum_{i\in I_{2}}\left\{ \int_{0}^{\dot{\Psi}%
_{it}^{(1)\prime}\dot{\Delta}_{t,v}}\left[ \left(\mathbf{1}\left\{
\epsilon_{it}\leq s\right\} -\mathbf{1}\left\{ \epsilon_{it}\leq 0\right\}
\right) -\mathbb{E}\left(\left(\mathbf{1}\left\{ \epsilon_{it}\leq s\right\}
-\mathbf{1}\left\{ \epsilon_{it}\leq 0\right\} \right) ds\bigg |\mathscr{D}%
_{e_{it}}\right) \right] \right\}  \notag \\
& +\frac{1}{N_{2}}\sum_{i\in I_{2}}\int_{0}^{\dot{\Psi}_{it}^{(1)\prime}\dot{%
\Delta}_{t,v}}\left(\mathbf{1}\left\{ w_{3,it}\leq s\right\} -\mathbf{1}%
\left\{ \epsilon_{it}\leq s\right\} \right) ds  \notag \\
& +\frac{1}{N_{2}}\sum_{i\in I_{2}}\int_{0}^{\dot{\Psi}_{it}^{(1)\prime}\dot{%
\Delta}_{t,v}}\left(\mathbf{1}\left\{ \epsilon_{it}\leq 0\right\}-\mathbf{1}%
\left\{ w_{3,it}\leq 0\right\} \right) ds  \notag \\
& =\frac{1}{N_{2}}\sum_{i\in I_{2}}\left[ \Psi_{it}^{0\prime}\dot{\Delta}%
_{t,v}\left(\tau -\mathbf{1}\left\{ \epsilon_{it}\leq 0\right\} \right) %
\right] +\frac{1}{N_{2}}\sum_{i\in I_{2}}\left[ \left(\dot{\Psi}%
_{it}^{(1)}-\Psi_{it}^{0}\right) ^{\prime}\dot{\Delta}_{t,v}\left(\tau -%
\mathbf{1}\left\{ \epsilon_{it}\leq 0\right\} \right) \right]  \notag \\
& +\frac{2}{N_{2}}\sum_{i\in I_{2}}\left[ \dot{\Psi}_{it}^{(1)\prime}\dot{%
\Delta}_{t,v}\left(\mathbf{1}\left\{ \epsilon_{it}\leq 0\right\} -\mathbf{1}%
\left\{ w_{3,it}\leq 0\right\} \right) \right]  \notag \\
& +\frac{1}{N_{2}}\sum_{i\in I_{2}}\int_{0}^{\dot{\Psi}_{it}^{(1)\prime}\dot{%
\Delta}_{t,v}}\mathbb{E}\left[ \left(\mathbf{1}\left\{ \epsilon_{it}\leq
s\right\} -\mathbf{1}\left\{ \epsilon_{it}\leq 0\right\} \right) \bigg |%
\mathscr{D}_{e_{it}}\right] ds  \notag \\
& +\frac{1}{N_{2}}\sum_{i\in I_{2}}\left\{ \int_{0}^{\dot{\Psi}%
_{it}^{(1)\prime}\dot{\Delta}_{t,v}}\left[ \left(\mathbf{1}\left\{
\epsilon_{it}\leq s\right\} -\mathbf{1}\left\{ \epsilon_{it}\leq 0\right\}
\right) -\mathbb{E}\left(\left(\mathbf{1}\left\{ \epsilon_{it}\leq s\right\}
-\mathbf{1}\left\{ \epsilon_{it}\leq 0\right\} \right) \bigg |\mathscr{D}%
_{e_{it}}\right) \right] ds\right\}  \notag \\
& +\frac{1}{N_{2}}\sum_{i\in I_{2}}\int_{0}^{\dot{\Psi}_{it}^{(1)\prime}\dot{%
\Delta}_{t,v}}\left(\mathbf{1}\left\{ w_{3,it}\leq s\right\} -\mathbf{1}%
\left\{ \epsilon_{it}\leq s\right\} \right) ds  \notag \\
& :=\sum_{m=1}^{6}B_{m,t}
\end{align}%
where $w_{3,it}=Y_{it}-v_{t,0}^{0\prime}O_{0}^{(1)\prime}\dot{u}%
_{i,0}^{(1)}-\sum_{j=1}^{p}v_{t,j}^{0\prime}O_{j}^{(1)\prime}\dot{u}%
_{i,j}^{(1)}X_{j,it}=Y_{it}-v_{t}^{0\prime}\dot{\Psi}_{it}^{(1)}=%
\epsilon_{it}-v_{t}^{0\prime}\left(\dot{\Psi}_{it}^{(1)}-\Psi_{it}^{0}%
\right) $, the last equality is by the fact that the third and the last
terms after the second equality are identical. Then we obtain $\left\vert
B_{4,t}\right\vert\leq \sum_{m\neq 4}\left\vert B_{m,t}\right\vert .$ By
Lemma \ref{Lem:B} and similar arguments for Theorem \ref{Thm2}(i), we obtain
that 
\begin{equation*}
\max_{t\in[T]}\left\Vert O_{j}^{(1)\prime}\dot{v}_{t,j}^{(1)}-v_{t,j}^{0}%
\right\Vert_{2}\leq \left\Vert\dot{\Delta}_{t,v}\right\Vert_{2}=O_{p}(%
\eta_{N}).
\end{equation*}
$\blacksquare $

\subsubsection{Proof of Statement (iii)}

In this proof we derive the linear expansion of $\{\dot{u}%
_{i,j}^{(1)}\}_{j\in [p]\cup \{0\}}$ for each $i\in I_{3}$. Recall that $%
\forall i\in I_{3}$, 
\begin{equation*}
\left\{ \dot{u}_{i,j}^{(1)}\right\}_{j\in [p]\cup \{0\}}=\argmin%
\limits_{\left\{ u_{i,j}\right\}_{j\in [p]\cup \{0\}}}\frac{1}{T}%
\sum_{t=1}^{T}\rho_{\tau }\left(Y_{it}-u_{i,0}^{\prime}\dot{v}%
_{t,0}^{(1)}-\sum_{j\in [p]}u_{i,j}^{\prime}\dot{v}_{t,j}^{(1)}X_{j,it}%
\right) .
\end{equation*}%
Define $\dot{\varpi}_{it}=\left(\dot{v}_{t,0}^{(1)\prime},X_{1,it}\dot{v}%
_{t,1}^{(1)\prime},\cdots ,X_{p,it}\dot{v}_{t,p}^{(1)\prime}\right)
^{\prime} $ and 
\begin{align*}
\dot{\mathbb{H}}_{i}\left(\{u_{i,j}\}_{j\in [p]\cup \{0\}}\right) & =\frac{1%
}{T}\sum_{t=1}^{T}\left[ \tau -\mathbf{1}\left\{ \epsilon_{it}\leq
g_{it}(\{u_{i,j}\}_{j\in [p]\cup \{0\}})\right\} \right] \dot{\varpi}_{it},
\\
\dot{\mathcal{H}}_{i}\left(\{u_{i,j}\}_{j\in [p]\cup \{0\}}\right) &=\frac{1%
}{T}\sum_{t=1}^{T}\mathbb{E}\left\{ \left[ \tau -\mathbf{1}%
\left\{\epsilon_{it}\leq g_{it}(\{u_{i,j}\}_{j\in [p]\cup \{0\}})\right\} %
\right] \dot{\varpi}_{it}\bigg |\mathscr{D}_{e_{i}}^{I_{1}\cup
I_{2}}\right\} , \\
& =\frac{1}{T}\sum_{t=1}^{T}\left\{ \left[ \tau -\mathfrak{F}%
_{it}\left(g_{it}(\{u_{i,j}\}_{j\in [p]\cup \{0\}})\right) \right] \dot{%
\varpi}_{it}\right\} , \\
\dot{\mathbb{W}}_{i}\left(\{\dot{u}_{i,j}^{(1)}\}_{j\in [p]\cup
\{0\}}\right) & =\dot{\mathbb{H}}_{i}\left(\{\dot{u}_{i,j}^{(1)}\}_{j\in[p]%
\cup \{0\}}\right) -\dot{\mathbb{H}}_{i}\left(\{O_{j}^{(1)}u_{i,j}^{0}\}_{j%
\in [p]\cup \{0\}}\right) \\
& -\left\{ \dot{\mathcal{H}}_{i}\left(\{\dot{u}_{i,j}^{(1)}\}_{j\in [p]\cup
\{0\}}\right) -\dot{\mathcal{H}}_{i}\left(\{O_{j}^{(1)}u_{i,j}^{0}\}_{j\in
[p]\cup \{0\}}\right) \right\} ,
\end{align*}
where $\mathscr{D}_{e_{i}}^{I_{1}\cup I_{2}}$ being the $\sigma$-field
generated by $\left\{ \epsilon_{i^{*}t},e_{i^{*}t}\right\}_{i^{*}\in
I_{1}\cup I_{2},t\in[T]}\bigcup \left\{ e_{it}\right\}_{t\in[T]}\bigcup
\left\{ V_{j}^{0}\right\}_{j\in [p]\cup \{0\}}\bigcup \left\{
W_{j}^{0}\right\}_{j\in [p]}$, and 
\begin{equation*}
g_{it}(\{u_{i,j}\}_{j\in [p]\cup \{0\}})=u_{i,0}^{\prime}\dot{v}%
_{t,0}^{(1)}+\sum_{j\in [p]}u_{i,j}^{\prime}\dot{v}%
_{t,j}^{(1)}X_{j,it}-u_{i,0}^{0\prime}v_{t,0}^{0}-\sum_{j\in
[p]}u_{i,j}^{0\prime}v_{t,j}^{0}X_{j,it}.
\end{equation*}%
Then we have 
\begin{align}
\dot{\mathbb{H}}_{i}\left(\{\dot{u}_{i,j}^{(1)}\}_{j\in [p]\cup\{0\}}\right)
& =\dot{\mathbb{H}}_{i}\left(\{O_{j}^{(1)}u_{i,j}^{0}\}_{j\in[p]\cup
\{0\}}\right) +\dot{\mathcal{H}}_{i}\left(\{\dot{u}_{i,j}^{(1)}\}_{j\in
[p]\cup \{0\}}\right) -\dot{\mathcal{H}}_{i}\left(\{O_{j}^{(1)}u_{i,j}^{0}%
\}_{j\in [p]\cup \{0\}}\right)  \notag  \label{B.25} \\
& +\dot{\mathbb{W}}_{i}\left(\{\dot{u}_{i,j}^{(1)}\}_{j\in
[p]\cup\{0\}}\right) .
\end{align}%
By Assumptions \ref{ass:1}(v) and \ref{ass:2} and Theorem \ref{Thm2}(ii), we
have $\max_{i\in I_{3},t\in[T]}\left\Vert \dot{\varpi}_{it}\right\Vert_{2}%
\leq C_{11}\xi_{N}$ a.s.. By the first order condition of the quantile
regression, we have 
\begin{equation}
\max_{i\in I_{3}}\left\Vert \dot{\mathbb{H}}_{i}\left(\{\dot{u}%
_{i,j}^{(1)}\}_{j\in [p]\cup \{0\}}\right) \right\Vert_{2}=O_{p}\left(\frac{1%
}{T}\max_{i\in I_{3},t\in[T]}\left\Vert \dot{\varpi}_{it}\right\Vert_{2}%
\right) =O_{p}\left(\frac{\xi_{N}}{T}\right) .  \label{B.26}
\end{equation}

Next, we show that $\max_{i\in I_{3}}\left\Vert\dot{\mathbb{W}}_{i}\left(\{%
\dot{u}_{i,j}^{(1)}\}_{j\in [p]\cup \{0\}}\right)\right\Vert_{2}
=o_{p}\left(\left(N\vee T\right) ^{-1/2}\right) $. Notice that 
\begin{align}
\dot{\mathbb{W}}_{i}\left(\{\dot{u}_{i,j}^{(1)}\}_{j\in [p]\cup\{0\}}\right)
& =\frac{1}{T}\sum_{t=1}^{T}\dot{\varpi}_{it}\left(\mathbf{1}\left\{
\epsilon_{it}\leq g_{it}(\{O_{j}^{(1)}u_{i,j}^{0}\}_{j\in [p]\cup
\{0\}})\right\} -\mathbf{1}\left\{ \epsilon_{it}\leq g_{it}(\{\dot{u}%
_{i,j}^{(1)}\}_{j\in [p]\cup \{0\}})\right\} \right)  \notag \\
& -\frac{1}{T}\sum_{t\in[T]}\dot{\varpi}_{it}\left[ \mathfrak{F}%
_{it}\left(g(\{O_{j}^{(1)}u_{i,j}^{0}\}_{j\in [p]\cup \{0\}})\right)-%
\mathfrak{F}_{it}\left(g_{it}(\{\dot{u}_{i,j}^{(1)}\}_{j\in
[p]\cup\{0\}})\right) \right]  \notag \\
& =\dot{\mathbb{W}}_{i}^{I}-\dot{\mathbb{W}}_{i}^{II}\left(\{\dot{u}%
_{i,j}^{(1)}\}_{j\in [p]\cup \{0\}}\right) ,  \label{B.27}
\end{align}%
where $\dot{\mathbb{W}}_{i}^{I}=\frac{1}{T}\sum_{t=1}^{T}\dot{\mathbb{W}}%
_{it}^{I}$ with 
\begin{align*}
&\dot{\mathbb{W}}_{it}^{I}=\dot{\varpi}_{it}\left\{ \left(\mathbf{1}\left\{
\epsilon_{it}\leq g_{it}(\{O_{j}^{(1)}u_{i,j}^{0}\}_{j\in [p]\cup
\{0\}})\right\} -\mathbf{1}\left\{ \epsilon_{it}\leq 0\right\}\right) -\left[
\mathfrak{F}_{it}\left(g_{it}(\{O_{j}^{(1)}u_{i,j}^{0}\}_{j\in [p]\cup
\{0\}})\right) -\mathfrak{F}_{it}(0)\right] \right\}
\end{align*}
and 
\begin{align*}
\dot{\mathbb{W}}_{i}^{II}\left(\{\dot{u}_{i,j}^{(1)}\}_{j\in [p]\cup
\{0\}}\right) &=\frac{1}{T}\sum_{t=1}^{T}\dot{\varpi}_{it}\biggl\{\left(%
\mathbf{1}\left\{ \epsilon_{it}\leq g_{it}(\{\dot{u}_{i,j}^{(1)}\}_{j\in
[p]\cup \{0\}})\right\} -\mathbf{1}\left\{\epsilon_{it}\leq 0\right\} \right)
\\
& -\left[ \mathfrak{F}_{it}\left(g_{it}(\{\dot{u}_{i,j}^{(1)}\}_{j\in
[p]\cup \{0\}})\right) -\mathfrak{F}_{it}(0)\right] \biggr\}.
\end{align*}%
Noting that 
\begin{equation*}
g_{it}(\{O_{j}^{(1)}u_{i,j}^{0}\}_{j\in [p]\cup
\{0\}})=\left(O_{0}^{(1)}u_{i,0}\right) ^{\prime}\left(\dot{v}%
_{t,0}^{(1)}-O_{0}^{(1)}v_{t,0}^{0}\right) +\sum_{j\in
[p]}\left(O_{j}^{(1)}u_{i,j}\right) ^{\prime}\left(\dot{v}%
_{t,j}^{(1)}-O_{j}^{(1)}v_{t,j}^{0}\right) X_{j,it},
\end{equation*}%
and $\max_{i\in I_{2},t\in[T]}\left\vert
g_{it}(\{O_{j}^{(1)}u_{i,j}^{0}\}_{j\in [p]\cup
\{0\}})\right\vert=O_{p}(\xi_{N}\eta_{N})$, we have, $\max_{i\in I_{3},t\in
[T]}\left\Vert \dot{\mathbb{W}}_{it}^{I}\right\Vert_{2}=O_{p}(\xi_{N})$, $%
\mathbb{E}\left(\mathbb{W}_{it}^{I}\bigg|\mathscr{D}_{e_{i}}^{I_{1}\cup
I_{2}}\right) =0$, and for some positive constants $C_{11}$, 
\begin{align*}
& \max_{i\in I_{3},t\in[T]}\left\Vert Var\left(\dot{\mathbb{W}}_{it}^{I}%
\bigg |\mathscr{D}_{e_{i}}^{I_{1}\cup I_{2}}\right) \right\Vert_{F} \\
& \leq \max_{i\in I_{3},t\in[T]}\left\Vert \mathbb{E}\left\{ \dot{\varpi}%
_{it}\dot{\varpi}_{it}^{\prime}\left(\mathbf{1}\left\{ \epsilon_{it}\leq
g(\{O_{j}^{(1)}u_{i,j}^{0}\}_{j\in [p]\cup \{0\}})\right\} -\mathbf{1}%
\left\{ \epsilon_{it}\leq 0\right\} \right) ^{2}\bigg |\mathscr{D}%
_{e_{i}}^{I_{1}\cup I_{2}}\right\} \right\Vert_{F} \\
& \leq \max_{i\in I_{3},t\in[T]}\left\Vert \dot{\varpi}_{it}\right%
\Vert_{2}^{2}\mathbb{E}\left(\mathbf{1}\left\{ 0\leq
\left\vert\epsilon_{it}\right\vert \leq \left\vert
g(\{O_{j}^{(1)}u_{i,j}^{0}\}_{j\in[p]\cup \{0\}})\right\vert \right\} \bigg |%
\mathscr{D}_{e_{i}}^{I_{1}\cup I_{2}}\right) =O_{p}(\xi_{N}^{3 }\eta_{N}),
\end{align*}%
and 
\begin{align*}
& \max_{i\in I_{3},t\in[T]}\sum_{s=t+1}^{T}\left\Vert Cov\left(\dot{\mathbb{W%
}}_{it}^{I},\dot{\mathbb{W}}_{is}^{I}\bigg |\mathscr{D}_{e_{i}}^{I_{1}\cup
I_{2}}\right) \right\Vert_{F} \\
& \lesssim \max_{i\in I_{3},t\in[T]}\sum_{s=t+1}^{T}\left[ \mathbb{E}%
\left(\left\Vert \dot{\mathbb{W}}_{it}^{I}\right\Vert_{F} ^{2+\vartheta }%
\bigg |\mathscr{D}_{e_{i}}^{I_{1}\cup I_{2}}\right) \right] ^{\frac{1}{%
2+\vartheta }}\left[ \mathbb{E}\left(\left\Vert \dot{\mathbb{W}}%
_{is}^{I}\right\Vert_{F}^{2+\vartheta }\bigg |\mathscr{D}_{e_{i}}^{I_{1}\cup
I_{2}}\right) \right] ^{\frac{1}{2+\vartheta }}\left[ \alpha (t-s)\right]
^{1-\frac{2}{2+\vartheta }} \\
& \leq \max_{i\in I_{3},t\in[T]}\left[ \mathbb{E}\left(\left\Vert \dot{%
\mathbb{W}}_{it}^{I}\right\Vert_{F} ^{2+\vartheta }\bigg |\mathscr{D}%
_{e_{i}}^{I_{1}\cup I_{2}}\right) \right] ^{\frac{2}{2+\vartheta }}\max_{t\in%
[T]}\sum_{s=t+1}^{T}\left[ \alpha (t-s)\right] ^{1-\frac{2}{2+\vartheta }}%
\text{a.s.} \\
& =O_{p}\left(\xi_{N}^{\frac{6+2\vartheta }{2+\vartheta }}\eta_{N}^{\frac{2}{%
2+\vartheta }}\right) ,
\end{align*}%
for any $\vartheta >0$, where the first inequality holds by Davydov's
inequality for conditional strong mixing processes, and the last equality
holds by the fact that $\mathbb{E}\left(\left\Vert \dot{\mathbb{W}}%
_{it}^{I}\right\Vert ^{2+\vartheta }\bigg |\mathscr{D}_{e_{i}}^{I_{1}\cup
I_{2}}\right) =O_{p}\left(\xi_{N}^{3+\vartheta }\eta_{N}\right) $ following
the similar argument as in (\ref{B.26}). Combining (\ref{B.26}) and (\ref%
{B.27}) yields 
\begin{equation*}
\max_{i\in I_{3},t\in[T]}\left\{ \left\Vert Var\left(\dot{\mathbb{W}}%
_{it}^{I}\right) \right\Vert_{F}+2\sum_{s>t}\left\Vert Cov\left(\dot{\mathbb{%
W}}_{it}^{I},\dot{\mathbb{W}}_{is}^{I}\right) \right\Vert_{F}\right\}
=O_{p}\left(\xi_{N}^{\frac{6+2\vartheta }{2+\vartheta }}\eta_{N}^{\frac{2}{%
2+\vartheta }}\right) .
\end{equation*}%
For any constant $C_{12}>0$, define 
\begin{align*}
\mathscr{A}_{5,i}(C_{12})& =\left\{ \max_{t\in[T]}\left\{ \left\Vert
Var\left(\dot{\mathbb{W}}_{it}^{I}\right)
\right\Vert_{F}+2\sum_{s>t}\left\Vert Cov\left(\dot{\mathbb{W}}_{it}^{I},%
\dot{\mathbb{W}}_{is}^{I}\right) \right\Vert_{F}\right\} \leq C_{12}\xi_{N}^{%
\frac{6+2\vartheta }{2+\vartheta }}\eta_{N}^{\frac{2}{2+\vartheta }}\right\}
, \\
\mathscr{A}_{5}(C_{12})& =\left\{ \max_{i\in I_{3},t\in[T]}\left\{\left\Vert
Var\left(\dot{\mathbb{W}}_{it}^{I}\right)
\right\Vert_{F}+2\sum_{s>t}\left\Vert Cov\left(\dot{\mathbb{W}}_{it}^{I},%
\dot{\mathbb{W}}_{is}^{I}\right) \right\Vert_{F}\right\} \leq C_{12}\xi_{N}^{%
\frac{6+2\vartheta }{2+\vartheta }}\eta_{N}^{\frac{2}{2+\vartheta }}\right\}.
\end{align*}
For any $e>0$, we can find a sufficiently large constants $C_{12}$ such that 
$\mathbb{P}(\mathscr{A}_{5}^{c}(C_{12}))\leq e$. Therefore, we have 
\begin{align}
& \mathbb{P}\left\{ \max_{i\in I_{3}}\left\Vert \frac{1}{T}\sum_{t=1}^{T}%
\dot{\mathbb{W}}_{it}^{I}\right\Vert_{2}>C_{13}\xi_{N}^{\frac{5+\vartheta}{%
2+\vartheta }}\left(\frac{\log (N\vee T)}{N\wedge T}\right) ^{\frac{1}{%
4+2\vartheta }}\sqrt{\frac{\log (N\vee T)}{T}}\right\}  \notag
\label{Wdot_exp1} \\
& \leq \mathbb{P}\left\{ \max_{i\in I_{3}}\left\Vert \frac{1}{T}%
\sum_{t=1}^{T}\dot{\mathbb{W}}_{it}^{I}\right\Vert_{2}>C_{13}\xi_{N}^{\frac{%
5+\vartheta }{2+\vartheta }}\left(\frac{\log (N\vee T)}{N\wedge T}\right) ^{%
\frac{1}{4+2\vartheta }}\sqrt{\frac{\log (N\vee T)}{T}},\mathscr{A}%
_{5}(C_{12})\right\} +e  \notag \\
& \leq \sum_{i\in I_{3}}\mathbb{P}\left\{ \left\Vert \frac{1}{T}%
\sum_{t=1}^{T}\dot{\mathbb{W}}_{it}^{I}\right\Vert_{2}>C_{13}\xi_{N}^{\frac{%
5+\vartheta }{2+\vartheta }}\left(\frac{\log (N\vee T)}{N\wedge T}\right) ^{%
\frac{1}{4+2\vartheta }}\sqrt{\frac{\log (N\vee T)}{T}},\mathscr{A}%
_{5}(C_{12})\right\} +e  \notag \\
& \leq \sum_{i\in I_{3}}\mathbb{P}\left\{ \left\Vert \frac{1}{T}%
\sum_{t=1}^{T}\dot{\mathbb{W}}_{it}^{I}\right\Vert_{2}>C_{13}\xi_{N}^{\frac{%
5+\vartheta }{2+\vartheta }}\left(\frac{\log (N\vee T)}{N\wedge T}\right) ^{%
\frac{1}{4+2\vartheta }}\sqrt{\frac{\log (N\vee T)}{T}},\mathscr{A}%
_{5,i}(C_{12})\right\} +e  \notag \\
& =\sum_{i\in I_{3}}\mathbb{E}\mathbb{P}\left\{ \left\Vert \frac{1}{T}%
\sum_{t=1}^{T}\dot{\mathbb{W}}_{it}^{I}\right\Vert_{2}>C_{13}\xi_{N}^{\frac{%
5+\vartheta }{2+\vartheta }}\left(\frac{\log (N\vee T)}{N\wedge T}\right) ^{%
\frac{1}{4+2\vartheta }}\sqrt{\frac{\log (N\vee T)}{T}}\Bigg|\mathscr{D}%
_{e_{i}}^{I_{1}\cup I_{2}}\right\} 1\{\mathscr{A}_{5,i}(C_{12})\}+e,
\end{align}%
where the second inequality is by the union bound, the third inequality is
by $\mathscr{A}_{5}(C_{12})\subset \mathscr{A}_{5,i}(C_{12})$, and the last
equality is owing to the fact that $\mathscr{A}_{5,i}(C_{12})$ is $%
\mathscr{D}_{e_{i}}^{I_{1}\cup I_{2}}$ measurable. Given $\mathscr{D}%
_{e_{i}}^{I_{1}\cup I_{2}}$, the randomness in $\dot{\mathbb{W}}_{i}^{I}$
only comes from $\{\epsilon_{it}\}_{t\in[T]}$, which are strong mixing given 
$\mathscr{D}_{e_{i}}^{I_{1}\cup I_{2}}$. Therefore, on $\mathscr{A}%
_{5,i}(C_{12})$, Lemma \ref{Lem:Bern}(ii) implies 
\begin{align}
& \mathbb{P}\left\{ \left\Vert\frac{1}{T}\sum_{t=1}^{T}\dot{\mathbb{W}}%
_{it}^{I}\right\Vert_{2}>C_{13}\xi_{N}^{\frac{5+\vartheta }{2+\vartheta }%
}\left(\frac{\log (N\vee T)}{N\wedge T}\right) ^{\frac{1}{4+2\vartheta }}%
\sqrt{\frac{\log (N\vee T)}{T}}\Bigg|\mathscr{D}_{e_{i}}^{I_{1}\cup
I_{2}}\right\}  \notag  \label{Wdot_exp2} \\
&=\mathbb{P}\left\{ \left\Vert\sum_{t=1}^{T}\dot{\mathbb{W}}%
_{it}^{I}\right\Vert_{2}>C_{13}\xi_{N}^{\frac{5+\vartheta }{2+\vartheta }%
}\left(\frac{\log (N\vee T)}{N\wedge T}\right) ^{\frac{1}{4+2\vartheta }}%
\sqrt{T\log (N\vee T)}\Bigg|\mathscr{D}_{e_{i}}^{I_{1}\cup I_{2}}\right\} 
\notag \\
& \leq \exp \left\{ -\frac{c_{12}C_{13}^{2}T\left(\frac{\log (N\vee T)}{%
N\wedge T}\right) ^{\frac{1}{2+\vartheta }}\xi_{N}^{\frac{10+2\vartheta}{%
2+\vartheta }}\log (N\vee T)}{C_{12}T\left(\frac{\log (N\vee T)}{N\wedge T}%
\right) ^{\frac{1}{2+\vartheta }}\xi_{N}^{\frac{10+2\vartheta}{2+\vartheta }%
}+C_{11}^{2}\xi_{N}^{2}+C_{11}C_{13}\sqrt{T\log (N\vee T)}\xi_{N}^{\frac{%
7+2\vartheta }{2+\vartheta }}\left(\frac{\log(N\vee T)}{N\wedge T}\right)^{%
\frac{1}{4+2\vartheta }}\left(\log T\right) ^{2}}\right\} ,
\end{align}%
which further implies 
\begin{equation*}
\mathbb{P}\left\{ \max_{i\in I_{3}}\left\Vert \frac{1}{T}\sum_{t=1}^{T}\dot{%
\mathbb{W}}_{it}^{I}\right\Vert_{2}>C_{13}\xi_{N}^{\frac{5+\vartheta }{%
2+\vartheta }}\left(\frac{\log (N\vee T)}{N\wedge T}\right) ^{\frac{1}{%
4+2\vartheta }}\sqrt{\frac{\log (N\vee T)}{T}}\right\} =o(1)+e.
\end{equation*}%
As $e$ is arbitrary, by Assumption \ref{ass:1}(ix), we obtain that 
\begin{equation*}
\max_{i\in I_{3}}\left\Vert \frac{1}{T}\sum_{t=1}^{T}\dot{\mathbb{W}}%
_{it}^{I}\right\Vert_{2}=O_{p}\left(C_{13}\xi_{N}^{\frac{5+\vartheta }{%
2+\vartheta }}\left(\frac{\log (N\vee T)}{N\wedge T}\right) ^{\frac{1}{%
4+2\vartheta }}\sqrt{\frac{\log (N\vee T)}{T}}\right) =o_{p}\left(\left(
N\vee T\right) ^{-1/2}\right) .
\end{equation*}

For $\dot{\mathbb{W}}_{i}^{II}\left(\{\dot{u}_{i,j}^{(1)}\}_{j\in [p]\cup
\{0\}}\right) $, we observe that 
\begin{align*}
g_{it}(\{\dot{u}_{i,j}^{(1)}\}_{j\in [p]\cup \{0\}})& =\dot{u}%
_{i,0}^{(1)\prime}\dot{v}_{t,0}^{(1)}+\sum_{j\in [p]}\dot{u}%
_{i,j}^{(1)\prime}\dot{v}_{t,j}^{(1)}X_{j,it}-u_{i,0}^{0\prime}v_{t,0}^{0}-%
\sum_{j\in [p]}u_{i,j}^{0\prime}v_{t,j}^{0}X_{j,it} \\
& =\left(\dot{u}_{i,0}^{(1)}-O_{0}^{(1)}u_{i,0}^{0}\right) ^{\prime}\dot{v}%
_{t,0}^{(1)}+\left(O_{0}^{(1)}u_{i,0}^{0}\right) ^{\prime}\left(\dot{v}%
_{t,0}^{(1)}-O_{0}^{(1)}v_{t,0}^{0}\right) \\
& +\sum_{j\in [p]}\left(\dot{u}_{i,j}^{(1)}-O_{j}^{(1)}u_{i,j}^{0}\right)
^{\prime}\dot{v}_{t,j}^{(1)}X_{j,it}+\sum_{j\in
[p]}\left(O_{j}^{(1)}u_{i,j}^{0}\right) ^{\prime}\left(\dot{v}%
_{t,j}^{(1)}-O_{j}^{(1)}v_{t,j}^{0}\right) X_{j,it} \\
& =\dot{\Delta}_{i,u}^{\prime}\dot{\Phi}_{it}^{(1)}+\dot{\Delta}%
_{t,v}^{\prime}\Psi_{it}^{0},
\end{align*}
where $\dot{\Delta}_{i,u}=\left(\dot{\Delta}_{i,0}^{\prime},\cdots ,\dot{%
\Delta}_{i,p}^{\prime}\right) ^{\prime}$ with $\dot{\Delta}%
_{i,j}=O_{j}^{(1)\prime}\dot{u}_{i,j}^{(1)}-u_{i,j}^{0}$, $\dot{\Delta}%
_{t,v}=\left(\dot{\Delta}_{t,0}^{\prime},\cdots ,\dot{\Delta}%
_{t,p}^{\prime}\right) ^{\prime}$ with $\dot{\Delta}_{t,j}=O_{j}^{(1)\prime}%
\dot{v}_{t,j}^{(1)}-v_{t,j}^{0}$, 
\begin{align*}
& \dot{\Phi}_{it}^{(1)}=\left[ \left(O_{0}^{(1)\prime}\dot{v}%
_{t,0}^{(1)}\right) ^{\prime},\left(O_{1}^{(1)\prime}\dot{v}%
_{t,1}^{(1)}X_{1,it}\right) ^{\prime},\cdots ,\left(O_{p}^{(1)\prime}\dot{v}%
_{t,p}^{(1)}X_{p,it}\right) ^{\prime}\right] ^{\prime},\quad \text{and} \\
&
\Psi_{it}^{0}=(u_{i,0}^{0\prime},u_{i,1}^{0\prime}X_{1,it},...,u_{i,p}^{0%
\prime}X_{p,it})^{\prime}.
\end{align*}
Unlike the analysis for $\dot{\mathbb{W}}_{it}^{I}$, to handle the
dependence between $\epsilon_{it}$ and $\dot{\Delta}_{i,u}$, for any
constant $C_{14}>0$, we first define an event set $\mathscr{A}%
_{6}(C_{14})=\left\{ \max_{i\in I_{3}}\left\Vert \dot{\Delta}%
_{i,u}\right\Vert_{2}\leq C_{14}\eta_{N}\right\} $ with $\mathbb{P}(%
\mathscr{A}_{6}^{c}(C_{14}))\leq e$ for any $e>0$ by Theorem \ref{Thm2}(i),
then we have 
\begin{align}
& \mathbb{P}\left(\max_{i\in I_{3}}\left\Vert \dot{\mathbb{W}}%
_{i}^{II}\left(\{\dot{u}_{i,j}^{(1)}\}_{j\in [p]\cup
\{0\}}\right)\right\Vert_{2}>C_{13}\xi_{N}^{\frac{5+\vartheta }{2+\vartheta }%
}\left(\frac{\log (N\vee T)}{N\wedge T}\right) ^{\frac{1}{4+2\vartheta }}%
\sqrt{\frac{\log (N\vee T)}{T}}\right)  \notag \\
& \leq \mathbb{P}\left(\max_{i\in I_{3}}\left\Vert \dot{\mathbb{W}}%
_{i}^{II}\left(\{\dot{u}_{i,j}^{(1)}\}_{j\in [p]\cup
\{0\}}\right)\right\Vert_{2}>C_{13}\xi_{N}^{\frac{5+\vartheta }{2+\vartheta }%
}\left(\frac{\log (N\vee T)}{N\wedge T}\right) ^{\frac{1}{4+2\vartheta }}%
\sqrt{\frac{\log (N\vee T)}{T}},\mathscr{A}_{6}(C_{14})\right) +e  \notag \\
& \leq \mathbb{P}\left(\sup_{s\in \mathbb{S}}\max_{i\in I_{3}}\left\Vert%
\overline{\mathbb{W}}_{i}^{II}\left(s\right) \right\Vert_{2}>C_{13}\xi_{N}^{%
\frac{5+\vartheta }{2+\vartheta }}\left(\frac{\log (N\vee T)}{N\wedge T}%
\right) ^{\frac{1}{4+2\vartheta }}\sqrt{\frac{\log (N\vee T)}{T}}\right) +e
\label{Wdot_1}
\end{align}%
with $\mathbb{S}=\left\{ s\in \mathbb{R}^{\left(\sum_{j\in
[p]\cup\{0\}}K_{j}\right) \times 1}:\left\Vert s\right\Vert_{2}\leq
C_{14}\eta_{N}\right\} $ and 
\begin{equation*}
\overline{\mathbb{W}}_{i}^{II}\left(s\right) =\frac{1}{T}\sum_{t=1}^{T}\dot{%
\varpi}_{it}\left\{ \left(\mathbf{1}\left\{ \epsilon_{it}\leq s^{\prime}\dot{\Phi}_{it}^{(1)}+\dot{\Delta}_{t,v}^{\prime}\Psi_{it}^{0}\right\} -\mathbf{1}%
\left\{ \epsilon_{it}\leq 0\right\} \right) -\left[ \mathfrak{F}%
_{it}\left(s^{\prime}\dot{\Phi}_{it}^{(1)}+\dot{\Delta}_{t,v}^{\prime}%
\Psi_{it}^{0}\right) -\mathfrak{F}_{it}(0)\right] \right\} .
\end{equation*}%
Similarly as in (\ref{B.11}), we sketch the proof. Divide $\mathbb{S}$ into $%
\mathbb{S}_{m}$ with center $s_{m}$ for $m=1,\cdots ,n_{\mathbb{S}}$ if $%
s\in \mathbb{S}_{m}$, then $\left\Vert s-s_{m}\right\Vert_{2}<\frac{%
\varepsilon }{T}$ and $n_{\mathbb{S}}\lesssim T^{\sum_{j\in
[p]\cup\{0\}}K_{j}}$. Then, $\forall s\in \mathbb{S}_{m}$, we have 
\begin{equation}
\left\Vert \overline{\mathbb{W}}_{i}^{II}\left(s\right) \right\Vert_{2}\leq
\left\Vert \overline{\mathbb{W}}_{i}^{II}\left(s_{m}\right)\right\Vert_{2}+%
\left\Vert \overline{\mathbb{W}}_{i}^{II}\left(s\right) -\overline{\mathbb{W}%
}_{i}^{II}\left(s_{m}\right) \right\Vert_{2},  \label{Wdot_2}
\end{equation}
with 
\begin{align}
& \max_{i\in I_{3},m\in [n_{\mathbb{S}}]}\sup_{s\in \mathbb{S}%
_{m}}\left\Vert \overline{\mathbb{W}}_{i}^{II}\left(s\right) -\overline{%
\mathbb{W}}_{i}^{II}\left(s_{m}\right) \right\Vert_{2}  \notag
\label{Wdot_3} \\
& \leq \max_{i\in I_{3},m\in [n_{\mathbb{S}}]}\sup_{s\in \mathbb{S}%
_{m}}\left\Vert \frac{1}{T}\sum_{t=1}^{T}\dot{\varpi}_{it}\left(\mathbf{1}%
\left\{ \epsilon_{it}\leq s^{\prime}\dot{\Phi}_{it}^{(1)}+\dot{\Delta}%
_{t,v}^{\prime}\Psi_{it}^{0}\right\} -\mathbf{1}\left\{ \epsilon_{it}\leq
s_{m}^{\prime}\dot{\Phi}_{it}^{(1)}+\dot{\Delta}_{t,v}^{\prime}\Psi_{it}^{0}%
\right\} \right) \right\Vert_{2}  \notag \\
& +\max_{i\in I_{3},m\in [n_{\mathbb{S}}]}\sup_{s\in \mathbb{S}%
_{m}}\left\Vert \frac{1}{T}\sum_{t=1}^{T}\dot{\varpi}_{it}\left[ \mathfrak{F}%
_{it}\left(s^{\prime}\dot{\Phi}_{it}^{(1)}+\dot{\Delta}_{t,v}^{\prime}%
\Psi_{it}^{0}\right) -\mathfrak{F}_{it}\left(s_{m}^{\prime}\dot{\Phi}_{it}^{(1)}+\dot{\Delta}_{t,v}^{\prime}\Psi_{it}^{0}\right) \right]
\right\Vert_{2}  \notag \\
& \leq \max_{i\in I_{3},m\in [n_{\mathbb{S}}]}\frac{1}{T}\sum_{t\in[T]%
}\left\Vert \dot{\varpi}_{it}\right\Vert_{2}\mathbb{E}\left(\mathbf{1}%
\left\{ \left\vert \epsilon_{it}-\dot{\Delta}_{t,v}^{\prime}%
\Psi_{it}^{0}-s_{m}^{\prime}\dot{\Phi}_{it}^{(1)}\right\vert \leq \frac{%
\varepsilon ||\dot{\Phi}_{it}^{(1)}||_{2}}{T}\right\} \bigg |\mathscr{D}%
_{e_{i}}^{I_{1}\cup I_{2}}\right) +\max_{i\in I_{2},m\in [n_{\mathbb{S}%
}]}\left\vert \overline{\mathbb{W}}_{i}^{III}(m)\right\vert  \notag \\
& +\max_{i\in I_{3},m\in [n_{\mathbb{S}}]}\sup_{s\in \mathbb{S}_{m}}\frac{1}{%
T}\sum_{t\in[T]}\left\Vert \dot{\varpi}_{it}\right\Vert_{2}||\dot{\Phi}_{it}^{(1)}||_{2}\left\Vert s-s_{m}\right\Vert_{2}
\end{align}
such that $\overline{\mathbb{W}}_{i}^{III}(m):= \frac{1}{T}\sum_{t\in[T]}%
\overline{\mathbb{W}}_{it}^{III}(m)$ and 
\begin{align*}
\overline{\mathbb{W}}_{it}^{III}(m)& :=\left\Vert \dot{\varpi}%
_{it}\right\Vert_{2}\biggl[\left(\mathbf{1}\left\{ \left\vert \epsilon_{it}-%
\dot{\Delta}_{t,v}^{\prime}\Psi_{it}^{0}-s_{m}^{\prime}\dot{\Phi}_{it}^{(1)}\right\vert \leq \frac{\varepsilon ||\dot{\Phi}_{it}^{(1)}||_{2}}{%
T}\right\} \right) \\
& -\mathbb{E}\left(\mathbf{1}\left\{ \left\vert \epsilon_{it}-\dot{\Delta}%
_{t,v}^{\prime}\Psi_{it}^{0}-s_{m}^{\prime}\dot{\Phi}_{it}^{(1)}\right\vert
\leq \frac{\varepsilon ||\dot{\Phi}_{it}^{(1)}||_{2}}{T}\right\} \bigg |%
\mathscr{D}_{e_{i}}^{I_{1}\cup I_{2}}\right) \biggr].
\end{align*}
Like (\ref{Lem:A:A3_mean}), we can show that $\max_{i\in I_{3},m\in [n_{%
\mathbb{S}}]}\frac{1}{T}\sum_{t\in[T]}\left\Vert \dot{\varpi}%
_{it}\right\Vert_{2}\mathbb{E}\left(\mathbf{1}\left\{ \left\vert
\epsilon_{it}-\dot{\Delta}_{t,v}^{\prime}\Psi_{it}^{0}-s_{m}^{\prime}\dot{\Phi}_{it}\right\vert \leq \frac{\varepsilon ||\dot{\Phi}_{it}^{(1)}||_{2}}{T%
}\right\} \bigg |\mathscr{D}_{e_{i}}^{I_{1}\cup I_{2}}\right) $ $=O_{p}(%
\frac{\varepsilon }{T})$ because $\dot{\Delta}_{t,v}^{\prime}%
\Psi_{it}^{0}+s_{m}^{\prime}\dot{\Phi}_{it}^{(1)}$ and $||\dot{\Phi}_{it}||$
are measurable in $\mathscr{D}_{e_{i}}^{I_{1}\cup I_{2}}$ and the
conditional density of $\epsilon_{it}$ given $\mathscr{D}_{e_{i}}^{I_{1}\cup
I_{2}}$ is bounded. Also we have 
\begin{equation*}
\max_{i\in I_{3},m\in [n_{\mathbb{S}}]}\sup_{s\in \mathbb{S}_{m}}\frac{1}{T}%
\sum_{t\in[T]}\left\Vert \dot{\varpi}_{it}\right\Vert_{2}||\dot{\Phi}_{it}^{(1)}||_{2}\left\Vert s-s_{m}\right\Vert_{2}=\frac{\varepsilon }{T}%
\max_{i\in I_{3}}\frac{1}{T}\sum_{t\in[T]}\left\Vert \dot{\varpi}%
_{it}\right\Vert_{2}||\dot{\Phi}_{it}^{(1)}||_{2}=O_{p}\left(\frac{%
\varepsilon }{T}\right) .
\end{equation*}
In addition, we note that 
\begin{align*}
& \max_{i\in I_{3},m\in [n_{\mathbb{S}}],t\in[T]}\left\Vert \dot{\varpi}%
_{it}\right\Vert_{2}=O_{p}(\xi_{N}),\quad \max_{i\in I_{3},m\in [n_{\mathbb{S%
}}],t\in[T]}Var\left(\overline{\mathbb{W}}_{it}^{III}(m)\bigg|\mathscr{D}%
_{e_{i}}^{I_{1}\cup I_{2}}\right)=O_{p}\left(\frac{\xi_{N}^{3 }\varepsilon }{%
T}\right) , \\
& \max_{i\in I_{3},m\in [n_{\mathbb{S}}],t\in [T]}\sum_{s=t+1}^{T}\left\vert
Cov\left(\overline{\mathbb{W}}_{it}^{III}(m),\overline{\mathbb{W}}%
_{is}^{III}(m)\bigg|\mathscr{D}_{e_{i}}^{I_{1}\cup I_{2}}\right) \right\vert
=O_{p}\left(\xi_{N}^{8/3}\left(\frac{\varepsilon }{T}\right) ^{2/3}\right) ,
\end{align*}%
and for any positive constant $C_{15}$ and $C_{16}$, define event set 
\begin{align*}
& \mathscr{A}_{8,N} \\
& = 
\begin{pmatrix}
& \operatornamewithlimits{\max}\limits_{i\in I_{3},m\in [n_{\mathbb{S}}],t\in%
[T]}\left\{ Var\left(\overline{\mathbb{W}}_{it}^{III}(m)\bigg|\mathscr{D}%
_{e_{i}}^{I_{1}\cup I_{2}}\right)+2\sum_{s=t+1}^{T}\left\vert Cov\left(%
\overline{\mathbb{W}}_{it}^{III}(m),\overline{\mathbb{W}}_{is}^{III}(m)\bigg|%
\mathscr{D}_{e_{i}}^{I_{1}\cup I_{2}}\right) \right\vert \right\} \leq
C_{15}\xi_{N}^{8/3}\left(\frac{\varepsilon }{T}\right) ^{2/3} \\ 
& \text{and}\quad \operatornamewithlimits{\max}\limits_{i\in I_{3},t\in[T]%
}\left\Vert \dot{\varpi}_{it}\right\Vert_{2}\leq C_{16}\xi_{N}%
\end{pmatrix}%
, \\
& \mathscr{A}_{8,N,i} \\
& = 
\begin{pmatrix}
& \operatornamewithlimits{\max}\limits_{m\in [n_{\mathbb{S}}],t\in[T]%
}\left\{ Var\left(\overline{\mathbb{W}}_{it}^{III}(m)\bigg|\mathscr{D}%
_{e_{i}}^{I_{1}\cup I_{2}}\right) +2\sum_{s=t+1}^{T}\left\vert Cov\left(%
\overline{\mathbb{W}}_{it}^{III}(m),\overline{\mathbb{W}}_{is}^{III}(m)\bigg|%
\mathscr{D}_{e_{i}}^{I_{1}\cup I_{2}}\right) \right\vert\right\} \leq
C_{15}\xi_{N}^{8/3}\left(\frac{\varepsilon }{T}\right) ^{2/3} \\ 
& \text{and}\quad \operatornamewithlimits{\max}\limits_{t\in [T]}\left\Vert 
\dot{\varpi}_{it}\right\Vert_{2}\leq C_{16}\xi_{N}%
\end{pmatrix}%
,
\end{align*}
with $\mathbb{P}\left(\mathscr{A}_{8,N}^{c}\right) \leq e$ for any positive $%
e$. Then we have 
\begin{align}
& \mathbb{P}\left(\max_{i\in I_{3},m\in [n_{\mathbb{S}}]}\left\vert 
\overline{\mathbb{W}}_{i}^{III}(m)\right\vert >C_{13}\xi_{N}^{\frac{%
5+\vartheta }{2+\vartheta }}\left(\frac{\log (N\vee T)}{N\wedge T}\right) ^{%
\frac{1}{4+2\vartheta }}\sqrt{\frac{\log (N\vee T)}{T}}\right)  \notag
\label{Wdot_4} \\
& \leq \mathbb{P}\left(\max_{i\in I_{3},m\in [n_{\mathbb{S}}]}\left\vert 
\overline{\mathbb{W}}_{i}^{III}(m)\right\vert >C_{13}\xi_{N}^{\frac{%
5+\vartheta }{2+\vartheta }}\left(\frac{\log (N\vee T)}{N\wedge T}\right) ^{%
\frac{1}{4+2\vartheta }}\sqrt{\frac{\log (N\vee T)}{T}},\mathscr{A}%
_{8,N}\right) +e  \notag \\
& \leq \sum_{i\in [I_{3}],m\in [n_{\mathbb{S}}]}\mathbb{P}\left(\left\vert 
\overline{\mathbb{W}}_{i}^{III}(m)\right\vert >C_{13}\xi_{N}^{\frac{%
5+\vartheta }{2+\vartheta }}\left(\frac{\log (N\vee T)}{N\wedge T}\right) ^{%
\frac{1}{4+2\vartheta }}\sqrt{\frac{\log (N\vee T)}{T}},\mathscr{A}%
_{8,N,i}\right) +e  \notag \\
& =\sum_{i\in [I_{3}],m\in [n_{\mathbb{S}}]}\mathbb{E}\mathbb{P}%
\left(\left\vert \overline{\mathbb{W}}_{i}^{III}(m)\right\vert
>C_{13}\xi_{N}^{\frac{5+\vartheta }{2+\vartheta }}\left(\frac{\log (N\vee T)%
}{N\wedge T}\right) ^{\frac{1}{4+2\vartheta }}\sqrt{\frac{\log (N\vee T)}{T}}%
\bigg|\mathscr{D}_{e_{i}}^{I_{1}\cup I_{2}}\right) \mathbf{1}\{\mathscr{A}%
_{8,N,i}\}+e  \notag \\
& =o(1)+e,
\end{align}%
where the last line is by Bernstein's inequality similar to (\ref{Wdot_exp2}%
). As $e$ is arbitrary, we have 
\begin{equation*}
\max_{i\in I_{3},m\in [n_{\mathbb{S}}]}\left\Vert \overline{\mathbb{W}}%
_{i}^{III}(m)\right\Vert_{2}=O_{p}\left(\xi_{N}^{\frac{5+\vartheta }{%
2+\vartheta }}\left(\frac{\log (N\vee T)}{N\wedge T}\right) ^{\frac{1}{%
4+2\vartheta }}\sqrt{\frac{\log (N\vee T)}{T}}\right) ,
\end{equation*}%
which implies that $\max_{i\in I_{3},m\in [n_{\mathbb{S}}]}\sup_{s\in 
\mathbb{S}_{m}}\left\Vert \overline{\mathbb{W}}_{i}^{II}\left(s\right) -%
\overline{\mathbb{W}}_{i}^{II}\left(s_{m}\right)
\right\Vert_{2}=O_{p}\left(\xi_{N}^{\frac{5+\vartheta }{2+\vartheta }}\left(%
\frac{\log (N\vee T)}{N\wedge T}\right) ^{\frac{1}{4+2\vartheta }}\sqrt{%
\frac{\log(N\vee T)}{T}}\right) $. Following the same argument in %
\eqref{Wdot_4}, we can show that 
\begin{equation*}
\max_{i\in I_{3},m\in [n_{\mathbb{S}}]}\left\Vert \overline{\mathbb{W}}%
_{i}^{II}\left(s_{m}\right) \right\Vert_{2}=O_{p}\left(\xi_{N}^{\frac{%
5+\vartheta }{2+\vartheta }}\left(\frac{\log (N\vee T)}{N\wedge T}\right) ^{%
\frac{1}{4+2\vartheta }}\sqrt{\frac{\log (N\vee T)}{T}}\right) ,
\end{equation*}%
which, combined with \eqref{Wdot_1} and \eqref{Wdot_2}, implies that 
\begin{align}
& \max_{i\in I_{3}}\left\Vert \dot{\mathbb{W}}_{i}^{II}\left(\{\dot{u}%
_{i,j}^{(1)}\}_{j\in [p]\cup \{0\}}\right)
\right\Vert_{2}=O_{p}\left(\xi_{N}^{\frac{5+\vartheta }{2+\vartheta }}\left(%
\frac{\log (N\vee T)}{N\wedge T}\right) ^{\frac{1}{4+2\vartheta }}\sqrt{%
\frac{\log(N\vee T)}{T}}\right) =o_{p}\left(\left(N\vee T\right)
^{-1/2}\right),\quad \text{and thus,}  \notag  \label{B.29} \\
& \max_{i\in I_{3}}\left\Vert \dot{\mathbb{W}}_{i}\left(\{\dot{u}%
_{i,j}^{(1)}\}_{j\in [p]\cup \{0\}}\right)
\right\Vert_{2}=o_{p}\left(\left(N\vee T\right) ^{-1/2}\right) .
\end{align}


Next, we observe that 
\begin{align}
& \dot{\mathcal{H}}_{i}\left(\{O_{j}^{(1)}u_{i,j}^{0}\}_{j=0}^{p}\right) -%
\dot{\mathcal{H}}_{i}\left(\{\dot{u}_{i,j}^{(1)}\}_{j=0}^{p}\right)  \notag
\label{B.30} \\
& =\frac{1}{T}\sum_{t=1}^{T}\mathfrak{f}_{it}\left(\dot{\Delta}%
_{t,v}^{\prime}\Psi_{it}^{0}\right) \dot{\varpi}_{it}\dot{\varpi}%
_{it}^{\prime}\dot{\Delta}_{i,u}+O_{p}\left(\max_{i\in i_{3}}\left\Vert \dot{%
\Delta}_{i,u}\right\Vert_{2}^{2}\right)  \notag \\
& =\frac{1}{T}\sum_{t=1}^{T}\mathfrak{f}_{it}\left(\dot{\Delta}%
_{t,v}^{\prime}\Psi_{it}^{0}\right) \dot{\varpi}_{it}\dot{\varpi}%
_{it}^{\prime}\dot{\Delta}_{i,u}+o_{p}\left(\left(N\vee T\right)
^{-1/2}\right) \quad \text{uniformly over $i\in I_{3}$},
\end{align}%
where the first equality holds by Taylor expansion and Lemma \ref{Lem16} and
the second equality is by Assumption \ref{ass:1}(ix). Combining (\ref{B.25}%
), (\ref{B.26}), (\ref{B.29}) and (\ref{B.30}), we have shown that 
\begin{equation}
\dot{\Delta}_{i,u}=\left[ \dot{D}_{i}^{I}\right] ^{-1}\dot{D}%
_{i}^{II}+o_{p}\left(\left(N\vee T\right) ^{-1/2}\right) \quad \text{%
uniformly over $i\in I_{3}$},  \label{Delta1}
\end{equation}%
where $\dot{D}_{i}^{I}:=\frac{1}{T}\sum_{t=1}^{T}\mathfrak{f}_{it}\left(\dot{%
\Delta}_{t,v}^{\prime}\Psi_{it}^{0}\right) \dot{\varpi}_{it}\dot{\varpi}%
_{it}^{\prime}\quad $and$\quad \dot{D}_{i}^{II}:=\frac{1}{T}\sum_{t=1}^{T}%
\left[ \tau -\mathbf{1}\left\{ \epsilon_{it}\leq \dot{\Delta}%
_{t,v}^{\prime}\Psi_{it}^{0}\right\} \right] \dot{\varpi}_{it}.$ Recall that 
$\varpi_{it}^{0}=\left(\left(O_{0}^{(1)}v_{t,0}^{0}\right)
^{\prime},\left(O_{1}^{(1)}v_{t,1}^{0}\right) ^{\prime}X_{1,it},\cdots
,\left(O_{1}^{(1)}v_{t,p}^{0}\right) ^{\prime}X_{p,it}\right) ^{\prime}$, $%
D_{i}^{I}=\frac{1}{T}\sum_{t=1}^{T}\mathfrak{f}_{it}(0)\varpi_{it}^{0}%
\varpi_{it}^{0\prime},\quad $and$\quad D_{i}^{II}=\frac{1}{T}\sum_{t=1}^{T}%
\left[\tau -\mathbf{1}\left\{ \epsilon_{it}\leq 0\right\} \right]
\varpi_{it}^{0}.$ Noting that 
\begin{align*}
\left[ \dot{D}_{i}^{I}\right] ^{-1}\dot{D}_{i}^{II}& =\left[ D_{i}^{I}\right]%
^{-1}D_{i}^{II}+\left[ D_{i}^{I}\right] ^{-1}\left(\dot{D}%
_{i}^{II}-D_{i}^{II}\right) +\left[ \left(\dot{D}_{i}^{I}\right)^{-1}-%
\left(D_{i}^{I}\right) ^{-1}\right] D_{i}^{II} \\
& +\left[ \left(\dot{D}_{i}^{I}\right) ^{-1}-\left(D_{i}^{I}\right) ^{-1}%
\right] \left(\dot{D}_{i}^{II}-D_{i}^{II}\right) ,
\end{align*}%
we have by (\ref{Delta1}) and Lemma \ref{Lem17}, uniformly 
\begin{equation*}
\dot{\Delta}_{i,u}=\left[ D_{i}^{I}\right] ^{-1}D_{i}^{II}+\left[ D_{i}^{I}%
\right] ^{-1}\frac{1}{T}\sum_{t=1}^{T}\left[ \mathbf{1}\left\{
\epsilon_{it}\leq 0\right\} -\mathbf{1}\left\{ \epsilon_{it}\leq \dot{\Delta}%
_{t,v}^{\prime}\Psi_{it}^{0}\right\} \right] \varpi_{it}^{0}+o_{p}\left(%
\left(N\vee T\right) ^{-1/2}\right)
\end{equation*}
uniformly over $i\in I_{3},$ where we use the fact that $\sqrt{\frac{\log
(N\vee T)}{T}}\xi_{N}\eta_{N}=o\left(\left(N\vee T\right)^{-1/2}\right) $ by
Assumption \ref{ass:1}(ix). $\blacksquare $

\subsection{Proof of Theorem \protect\ref{Thm3}}

In this section, we extend the distribution theory of the least squares
framework in \cite{chernozhukov2019inference} to the quantile regression
framework and obtain the uniform error bound. We assume the model has only
one regressor in this section for notation simplicity.

\subsubsection{Proof of Statement (i)}

For $\forall i\in I_{3}$, recall from (\ref{debias_1}) that 
\begin{equation}
\left\{ \hat{u}_{i,j}^{(3,1)}\right\}_{j\in [p]}=\argmin\limits_{\left\{
u_{i,j}\right\}_{j\in [p]}}\frac{1}{T}\sum_{t=1}^{T}\rho_{\tau }\left(\tilde{%
Y}_{it}-u_{i,0}^{\prime}\dot{v}_{t,0}^{(1)}-u_{i,1}^{\prime}\dot{v}%
_{t,1}^{(1)}\hat{e}_{1,it}\right) .  \label{C.1}
\end{equation}
where $\tilde{Y}_{it}=Y_{it}-\hat{\mu}_{1,it}\dot{u}_{i,1}^{(1)\prime}\dot{v}%
_{t,1}^{(1)}$. Let $\hat{\Delta}_{i,u}=%
\begin{bmatrix}
\hat{u}_{i,0}^{(3,1)}-O_{0}^{(1)}u_{i,0}^{0} \\ 
\hat{u}_{i,1}^{(3,1)}-O_{1}^{(1)}u_{i,1}^{0}%
\end{bmatrix}%
$ and $\dot{\omega}_{it}= 
\begin{bmatrix}
\dot{v}_{t,0}^{(1)} \\ 
\dot{v}_{t,1}^{(1)}\hat{e}_{1,it}%
\end{bmatrix}%
.$ With generic $(u_{i,0},u_{i,1},\mathrm{u}_{i,1})$, define%
\begin{align}
&\hat{\mathbb{H}}_{i}\left(u_{i,0},u_{i,1},\mathrm{u}_{i,1}\right)=\frac{1}{T%
}\sum_{t=1}^{T}\left[\tau-\mathbf{1}\left\{
\epsilon_{it}\leq\iota_{it}\left(u_{i,0},u_{i,1},\mathrm{u}_{i,1}\right)
,\right\} \right] \dot{\omega}_{it}\text{ and}  \notag \\
&\iota_{it}\left(u_{i,0},u_{i,1},\mathrm{u}_{i,1}\right)=u_{i,0}^{\prime}%
\dot{v}_{t,0}^{(1)}-u_{i,0}^{0\prime}v_{t,0}^{0}+\hat{\mu}_{1,it}\mathrm{u}%
_{i,1}^{\prime}\dot{v}_{t,1}^{(1)}-\mu_{1,it}u_{i,1}^{0\prime}v_{t,1}^{0}+%
\hat{e}_{1,it}u_{i,1}^{\prime}\dot{v}_{t,1}^{(1)}-e_{1,it}u_{i,1}^{0%
\prime}v_{t,1}^{0}.  \label{C.2}
\end{align}%
We can see that 
\begin{equation*}
\hat{\mathbb{H}}_{i}\left(u_{i,0},u_{i,1},\dot{u}_{i,1}^{(1)}\right) =\frac{1%
}{T}\sum_{t=1}^{T}\left[ \tau -\mathbf{1}\left\{ \tilde{Y}%
_{it}-u_{i,0}^{\prime}\dot{v}_{t,0}^{(1)}-u_{i,1}^{\prime}\dot{v}_{t,1}^{(1)}%
\hat{e}_{1,it}\leq 0\right\} \right] \dot{\omega}_{it}
\end{equation*}%
is the first order subgradient of (\ref{C.1}). In addition, we define 
\begin{align*}
\hat{\mathcal{H}}_{i}\left(u_{i,0},u_{i,1},\mathrm{u}_{i,1}\right) & =\frac{1%
}{T}\sum_{t=1}^{T}\mathbb{E}\left\{ \left[ \tau -\mathbf{1}\left\{
\epsilon_{it}\leq \left(\iota_{it}\left(u_{i,0},u_{i,1},\mathrm{u}%
_{i,1}\right)\right) \right\} \right] \dot{\omega}_{it}\bigg|\mathscr{D}%
_{e}^{I_{1}\cup I_{2}}\right\} \\
& =\frac{1}{T}\sum_{t=1}^{T}\left[ \tau
-F_{it}\left(\iota_{it}\left(u_{i,0},u_{i,1},\mathrm{u}_{i,1}\right) \right) %
\right] \dot{\omega}_{it},
\end{align*}%
and 
\begin{align*}
\hat{\mathbb{W}}_{i}\left(\hat{u}_{i,0}^{(3,1)},\hat{u}_{i,1}^{(3,1)},\dot{u}%
_{i,1}^{(1)}\right) & =\hat{\mathbb{H}}_{i}\left(\hat{u}_{i,0}^{(3,1)},\hat{u%
}_{i,1}^{(3,1)},\dot{u}_{i,1}^{(1)}\right) -\hat{\mathbb{H}}%
_{i}\left(O_{0}^{(1)}u_{i,0}^{0},O_{1}^{(1)}u_{i,1}^{0},\dot{u}%
_{i,1}^{(1)}\right) \\
& -\left\{ \hat{\mathcal{H}}_{i}\left(\hat{u}_{i,0}^{(3,1)},\hat{u}%
_{i,1}^{(3,1)},\dot{u}_{i,1}^{(1)}\right) -\hat{\mathcal{H}}%
_{i}\left(O_{0}^{(1)}u_{i,0}^{0},O_{1}^{(1)}u_{i,1}^{0},\dot{u}%
_{i,1}^{(1)}\right) \right\} ,
\end{align*}%
where $\mathscr{D}_{e}^{I_{1}\cup I_{2}}$ is the $\sigma$-field generated by 
$\left\{ \epsilon_{it}\right\}_{i\in I_{1}\cup I_{2},t\in [T]}\bigcup
\left\{ e_{it}\right\}_{i\in [N],t\in[T]}\bigcup\left\{
V_{j}^{0}\right\}_{j\in [p]\cup \{0\}}\bigcup \left\{W_{j}^{0}\right\}_{j\in
[p]}$. Then it is clear that 
\begin{align}
\hat{\mathbb{H}}_{i}\left(\hat{u}_{i,0}^{(3,1)},\hat{u}_{i,1}^{(3,1)},\dot{u}%
_{i,1}^{(1)}\right) & =\hat{\mathbb{H}}_{i}%
\left(O_{0}^{(1)}u_{i,0}^{0},O_{1}^{(1)}u_{i,1}^{0},\dot{u}%
_{i,1}^{(1)}\right) )+\hat{\mathcal{H}}_{i}\left(\hat{u}_{i,0}^{(3,1)},\hat{u%
}_{i,1}^{(3,1)},\dot{u}_{i,1}^{(1)}\right) -\hat{\mathcal{H}}%
_{i}\left(O_{0}^{(1)}u_{i,0}^{0},O_{1}^{(1)}u_{i,1}^{0},\dot{u}%
_{i,1}^{(1)}\right)  \notag  \label{C.3} \\
& +\hat{\mathbb{W}}_{i}\left(\hat{u}_{i,0}^{(3,1)},\hat{u}_{i,1}^{(3,1)},%
\dot{u}_{i,1}^{(1)}\right) .
\end{align}%
For specific $u_{i,0}$ and $u_{i,1}$, let $u_{i}=\left(u_{i,0}^{%
\prime},u_{i,1}^{\prime}\right) ^{\prime}$. Following similar arguments as
used in the proof of Lemma \ref{Lem16}, the second order partial derivative
of the function $\hat{\mathcal{H}}_{i}(\cdot)$ with respect to $u_{i}$ at
the true value can be shown to be bounded in probability. By Taylor
expansion, it yields 
\begin{equation}
\hat{\mathcal{H}}_{i}\left(\hat{u}_{i,0}^{(3,1)},\hat{u}_{i,1}^{(3,1)},\dot{u%
}_{i,1}^{(1)}\right) -\hat{\mathcal{H}}_{i}%
\left(O_{0}^{(1)}u_{i,0}^{0},O_{1}^{(1)}u_{i,1}^{0},\dot{u}%
_{i,1}^{(1)}\right) =\frac{\partial \hat{\mathcal{H}}_{i}%
\left(O_{0}^{(1)}u_{i,0}^{0},O_{1}^{(1)}u_{i,1}^{0},\dot{u}%
_{i,1}^{(1)}\right) }{\partial u_{i}^{\prime}}\hat{\Delta}_{i,u}+R_{i},
\label{C.4}
\end{equation}
where $\max_{i\in I_{3}}|R_{i}|\lesssim \max_{i\in I_{3}}\left\Vert \hat{%
\Delta}_{i,u}\right\Vert_{2}^{2}$ and 
\begin{equation*}
\frac{\partial \hat{\mathcal{H}}_{i}%
\left(O_{0}^{(1)}u_{i,0}^{0},O_{1}^{(1)}u_{i,1}^{0},\dot{u}%
_{i,1}^{(1)}\right) }{\partial u_{i}^{\prime}}=-\frac{1}{T}%
\sum_{t=1}^{T}f_{it}\left[ \iota_{it}%
\left(O_{0}^{(1)}u_{i,0}^{0},O_{1}^{(1)}u_{i,1}^{0},\dot{u}%
_{i,1}^{(1)}\right) \right] 
\begin{bmatrix}
\dot{v}_{t,0}^{(1)}\dot{v}_{t,0}^{(1)\prime} & \hat{e}_{1,it}\dot{v}%
_{t,0}^{(1)}\dot{v}_{t,1}^{(1)\prime} \\ 
\hat{e}_{1,it}\dot{v}_{t,1}^{(1)}\dot{v}_{t,0}^{(1)\prime} & \hat{e}%
_{1,it}^{2}\dot{v}_{t,1}^{(1)}\dot{v}_{t,1}^{(1)\prime}%
\end{bmatrix}%
=-\dot{D}_{i}^{F}.
\end{equation*}

Combing (\ref{C.3}) and (\ref{C.4}), we have 
\begin{align}
\hat{\Delta}_{i,u}& =\left(\dot{D}_{i}^{F}\right) ^{-1}\Bigg\{\hat{\mathbb{H}%
}_{i}\left(O_{0}^{(1)}u_{i,0}^{0},O_{1}^{(1)}u_{i,1}^{0},\dot{u}%
_{i,1}^{(1)}\right) +\hat{\mathbb{W}}_{i}\left(\hat{u}_{i,0}^{(3,1)},\hat{u}%
_{i,1}^{(3,1)},\dot{u}_{i,1}^{(1)}\right) -\hat{\mathbb{H}}_{i}\left(\hat{u}%
_{i,0}^{(3,1)},\hat{u}_{i,1}^{(3,1)},\dot{u}_{i,1}^{(1)}\right) +R_{i}\Bigg\}
\notag  \label{C.5} \\
& =\left(\dot{D}_{i}^{F}\right) ^{-1}\left\{ \hat{\mathbb{H}}%
_{i}\left(O_{0}^{(1)}u_{i,0}^{0},O_{1}^{(1)}u_{i,1}^{0},\dot{u}%
_{i,1}^{(1)}\right) +\hat{\mathbb{W}}_{i}\left(\hat{u}_{i,0}^{(3,1)},\hat{u}%
_{i,1}^{(3,1)},\dot{u}_{i,1}^{(1)}\right) +o_{p}\left(\left(N\vee T\right)
^{-1/2}\right) \right\} ,
\end{align}%
uniformly over $i\in I_{3}$, where the second line is due to the fact that 
\begin{equation*}
\max_{i\in I_{3}}\left\Vert \hat{\mathbb{H}}_{i}\left(\hat{u}_{i,0}^{(3,1)},%
\hat{u}_{i,1}^{(3,1)},\dot{u}_{i,1}^{(1)}\right)\right\Vert_{\max
}=O_{p}\left(\frac{\xi_{N}}{T}\right) \quad \text{and}\quad \max_{i\in
I_{3}}\left\Vert \hat{\Delta}_{i,u}\right\Vert_{2}^{2}=O_{p}\left(%
\eta_{N}^{2}\right) =o_{p}\left(\left(N\vee T\right) ^{-1/2}\right) ,
\end{equation*}%
following similar arguments as in (\ref{B.25}) and the proof of Theorem \ref%
{Thm2}.

Next, we analyze the term $\hat{\mathbb{H}}_{i}%
\left(O_{0}^{(1)}u_{i,0}^{0},O_{1}^{(1)}u_{i,1}^{0},\dot{u}%
_{i,1}^{(1)}\right)$ in (\ref{C.4}), which can be written as 
\begin{align}
& \hat{\mathbb{H}}_{i}\left(O_{0}^{(1)}u_{i,0}^{0},O_{1}^{(1)}u_{i,1}^{0},%
\dot{u}_{i,1}^{(1)}\right)  \notag  \label{C.6} \\
& =\frac{1}{T}\sum_{t=1}^{T}\dot{\omega}_{it}\left(\tau -\mathbf{1}%
\left\{\epsilon_{it}\leq
\iota_{it}\left(O_{0}^{(1)}u_{i,0}^{0},O_{1}^{(1)}u_{i,1}^{0},\dot{u}%
_{i,1}^{(1)}\right)\right\} \right)  \notag \\
& =\frac{1}{T}\sum_{t=1}^{T}\dot{\omega}_{it}\left(\tau -\mathbf{1}%
\left\{\epsilon_{it}\leq 0\right\} \right) +\frac{1}{T}\sum_{t=1}^{T}\dot{%
\omega}_{it}\left(\mathbf{1}\left\{ \epsilon_{it}\leq 0\right\} -\mathbf{1}%
\left\{ \epsilon_{it}\leq
\iota_{it}\left(O_{0}^{(1)}u_{i,0}^{0},O_{1}^{(1)}u_{i,1}^{0},\dot{u}%
_{i,1}^{(1)}\right)\right\} \right)  \notag \\
& =\frac{1}{T}\sum_{t=1}^{T}\dot{\omega}_{it}\left[ \tau -\mathbf{1}%
\left\{\epsilon_{it}\leq 0\right\} \right) +\frac{1}{T}\sum_{t=1}^{T}\dot{%
\omega}_{it}\left\{ F_{it}(0)-F_{it}\left[ \iota_{it}%
\left(O_{0}^{(1)}u_{i,0}^{0},O_{1}^{(1)}u_{i,1}^{0},\dot{u}%
_{i,1}^{(1)}\right) \right] \right\}  \notag \\
& +\frac{1}{T}\sum_{t=1}^{T}\dot{\omega}_{it}\left(\mathbf{1}%
\left\{\epsilon_{it}\leq 0\right\} -F_{it}(0)\right]  \notag \\
& -\frac{1}{T}\sum_{t=1}^{T}\dot{\omega}_{it}\left\{ \mathbf{1}%
\left\{\epsilon_{it}\leq
\iota_{it}\left(O_{0}^{(1)}u_{i,0}^{0},O_{1}^{(1)}u_{i,1}^{0},\dot{u}%
_{i,1}^{(1)}\right)\right\} -F_{it}\left[ \iota_{it}%
\left(O_{0}^{(1)}u_{i,0}^{0},O_{1}^{(1)}u_{i,1}^{0},\dot{u}%
_{i,1}^{(1)}\right) \right] \right\} .
\end{align}%
For the second term after the last equality, we notice that 
\begin{align}
& \frac{1}{T}\sum_{t=1}^{T}\dot{\omega}_{it}\left\{ F_{it}(0)-F_{it}\left[%
\iota_{it}\left(O_{0}^{(1)}u_{i,0}^{0},O_{1}^{(1)}u_{i,1}^{0},\dot{u}%
_{i,1}^{(1)}\right) \right] \right\}  \notag  \label{C.7} \\
& =-\frac{1}{T}\sum_{t=1}^{T}\dot{\omega}_{it}f_{it}\left(\tilde{\iota}%
_{it}\right) 
\begin{bmatrix}
(\dot{v}_{t,0}^{(1)}-O_{0}^{(1)}v_{t,0}^{0})^{\prime} & (\hat{e}_{1,it}\dot{v%
}_{t,1}^{(1)}-e_{1,it}O_{1}^{(1)}v_{t,1}^{0})^{\prime}%
\end{bmatrix}
\begin{bmatrix}
O_{0}^{(1)}u_{i,0}^{0} \\ 
O_{1}^{(1)}u_{i,1}^{0}%
\end{bmatrix}
\notag \\
& -\frac{1}{T}\sum_{t=1}^{T}\dot{\omega}_{it}f_{it}\left(\tilde{\iota}%
_{it}\right) \left(\hat{\mu}_{1,it}\dot{u}_{i,1}^{(1)\prime}\dot{v}%
_{t,1}^{(1)}-\mu_{1,it}u_{i,1}^{0\prime}v_{t,1}^{0}\right)  \notag \\
& :=\dot{D}_{i}^{J}%
\begin{bmatrix}
O_{0}^{(1)}u_{i,0}^{0} \\ 
O_{1}^{(1)}u_{i,1}^{0}%
\end{bmatrix}%
+\frac{1}{T}\sum_{t=1}^{T}\dot{\omega}_{it}f_{it}\left(\tilde{\iota}%
_{it}\right) \left(\mu_{1,it}u_{i,1}^{0\prime}v_{t,1}^{0}-\hat{\mu}_{1,it}%
\dot{u}_{i,1}^{(1)\prime}\dot{v}_{t,1}^{(1)}\right) ,
\end{align}%
where $\left\vert \tilde{\iota}_{it}\right\vert $ lies between $0$ and $%
\left\vert \iota_{it}\left(O_{0}^{(1)}u_{i,0}^{0},O_{1}^{(1)}u_{i,1}^{0},%
\dot{u}_{i,1}^{(1)}\right) \right\vert $ and 
\begin{equation*}
\dot{D}_{i}^{J}=\frac{1}{T}\sum_{t=1}^{T}f_{it}(\tilde{\iota}_{it}) 
\begin{bmatrix}
\dot{v}_{t,0}^{(1)}\left(O_{0}^{(1)}v_{t,0}^{0}-\dot{v}_{t,0}^{(1)}\right)^{%
\prime} & \dot{v}_{t,0}^{(1)}\left(e_{1,it}O_{1}^{(1)}v_{t,1}^{0}-\hat{e}%
_{it}\dot{v}_{t,1}^{(1)}\right) ^{\prime}\notag \\ 
\hat{e}_{1,it}\dot{v}_{t,1}^{(1)}\left(O_{0}^{(1)}v_{t,0}^{0}-\dot{v}%
_{t,0}^{(1)}\right) ^{\prime} & \hat{e}_{1,it}\dot{v}_{t,1}^{(1)}%
\left(e_{1,it}O_{1}^{(1)}v_{t,1}^{0}-\hat{e}_{1,it}\dot{v}%
_{t,1}^{(1)}\right) ^{\prime}%
\end{bmatrix}%
.
\end{equation*}
The first equality above is due to the mean-value theorem and the definition
for $\iota_{it}\left(O_{0}^{(1)}u_{i,0}^{0},O_{1}^{(1)}u_{i,1}^{0},\dot{u}%
_{i,1}^{(1)}\right) $ in (\ref{C.2}). Inserting (\ref{C.7}) into (\ref{C.6}%
), we obtain that 
\begin{align}
\hat{\mathbb{H}}_{i}\left(O_{0}^{(1)}u_{i,0}^{0},O_{1}^{(1)}u_{i,1}^{0},\dot{%
u}_{i,1}^{(1)}\right) & =\dot{D}_{i}^{J}%
\begin{bmatrix}
O_{0}^{(1)}u_{i,0}^{0} \\ 
O_{1}^{(1)}u_{i,1}^{0}%
\end{bmatrix}%
+\frac{1}{T}\sum_{t=1}^{T}\dot{\omega}_{it}\left[ \tau -\mathbf{1}%
\left\{\epsilon_{it}\leq 0\right\} \right)  \notag  \label{C.8} \\
& +\frac{1}{T}\sum_{t=1}^{T}\dot{\omega}_{it}f_{it}\left(\tilde{\iota}%
_{it}\right) \left(\mu_{1,it}u_{i,1}^{0\prime}v_{t,1}^{0}-\hat{\mu}_{1,it}%
\dot{u}_{i,1}^{(1)\prime}\dot{v}_{t,1}^{(1)}\right)  \notag \\
& +\frac{1}{T}\sum_{t=1}^{T}\dot{\omega}_{it}\Bigg\{\left[ \mathbf{1}%
\left\{\epsilon_{it}\leq 0\right\} -\mathbf{1}\left\{ \epsilon_{it}\leq
\iota_{it}\left(O_{0}^{(1)}u_{i,0}^{0},O_{1}^{(1)}u_{i,1}^{0},\dot{u}%
_{i,1}^{(1)}\right) \right\} \right]  \notag \\
& -\left(F_{it}(0)-F_{it}\left[ \iota_{it}%
\left(O_{0}^{(1)}u_{i,0}^{0},O_{1}^{(1)}u_{i,1}^{0},\dot{u}%
_{i,1}^{(1)}\right) \right] \right) \Bigg\}  \notag \\
& :=\dot{D}_{i}^{J} 
\begin{bmatrix}
O_{0}^{(1)}u_{i,0}^{0} \\ 
O_{1}^{(1)}u_{i,1}^{0}%
\end{bmatrix}%
+\hat{\mathbb{I}}_{1,i}+\hat{\mathbb{I}}_{2,i}+\hat{\mathbb{I}}_{3,i}.
\end{align}

Combining (\ref{C.5}) and (\ref{C.8}), we obtain that 
\begin{equation}
\hat{\Delta}_{i,u}=\left(\dot{D}_{i}^{F}\right) ^{-1}\left\{ \dot{D}_{i}^{J} 
\begin{bmatrix}
O_{0}^{(1)}u_{i,0}^{0} \\ 
O_{1}^{(1)}u_{i,1}^{0}%
\end{bmatrix}%
+\hat{\mathbb{I}}_{1,i}+\hat{\mathbb{I}}_{2,i}+\hat{\mathbb{I}}_{3,i}+\hat{%
\mathbb{W}}_{i}\left(\hat{u}_{i,0}^{(3,1)},\hat{u}_{i,1}^{(3,1)},\dot{u}%
_{i,1}^{(1)}\right) +o_{p}\left(\left(N\vee T\right) ^{-1/2}\right)\right\} ,
\label{C.9}
\end{equation}%
where the $o_{p}\left(\left(N\vee T\right) ^{-1/2}\right) $ term holds
uniformly over $i\in I_{3}$. To prove Theorem \ref{Thm3}(i), we analyze each
term in (\ref{C.9}) step by step.\medskip

\noindent\textbf{Step 1: Uniform Convergence for $\dot{D}_i^F$ and $\dot{D}%
_i^J$}. \newline
Define 
\begin{align*}
&D_{i}^{F}=\frac{1}{T}\sum_{t=1}^{T}f_{it}(0) 
\begin{bmatrix}
O_{0}^{(1)}v_{t,0}^{0}v_{t,0}^{0\prime}O_{0}^{(1)\prime} & 0 \\ 
0 & e_{1,it}^{2}O_{1}^{(1)}v_{t,1}^{0}v_{t,1}^{0\prime}O_{1}^{(1)\prime}%
\end{bmatrix}%
, \\
&D_{i}^{J}=\frac{1}{T}\sum_{t=1}^{T}f_{it}(0) 
\begin{bmatrix}
\dot{v}_{t,0}^{(1)}\left(O_{0}^{(1)}v_{t,0}^{0}-\dot{v}_{t,0}^{(1)}\right)^{%
\prime} & 0 \\ 
0 & e_{1,it}^{2}O_{1}^{(1)}v_{t,1}^{0}\left(O_{1}^{(1)}v_{t,1}^{0}-\dot{v}%
_{t,1}^{(1)}\right)^{\prime}%
\end{bmatrix}%
.
\end{align*}
Lemmas \ref{Lem21} and \ref{Lem22} show that 
\begin{align*}
& \max_{i \in I_3}||\dot{D}_{i}^{F}-D_{i}^{F}||_F=O_{p}(\eta_{N}) \quad 
\text{and} \quad \max_{i \in I_3}||\dot{D}_{i}^{J}-D_{i}^{J}||_F=\left\Vert%
\begin{bmatrix}
O_{p}(\eta_{N}^{2}) & O_{p}(\eta_{N}) \\ 
O_{p}(\eta_{N}^{2}) & O_{p}(\eta_{N}^{2})%
\end{bmatrix}
\right\Vert_{F},
\end{align*}
with $\eta_{N}=\frac{\sqrt{\log\left(N\vee T\right)}\xi_{N}^{2}}{\sqrt{%
N\wedge T}}$.

\noindent\textbf{Step 2: Uniform Convergence for $\hat{\mathbb{I}}_{1,i}$}.%
\newline
Let $\omega_{it}^{0}= 
\begin{bmatrix}
O_{0}^{(1)}v_{t,0}^{0} \\ 
O_{1}^{(1)}v_{t,1}^{0}e_{1,it}%
\end{bmatrix}%
$. Then we can see that 
\begin{equation*}
\hat{\mathbb{I}}_{1,i}=\frac{1}{T}\sum_{t=1}^{T}\omega_{it}^{0}\left(\tau -%
\mathbf{1}\left\{ \epsilon_{it}\leq 0\right\} \right) +\frac{1}{T}%
\sum_{t=1}^{T}\left(\dot{\omega}_{it}-\omega_{it}^{0}\right) \left(\tau -%
\mathbf{1}\left\{ \epsilon_{it}\leq 0\right\} \right) .
\end{equation*}
Noting that 
\begin{align}
\dot{\omega}_{it}-\omega_{it}^{0}& = 
\begin{bmatrix}
\dot{v}_{t,0}^{(1)}-O_{0}^{(1)}v_{t,0}^{0} \\ 
\dot{v}_{t,1}^{(1)}\hat{e}_{1,it}-O_{1}^{(1)}v_{t,1}^{0}e_{1,it}^{0}%
\end{bmatrix}
\notag  \label{C.11} \\
& = 
\begin{bmatrix}
\dot{v}_{t,0}^{(1)}-O_{0}^{(1)}v_{t,0}^{0} \\ 
\left(\dot{v}_{t,1}^{(1)}-O_{0}^{(1)}v_{t,1}^{0}\right) \left(\hat{e}%
_{1,it}-e_{1,it}\right) +e_{1,it}\left(\dot{v}%
_{t,1}^{(1)}-O_{0}^{(1)}v_{t,1}^{0}\right) +O_{1}^{(1)}v_{t,1}^{(1)}\left(%
\hat{e}_{1,it}-e_{1,it}\right)%
\end{bmatrix}%
,
\end{align}%
we have 
\begin{equation}
\max_{i\in I_{3}}\frac{1}{T}\sum_{t\in[T]}\left\Vert \dot{\omega}%
_{it}-\omega_{it}^{0}\right\Vert_{2}=O_{p}(\eta_{N}).  \label{C.12}
\end{equation}%
In addition, $\max_{i\in I_{3}}\left\Vert \frac{1}{T}\sum_{t=1}^{T}\left(%
\dot{\omega}_{it}-\omega_{it}^{0}\right) \left(\tau -\mathbf{1}%
\left\{\epsilon_{it}\leq 0\right\} \right)
\right\Vert_{2}=O_{p}\left(\eta_{N}^{2}\right) $ by Lemma \ref{Lem23}. It
follows that 
\begin{equation}
\hat{\mathbb{I}}_{1,i}=\frac{1}{T}\sum_{t=1}^{T}\omega_{it}^{0}\left(\tau -%
\mathbf{1}\left\{ \epsilon_{it}\leq 0\right\} \right)
+O_{p}\left(\eta_{N}^{2}\right) ,  \label{C.13}
\end{equation}%
uniformly over $i\in I_{3}$.

\noindent\textbf{Step 3: Uniform Convergence for $\hat{\mathbb{I}}_{2,i}$}.%
\newline
Note that 
\begin{align}
\hat{\mathbb{I}}_{2,i}& =\frac{1}{T}\sum_{t=1}^{T}\dot{\omega}%
_{it}f_{it}\left(\tilde{\iota}_{it}\right)
\left(\mu_{1,it}u_{i,1}^{0\prime}v_{t,1}^{0}-\hat{\mu}_{1,it}\dot{u}%
_{i,1}^{(1)\prime}\dot{v}_{t,1}^{(1)}\right)  \notag  \label{I2} \\
& =\frac{1}{T}\sum_{t=1}^{T}\omega_{it}^{0}f_{it}(0)\left(%
\mu_{1,it}u_{i,1}^{0\prime}v_{t,1}^{0}-\hat{\mu}_{1,it}\dot{u}%
_{i,1}^{(1)\prime}\dot{v}_{t,1}^{(1)}\right)  \notag \\
& +\frac{1}{T}\sum_{t=1}^{T}\left(\dot{\omega}_{it}-\omega_{it}^{0}\right) %
\left[ f_{it}\left(\tilde{\iota}_{it}\right) -f_{it}(0)\right]
\left(\mu_{1,it}u_{i,1}^{0\prime}v_{t,1}^{0}-\hat{\mu}_{1,it}\dot{u}%
_{i,1}^{(1)\prime}\dot{v}_{t,1}^{(1)}\right)  \notag \\
& +\frac{1}{T}\sum_{t=1}^{T}\left(\dot{\omega}_{it}-\omega_{it}^{0}%
\right)f_{it}(0)\left(\mu_{1,it}u_{i,1}^{0\prime}v_{t,1}^{0}-\hat{\mu}_{1,it}%
\dot{u}_{i,1}^{(1)\prime}\dot{v}_{t,1}^{(1)}\right)  \notag \\
& +\frac{1}{T}\sum_{t=1}^{T}\omega_{it}^{0}\left[ f_{it}\left(\tilde{\iota}%
_{it}\right) -f_{it}(0)\right] \left(\mu_{1,it}u_{i,1}^{0\prime}v_{t,1}^{0}-%
\hat{\mu}_{1,it}\dot{u}_{i,1}^{(1)\prime}\dot{v}_{t,1}^{(1)}\right)
\end{align}%
where 
\begin{align*}
& \max_{i\in I_{3}}\left\Vert \frac{1}{T}\sum_{t=1}^{T}\left(\dot{\omega}%
_{it}-\omega_{it}^{0}\right) \left[ f_{it}\left(\tilde{\iota}%
_{it}\right)-f_{it}(0)\right] \left(\mu_{1,it}u_{i,1}^{0\prime}v_{t,1}^{0}-%
\hat{\mu}_{1,it}\dot{u}_{i,1}^{(1)\prime}\dot{v}_{t,1}^{(1)}\right)
\right\Vert_{2}=O_{p}(\eta_{N}^{3}), \\
& \max_{i\in I_{3}}\left\Vert \frac{1}{T}\sum_{t=1}^{T}\left(\dot{\omega}%
_{it}-\omega_{it}^{0}\right)
f_{it}(0)\left(\mu_{1,it}u_{i,1}^{0\prime}v_{t,1}^{0}-\hat{\mu}_{1,it}\dot{u}%
_{i,1}^{(1)\prime}\dot{v}_{t,1}^{(1)}\right)
\right\Vert_{2}=O_{p}(\eta_{N}^{2}), \\
& \max_{i\in I_{3}}\left\Vert \frac{1}{T}\sum_{t=1}^{T}\omega_{it}^{0}\left[%
f_{it}\left(\tilde{\iota}_{it}\right) -f_{it}(0)\right] \left(%
\mu_{1,it}u_{i,1}^{0\prime}v_{t,1}^{0}-\hat{\mu}_{1,it}\dot{u}%
_{i,1}^{(1)\prime}\dot{v}_{t,1}^{(1)}\right)
\right\Vert_{2}=O_{p}(\eta_{N}^{2}).
\end{align*}%
To see why these three equalities hold, we focus on the third one. By
Cauchy's inequality, Theorem \ref{Thm2}, and Lemma \ref{Lem20}, we have 
\begin{align*}
& \max_{i\in I_{3}}\left\Vert \frac{1}{T}\sum_{t=1}^{T}\omega_{it}^{0}\left[%
f_{it}\left(\tilde{\iota}_{it}\right) -f_{it}(0)\right] \left(%
\mu_{1,it}u_{i,1}^{0\prime}v_{t,1}^{0}-\hat{\mu}_{1,it}\dot{u}%
_{i,1}^{(1)\prime}\dot{v}_{t,1}^{(1)}\right) \right\Vert_{2} \\
& \lesssim \sqrt{\max_{i\in I_{3}}\frac{1}{T}\sum_{t\in [T]}\left\Vert
\omega_{it}^{0}\right\Vert_{2}^{2}\left\vert \tilde{\iota}_{it}\right\vert
^{2}}\sqrt{\max_{i\in I_{3}}\frac{1}{T}\sum_{t\in [T]}\left\vert
\mu_{1,it}u_{i,1}^{0\prime}v_{t,1}^{0}-\hat{\mu}_{1,it}\dot{u}%
_{i,1}^{(1)\prime}\dot{v}_{t,1}^{(1)}\right\vert ^{2}}=O_{p}\left(%
\eta_{N}^{2}\right) .
\end{align*}%
For the first term on the right hand side (RHS) of the second equality of (%
\ref{I2}), we have by Lemma \ref{Lem24} 
\begin{align*}
& \max_{i\in I_{3}}\left\Vert \frac{1}{T}%
\sum_{t=1}^{T}O_{0}^{(1)}v_{t,0}^{0}f_{it}(0)\left(\mu_{1,it}u_{i,1}^{0%
\prime}v_{t,1}^{0}-\hat{\mu}_{1,it}\dot{u}_{i,1}^{(1)\prime}\dot{v}%
_{t,1}^{(1)}\right) \right\Vert_{2}=O_{p}(\eta_{N}),\text{ } \\
& \max_{i\in I_{3}}\left\Vert \frac{1}{T}%
\sum_{t=1}^{T}e_{1,it}O_{1}^{(1)}v_{t,1}^{0}f_{it}(0)\left(%
\mu_{1,it}u_{i,1}^{0\prime}v_{t,1}^{0}-\hat{\mu}_{1,it}\dot{u}%
_{i,1}^{(1)\prime}\dot{v}_{t,1}^{(1)}\right)
\right\Vert_{2}=O_{p}\left(\eta_{N}^{2}\right) ,
\end{align*}%
and thus $\max_{i\in I_{3}}\left\Vert \hat{\mathbb{I}}_{2,i}\right\Vert_{2}=%
\left\Vert 
\begin{bmatrix}
O_{p}(\eta_{N}) \\ 
O_{p}(\eta_{N}^{2})%
\end{bmatrix}%
\right\Vert_{2}$.

\noindent\textbf{Step 4: Uniform Convergence for $\hat{\mathbb{I}}_{3,i}$}.%
\newline
Note that 
\begin{align}
\hat{\mathbb{I}}_{3,i}& =\frac{1}{T}\sum_{t=1}^{T}\dot{\omega}_{it}\Bigg\{%
\left[ \mathbf{1}\left\{ \epsilon_{it}\leq 0\right\} -\mathbf{1}%
\left\{\epsilon_{it}\leq
\iota_{it}\left(O_{0}^{(1)}u_{i,0}^{0},O_{1}^{(1)}u_{i,1}^{0},\dot{u}%
_{i,1}^{(1)}\right)\right\} \right]  \notag  \label{step4_1} \\
& -\left(F_{it}(0)-F_{it}\left[ \iota_{it}%
\left(O_{0}^{(1)}u_{i,0}^{0},O_{1}^{(1)}u_{i,1}^{0},\dot{u}%
_{i,1}^{(1)}\right) \right] \right) \Bigg\}  \notag \\
& =\frac{1}{T}\sum_{t=1}^{T}\omega_{it}^{0}\Bigg\{\left[ \mathbf{1}%
\left\{\epsilon_{it}\leq 0\right\} -\mathbf{1}\left\{ \epsilon_{it}\leq
\iota_{it}\left(O_{0}^{(1)}u_{i,0}^{0},O_{1}^{(1)}u_{i,1}^{0},\dot{u}%
_{i,1}^{(1)}\right) \right\} \right]  \notag \\
& -\left(F_{it}(0)-F_{it}\left[ \iota_{it}%
\left(O_{0}^{(1)}u_{i,0}^{0},O_{1}^{(1)}u_{i,1}^{0},\dot{u}%
_{i,1}^{(1)}\right) \right] \right) \Bigg\}  \notag \\
& +\frac{1}{T}\sum_{t=1}^{T}\left(\dot{\omega}_{it}-\omega_{it}^{0}\right) %
\Bigg\{\left[ \mathbf{1}\left\{ \epsilon_{it}\leq 0\right\} -\mathbf{1}%
\left\{ \epsilon_{it}\leq
\iota_{it}\left(O_{0}^{(1)}u_{i,0}^{0},O_{1}^{(1)}u_{i,1}^{0},\dot{u}%
_{i,1}^{(1)}\right)\right\} \right]  \notag \\
& -\left(F_{it}(0)-F_{it}\left[ \iota_{it}%
\left(O_{0}^{(1)}u_{i,0}^{0},O_{1}^{(1)}u_{i,1}^{0},\dot{u}%
_{i,1}^{(1)}\right) \right] \right) \Bigg\}.
\end{align}%
By (\ref{Lem25.3}), we can show that 
\begin{align}
&\iota_{it}\left(O_{0}^{(1)}u_{i,0}^{0},O_{1}^{(1)}u_{i,1}^{0},\dot{u}%
_{i,1}^{(1)}\right) \leq R_{\iota ,it}^{1}\left(\left\vert
\mu_{1,it}\right\vert +\left\vert e_{1,it}\right\vert \right)
+R_{\iota,it}^{2}\quad \text{with}  \notag  \label{step4_2} \\
&\max_{i\in I_{3},t\in[T]}\left\vert R_{\iota
,it}^{1}\right\vert=O_{p}(\eta_{N}),\max_{i\in I_{3},t\in[T]}\left\vert
R_{\iota,it}^{2}\right\vert =O_{p}(\eta_{N}).
\end{align}%
For the second term on the RHS of the second equality in (\ref{step4_1}), we
notice that 
\begin{align}
& \max_{i\in I_{3}}\bigg\Vert\frac{1}{T}\sum_{t=1}^{T}\left(\dot{\omega}%
_{it}-\omega_{it}^{0}\right) \Bigg\{\left[ \mathbf{1}\left\{
\epsilon_{it}\leq 0\right\} -\mathbf{1}\left\{ \epsilon_{it}\leq
\iota_{it}\left(O_{0}^{(1)}u_{i,0}^{0},O_{1}^{(1)}u_{i,1}^{0},\dot{u}%
_{i,1}^{(1)}\right)\right\} \right]  \notag  \label{step4_3} \\
& -\left(F_{it}(0)-F_{it}\left[ \iota_{it}%
\left(O_{0}^{(1)}u_{i,0}^{0},O_{1}^{(1)}u_{i,1}^{0},\dot{u}%
_{i,1}^{(1)}\right) \right] \right) \Bigg\}\Bigg\Vert_{2}  \notag \\
& \leq \max_{i\in I_{3}}\frac{1}{T}\sum_{t=1}^{T}\left\Vert \dot{\omega}%
_{it}-\omega_{it}^{0}\right\Vert_{2}\mathbf{1}\left\{ 0\leq
\left\vert\epsilon_{it}\right\vert \leq \left\vert
\iota_{it}\left(O_{0}^{(1)}u_{i,0}^{0},O_{1}^{(1)}u_{i,1}^{0},\dot{u}%
_{i,1}^{(1)}\right)\right\vert \right\}  \notag \\
& +\max_{i\in I_{3}}\frac{1}{T}\sum_{t=1}^{T}\left\Vert \dot{\omega}%
_{it}-\omega_{it}^{0}\right\Vert_{2}\left\vert F_{it}\left[
\iota_{it}\left(O_{0}^{(1)}u_{i,0}^{0},O_{1}^{(1)}u_{i,1}^{0},\dot{u}%
_{i,1}^{(1)}\right) \right] -F_{it}(0)\right\vert  \notag \\
& \leq \max_{i\in I_{3}}\frac{1}{T}\sum_{t=1}^{T}\left\Vert \dot{\omega}%
_{it}-\omega_{it}^{0}\right\Vert_{2}\mathbf{1}\left\{ 0\leq
\left\vert\epsilon_{it}\right\vert \leq \left\vert
\iota_{it}\left(O_{0}^{(1)}u_{i,0}^{0},O_{1}^{(1)}u_{i,1}^{0},\dot{u}%
_{i,1}^{(1)}\right)\right\vert \right\} +O_{p}(\eta_{N}^{2}),
\end{align}%
where the last line is by (\ref{C.11}), (\ref{step4_2}) and Assumption \ref%
{ass:1}(iv).

Define the event $\mathscr{A}_{9}(M):=\left\{ \max_{i\in I_{3},t\in
[T]}\left\vert \iota_{it}\left(O_{0}^{(1)}u_{i,0}^{0},O_{1}^{(1)}u_{i,1}^{0},%
\dot{u}_{i,1}^{(1)}\right)\right\vert \leq M\eta_{N}\left(\left\vert
\mu_{1,it}\right\vert+\left\vert e_{1,it}\right\vert +1\right) \right\} $
with $\mathbb{P}\left\{\mathscr{A}_{9}^{c}(M)\right\} \leq e$ for any $e>0$.
Then for a large enough constant $C_{17}$, we have 
\begin{align}
& \mathbb{P}\left(\max_{i\in I_{3}}\frac{1}{T}\sum_{t=1}^{T}\left\Vert \dot{%
\omega}_{it}-\omega_{it}^{0}\right\Vert_{2}\mathbf{1}\left\{ 0\leq\left\vert
\epsilon_{it}\right\vert \leq \left\vert
\iota_{it}\left(O_{0}^{(1)}u_{i,0}^{0},O_{1}^{(1)}u_{i,1}^{0},\dot{u}%
_{i,1}^{(1)}\right)\right\vert \right\} >C_{17}\eta_{N}^{2}\right)  \notag
\label{step4_4} \\
& \leq \mathbb{P}\left(\max_{i\in I_{3}}\frac{1}{T}\sum_{t=1}^{T}\left\Vert 
\dot{\omega}_{it}-\omega_{it}^{0}\right\Vert_{2}\mathbf{1}\left\{
0\leq\left\vert \epsilon_{it}\right\vert \leq \left\vert
\iota_{it}\left(O_{0}^{(1)}u_{i,0}^{0},O_{1}^{(1)}u_{i,1}^{0},\dot{u}%
_{i,1}^{(1)}\right)\right\vert \right\} >C_{17}\eta_{N}^{2}\bigg|\mathscr{A}%
_{9}(M)\right) +e  \notag \\
& \leq \mathbb{P}\left(\max_{i\in I_{3}}\frac{1}{T}\sum_{t=1}^{T}\left\Vert 
\dot{\omega}_{it}-\omega_{it}^{0}\right\Vert_{2}\mathbf{1}\left\{
0\leq\left\vert \epsilon_{it}\right\vert \leq M\eta_{N}\left(\left\vert
\mu_{1,it}\right\vert +\left\vert e_{1,it}\right\vert +1\right)
\right\}>C_{17}\eta_{N}^{2}\right) +e  \notag \\
& \leq \mathbb{E}\Bigg\{\mathbb{P}\bigg(\max_{i\in I_{3}}\frac{1}{T}%
\sum_{t=1}^{T}\left\Vert \dot{\omega}_{it}-\omega_{it}^{0}\right\Vert_{2}%
\bigg[\mathbf{1}_{it}-\bar{\mathbf{1}}_{it}\bigg]>\frac{C_{17}\eta_{N}^{2}}{2%
}\bigg|\mathscr{D}_{e}^{I_{1}\cup I_{2}}\bigg)\Bigg\}  \notag \\
& +\mathbb{E}\left\{ \mathbb{P}\left(\max_{i\in I_{3}}\frac{1}{T}%
\sum_{t=1}^{T}\mathbb{E}\left\{ \left\Vert \dot{\omega}_{it}-\omega_{it}^{0}%
\right\Vert_{2}\mathbf{1}\left\{ 0\leq \left\vert \epsilon_{it}\right\vert
\leq M\eta_{N}\left(\left\vert \mu_{1,it}\right\vert+\left\vert
e_{1,it}\right\vert +1\right) \right\} \bigg|\mathscr{D}_{e}^{I_{1}\cup
I_{2}}\right\} >\frac{C_{17}\eta_{N}^{2}}{2}\right)\right\} +e  \notag \\
& =o(1)+e
\end{align}%
where $\mathbf{1}_{it}-\bar{\mathbf{1}}_{it}:=\mathbf{1}\left\{
0\leq\left\vert \epsilon_{it}\right\vert \leq M\eta_{N}\left(\left\vert
\mu_{1,it}\right\vert +\left\vert e_{1,it}\right\vert +1\right) \right\} -%
\mathbb{E}\left(\mathbf{1}\left\{ 0\leq \left\vert \epsilon_{it}\right\vert
\leq M\eta_{N}\left(\left\vert \mu_{1,it}\right\vert+\left\vert
e_{1,it}\right\vert +1\right) \right\} \bigg|\mathscr{D}_{e}^{I_{1}\cup
I_{2}}\right) $, the last line holds by the fact that $\max_{i\in I_{3}}%
\frac{1}{T}\sum_{t=1}^{T}\mathbb{E}\left\{ \left\Vert \dot{\omega}%
_{it}-\omega_{it}^{0}\right\Vert_{2}\mathbf{1}\left\{ 0\leq\left\vert
\epsilon_{it}\right\vert \leq M\eta_{N}\left(\left\vert
\mu_{1,it}\right\vert +\left\vert e_{1,it}\right\vert +1\right) \right\} %
\bigg|\mathscr{D}_{e}^{I_{1}\cup I_{2}}\right\} $ $=O_{p}(\eta_{N}^{2})$ and 
\begin{align*}
& \mathbb{P}\left(\max_{i\in I_{3}}\frac{1}{T}\sum_{t=1}^{T}\left\Vert \dot{%
\omega}_{it}-\omega_{it}^{0}\right\Vert_{2}\left[ \mathbf{1}_{it}-\bar{%
\mathbf{1}}_{it}\right] >C_{17}\eta_{N}^{2}\bigg|\mathscr{D}_{e}^{I_{1}\cup
I_{2}}\right) \\
& \lesssim \exp \left(-\frac{T^{2}\eta_{N}^{4}}{T\xi_{N}^{2}\eta_{N}^{2}+T%
\eta_{N}^{3}\xi_{N}\log T\log \log T}\right) =o(1)
\end{align*}%
by Bernstein's inequality in Lemma \ref{Lem:Bern}(i). Combining (\ref%
{step4_3}) and (\ref{step4_4}), we have shown the second term on the RHS of
the second equality in (\ref{step4_1}) is $O_{p}(\eta_{N}^{2})$ uniformly
over $i\in I_{3}$. This result, in conjunction with Lemma \ref{Lem24}, (\ref%
{step4_1}) and Assumption \ref{ass:1}(ix), implies that $\max_{i\in
I_{3}}\left\Vert \hat{\mathbb{I}}_{3,i}\right\Vert =o_{p}\left(\left(N\vee
T\right) ^{-\frac{1}{2}}\right) .$

\noindent\textbf{Step 5: Uniform Convergence for $\hat{\mathbb{W}}_{i}\left(%
\hat{u}_{i,0}^{(3,1)},\hat{u}_{i,1}^{(3,1)},\dot{u}_{i,1}^{(1)}\right) $}.%
\newline
Note that 
\begin{align*}
& \hat{\mathbb{W}}_{i}\left(\hat{u}_{i,0}^{(3,1)},\hat{u}_{i,1}^{(3,1)},\dot{%
u}_{i,1}^{(1)}\right) \\
& =\frac{1}{T}\sum_{t=1}^{T}\dot{\omega}_{it}\left(\mathbf{1}%
\left\{\epsilon_{it}\leq
\iota_{it}\left(O_{0}^{(1)}u_{i,0}^{0},O_{1}^{(1)}u_{i,1}^{0},\dot{u}%
_{i,1}^{(1)}\right)\right\} -\mathbf{1}\left\{ \epsilon_{it}\leq
\iota_{it}\left(\hat{u}_{i,0}^{(3,1)},\hat{u}_{i,1}^{(3,1)},\dot{u}%
_{i,1}^{(1)}\right) \right\} \right) \\
& -\frac{1}{T}\sum_{t=1}^{T}\dot{\omega}_{it}\left(F_{it}\left[
\iota_{it}\left(O_{0}^{(1)}u_{i,0}^{0},O_{1}^{(1)}u_{i,1}^{0},\dot{u}%
_{i,1}^{(1)}\right) \right] -F_{it}\left[ \iota_{it}\left(\hat{u}%
_{i,0}^{(3,1)},\hat{u}_{i,1}^{(3,1)},\dot{u}_{i,1}^{(1)}\right) \right]
\right) \\
& =\frac{1}{T}\sum_{t=1}^{T}\left(\dot{\omega}_{it}-\omega_{it}^{0}\right)%
\left(\mathbf{1}\left\{ \epsilon_{it}\leq
\iota_{it}\left(O_{0}^{(1)}u_{i,0}^{0},O_{1}^{(1)}u_{i,1}^{0},\dot{u}%
_{i,1}^{(1)}\right)\right\} -\mathbf{1}\left\{ \epsilon_{it}\leq
\iota_{it}\left(\hat{u}_{i,0}^{(3,1)},\hat{u}_{i,1}^{(3,1)},\dot{u}%
_{i,1}^{(1)}\right) \right\} \right) \\
& -\frac{1}{T}\sum_{t=1}^{T}\left(\dot{\omega}_{it}-\omega_{it}^{0}\right)%
\left(F_{it}\left[ \iota_{it}%
\left(O_{0}^{(1)}u_{i,0}^{0},O_{1}^{(1)}u_{i,1}^{0},\dot{u}%
_{i,1}^{(1)}\right) \right] -F_{it}\left[ \iota_{it}\left(\hat{u}%
_{i,0}^{(3,1)},\hat{u}_{i,1}^{(3,1)},\dot{u}_{i,1}^{(1)}\right) \right]
\right) \\
& +\frac{1}{T}\sum_{t=1}^{T}\omega_{it}^{0}\Bigg\{\left(\mathbf{1}%
\left\{\epsilon_{it}\leq
\iota_{it}\left(O_{0}^{(1)}u_{i,0}^{0},O_{1}^{(1)}u_{i,1}^{0},\dot{u}%
_{i,1}^{(1)}\right)\right\} -\mathbf{1}\left\{ \epsilon_{it}\leq
\iota_{it}\left(\hat{u}_{i,0}^{(3,1)},\hat{u}_{i,1}^{(3,1)},\dot{u}%
_{i,1}^{(1)}\right) \right\} \right) \\
& -\left(F_{it}\left[ \iota_{it}%
\left(O_{0}^{(1)}u_{i,0}^{0},O_{1}^{(1)}u_{i,1}^{0},\dot{u}%
_{i,1}^{(1)}\right) \right] -F_{it}\left[ \iota_{it}\left(\hat{u}%
_{i,0}^{(3,1)},\hat{u}_{i,1}^{(3,1)},\dot{u}_{i,1}^{(1)}\right) \right]
\right) \Bigg\} \\
& :=\frac{1}{T}\sum_{t=1}^{T}\left(\dot{\omega}_{it}-\omega_{it}^{0}\right)
\left(\mathbf{1}\left\{ \epsilon_{it}\leq
\iota_{it}\left(O_{0}^{(1)}u_{i,0}^{0},O_{1}^{(1)}u_{i,1}^{0},\dot{u}%
_{i,1}^{(1)}\right) \right\} -\mathbf{1}\left\{ \epsilon_{it}\leq
\iota_{it}\left(\hat{u}_{i,0}^{(3,1)},\hat{u}_{i,1}^{(3,1)},\dot{u}%
_{i,1}^{(1)}\right) \right\} \right) \\
& -\frac{1}{T}\sum_{t=1}^{T}\left(\dot{\omega}_{it}-\omega_{it}^{0}\right)%
\left(F_{it}\left[ \iota_{it}%
\left(O_{0}^{(1)}u_{i,0}^{0},O_{1}^{(1)}u_{i,1}^{0},\dot{u}%
_{i,1}^{(1)}\right) \right] -F_{it}\left[ \iota_{it}\left(\hat{u}%
_{i,0}^{(3,1)},\hat{u}_{i,1}^{(3,1)},\dot{u}_{i,1}^{(1)}\right) \right]
\right) \\
& +\hat{\mathbb{W}}_{i}^{I}-\hat{\mathbb{W}}_{i}^{II},
\end{align*}%
where we define 
\begin{align*}
\hat{\mathbb{W}}_{i}^{I}& =\frac{1}{T}\sum_{t=1}^{T}\omega_{it}^{0}\Bigg\{%
\mathbf{1}\left\{ \epsilon_{it}\leq
\iota_{it}\left(O_{0}^{(1)}u_{i,0}^{0},O_{1}^{(1)}u_{i,1}^{0},\dot{u}%
_{i,1}^{(1)}\right)\right\} -\mathbf{1}\left\{ \epsilon_{it}\leq 0\right\} \\
& -\bigg\{F_{it}\left[ \iota_{it}%
\left(O_{0}^{(1)}u_{i,0}^{0},O_{1}^{(1)}u_{i,1}^{0},\dot{u}%
_{i,1}^{(1)}\right) \right] -F_{it}(0)\bigg\}\Bigg\}\quad \text{and} \\
\hat{\mathbb{W}}_{i}^{II}& =\frac{1}{T}\sum_{t=1}^{T}\omega_{it}^{0}\Bigg\{%
\mathbf{1}\left\{ \epsilon_{it}\leq 0\right\} -\mathbf{1}\left\{
\epsilon_{it}\leq \iota_{it}\left(\hat{u}_{i,0}^{(3,1)},\hat{u}%
_{i,1}^{(3,1)},\dot{u}_{i,1}^{(1)}\right) \right\} \\
& -\bigg\{F_{it}(0)-F_{it}\left[ \iota_{it}\left(\hat{u}_{i,0}^{(3,1)},\hat{u%
}_{i,1}^{(3,1)},\dot{u}_{i,1}^{(1)},\right) \right] \bigg\}\Bigg\}.
\end{align*}%
We first observe that 
\begin{align}
& \iota_{it}\left(O_{0}^{(1)}u_{i,0}^{0},O_{1}^{(1)}u_{i,1}^{0},\dot{u}%
_{i,1}^{(1)}\right) -\iota_{it}\left(\hat{u}_{i,0}^{(3,1)},\hat{u}%
_{i,1}^{(3,1)},\dot{u}_{i,1}^{(1)}\right)  \notag  \label{step5_1} \\
& =\left(\hat{u}_{i,0}^{(3,1)}-O_{0}^{(1)}u_{i,0}^{0}\right) ^{\prime}\dot{v}%
_{t,0}^{(1)}+\left(\hat{u}_{i,1}^{(3,1)}-O_{1}^{(1)}u_{i,1}^{0}\right)
^{\prime}\dot{v}_{t,1}^{(1)}\hat{e}_{1,it}  \notag \\
& =\left(O_{0}^{(1)\prime}\hat{u}_{i,0}^{(3,1)}-u_{i,0}^{0}\right)
^{\prime}v_{t,0}^{0}+\left(O_{1}^{(1)\prime}\hat{u}%
_{i,1}^{(3,1)}-u_{i,1}^{0}\right)^{\prime}v_{t,1}^{(1)}e_{1,it}+O_{p}(%
\eta_{N}^{2})  \notag \\
& =R_{\iota ,it}^{3}e_{1,it}+R_{\iota ,it}^{4}
\end{align}%
such that $\max_{i\in I_{3},t\in[T]}\left\vert R_{\iota,it}^{3}\right\vert
=O_{p}(\eta_{N})$ and $\max_{i\in I_{3},t\in [T]}\left\vert R_{\iota
,it}^{4}\right\vert =O_{p}(\eta_{N})$. As in Step 4, we can show that 
\begin{equation*}
\max_{i\in I_{3}}\left\Vert \frac{1}{T}\sum_{t=1}^{T}\left(\dot{\omega}%
_{it}-\omega_{it}^{0}\right) \left(\mathbf{1}\left\{
\epsilon_{it}\leq\iota_{it}%
\left(O_{0}^{(1)}u_{i,0}^{0},O_{1}^{(1)}u_{i,1}^{0},\dot{u}%
_{i,1}^{(1)}\right) \right\} -\mathbf{1}\left\{ \epsilon_{it}\leq
\iota_{it}\left(\hat{u}_{i,0}^{(3,1)},\hat{u}_{i,1}^{(3,1)},\dot{u}%
_{i,1}^{(1)}\right) \right\} \right) \right\Vert =O_{p}(\eta_{N}^{2})
\end{equation*}%
and 
\begin{equation*}
\max_{i\in I_{3}}\left\Vert \frac{1}{T}\sum_{t=1}^{T}\left(\dot{\omega}%
_{it}-\omega_{it}^{0}\right) \left(F_{it}\left[ \iota_{it}%
\left(O_{0}^{(1)}u_{i,0}^{0},O_{1}^{(1)}u_{i,1}^{0},\dot{u}%
_{i,1}^{(1)}\right) \right] -F_{it}\left[ \iota_{it}\left(\hat{u}%
_{i,0}^{(3,1)},\hat{u}_{i,1}^{(3,1)},\dot{u}_{i,1}^{(1)}\right) \right]
\right) \right\Vert =O_{p}(\eta_{N}^{2}).
\end{equation*}%
Then by Lemma \ref{Lem25}, Lemma \ref{Lem26} and Assumption \ref{ass:1}(ix),
we obtain that 
\begin{equation*}
\max_{i\in I_{3}}\left\Vert \hat{\mathbb{W}}_{i}\left(\hat{u}_{i,0}^{(3,1)},%
\hat{u}_{i,1}^{(3,1)},\dot{u}_{i,1}^{(1)}\right)
\right\Vert_{2}=o_{p}\left(\left(N\vee T\right) ^{-\frac{1}{2}}\right) .
\end{equation*}

\noindent\textbf{Step 6: Distribution Theory for $\hat{\Delta}_{i,u}$}%
\newline
Combining the above results, we have that uniformly over $i\in I_{3}$, 
\begin{align}
\begin{bmatrix}
\hat{u}_{i,0}^{(3,1)} \\ 
\hat{u}_{i,1}^{(3,1)}%
\end{bmatrix}
& = 
\begin{bmatrix}
O_{0}^{(1)}u_{i,0}^{0} \\ 
O_{1}^{(1)}u_{i,1}^{0}%
\end{bmatrix}%
+\left(\dot{D}_{i}^{F}\right) ^{-1}\left\{ \dot{D}_{i}^{J}%
\begin{bmatrix}
O_{0}^{(1)}u_{i,0}^{0} \\ 
O_{1}^{(1)}u_{i,1}^{0}%
\end{bmatrix}%
+\hat{\mathbb{I}}_{1,i}+\hat{\mathbb{I}}_{2,i}+\mathbb{I}_{3,i}+\hat{\mathbb{%
W}}_{i}\left(\hat{u}_{i,0}^{(3,1)},\hat{u}_{i,1}^{(3,1)},\dot{u}%
_{i,1}^{(1)}\right) +o_{p}\left(\left(N\vee T\right) ^{-\frac{1}{2}}\right)
\right\}  \notag  \label{step6_1} \\
& =\left[ I_{K_{0}+K_{1}}+\left(D_{i}^{F}\right) ^{-1}D_{i}^{J}\right] 
\begin{bmatrix}
O_{0}^{(1)}u_{i,0}^{0} \\ 
O_{1}^{(1)}u_{i,1}^{0}%
\end{bmatrix}%
+\left[ \left(\dot{D}_{i}^{F}\right) ^{-1}\dot{D}_{i}^{J}-\left(D_{i}^{F}%
\right) ^{-1}D_{i}^{J}\right] 
\begin{bmatrix}
O_{0}^{(1)}u_{i,0}^{0} \\ 
O_{1}^{(1)}u_{i,1}^{0}%
\end{bmatrix}
\notag \\
& +\left(D_{i}^{F}\right) ^{-1}\frac{1}{T}\sum_{t=1}^{T}\omega_{it}^{0}%
\left(\tau -\mathbf{1}\left\{ \epsilon_{it}\leq 0\right\} \right)+\left[
\left(\dot{D}_{i}^{F}\right) ^{-1}-\left(D_{i}^{F}\right) ^{-1}\right] \frac{%
1}{T}\sum_{t=1}^{T}\omega_{it}^{0}\left(\tau -\mathbf{1}\left\{
\epsilon_{it}\leq 0\right\} \right)  \notag \\
& +\left(D_{i}^{F}\right) ^{-1}\hat{\mathbb{I}}_{2,i}+o_{p}\left(\left(N\vee
T\right) ^{-\frac{1}{2}}\right) .
\end{align}%
Owing to the fact that 
\begin{equation}
\max_{i\in I_{3}}\left\Vert \frac{1}{T}\sum_{t=1}^{T}\omega_{it}^{0}\left(%
\tau -\mathbf{1}\left\{ \epsilon_{it}\leq 0\right\} \right)
\right\Vert_{2}=O_{p}\left(\sqrt{\frac{\log \left(N\vee T\right) }{T}}\right)
\label{CLT_max}
\end{equation}%
by similar arguments as (\ref{Lem:F2}) using Bernstein's inequality in Lemma %
\ref{Lem:Bern}(ii) and Lemma \ref{Lem21}, we notice that 
\begin{equation*}
\max_{i\in I_{3}}\left\Vert \left[ \left(\dot{D}_{i}^{F}\right)^{-1}-%
\left(D_{i}^{F}\right) ^{-1}\right] \frac{1}{T}\sum_{t=1}^{T}\omega_{it}^{0}%
\left(\tau -\mathbf{1}\left\{ \epsilon_{it}\leq 0\right\}
\right)\right\Vert_{2}=o_{p}\left(\left(N\vee T\right) ^{-\frac{1}{2}%
}\right) .
\end{equation*}

Next, we define 
\begin{align*}
& D^{F}=O^{(1)}%
\begin{bmatrix}
\frac{1}{T}\sum_{t=1}^{T}\mathbb{E}\left[ f_{it}(0)\bigg|\mathscr{D}\right]%
v_{t,0}^{0}v_{t,0}^{0\prime} & 0 \\ 
0 & \frac{1}{T}\sum_{t=1}^{T}\mathbb{E}\left[ e_{1,it}^{2}f_{it}(0)\bigg|%
\mathscr{D}\right] v_{t,1}^{0}v_{t,1}^{0\prime}%
\end{bmatrix}%
O^{(1)\prime}, \\
& D^{J}=%
\begin{bmatrix}
\frac{1}{T}\sum_{t=1}^{T}\mathbb{E}\left[ f_{it}(0)\bigg|\mathscr{D}\right] 
\dot{v}_{t,0}^{(1)}\left(O_{0}^{(1)}v_{t,0}^{0}-\dot{v}_{t,0}^{(1)}\right)^{%
\prime} & 0 \\ 
0 & O_{1}^{(1)}\frac{1}{T}\sum_{t=1}^{T}\mathbb{E}\left[e_{1,it}^{2}f_{it}(0)%
\bigg|\mathscr{D}\right] v_{t,1}^{0}\left(O_{1}^{(1)}v_{t,1}^{0}-\dot{v}%
_{t,1}^{(1)}\right) ^{\prime}%
\end{bmatrix}%
,
\end{align*}
where $O^{(1)}=diag\left(O_{0}^{(1)},O_{1}^{(1)}\right) .$ Here $D^{F}$ and $%
D^{J}$ do not depend on $i$ owing to stationary assumption of sequence $%
\left\{ f_{it},f_{it}(0)e_{j,it}\right\}_{j\in [p]}$ conditional on all
factors in Assumption \ref{ass:10}(iii). By Bernstein's inequality
conditional on all factors similarly as in (\ref{Lem:F2}), we can show that $%
\max_{i\in I_{3}}\left\Vert D_{i}^{F}-D^{F}\right\Vert_{F}=O_{p}\left(\sqrt{%
\frac{\log (N\vee T)}{T}}\xi_{N}\right) $. Analogously, by Bernstein's
inequality conditional on $\mathscr{D}^{I_{1}\cup I_{2}}$, we can show that $%
\max_{i\in I_{3}}\left\Vert D_{i}^{J}-D_{j}\right\Vert_{F}=O_{p}\left(\sqrt{%
\frac{\log (N\vee T)}{T}}\eta_{N}\xi_{N}\right) $. Then it follows that 
\begin{equation*}
\max_{i\in I_{3}}\left\Vert \left(D_{i}^{F}\right)
^{-1}D_{i}^{J}-\left(D^{F}\right)
^{-1}D^{J}\right\Vert_{F}=O_{p}(\eta_{N}^{2}).
\end{equation*}%
In addition, uniformly over $i\in I_{3}$, 
\begin{align*}
& \left(\dot{D}_{i}^{F}\right) ^{-1}\dot{D}_{i}^{J}-\left(D_{i}^{F}%
\right)^{-1}D_{i}^{J} \\
& =\left[ \left(\dot{D}_{i}^{F}\right) ^{-1}-\left(D_{i}^{F}\right) ^{-1}%
\right] \left[ \dot{D}_{i}^{J}-D_{i}^{J}\right] +D_{i}^{J}\left[ \left(\dot{D%
}_{i}^{F}\right) ^{-1}-\left(D_{i}^{F}\right) ^{-1}\right]
+\left(D_{i}^{F}\right) ^{-1}\left[ \dot{D}_{i}^{J}-D_{i}^{J}\right] \\
& =\left(D_{i}^{F}\right) ^{-1}\left[ \dot{D}_{i}^{J}-D_{i}^{J}\right]%
+O_{p}(\eta_{N}^{2}) \\
& =%
\begin{bmatrix}
O_{p}\left(\eta_{N}^{2}\right) & O_{p}(\eta_{N}) \\ 
O_{p}\left(\eta_{N}^{2}\right) & O_{p}\left(\eta_{N}^{2}\right)%
\end{bmatrix}%
,
\end{align*}%
where the upper right block is dominated by $\frac{1}{T}%
\sum_{t=1}^{T}f_{it}(0)O_{0}^{(1)}v_{t,0}^{0}v_{t,1}^{0\prime}O_{1}^{(1)%
\prime}\left(e_{1,it}-\hat{e}_{1,it}\right) $ in the analysis of $%
J_{1,i}^{2} $ in (\ref{Lem:J_eta}).

Let $I_{K_{0}+K_{1}}+\left(D^{F}\right) ^{-1}D^{J}=O_{u,1}^{(1)}:= 
\begin{pmatrix}
\bar{O}_{u,0} & 0 \\ 
0 & \bar{O}_{u,1}%
\end{pmatrix}%
$, $O_{u,0}^{(1)}=\bar{O}_{u,0}O_{0}^{(1)}$, and $O_{u,1}^{(1)}=\bar{O}%
_{u,1}O_{1}^{(1)}$. Combining the above arguments, we obtain that 
\begin{align}
\hat{u}_{i,1}^{(3,1)}-O_{u,1}^{(1)}u_{i,1}^{0}& =O_{1}^{(1)}\hat{V}%
_{u_{1}}^{-1}\frac{1}{T}\sum_{t=1}^{T}e_{1,it}v_{t,1}^{0}\left(\tau -\mathbf{%
1}\left\{ \epsilon_{it}\leq 0\right\} \right) +\mathcal{R}_{i,u}^{1},
\label{step6_2} \\
\hat{u}_{i,0}^{(3,1)}-O_{u,0}^{(1)}u_{i,0}^{0}& =O_{0}^{(1)}\hat{V}%
_{u_{0},i}^{-1}\frac{1}{T}\sum_{t=1}^{T}v_{t,0}^{0}\left(\tau -\mathbf{1}%
\left\{ \epsilon_{it}\leq 0\right\} \right)  \notag \\
& +O_{0}^{(1)}\hat{V}_{u_{0},i}^{-1}\frac{1}{T}%
\sum_{t=1}^{T}f_{it}(0)v_{t,0}^{0}v_{t,1}^{0\prime}u_{i,1}^{0}\left(e_{1,it}-%
\hat{e}_{1,it}\right)  \notag \\
& +O_{0}^{(1)}\hat{V}_{u_{0},i}^{-1}\frac{1}{T}%
\sum_{t=1}^{T}f_{it}(0)v_{t,0}^{0}\left(\mu_{1,it}u_{i,1}^{0%
\prime}v_{t,1}^{0}-\hat{\mu}_{1,it}\dot{u}_{i,1}^{(1)\prime}\dot{v}%
_{t,1}^{(1)}\right) +\mathcal{R}_{i,u}^{0}  \label{step6_3}
\end{align}%
such that $\max_{i\in I_{3}}\left\vert \mathcal{R}_{i,u}^{0}\right%
\vert=o_{p}\left(\left(N\vee T\right) ^{-\frac{1}{2}}\right) $, $\max_{i\in
I_{3}}\left\vert \mathcal{R}_{i,u}^{1}\right\vert =o_{p}\left(\frac{1}{\sqrt{%
T}}\right) $, $\hat{V}_{u_{0}}=\frac{1}{T}\sum_{t=1}^{T}\mathbb{E}\left[
f_{it}(0)\big|\mathscr{D}\right] v_{t,0}^{0}v_{t,0}^{0\prime}$ and $\hat{V}%
_{u_{1}}=\frac{1}{T}\sum_{t=1}^{T}\mathbb{E}\left[f_{it}(0)e_{1,it}^{2}\big|%
\mathscr{D}\right] v_{t,1}^{0}v_{t,1}^{0\prime}$. From (\ref{step6_2}),
owing to the fact that $\hat{V}_{u_{1}}$ is bounded a.s. and $\max_{i\in
I_{3}}\left\Vert \frac{1}{T}\sum_{t=1}^{T}e_{1,it}v_{t,1}^{0}\left(\tau -%
\mathbf{1}\left\{ \epsilon_{it}\leq 0\right\} \right)
\right\Vert_{2}=O_{p}\left(\sqrt{\frac{\log(N\vee T)}{T}}\right) $ by
Bernstein's inequality in Lemma \ref{Lem:Bern}(ii), we obtain that 
\begin{equation}
\max_{i\in I_{3}}\left\Vert \hat{u}_{i,1}^{(3,1)}-O_{u,1}^{(1)}u_{i,1}^{0}%
\right\Vert_{2}=O_{p}\left(\sqrt{\frac{\log (N\vee T)}{T}}\right)\quad \text{%
and }\max_{i\in I_{3}}\left\Vert \hat{u}%
_{i,0}^{(3,1)}-O_{u,0}^{(1)}u_{i,0}^{0}\right\Vert_{2}=O_{p}(\eta_{N}).
\label{max:u}
\end{equation}

Last, noting that $O_{1}^{(1)}$ is a rotation matrix and the normal
distribution is invariant to rotation, for each $i\in I_{3},$ we have 
\begin{equation*}
\sqrt{T}\left(\hat{u}_{i,1}^{(3,1)}-O_{u,1}^{(1)}u_{i,1}^{0}\right)%
\rightsquigarrow \mathcal{N}(0,\Sigma_{u_{1},i}),
\end{equation*}
where $\Sigma_{u_{1},i}=O_{1}^{(1)}V_{u_{1}}^{-1}%
\Omega_{u_{1}}V_{u_{1}}^{-1}O_{1}^{(1)\prime}$, $V_{u_{1},i}=\frac{1}{T}%
\sum_{t=1}^{T}\mathbb{E}\left(f_{it}(0)e_{1,it}^{2}v_{t,1}^{0}v_{t,1}^{0%
\prime}\right) ,$ and 
\begin{equation*}
\Omega_{u_{1},i}=Var\left[ \frac{1}{\sqrt{T}}%
\sum_{t=1}^{T}e_{1,it}v_{t,1}^{0}\left(\tau -\mathbf{1}\left\{
\epsilon_{it}\leq 0\right\} \right) \right].\quad\blacksquare
\end{equation*}

\subsubsection{Proof of Statement (ii)}

Steps for the proof for statement (ii) are the same as those in the proof of
statement (i). Hence, we only sketch the proof. Recall from (\ref{debias_2})
that $\forall t\in[T]$, 
\begin{equation*}
\{\hat{v}_{t,j}^{(3,1)}\}_{j\in [p]}=\argmin\limits_{\{v_{t,j}\}_{j\in[p]}}%
\frac{1}{N_{3}}\sum_{i\in I_{3}}\rho_{\tau }\left(\hat{Y}_{it}-v_{t,0}^{%
\prime}\hat{u}_{i,0}^{(3,1)}-v_{t,1}^{\prime}\hat{u}_{i,1}^{(3,1)}\hat{e}%
_{1,it}\right) ,
\end{equation*}%
where $\hat{Y}_{it}=Y_{it}-\hat{\mu}_{1,it}\hat{u}_{i,1}^{(3,1)\prime}\dot{v}%
_{t,1}^{(1)}$. Let 
\begin{equation*}
\hat{\Delta}_{t,v}=%
\begin{bmatrix}
\hat{v}_{t,0}^{(3,1)}-\left(O_{u,0}^{(1)}\right) ^{\prime-1}v_{t,0}^{0} \\ 
\hat{v}_{t,1}^{(3,1)}-\left(O_{u,1}^{(1)}\right) ^{\prime-1}v_{t,1}^{0}%
\end{bmatrix}%
\text{and }\hat{\omega}_{it}=%
\begin{bmatrix}
\hat{u}_{i,0}^{(3,1)} \\ 
\hat{u}_{i,1}^{(3,1)}\hat{e}_{1,it}%
\end{bmatrix}%
.
\end{equation*}
For generic $(u_{i,0},u_{i,1},v_{i,0},v_{i,1})$, define 
\begin{equation*}
\hat{\mathbb{S}}_{t}\left(u_{i,0},u_{i,1},v_{i,0},v_{i,1}\right) =\frac{1}{%
N_{3}}\sum_{i\in I_{3}}\left[ \tau -\mathbf{1}\left\{
\epsilon_{it}\leq\varrho_{it}\left(u_{i,0},u_{i,1},v_{i,0},v_{i,1}\right)
\right\} \right] \hat{\omega}_{it},
\end{equation*}%
with 
\begin{equation*}
\varrho_{it}\left(u_{i,0},u_{i,1},v_{i,0},v_{i,1}\right)
=u_{i,0}^{\prime}v_{t,0}-u_{i,0}^{0\prime}v_{t,0}^{0}+\hat{\mu}%
_{1,it}u_{i,1}^{\prime}\dot{v}_{t,1}^{(1)}-\mu_{1,it}u_{i,1}^{0%
\prime}v_{t,1}^{0}+\hat{e}_{1,it}u_{i,1}^{\prime}v_{t,1}-e_{1,it}u_{i,1}^{0%
\prime}v_{t,1}^{0}.
\end{equation*}%
We also define 
\begin{align*}
\hat{\mathcal{S}}_{t}\left(u_{i,0},u_{i,1},v_{i,0},v_{i,1}\right) & =\frac{1%
}{N_{3}}\sum_{i\in I_{3}}\mathbb{E}\left\{ \left[ \tau -\mathbf{1}%
\left\{\epsilon_{it}\leq
\left(\varrho_{it}\left(u_{i,0},u_{i,1},v_{i,0},v_{i,1}\right) \right)
\right\} \right] \dot{\omega}_{it}\bigg|\mathscr{D}_{e}^{I_{1}\cup
I_{2}}\right\} \\
& =\frac{1}{N_{3}}\sum_{i\in I_{3}}\left[ \tau
-F_{it}\left(\varrho_{it}\left(u_{i,0},u_{i,1},v_{i,0},v_{i,1}\right)
\right) \right] \dot{\omega}_{it}
\end{align*}%
and 
\begin{align*}
\hat{\mathbb{M}}_{t}\left(\hat{u}_{i,0}^{(3,1)},\hat{u}_{i,1}^{(3,1)},\hat{v}%
_{t,0}^{(3,1)},\hat{v}_{t,1}^{(3,1)}\right) & =\hat{\mathbb{S}}_{t}\left(%
\hat{u}_{i,0}^{(3,1)},\hat{u}_{i,1}^{(3,1)},\hat{v}_{t,0}^{(3,1)},\hat{v}%
_{t,1}^{(3,1)}\right) -\hat{\mathbb{S}}_{t}\left(\hat{u}_{i,0}^{(3,1)},\hat{u%
}_{i,1}^{(3,1)},\left(O_{u,0}^{(1)}\right)
^{\prime-1}v_{t,0}^{0},\left(O_{u,1}^{(1)}\right)
^{\prime-1}v_{t,1}^{0}\right) \\
& -\left\{ \hat{\mathcal{S}}_{t}\left(\hat{u}_{i,0}^{(3,1)},\hat{u}%
_{i,1}^{(3,1)},\hat{v}_{t,0}^{(3,1)},\hat{v}_{t,1}^{(3,1)}\right) -\hat{%
\mathcal{S}}_{t}\left(\hat{u}_{i,0}^{(3,1)},\hat{u}_{i,1}^{(3,1)},%
\left(O_{u,0}^{(1)}\right)
^{\prime-1}v_{t,0}^{0},\left(O_{u,1}^{(1)}\right)^{\prime-1}v_{t,1}^{0}%
\right) \right\} .
\end{align*}%
Then a similar result to that in (\ref{C.9}) holds: 
\begin{equation}
\hat{\Delta}_{t,v}=\left(\hat{D}_{t}^{F}\right) ^{-1}\left\{ \hat{D}_{t}^{J} 
\begin{bmatrix}
\left(O_{u,0}^{(1)}\right) ^{\prime-1}v_{t,0}^{0} \\ 
\left(O_{u,1}^{(1)}\right) ^{\prime-1}v_{t,1}^{0}%
\end{bmatrix}%
+\hat{\mathbb{I}}_{4,t}+\hat{\mathbb{I}}_{5,t}+\hat{\mathbb{I}}_{6,t}+\hat{%
\mathbb{M}}_{t}\left(\hat{u}_{i,0}^{(3,1)},\hat{u}_{i,1}^{(3,1)},\hat{v}%
_{t,0}^{(1)},\hat{v}_{t,1}^{(3,1)}\right) +o_{p}\left(\left(N\vee T\right)^{-%
\frac{1}{2}}\right) \right\} ,  \label{C.23}
\end{equation}
where 
\begin{align*}
\hat{D}_{t}^{F}& =\frac{1}{N_{3}}\sum_{i\in I_{3}}f_{it}\left[
\varrho_{it}\left(\hat{u}_{i,0}^{(3,1)},\hat{u}_{i,1}^{(3,1)},%
\left(O_{u,0}^{(1)}\right)
^{\prime-1}v_{t,0}^{0},\left(O_{u,1}^{(1)}\right)^{\prime-1}v_{t,1}^{0}%
\right) \right] 
\begin{bmatrix}
\hat{u}_{i,0}^{(3,1)}\hat{u}_{i,0}^{(3,1)\prime} & \hat{e}_{1,it}\hat{u}%
_{i,0}^{(3,1)}\hat{u}_{i,1}^{(3,1)\prime} \\ 
\hat{e}_{1,it}\hat{u}_{i,1}^{(3,1)}\hat{u}_{i,0}^{(3,1)\prime} & \hat{e}%
_{1,it}^{2}\hat{u}_{i,1}^{(3,1)}\hat{u}_{i,1}^{(3,1)\prime}%
\end{bmatrix}
, \\
\hat{D}_{t}^{J}& =\frac{1}{N_{3}}\sum_{i\in I_{3}}f_{it}(\tilde{\varrho}%
_{it}) 
\begin{bmatrix}
\hat{u}_{i,0}^{(3,1)}\left(O_{u,0}^{(1)}u_{i,0}^{0}-\hat{u}%
_{i,0}^{(3,1)}\right) ^{\prime} & \hat{u}_{i,0}^{(3,1)}%
\left(e_{1,it}O_{u,1}^{(1)}u_{i,1}^{0}-\hat{e}_{it}\hat{u}%
_{i,1}^{(3,1)}\right)^{\prime}\notag \\ 
\hat{e}_{1,it}\hat{u}_{i,1}^{(3,1)}\left(O_{u,0}^{(1)}u_{i,0}^{0}-\hat{u}%
_{i,0}^{(3,1)}\right) ^{\prime} & \hat{e}_{1,it}\hat{u}_{i,1}^{(3,1)}%
\left(e_{1,it}O_{u,1}^{(1)}u_{i,1}^{0}-\hat{e}_{1,it}\hat{u}%
_{i,1}^{(3,1)}\right)^{\prime}%
\end{bmatrix}%
, \\
\hat{\mathbb{I}}_{4,t}& =\frac{1}{N_{3}}\sum_{i\in I_{3}}\hat{\omega}_{it}%
\left[ \tau -\mathbf{1}\left\{ \epsilon_{it}\leq 0\right\} \right) , \\
\hat{\mathbb{I}}_{5,t}& =\frac{1}{N_{3}}\sum_{i\in I_{3}}\hat{\omega}%
_{it}f_{it}\left(\tilde{\varrho}_{it}\right)
\left(\mu_{1,it}u_{i,1}^{0\prime}v_{t,1}^{0}-\hat{\mu}_{1,it}\hat{u}%
_{i,1}^{(3,1)\prime}\dot{v}_{t,1}^{(1)}\right) , \\
\hat{\mathbb{I}}_{6,t}& =\frac{1}{N_{3}}\sum_{i\in I_{3}}\hat{\omega}_{it}%
\Bigg\{\left[ \mathbf{1}\left\{ \epsilon_{it}\leq 0\right\} -\mathbf{1}%
\left\{ \epsilon_{it}\leq \varrho_{it}\left(\hat{u}_{i,0}^{(3,1)},\hat{u}%
_{i,1}^{(3,1)},\left(O_{u,0}^{(1)}\right)
^{\prime-1}v_{t,0}^{0},\left(O_{u,1}^{(1)}\right)
^{\prime-1}v_{t,1}^{0}\right) \right\} \right] \\
& -\left(f_{it}(0)-f_{it}\left[ \varrho_{it}\left(\hat{u}_{i,0}^{(3,1)},\hat{%
u}_{i,1}^{(3,1)},\left(O_{u,0}^{(1)}\right)
^{\prime-1}v_{t,0}^{0},\left(O_{u,1}^{(1)}\right)
^{\prime-1}v_{t,1}^{0}\right) \right] \right) \Bigg\},
\end{align*}%
and $\left\vert \tilde{\varrho}_{it}\right\vert $ lies between $0$ and $%
\left\vert \varrho_{it}\left(\hat{u}_{i,0}^{(3,1)},\hat{u}%
_{i,1}^{(3,1)},\left(O_{u,0}^{(1)}\right)
^{\prime-1}v_{t,0}^{0},\left(O_{u,1}^{(1)}\right)
^{\prime-1}v_{t,1}^{0}\right) \right\vert $.

Then, we derive the linear expansion for $\hat{\Delta}_{t,v}$ by analyzing
each term in (\ref{C.23}). Define 
\begin{align*}
& D_{t}^{F}=\frac{1}{N_{3}}\sum_{i\in I_{3}}f_{it}(0)%
\begin{bmatrix}
O_{u,0}^{(1)}u_{i,0}^{0}u_{i,0}^{0\prime}O_{u,0}^{(1)\prime} & 0 \\ 
0 & e_{1,it}^{2}O_{u,1}^{(1)}u_{i,1}^{0}u_{i,1}^{0\prime}O_{u,1}^{(1)\prime}%
\end{bmatrix}%
, \\
& D_{t}^{J}=\frac{1}{N_{3}}\sum_{i\in I_{3}}f_{it}(0)%
\begin{bmatrix}
\hat{u}_{i,0}^{(3,1)}\left(O_{u,0}^{(1)}u_{i,0}^{0}-\hat{u}%
_{i,0}^{(3,1)}\right) ^{\prime} & 0 \\ 
0 & e_{1,it}^{2}O_{u,1}^{(1)}u_{i,1}^{0}\left(O_{u,1}^{(1)}u_{i,1}^{0}-\hat{u%
}_{i,1}^{(3,1)}\right) ^{\prime}%
\end{bmatrix}%
,
\end{align*}
such that 
\begin{equation*}
\max_{t\in[T]}\left\Vert \hat{D}_{t}^{F}-D_{t}^{F}\right\Vert_{F}=O_{p}(%
\eta_{N})\text{ and }\max_{t\in[T]}\left\Vert \hat{D}_{t}^{J}-D_{t}^{J}%
\right\Vert_{F}=\left\Vert 
\begin{bmatrix}
O_{p}\left(\eta_{N}^{2}\right) & O_{p}(\eta_{N}) \\ 
O_{p}\left(\eta_{N}^{2}\right) & O_{p}\left(\eta_{N}^{2}\right)%
\end{bmatrix}%
\right\Vert_{F}
\end{equation*}%
by Lemma \ref{Lem27}. Let $\omega_{it}^{\ast }=%
\begin{bmatrix}
O_{u,0}^{(1)}u_{i,0}^{0} \\ 
e_{1,it}O_{u,1}^{(1)}u_{i,1}^{0}%
\end{bmatrix}%
$. We can show that 
\begin{equation*}
\hat{\mathbb{I}}_{4,t}=\frac{1}{N_{3}}\sum_{i\in I_{3}}\omega_{it}^{\ast }%
\left[ \tau -\mathbf{1}\left\{ \epsilon_{it}\leq 0\right\} \right]%
+O_{p}\left(\eta_{N}^{2}\right)
\end{equation*}%
uniformly over $t\in[T]$ by analogous analysis as in Step 2 in the previous
subsection. With Lemma \ref{Lem28} and by similar arguments as in Step 3 in
the previous subsection, we obtain 
\begin{equation*}
\max_{t\in[T]}\left\Vert \hat{\mathbb{I}}_{5,t}\right\Vert_{2}=\left\Vert 
\begin{bmatrix}
O_{p}(\eta_{N}) \\ 
o_{p}\left(\left(N\vee T\right) ^{-\frac{1}{2}}\right)%
\end{bmatrix}%
\right\Vert_{2}
\end{equation*}
uniformly over $t\in[T]$.

Next, for $\hat{\mathbb{I}}_{6,t}$, From (\ref{I6_2}), we first note that 
\begin{align}
& \varrho_{it}\left(\hat{u}_{i,0}^{(3,1)},\hat{u}_{i,1}^{(3,1)},%
\left(O_{u,0}^{(1)}\right)
^{\prime-1}v_{t,0}^{0},\left(O_{u,1}^{(1)}\right)^{\prime-1}v_{t,1}^{0}%
\right) \leq R_{\varrho ,it}^{1}\left(\left\vert \mu_{1,it}\right\vert
+\left\vert e_{1,it}\right\vert \right) +R_{\varrho,it}^{2}\quad \text{with}
\notag  \label{rho} \\
& \max_{i\in I_{3},t\in[T]}\left\vert R_{\varrho
,it}^{1}\right\vert=O_{p}(\eta_{N}),\max_{i\in I_{3},t\in[T]}\left\vert
R_{\varrho,it}^{2}\right\vert =O_{p}(\eta_{N}).
\end{align}
We then observe that, uniformly over $t\in[T]$, 
\begin{align*}
\hat{\mathbb{I}}_{6,t}& =\frac{1}{N_{3}}\sum_{i\in I_{3}}\hat{\omega}_{it}%
\Bigg\{\left[ \mathbf{1}\left\{ \epsilon_{it}\leq 0\right\} -\mathbf{1}%
\left\{ \epsilon_{it}\leq \varrho_{it}\left(\hat{u}_{i,0}^{(3,1)},\hat{u}%
_{i,1}^{(3,1)},\left(O_{u,0}^{(1)}\right)
^{\prime-1}v_{t,0}^{0},\left(O_{u,1}^{(1)}\right)
^{\prime-1}v_{t,1}^{0}\right) \right\} \right] \\
& -\left(F_{it}(0)-F_{it}\left[ \varrho_{it}\left(\hat{u}_{i,0}^{(3,1)},\hat{%
u}_{i,1}^{(3,1)},\left(O_{u,0}^{(1)}\right)
^{\prime-1}v_{t,0}^{0},\left(O_{u,1}^{(1)}\right)
^{\prime-1}v_{t,1}^{0}\right) \right] \right) \Bigg\} \\
& =\frac{1}{N_{3}}\sum_{i\in I_{3}}\omega_{it}^{\ast }\Bigg\{\left[ \mathbf{1%
}\left\{ \epsilon_{it}\leq 0\right\} -\mathbf{1}\left\{ \epsilon_{it}\leq
\varrho_{it}\left(\hat{u}_{i,0}^{(3,1)},\hat{u}_{i,1}^{(3,1)},%
\left(O_{u,0}^{(1)}\right)
^{\prime-1}v_{t,0}^{0},\left(O_{u,1}^{(1)}\right)^{\prime-1}v_{t,1}^{0}%
\right) \right\} \right] \\
& -\left(F_{it}(0)-F_{it}\left[ \varrho_{it}\left(\hat{u}_{i,0}^{(3,1)},\hat{%
u}_{i,1}^{(3,1)},\left(O_{u,0}^{(1)}\right)
^{\prime-1}v_{t,0}^{0},\left(O_{u,1}^{(1)}\right)
^{\prime-1}v_{t,1}^{0}\right) \right] \right) \Bigg\} \\
& +\frac{1}{N_{3}}\sum_{i\in I_{3}}\left(\hat{\omega}_{it}-\omega_{it}^{\ast
}\right) \Bigg\{\left[ \mathbf{1}\left\{ \epsilon_{it}\leq 0\right\} -%
\mathbf{1}\left\{ \epsilon_{it}\leq \varrho_{it}\left(\hat{u}_{i,0}^{(3,1)},%
\hat{u}_{i,1}^{(3,1)},\left(O_{u,0}^{(1)}\right)
^{\prime-1}v_{t,0}^{0},\left(O_{u,1}^{(1)}\right)
^{\prime-1}v_{t,1}^{0}\right)\right\} \right] \\
& -\left(F_{it}(0)-F_{it}\left[ \varrho_{it}\left(\hat{u}_{i,0}^{(3,1)},\hat{%
u}_{i,1}^{(3,1)},\left(O_{u,0}^{(1)}\right)
^{\prime-1}v_{t,0}^{0},\left(O_{u,1}^{(1)}\right)
^{\prime-1}v_{t,1}^{0}\right) \right] \right) \Bigg\} \\
& =\frac{1}{N_{3}}\sum_{i\in I_{3}}\omega_{it}^{\ast }\Bigg\{\left[ \mathbf{1%
}\left\{ \epsilon_{it}\leq 0\right\} -\mathbf{1}\left\{ \epsilon_{it}\leq
\varrho_{it}\left(\hat{u}_{i,0}^{(3,1)},\hat{u}_{i,1}^{(3,1)},%
\left(O_{u,0}^{(1)}\right)
^{\prime-1}v_{t,0}^{0},\left(O_{u,1}^{(1)}\right)^{\prime-1}v_{t,1}^{0}%
\right) \right\} \right] \\
& -\left(F_{it}(0)-F_{it}\left[\varrho_{it}\left(\hat{u}_{i,0}^{(3,1)},\hat{u%
}_{i,1}^{(3,1)},\left(O_{u,0}^{(1)}\right)
^{\prime-1}v_{t,0}^{0},\left(O_{u,1}^{(1)}\right)
^{\prime-1}v_{t,1}^{0}\right) \right] \right) \Bigg\}+o_{p}\left(\left(N\vee
T\right) ^{-\frac{1}{2}}\right) \\
& =%
\begin{bmatrix}
O_{p}(\eta_{N}) \\ 
o_{p}\left(\left(N\vee T\right) ^{-\frac{1}{2}}\right)%
\end{bmatrix}%
,
\end{align*}%
where the third equality holds by (\ref{rho}) and similar arguments as used
in (\ref{step4_3}) and (\ref{step4_4}), and the last equality holds by Lemma %
\ref{Lem29}.

Similarly, combining Lemma \ref{Lem29} and Lemma \ref{Lem30}, it yields 
\begin{equation*}
\max_{t\in[T]}\left\Vert \hat{\mathbb{M}}_{t}\left(\hat{u}_{i,0}^{(3,1)},%
\hat{u}_{i,1}^{(3,1)},\hat{v}_{t,0}^{(3,1)},\hat{v}_{t,1}^{(3,1)}\right)
\right\Vert_{2}=\left\Vert 
\begin{bmatrix}
O_{p}(\eta_{N}) \\ 
o_{p}\left(\left(N\vee T\right) ^{-\frac{1}{2}}\right)%
\end{bmatrix}%
\right\Vert_{2}.
\end{equation*}
Let $\hat{V}_{v_{1},t}^{3}=\frac{1}{N_{3}}\sum_{i\in
I_{3}}f_{it}(0)e_{1,it}^{2}u_{i,1}^{0}u_{i,1}^{0\prime},$ and 
\begin{equation*}
O_{v_{1},t}^{(1)}=\left\{ I_{K_{1}}+\left(O_{u,1}^{(1)\prime}\right) ^{-1} 
\left[\hat{V}_{v_{1},t}^{3}\right] ^{-1}\left[ \frac{1}{N_{3}}\sum_{i\in
I_{3}}f_{it}(0)e_{1,it}^{2}u_{i,1}^{0}\left(O_{u,1}^{(1)}u_{i,1}^{0}-\hat{u}%
_{i,1}^{(3,1)}\right) ^{\prime}\right] \right\}
\left(O_{u,1}^{(1)\prime}\right) ^{-1}.
\end{equation*}
Combining arguments above from (\ref{C.23}), we have 
\begin{equation*}
\hat{v}_{t,1}^{(3,1)}-O_{v_{1},t}^{(1)}v_{t,1}^{0}=\left(O_{u,1}^{(1)\prime}%
\right)^{-1}\left(\hat{V}_{v_{1},t}^{3}\right) ^{-1}\frac{1}{N_{3}}%
\sum_{i\in I_{3}}e_{1,it}u_{i,1}^{0}\left(\tau -\mathbf{1}\left\{
\epsilon_{it}\leq 0\right\} \right) +\mathcal{R}_{t,v}^{1},
\end{equation*}
such that $\max_{t\in[T]}\left\vert \mathcal{R}_{t,v}^{1}\right\vert=o_{p}%
\left(\left(N\vee T\right) ^{-\frac{1}{2}}\right) $. This, in conjunction
with the result in Lemma \ref{Lem31}, i.e., $\max_{t\in [T]}\left\Vert
O_{v_{1},t}^{(1)}-\left(O_{u,1}^{(1)\prime}\right)
^{-1}\right\Vert_{F}=o_{p}\left(\left(N\vee T\right) ^{-\frac{1}{2}}\right)$%
, implies that 
\begin{align}
\hat{v}_{t,1}^{(3,1)}-\left(O_{u,1}^{(1)}\right)
^{\prime-1}v_{t,1}^{0}&=\left(O_{u,1}^{(1)\prime}\right) ^{-1}\left(\hat{V}%
_{v_{1},t}^{3}\right)^{-1}\frac{1}{N_{3}}\sum_{i\in
I_{3}}e_{1,it}u_{i,1}^{0}\left(\tau -\mathbf{1}\left\{ \epsilon_{it}\leq
0\right\} \right) +\mathcal{R}_{t,v}^{1},  \notag  \label{C.25} \\
& =O_{1}^{(1)}\left(\hat{V}_{v_{1},t}^{3}\right) ^{-1}\frac{1}{N_{3}}%
\sum_{i\in I_{3}}e_{1,it}u_{i,1}^{0}\left(\tau -\mathbf{1}\left\{
\epsilon_{it}\leq 0\right\} \right) +\mathcal{R}_{t,v}^{1}.
\end{align}%
where the second line holds by the fact that $\left\Vert
O_{u,1}^{(1)}-O_{1}^{(1)}\right\Vert_{F}=O_{p}(\eta_{N})$, $\hat{V}%
_{v_{1},t}^{I}$ is uniformly bounded and that 
\begin{equation*}
\max_{t\in[T]}\left\Vert \frac{1}{N_{3}}\sum_{i\in
I_{3}}e_{1,it}u_{i,1}^{0}\left(\tau -\mathbf{1}\left\{ \epsilon_{it}\leq
0\right\} \right) \right\Vert_{2}=O_{p}\left(\sqrt{\frac{\log (N\vee T)}{N}}%
\right)
\end{equation*}
by similar arguments as in (\ref{CLT_max}). Further define 
\begin{align*}
&
\Sigma_{v_{1}}=O_{1}^{(1)}V_{v_{1}}^{-1}%
\Omega_{u_{1}}V_{u_{1}}^{-1}O_{1}^{(1)\prime},\quad V_{v_{1}}=\frac{1}{N_{3}}%
\sum_{i\in I_{3}}\mathbb{E}\left(f_{it}(0)e_{1,it}^{2}u_{i,1}^{0}u_{i,1}^{0%
\prime}\right) ,\text{ and} \\
& \Omega_{v_{1}}=\tau \left(1-\tau \right) \frac{1}{N_{3}}\sum_{i\in I_{3}}%
\mathbb{E}\left(e_{1,it}^{2}u_{i,1}^{0}u_{i,1}^{0\prime}\right) .
\end{align*}%
Then we have 
\begin{align}
& \sqrt{N_{3}}\left(\hat{v}_{t,1}^{(3,1)}-\left(O_{u,1}^{(1)}\right)
^{\prime-1}v_{t,1}^{0}\right) \rightsquigarrow \mathcal{N}(0,\Sigma_{v_{1}}),
\notag  \label{max:v} \\
& \max_{t\in[T]}\left\Vert \hat{v}_{t,1}^{(3,1)}-O_{u,1}^{(1)}v_{t,1}^{0}%
\right\Vert_{2}=O_{p}\left(\sqrt{\frac{\log (N\vee T)}{N}}\right) ,\quad 
\text{and }\max_{t\in[T]}\left\Vert \hat{v}%
_{t,0}^{(3,1)}-O_{u,0}^{(1)}v_{t,0}^{0}\right\Vert_{2}=O_{p}(\eta_{N}),
\end{align}%
where the second line holds by Bernstein's inequality with independent data
and is similar to (\ref{max:u}). $\blacksquare $

\subsection{Proof of Proposition \protect\ref{Pro4}}

\subsubsection{Proof of Statement (i)}

Focusing on the slope estimators for $i\in I_{3}$, we notice that $\hat{%
\Theta}_{j,it}=\frac{1}{2}\left\{ \hat{u}_{t,j}^{(3,1)\prime}\hat{v}%
_{t,j}^{(3,1)}+\hat{u}_{t,j}^{(3,2)\prime}\hat{v}_{t,j}^{(3,2)}\right\} .$
It follows that 
\begin{align*}
& \hat{\Theta}_{j,it}-\Theta_{j,it}^{0} \\
& =\frac{1}{2}\left\{ \left(\hat{u}_{i,j}^{(3,1)}-O_{u,j}^{(1)}u_{i,j}^{0}%
\right) ^{\prime}\left(\hat{v}_{t,j}^{(3,1)}-O_{u,j}^{(1)}v_{t,j}^{0}%
\right)+\left(O_{u,j}^{(1)}u_{i,j}^{0}\right) ^{\prime}\left(\hat{v}%
_{t,j}^{(3,1)}-O_{u,j}^{(1)}v_{t,j}^{0}\right) +\left(\hat{u}%
_{i,j}^{(3,1)}-O_{u,j}^{(1)}u_{i,j}^{0}\right)
^{\prime}O_{u,j}^{(1)}v_{t,j}^{0}\right\} \\
& +\frac{1}{2}\left\{ \left(\hat{u}_{i,j}^{(3,2)}-O_{u,j}^{(2)}u_{i,j}^{0}%
\right) ^{\prime}\left(\hat{v}_{t,j}^{(3,2)}-O_{u,j}^{(2)}v_{t,j}^{0}%
\right)+\left(O_{u,j}^{(2)}u_{i,j}^{0}\right) ^{\prime}\left(\hat{v}%
_{t,j}^{(3,2)}-O_{u,j}^{(2)}v_{t,j}^{0}\right) +\left(\hat{u}%
_{i,j}^{(3,2)}-O_{u,j}^{(2)}u_{i,j}^{0}\right)
^{\prime}O_{u,j}^{(1)}v_{t,j}^{0}\right\} \\
& =u_{i,j}^{0\prime}\left(\hat{V}_{v_{j},t}^{(3)}\right) ^{-1}\frac{1}{N_{3}}%
\sum_{i\in I_{3}}e_{j,it}u_{i,j}^{0}\left(\tau -\mathbf{1}\left\{
\epsilon_{it}\leq 0\right\} \right) +v_{t,j}^{0\prime}\hat{V}_{u_{j}}^{-1}%
\frac{1}{T}\sum_{t=1}^{T}e_{j,it}v_{t,j}^{0}\left(\tau -\mathbf{1}\left\{
\epsilon_{it}\leq 0\right\} \right) +\mathcal{R}_{it}^{j} \\
&=u_{i,j}^{0\prime}\left(\hat{V}_{v_{j},t}^{(3)}\right) ^{-1}\frac{1}{N_{3}}%
\sum_{i\in I_{3}}\xi_{j,it}^{0} +v_{t,j}^{0\prime}\hat{V}_{u_{j}}^{-1}\frac{1%
}{T}\sum_{t=1}^{T}b_{j,it}^{0}+\mathcal{R}_{it}^{j},
\end{align*}
such that $\max_{i\in I_{3},t\in[T]}\left\vert \mathcal{R}%
_{it}^{j}\right\vert =o_{p}\left(\left(N\vee T\right) ^{-1/2}\right) $ by
Theorem \ref{Thm3} and the second equality above combines Theorem \ref{Thm3}
and the fact that $\left\Vert
O_{u,1}^{(1)}-O_{1}^{(1)}\right\Vert_{F}=O_{p}(\eta_{N})$. With similar
results hold for slope estimators for subsamples $I_{1}$ and $I_{2}$, and
then we obtain the statement (i).

\subsubsection{Proof of Statement (ii)}

Combining (\ref{max:u}), (\ref{max:v}) and Lemma \ref{Lem:bounded u&v_tilde}%
(i), it's clear that 
\begin{equation}
\max_{i\in I_{3},t\in[T]}\left\vert \hat{\Theta}_{j,it}-\Theta_{j,it}^{0}%
\right\vert =O_{p}\left(\sqrt{\frac{\log (N\vee T)}{N\wedge T}}\right) \text{
}\forall j\in [p]\quad \text{and}\quad\max_{i\in I_{3},t\in[T]}\left\vert 
\hat{\Theta}_{0,it}-\Theta_{0,it}^{0}\right\vert =O_{p}(\eta_{N}).\text{ }
\label{max:theta}
\end{equation}

\subsubsection{Proof of Statement (iii)}

For $i\in I_{a}$ and $a\in[3]$, with the distribution theory defined in
Theorem \ref{Thm3}, we notice that 
\begin{align*}
\left(\frac{1}{T}v_{t,j}^{0\prime}\Xi_{u_{j},i}v_{t,j}^{0}+\frac{1}{N_{a}}%
u_{i,j}^{0\prime}\Xi_{v_{j}}^{a}u_{i,j}^{0}\right)^{-1/2}\left(\hat{\Theta}%
_{j,it}-\Theta_{j,it}^{0}\right)\rightsquigarrow \mathcal{N}(0,1),
\end{align*}
which leads to the proof. $\blacksquare$

\subsection{Proof of Theorem \protect\ref{Thm5}}

\subsubsection{Proof of Statement (i)}

The proof is analogous to that in \cite{castagnetti2015inference} and \cite%
{lu2021uniform}. Recall that $S_{u_{j}}=\max (S_{u_{j}}^{(1,2)},$ $%
S_{u_{j}}^{(2,3)},$ $S_{u_{j}}^{(3,1)})$. For $a\in \lbrack 3]$ and $b\in
\lbrack 3]\setminus \{a\}$, Theorem \ref{Thm3} shows that 
\begin{equation*}
\hat{u}_{i,j}^{(a,b)}-O_{u,j}^{(b)}u_{i,j}^{0}=O_{j}^{(b)}\hat{V}%
_{u_{j}}^{-1}\frac{1}{T}\sum_{t=1}^{T}b_{j,it}^{0}+\mathcal{R}_{i,u}^{j}%
\text{ }\forall i\in I_{a},
\end{equation*}%
where $\max_{i\in I_{a}}\left\vert \mathcal{R}_{i,u}^{j}\right\vert
=o_{p}\left( \left( N\vee T\right) ^{-1/2}\right) $. Recall that $\hat{\bar{u%
}}_{j}^{(a)}=\frac{1}{N_{a}}\sum_{i\in I_{a}}\hat{u}_{i,j}^{(a,b)}$. Under $%
H_{0}^{I}:u_{i,j}^{0}=u_{j},\forall i\in \lbrack N]$, we have 
\begin{equation}
\hat{\bar{u}}_{j}^{(a,b)}-O_{u,j}^{(b)}u_{j}=O_{j}^{(b)}\hat{V}_{u_{j}}^{-1}%
\frac{1}{N_{a}T}\sum_{i\in I_{a}}\sum_{t=1}^{T}b_{j,it}^{0}+\frac{1}{N_{a}}%
\sum_{i\in I_{a}}\mathcal{R}_{i,u}^{j}=o_{p}\left( \left( N\vee T\right)
^{-1/2}\right) ,  \label{E.1}
\end{equation}%
where the last equality holds by a simple application of Bernstein's
inequality. Note that 
\begin{align}
& T\left( \hat{u}_{i,j}^{(a,b)}-\hat{\bar{u}}_{j}^{(a,b)}\right) ^{\prime
}\left( \hat{\Sigma}_{u_{j}}\right) ^{-1}\left( \hat{u}_{i,j}^{(a,b)}-\hat{%
\bar{u}}_{j}^{(a,b)}\right)  \notag  \label{E.2} \\
& =T\left( \hat{u}_{i,j}^{(a,b)}-O_{u,j}^{(b)}u_{j}\right) ^{\prime }\left( 
\hat{\Sigma}_{u_{j}}\right) ^{-1}\left( \hat{u}%
_{i,j}^{(a,b)}-O_{u,j}^{(b)}u_{j}\right) +T\left( \hat{\bar{u}}%
_{j}^{(a,b)}-O_{u,j}^{(b)}u_{j}\right) ^{\prime }\left( \hat{\Sigma}%
_{u_{j}}\right) ^{-1}\left( \hat{\bar{u}}_{j}^{(a,b)}-O_{u,j}^{(b)}u_{j}%
\right)  \notag \\
& -2T\left( \hat{u}_{i,j}^{(a,b)}-O_{u,j}^{(b)}u_{j}\right) ^{\prime }\left( 
\hat{\Sigma}_{u_{j}}\right) ^{-1}\left( \hat{\bar{u}}%
_{j}^{(a,b)}-O_{u,j}^{(b)}u_{j}\right) :=I_{1ij}+I_{2ij}-2I_{3ij}.
\end{align}%
For $I_{2ij},$ we have 
\begin{equation}
\max_{i\in I_{a}}\left\vert I_{2ij}\right\vert \leq \max_{i\in
I_{a}}T\left\Vert \hat{\bar{u}}_{j}^{(a,b)}-O_{u,j}^{(b)}u_{j}\right\Vert
_{2}^{2}\left\{ \lambda _{\min }(\Sigma _{u_{j}})+o_{p}(1)\right\}
^{-1}=o_{p}(1)  \label{E.3}
\end{equation}%
by Lemma \ref{Lem:covhat}, Assumption \ref{ass:12} and (\ref{E.1}). For $%
I_{3ij},$ we have 
\begin{align}
\max_{i\in I_{a}}\left\vert I_{3ij}\right\vert & \leq T\left( \max_{i\in
I_{a}}\left\Vert \hat{u}_{i,j}^{(a,b)}-O_{u,j}^{(b)}u_{j}\right\Vert
_{2}\right) \left( \left\Vert \hat{\bar{u}}_{j}^{(a,b)}-O_{u,j}^{(b)}u_{j}%
\right\Vert _{2}\right) \left[ \lambda _{\min }(\Sigma _{u_{j}})+o_{p}(1)%
\right] ^{-1}  \notag \\
& =TO_{p}\left( \sqrt{\frac{\log (N\vee T)}{T}}\right) o_{p}\left( \left(
N\vee T\right) ^{-1/2}\right) =o_{p}(1)  \label{E.4}
\end{align}%
by (\ref{max:u}) and (\ref{E.1}). It suffices to study $I_{1ij}$ below.

Now, let $\mathbb{Z}_{\mathfrak{B}}^{(1)}=\left( \mathbb{Z}_{\mathfrak{B}%
,1}^{(1)\prime },\cdots ,\mathbb{Z}_{\mathfrak{B},N}^{(1)\prime }\right)
^{\prime }$, where $\mathbb{Z}_{\mathfrak{B},i}^{(1)}\sim \mathcal{N}\left(
0,O_{j}^{(1)\prime }\Sigma _{u_{j}}O_{j}^{(1)}\right) $ for $i\in I_{3}$, $%
\mathbb{Z}_{\mathfrak{B},i}^{(1)}\sim \mathcal{N}\left( 0,O_{j}^{(3)\prime
}\Sigma _{u_{j}}O_{j}^{(3)}\right) $ for $i\in I_{2}$ and $\mathbb{Z}_{%
\mathfrak{B},i}^{(1)}\sim \mathcal{N}\left( 0,O_{j}^{(2)\prime }\Sigma
_{u_{j}}O_{j}^{(2)}\right) $ for $i\in I_{1}$. Note that 
\begin{align}
I_{1ij}& =\left( O_{j}^{(b)}\mathbb{Z}_{\mathfrak{B},i}^{(1)}+o_{p(1)}+%
\mathcal{R}_{i,u}^{j}\right) ^{\prime }\Sigma _{u_{j}}^{-1}\left( O_{j}^{(b)}%
\mathbb{Z}_{\mathfrak{B},i}^{(1)}+o_{p}(1)+\mathcal{R}_{i,u}^{j}\right) 
\notag \\
& +T\left( \hat{u}_{i,j}^{(a,b)}-O_{u,j}^{(b)}u_{j}\right) ^{\prime }\left[ 
\hat{\Sigma}_{u_{j}}^{-1}-\Sigma _{u_{j}}^{-1}\right] \left( \hat{u}%
_{i,j}^{(a,b)}-O_{u,j}^{(b)}u_{j}\right)  \notag \\
& =\mathbb{Z}_{\mathfrak{B},i}^{(1)\prime }\left( O_{j}^{(b)\prime }\Sigma
_{u_{j}}O_{j}^{(b)}\right) ^{-1}\mathbb{Z}_{\mathfrak{B},i}^{(1)}+o_{p}(1)%
\quad \text{uniformly over $i\in \lbrack N]$.}  \label{E.5}
\end{align}%
%
It follows that 
\begin{equation}
S_{u_{j}}=\max_{i\in \lbrack N]}\left\{ \begin{aligned}
\mathbb{Z}_{\mathfrak{B},i}^{(1)\prime}\left(O_{j}^{(2)\prime}%
\Sigma_{u_{j}}O_{j}^{(2)}\right)^{-1}\mathbb{Z}_{%
\mathfrak{B},i}^{(1)}+o_{p}(1),\quad\forall i\in I_{1},\\
\mathbb{Z}_{\mathfrak{B},i}^{(1)\prime}\left(O_{j}^{(3)\prime}%
\Sigma_{u_{j}}O_{j}^{(3)}\right)^{-1}\mathbb{Z}_{%
\mathfrak{B},i}^{(1)}+o_{p}(1),\quad\forall i\in I_{2},\\
\mathbb{Z}_{\mathfrak{B},i}^{(1)\prime}\left(O_{j}^{(1)\prime}%
\Sigma_{u_{j}}O_{j}^{(1)}\right)^{-1}\mathbb{Z}_{%
\mathfrak{B},i}^{(1)}+o_{p}(1),\quad\forall i\in I_{3}, \end{aligned}\right.
\label{E.6}
\end{equation}%
with $\mathbb{Z}_{\mathfrak{B},i}^{(1)\prime }\left( O_{j}^{(2)\prime
}\Sigma _{u_{j}}O_{j}^{(2)}\right) ^{-1}\mathbb{Z}_{\mathfrak{B}%
,i}^{(1)}\rightarrow \chi ^{2}(1)$ for each $i\in I_{1}$, $\mathbb{Z}_{%
\mathfrak{B},i}^{(1)\prime }\left( O_{j}^{(3)\prime }\Sigma
_{u_{j}}O_{j}^{(3)}\right) ^{-1}\mathbb{Z}_{\mathfrak{B},i}^{(1)}\rightarrow
\chi ^{2}(1)$ for each $i\in I_{2}$, and $\mathbb{Z}_{\mathfrak{B}%
,i}^{(1)\prime }\left( O_{j}^{(1)\prime }\Sigma _{u_{j}}O_{j}^{(1)}\right)
^{-1}\mathbb{Z}_{\mathfrak{B},i}^{(1)}\rightarrow \chi ^{2}(1)$ for each $%
i\in I_{3}$. As in \cite{castagnetti2015inference}, we can conclude that 
\begin{equation}
\mathbb{P}\left( \frac{1}{2}S_{u_{j}}\leq x+\mathsf{b}(N)\right) \rightarrow
e^{-e^{-x}}\ \text{as }\left( N,T\right) \rightarrow \infty .  \label{E.7}
\end{equation}

For the test statistic for $H_{0}^{II}$, the proof is similar. Let $\mathbb{Z%
}_{\mathfrak{B}}^{(2)}=\left( \mathbb{Z}_{\mathfrak{B},1}^{(2)\prime
},\cdots ,\mathbb{Z}_{\mathfrak{B},T}^{(2)\prime }\right) ^{\prime }$, where 
$\mathbb{Z}_{\mathfrak{B},t}^{(2)}\sim N\left( 0,O_{j}^{(1)\prime }\Sigma
_{v_{j}}O_{j}^{(1)}\right) $. As in (\ref{E.6}), we can show that 
\begin{equation*}
S_{v_{j}}^{(3,1)}=\max_{t\in \lbrack T]}\mathbb{Z}_{\mathfrak{B}%
,t}^{(2)\prime }\left( O_{j}^{(1)\prime }\Sigma _{v_{j}}O_{j}^{(1)}\right)
^{-1}\mathbb{Z}_{\mathfrak{B},t}^{(2)}+o_{p}(1)\text{.}
\end{equation*}%
By the strong mixing condition in Assumption \ref{ass:1}(iii), we have $%
\max_{j\in \lbrack p],i\in \lbrack N]}\left\Vert Cov\left( \mathfrak{b}%
_{j,it}^{(2)},\mathfrak{b}_{j,is}^{(2)}\right) \right\Vert _{\infty }\log
\left( t-s\right) =o(1)$ as $t-s\rightarrow \infty $ by Davydov's
inequality. Then by Theorem 3.5.1 in \cite{leadbetter1988extremal}, we have
that 
\begin{equation*}
\mathbb{P}\left( \frac{1}{2}\left( \max_{t\in \lbrack T]}\mathbb{Z}_{%
\mathfrak{B},t}^{(2)\prime }\left( O_{j}^{(1)\prime }\Sigma
_{v_{j}}O_{j}^{(1)}\right) ^{-1}\mathbb{Z}_{\mathfrak{B},t}^{(2)}\right)
\leq x+\mathsf{b}(T)\right) \rightarrow e^{-e^{-x}},
\end{equation*}%
which implies that 
\begin{equation*}
\mathbb{P}\left( \frac{1}{2}S_{v_{j}}^{(3,1)}\leq x+\mathsf{b}(T)\right)
\rightarrow e^{-e^{-x}}\text{ as }\left( N,T\right) \rightarrow \infty .\quad
\end{equation*}%
Recall that $\tilde{S}_{v_{j}}^{(a,b)}=\frac{1}{2}S_{v_{j}}^{(a,b)}-\mathsf{b%
}(T)$ and $S_{v_{j}}=\max (\tilde{S}_{v_{j}}^{(1,2)},\tilde{S}%
_{v_{j}}^{(2,3)},\tilde{S}_{v_{j}}^{(3,1)}).$ Noting that $S_{v_{j}}$ is
asymptotically distributed as the maximum of three independent Gumbel random
variables under $H_{0}^{II}$, we have $\mathbb{P(}S_{v_{j}}\leq
x)\rightarrow e^{-3e^{-x}}$ as $\left( N,T\right) \rightarrow \infty .$

\subsubsection{Proof of Statement (ii)}

Under $H_{1}^{I}$, we have that $\forall i\in I_{a}$, 
\begin{align*}
\hat{u}_{i,j}^{(a,b)}-\hat{\bar{u}}_{j}^{(a,b)}& =\left( \hat{u}%
_{i,j}^{(a,b)}-O_{u,j}^{(b)}u_{i,j}^{0}\right) +O_{u,j}^{(b)}\left(
u_{i,j}^{0}-u_{j}\right) -\left( \hat{\bar{u}}%
_{j}^{(a,b)}-O_{u,j}^{(b)}u_{j}\right) \\
& =\left( \hat{u}_{i,j}^{(a,b)}-O_{u,j}^{(b)}u_{i,j}^{0}\right)
+O_{u,j}^{(b)}c_{i,j}^{u}-\left( \hat{\bar{u}}%
_{j}^{(a,b)}-O_{u,j}^{(b)}u_{j}\right) .
\end{align*}%
Then 
\begin{align*}
& T\left( \hat{u}_{i,j}^{(a,b)}-\hat{\bar{u}}_{j}^{(a,b)}\right) ^{\prime
}\left( \hat{\Sigma}_{u_{j}}\right) ^{-1}\left( \hat{u}_{i,j}^{(a,b)}-\hat{%
\bar{u}}_{j}^{(a,b)}\right) \\
& =T\left( \hat{u}_{i,j}^{(a,b)}-O_{u,j}^{(b)}u_{j}\right) ^{\prime }\hat{%
\Sigma}_{u_{j}}^{-1}\left( \hat{u}_{i,j}^{(1)}-O_{u,j}^{(b)}u_{j}\right)
+T\left( \hat{\bar{u}}_{j}^{(1)}-O_{u,j}^{(b)}u_{j}\right) ^{\prime }\hat{%
\Sigma}_{u_{j}}^{-1}\left( \hat{\bar{u}}_{j}^{(a,b)}-O_{u,j}^{(b)}u_{j}%
\right) \\
& -2T\left( \hat{u}_{i,j}^{(a,b)}-O_{u,j}^{(b)}u_{j}\right) ^{\prime }\hat{%
\Sigma}_{u_{j}}^{-1}\left( \hat{\bar{u}}_{j}^{(1)}-O_{u,j}^{(b)}u_{j}\right)
+2T\left( O_{u,j}^{(b)}c_{i,j}^{u}\right) ^{\prime }\hat{\Sigma}%
_{u_{j}}^{-1}\left( \hat{u}_{i,j}^{(a,b)}-O_{u,j}^{(b)}u_{j}\right) \\
& +T\left( O_{u,j}^{(b)}c_{i,j}^{u}\right) ^{\prime }\hat{\Sigma}%
_{u_{j}}^{-1}O_{u,j}^{(b)}c_{i,j}^{u}-2T\left(
O_{u,j}^{(b)}c_{i,j}^{u}\right) ^{\prime }\hat{\Sigma}_{u_{j}}^{-1}\left( 
\hat{\bar{u}}_{j}^{(a,b)}-O_{u,j}^{(b)}u_{j}\right) \\
:&
=S_{u_{j},i,1}^{(b)}+S_{u_{j},2}^{(b)}+S_{u_{j},i,3}^{(b)}+S_{u_{j},i,4}^{(b)}+S_{u_{j},i,5}^{(b)}+S_{u_{j},i,6}^{(b)},
\end{align*}%
where $\max_{i\in I_{a}}\left\vert S_{u_{j},i,1}^{(b)}\right\vert
=O_{p}\left( \log N\right) $ by (\ref{E.5}) and (\ref{E.7}), $\left\vert
S_{u_{j},2}^{(b)}\right\vert =o_{p}\left( 1\right) $ by (\ref{E.3}), $%
\max_{i\in I_{3}}\left\vert S_{u_{j},i,3}^{(b)}\right\vert =o_{p}\left(
1\right) $ by (\ref{E.4}). Next, 
\begin{equation*}
\max_{i\in I_{a}}\left\vert S_{u_{j},i,5}^{(b)}\right\vert =\max_{i\in
I_{a}}T\left( O_{u,j}^{(b)}c_{i,j}^{u}\right) ^{\prime }\Sigma
_{u_{j}}^{-1}O_{u,j}^{(b)}c_{i,j}^{u}\left\{ 1+o_{p}(1)\right\} \gtrsim 
\left[ \max_{i\in I_{a}}\lambda _{\max }\left( \Sigma _{u,j}\right) \right]
^{-1}T\max_{i\in I_{a}}\left\Vert c_{i,j}^{u}\right\Vert _{2}^{2}\asymp
T\max_{i\in I_{a}}\left\Vert c_{i,j}^{u}\right\Vert _{2}^{2},
\end{equation*}%
which diverges to infinity at the rate faster than $\log N$ by condition in
statement (ii). By Cauchy-Schwarz inequality, we obtain that 
\begin{align*}
& \max_{i\in I_{a}}\left\vert S_{u_{j},i,4}^{(b)}\right\vert \lesssim
\max_{i\in I_{a}}\left( S_{u_{j},i,1}^{(b)}\right) ^{1/2}\left(
S_{u_{j},i,5}^{(b)}\right) ^{1/2}=o_{p}\left( \max_{i\in I_{a}}\left\vert
S_{u_{j},i,5}^{(b)}\right\vert \right) ,\text{ and} \\
& \max_{i\in I_{a}}\left\vert S_{u_{j},i,6}^{(b)}\right\vert \lesssim
\max_{i\in I_{a}}\left( S_{u_{j},i,2}^{a}\right) ^{1/2}\left(
S_{u_{j},i,5}^{(b)}\right) ^{1/2}=o_{p}\left( \max_{i\in I_{a}}\left\vert
S_{u_{j},i,5}^{(b)}\right\vert \right) .
\end{align*}%
It follows that $\mathbb{P}\left\{ S_{u_{j}}>c_{\alpha ,N}\right\}
\rightarrow 1$ as $c_{\alpha ,N}\asymp \log N$, and the final result follows.

The power of the test statistic $S_{v_{j}}$ can be analyzed analogously. $%
\blacksquare$

\subsection{Proof of Theorem \protect\ref{Thm6}}

We first derive the linear expansion of $\hat{\Theta}_{j,it}^{\ast
}-\Theta_{j,it}^{\ast }$ for $i\in I_{3}$, and similar results hold for $%
i\in I_{1}\cup I_{2}$. Let $\bar{v}_{j}^{0}:=\frac{1}{T}\sum_{t\in[T]%
}v_{t,j}^{0}\ $and $\bar{u}_{j}^{0,I_{a}}:=\frac{1}{N_{a}}\sum_{i\in
I_{a}}u_{i,j}^{0}$. By Proposition \ref{Pro4}, we obtain that uniformly in $%
i\in I_{a},$ 
\begin{align}
&\hat{\bar{\Theta}}_{j,i\cdot }-\bar{\Theta}_{j,i\cdot } =\frac{1}{T}\sum_{t%
\in[T]}\left(\hat{\Theta}_{j,it}-\Theta_{j,it}^{0}\right)  \notag \\
&=\frac{1}{T}\sum_{t\in[T]}u_{i,j}^{0\prime}\left(\hat{V}_{v_{j},t}^{(a)}%
\right) ^{-1}\frac{1}{N_{a}}\sum_{i^{*}\in I_{a}}\xi_{j,i^{*}t}^{0}+\frac{1}{%
T}\sum_{t\in[T]}v_{t,j}^{0\prime}\hat{V}_{u_{j}}^{-1}\frac{1}{T}%
\sum_{t=1}^{T}b_{j,it}^{0} +\frac{1}{T}\sum_{t\in[T]}\mathcal{R}_{it}^{j} 
\notag \\
&=\bar{v}_{j}^{0\prime}\hat{V}_{u_{j}}^{-1}\frac{1}{T}%
\sum_{t=1}^{T}b_{j,it}^{0}+o_{p}\left(\left(N\vee T\right)^{-1/2}\right) ,
\label{thetabar_i}
\end{align}%
where the third equality holds by the fact that 
\begin{align*}
& \max_{i\in I_{a}}\left\vert \frac{1}{T}\sum_{t\in
[T]}u_{i,j}^{0\prime}\left(\hat{V}_{v_{j},t}^{(a)}\right) ^{-1}\frac{1}{N_{a}%
}\sum_{i^{*}\in I_{a}}\xi_{j,i^{*}t}^{0} \right\vert=\max_{i\in
I_{a}}\left\vert \frac{1}{N_{a}T}\sum_{i^{*}\in I_{a}}\sum_{t\in[T]%
}u_{i,j}^{0\prime}\left(\hat{V}_{v_{j},t}^{(a)}\right)
^{-1}\xi_{j,i^{*}t}^{0} \right\vert \\
&=O_{p}\left(\sqrt{\frac{\log N}{NT}}\xi_{N}\right) =o_{p}\left(\left(N\vee
T\right) ^{-1/2}\right),
\end{align*}
by conditional Bernstein's inequality given $\mathscr{D}_{e}$ in Lemma \ref%
{Lem:Bern}(i), Assumption \ref{ass:1}(i)-(ii) and Assumption \ref{ass:1}%
(ix). Similarly, uniformly in $t\in \left[T\right]$, we have 
\begin{align}
&\hat{\bar{\Theta}}_{j,\cdot t}^{I_{a}}-\bar{\Theta}_{j,\cdot t}^{I_{a}} =%
\frac{1}{N_{a}}\sum_{i\in I_{3}}\left(\hat{\Theta}_{j,it}-\Theta_{j,it}^{0}%
\right)  \notag \\
&=\frac{1}{N_{a}}\sum_{i\in I_{a}}u_{i,j}^{0\prime}\left(\hat{V}%
_{v_{j},t}^{(a)}\right) ^{-1}\frac{1}{N_{a}}\sum_{i\in I_{a}}\xi_{j,it}^{0}+%
\frac{1}{N_{a}}\sum_{i\in I_{a}}v_{t,j}^{0\prime}\hat{V}_{u_{j}}^{-1}\frac{1%
}{T}\sum_{t^{*}=1}^{T}b_{j,it^{*}}^{0}+\frac{1}{N_{a}}\sum_{i\in I_{a}}%
\mathcal{R}_{it}^{j}  \notag \\
&=\bar{u}_{j}^{0,I_{a}\prime }\left(\hat{V}_{v_{j},t}^{(a)}\right) ^{-1}%
\frac{1}{N_{a}}\sum_{i\in I_{a}}\xi_{j,it}^{0} +o_{p}\left(\left(N\vee
T\right)^{-1/2}\right) ,  \label{thetabar_t}
\end{align}%
and 
\begin{align}
&\hat{\bar{\Theta}}_{j}^{I_{a}}-\bar{\Theta}_{j}^{I_{a}} =\frac{1}{N_{a}T}%
\sum_{i\in I_{a}}\sum_{t\in[T]}\left(\hat{\Theta}_{j,it}-\Theta_{j,it}^{0}%
\right)  \notag \\
&=\frac{1}{N_{a}T}\sum_{i\in I_{a}}\sum_{t\in[T]}u_{i,j}^{0\prime}\left(\hat{%
V}_{v_{j},t}^{(a)}\right) ^{-1}\frac{1}{N_{a}}\sum_{i\in
I_{a}}\xi_{j,it}^{0}+\frac{1}{N_{a}T}\sum_{i\in I_{a}}\sum_{t\in[T]%
}v_{t,j}^{0\prime}\hat{V}_{u_{j}}^{-1}\frac{1}{T}\sum_{t=1}^{T}b_{j,it}^{0} +%
\frac{1}{N_{a}T}\sum_{i\in I_{a}}\sum_{t\in[T]}\mathcal{R}_{it}^{j}  \notag
\\
&=o_{p}\left(\left(N\vee T\right) ^{-1/2}\right) .  \label{thetabar}
\end{align}%
Combining (\ref{thetabar_i})-(\ref{thetabar}), we obtain that $\forall i\in
I_{a}$ and $t\in \left[T\right]$, 
\begin{align*}
&\hat{\Theta}_{j,it}^{\ast }-\Theta_{j,it}^{\ast } =\hat{\Theta}%
_{j,it}-\Theta_{j,it}^{0}-\left(\hat{\bar{\Theta}}_{j,i\cdot }-\bar{\Theta}%
_{j,i\cdot }\right) -\left(\hat{\bar{\Theta}}_{j,\cdot t}^{I_{a}}-\bar{\Theta%
}_{j,\cdot t}^{I_{a}}\right) +\left(\hat{\bar{\Theta}}_{j}^{I_{a}}-\bar{%
\Theta}_{j}^{I_{a}}\right) \\
& =\left(u_{i,j}^{0}-\bar{u}_{j}^{0,I_{a}}\right) ^{\prime}\left(\hat{V}%
_{v_{j},t}^{(a)}\right) ^{-1}\frac{1}{N_{a}}\sum_{i\in
I_{a}}\xi_{j,it}^{0}+\left(v_{t,j}^{0}-\bar{v}_{j}^{0}\right) ^{\prime}\hat{V%
}_{u_{j}}^{-1}\frac{1}{T}\sum_{t=1}^{T}b_{j,it}^{0}+\bar{\mathcal{R}}%
_{it}^{j},
\end{align*}
such that $\max_{i\in I_{3},t\in[T]}\left\vert \bar{\mathcal{R}}%
_{it}^{j}\right\vert =o_{p}\left(\left(N\vee T\right) ^{-1/2}\right) $. It
follows that $\hat{\Theta}_{j,it}^{\ast
}-\Theta_{j,it}^{\ast}\rightsquigarrow \mathcal{N}(0,\Sigma_{j,it}^{*})$
with 
\begin{equation*}
\Sigma_{j,it}^{*}=\sum_{a\in[3]}\frac{1}{N_{a}}%
\left(O_{j}^{(b)}u_{i,j}^{0}-O_{j}^{(b)}\bar{u}_{j}^{0}\right)
^{\prime}\Sigma_{v_{j}}\left(O_{j}^{(b)}u_{i,j}^{0}-O_{j}^{(b)}\bar{u}%
_{j}^{0}\right)\mathbf{1}_{ia} +\frac{1}{T}%
\left(O_{j}^{(b)}v_{t,j}^{0}-O_{j}^{(b)}\bar{v}_{j}^{0}\right)
^{\prime}\Sigma_{u_{j}}\left(O_{j}^{(b)}v_{t,j}^{0}-O_{j}^{(b)}\bar{v}%
_{j}^{0}\right) ,
\end{equation*}%
and $\Sigma_{u_{j}}$ and $\Sigma_{v_{j}}$ are as defined in Theorem \ref%
{Thm3}. The reason why $\Sigma_{j,it}^{*}$ is not indexed with $b$ is owing
to the fact that $O_{j}^{(b)}$ shown in the right side of the equality can
be absorbed by $O_{j}^{(b)}$ not shown in $\Sigma_{u_{j}}$ and $%
\Sigma_{v_{j}}$.

Define $\hat{\Sigma}_{j,it}^{*}=\frac{1}{2}\operatornamewithlimits{\sum}_{b}%
\hat{\Sigma}_{j,it}^{(b)*}$ with 
\begin{equation*}
\hat{\Sigma}_{j,it}^{(b)*}=\sum_{a\in[3]}\left[\frac{1}{N_{a}}\left(\hat{u}%
_{i,j}^{(a,b)}-\hat{\bar{u}}_{j}^{(a,b)}\right) ^{\prime}\hat{\Sigma}%
_{v_{j}}\left(\hat{u}_{i,j}^{(a,b)}-\bar{u}_{j}^{(a,b)}\right)\mathbf{1}%
_{ia} +\frac{1}{T}\left(\hat{v}_{t,j}^{(a,b)}-\hat{\bar{v}}%
_{j}^{(a,b)}\right) ^{\prime}\hat{\Sigma}_{u_{j}}\left(\hat{v}_{t,j}-\hat{%
\bar{v}}_{j}^{(a,b)}\right)\right],
\end{equation*}%
where $\hat{\bar{u}}_{j}^{(a,b)}=\frac{1}{N_{a}}\sum_{i\in I_{a}}\hat{u}%
_{i,j}^{(a,b)}$ and $\hat{\bar{v}}_{j}^{(a,b)}=\frac{1}{T}\sum_{t\in[T]}\hat{%
v}_{t,j}^{(a,b)}$. By Theorem \ref{Thm3} and Lemma \ref{Lem:covhat}, we have 
\begin{equation*}
\max_{j\in[p],i\in [N],t\in[T]}\left\vert \hat{\Sigma}_{j,it}^{*}-%
\Sigma_{j,it}^{*}\right\vert =o_{p}(1).
\end{equation*}%
By arguments as used in the proof of Theorem \ref{Thm3}, we have that as $%
\left(N,T\right) \rightarrow \infty ,$ $\mathbb{P}\left(\frac{1}{2}%
S_{NT}\leq x+\mathsf{b}(NT)\right) \rightarrow e^{-e^{-x}}\ $under $%
H_{0}^{III},$ and $\mathbb{P}\left(S_{NT}>c_{\alpha ,3\cdot
NT}\right)\rightarrow 1$ under $H_{1}^{III}$ provided $\frac{N\wedge T}{\log
NT}\max_{i\in [N],t\in[T]}\left\vert \Theta_{j,it}^{\ast}\right\vert
^{2}\rightarrow \infty $. $\blacksquare$


\section{Some Technical Lemmas}

In this section we state and prove some technical lemmas that are used in
the proofs of the main results in the paper.

\subsection{Lemmas for the Proof of Theorem \protect\ref{Thm1}}

\begin{lemma}
{\small \label{Lem:matrix Bern} } Consider a matrix sequence $%
\left\{A_{i},i=1,\cdots,N\right\}$ whose values are symmetric matrices with
dimension $d$,

\begin{itemize}
\item[(i)] suppose $\left\{ A_{i},i=1,\cdots ,N\right\} $ is independent
with $\mathbb{E}\left(A_{i}\right) =0$ and $\left\Vert
A_{i}\right\Vert_{op}\leq M$ a.s. Let $\sigma ^{2}=\left\Vert \sum_{i\in [N]}%
\mathbb{E}\left(A_{i}^{2}\right) \right\Vert_{op}$. Then for all $t>0$, we
have 
\begin{equation*}
\mathbb{P}\left(\left\Vert \sum_{i\in [N]}A_{i}\right\Vert_{op}>t\right)
\leq d\cdot \exp \left\{ -\frac{t^{2}/2}{\sigma ^{2}+Mt/3}\right\} .
\end{equation*}

\item[(ii)] suppose $\left\{ A_{i},i=1,\cdots ,N\right\} $ is sequence of
martingale difference matrices with $\mathbb{E}_{i-1}\left(A_{i}\right) =0$
and $\left\Vert A_{i}\right\Vert_{op}\leq M$ a.s., where $\mathbb{E}_{i-1}$
denotes $\mathbb{E}\left(\cdot |\mathscr{F}_{i-1}\right) $, where $\left\{%
\mathscr{F}_{i}:i\leq N\right\} $ denotes the filtration that is clear from
the context. Let $\left\Vert \sum_{i\in [N]}\mathbb{E}_{i-1}\left(A_{i}^{2}%
\right) \right\Vert_{op}\leq \sigma ^{2}$. Then for all $t>0$, we have 
\begin{equation*}
\mathbb{P}\left(\left\Vert \sum_{i\in [N]}A_{i}\right\Vert_{op}>t\right)
\leq d\cdot \exp \left\{ -\frac{t^{2}/2}{\sigma ^{2}+Mt/3}\right\} .
\end{equation*}
\end{itemize}
\end{lemma}

\begin{proof}
Lemma \ref{Lem:matrix Bern}(i) and (ii) are Matrix Bernstein inequality and
Matrix Freedman inequality, which are\ respectively stated in Theorem 1.3
and Corollary 4.2 in \cite{tropp2011user}.
\end{proof}

\begin{lemma}
{\small \label{Lem:matrix op} } Consider a specific matrix $A\in \mathbb{R}%
^{N\times T}$ whose rows (denoted as $A_{i}^{\prime}$ where $A_{i}\in 
\mathbb{R}^{T}$) are independent random vectors in $\mathbb{R}^{T}$ with $%
\mathbb{E}A_{i}=0$ and $\Sigma_{i}=\mathbb{E}\left(A_{i}A_{i}^{\prime}%
\right) $. Suppose $\max_{i\in [N]}\left\Vert A_{i}\right\Vert_{2}\leq \sqrt{%
m}$ a.s. and $\max_{i\in [N]}\left\Vert \Sigma_{i}\right\Vert_{op}\leq M $
for some positive constant $M$. Then for every $t>0$, with probability $%
1-2T\exp \left(-c_{1}t^{2}\right) $, we have 
\begin{equation*}
\left\Vert A\right\Vert_{op}\leq \sqrt{NM}+t\sqrt{m+M},
\end{equation*}
where $c_{1}$ is an absolute constant.
\end{lemma}

\begin{proof}
The proof follows similar arguments as used in the proof of Theorem 5.41 in 
\cite{vershynin2010introduction}. Define $Z_{i}:=\frac{1}{N}%
\left(A_{i}A_{i}^{\prime}-\Sigma_{i}\right) \in \mathbb{R}^{T\times T}$. We
notice that $\left\{ Z_{i}\right\} $ is an independent sequence with $%
\mathbb{E}\left(Z_{i}\right) =0$. To use the matrix Bernstein inequality, we
analyze $\left\Vert Z_{i}\right\Vert_{op}$ and $\left\Vert \sum_{i\in [N]}%
\mathbb{E}\left(X_{i}^{2}\right) \right\Vert_{op}$ as follows. First, note
that uniformly over $i,$ 
\begin{equation}
\left\Vert Z_{i}\right\Vert_{op}\leq \frac{1}{N}\left(\left\Vert
A_{i}A_{i}^{\prime}\right\Vert_{op}+\left\Vert
\Sigma_{i}\right\Vert_{op}\right) \leq \frac{1}{N}\left(\left\Vert
A_{i}\right\Vert_{2}^{2}+\left\Vert \Sigma_{i}\right\Vert_{op}\right) \leq 
\frac{m+M}{N},\quad \text{a.s.}  \label{Lem2.1}
\end{equation}%
Next, noting that $\mathbb{E}\left[ \left(A_{i}A_{i}^{\prime}\right) ^{2}%
\right] =\mathbb{E}\left[ \left\Vert A_{i}\right\Vert_{2}A_{i}A_{i}^{\prime}%
\right] \leq m\Sigma_{i}$ and $Z_{i}^{2}=\frac{1}{N^{2}}[\left(A_{i}A_{i}^{%
\prime}\right) ^{2}-A_{i}A_{i}^{\prime}\Sigma_{i}$ $-\Sigma_{i}A_{i}A_{i}^{%
\prime}+\Sigma_{i}^{2}],$ we have 
\begin{align*}
\left\Vert \mathbb{E}\left(Z_{i}^{2}\right) \right\Vert_{op}& =\frac{1}{N^{2}%
}\left\Vert \mathbb{E}\left[ \left(A_{i}A_{i}^{\prime}\right)^{2}-%
\Sigma_{i}^{2}\right] \right\Vert_{op}\leq \frac{1}{N^{2}}\left\{\left\Vert 
\mathbb{E}\left[ \left(A_{i}A_{i}^{\prime}\right) ^{2}\right]%
\right\Vert_{op}+\left\Vert \Sigma_{i}\right\Vert_{op}^{2}\right\} \\
& \leq \frac{1}{N^{2}}\left(m\left\Vert
\Sigma_{i}\right\Vert_{op}+\left\Vert \Sigma_{i}\right\Vert_{op}^{2}\right)
\leq \frac{mM+M^{2}}{N^{2}}\text{ uniformly in }i.
\end{align*}
It follows that 
\begin{equation}
\left\Vert \sum_{i\in [N]}\mathbb{E}\left(Z_{i}^{2}\right)\right\Vert_{op}%
\leq N\max_{i\in [N]}\left\Vert \mathbb{E}\left(Z_{i}^{2}\right)
\right\Vert_{op}\leq \frac{mM+M^{2}}{N}.  \label{Lem2.2}
\end{equation}
Let $\varepsilon=\max \left(\sqrt{M}\delta ,\delta ^{2}\right) $ with $%
\delta =t\sqrt{\frac{m+M}{N}}$. By (\ref{Lem2.1})-(\ref{Lem2.2}) and the
matrix Bernstein inequality in Lemma \ref{Lem:matrix Bern}(i), we have 
\begin{align*}
\mathbb{P}\left\{ \left\Vert \frac{1}{N}\left(A^{\prime}A-\sum_{i\in
[N]}\Sigma_{i}\right) \right\Vert_{op}\geq \varepsilon \right\} & =\mathbb{P}%
\left(\left\Vert \sum_{i\in [N]}Z_{i}\right\Vert_{op}\geq\varepsilon \right)
\\
& \leq 2T\exp \left\{ -c_{1}\min \left(\frac{\varepsilon ^{2}}{\frac{mM+M^{2}%
}{N}},\frac{\varepsilon }{\frac{m+M}{N}}\right) \right\} \\
& \leq 2T\exp \left\{ -c_{1}\min \left(\frac{\varepsilon ^{2}}{M}%
,\varepsilon \right) \frac{N}{m+M}\right\} \\
& \leq 2T\exp \left\{ -\frac{c_{1}\delta ^{2}N}{m+M}\right\} =2T\exp
\left\{-c_{1}t^{2}\right\} ,
\end{align*}%
for some positive constant $c_{1}$, where the third inequality is due to the
fact that 
\begin{align*}
\min \left(\frac{\varepsilon ^{2}}{M},\varepsilon \right) & =\min \left(\max
\left(\delta ^{2},\delta ^{4}/M\right) ,\max \left(\sqrt{M}\delta,\delta
\right) \right) \\
& =\left\{ \begin{aligned}
&\min\left(\delta^{2},\sqrt{M}\delta\right)=\delta^{2}\quad\text{if}\quad%
\delta^{2}\geq \frac{\delta^{4}}{M},\\
&\min\left(\delta^{4}/M,\delta^{2}\right)=\delta^{2},\quad
\text{if}\quad\delta^{2}< \frac{\delta^{4}}{M}. \end{aligned}\right. .
\end{align*}%
That is, 
\begin{equation}
\left\Vert \frac{1}{N}A^{\prime}A-\frac{1}{N}\sum_{i\in
[N]}\Sigma_{i}\right\Vert_{op}\leq \max \left(\sqrt{M}\delta ,\delta
^{2}\right)  \label{Lem2.3}
\end{equation}%
with probability $1-\exp \left(-c_{1}t^{2}\right) $. Combining the fact that 
$\left\Vert \Sigma_{i}\right\Vert \leq M$ uniformly over $i$ and (\ref%
{Lem2.3}), we show that 
\begin{align*}
\frac{1}{N}\left\Vert A\right\Vert_{op}^{2}& =\left\Vert \frac{1}{N}%
A^{\prime}A\right\Vert_{op}\leq \left\Vert \frac{1}{N}\sum_{i\in
[N]}\Sigma_{i}\right\Vert_{op}+\left\Vert \frac{1}{N}A^{\prime}A-\frac{1}{N}%
\sum_{i\in[N]}\Sigma_{i}\right\Vert_{op} \\
& \leq \max_{i\in [N]}\left\Vert \Sigma_{i}\right\Vert_{op}+\sqrt{M}\delta
+\delta ^{2} \\
& \leq M+\sqrt{M}t\sqrt{\frac{m+M}{N}}+t^{2}\frac{m+M}{N}\leq \left(\sqrt{M}%
+t\sqrt{\frac{m+M}{N}}\right) ^{2}.
\end{align*}
It follows that $\left\Vert A\right\Vert_{op}\leq \sqrt{NM}+t\sqrt{m+M}.$
\end{proof}

\begin{lemma}
{\small \label{Lem:score op} } Recall $a_{it}=\tau -\mathbf{1}%
\left\{\epsilon_{it}\leq 0\right\} $ and $a=\{a_{it}\}\in \mathbb{R}%
^{N\times T}.$ Under Assumption \ref{ass:1}, we have $\left\Vert X_{j}\odot
a\right\Vert_{op}=O_{p}\left(\sqrt{N}+\sqrt{T\log T}\right) $ $\forall j\in
[p]$ and $\left\Vert a\right\Vert_{op}=O_{p}\left(\sqrt{N}+\sqrt{T\log T}%
\right) $.
\end{lemma}

\begin{proof}
We focus on $\left\Vert X_{j}\odot a\right\Vert _{op}$ as the result for $%
||a||_{op}$ can be derived in the same manner. We first note that,
conditional $\mathscr{D}$, the $i$-th row of $X_{j}\odot a$ only depends on $%
\{e_{it},\epsilon _{it}\}_{t\in \lbrack T]}$, which are independent across $%
i $. Therefore $X_{j}\odot a$ has independent rows, denoted as $%
A_{i}=X_{j,i}\odot a_{i}$, given $\mathscr{D}$, where $X_{j,i}$ and $a_{i}$
being the $i$-th row of matrix $X_{j}$ and $a$, respectively. In addition,
for the $t$-th element of $A_{i}$, we have 
\begin{equation*}
\mathbb{E}\left[ X_{j,it}a_{it}\bigg|\mathscr{D}\right] =\mathbb{E}\left\{
X_{j,it}\mathbb{E}\left[ a_{it}\bigg|\mathscr{D}_{e}\right] \bigg|\mathscr{D}%
\right\} =0,
\end{equation*}%
where the second equality holds by Assumption \ref{ass:1}(ii) and the fact
that given $\mathscr{D}_{e}$, $X_{j,it}$ is known. Therefore, In order to
apply Lemma \ref{Lem:matrix Bern}, conditionally on $\mathscr{D}$, we only
need to upper bound $\left\Vert A_{i}\right\Vert _{2}$ and $\mathbb{E}\left[
A_{i}A_{i}^{\prime }\bigg|\mathscr{D}\right] $.

First, under Assumption \ref{ass:1}(iv), we have $\frac{1}{T}\sum_{t\in[T]%
}\left(X_{j,it}a_{it}\right) ^{2}\leq \frac{1}{T}\sum_{t\in[T]%
}X_{j,it}^{2}\leq c_{2}\quad $a.s. for some positive constant $c_{2}$, which
implies 
\begin{equation}
||A_{i}||_{2}=\left\Vert X_{j,i}\odot a_{i}\right\Vert_{2}\leq c_{2}\sqrt{T}%
\quad \text{a.s.}  \label{Lem3.1}
\end{equation}

Second, let $\Sigma_{i}=\mathbb{E}\left\{ \left[ \left(X_{j,i}\odot
a_{i}\right) \left(X_{j,i}\odot a_{i}\right) ^{\prime}\right] \bigg|%
\mathscr{D}\right\} $ with $\left(t,s\right)^{th}$ element being $\mathbb{E}%
\left(X_{j,it}X_{j,is}a_{it}a_{is}\bigg|\mathscr{D}\right) .$ Recall that $%
\left\Vert \cdot \right\Vert_{1}$ and $\left\Vert \cdot \right\Vert_{\infty
} $ are matrix norms induced by 1- and $\infty -$norms, i.e., 
\begin{equation*}
\left\Vert \Sigma_{i}\right\Vert_{1}=\max_{s\in [T]}\sum_{t\in[T]}\left\vert 
\mathbb{E}\left(X_{j,it}X_{j,is}a_{it}a_{is}\bigg|\mathscr{D}\right)
\right\vert ,\quad \left\Vert \Sigma_{i}\right\Vert_{\infty }=\max_{t\in[T]%
}\sum_{s\in [T]}\left\vert \mathbb{E}\left(X_{j,it}X_{j,is}a_{it}a_{is}\bigg|%
\mathscr{D}\right) \right\vert .
\end{equation*}%
By Davydov's inequality for conditional strong mixing sequence, we can show
that 
\begin{align*}
& \max_{s\in [T]}\sum_{t\in[T]}\left\vert \mathbb{E}%
\left(X_{j,it}X_{j,is}a_{it}a_{is}\bigg|\mathscr{D}\right) \right\vert \\
& =\max_{s\in [T]}\sum_{t\in[T]}\left\vert
Cov\left(X_{j,it}a_{it},X_{j,is}a_{is}\bigg|\mathscr{D}\right) \right\vert \\
& \leq \max_{s\in [T]}\sum_{t\in[T]}\left\{ \mathbb{E}\left[\left\vert
X_{j,it}a_{it}\right\vert ^{q}\bigg|\mathscr{D}\right] \right\}^{1/q}\left\{ 
\mathbb{E}\left[ \left\vert X_{j,is}a_{is}\}\right\vert ^{q}\bigg|\mathscr{D}%
\right] \right\} ^{1/q}\times \alpha \left(t-s\right)^{\left(q-2\right) /q}
\\
& \leq \max_{i\in [N],t\in[T]}\left\{ \mathbb{E}\left[\left\vert
X_{j,it}^{q}\right\vert \bigg|\mathscr{D}\right] \right\}^{2/q}\max_{s\in
[T]}\sum_{t\in[T]}\alpha \left(t-s\right)^{\left(q-2\right) /q} \\
& \leq c_{3}~a.s.,
\end{align*}%
where 
$c_{3}$ is a positive constant which does not depend on $i$ and the last
line is by Assumption \ref{ass:1}(iii) and \ref{ass:1}(iv). Similarly, we
have $\max_{t\in[T]}\sum_{s\in [T]}\left\vert \mathbb{E}%
\left(X_{j,it}X_{j,is}a_{it}a_{is}\bigg|\mathscr{D}\right) \right\vert \leq
c_{3}\quad $a.s. Then by Corollary 2.3.2 in \cite{golub1996matrix}, we have 
\begin{equation}
\max_{i\in [N]}\left\Vert \Sigma_{i}\right\Vert_{op}\leq \sqrt{\left\Vert
\Sigma_{i}\right\Vert_{1}\left\Vert \Sigma_{i}\right\Vert_{\infty }}\leq
c_{3}\quad \text{a.s.}  \label{Lem3.2}
\end{equation}%
Combining (\ref{Lem2.1}), (\ref{Lem2.2}), and Lemma \ref{Lem:matrix Bern}
with $t=\sqrt{\log T}$, we obtain the desired result.
\end{proof}

\bigskip

Recall that $\Delta_{\Theta_{j}}=\Theta_{j}-\Theta_{j}^{0}$ for any $%
\Theta_{j}$ and define 
\begin{equation*}
\mathcal{R}(C_{1}):= \left\{
\{\Delta_{\Theta_{j}}\}_{j=0}^{p}:\sum_{j=0}^{p}\left\Vert\mathcal{P}%
_{j}^{\bot}(\Delta_{\Theta_{j}})\right\Vert_{\ast}\leq
C_{1}\sum_{j=0}^{p}\left\Vert\mathcal{P}_{j}(\Delta_{\Theta_{j}})\right%
\Vert_{\ast}\right\}.
\end{equation*}

\begin{lemma}
{\small \label{Lem:RS} } Suppose Assumptions \ref{ass:1}-\ref{ass:3} hold.
Then $\{\tilde{\Delta}_{\Theta_{j}}\}_{j=0}^{p}\in \mathcal{R}(3)$ w.p.a.1.
\end{lemma}

\begin{proof}
Define $\mathbb{Q}_{\tau }\left(\left\{ \Theta_{j}\right\}_{j=0}^{p}\right) =%
\frac{1}{NT}\sum_{i=1}^{N}\sum_{t=1}^{T}\rho_{\tau}\left(Y_{it}-%
\Theta_{0,it}-\sum_{j=1}^{p}X_{j,it}\Theta_{j,it}\right)$ for generic $%
\left\{ \Theta_{j}\right\}_{j=0}^{p}$. By the definition of the nuclear norm
estimator in (\ref{pre rank}), we have 
\begin{equation}
\mathbb{Q}_{\tau }\left(\left\{ \Theta_{j}^{0}\right\}_{j=0}^{p}\right) -%
\mathbb{Q}_{\tau }\left(\left\{ \tilde{\Theta}_{j}\right\}_{j=0}^{p}\right)
+\sum_{j=0}^{p}\nu_{j}\left[ \left\Vert
\Theta_{j}^{0}(\tau)\right\Vert_{\ast }-\left\Vert \tilde{\Theta}%
_{j}(\tau)\right\Vert_{\ast }\right] \geq 0.  \label{Lem4.1}
\end{equation}
In addition, we have 
\begin{align}
& \mathbb{Q}_{\tau }\left(\left\{ \Theta_{j}^{0}\right\}_{j=0}^{p}\right)-%
\mathbb{Q}_{\tau }\left(\left\{ \tilde{\Theta}_{j}\right\}_{j=0}^{p}\right) 
\notag \\
& =-\left\{ \frac{1}{NT}\sum_{i=1}^{N}\sum_{t=1}^{T}a_{it}\left(\tilde{\Delta%
}_{\Theta_{0,it}}+\sum_{j=1}^{p}X_{j,it}\tilde{\Delta}_{\Theta_{j,it}}%
\right) +\int_{0}^{\tilde{\Delta}_{\Theta_{0,it}}+\sum_{j=1}^{p}X_{j,it}%
\tilde{\Delta}_{\Theta_{j,it}}}\mathbf{1}\left\{ \epsilon_{it}\leq s\right\}
-\mathbf{1}\left\{ \epsilon_{it}\leq 0\right\} ds\right\}  \notag \\
& \leq \left\vert \frac{1}{NT}\sum_{i=1}^{N}\sum_{t=1}^{T}a_{it}\left(\tilde{%
\Delta}_{\Theta_{0,it}}+\sum_{j=1}^{p}X_{j,it}\tilde{\Delta}%
_{\Theta_{j,it}}\right) \right\vert  \notag \\
& \leq \sum_{j=1}^{p}\left\vert \frac{1}{NT}\sum_{i=1}^{N}%
\sum_{t=1}^{T}a_{it}X_{j,it}\tilde{\Delta}_{\Theta_{j,it}}\right\vert+\left%
\vert \frac{1}{NT}\sum_{i=1}^{N}\sum_{t=1}^{T}a_{it}\tilde{\Delta}%
_{\Theta_{0,it}}\right\vert  \notag \\
& =\sum_{j=1}^{p}\frac{1}{NT}\left\vert tr\left[ \tilde{\Delta}%
_{\Theta_{j}}^{\prime}\left(X_{j}\odot a\right) \right] \right\vert +\frac{1%
}{NT}\left\vert tr\left(\tilde{\Delta}_{\Theta_{0}}^{\prime}a\right)
\right\vert  \notag \\
& \leq \sum_{j=1}^{p}\frac{1}{NT}\left\Vert \tilde{\Delta}%
_{\Theta_{j}}\right\Vert_{\ast }\left\Vert X_{j}\odot a\right\Vert_{op}+%
\frac{1}{NT}\left\Vert \tilde{\Delta}_{\Theta_{0}}\right\Vert_{\ast
}\left\Vert a\right\Vert_{op}  \notag \\
& \leq c_{4}\sum_{j=0}^{p}\frac{(\sqrt{N}\vee \sqrt{T\log T})}{NT}\left\Vert%
\tilde{\Delta}_{\Theta_{j}}\right\Vert_{\ast },\quad \text{w.p.a.1},
\label{Lem4.2}
\end{align}%
where the first equality holds by Knight's identity in \cite%
{knight1998limiting} which states that 
\begin{equation}
\rho_{\tau }(u-v)-\rho_{\tau }(u)=v(\tau -\mathbf{1}\left\{ u\leq 0\right\}
)+\int_{0}^{v}\left(\mathbf{1}\left\{ u\leq s\right\} -\mathbf{1}\{u\leq
0\}\right) ds,  \label{Lem4.3}
\end{equation}%
the first inequality is due to the fact that the second term in the bracket
of the second line is non-negative, the second inequality holds by triangle
inequality, the third inequality is by the fact that $tr(AB)\leq
||A||_{op}||B||_{\ast }$, and the last inequality holds by Lemma \ref%
{Lem:score op}.

Combining \eqref{Lem4.1} and \eqref{Lem4.2}, w.p.a.1, we have 
\begin{equation}
0\leq c_{4}\sum_{j=0}^{p}\left\{ \frac{(\sqrt{N}\vee \sqrt{T\log T})}{NT}||%
\tilde{\Delta}_{\Theta_{j}}||_{\ast }+\nu_{j}\left(\left\Vert
\Theta_{j}^{0}\right\Vert_{\ast }-\left\Vert \tilde{\Theta}%
_{j}\right\Vert_{\ast}\right) \right\} .  \label{Lem4.4}
\end{equation}%
Besides, we can show that 
\begin{align}
\left\Vert \tilde{\Theta}_{j}\right\Vert_{\ast }& =\left\Vert \tilde{\Delta}%
_{\Theta_{j}}+\Theta_{j}^{0}\right\Vert_{\ast }=\left\Vert \Theta_{j}^{0}+%
\mathcal{P}_{j}^{\bot }(\tilde{\Delta}_{\Theta_{j}})+\mathcal{P}_{j}(\tilde{%
\Delta}_{\Theta_{j}})\right\Vert_{\ast }  \notag \\
& \geq \left\Vert \Theta_{j}^{0}+\mathcal{P}_{j}^{\bot }(\tilde{\Delta}%
_{\Theta_{j}})\right\Vert_{\ast }-\left\Vert \mathcal{P}_{j}(\tilde{\Delta}%
_{\Theta_{j}})\right\Vert_{\ast }=\left\Vert \Theta_{j}^{0}\right\Vert_{\ast
}+\left\Vert \mathcal{P}_{j}^{\bot }(\tilde{\Delta}_{\Theta_{j}})\right%
\Vert_{\ast }-\left\Vert \mathcal{P}_{j}(\tilde{\Delta}_{\Theta_{j}})\right%
\Vert_{\ast }  \label{Lem4.5}
\end{align}%
where the second equality holds by Lemma D.2(i) in \cite%
{chernozhukov2019inference}, the first inequality holds by triangle
inequality, and the last equality is by the construction of the linear space 
$\mathcal{P}_{j}^{\bot }$ and $\mathcal{P}_{j}$. Then combining (\ref{Lem4.4}%
) and (\ref{Lem4.5}), we obtain 
\begin{align*}
& \sum_{j=0}^{p}\nu_{j}\left\{ \left\Vert \mathcal{P}_{j}^{\bot }\left(%
\tilde{\Delta}_{\Theta_{j}}\right) \right\Vert_{\ast }-\left\Vert \mathcal{P}%
_{j}\left(\tilde{\Delta}_{\Theta_{j}}\right) \right\Vert_{\ast}\right\} \leq
c_{4}\sum_{j=0}^{p}\frac{(\sqrt{N}\vee \sqrt{T\log T})}{NT}\left\Vert \tilde{%
\Delta}_{\Theta_{j}}\right\Vert_{\ast } \\
& =c_{4}\sum_{j=0}^{p}\frac{(\sqrt{N}\vee \sqrt{T\log T})}{NT}%
\left\{\left\Vert \mathcal{P}_{j}\left(\tilde{\Delta}_{\Theta_{j}}\right)%
\right\Vert_{\ast }+\left\Vert \mathcal{P}_{j}^{\bot }\left(\tilde{\Delta}%
_{\Theta_{j}}\right) \right\Vert_{\ast }\right\} ,~\text{w.p.a.1}
\end{align*}%
By setting $\nu_{j}=\frac{2c_{4}(\sqrt{N}\vee \sqrt{T\log T})}{NT}$, we
obtain $\sum_{j=0}^{p}\left\Vert \mathcal{P}_{j}^{\bot }\left(\tilde{\Delta}%
_{\Theta_{j}}\right) \right\Vert_{\ast }\leq 3\sum_{j=0}^{p}\left\Vert%
\mathcal{P}_{j}\left(\tilde{\Delta}_{\Theta_{j}}\right) \right\Vert_{\ast}~$%
w.p.a.1$.$
\end{proof}

Recall $\mathscr{G}_{i,t-1}$ is the $\sigma$-field generated by $%
\{V_{j}^{0}\}_{j\in [p]\cup \{0\}}$, $\{W_{j}^{0}\}_{j\in [p]}$, $%
\{e_{ls}\}_{l\leq i-1,s\in [T]}$, $\{e_{is}\}_{s\leq t}$, $%
\{\epsilon_{ls}\}_{l\leq i-1,s\in [T]}$, and $\{\epsilon_{is}\}_{s\leq t-1}$
and $\mathsf{F}_{it}(\cdot)$ and $\mathsf{f}_{it}(\cdot)$ are the
conditional CDF and PDF of $\epsilon_{it}$ given $\mathscr{G}_{i,t-1}$,
respectively. Specifically, we note that $(\{X_{j,it}\}_{j\in
[p]},\{\Theta_{j,it}^{0}\}_{j\in [p]\cup\{0\}})$ are measurable w.r.t. $%
\mathscr{G}_{i,t-1}$.

\begin{lemma}
{\small \label{Lem:ind diff} } For all $u_{1},u_{2}\in \mathbb{R}$ and all $%
c_{5}\in (0,1]$, we have

\begin{itemize}
\item[(i)] $\int_{0}^{u_{2}}\left(\mathbf{1}\left\{u_{1}\leq z\right\}- 
\mathbf{1}\left\{u_{1}\leq 0\right\} \right)dz\geq\int_{0}^{c_{5}u_{2}}\left(%
\mathbf{1}\left\{u_{1}\leq z\right\}- \mathbf{1}\left\{u_{1}\leq 0\right\}
\right)dz\geq 0$,

\item[(ii)] $\int_{0}^{u_{2}}\left\{ \mathsf{F}_{it}\left(u_{1}+z\right)-%
\mathsf{F}_{it}\left(u_{1}\right) \right\}dz\geq\int_{0}^{c_{5}u_{2}}\left\{ 
\mathsf{F}_{it}\left(u_{1}+z\right)-\mathsf{F}_{it}\left(u_{1}\right)
\right\}dz\geq 0$.
\end{itemize}
\end{lemma}

\begin{proof}
Statement (i) is just \citet[Lemma A2]{feng2019regularized}. To prove
statement (ii), notice that if $u_{2}\geq 0$, then $z\geq 0$ and $\mathsf{F}%
_{it}\left(u_{1}+z\right) -\mathsf{F}_{it}\left(u_{1}\right) \geq 0$ for all 
$z\in [0,u_{2}]$, which leads to the existence of the second inequality
naturally: 
\begin{align*}
& \int_{0}^{u_{2}}\left\{ \mathsf{F}_{it}\left(u_{1}+z\right) -\mathsf{F}%
_{it}\left(u_{1}\right) \right\} dz-\int_{0}^{c_{5}u_{2}}\left\{ \mathsf{F}%
_{it}\left(u_{1}+z\right) -\mathsf{F}_{it}\left(u_{1}\right) \right\} dz \\
& =\int_{c_{5}u_{2}}^{u_{2}}\left\{ \mathsf{F}_{it}\left(u_{1}+z\right) -%
\mathsf{F}_{it}\left(u_{1}\right) \right\} dz\geq 0.
\end{align*}%
On the other hand, if $u_{2}<0$, we have 
\begin{align*}
& \int_{0}^{u_{2}}\left\{ \mathsf{F}_{it}\left(u_{1}+z\right) -\mathsf{F}%
_{it}\left(u_{1}\right) \right\} dz-\int_{0}^{c_{5}u_{2}}\left\{ \mathsf{F}%
_{it}\left(u_{1}+z\right) -\mathsf{F}_{it}\left(u_{1}\right) \right\} dz \\
& =\int_{u_{2}}^{0}\left\{ \mathsf{F}_{it}\left(u_{1}\right) -\mathsf{F}%
_{it}\left(u_{1}+z\right) \right\} dz-\int_{c_{5}u_{2}}^{0}\left\{ \mathsf{F}%
_{it}\left(u_{1}\right) -\mathsf{F}_{it}\left(u_{1}+z\right) \right\} dz \\
& =\int_{u_{2}}^{c_{5}u_{2}}\left\{ \mathsf{F}_{it}\left(u_{1}\right) 5%
\mathsf{F}_{it}\left(u_{1}+z\right) \right\} dz\geq 0,
\end{align*}%
where the last inequality holds as the same reason for $u_{2}\geq 0$ case.
\end{proof}

\begin{lemma}
{\small \label{Lem:exp lower} } Under Assumptions \ref{ass:1}-\ref{ass:4},
for any $\left\{ \Delta_{\Theta_{j}}\right\}_{j=0}^{p}\in \mathcal{R}%
(3,C_{2})$ such that $||\Delta_{\Theta_{j}}||_{\max }\leq M$ for some
constant $M>0$, we have 
\begin{equation*}
Q_{\tau }\left(\left\{
\Theta_{j}^{0}+\Delta_{\Theta_{j}}\right\}_{j=0}^{p}\right) -Q_{\tau
}\left(\left\{ \Theta_{j}^{0}\right\}_{j=0}^{p}\right) \geq \frac{c_{7}C_{3}%
}{NT\xi_{N}^{2}}\sum_{j=0}^{p}\left\Vert
\Delta_{\Theta_{j}}\right\Vert_{F}^{2}-\frac{c_{7}C_{4}}{NT\xi_{N}^{2}}%
\left(N+T\right) ~w.p.a.1,
\end{equation*}%
where $Q_{\tau }(\left\{ \Theta_{j}\right\}_{j=0}^{p})=\frac{1}{NT}%
\sum_{i=1}^{N}\sum_{t=1}^{T}\mathbb{E}\left[ \rho_{\tau
}\left(Y_{it}-\Theta_{0,it}-\sum_{j=1}^{p}X_{j,it}\Theta_{j,it}\right) 
\bigg
|\mathscr{G}_{i,t-1}\right]$, $c_{7}=\frac{\underline{\mathsf{f}}c_{6}^{2}}{4%
}$ with $c_{6}$ being a positive constant between 0 and 1.
\end{lemma}

\begin{proof}
We can choose a sufficiently large constant $M$ such that $c_{6}:=\frac{3%
\underline{\mathsf{f}}}{2\bar{\mathsf{f}}^{\prime }M(1+p)}\in \left(0,1%
\right]$. Then we have 
\begin{align}
& Q_{\tau }\left(\left\{
\Theta_{j}^{0}+\Delta_{\Theta_{j}}\right\}_{j=0}^{p}\right) -Q_{\tau
}\left(\left\{ \Theta_{j}^{0}\right\}_{j=0}^{p}\right)  \notag
\label{Lem6.1} \\
& =\frac{1}{NT}\sum_{i=1}^{N}\sum_{t=1}^{T}\mathbb{E}\left\{
\rho_{\tau}\left(Y_{it}-\Theta_{0,it}-\sum_{j=1}^{p}X_{j,it}\Theta_{j,it}%
\right)-\rho_{\tau
}\left(Y_{it}-\Theta_{0,it}^{0}-\sum_{j=1}^{p}X_{j,it}\Theta_{j,it}^{0}%
\right) \bigg |\mathscr{G}_{i,t-1}\right\}  \notag \\
& =\frac{1}{NT}\sum_{i=1}^{N}\sum_{t=1}^{T}\mathbb{E}\left\{
\left(\Delta_{\Theta_{0,it}}+\sum_{j=1}^{p}X_{j,it}\Delta_{\Theta_{j,it}}%
\right)\left(\tau -\mathbf{1}\left\{ \epsilon_{it}\leq 0\right\} \right) %
\bigg |\mathscr{G}_{i,t-1}\right\}  \notag \\
& +\frac{1}{NT}\sum_{i=1}^{N}\sum_{t=1}^{T}\mathbb{E}\left\{\int_{0}^{%
\Delta_{\Theta_{0,it}}+\sum_{j=1}^{p}X_{j,it}\Delta_{\Theta_{j,it}}}\left(%
\mathbf{1}\left\{ \epsilon_{it}\leq s\right\} -\mathbf{1}\left\{
\epsilon_{it}\leq 0\right\} \right) ds\bigg |\mathscr{G}_{i,t-1}\right\} 
\notag \\
& =\frac{1}{NT}\sum_{i=1}^{N}\sum_{t=1}^{T}\mathbb{E}\left\{\int_{0}^{%
\Delta_{\Theta_{0,it}}+\sum_{j=1}^{p}X_{j,it}\Delta_{\Theta_{j,it}}}\left(%
\mathbf{1}_{\epsilon_{it}\leq s}-\mathbf{1}_{\epsilon_{it}\leq 0}\right) ds%
\bigg |\mathscr{G}_{i,t-1}\right\}  \notag \\
& \geq \frac{1}{NT}\sum_{i=1}^{N}\sum_{t=1}^{T}\mathbb{E}\left\{%
\int_{0}^{c_{6}\xi_{N}^{-1}\left(\Delta_{\Theta_{0,it}}+%
\sum_{j=1}^{p}X_{j,it}\Delta_{\Theta_{j,it}}\right) }\left(\mathbf{1}%
_{\epsilon_{it}\leq s}-\mathbf{1}_{\epsilon_{it}\leq 0}\right) ds\bigg |%
\mathscr{G}_{i,t-1}\right\}  \notag \\
& =\frac{1}{NT}\sum_{i=1}^{N}\sum_{t=1}^{T}\int_{0}^{c_{6}\xi_{N}^{-1}\left(%
\Delta_{\Theta_{0,it}}+\sum_{j=1}^{p}X_{j,it}\Delta_{\Theta_{j,it}}\right) }%
\left[ \mathsf{F}_{it}\left(s\right) -\mathsf{F}_{it}\left(0\right) \right]
ds  \notag \\
& =\frac{1}{NT}\sum_{i=1}^{N}\sum_{t=1}^{T}\int_{0}^{c_{6}\xi_{N}^{-1}\left(%
\Delta_{\Theta_{0,it}}+\sum_{j=1}^{0}X_{j,it}\Delta_{\Theta_{j,it}}\right) }%
\left[ s\mathsf{f}_{it}\left(0\right) +\frac{s^{2}}{2}\mathsf{f}%
_{it}^{\prime }\left(\tilde{s}\right) \right] ds  \notag \\
& \geq \frac{\underline{\mathsf{f}}c_{6}^{2}}{2NT\xi_{N}^{2}}%
\sum_{i=1}^{N}\sum_{t=1}^{T}\left(\Delta_{\Theta
_{0,it}}+\sum_{j=1}^{p}X_{j,it}\Delta_{\Theta_{j,it}}\right) ^{2}-\frac{\bar{%
\mathsf{f}}^{\prime }c_{6}^{3}}{6NT\xi_{N}^{3 }}\sum_{i=1}^{N}\sum_{t=1}^{T}%
\left\vert
\Delta_{\Theta_{0,it}}+\sum_{j=1}^{p}X_{j,it}\Delta_{\Theta_{j,it}}\right%
\vert ^{3}  \notag \\
& =\frac{\underline{\mathsf{f}}c_{6}^{2}}{4NT\xi_{N}^{2}}\sum_{i=1}^{N}%
\sum_{t=1}^{T}\left(\Delta_{\Theta_{0,it}}+\sum_{j=1}^{p}X_{j,it}\Delta_{%
\Theta_{j,it}}\right) ^{2}  \notag \\
& +\frac{1}{NT\xi_{N}^{2}}\sum_{i=1}^{N}\sum_{t=1}^{T}\left\{ \frac{%
\underline{\mathsf{f}}c_{6}^{2}}{4}\left(\Delta_{\Theta_{0,it}}+%
\sum_{j=1}^{p}X_{j,it}\Delta_{\Theta_{j,it}}\right) ^{2}\left(1-\frac{2c_{6}%
\bar{\mathsf{f}}^{\prime }}{3\underline{\mathsf{f}}\xi_{N}}\left\vert
\Delta_{\Theta_{0,it}}+\sum_{j=1}^{p}X_{j,it}\Delta_{\Theta_{j,it}}\right%
\vert \right) \right\}  \notag \\
& \geq \frac{\underline{\mathsf{f}}c_{6}^{2}}{4NT\xi_{N}^{2}}%
\sum_{i=1}^{N}\sum_{t=1}^{T}\left\{
\Delta_{\Theta_{0,it}}+\sum_{j=1}^{p}X_{j,it}\Delta_{\Theta_{j,it}}\right\}
^{2}  \notag \\
& =\frac{\underline{\mathsf{f}}c_{6}^{2}}{4NT\xi_{N}^{2}}\left\Vert
\Delta_{\Theta_{0}}+\sum_{j=1}^{p}X_{j}\odot
\Delta_{\Theta_{j}}\right\Vert_{F}^{2}  \notag \\
& \geq \frac{\underline{\mathsf{f}}c_{6}^{2}}{4NT\xi_{N}^{2}}%
\left\{C_{3}\sum_{j=0}^{p}\left\Vert
\Delta_{\Theta_{j}}\right\Vert_{F}^{2}-C_{4}\left(N+T\right) \right\} \quad 
\text{w.p.a.1}  \notag \\
& =\frac{c_{7}C_{3}}{NT\xi_{N}^{2}}\sum_{j=0}^{p}\left\Vert
\Delta_{\Theta_{j}}\right\Vert_{F}^{2}-\frac{c_{7}C_{4}}{NT\xi_{N}^{2}}%
\left(N+T\right)
\end{align}%
where the second equality is by (\ref{Lem4.3}), the first inequality is by
Lemma \ref{Lem:ind diff} and the fact that $c_{6}/\xi_{N}\leq 1$, the fifth
equality is by the mean-value theorem, the third inequality is by the fact
that 
\begin{equation*}
1-\frac{2c_{6}\bar{\mathsf{f}}^{\prime }}{3\underline{\mathsf{f}}\xi_{N}}%
\left\vert
\Delta_{\Theta_{0,it}}+\sum_{j=1}^{p}X_{j,it}\Delta_{\Theta_{j,it}}\right%
\vert \geq 1-\frac{2c_{6}\bar{\mathsf{f}}^{\prime }}{3\underline{\mathsf{f}}%
\xi_{N}}M(1+p)\xi_{N}\geq 0,
\end{equation*}%
and the fourth inequality holds under Assumption \ref{ass:4}. This concludes
the proof.
\end{proof}

\begin{lemma}
{\small \label{Lem:delta nuclear} } Under Assumptions \ref{ass:1}--\ref%
{ass:4}, for any $\left\{ \Delta_{\Theta_{j}}\right\}_{j=0}^{p}\in\mathcal{R}%
(3,C_{2})$, we have $\left\Vert \Delta_{\Theta_{j}}\right\Vert_{\ast }\leq
c_{8}\sum_{j=0}^{p}\left\Vert
\Delta_{\Theta_{j}}\right\Vert_{F}\,\,\,\forall j\in \left\{ 0,\cdots
,p\right\} $ where $c_{8}=4\sqrt{2\bar{K}}$.
\end{lemma}

\begin{proof}
For $\forall j\in \left\{ 0,\cdots ,p\right\} $, we obtain that 
\begin{align*}
\left\Vert \Delta_{\Theta_{j}}\right\Vert_{\ast }& =\left\Vert \mathcal{P}%
_{j}(\Delta_{\Theta_{j}})\right\Vert_{\ast }+\left\Vert \mathcal{P}%
_{j}^{\bot }(\Delta_{\Theta_{j}})\right\Vert_{\ast }\leq \left\Vert \mathcal{%
P}_{j}(\Delta_{\Theta_{j}})\right\Vert_{\ast}+3\sum_{j=0}^{p}\left\Vert 
\mathcal{P}_{j}(\Delta_{\Theta_{j}})\right\Vert_{\ast } \\
& \leq 4\sum_{j=0}^{p}\left\Vert \mathcal{P}_{j}(\Delta_{\Theta_{j}})\right%
\Vert_{\ast }\leq 4\sum_{j=0}^{p}\sqrt{2K_{j}}\left\Vert\mathcal{P}%
_{j}\left(\Delta_{\Theta_{j}}\right) \right\Vert_{F} \\
& \leq 4\sqrt{2\bar{K}}\sum_{j=0}^{p}\left\Vert
\Delta_{\Theta_{j}}\right\Vert_{F}:=c_{8}\sum_{j=0}^{p}\left\Vert
\Delta_{\Theta_{j}}\right\Vert_{F},
\end{align*}%
where the first equality is by 
\citet[Lemma
D.2(i)]{chernozhukov2019inference}, the first inequality is by the
definition of $\mathcal{R}(C_{1},C_{2})$, and the last two inequalities
follow the facts that $\left\Vert A\right\Vert_{\ast }\leq \sqrt{\text{rank}%
(A)}\left\Vert A\right\Vert_{F}$ for any matrix $A\in \mathbb{R}^{N\times T}$
and $\text{rank}\left(\mathcal{P}_{j}\left(\Delta_{\Theta_{j}}\right)\right)
\leq 2K_{j}$, which hold by 
\citet[ Lemma
D.2.(iii)]{chernozhukov2019inference}.
\end{proof}

\bigskip

Let $\mathcal{Z}$ be a separable metric space, $\left\{
Z_{1},\cdots,Z_{n}\right\} $ be a sequence of random variables in $\mathcal{Z%
}$ adapted to the filtration $\{\mathscr{F}_{t}\}_{t\in [n]}$, $\mathcal{F}%
=\left\{ f:\mathcal{Z}\rightarrow \mathbb{R}\right\} $ be a set of bounded
real valued functions on $\mathbb{R}$, and $u_{1},\cdots ,u_{n}$ be i.i.d.
Rademacher random variables. Here we allow the dependence of sequence $%
\left\{ Z_{1},\cdots ,Z_{n}\right\} $. Similarly as in \cite%
{rakhlin2015sequential}, we define a $\mathcal{Z}$-valued tree $\mathbf{z}$
of depth $n$ with the sequence $\left(\mathbf{z}_{1},\cdots ,\mathbf{z}%
_{n}\right) $ such that $\mathbf{z}_{t}:\left\{ u_{1},\cdots,u_{t-1}\right\}
\rightarrow \mathcal{Z}$. For simplicity, we denote as $\mathbf{z}%
_{t}\left(u\right) =\mathbf{z}_{t}\left(u_{1},\cdots,u_{t-1}\right) $ and $%
\mathbb{E}_{t}\left[ \cdot \right] =\mathbb{E}\left[\cdot \big|\mathscr{F}%
_{t}\right]$ for short. Also, denote $u_{1:t}:=\left(u_{1},\cdots
,u_{t}\right) $ and similarly for $Z_{1:t}$.

\begin{lemma}
{\small \label{Lem:symmetrization} } Let $\mathcal{F}$ be a class of
functions. For any $\alpha >0$, it holds that 
\begin{equation*}
\beta_{n}\mathbb{P}\left\{ \sup_{f\in \mathcal{F}}\left\vert \frac{1}{n}%
\sum_{t=1}^{n}\left(f\left(Z_{t}\right) -\mathbb{E}_{t-1}\left[
f\left(Z_{t}\right) \right] \right) \right\vert >\alpha \right\} \leq 2\sup_{%
\mathbf{z}}\mathbb{P}\left\{ \sup_{f\in \mathcal{F}}\left\vert \frac{1}{n}%
\sum_{t=1}^{n}u_{t}f\left(\mathbf{z}_{t}\left(u\right) \right) \right\vert>%
\frac{\alpha }{4}\right\} .
\end{equation*}%
where $\beta_{n}\geq 1-\sup_{f\in \mathcal{F}}\frac{4\sum_{t=1}^{n}Var%
\left(f\left(Z_{t}\right) \big |\mathscr{F}_{t-1}\right) }{n^{2}\alpha ^{2}}$
and the outer supremum is taken over all $\mathcal{Z}$-valued tree of depth $%
n$.\footnote{%
This means the supremum is taken over all $\mathbf{z}=\{\mathbf{z}%
_{t}(\cdot)\}_{t\in [n]}$.}
\end{lemma}

\begin{proof}
Let $Z_{1:n}^{\prime }$ be a decoupled sequence tangent to $Z_{1:n}$. For
the sequence of random variables $\left\{ Z_{t}:t\in [n]\right\} $ adapted
to the filtration $\left\{ \mathscr{F}_{t}:t\in [n]\right\} $, the sequence $%
Z_{1:n}^{\prime }=\left\{ Z_{t}^{*},t\in [n]\right\} $ is said to be a
decoupled sequence tangent to $\left\{Z_{t}:t\in [n]\right\} $ if for each $%
t\in [n]$, $Z_{t}^{*}$ is generated from the conditional distribution of $%
Z_{t}$ given $\mathscr{F}_{t-1}$ and independent of everything else. This
means the sequence $Z_{1:n}^{\prime }$ is conditionally independent given $%
\mathscr{F}_{n}$ and for any measurable function $f$ of $Z_{t}^{*}$, 
\begin{equation}
\mathbb{E}(f(Z_{t}^{*})|\mathscr{F}_{n})=\mathbb{E}(f(Z_{t}^{*})|\mathscr{F}%
_{t-1})~a.s.  \label{eq:z'}
\end{equation}

For $\forall f\in \mathcal{F}$, with Chebyshev's inequality, we have 
\begin{align*}
& \mathbb{P}\left( \frac{1}{n}\left\vert \sum_{t=1}^{n}f\left( Z_{t}^{\ast
}\right) -\mathbb{E}_{t-1}\left[ f\left( Z_{t}^{\ast }\right) \right]
\right\vert >\alpha /2\bigg |\mathscr{F}_{n}\right) \\
& \leq \frac{4\mathbb{E}\left\{ \left( \sum_{t=1}^{n}f\left( Z_{t}^{\ast
}\right) -\mathbb{E}_{t-1}\left[ f\left( Z_{t}^{\ast }\right) \right]
\right) ^{2}\bigg|\mathscr{F}_{n}\right\} }{n^{2}\alpha ^{2}} \\
& =\frac{4\sum_{t=1}^{n}\mathbb{E}\left\{ \left( f\left( Z_{t}^{\ast
}\right) -\mathbb{E}_{t-1}\left[ f\left( Z_{t}^{\ast }\right) \right]
\right) ^{2}\bigg|\mathscr{F}_{n}\right\} }{n^{2}\alpha ^{2}} \\
& =\frac{4\sum_{t=1}^{n}Var\left( f\left( Z_{t}\right) \big |\mathscr{F}%
_{t-1}\right) }{n^{2}\alpha ^{2}},
\end{align*}%
where the first equality holds by the fact that, given $\mathscr{F}_{n}$, $%
\{Z_{t}^{\ast }\}_{t\in \lbrack n]}$ are independent and the last equality
holds by \eqref{eq:z'}. This implies 
\begin{align*}
\beta _{n}& :=\inf_{f\in \mathcal{F}}\mathbb{P}\left( \frac{1}{n}\left\vert
\sum_{t=1}^{n}f\left( Z_{t}^{\ast }\right) -\mathbb{E}_{t-1}\left[ f\left(
Z_{t}^{\ast }\right) \right] \right\vert \leq \alpha /2\bigg |\mathscr{F}%
_{n}\right) \\
& =1-\sup_{f\in \mathcal{F}}\mathbb{P}\left( \frac{1}{n}\left\vert
\sum_{t=1}^{n}f\left( Z_{t}^{\ast }\right) -\mathbb{E}_{t-1}\left[ f\left(
Z_{t}^{\ast }\right) \right] \right\vert >\alpha /2\bigg|\mathscr{F}%
_{n}\right) \\
& \geq 1-\sup_{f\in \mathcal{F}}\frac{4\sum_{t=1}^{n}Var\left( f\left(
Z_{t}\right) \big |\mathscr{F}_{t-1}\right) }{n^{2}\alpha ^{2}}.
\end{align*}%
Let function $f^{\ast }$ be the function that maximizes $\frac{1}{n}%
\left\vert \sum_{t=1}^{n}\left[ f\left( Z_{t}^{\ast }\right) -\mathbb{E}%
_{t-1}\left( f\left( Z_{t}^{\ast }\right) \right) \right] \right\vert $
condition on $\mathscr{F}_{n}$, and define the event $A_{1}=\left\{
\sup_{f\in \mathcal{F}}\frac{1}{n}\left\vert \sum_{t=1}^{n}\left[ f\left(
Z_{t}\right) -\mathbb{E}_{t-1}\left( f\left( Z_{t}\right) \right) \right]
\right\vert >\alpha \right\} $. Then we obtain that 
\begin{equation*}
\beta _{n}\leq \mathbb{P}\left( \frac{1}{n}\left\vert \sum_{t=1}^{n}f^{\ast
}\left( Z_{t}^{\ast }\right) -\mathbb{E}_{t-1}\left[ f^{\ast }\left(
Z_{t}\right) \right] \right\vert \leq \alpha /2\bigg |\mathscr{F}_{n}\right)
,
\end{equation*}%
where the inequality follows by the definition of $\beta _{n}$ and the fact
that $\mathbb{E}_{t-1}\left[ f^{\ast }\left( Z_{t}^{\ast }\right) \right] =%
\mathbb{E}_{t-1}\left[ f^{\ast }\left( Z_{t}\right) \right] $. As $A_{1}\in %
\mathscr{F}_{n}$, we have 
\begin{equation*}
\beta _{n}\leq \mathbb{P}\left( \frac{1}{n}\left\vert \sum_{t=1}^{n}f^{\ast
}\left( Z_{t}^{\ast }\right) -\mathbb{E}_{t-1}\left[ f^{\ast }\left(
Z_{t}^{\ast }\right) \right] \right\vert \leq \alpha /2\bigg |A_{1}\right) .
\end{equation*}%
It follows that 
\begin{align}
& \beta _{n}\mathbb{P}\left\{ \sup_{f\in \mathcal{F}}\left\vert \frac{1}{n}%
\sum_{t=1}^{n}\left( f\left( Z_{t}\right) -\mathbb{E}_{t-1}\left[ f\left(
Z_{t}\right) \right] \right) \right\vert >\alpha \right\}  \notag
\label{Lem8.1} \\
& \leq \mathbb{P}\left( \frac{1}{n}\left\vert \sum_{t=1}^{n}f^{\ast }\left(
Z_{t}^{\ast }\right) -\mathbb{E}_{t-1}\left[ f^{\ast }\left( Z_{t}^{\ast
}\right) \right] \right\vert \leq \alpha /2\bigg |A_{1}\right) \mathbb{P}%
\left( A_{1}\right)  \notag \\
& =\mathbb{P}\left( \left\{ \frac{1}{n}\left\vert \sum_{t=1}^{n}f^{\ast
}\left( Z_{t}^{\ast }\right) -\mathbb{E}_{t-1}\left[ f^{\ast }\left(
Z_{t}^{\ast }\right) \right] \right\vert \leq \alpha /2\right\} \bigcup
A_{1}\right)  \notag \\
& \leq \mathbb{P}\left( \frac{1}{n}\left\vert \sum_{t=1}^{n}f^{\ast }\left(
Z_{t}\right) -f^{\ast }\left( Z_{t}^{\ast }\right) \right\vert \geq \alpha
/2\right)  \notag \\
& \leq \mathbb{P}\left( \frac{1}{n}\sup_{f\in \mathcal{F}}\left\vert
\sum_{t=1}^{n}f\left( Z_{t}\right) -f\left( Z_{t}^{\ast }\right) \right\vert
\geq \alpha /2\right) ,
\end{align}%
where the second inequality holds by the implication rule. Let $\phi \left(
\cdot \right) =\mathbf{1}\left\{ \cdot >n\alpha /2\right\} .$ By Lemma 18 in 
\cite{rakhlin2015sequential}, we have 
\begin{align}
& \mathbb{P}\left( \frac{1}{n}\sup_{f\in \mathcal{F}}\left\vert
\sum_{t=1}^{n}f\left( Z_{t}\right) -f\left( Z_{t}^{\ast }\right) \right\vert
\geq \alpha /2\right)  \notag \\
& =\mathbb{E}\mathbf{1}\left\{ \frac{1}{n}\sup_{f\in \mathcal{F}}\left\vert
\sum_{t=1}^{n}f\left( Z_{t}\right) -f\left( Z_{t}^{\ast }\right) \right\vert
\geq \alpha /2\right\}  \notag \\
& \leq \sup_{z_{1},z_{1}^{\prime }}\mathbb{E}_{u_{1}}\cdots
\sup_{z_{n},z_{n}^{\prime }}\mathbb{E}_{u_{n}}\mathbf{1}\left\{ \sup_{f\in 
\mathcal{F}}\left\vert \sum_{t=1}^{n}u_{t}\left[ f\left( z_{t}\right)
-f\left( Z_{t}^{\ast }\right) \right] \right\vert \geq n\alpha /2\right\} 
\notag \\
& \leq \sup_{z_{1},z_{1}^{\prime }}\mathbb{E}_{u_{1}}\cdots
\sup_{z_{n},z_{n}^{\prime }}\mathbb{E}_{u_{n}}\mathbf{1}\left\{ \sup_{f\in 
\mathcal{F}}\left\vert \sum_{t=1}^{n}u_{t}f\left( z_{t}\right) \right\vert
\geq n\alpha /4\right\} +\sup_{z_{1},z_{1}^{\prime }}\mathbb{E}%
_{u_{1}}\cdots \sup_{z_{n},z_{n}^{\prime }}\mathbb{E}_{u_{n}}\mathbf{1}%
\left\{ \sup_{f\in \mathcal{F}}\left\vert \sum_{t=1}^{n}u_{t}f\left(
Z_{t}^{\ast }\right) \right\vert \geq n\alpha /4\right\}  \notag \\
& =2\sup_{z_{1},z_{1}^{\prime }}\mathbb{E}_{u_{1}}\cdots
\sup_{z_{n},z_{n}^{\prime }}\mathbb{E}_{u_{n}}\mathbf{1}\left\{ \sup_{f\in 
\mathcal{F}}\left\vert \sum_{t=1}^{n}u_{t}f\left( z_{t}\right) \right\vert
\geq n\alpha /4\right\}  \notag \\
& =2\sup_{\mathbf{z}}\mathbb{P}\left\{ \sup_{f\in \mathcal{F}}\left\vert 
\frac{1}{n}\sum_{t=1}^{n}u_{t}f\left( \mathbf{z}_{t}\left( u\right) \right)
\right\vert >\frac{\alpha }{4}\right\} ,  \label{Lem8.2}
\end{align}%
where the standard text $z_{t}$ and $Z_{t}^{\ast }$ are $\mathcal{Z}$%
-valued, the bold text $\mathbf{z}_{t}=\mathbf{z}_{t}(u)$ is the $t$-th root
of a tree $\mathbf{z}(\cdot )$ (i.e., a function of $u_{1:t-1}$), and the
outer supremum in the last line is taken over all $\mathcal{Z}$-valued tree
of depth $n$. Combining (\ref{Lem8.1}) and (\ref{Lem8.2}), we can conclude
that 
\begin{equation*}
\beta _{n}\mathbb{P}\left\{ \sup_{f\in \mathcal{F}}\left\vert \frac{1}{n}%
\sum_{t=1}^{n}\left( f\left( Z_{t}\right) -\mathbb{E}_{t-1}\left[ f\left(
Z_{t}\right) \right] \right) \right\vert >\alpha \right\} \leq 2\sup_{%
\mathbf{z}}\mathbb{P}\left\{ \sup_{f\in \mathcal{F}}\left\vert \frac{1}{n}%
\sum_{t=1}^{n}u_{t}f\left( \mathbf{z}_{t}\left( u\right) \right) \right\vert
>\frac{\alpha }{4}\right\} .
\end{equation*}
\end{proof}

The next lemma is an extension of the contraction principle, i.e., %
\citet[Theorem 4.12]{ledoux1991probability}, to the case with sequential
symmetrization.

\begin{lemma}
{\small \label{Lem:contraction} } Let function $F:\mathbb{R}_{+}\rightarrow 
\mathbb{R}_{+}$ be convex and increasing and $\phi _{t}:\mathbb{R}%
\rightarrow \mathbb{R}$ be contractions such that $\phi _{t}\left( 0\right)
=0$, $\mathbf{z}_{t}$ is the $t$-th root of a tree ($\mathbf{z}$) which
depends on $\left\{ u_{1},\cdots ,u_{t-1}\right\} $. Then we have 
\begin{equation*}
\mathbb{E}F\left\{ \frac{1}{2}\sup_{f\in \mathcal{F}}\left\vert
\sum_{t=1}^{n}u_{t}\phi _{t}\left( f\left( \mathbf{z}_{t}(u)\right) \right)
\right\vert \right\} \leq \mathbb{E}F\left\{ \sup_{f\in \mathcal{F}%
}\left\vert \sum_{t=1}^{n}u_{t}f\left( \mathbf{z}_{t}(u)\right) \right\vert
\right\} .
\end{equation*}
\end{lemma}


\begin{proof}
We first consider the statement without the absolute value. Let function $G:%
\mathbb{R}\rightarrow \mathbb{R}$ be convex and increasing. We observe that 
\begin{equation*}
\mathbb{E}G\left\{ \sup_{f\in \mathcal{F}}\sum_{t=1}^{n}u_{t}\phi_{t}\left(f%
\left(\mathbf{z}_{t}(u)\right) \right) \right\} =\mathbb{E}\left\{ \mathbb{E}%
\left[ G\left\{ \sup_{f\in \mathcal{F}}\sum_{t=1}^{n}u_{t}\phi_{t}\left(f%
\left(\mathbf{z}_{t}(u)\right) \right) \right\} \bigg|u_{1:n-1}\right]
\right\}
\end{equation*}%
and 
\begin{align*}
\mathbb{E}\left\{ G\left(\sup_{f\in \mathcal{F}}\sum_{t=1}^{n}u_{t}\phi_{t}%
\left(f\left(\mathbf{z}_{t}(u)\right) \right) \bigg|u_{1:n-1}\right)\right\}
& =\mathbb{E}\left\{ G\left(\sup_{f\in \mathcal{F}}\sum_{t=1}^{n-1}u_{t}%
\phi_{t}\left(f\left(\mathbf{z}_{t}(u)\right)\right)
+u_{n}\phi_{n}\left(f\left(\mathbf{z}_{n}(u)\right) \right) \bigg|%
u_{1:n-1}\right) \right\} \\
& =\mathbb{E}\left\{ G\left(\sup_{k_{1},k_{2}\in \mathcal{K}%
}(k_{1}+u_{n}\phi_{n}\left(k_{2}\right) )\right) \bigg |u_{1:n-1}\right\} ,
\end{align*}%
where $k_{1}=\sum_{t=1}^{n-1}u_{t}\phi_{t}\left(f\left(\mathbf{z}%
_{t}(u)\right) \right) $, $k_{2}=f\left(\mathbf{z}_{n}(u)\right) $, and $%
\mathcal{K}=\{(k_{1},k_{2}):f\in \mathcal{F}\}\subset \mathbb{R}^{2}$. We
also note that $k_{1}$ and $k_{2}$ only depend on $u_{1:n-1}$ and is
independent of $u_{n}$. The proof in 
\citet[Theorem
4.12]{ledoux1989comparison} shows 
\begin{equation*}
\mathbb{E}G\left\{ \sup_{k_{1},k_{2}\in \mathcal{K}}k_{1}+u_{2}\phi_{2}%
\left(k_{2}\right) \right\} \leq \mathbb{E}G\left\{\sup_{k_{1},k_{2}\in 
\mathcal{K}}k_{1}+u_{2}k_{2}\right\} ,
\end{equation*}%
which implies 
\begin{equation*}
\mathbb{E}\left\{ G\left(\sup_{f\in \mathcal{F}}\sum_{t=1}^{n}u_{t}\phi_{t}%
\left(f\left(\mathbf{z}_{t}(u)\right) \right) \bigg|u_{1:n-1}\right)\right\}
\leq \mathbb{E}\left\{ G\left(\sup_{f\in \mathcal{F}}\sum_{t=1}^{n-1}u_{t}%
\phi_{t}\left(f\left(\mathbf{z}_{t}(u)\right)\right) +u_{n}f\left(\mathbf{z}%
_{n}(u)\right) \bigg|u_{1:n-1}\right)\right\} .
\end{equation*}%
Taking expectation on both sides, we have 
\begin{equation*}
\mathbb{E}\left\{ G\left(\sup_{f\in \mathcal{F}}\sum_{t=1}^{n}u_{t}\phi_{t}%
\left(f\left(\mathbf{z}_{t}(u)\right) \right) \right) \right\} \leq\mathbb{E}%
\left\{ G\left(\sup_{f\in \mathcal{F}}\sum_{t=1}^{n-1}u_{t}\phi_{t}\left(f%
\left(\mathbf{z}_{t}(u)\right) \right) +u_{n}f\left(\mathbf{z}_{n}(u)\right)
\right) \right\} .
\end{equation*}

Next, let $k_{1}=\sum_{t=1}^{n-2}u_{t}\phi_{t}\left(f\left(\mathbf{z}%
_{t}(u)\right) \right) $, $k_{2}=f(\mathbf{z}_{n-1}(u))$ and $%
k_{3}(u_{n-1})=u_{n}f(\mathbf{z}_{n}(u))$. We emphasize that $k_{1}$ and $%
k_{2}$ only depend on $u_{1:n-2}$ and $u_{n}$ while $k_{3}$ also depends on $%
u_{n-1} $. The dependence of $(k_{1},k_{2},k_{3}(\cdot))$ on $f$ is made
implicit for notation simplicity. Furthermore, given the fact that $u_{n-1}$
only takes values \{-1,1\}, we have 
\begin{equation*}
k_{3}(u_{n-1})=\frac{u_{n-1}+1}{2}k_{3}(1)+\frac{1-u_{n-1}}{2}k_{3}(-1)=%
\frac{k_{3}(1)+k_{3}(-1)}{2}+u_{n-1}\frac{k_{3}(1)-k_{3}(-1)}{2}.
\end{equation*}%
Given these notation and conditioning on $(u_{1:n-2},u_{n})$, we have 
\begin{align*}
& \mathbb{E}_{u_{n-1}}\left\{ G\left(\sup_{f\in \mathcal{F}%
}\sum_{t=1}^{n-1}u_{t}\phi_{t}\left(f\left(\mathbf{z}_{t}(u)\right)\right)
+u_{n}f\left(\mathbf{z}_{n}(u)\right) \right) \right\} \\
& =\mathbb{E}_{u_{n-1}}\left[ G\left(\sup_{f\in \mathcal{F}%
}k_{1}+u_{n-1}\phi_{n-1}(k_{2})+k_{3}(u_{n-1})\right) \right] \\
& =\mathbb{E}_{u_{n-1}}\left[ G\left(\sup_{f\in \mathcal{F}}\left(k_{1}+%
\frac{k_{3}(1)-k_{3}(-1)}{2}\right) +u_{n-1}\left(\phi_{n-1}(k_{2})+\frac{%
k_{3}(1)-k_{3}(-1)}{2}\right) \right) \right] \\
& =\mathbb{E}_{u_{n-1}}\left[ G\left(\sup_{(h_{1},h_{2},h_{3})\in \mathcal{H}%
}h_{1}+u_{n-1}\left(\phi_{n-1}(h_{2})+h_{3}\right) \right) \right] ,
\end{align*}%
where $\mathbb{E}_{u_{n-1}}$ means the expectation is taken conditionally on 
$(u_{1:n-2},u_{n})$, $h_{1}=\sum_{t=1}^{n-2}u_{t}\phi_{t}\left(f\left(%
\mathbf{z}_{t}(u)\right) \right) +\frac{u_{n}(f(\mathbf{z}%
_{n}(u_{1:n-2},1))+f(\mathbf{z}_{n}(u_{1:n-2},-1)))}{2}$, $h_{2}=k_{2}$, $%
h_{3}=\frac{u_{n}(f(\mathbf{z}_{n}(u_{1:n-2},1))-f(\mathbf{z}%
_{n}(u_{1:n-2},-1)))}{2}$, and $\mathcal{H}=((h_{1},h_{2},h_{3}):f\in 
\mathcal{F})\in \mathbb{R}^{3}$. Suppose $(h_{1}^{\ast },h_{2}^{\ast
},h_{3}^{\ast })\in \mathcal{H}$ and $(h_{1}^{\dagger },h_{2}^{\dagger
},h_{3}^{\dagger }) \in \mathcal{H}$ achieve the supremum of 
\begin{equation*}
h_{1}+\left(\phi_{n-1}(h_{2})+h_{3}\right) \quad \text{and}\quad
h_{1}-\left(\phi_{n-1}(h_{2})+h_{3}\right) ,\quad \text{respectively}.
\end{equation*}%
Then, we have 
\begin{align*}
& \mathbb{E}_{u_{n-1}}\left[ G\left(\sup_{(h_{1},h_{2},h_{3})\in \mathcal{H}%
}h_{1}+u_{n-1}\left(\phi_{n-1}(h_{2})+h_{3}\right) \right) \right] \\
& =\frac{1}{2}G((h_{1}^{\ast }+h_{3}^{\ast })+\phi_{n-1}(h_{2}^{\ast }))+%
\frac{1}{2}G((h_{1}^{\dagger }-h_{3}^{\dagger })-\phi_{n-1}(h_{2}^{\dagger}))
\\
& \leq \frac{1}{2}G((h_{1}^{\ast }+h_{3}^{\ast })+h_{2}^{\ast })+\frac{1}{2}%
G((h_{1}^{\dagger }-h_{3}^{\dagger })-h_{2}^{\dagger }) \\
& \leq \mathbb{E}_{u_{n-1}}\left[ G\left(\sup_{(h_{1},h_{2},h_{3})\in%
\mathcal{H}}h_{1}+u_{n-1}\left(h_{2}+h_{3}\right) \right) \right] ,
\end{align*}%
where the first inequality is by the fact proved in the proof of %
\citet[Theorem 4.12]{ledoux1991probability} that for any $%
t_{1},s_{1},t_{2},s_{2}$, 
\begin{equation*}
\frac{1}{2}G(t_{1}+\phi_{n-1}(t_{2}))+\frac{1}{2}G(s_{1}-\phi_{n-1}(s_{2}))%
\leq \frac{1}{2}G(t_{1}+t_{2})+\frac{1}{2}G(s_{1}-s_{2}).
\end{equation*}

Plugging back the definition of $h_1,h_2,h_3$, we have 
\begin{align*}
&\mathbb{E}_{u_{n-1}}\left\{G\left(\sup_{f\in\mathcal{F}%
}\sum_{t=1}^{n-1}u_{t}\phi_{t}\left(f\left(\mathbf{z}_t(u)\right)
\right)+u_{n}f\left(\mathbf{z}_n(u)\right)\right)\right\} \leq \mathbb{E}%
_{u_{n-1}}\left\{G\left(\sup_{f\in\mathcal{F}}\sum_{t=1}^{n-2}
u_{t}\phi_{t}\left(f\left(\mathbf{z}_t(u)\right) \right) +\sum_{t=n-1}^n
u_{t}f\left(\mathbf{z}_{t}(u)\right)\right)\right\}.
\end{align*}
Taking expectation on both sides, we have 
\begin{align*}
&\mathbb{E}\left\{G\left(\sup_{f\in\mathcal{F}}\sum_{t=1}^{n-1}u_{t}\phi_{t}%
\left(f\left(\mathbf{z}_{t}(u)\right) \right)+u_{n}f\left(\mathbf{z}%
_{n}(u)\right)\right)\right\} \leq \mathbb{E}\left\{G\left(\sup_{f\in%
\mathcal{F}}\sum_{t=1}^{n-2} u_{t}\phi_{t}\left(f\left(\mathbf{z}%
_{t}(u)\right) \right) +\sum_{t=n-1}^{n} u_{t}f\left(\mathbf{z}%
_{t}(u)\right)\right)\right\}.
\end{align*}

We can repeat a similar argument by taking conditional expectations given $%
(u_{1:t-1},u_{t+1:n})$ and removing $\phi_t$ for all $t=n-2,\cdots,2$. This
leads to the result that 
\begin{align}
&\mathbb{E}\left\{G\left(\sup_{f\in\mathcal{F}}\sum_{t=1}^{n}u_{t}\phi_{t}%
\left(f\left(\mathbf{z}_t(u)\right) \right)\right)\right\} \leq\mathbb{E}%
\left\{G\left(\sup_{f\in\mathcal{F}}\sum_{t=1}^n u_{t}f\left(\mathbf{z}%
_{t}(u)\right)\right)\right\}.  \label{eq:no_abs}
\end{align}

Next, we come back to the case with the absolute value. Note that 
\begin{align}  \label{Lem9.6}
&\mathbb{E}F\left[\frac{1}{2}\sup_{f\in\mathcal{F}}\left\vert
\sum_{t=1}^{n}u_{t}\phi_{t}\left(f\left(\mathbf{z}_t(u)\right) \right)
\right\vert\right]  \notag \\
&\leq\frac{1}{2}\left\{\mathbb{E}F\left[\sup_{f\in\mathcal{F}%
}\left(\sum_{t=1}^{n} u_{t}\phi_{t}\left(f\left(\mathbf{z}_t(u)\right)
\right)\right)^{+}\right]+\mathbb{E}F\left[\sup_{f\in\mathcal{F}%
}\left(\sum_{t=1}^{n} u_{t}\phi_{t}\left(f\left(\mathbf{z}_t(u)\right)
\right)\right)^{-}\right]\right\}  \notag \\
&=\frac{1}{2}\left\{\mathbb{E}F\left[\sup_{f\in\mathcal{F}%
}\left(\sum_{t=1}^{n} u_{t}\phi_{t}\left(f\left(\mathbf{z}_t(u)\right)
\right)\right)^{+}\right]+\mathbb{E}F\left[\sup_{f\in\mathcal{F}%
}\left(\sum_{t=1}^{n} u_{t}^{\ast}\phi_{t}\left(f\left(\mathbf{z}%
^{\ast}_t(u^{\ast})\right) \right) \right)^{+}\right]\right\}  \notag \\
&=\frac{1}{2}\left\{\mathbb{E}F\left[\left(\sup_{f\in\mathcal{F}%
}\sum_{t=1}^{n} u_{t}\phi_{t}\left(f\left(\mathbf{z}_t(u)\right)
\right)\right)^{+}\right]+\mathbb{E}F\left[\left(\sup_{f\in\mathcal{F}%
}\sum_{t=1}^{n} u_{t}^{\ast}\phi_{t}\left(f\left(\mathbf{z}%
^{\ast}_t(u^{\ast})\right) \right) \right)^{+}\right]\right\}  \notag \\
&\leq \frac{1}{2}\left\{\mathbb{E}F\left[\left(\sup_{f\in\mathcal{F}%
}\sum_{t=1}^{n} u_{t}\left(f\left(\mathbf{z}_t(u)\right) \right) \right)^{+}%
\right]+\mathbb{E}F\left[\left(\sup_{f\in\mathcal{F}}\sum_{t=1}^{n}u_{t}^{%
\ast}f\left(\mathbf{z}^{\ast}_t(u^{\ast})\right) \right)^{+}\right]\right\} 
\notag \\
&= \frac{1}{2}\left\{\mathbb{E}F\left[\sup_{f\in\mathcal{F}%
}\left(\sum_{t=1}^{n} u_{t}\left(f\left(\mathbf{z}_t(u)\right)
\right)\right)^{+}\right]+\mathbb{E}F\left[\sup_{f\in\mathcal{F}%
}\left(\sum_{t=1}^{n} u_{t}^{\ast}f\left(\mathbf{z}^{\ast}_t(u^{\ast})%
\right)\right)^{+}\right]\right\}  \notag \\
& = \frac{1}{2}\left\{\mathbb{E}F\left[\sup_{f\in\mathcal{F}%
}\left(\sum_{t=1}^{n} u_{t}\left(f\left(\mathbf{z}_t(u)\right) \right)
\right)^{+}\right]+\mathbb{E}F\left[\sup_{f\in\mathcal{F}}\left(%
\sum_{t=1}^{n}u_{t}f\left(\mathbf{z}_t(u)\right) \right)^{-}\right]\right\} 
\notag \\
& \leq \mathbb{E}F\left[\sup_{f\in\mathcal{F}}\left\vert
\sum_{t=1}^{n}u_{t}f\left(\mathbf{z}_t(u)\right) \right\vert\right],
\end{align}
where $u_{t}^{\ast}=-u_{t}$, $\mathbf{z}_{t}^{\ast}(u) = \mathbf{z}_t(-u)$,
the first inequality is by the convexity of $F$, the first and second
equalities are by the fact that $(v)^{-} =(-v)^{+}$ for any $v$, and the
second inequality is by \eqref{eq:no_abs} and the fact that $F((\cdot)^+)$
is convex and increasing. This leads to the desired result.
\end{proof}

\bigskip

Define Rademacher sequence $u=\left(u_{11},\cdots
,u_{1T},\cdots,u_{N1},\cdots ,u_{NT}\right) =\left(u_{\left(1\right)
},\cdots ,u_{\left(NT\right) }\right) \in \mathbb{R}^{NT\times 1}$. In the
matrix notation, let $U=\left\{ u_{it}\right\} \in \mathbb{R}^{N\times T}$.
By vectorization, for a sequence of independent variables $X_{j,it}$, we
define 
\begin{align*}
& \left(X_{j,11},\cdots ,X_{j,1T},\cdots ,X_{j,N1},\cdots
,X_{j,NT}\right)=\left(X_{j,\left(1\right) },\cdots ,X_{j,\left(NT\right)
}\right) , \\
& \left(\epsilon_{11},\cdots ,\epsilon_{1T},\cdots
,\epsilon_{N1},\cdots,\epsilon_{NT}\right) =\left(\epsilon_{\left(1\right)
},\cdots ,\epsilon_{\left(NT\right) }\right) ,
\end{align*}%
$\forall j\in [p]$. Using the binary tree representation, let $%
x_{j,\left(l\right) }^{\ast }$ be $(l)^{th}$ root of the tree which takes
values in the support of $X_{j,\left(l\right) }$, i.e., $x_{j,\left(l\right)
}^{\ast }:u^{l-1}\mapsto [-\xi_{N},\xi_{N}]$ for $u^{l-1}:=\left(u_{\left(1%
\right)},\cdots,u_{\left(l-1\right) }\right) $ such that $\max_{i\in
[N]}\sum_{t\in[T]}\left(x_{j,it}^{\ast}\right) ^{2}\leq MT$ and $\max_{t\in[T%
]}\sum_{i\in [N]}\left(x_{j,it}^{\ast }\right) ^{2\ell }\leq MN$ for some
fixed constant $M<\infty $ and $\ell =1,2$. Similar notation follows for $%
\epsilon_{\left(l\right) }^{\ast }=\epsilon_{\left(l\right) }^{\ast
}\left(u^{l-1}\right)$. In the matrix notation, let $x_{j}^{\ast }=\left\{
x_{j,it}^{\ast}\right\} \in \mathbb{R}^{N\times T}$ such that $%
x_{j,it}^{\ast}=x_{j,\left(l\right) }$ with $i=\lceil \frac{l}{T}\rceil $
and $t=l-\left(i-1\right) T$.

\begin{lemma}
{\small \label{Lem:exp tree} } Under Assumption \ref{ass:1}, for $j\in[p]$,
there exists an absolute constant $C$ that is independent of the trees $%
(x^{\ast },\epsilon ^{\ast })$ such that when $\log (N\vee T)\geq 2$, 
\begin{equation*}
\mathbb{E}\exp \left(\frac{\left\Vert U\right\Vert_{op}}{\sqrt{(N\vee T)\log
(N\vee T)}}\right) \leq C\quad \text{and}\quad \mathbb{E}\exp \left(\frac{%
\left\Vert U\odot x_{j}^{\ast }\right\Vert_{op}}{\sqrt{(N\vee T)\log(N\vee T)%
}}\right) \leq C.
\end{equation*}
\end{lemma}

\begin{proof}
The proof here follows similarly as Lemma \ref{Lem:matrix op} except that we
have martingale difference matrices rather than independent matrices. For a
specific $j$, let $A=U\odot x_{j}^{\ast }=\left(A_{1},\cdots
,A_{N}\right)^{\prime}\in \mathbb{R}^{N\times T}$, $\mathscr{F}_{i}$ be the $%
\sigma$-field generated by $\{u_{i^{*}t}\}_{i^{*}\leq i,t\in[T]}$, $\mathbb{E%
}_{i}(\cdot)=\mathbb{E}(\cdot |\mathscr{F}_{i})$, $\Sigma_{i}=\mathbb{E}%
_{i-1}\left(A_{i}A_{i}^{\prime}\right) $ and $Z_{i}=\frac{1}{N}%
\left(A_{i}A_{i}^{\prime}-\Sigma_{i}\right) $, with $\frac{1}{N}%
\left(A^{\prime}A-\sum_{i\in [N]}\Sigma_{i}\right) =\sum_{i\in [N]}Z_{i}$.
Note that $\mathbb{E}_{i-1}\left(Z_{i}\right) =0,$ 
\begin{equation}
\max_{i\in [N]}\left\Vert A_{i}\right\Vert_{2}=\max_{i\in [N]}\sqrt{\sum_{t%
\in[T]}\left(x_{j,it}^{\ast }u_{it}\right) ^{2}}=\max_{i\in [N]}\sqrt{\sum_{t%
\in[T]}\left(x_{j,it}^{\ast}\right) ^{2}}\leq \sqrt{MT}~a.s.,
\label{Lem10.1}
\end{equation}
\begin{equation}
\max_{i\in [N]}\left\Vert \Sigma_{i}\right\Vert_{op}=\max_{i\in[N]%
}\left\Vert \text{diag}\left(\left(x_{j,i1}^{\ast }\right)^{2},\cdots
,\left(x_{j,iT}^{\ast }\right) ^{2}\right) \right\Vert_{op}\leq
\xi_{N}^{2}~a.s.,  \label{Lem10.2}
\end{equation}
and for $\ell =1,2,$ 
\begin{equation}
\left\Vert \sum_{i\in [N]}\Sigma_{i}^{\ell }\right\Vert_{op}=\left\Vert 
\text{diag}\left(\sum_{i\in [N]}\left(x_{j,i1}^{\ast }\right) ^{2\ell
},\cdots ,\sum_{i\in [N]}\left(x_{j,iT}^{\ast }\right) ^{2\ell }\right)
\right\Vert_{op}\leq MN~a.s..  \label{Lem10.2a}
\end{equation}

Combining (\ref{Lem10.1}) and (\ref{Lem10.2}) yields 
\begin{eqnarray}
\max_{i\in [N]}\left\Vert Z_{i}\right\Vert_{op} &\leq &\max_{i\in[N]}\frac{1%
}{N}\left(\left\Vert A_{i}A_{i}^{\prime}\right\Vert_{op}+\left\Vert
\Sigma_{i}\right\Vert_{op}\right)  \notag \\
&\leq &\frac{1}{N}\left(\max_{i\in [N]}\left\Vert
A_{i}\right\Vert_{2}^{2}+\max_{i\in [N]}\left\Vert
\Sigma_{i}\right\Vert_{op}\right) \leq \frac{MT+\xi_{N}^{2}}{N}\,\,a.s.
\label{Lem10.3}
\end{eqnarray}%
In addition,%
\begin{align}
\left\Vert \sum_{i\in [N]}\mathbb{E}_{i-1}\left(Z_{i}^{2}\right)\right%
\Vert_{op}& =\left\Vert \sum_{i\in [N]}\mathbb{E}_{i-1}\left\{ \frac{1}{N^{2}%
}\left[ \left(A_{i}A_{i}^{\prime}\right) ^{2}-\Sigma_{i}^{2}\right] \right\}
\right\Vert_{op}  \notag \\
& \leq \frac{1}{N^{2}}\left[ \sum_{i\in [N]}\left\Vert \mathbb{E}%
_{i-1}\left(\left\Vert
A_{i}\right\Vert_{2}^{2}A_{i}A_{i}^{\prime}\right)\right\Vert_{op}+\left%
\Vert \sum_{i\in [N]}\Sigma_{i}^{2}\right\Vert_{op}\right]  \notag \\
& \leq \frac{1}{N^{2}}\left[ MT\left\Vert \sum_{i\in
[N]}\Sigma_{i}\right\Vert_{op}+\left\Vert \sum_{i\in
[N]}\Sigma_{i}^{2}\right\Vert_{op}\right]  \notag \\
& \leq \frac{\left(MT+1\right) M}{N},  \label{Lem10.4}
\end{align}%
where the last inequality holds by (\ref{Lem10.1}) and (\ref{Lem10.2a}).

Combining (\ref{Lem10.3}) and (\ref{Lem10.4}), by matrix Bernstein's
inequality in Lemma \ref{Lem:matrix Bern}(ii), for some sufficiently large
constant $\bar{c}$ that depends on $M$, we have 
\begin{align*}
& \mathbb{P}\left(\left\Vert \sum_{i\in [N]}Z_{i}\right\Vert_{op}>\bar{c}%
\frac{(N\vee T)}{N}\log (N\vee T)\right) \\
& \leq T\exp \left\{ -\frac{\frac{1}{2}\bar{c}^{2}\left(\frac{(N\vee T)}{N}%
\right) ^{2}\left(\log (N\vee T)\right) ^{2}}{\frac{\left(MT+1\right) M}{N}+%
\frac{\frac{MT+\xi_{N}^{2}}{N}\bar{c}\frac{(N\vee T)}{N}\log (N\vee T)}{3}}%
\right\} =\exp \left(-\left(\frac{\bar{c}}{2}-1\right) \log ((N\vee
T))\right) ,
\end{align*}
which implies that with probability greater than $1-\exp (-(\bar{c}/2-1)\log
(N\vee T))$, 
\begin{align*}
& \left\Vert \frac{1}{N}\left(A^{\prime}A-\sum_{i\in [N]}\Sigma_{i}\right)
\right\Vert_{op}\leq \bar{c}\left(\frac{(N\vee T)}{N}\log(N\vee T)\right) ,
\\
& \left\Vert \frac{1}{N}A^{\prime}A\right\Vert_{op}\leq \left\Vert \frac{1}{N%
}\sum_{i\in [N]}\Sigma_{i}\right\Vert_{op}+\left\Vert \frac{1}{N}%
\left(A^{\prime}A-\sum_{i\in [N]}\Sigma_{i}\right) \right\Vert_{op}\leq M+%
\bar{c}\left(\frac{(N\vee T)}{N}\log (N\vee T)\right) ,\quad \text{and} \\
& \left\Vert A\right\Vert_{op}=\left\Vert U\odot x_{j}^{\ast
}\right\Vert_{op}\leq \sqrt{(1+\bar{c})(N\vee T)\log (N\vee T)}.
\end{align*}%
%
%
%
%
%
%
%
%
%
%

Consequently, when $\log (N\vee T)\geq 2$, we have 
\begin{align*}
& \mathbb{E}\exp \left(\frac{\left\Vert U\odot x_{j}^{\ast }\right\Vert_{op}%
}{\sqrt{(N\vee T)\log (N\vee T)}}\right) \\
& =\int_{0}^{\infty }\exp (u)\mathbb{P}\left(\frac{\left\Vert U\odot
x_{j}^{\ast }\right\Vert_{op}}{\sqrt{(N\vee T)\log (N\vee T)}}\geq u\right)du
\\
& =\left(\int_{0}^{2}+\int_{2}^{\infty }\right) \exp (u)\mathbb{P}%
\left(\left\Vert U\odot x_{j}^{\ast }\right\Vert_{op}\geq u\sqrt{(N\vee
T)\log(N\vee T)}\right) du \\
& \leq \int_{0}^{2}\exp (u)du+\int_{2}^{\infty }\exp (u)\exp \left(-\frac{%
(u^{2}-3)}{2}\log (N\vee T)\right) du \\
& \leq \int_{0}^{2}\exp (u)du+\int_{2}^{\infty }\exp \left(-\left(u-\frac{1}{%
2}\right) ^{2}+\frac{13}{4}\right) du \\
& \leq \exp (2)-1+\sqrt{2\pi }\exp \left(\frac{13}{4}\right) :=C,
\end{align*}%
where the first inequality is by the fact that 
\begin{equation*}
\mathbb{P}\left(\frac{\left\Vert U\odot x_{j}^{\ast }\right\Vert_{op}}{\sqrt{%
(N\vee T)\log (N\vee T)}}\geq \sqrt{1+\bar{c}}\right) \leq \exp \left(-\left(%
\frac{\bar{c}}{2}-1\right) \log ((N\vee T))\right)
\end{equation*}%
and by letting $u=\sqrt{1+\bar{c}}\geq 2$ for large $\bar{c}$. Similarly, we
can show that when $\log (N\vee T)\geq 2$, $\mathbb{E}\exp \left(\frac{%
\left\Vert U\right\Vert_{op}}{\sqrt{(N\vee T)\log (N\vee T)}}\right) \leq C$
for some absolute constant $C$.
\end{proof}

Recall that $\tilde{\rho}_{it}\left(\left\{\Delta_{\Theta_{j},it},X_{j,it}%
\right\}_{j=0}^{p},\epsilon_{it}\right)$ is defined in \eqref{eq:rhotilde}
and $\mathscr{G}_{i,t}$ is defined in Assumption \ref{ass:1}.

\begin{lemma}
{\small \label{Lem:empirical process} } If Assumptions \ref{ass:1}-\ref%
{ass:4} hold, then we have 
\begin{equation*}
\sup_{\left\{ \Delta_{\Theta_{j}}\right\}_{j=0}^{p}\in \mathcal{R}%
\left(3,C_{2}\right) }\frac{\left\vert \frac{1}{NT}\sum_{i=1}^{N}%
\sum_{t=1}^{T}\tilde{\rho}_{it}\left(\left\{
\Delta_{\Theta_{j},it},X_{j,it}\right\}_{j=0}^{p},\epsilon_{it}\right)
\right\vert }{\sum_{j=0}^{p}\left\Vert\Delta_{\Theta_{j}}\right\Vert_{F}}%
=O_{p}\left(a_{NT}\right) ,
\end{equation*}%
where $a_{NT}=\frac{\sqrt{\left(N\vee T\right) \log \left(N\vee T\right) }}{%
NT}$.
\end{lemma}

\begin{proof}
Let $n=NT$ and for $l\in [n]$, $Z_{l}=(\{X_{j,it}\}_{j\in
[p]},\epsilon_{it}) $ with $i=\lceil \frac{l}{T}\rceil $ and $%
t=l-\left(i-1\right) T$, and $\mathscr{F}_{l}=\mathscr{G}_{it}$. Then, $%
Z_{l} $ is adapted to the filtration $\{\mathscr{F}_{l}\}_{l\in [n]}$. Lemma %
\ref{Lem:symmetrization} implies 
\begin{align}
& \beta_{NT}\mathbb{P}\left\{ \sup_{\left\{
\Delta_{\Theta_{j}}\right\}_{j=0}^{p}\in \mathcal{R}\left(3,C_{2}\right)
}\left\vert \frac{1}{NT}\sum_{i=1}^{N}\sum_{t=1}^{T}\frac{\tilde{\rho}%
_{it}\left(\left\{
\Delta_{\Theta_{j},it},X_{j,it}\right\}_{j=0}^{p},\epsilon_{it}\right) }{%
\sum_{j=0}^{p}\left\Vert \Delta_{\Theta_{j}}\right\Vert_{F}}%
\right\vert>C_{5}a_{NT}\right\}  \notag  \label{Lem11.1} \\
\leq & 2\sup_{x^{\ast },\epsilon ^{\ast }}\mathbb{P}\left\{
\sup_{\left\{\Delta_{\Theta_{j}}\right\}_{j=0}^{p}\in \mathcal{R}%
\left(3,C_{2}\right)}\left\vert \frac{1}{NT}\sum_{i=1}^{N}\sum_{t=1}^{T}%
\frac{u_{it}\rho_{it}\left(\left\{ \Delta_{\Theta_{j},it},x_{j,it}^{\ast
}\right\}_{j=0}^{p},\epsilon_{it}^{\ast }\right) }{\sum_{j=0}^{p}\left\Vert
\Delta_{\Theta_{j}}\right\Vert_{F}}\right\vert >\frac{C_{5}a_{NT}}{4}%
\right\} ,
\end{align}%
for some positive constant $C_{5}$, where the outer supremum on the RHS of
the above display is taken over $[-\xi_{N},\xi_{N}]^{p}\times \mathbb{R}$%
-valued trees with depth $n$ and 
\begin{equation*}
\beta_{NT}=1-\sup_{\left\{ \Delta_{\Theta_{j}}\right\}_{j=0}^{p}\in\mathcal{R%
}\left(3,C_{2}\right) }\frac{4}{C_{5}^{2}\left(NT\right)^{2}a_{NT}^{2}}%
\sum_{(i,t)\in [N]\times [T]}Var\left(\frac{\tilde{\rho}_{it}\left(\left\{
\Delta_{\Theta_{j},it},X_{j,it}\right\}_{j=0}^{p},\epsilon_{it}\right) }{%
\sum_{j=0}^{p}\left\Vert \Delta_{\Theta_{j}}\right\Vert_{F}}\bigg |%
\mathscr{G}_{i,t-1}\right) .
\end{equation*}%
We first note 
\begin{align*}
&\sum_{(i,t)\in [N]\times [T]}Var\left(\frac{\tilde{\rho}_{it}\left(\left\{
\Delta_{\Theta_{j},it},X_{j,it}\right\}_{j=0}^{p},\epsilon_{it}\right) }{%
\sum_{j=0}^{p}\left\Vert \Delta_{\Theta_{j}}\right\Vert_{F}}\bigg |%
\mathscr{G}_{i,t-1}\right) \\
& \leq\sum_{(i,t)\in [N]\times [T]}\mathbb{E}\left[ \left(\frac{%
\rho_{it}\left(\left\{
\Delta_{\Theta_{j},it},X_{j,it}\right\}_{j=0}^{p},\epsilon_{it}\right) }{%
\sum_{j=0}^{p}\left\Vert \Delta_{\Theta_{j}}\right\Vert_{F}}\right) ^{2}%
\bigg |\mathscr{G}_{i,t-1}\right] \\
& \leq 2\sum_{(i,t)\in [N]\times [T]}\left(\frac{\Delta_{\Theta_{0},it}+%
\sum_{j=1}^{p}X_{j,it}\Theta_{j,it}}{\sum_{j=0}^{p}\left\Vert
\Delta_{\Theta_{j}}\right\Vert_{F}}\right) ^{2} \\
& \lesssim \sum_{i=1}^{N}\sum_{t=1}^{T}\frac{\sum_{j=1}^{p}\left(X_{j,it}%
\Delta_{\Theta_{j},it}\right) ^{2}+\Delta_{\Theta_{0},it}^{2}}{%
\sum_{j=0}^{p}\left\Vert \Delta_{\Theta_{j}}\right\Vert_{F}^{2}}\leq
c_{10}\xi_{N}^{2}
\end{align*}%
with some positive constant $c_{10}$, where the first inequality holds by
Jensen inequality, the second inequality is by $\left\vert
\rho_{\tau}(u)-\rho_{\tau }(v)\right\vert \leq 2\left\vert u-v\right\vert $,
and the last line holds by Assumption \ref{ass:1}(iv) and the fact that $%
\left(\sum_{j=0}^{p}\left\Vert
\Delta_{\Theta_{j}}\right\Vert_{F}\right)^{2}=O\left(\sum_{j=0}^{p}\left%
\Vert \Delta_{\Theta_{j}}\right\Vert_{F}^{2}\right) $. Therefore, we have 
\begin{align}
\beta_{NT}& \geq 1-\sup_{\left\{ \Delta_{\Theta_{j}}\right\}_{j=0}^{p}\in 
\mathcal{R}\left(3,C_{2}\right) }\frac{4}{C_{5}^{2}\left(NT\right)
^{2}a_{NT}^{2}}\sum_{(i,t)\in [N]\times [T]}Var\left(\frac{\tilde{\rho}%
_{it}\left(\left\{
\Delta_{\Theta_{j},it},X_{j,it}\right\}_{j=0}^{p},\epsilon_{it}\right) }{%
\sum_{j=0}^{p}\left\Vert \Delta_{\Theta_{j}}\right\Vert_{F}}\right)  \notag
\\
& \geq 1-\sup_{\left\{ \Delta_{\Theta_{j}}\right\}_{j=0}^{p}\in \mathcal{R}%
\left(3,C_{2}\right) }\frac{4c_{10}\xi_{N}^{2}}{C_{5}^{2}\left(N\vee
T\right) \log \left(N\vee T\right) }  \notag \\
& \geq 1-O\left(\frac{\xi_{N}^{2}}{\left(N\vee T\right) \log (N\vee T)}%
\right) \rightarrow 1,
\end{align}
where the last line is by Assumption \ref{ass:1}(ix).

Define 
\begin{align*}
& \mathcal{A}_{0}=\sup_{\left\{ \Delta_{\Theta_{j}}\right\}_{j=0}^{p}\in%
\mathcal{R}\left(3,C_{2}\right) }\left\vert \frac{1}{NT}\sum_{i=1}^{N}%
\sum_{t=1}^{T}\frac{u_{it}\Delta_{\Theta_{0},it}}{\sum_{j=0}^{p}\left\Vert%
\Delta_{\Theta_{j}}\right\Vert_{F}}\right\vert , \\
& \mathcal{A}_{j}=\sup_{\left\{ \Delta_{\Theta_{j}}\right\}_{j=0}^{p}\in%
\mathcal{R}\left(3,C_{2}\right) }\left\vert \frac{1}{NT}\sum_{i=1}^{N}%
\sum_{t=1}^{T}\frac{u_{it}x_{j,it}^{\ast }\Delta_{\Theta_{j},it}}{%
\sum_{j=0}^{p}\left\Vert \Delta_{\Theta_{j}}\right\Vert_{F}}\right\vert
,\quad\forall j\in [p], \\
& \mathcal{A}_{p+1}=\sup_{\left\{ \Delta_{\Theta_{j}}\right\}_{j=0}^{p}\in 
\mathcal{R}\left(3,C_{2}\right) }\left\vert \frac{\frac{1}{NT}%
\sum_{i=1}^{N}\sum_{t=1}^{T}u_{it}\phi_{it}\left(\Delta_{\Theta_{0},it}+%
\sum_{j=1}^{p}x_{j,it}^{\ast }\Delta_{\Theta_{j},it}\right) }{%
\sum_{j=0}^{p}\left\Vert \Delta_{\Theta_{j}}\right\Vert_{F}}\right\vert ,
\end{align*}%
where $\phi_{it}\left(u\right) =\left(\epsilon_{it}^{\ast
}-u\right)^{-}-\left(\epsilon_{it}^{\ast }\right) ^{-}$. Notice that 
\begin{equation*}
\rho_{it}\left(\left\{ \Delta_{\Theta_{j},it},x_{j,it}^{\ast
}\right\}_{j=0}^{p},\epsilon_{it}^{\ast }\right) =\tau
\left(\Delta_{\Theta_{0},it}+\sum_{j=1}^{p}x_{j,it}^{\ast
}\Delta_{\Theta_{j},it}\right)
+\phi_{it}\left(\Delta_{\Theta_{0},it}+\sum_{j=1}^{p}x_{j,it}^{\ast
}\Delta_{\Theta_{j},it}\right) ,
\end{equation*}%
we obtain that 
\begin{align}
& \sup_{x^{\ast },\epsilon ^{\ast }}\mathbb{P}\left\{
\sup_{\left\{\Delta_{\Theta_{j}}\right\}_{j=0}^{p}\in \mathcal{R}%
\left(3,C_{2}\right)}\left\vert \frac{1}{NT}\sum_{i=1}^{N}\sum_{t=1}^{T}%
\frac{u_{it}\rho_{it}\left(\left\{ \Delta_{\Theta_{j},it},x_{j,it}^{\ast
}\right\}_{j=0}^{p},\epsilon_{it}^{\ast }\right) }{\sum_{j=0}^{p}\left\Vert
\Delta_{\Theta_{j}}\right\Vert_{F}}\right\vert >\frac{C_{5}a_{NT}}{4}\right\}
\notag  \label{Lem11.2} \\
\leq & \sum_{j=0}^{p}\sup_{x^{\ast }}\mathbb{P}\left\{ \tau \mathcal{A}_{j}>%
\frac{C_{5}a_{NT}}{4\left(p+2\right) }\right\} +\sup_{x^{\ast
},\epsilon^{\ast }}\mathbb{P}\left\{ \mathcal{A}_{p+1}>\frac{C_{5}a_{NT}}{%
4\left(p+2\right) }\right\} .
\end{align}

We first bound $\mathcal{A}_{j}$ for $j\in [p]$. We have 
\begin{align}
\mathcal{A}_{j}& =\frac{1}{NT}\sup_{\left\{
\Delta_{\Theta_{j}}\right\}_{j=0}^{p}\in \mathcal{R}\left(3,C_{2}\right)
}\left\vert\sum_{i=1}^{N}\sum_{t=1}^{T}\frac{u_{it}x_{j,it}^{\ast
}\Delta_{\Theta_{j},it}}{\sum_{j=0}^{p}\left\Vert
\Delta_{\Theta_{j}}\right\Vert_{F}}\right\vert \leq \frac{1}{NT}%
\sup_{\left\{ \Delta_{\Theta_{j}}\right\}_{j=0}^{p}\in \mathcal{R}%
\left(3,C_{2}\right) }\frac{\left\vert tr\left[\Delta_{\Theta_{j}}^{\prime}%
\left(U\odot x_{j}^{\ast }\right) \right]\right\vert }{\sum_{j=0}^{p}\left%
\Vert \Delta_{\Theta_{j}}\right\Vert_{F}}  \notag \\
& \leq \frac{1}{NT}\left\Vert U\odot x_{j}^{\ast
}\right\Vert_{op}\sup_{\left\{ \Delta_{\Theta_{j}}\right\}_{j=0}^{p}\in 
\mathcal{R}\left(3,C_{2}\right) }\frac{\left\Vert
\Delta_{\Theta_{j}}\right\Vert_{\ast }}{\sum_{j=0}^{p}\left\Vert
\Delta_{\Theta_{j}}\right\Vert_{F}}\leq \frac{c_{8}}{NT}\left\Vert U\odot
x_{j}^{\ast }\right\Vert_{op},  \label{Lem11.4}
\end{align}%
where the first inequality holds by $tr(AB)\leq \left\Vert
A\right\Vert_{op}\left\Vert B\right\Vert_{\ast }$ and the second inequality
holds by Lemma \ref{Lem:delta nuclear}. Then 
\begin{align}
\sup_{x^{\ast },\epsilon ^{\ast }}\mathbb{P}\left\{ \tau \mathcal{A}_{j}>%
\frac{C_{5}a_{NT}}{4\left(p+2\right) }\right\} & \leq \sup_{x^{\ast }}%
\mathbb{P}\left\{ \frac{\left\Vert U\odot x_{j}^{\ast }\right\Vert }{\sqrt{%
\left(N\vee T\right) \log \left(N\vee T\right) }}>\frac{C_{5}}{%
4c_{8}\tau\left(p+2\right) }\right\}  \notag \\
& \leq \sup_{x^{\ast }}\left\{ \exp \left(-\frac{C_{5}}{4c_{8}\tau
\left(p+2\right) }\right) \mathbb{E}\left[ \exp \left(\frac{\left\Vert
U\odot x_{j}^{\ast }\right\Vert }{\sqrt{\left(N\vee T\right) \log
\left(N\vee T\right) }}\right) \right] \right\}  \notag \\
& \leq C\exp \left(-\frac{C_{5}}{4c_{8}\tau \left(p+2\right) }\right)
\label{eq:Aj}
\end{align}%
for some absolute constant $C$ that independent of $(x^{\ast
},\epsilon^{\ast })$, where the last inequality holds by Lemma \ref{Lem:exp
tree}. Similarly, we can establish 
\begin{equation}
\sup_{x^{\ast },\epsilon ^{\ast }}\mathbb{P}\left\{ \tau \mathcal{A}_{0}>%
\frac{C_{5}a_{NT}}{4\left(p+2\right) }\right\} \leq C\exp \left(-\frac{C_{5}%
}{4c_{8}\tau \left(p+2\right) }\right) .  \label{eq:A0}
\end{equation}




Next, we turn to $\mathcal{A}_{p+1}$. We have 
\begin{equation}
\sup_{x^{\ast },\epsilon ^{\ast }}\mathbb{P}\left\{ \mathcal{A}_{p+1}>\frac{%
C_{5}a_{NT}}{4\left(p+2\right) }\right\} \leq \exp \left(-\frac{C_{5}}{%
4c_{8}(p+1)\left(p+2\right) }\right) \sup_{x^{\ast },\epsilon ^{\ast }}%
\mathbb{E}\left\{ \exp \left(\frac{NT\mathcal{A}_{p+1}}{c_{8}(p+1)\sqrt{%
\left(N\vee T\right) \log \left(N\vee T\right) }}\right) \right\} .
\label{Lem11.5}
\end{equation}%
Because $\phi_{it}(\cdot)$ is a contraction, Lemma \ref{Lem:contraction}
implies 
\begin{align}
& \mathbb{E}\left\{ \exp \left(\frac{NT\mathcal{A}_{p+1}}{c_{8}(p+1)\sqrt{%
\left(N\vee T\right) \left(\log (N\vee T)\right) }}\right) \right\}  \notag
\label{Lem11.6} \\
& =\mathbb{E}\left\{ \exp \left[ \frac{1}{c_{8}(p+1)\sqrt{\left(N\vee
T\right) \left(\log (N\vee T)\right) }}\sup_{\left\{
\Delta_{\Theta_{j}}\right\}_{j=0}^{p}\in \mathcal{R}\left(3,C_{2}\right)
}\left\vert\frac{\sum_{i=1}^{N}\sum_{t=1}^{T}u_{it}\phi_{it}\left(\Delta_{%
\Theta_{0},it}+\sum_{j=1}^{p}x_{j,it}^{\ast }\Delta_{\Theta_{j},it}\right) }{%
\sum_{j=0}^{p}\left\Vert \Delta_{\Theta_{j}}\right\Vert_{F}}\right\vert %
\right] \right\}  \notag \\
& \leq \mathbb{E}\left\{ \exp \left[ \frac{1}{c_{8}(p+1)\sqrt{\left(N\vee
T\right) \left(\log (N\vee T)\right) }}\sup_{\left\{
\Delta_{\Theta_{j}}\right\}_{j=0}^{p}\in \mathcal{R}\left(3,C_{2}\right)
}\left\vert\sum_{i=1}^{N}\sum_{t=1}^{T}\frac{u_{it}(\Delta_{\Theta_{0},it}+%
\sum_{j=1}^{p}x_{j,it}^{\ast }\Delta_{\Theta_{j},it})}{\sum_{j=0}^{p}\left%
\Vert \Delta_{\Theta_{j}}\right\Vert_{F}}\right\vert \right] \right\}  \notag
\\
& \leq \mathbb{E}\left\{ \exp \left[ \frac{\left(\left\Vert
U\right\Vert_{op}+\sum_{j\in [p]}\left\Vert U\odot x_{j}^{\ast
}\right\Vert_{op}\right) }{(p+1)\sqrt{\left(N\vee T\right) \left(\log (N\vee
T)\right) }}\right] \right\}  \notag \\
& \leq \mathbb{E}\left\{ \exp \left[ \frac{\left(\left\Vert
U\right\Vert_{op}\right) }{\sqrt{\left(N\vee T\right) \left(\log (N\vee
T)\right) }}\right] \right\} ^{1/(1+p)}\Pi_{j\in [p]}\left[ \mathbb{E}%
\left\{\exp \left[ \frac{\left(\left\Vert U\odot x_{j}^{\ast
}\right\Vert_{op}\right) }{\sqrt{\left(N\vee T\right) \left(\log (N\vee
T)\right) }}\right] \right\} ^{1/(1+p)}\right] \leq C
\end{align}%
for some absolute constant $C$, where the first inequality is by Lemma \ref%
{Lem:contraction}, the second inequality is by (\ref{Lem11.4}), the third
inequality is due to the fact that, for random variables $\{A_{i}\}_{i\in[p+1%
]}$, 
\begin{equation*}
\mathbb{E}(\Pi_{i\in [p+1]}|A_{i}|)\leq \Pi_{i\in [p+1]}[\mathbb{E}%
|A_{i}^{1+p}|]^{1/(1+p)},
\end{equation*}%
and the final inequality is by Lemma \ref{Lem:exp tree} with an absolute
constant $C$ that is independent of $(x^{\ast },\epsilon ^{\ast })$.

Combining (\ref{Lem11.5}) and (\ref{Lem11.6}), we have 
\begin{equation*}
\sup_{x^{\ast },\epsilon ^{\ast }}\mathbb{P}\left\{ \mathcal{A}_{p+1}>\frac{%
C_{5}a_{NT}}{4\left(p+2\right) }\right\} \leq C\exp \left(-\frac{C_{5}}{%
4c_{8}\left(p+1\right) (P+2)}\right) ,
\end{equation*}%
which, combined with \eqref{Lem11.1}, \eqref{Lem11.2}, \eqref{eq:Aj}, and %
\eqref{eq:A0}, further implies that 
\begin{align*}
& \mathbb{P}\left\{ \sup_{\left\{ \Delta_{\Theta_{j}}\right\}_{j=0}^{p}\in 
\mathcal{R}\left(3,C_{2}\right) }\left\vert \frac{1}{NT}\sum_{i=1}^{N}%
\sum_{t=1}^{T}\frac{\tilde{\rho}_{it}\left(\left\{
\Delta_{\Theta_{j},it},X_{j,it}\right\}_{j=0}^{p},\epsilon_{it}\right) }{%
\sum_{j=0}^{p}\left\Vert \Delta_{\Theta_{j}}\right\Vert_{F}}%
\right\vert>C_{5}a_{NT}\right\} \\
& \leq 2\beta_{NT}^{-1}\sup_{x^{\ast },\epsilon ^{\ast }}\mathbb{P}%
\left\{\sup_{\left\{ \Delta_{\Theta_{j}}\right\}_{j=0}^{p}\in \mathcal{R}%
\left(3,C_{2}\right) }\left\vert \frac{1}{NT}\sum_{i=1}^{N}\sum_{t=1}^{T}%
\frac{u_{it}\rho_{it}\left(\left\{
\Delta_{\Theta_{j},it},x_{j,it}^{\ast}\right\}_{j=0}^{p},\epsilon_{it}^{\ast
}\right) }{\sum_{j=0}^{p}\left\Vert \Delta_{\Theta_{j}}\right\Vert_{F}}%
\right\vert >\frac{C_{5}a_{NT}}{4}\right\} \\
& \leq C(1+o(1))\left[ (p+1)\exp \left(-\frac{C_{5}}{4c_{8}\tau
\left(p+2\right) }\right) +\exp \left(-\frac{C_{5}}{4c_{8}(p+1)\left(p+2%
\right) }\right) \right] .
\end{align*}%
The RHS of the last inequality converges to zero as $C_{5}\rightarrow \infty$%
, which implies that 
\begin{equation*}
\sup_{\left\{ \Delta_{\Theta_{j}}\right\}_{j=0}^{p}\in \mathcal{R}%
\left(3,C_{2}\right) }\frac{\left\vert \frac{1}{NT}\sum_{i=1}^{N}%
\sum_{t=1}^{T}\tilde{\rho}_{it}\left(\left\{
\Delta_{\Theta_{j},it},X_{j,it}\right\}_{j=0}^{p},\epsilon_{it}\right)
\right\vert }{\sum_{j=0}^{p}\left\Vert\Delta_{\Theta_{j}}\right\Vert_{F}}%
=O_{p}\left(a_{NT}\right) .
\end{equation*}
\end{proof}

\subsection{Lemmas for the Proof of Theorem \protect\ref{Thm2}}

\begin{lemma}
{\small \label{Lem:Bern} } Let $\left\{\Upsilon_{t}, t=1,\cdots,T\right\}$
be a zero-mean strong mixing process, not necessarily stationary, with the
mixing coefficients satisfying $\alpha(z)\leq c_{\alpha}\rho^{z}$ for some $%
c_{\alpha}>0$ and $\rho\in(0,1)$. If $\sup_{1\leq t\leq T}\left\vert
\Upsilon_{t}\right\vert\leq M_{T}$, then there exist a constant $c_{9}$
depending on $c_{\alpha}$ and $\rho$ such that for any $T\geq 2$ and $d>0$,

\begin{itemize}
\item[(i)] $\mathbb{P}\left\{ \left\vert
\sum_{t=1}^{T}\Upsilon_{t}\right\vert > d \right\}\leq \exp\left\{ -\frac{%
c_{9}d^{2}}{M_{T}^{2}T+dM_{T}\left(\log T\right)\left(\log\log T\right)}%
\right\}$.

\item[(ii)] $\mathbb{P}\left\{ \left\vert
\sum_{t=1}^{T}\Upsilon_{t}\right\vert > d \right\}\leq \exp\left\{ -\frac{%
c_{9}d^{2}}{\upsilon_{0}^{2}T+M_{T}^{2}+d M_{T}\left(\log T\right)^{2}}%
\right\}$,
\end{itemize}

with $\upsilon_{0}^{2}=\sup_{t\in[T]}\left[Var(\Upsilon_{t})+2\sum_{s>t}%
\left\vert Cov(\Upsilon_{t},\Upsilon_{s}) \right\vert \right]$.
\end{lemma}

\begin{proof}
The proof is the same as that of Theorems 1 and 2 in \cite%
{merlevede2009bernstein}) with the condition assumed $\alpha (a)\leq
\exp\left\{ -2ca\right\} $ for some $c>0$ changed with $c_{\alpha }=1$ and $%
\rho=\exp \left\{ -2c\right\} $ in our lemma instead.
\end{proof}


\begin{lemma}
{\small \label{Lem:bounded u&v_tilde} } Suppose Assumptions \ref{ass:1}-\ref%
{ass:4} hold. Then, for $j\in\left\{0,\cdots,p\right\}$, we have

\begin{itemize}
\item[(i)] $\max_{i\in [N]}\left\Vert u_{i,j}^{0} \right\Vert_{2}\leq M$ and 
$\max_{t\in[T]}\left\Vert v_{t,j}^{0} \right\Vert_{2}\leq \frac{M}{%
\sigma_{K_{j},j}}\leq \frac{M}{c_{\sigma}}$.

\item[(ii)] $\max_{t\in[T]}\left\Vert O_{j}^{\prime}\tilde{v}%
_{t,j}\right\Vert_{2}\leq \frac{2M}{\sigma_{K_{j},j}}\leq \frac{2M}{%
c_{\sigma}}$ and $\max_{t\in[T]}\left\Vert O_{j}^{(1)\prime}\tilde{v}%
_{t,j}^{(1)}\right\Vert_{2}\leq \frac{2M}{\sigma_{K_{j},j}}\leq \frac{2M}{%
c_{\sigma}}$ w.p.a.1.

\item[(iii)] $\max_{i\in I_{2}}\frac{1}{T}\sum_{t\in[T]}\left\Vert \tilde{%
\phi}_{it}^{(1)}\right\Vert_{2}^{2}\leq \frac{4M^{2}}{c_{\sigma }}+\frac{%
4M^{2}pC}{c_{\sigma }}$ w.p.a.1.
\end{itemize}
\end{lemma}

\begin{proof}
(i) Recall that $\frac{1}{\sqrt{NT}}\Theta_{j}^{0}=\mathcal{U}%
_{j}^{0}\Sigma_{j}^{0}\mathcal{V}_{j}^{0\prime},$ $U_{j}^{0}=\sqrt{N}%
\mathcal{U}_{j}^{0}\Sigma_{j}^{0}$ and $V_{j}=\sqrt{T}\mathcal{V}_{j}$. Then
we have 
\begin{equation}
\frac{1}{\sqrt{T}}\Theta_{j}^{0}\mathcal{V}_{j}^{0}=\sqrt{N}\mathcal{U}%
_{j}^{0}\Sigma_{j}^{0}=U_{j}^{0}\quad \text{and}\quad\frac{1}{\sqrt{N}}%
\mathcal{U}_{j}^{0\prime}\Theta_{j}^{0}=\sqrt{T}\Sigma_{j}^{0}\mathcal{V}%
_{j}^{0\prime}=\Sigma_{j}^{0}V_{j}^{0\prime}.  \label{Lem12.1}
\end{equation}%
Hence, it's natural to see that 
\begin{equation*}
\left\Vert u_{i,j}^{0}\right\Vert_{2}=\frac{1}{\sqrt{T}}\left\Vert \left[%
\Theta_{j}^{0}\mathcal{V}_{j}^{0}\right]_{i.}\right\Vert_{2}\leq \frac{1}{%
\sqrt{T}}\left\Vert \left[ \Theta_{j}\right]_{i.}\right\Vert_{2}\leq M,
\end{equation*}%
where the first inequality is due to the fact that $\mathcal{V}_{j}$ is the
unitary matrix and the last inequality holds by Assumption \ref{ass:2}.
Since the upper bound $M$ is not dependent on $i$, this result holds
uniformly in $i$. Analogously, 
\begin{equation*}
\left\Vert v_{t,j}^{0}\right\Vert_{2}\leq \frac{1}{\sqrt{N}}%
c_{\sigma}^{-1}\left\Vert \left[ \mathcal{U}_{j}^{0\prime}\Theta_{j}^{0}%
\right]_{.t}\right\Vert_{2}\leq \frac{1}{\sqrt{N}}c_{\sigma }^{-1}\left\Vert %
\left[\Theta_{j}^{0}\right]_{.t}\right\Vert_{2}\leq \frac{M}{c_{\sigma }}.
\end{equation*}

(ii) As in (\ref{Lem12.1}), we have 
\begin{equation*}
\frac{1}{\sqrt{N}}\tilde{\mathcal{U}}_{j}^{\prime}\tilde{\Theta}_{j}=\sqrt{T}%
\tilde{\Sigma}_{j}\tilde{\mathcal{V}}_{j}^{\prime}=\tilde{\Sigma}_{j}\tilde{V%
}_{j}^{\prime}.
\end{equation*}%
It follows that 
\begin{equation*}
\left\Vert O_{j}^{\prime}\tilde{v}_{t,j}\right\Vert_{2}\leq \frac{1}{\sqrt{N}%
}\tilde{\sigma}_{K_{j},j}^{-1}\left\Vert \left[ \tilde{\mathcal{U}}%
_{j}^{\prime}\tilde{\Theta}_{j}\right]_{.t}\right\Vert_{2}\leq \frac{1}{%
\sqrt{N}}\tilde{\sigma}_{K_{j},j}^{-1}\left\Vert \left[ \tilde{\Theta}_{j}%
\right]_{.t}\right\Vert_{2}\leq \frac{2M}{c_{\sigma }},
\end{equation*}%
where the last inequality holds due to the fact that $\max_{k\in
[K_{j}]}\left\vert \tilde{\sigma}_{k,j}^{-1}-\Sigma_{k,j}^{-1}\right\vert
\leq \Sigma_{K_{j},j}^{-1}$ w.p.a.1. and the bounded parameter space where $%
\tilde{\Theta}_{j}$ lies in by Assumption \ref{ass:2} and ADMM algorithm
proposed in the last section. The upper bound of $\max_{t\in[T]}\left\Vert
O_{j}^{(1)\prime}\tilde{v}_{t,j}^{(1)}\right\Vert_{2}$ follows the same
argument as above.

(iii) We observe that 
\begin{equation*}
\max_{i\in I_{2}}\frac{1}{T}\sum_{t\in[T]}\left\Vert \tilde{\Phi}%
_{it}^{(1)}\right\Vert_{2}^{2}\leq \frac{1}{T}\sum_{t\in [T]}\left\Vert
O_{0}^{\prime}\tilde{v}_{t,0}^{(1)}\right\Vert_{2}^{2}+\max_{i\in
I_{2}}\sum_{j\in [p]}\frac{1}{T}\sum_{t\in[T]}\left\Vert O_{j}^{\prime}%
\tilde{v}_{t,j}^{(1)}\right\Vert_{2}^{2}\left\vert X_{j,it}\right\vert
^{2}\leq \frac{4M^{2}}{c_{\sigma }^{2}}+\frac{4M^{2}pC}{c_{\sigma }^{2}}%
\quad \text{w.p.a.1},
\end{equation*}%
where the last inequality holds by Lemma \ref{Lem:bounded u&v_tilde}(ii) and
Assumption \ref{ass:1}(iv).
\end{proof}

\begin{lemma}
{\small \label{Lem:phi eig} } Under Assumptions \ref{ass:1}--\ref{ass:5}, we
have

\begin{itemize}
\item[(i)] $\operatornamewithlimits{\min}\limits_{i\in
I_{2}}\lambda_{\min}\left(\tilde{\Phi}_{i}^{(1)}\right)\geq \frac{c_{\phi}}{2%
}$, $\operatornamewithlimits{\max}\limits_{i\in I_{2}}\lambda_{\max}\left(%
\tilde{\Phi}_{i}^{(1)}\right)\leq 2C_{\phi}$ w.p.a.1,

\item[(ii)] For $\forall j\in[p]$, $\operatornamewithlimits{\max}%
\limits_{i\in I_{2}}\frac{1}{T}\sum_{t\in[T]}\left[X_{j,it}^{2}-\mathbb{E}%
\left(X_{j,it}^{2}\bigg|\mathscr{D}_{\left\{e_{is}\right\}_{s<t}}^{I_{1}}
\right)\right]\left\Vert \tilde{O}_{j}^{(1)\prime}\tilde{v}%
_{t,j}^{(1)}-v_{t,j}^{0}\right\Vert_{2}^{2}=O_{p}(\eta_{N}^{2})$,

\item[(iii)] $\operatornamewithlimits{\max}\limits_{i\in I_{2}}\frac{1}{T}%
\sum_{t\in[T]}\left\Vert\tilde{\Phi}_{it}^{(1)}-\Phi_{it}^{0}\right%
\Vert_{2}^{2}=O_p(\eta_{N}^{2})$.
\end{itemize}
\end{lemma}

\begin{proof}
(i) Recall that 
\begin{align*}
& \Phi_{i}=\frac{1}{T}\sum_{t=1}^{T}\Phi_{it}^{0}\Phi_{it}^{0\prime}\quad 
\text{with}\quad
\Phi_{it}^{0}=(v_{t,0}^{0\prime},v_{t,1}^{0\prime}X_{1,it},...,v_{t,p}^{0%
\prime}X_{p,it})^{\prime},\quad \text{and} \\
& \tilde{\Phi}_{i}^{(1)}=\frac{1}{T}\sum_{t=1}^{T}\tilde{\Phi}_{it}^{(1)}%
\tilde{\Phi}_{it}^{(1)\prime}\quad \text{with}\quad \tilde{\Phi}_{it}^{(1)}=%
\left[ \left(O_{0}^{(1)\prime}\tilde{v}_{t,0}^{(1)}\right)
^{\prime},\left(O_{1}^{(1)\prime}\tilde{v}_{t,1}^{(1)}X_{1,it}\right)
^{\prime},\cdots,\left(O_{p}^{(1)\prime}\tilde{v}_{t,p}^{(1)}X_{p,it}\right)
^{\prime}\right] ^{\prime }.
\end{align*}%
Uniformly over $i\in I_{2}$, it is clear that 
\begin{align*}
& \left\Vert \tilde{\Phi}_{i}^{(1)}-\Phi_{i}\right\Vert_{F} \\
& \lesssim \frac{4M}{c_{\sigma }T}\sum_{t=1}^{T}\left\Vert O_{0}^{(1)\prime}%
\tilde{v}_{t,0}^{(1)}-v_{t,0}^{0}\right\Vert_{2}+\frac{4M}{c_{\sigma }T}%
\sum_{j=1}^{p}\sum_{t=1}^{T}\left\Vert O_{j}^{(1)\prime}\tilde{v}%
_{t,j}^{(1)}-v_{t,j}^{0}\right\Vert_{2}\left\vert X_{j,it}\right\vert \\
& \leq \frac{4M}{c_{\sigma }}\frac{1}{\sqrt{T}}\left\Vert O_{0}^{(1)\prime}%
\tilde{V}_{0}^{(1)}-V_{0}^{0}\right\Vert_{F}+\frac{4M^{2}}{c_{\sigma }}%
\sum_{j=1}^{p}\frac{1}{\sqrt{T}}\left\Vert O_{j}^{(1)\prime}\tilde{V}%
_{j}^{(1)}-V_{j}^{0}\right\Vert_{F}\left(\frac{1}{T}\sum_{t\in
[T]}\left\vert X_{j,it}\right\vert ^{2}\right)
^{1/2}=O_{p}\left(\eta_{N}\right) ,
\end{align*}%
where the third line holds by Lemma \ref{Lem:bounded u&v_tilde}(i) and
Assumption \ref{ass:1}(iv). It follows that 
\begin{equation*}
\min_{i\in I_{2}}\lambda_{\min }\left[ \tilde{\Phi}_{i}^{(1)}\right]
\geq\min_{i\in I_{2}}\lambda_{\min }\left[ \Phi_{i}\right]
-O\left(\eta_{N}\right) \geq \frac{c_{\phi }}{2},\quad \text{w.p.a.1}
\end{equation*}
and 
\begin{equation*}
\max_{i\in I_{2}}\lambda_{\max }\left[ \tilde{\Phi}_{i}^{(1)}\right]
\leq\max_{i\in I_{2}}\lambda_{\max }\left[ \Phi_{i}\right]
+O\left(\eta_{N}\right) \leq 2C_{\phi },\quad \text{w.p.a.1}.
\end{equation*}

(ii) Let $I_{j,i}:=\frac{1}{T}\sum_{t\in[T]}I_{j,it}$ such that $I_{j,it}=%
\left[X_{j,it}^{2}-\mathbb{E}\left(X_{j,it}^{2}\bigg|\mathscr{D}%
_{\left\{e_{is}\right\}_{s<t}}^{I_{1}} \right)\right]\left\Vert \tilde{O}%
_{j}^{(1)\prime}\tilde{v}_{t,j}^{(1)}-v_{t,j}^{0}\right\Vert_{2}^{2}$. Then
for a constant $c$, we have 
\begin{align*}
&\mathbb{P}\left(\max_{i\in I_{2}}\left\vert\sum_{t\in[T]}
I_{j,it}\right\vert>cT\eta_{N}^{2} \right)\leq\sum_{i\in I_{2}}\mathbb{P}%
\left(\left\vert\sum_{t\in[T]} I_{j,it}\right\vert>cT\eta_{N}^{2} \right) \\
&=\sum_{i\in I_{2}}\mathbb{E}\mathbb{P}\left(\left\vert\sum_{t\in[T]}
I_{j,it}\right\vert>cT\eta_{N}^{2}\bigg|\mathscr{D}_{\left\{e_{is}\right%
\}_{s<T}}^{I_{1}}\right) \\
&\leq 2\sum_{i\in I_{2}}\exp\left\{-\frac{2(cT\eta_{N}^{2})^{2}}{\sum_{t\in[T%
]}\left[2\xi_{N}^{2}\left\Vert \tilde{O}_{j}^{(1)\prime}\tilde{v}%
_{t,j}^{(1)}-v_{t,j}^{0}\right\Vert_{2}^{2} \right]^{2}}\right\} \\
&\leq 2\exp\left\{-\frac{2(cT\eta_{N}^{2})^{2}}{4\xi_{N}^{4}\left[\frac{M^{2}%
}{c_{\sigma}^{2}}+\frac{4M^{2}}{c_{\sigma}^{2}} \right]\sum_{t\in[T]%
}\left\Vert \tilde{O}_{j}^{(1)\prime}\tilde{v}_{t,j}^{(1)}-v_{t,j}^{0}\right%
\Vert_{2}^{2} }+\log N\right\}=o(1),
\end{align*}
where the first inequality combines the fact that $I_{j,it}$ is the
martingale difference sequence, Assumption \ref{ass:1}(v) and the
Azuma-Hoeffding inequality in \citet[Corollary 2.20]{wainwright2019high}.
The last inequality is by Lemma \ref{Lem:bounded u&v_tilde}(i) and \ref%
{Lem:bounded u&v_tilde}(ii), and the finial result is by the definition of $%
\eta_{N}$.

(iii) Note that 
\begin{align*}
&\max_{i\in I_{2}}\frac{1}{T}\sum_{t\in[T]}\left\Vert \tilde{\Phi}%
_{it}^{(1)}-\Phi_{it}^{0}\right\Vert_{2} ^{2}\leq\frac{1}{T}\sum_{t\in[T]%
}\left\Vert \tilde{O}_{0}^{(1)\prime}\tilde{v}_{t,0}^{(1)}-v_{t,0}^{0}\right%
\Vert_{2}^{2}+p\max_{i\in I_{2},j\in \left[p\right] }\frac{1}{T}\sum_{t\in[T]%
}\left\vert X_{j,it}\right\vert^{2}\left\Vert \tilde{O}_{j}^{(1)\prime}%
\tilde{v}_{t,j}^{(1)}-v_{t,j}^{0}\right\Vert_{2}^{2} \\
&=\frac{1}{T}\sum_{t\in[T]}\left\Vert \tilde{O}_{0}^{(1)\prime}\tilde{v}%
_{t,0}^{(1)}-v_{t,0}^{0}\right\Vert_{2}^{2}+p\max_{i\in I_{2},j\in \left[p%
\right] }\frac{1}{T}\sum_{t\in[T]}\left[X_{j,it}^{2}-\mathbb{E}%
\left(X_{j,it}^{2}\bigg|\mathscr{D}_{\left\{e_{is}\right\}_{s<t}}^{I_{1}}
\right)\right]\left\Vert \tilde{O}_{j}^{(1)\prime}\tilde{v}%
_{t,j}^{(1)}-v_{t,j}^{0}\right\Vert_{2}^{2} \\
&+\max_{i\in I_{2},j\in \left[p\right] }\frac{1}{T}\sum_{t\in[T]}\mathbb{E}%
\left(X_{j,it}^{2}\bigg|\mathscr{D}_{\left\{e_{is}\right\}_{s<t}}^{I_{1}}
\right)\left\Vert \tilde{O}_{j}^{(1)\prime}\tilde{v}_{t,j}^{(1)}-v_{t,j}^{0}%
\right\Vert_{2}^{2} \\
& \leq \frac{1}{T}\sum_{t\in[T]}\left\Vert \tilde{O}_{0}^{(1)\prime}\tilde{v}%
_{t,0}^{(1)}-v_{t,0}^{0}\right\Vert_{2}^{2}+O_{p}(\eta_{N}^{2})+pM\max_{j\in %
\left[ p\right]}\frac{1}{T}\sum_{t\in[T]}\left\Vert \tilde{O}_{j}^{(1)\prime}%
\tilde{v}_{t,j}^{(1)}-v_{t,j}^{0}\right\Vert_{2}^{2}\quad \\
& =\frac{1}{T}\left\Vert O_{0}^{(1)}\tilde{V}_{0}^{(1)}-V_{0}^{0}\right%
\Vert_{F}^{2}+p\max_{j\in \left[ p\right] }\frac{1}{T}\left\Vert O_{j}^{(1)}%
\tilde{V}_{j}^{(1)}-V_{j}^{0}\right\Vert_{F}^{2}+O_{p}(\eta_{N}^{2}) \\
& =O_{p}(\eta_{N}^{2}) ,
\end{align*}
where the the second inequality holds by Assumption \ref{ass:1}(iv) and
Lemma \ref{Lem:phi eig}(ii), and the last equality holds by the Theorem \ref%
{Thm1}(ii).
\end{proof}

\begin{lemma}
\label{Lem:A} Recall $\{A_{1,i},\cdots,A_{7,i}\}_{i\in I_{2}}$ and $%
q_{i}^{I} $ defined in \eqref{A:obj diff} and \eqref{eq:qiI}, respectively.
Suppose Assumptions \ref{ass:1}--\ref{ass:5} hold. Then for any constant $%
c_{11}<\min (\frac{3\underline{\mathfrak{f}}}{\bar{\mathfrak{f}}^{\prime }}%
,1)$, we have 
\begin{align*}
& \max_{i\in I_{2}}(|A_{m,i}|/\left\Vert \dot{\Delta}_{i,u}\right%
\Vert_{2})=O_{p}\left(\eta_{N}\right) ,\forall m\in\{1,2,3,5,6,7\}\quad 
\text{and} \\
& |A_{4,i}|\geq \min \left(\frac{\left(3c_{11}^{2}\underline{\mathfrak{f}}%
-c_{11}^{3}\bar{\mathfrak{f}}^{\prime }\right) c_{\phi }\left\Vert \dot{%
\Delta}_{i,u}\right\Vert_{2}^{2}}{12},\frac{\left(3c_{11}^{2}\underline{%
\mathfrak{f}}-c_{11}^{3}\bar{\mathfrak{f}}^{\prime }\right) \sqrt{c_{\phi }}%
q_{i}^{I}\left\Vert \dot{\Delta}_{i,u}\right\Vert_{2}}{6\sqrt{2}}%
\right),\quad \forall i\in I_{2},\quad w.p.a.1.
\end{align*}
\end{lemma}

\begin{proof}
Recall that $w_{1,it}=\epsilon_{it}-u_{i}^{0\prime}\left(\tilde{\Phi}%
_{it}^{(1)}-\Phi_{it}^{0}\right) .$ Let $w_{2,it}=\tilde{\Phi}%
_{it}^{(1)\prime}\dot{\Delta}_{i,u}$. For some positive constant $%
c_{11}\in(0,1]$, we observe that 
\begin{align}
A_{4,i}& =\frac{1}{T}\sum_{t=1}^{T}\int_{0}^{w_{2,it}}\mathbb{E}\left(%
\mathbf{1}\left\{ \epsilon_{it}\leq s\right\} -\mathbf{1}\left\{
\epsilon_{it}\leq 0\right\} \bigg |\mathscr{D}_{e_{i}}^{I_{1}}\right) ds 
\notag  \label{B.2} \\
& =\frac{1}{T}\sum_{t=1}^{T}\int_{0}^{w_{2,it}}\left[ \mathfrak{F}_{it}(s)-%
\mathfrak{F}_{it}(0)\right] ds\geq \frac{1}{T}\sum_{t=1}^{T}%
\int_{0}^{c_{11}w_{2,it}}\left[ \mathfrak{F}_{it}(s)-\mathfrak{F}_{it}(0)%
\right] ds  \notag \\
& =\frac{1}{T}\sum_{t=1}^{T}\int_{0}^{c_{11}w_{2,it}}\left[ s\mathfrak{f}%
_{it}(0)+\frac{s^{2}}{2}\mathfrak{f}_{it}^{\prime }(\tilde{s})\right] ds 
\notag \\
& \geq \frac{1}{T}\sum_{t=1}^{T}\left[ \frac{c_{11}^{2}\underline{\mathfrak{f%
}}\left(\tilde{\Phi}_{it}^{(1)\prime}\dot{\Delta}_{i,u}\right) ^{2}}{2}-%
\frac{c_{11}^{3}\bar{\mathfrak{f}}^{\prime }\left\vert \tilde{\Phi}%
_{it}^{(1)\prime}\dot{\Delta}_{i,u}\right\vert ^{3}}{6}\right]
\end{align}%
where $\tilde{s}\in (0,s)$. Here, due to Assumption \ref{ass:1}(i), the
conditional CDF of $\epsilon_{it}$ given $\mathscr{D}_{e_{i}}^{I_{1}}$ is
the same as that given the $\sigma$-field generated by $\left\{
\left\{e_{it}\right\}_{t\in[T]}\bigcup \left\{ V_{j}^{0}\right\}_{j\in[p]%
\cup \{0\}}\bigcup \left\{ W_{j}^{0}\right\}_{j\in [p]}\right\} $, which
leads to the second equality of the above display by Assumption \ref{ass:1}%
(vii); the first inequality holds by Lemma \ref{Lem:ind diff}(ii) for any $%
c_{11}\in \left(0,1\right]$. Here, we choose $c_{11}$ such that $c_{11}<%
\frac{3\underline{\mathfrak{f}}}{\bar{\mathfrak{f}}^{\prime }}$; and the
last inequality holds by Assumption \ref{ass:1}(vii).

Let $q_{i}^{II}=\left[ \frac{1}{T}\sum_{t\in[T]}\left(\tilde{\Phi}%
_{it}^{(1)\prime}\dot{\Delta}_{i,u}\right) ^{2}\right] ^{\frac{1}{2}}$ and
recall that $q_{i}^{I}=\operatornamewithlimits{\inf}\limits_{\Delta}\frac{%
\left[ \frac{1}{T}\sum_{t\in[T]}\left(\tilde{\Phi}_{it}^{(1)\prime}\Delta
\right) ^{2}\right] ^{\frac{3}{2}}}{\frac{1}{T}\sum_{t\in[T]}\left\vert 
\tilde{\Phi}_{it}^{(1)\prime}\Delta \right\vert ^{3}}$. If $q_{i}^{I}\geq
q_{i}^{II}$, we notice that $\frac{1}{T}\sum_{t\in[T]}\left\vert \tilde{\Phi}%
_{it}^{(1)\prime}\Delta \right\vert ^{3}<\left(q_{i}^{II}\right) ^{2}$ and $%
A_{4,i}\geq \frac{c_{11}^{2}\underline{\mathfrak{f}}\left(q_{i}^{II}%
\right)^{2}}{2}-\frac{c_{11}^{3}\bar{\mathfrak{f}}^{\prime
}\left(q_{i}^{II}\right) ^{2}}{6}=\frac{3c_{11}^{2}\underline{\mathfrak{f}}%
-c_{11}^{3}\bar{\mathfrak{f}}^{\prime }}{6}\left(q_{i}^{II}\right) ^{2}$. If 
$q_{i}^{I}<q_{i}^{II}$, we have $\left[ \frac{1}{T}\sum_{t\in [T]}\left(%
\tilde{\Phi}_{it}^{(1)\prime}\Delta_{i,u}^{\ast }\right) ^{2}\right] ^{\frac{%
1}{2}}=q_{i}^{I}$ with $\Delta_{i,u}^{\ast }=\frac{q_{i}^{II}\dot{\Delta}%
_{i,u}}{q_{i}^{I}}$. Define the function 
\begin{equation*}
F\left(\Delta \right) =\frac{1}{T}\sum_{t=1}^{T}\int_{0}^{\tilde{\Phi}%
_{it}^{(1)\prime}\Delta }\left[ \mathfrak{F}_{it}(s)-\mathfrak{F}_{it}(0)%
\right] ds.
\end{equation*}
Note that the second-order derivative of function $F\left(\Delta \right)$ is
no less than zero, which implies $F\left(\Delta \right) $ is convex.
Therefore, we have 
\begin{align*}
F\left(\dot{\Delta}_{i,u}\right) & =F\left(\frac{q_{i}^{I}\Delta_{i,u}^{\ast
}}{q_{i}^{II}}\right) \geq \frac{q_{i}^{I}}{q_{i}^{II}}F\left(\Delta_{i,u}^{%
\ast }\right) \geq \frac{q_{i}^{II}}{q_{i}^{I}}\frac{3c_{11}^{2}\underline{%
\mathfrak{f}}-c_{11}^{3}\bar{\mathfrak{f}}^{\prime }}{6}\frac{1}{T}\sum_{t\in%
[T]}\left(\tilde{\Phi}_{it}^{(1)\prime}\Delta_{i,u}^{\ast }\right) ^{2}=%
\frac{\left(3c_{11}^{2}\underline{\mathfrak{f}}-c_{11}^{3}\bar{\mathfrak{f}}%
^{\prime }\right) q_{i}^{I}q_{i}^{II}}{6}.
\end{align*}
Combining these two cases, we have 
\begin{align}
A_{4,i}& \geq \min \left(\frac{3c_{11}^{2}\underline{\mathfrak{f}}-c_{11}^{3}%
\bar{\mathfrak{f}}^{\prime }}{6}\left(q_{i}^{II}\right) ^{2},\frac{%
\left(3c_{11}^{2}\underline{\mathfrak{f}}-c_{11}^{3}\bar{\mathfrak{f}}%
^{\prime }\right) q_{i}^{I}q_{i}^{II}}{6}\right)  \notag  \label{B.4} \\
& \geq \min \left(\frac{\left(3c_{11}^{2}\underline{\mathfrak{f}}-c_{11}^{3}%
\bar{\mathfrak{f}}^{\prime }\right) c_{\phi }\left\Vert \dot{\Delta}%
_{i,u}\right\Vert_{2}^{2}}{12},\frac{\left(3c_{11}^{2}\underline{\mathfrak{f}%
}-c_{11}^{3}\bar{\mathfrak{f}}^{\prime }\right) \sqrt{c_{\phi }}%
q_{i}^{I}\left\Vert \dot{\Delta}_{i,u}\right\Vert_{2}}{6\sqrt{2}}\right) ,
\end{align}%
where the second inequality holds by Lemma \ref{Lem:phi eig}(i).

As for $\left\vert A_{1,i}\right\vert $, we notice that 
\begin{align}
\max_{i\in I_{2}}\left(\left\vert A_{1,i}\right\vert /\left\Vert \dot{\Delta}%
_{i,u}\right\Vert_{2}\right) & =\max_{i\in I_{2}}\frac{\left\vert \frac{1}{T}%
\sum_{t=1}^{T}\mathbb{E}\left\{ \tilde{\Phi}_{it}^{(1)\prime}\left(\tau -%
\mathbf{1}\left\{ \epsilon_{it}\leq u_{i}^{0\prime}\left(\tilde{\Phi}%
_{it}^{(1)}-\Phi_{it}^{0}\right) \right\} \right) \bigg |\mathscr{D}%
_{e_{i}}^{I_{1}}\right\} \dot{\Delta}_{i,u}\right\vert }{\left\Vert \dot{%
\Delta}_{i,u}\right\Vert_{2}}  \notag  \label{Lem:A:A1} \\
& =\max_{i\in I_{2}}\frac{\left\vert \frac{1}{T}\sum_{t=1}^{T}\tilde{\Phi}%
_{it}^{(1)\prime}\dot{\Delta}_{i,u}\left(\tau -\mathfrak{F}_{it}\left[%
u_{i}^{0\prime}\left(\tilde{\Phi}_{it}^{(1)}-\Phi_{it}^{0}\right) \right]%
\right) \right\vert }{\left\Vert \dot{\Delta}_{i,u}\right\Vert_{2}}  \notag
\\
& =\max_{i\in I_{2}}\frac{\left\vert \frac{1}{T}\sum_{t=1}^{T}\tilde{\Phi}%
_{it}^{(1)\prime}\dot{\Delta}_{i,u}\left[ \mathfrak{f}_{it}(s_{it})u_{i}^{0%
\prime}\left(\tilde{\Phi}_{it}^{(1)}-\Phi_{it}^{0}\right) \right]\right\vert 
}{\left\Vert \dot{\Delta}_{i,u}\right\Vert_{2}}  \notag \\
& \leq \max_{i\in I_{2}}\frac{\bar{\mathfrak{f}}}{T}\sum_{t\in
[T]}\left\Vert \tilde{\Phi}_{it}^{(1)}\right\Vert_{2}\left\Vert \tilde{\Phi}%
_{it}^{(1)}-\Phi_{it}^{0}\right\Vert_{2}\left\Vert u_{i}^{0}\right\Vert_{2} 
\notag \\
& \leq \max_{i\in I_{2}}\sqrt{\frac{1}{T}\sum_{t\in[T]}\left\Vert\tilde{\Phi}%
_{it}^{(1)}\right\Vert ^{2}}\sqrt{\frac{1}{T}\sum_{t\in [T]}\left\Vert 
\tilde{\Phi}_{it}^{(1)}-\Phi_{it}^{0}\right\Vert ^{2}}\left\Vert
u_{i}^{0}\right\Vert_{2}  \notag \\
& \leq O_{p}(\eta_{N}),
\end{align}
where the second and third equalities hold by Assumption \ref{ass:1}(vii)
and mean-value theorem with some $\left\vert s_{it}\right\vert \in
\left(0,\left\vert u_{i}^{0\prime}(\tilde{\Phi}_{it}^{(1)}-\Phi_{it}^{0})%
\right\vert \right) ,$ the second inequality holds by Cauchy-Schwarz
inequality, and the third inequality holds by Lemmas \ref{Lem:bounded
u&v_tilde}(i), \ref{Lem:bounded u&v_tilde}(iii) and \ref{Lem:phi eig}(ii).

For $A_{3,i}$, note that 
\begin{align}
A_{3,i}& =\frac{1}{T}\sum_{t\in[T]}\int_{0}^{\tilde{\Phi}_{it}^{(1)\prime}%
\dot{\Delta}_{i,u}}\left(\mathbf{1}\left\{ \epsilon_{it}\leq s\right\} -%
\mathbf{1}\left\{ \epsilon_{it}\leq 0\right\} \right)-\left\{ \mathbb{E}%
\left(\mathbf{1}\left\{ \epsilon_{it}\leq s\right\} -\mathbf{1}\left\{
\epsilon_{it}\leq 0\right\} \bigg |\mathscr{D}_{e_{i}}^{I_{1}}\right)
\right\} ds  \notag  \label{A3} \\
& =\int_{0}^{1}\biggl\{\frac{1}{T}\sum_{t\in[T]}\tilde{\Phi}_{it}^{(1)}\left(%
\mathbf{1}\left\{ \epsilon_{it}\leq \tilde{\Phi}_{it}^{(1)\prime}\dot{\Delta}%
_{i,u}s^{\ast }\right\} -\mathbf{1}\left\{\epsilon_{it}\leq 0\right\} \right)
\notag \\
& -\left[ \mathbb{E}\left(\mathbf{1}\left\{ \epsilon_{it}\leq \tilde{\Phi}%
_{it}^{(1)\prime}\dot{\Delta}_{i,u}s^{\ast }\right\} -\mathbf{1}%
\left\{\epsilon_{it}\leq 0\right\} \bigg |\mathscr{D}_{e_{i}}^{I_{1}}\right) %
\right] \biggr\}^{\prime}\dot{\Delta}_{i,u}ds^{\ast }  \notag \\
& \leq \sup_{s\in \mathbb{R}}\left\Vert
A_{3,i}^{I}(s)\right\Vert_{2}\left\Vert \dot{\Delta}_{i,u}\right\Vert_{2}
\end{align}%
by change of variables with $s^{\ast }=\frac{s}{\tilde{\Phi}_{it}^{(1)\prime}%
\dot{\Delta}_{i,u}}$, $A_{3,i}^{I}(s)=\frac{1}{T}\sum_{t\in
[T]}A_{3,it}^{I}(s)$ and 
\begin{equation*}
A_{3,it}^{I}(s)=\tilde{\Phi}_{it}^{(1)}\left[ \left(\mathbf{1}%
\left\{\epsilon_{it}\leq s\right\} -\mathbf{1}\left\{ \epsilon_{it}\leq
0\right\}\right) -\mathbb{E}\left(\mathbf{1}\left\{ \epsilon_{it}\leq
s\right\} -\mathbf{1}\left\{ \epsilon_{it}\leq 0\right\} \bigg |\mathscr{D}%
_{e_{i}}^{I_{1}}\right) \right] .
\end{equation*}
Below, we aim to show $\operatornamewithlimits{\sup}_{s\in (-\infty
,+\infty)}\max_{i\in I_{2}}\left\Vert
A_{3,i}^{I}(s)\right\Vert_{2}=O_{p}\left(\eta_{N}\right) .$

When $\left\vert s\right\vert >T^{1/4}$, we notice that 
\begin{align*}
\sup_{\left\vert s\right\vert >T^{1/4}}\max_{i\in I_{2}}\left\Vert
A_{3,i}^{I}(s)\right\Vert_{2}& \leq \sup_{\left\vert
s\right\vert>T^{1/4}}\max_{i\in I_{2}}\left\Vert \frac{1}{T}\sum_{t\in[T]}%
\tilde{\Phi}_{it}^{(1)}\left[ \mathbf{1}\left\{ \epsilon_{it}\leq s\right\} -%
\mathbb{E}\left(\mathbf{1}\left\{ \epsilon_{it}\leq s\right\} \bigg |%
\mathscr{D}_{e_{i}}^{I_{1}}\right) \right] \right\Vert_{2} \\
& +\max_{i\in I_{2}}\left\Vert \frac{1}{T}\sum_{t\in[T]}\tilde{\Phi}%
_{it}^{(1)}\left[ \mathbf{1}\left\{ \epsilon_{it}\leq 0\right\} -\mathbb{E}%
\left(\mathbf{1}\left\{ \epsilon_{it}\leq 0\right\} \bigg |\mathscr{D}%
_{e_{i}}^{I_{1}}\right) \right] \right\Vert_{2} \\
& \leq \max_{i\in I_{2}}\frac{1}{T}\sum_{t\in[T]}\left\Vert \tilde{\Phi}%
_{it}^{(1)}\right\Vert_{2}\mathbf{1}\left\{ \epsilon_{it}>T^{1/4}\right\}
+\max_{i\in I_{2}}\frac{1}{T}\sum_{t\in [T]}\left\Vert \tilde{\Phi}%
_{it}^{(1)}\right\Vert_{2}\mathbb{E}\left(\mathbf{1}\left\{
\epsilon_{it}>T^{1/4}\right\} \bigg |\mathscr{D}_{e_{i}}^{I_{1}}\right) \\
& +\max_{i\in I_{2}}\left\Vert \frac{1}{T}\sum_{t\in[T]}\tilde{\Phi}%
_{it}^{(1)}\left[ \mathbf{1}\left\{ \epsilon_{it}\leq 0\right\} -\mathbb{E}%
\left(\mathbf{1}\left\{ \epsilon_{it}\leq 0\right\} \bigg |\mathscr{D}%
_{e_{i}}^{I_{1}}\right) \right] \right\Vert_{2} \\
& \leq 2\max_{i\in I_{2}}\frac{1}{T}\sum_{t\in[T]}\left\Vert \tilde{\Phi}%
_{it}^{(1)}\right\Vert_{2}\mathbb{E}\left(\mathbf{1}\left\{
\epsilon_{it}>T^{1/4}\right\} \bigg |\mathscr{D}_{e_{i}}^{I_{1}}\right) \\
& +\max_{i\in I_{2}}\left\vert \frac{1}{T}\sum_{t\in[T]}\left\Vert \tilde{\Phi}_{it}^{(1)}\right\Vert_{2}\left[ \mathbf{1}\left\{
\epsilon_{it}>T^{1/4}\right\} -\mathbb{E}\left(\mathbf{1}\left\{
\epsilon_{it}>T^{1/4}\right\} \bigg |\mathscr{D}_{e_{i}}^{I_{1}}\right) %
\right]\right\vert \\
& +\max_{i\in I_{2}}\left\vert \frac{1}{T}\sum_{t\in[T]}\left\Vert\tilde{\Phi}_{it}^{(1)}\right\Vert_{2}\left[ \mathbf{1}\left\{ \epsilon_{it}\leq
0\right\} -\mathbb{E}\left(\mathbf{1}\left\{ \epsilon_{it}\leq 0\right\} %
\bigg |\mathscr{D}_{e_{i}}^{I_{1}}\right) \right] \right\vert ,
\end{align*}
where the first and third inequalities hold by triangle inequality. Besides,
we observe that 
\begin{align}
& \max_{i\in I_{2}}\frac{1}{T}\sum_{t\in[T]}\left\Vert \tilde{\Phi}_{it}^{(1)}\right\Vert_{2}\mathbb{E}\left(\mathbf{1}\left\{
\epsilon_{it}>T^{1/4}\right\} \bigg |\mathscr{D}_{e_{i}}^{I_{1}}\right) 
\notag \\
& =\max_{i\in I_{2}}\frac{1}{T}\sum_{t\in[T]}\left\Vert \tilde{\Phi}_{it}^{(1)}\right\Vert_{2}\mathbb{P}\left(\epsilon_{it}>T^{1/4}\bigg|%
\mathscr{D}_{e_{i}}^{I_{1}}\right) \leq \max_{i\in I_{2}}\frac{1}{T}\sum_{t%
\in[T]}\left\Vert \tilde{\Phi}_{it}^{(1)}\right\Vert_{2}\frac{\mathbb{E}%
\left(\epsilon_{it}^{2}\bigg|\mathscr{D}_{e_{i}}^{I_{1}}\right) }{\sqrt{T}} 
\notag \\
& \leq T^{-1/2}\max_{i\in I_{2}}\sqrt{\frac{1}{T}\sum_{t\in [T]}\left\Vert 
\tilde{\Phi}_{it}^{(1)}\right\Vert_{2}^{2}}\max_{i\in I_{2}}\sqrt{\frac{1}{T}%
\sum_{t\in[T]}\left[ \mathbb{E}\left(\epsilon_{it}^{2}\bigg|\mathscr{D}%
_{e_{i}}\right) \right] ^{2}}=O_{p}\left(T^{-1/2}\right) ,  \label{B.9}
\end{align}%
where the last equality holds by Lemma \ref{Lem:bounded u&v_tilde}(iii) and
Assumption \ref{ass:1}(iv). For a positive constant $c_{12}$, 
\begin{align*}
& \mathbb{P}\left(\max_{i\in I_{2}}\left\vert \frac{1}{T}\sum_{t\in
[T]}\left\Vert \tilde{\Phi}_{it}^{(1)}\right\Vert_{2}\left[ \mathbf{1}%
\left\{ \epsilon_{it}>T^{1/4}\right\} -\mathbb{E}\left(\mathbf{1}%
\left\{\epsilon_{it}>T^{1/4}\right\} \bigg |\mathscr{D}_{e_{i}}^{I_{1}}%
\right) \right] \right\vert >c_{12}\frac{\xi_{N}\sqrt{\log (N\vee T)}}{\sqrt{%
T}}\right) \\
& \leq \sum_{i\in I_{2}}\mathbb{E}\mathbb{P}\left(\left\vert \frac{1}{T}%
\sum_{t\in[T]}\left\Vert \tilde{\Phi}_{it}^{(1)}\right\Vert_{2}\left[ 
\mathbf{1}\left\{ \epsilon_{it}>T^{1/4}\right\} -\mathbb{E}\left(\mathbf{1}%
\left\{ \epsilon_{it}>T^{1/4}\right\} \bigg |\mathscr{D}_{e_{i}}^{I_{1}}%
\right) \right] \right\vert >c_{12}\frac{\xi_{N}\sqrt{\log(N\vee T)}}{\sqrt{T%
}}\bigg|\mathscr{D}_{e_{i}}^{I_{1}}\right) \\
& \leq \sum_{i\in I_{2}}\exp \left\{ -\frac{c_{9}c_{12}^{2}T\xi_{N}^{2}%
\log(N\vee T)}{\frac{16TM^{2}}{c_{\sigma }^{2}}\left(1+p\xi_{N}^{2}\right) +%
\frac{4Mc_{12}}{c_{\sigma }}\sqrt{1+p\xi_{N}}\sqrt{T}\xi_{N}\sqrt{\log(N\vee
T)}\log T\left(\log \log T\right) }\right\} =o(1),
\end{align*}%
where the first inequality holds by the union bound and Assumption \ref%
{ass:1}(i), and the second inequality holds by Assumption \ref{ass:1}(iii)
and the conditional Bernstein's inequality in Lemma \ref{Lem:Bern}(i) with
the fact that $\max_{i\in I_{2},t\in[T]}$ $\left\Vert \tilde{\Phi}%
_{it}^{(1)}\right\Vert_{2}\leq \frac{2M}{c_{\sigma }}\sqrt{1+p\xi_{N}^{2}}$
w.p.a.1. Here, we can apply the conditional Bernstein's inequality because $%
\tilde{\Phi}_{it}^{(1)}$ and $\mathbb{E}\left(\mathbf{1}\left\{
\epsilon_{it}>T^{1/4}\right\} \bigg |\mathscr{D}_{e_{i}}^{I_{1}}\right) $
are deterministic given $\left\{ \mathscr{D}_{e_{i}}^{I_{1}}\right\}_{i\in
I_{2},t\in[T]}$ so that the only randomness comes from $\{\epsilon_{it}\}_{t%
\in[T]}$. Furthermore, the joint distribution of $\{\epsilon_{it}\}_{t\in[T]%
} $ given $\mathscr{D}_{e_{i}}^{I_{1}}$ is the same as that given the $%
\sigma $-field generated by $\mathscr{D}_{e_{i}}$ due to the independence
structure assumed in Assumption \ref{ass:1}(i). Last, given the $\sigma$%
-field generated by $\mathscr{D}_{e_{i}}$, $\{\epsilon_{it}\}_{t\in[T]}$ is
strong mixing with mixing coefficient $\alpha_{i}(\cdot)$ as assumed by
Assumption \ref{ass:1}(iii). Similarly, we obtain that 
\begin{equation*}
\max_{i\in I_{2}}\left\vert \frac{1}{T}\sum_{t\in[T]}\left\Vert\tilde{\Phi}%
_{it}^{(1)}\right\Vert_{2}\left[ \mathbf{1}\left\{ \epsilon_{it}\leq
0\right\} -\mathbb{E}\left(\mathbf{1}\left\{ \epsilon_{it}\leq 0\right\} %
\bigg |\mathscr{D}_{e_{i}}^{I_{1}}\right) \right] \right\vert=O_{p}\left(%
\frac{\sqrt{\log (N\vee T)}\xi_{N}}{\sqrt{T}}\right) ,
\end{equation*}%
which implies 
\begin{equation}
\sup_{\left\vert s\right\vert >T^{1/4}}\max_{i\in I_{2}}\left\vert
A_{3,i}^{I}(s)\right\vert =O_{p}(\eta_{N}).  \label{B.10}
\end{equation}

For $|s|\leq T^{1/4}$, let $\mathbb{S}=[-T^{1/4},T^{1/4}]$ and divide $%
\mathbb{S}$ into $\mathbb{S}_{m}$ for $m=1,\cdots,n_{\mathbb{S}}$ such that $%
\left\vert s-\bar{s}\right\vert<\frac{\varepsilon}{T}$ for $s$ and $\bar{s}%
\in\mathbb{S}_{m}$ and $n_{\mathbb{S}}\asymp T^{5/4}$. Let $s_{m}\in\mathbb{S%
}_{m}$. For any $s\in\mathbb{S}_{m}$, we have 
\begin{align}  \label{B.11}
\left\Vert \frac{1}{T}\sum_{t\in[T]} A_{3,it}^{I}(s)
\right\Vert_{2}\leq\left\Vert \frac{1}{T}\sum_{t\in[T]}
A_{3,it}^{I}\left(s_{m}\right)\right\Vert_{2}+\left\Vert \frac{1}{T}\sum_{t%
\in[T]} \left[A_{3,it}^{I}\left(s\right) -A_{3,it}^{I}\left(s_{m}\right) %
\right]\right\Vert_{2},
\end{align}
such that 
\begin{align}  \label{Lem:A:A3:diff}
&\max_{i\in I_{2},m \in [n_{\mathbb{S}}]}\sup_{s \in \mathbb{S}%
_{m}}\left\Vert \frac{1}{T}\sum_{t\in[T]} \left[A_{3,it}^{I}\left(s%
\right)-A_{3,it}^{I}\left(s_{m}\right) \right]\right\Vert_{2}  \notag \\
&=\max_{i\in I_{2},m \in [n_{\mathbb{S}}]}\sup_{s \in \mathbb{S}%
_{m}}\left\Vert \frac{1}{T}\sum_{t\in[T]}\tilde{\Phi}_{it}^{(1)}\left(%
\mathbf{1}\left\{\epsilon_{it}\leq s\right\}-\mathbf{1}\left\{\epsilon_{it}%
\leq s_{m}\right\}\right)- \tilde{\Phi}_{it}^{(1)}\left[ \mathbb{E}\left(%
\mathbf{1}\left\{\epsilon_{it}\leq s\right\}-\mathbf{1}\left\{\epsilon_{it}%
\leq s_{m}\right\}\bigg | \mathscr{D}^{I_{1}}_{e_{i}}\right)\right]%
\right\Vert_{2}  \notag \\
& \leq \max_{i\in I_{2},m \in [n_{\mathbb{S}}]}\frac{2}{T}\sum_{t\in[T]%
}\left\Vert \tilde{\Phi}_{it}^{(1)}\right\Vert_2 \mathbb{E}\left(\mathbf{1}%
\left\{\epsilon_{it}\in \left[s_m - \frac{\varepsilon}{T}, s_m + \frac{%
\varepsilon}{T}\right] \right\}\bigg | \mathscr{D}^{I_{1}}_{e_{i}}\right) 
\notag \\
& + \max_{i\in I_{2},m \in [n_{\mathbb{S}}]} \left\vert \frac{1}{T}%
\sum_{t\in [T]}\left\Vert \tilde{\Phi}_{it}^{(1)} \right\Vert_2 \left[\left(%
\mathbf{1}\left\{\epsilon_{it}\in \left[s_m - \frac{\varepsilon}{T}, s_m + 
\frac{\varepsilon}{T}\right] \right\}\right)- \mathbb{E}\left(\mathbf{1}%
\left\{\epsilon_{it}\in \left[s_m - \frac{\varepsilon}{T}, s_m + \frac{%
\varepsilon}{T}\right] \right\}\bigg | \mathscr{D}^{I_{1}}_{e_{i}}\right) %
\right] \right\vert  \notag \\
&:=\max_{i\in I_{2},m\in[n_{\mathbb{S}}]}\frac{2}{T}\sum_{t\in[T]}\left\Vert 
\tilde{\Phi}_{it}^{(1)}\right\Vert_2 \mathbb{E}\left(\mathbf{1}%
\left\{\epsilon_{it}\in \left[s_m - \frac{\varepsilon}{T}, s_m + \frac{%
\varepsilon}{T}\right] \right\}\bigg | \mathscr{D}^{I_{1}}_{e_{i}}\right)+%
\max_{i\in I_{2},m \in [n_{\mathbb{S}}]}\left\vert
A_{3,i}^{II}(m)\right\vert.
\end{align}
For the first term in (\ref{Lem:A:A3:diff}), we notice that 
\begin{align}  \label{Lem:A:A3_mean}
&\max_{i\in I_{2},m \in [n_{\mathbb{S}}]}\frac{2}{T}\sum_{t\in[T]}\left\Vert 
\tilde{\Phi}_{it}^{(1)}\right\Vert_2 \mathbb{E}\left(\mathbf{1}%
\left\{\epsilon_{it}\in \left[s_m - \frac{\varepsilon}{T}, s_m + \frac{%
\varepsilon}{T}\right] \right\}\bigg | \mathscr{D}^{I_{1}}_{e_{i}}\right) 
\notag \\
&=\max_{i\in I_{2},m \in [n_{\mathbb{S}}]}\frac{2}{T}\sum_{t\in[T]%
}\left\Vert \tilde{\Phi}_{it}^{(1)}\right\Vert_2\left(\mathfrak{F}_{it}(s_m+%
\frac{\varepsilon}{T})-\mathfrak{F}_{it}(s_m-\frac{\varepsilon}{T}) \right) 
\notag \\
&\leq \max_{i\in I_{2}}\frac{2}{T}\sum_{t\in[T]}\left\Vert \tilde{\Phi}%
_{it}^{(1)}\right\Vert_2\frac{2\varepsilon}{T}  \notag \\
&\leq\frac{6M\varepsilon}{c_{\sigma}T}\left(1+\sqrt{p}C\right)\quad\text{%
w.p.a.1},
\end{align}
where the first inequality is by mean-value theorem, and the second
inequality is due to the fact that 
\begin{align*}
\max_{i\in I_{2}}\frac{1}{T}\sum_{t\in[T]}\left\Vert \tilde{\Phi}%
_{it}^{(1)}\right\Vert_2&\leq\max_{i\in I_{2}}\frac{1}{T}\sum_{t\in[T]}\sqrt{%
\frac{4M^{2}}{c_{\sigma}}\left(1+p\max_{j\in[p]}\left\vert
X_{j,it}\right\vert^{2} \right)} \\
&\leq\max_{i\in I_{2}}\frac{2M}{c_{\sigma}}\frac{1}{T}\sum_{t\in[T]}\left(1+%
\sqrt{p}\max_{j\in[p]}\left\vert X_{j,it}\right\vert \right) \\
&\leq\frac{2M\left(1+\sqrt{p}C\right)}{c_{\sigma}}\quad\text{w.p.a.1}.
\end{align*}
For $A_{3,i}^{II}(m)$, let $A_{3,i}^{II}(m)=\frac{1}{T}\sum_{t\in[T]%
}A_{3,it}^{II}(m)$ with 
\begin{align*}
A_{3,it}^{II}(m)=\left\Vert\tilde{\Phi}_{it}^{(1)} \right\Vert_2 \left[\left(%
\mathbf{1}\left\{\epsilon_{it}\in \left[s_m - \frac{\varepsilon}{T}, s_m +%
\frac{\varepsilon}{T}\right] \right\}\right)- \mathbb{E}\left(\mathbf{1}%
\left\{\epsilon_{it}\in \left[s_m - \frac{\varepsilon}{T}, s_m + \frac{%
\varepsilon}{T}\right] \right\}\bigg | \mathscr{D}^{I_{1}}_{e_{i}}\right) %
\right].
\end{align*}
We first observe that 
\begin{align*}
\max_{i\in I_{2},m \in [n_{\mathbb{S}}],t\in[T]}Var\left(A_{3,it}^{II}(m)%
\bigg|\mathscr{D}^{I_{1}}_{e_{i}} \right)&\leq\max_{i \in I_2, t\in[T]%
}\left\Vert\tilde{\Phi}_{it}^{(1)} \right\Vert_2^2\mathbb{E}\left(\mathbf{1}%
\left\{\epsilon_{it}\in \left[s_m - \frac{\varepsilon}{T}, s_m + \frac{%
\varepsilon}{T}\right] \right\}\bigg | \mathscr{D}^{I_{1}}_{e_{i}}\right) \\
&\leq\max_{i \in I_2, t\in[T]}\left\Vert\tilde{\Phi}_{it}^{(1)}\right%
\Vert_2^2 \frac{2\varepsilon}{T}\lesssim\frac{\xi_{N}^{2}\varepsilon}{T},
\quad w.p.a.1,
\end{align*}
where the first inequality is by $Var(x)\leq \mathbb{E}(x^{2})$ for any
random variable $x$ and the second inequality is by the mean-value theorem
and Assumption \ref{ass:1}(v). Similarly, we have 
\begin{align*}
& \max_{i\in I_{2},m \in [n_{\mathbb{S}}], t\in[T]}\sum_{s=t+1}^{T}\left%
\vert Cov\left(A_{3,it}^{II}(m),A_{3,is}^{II}(m)\bigg|\mathscr{D}%
^{I_{1}}_{e_{i}} \right)\right\vert \\
&\leq\max_{i\in I_{2},m \in [n_{\mathbb{S}}],t\in[T]}\sum_{s=t+1}^{T}\mathbb{%
E}\left(\left\vert A_{3,it}^{II}(m)\right\vert^{3}\bigg|\mathscr{D}%
^{I_{1}}_{e_{i}} \right)^{1/3}\mathbb{E}\left(\left\vert
A_{3,is}^{II}(m)\right\vert^{3}\bigg|\mathscr{D}^{I_{1}}_{e_{i}}\right)^{1/3}%
\left[\alpha(t-s) \right]^{1/3} \\
&\lesssim \max_{i\in I_{2},m \in [n_{\mathbb{S}}],t\in[T]}\mathbb{E}%
\left(\left\vert A_{3,it}^{II}(m)\right\vert^{3}\bigg|\mathscr{D}%
^{I_{1}}_{e_{i}} \right)^{2/3}\lesssim\xi_{N}^{2}\left(\frac{\varepsilon}{T}%
\right)^{2/3}, \quad w.p.a.1.
\end{align*}
It follows that 
\begin{align*}
\max_{i\in I_{2},m \in [n_{\mathbb{S}}],t\in[T]}\left\{Var%
\left(A_{3,it}^{II}(m)\bigg|\mathscr{D}^{I_{1}}_{e_{i}}\right)+2%
\sum_{s=t+1}^{T}\left\vert Cov\left(A_{3,it}^{II}(m),A_{3,is}^{II}(m)\bigg|%
\mathscr{D}^{I_{1}}_{e_{i}}\right)\right\vert \right\}\leq
c_{13}\xi_{N}^{2}\left(\frac{\varepsilon}{T}\right)^{2/3}, \quad w.p.a.1.
\end{align*}
and $\max_{i\in I_{2},m \in [n_{\mathbb{S}}],t\in[T]}\left\vert
A_{3,it}^{II}(m)\right\vert\leq C_{13}\xi_{N}$ w.p.a.1 for positive
constants $c_{13}$ and $c_{14}$. Denote the events 
\begin{align*}
& \mathscr{A}_{3,N} = 
\begin{pmatrix}
& \max_{i\in I_{2},m \in [n_{\mathbb{S}}],t\in[T]}\left\{Var%
\left(A_{3,it}^{II}(m)\bigg|\mathscr{D}^{I_{1}}_{e_{i}}\right)+2%
\sum_{s=t+1}^{T}\left\vert Cov\left(A_{3,it}^{II}(m),A_{3,is}^{II}(m)\bigg|%
\mathscr{D}^{I_{1}}_{e_{i}}\right)\right\vert \right\}\leq
c_{13}\xi_{N}^{2}\left(\frac{\varepsilon}{T}\right)^{2/3} \\ 
& \text{and} \quad \max_{i\in I_{2},t\in[T]}\left\Vert \tilde{\Phi}%
_{it}^{(1)}\right\Vert_2\leq c_{14}\xi_{N}%
\end{pmatrix}
\\
& \mathscr{A}_{3,N,i} = 
\begin{pmatrix}
& \max_{m \in [n_{\mathbb{S}}],t\in[T]}\left\{Var\left(A_{3,it}^{II}(m)\bigg|%
\mathscr{D}^{I_{1}}_{e_{i}} \right)+2\sum_{s=t+1}^{T}\left\vert
Cov\left(A_{3,it}^{II}(m),A_{3,is}^{II}(m)\bigg|\mathscr{D}%
^{I_{1}}_{e_{i}}\right)\right\vert \right\}\leq c_{13}\xi_{N}^{2}\left(\frac{%
\varepsilon}{T}\right)^{2/3} \\ 
& \text{and} \quad \max_{t\in[T]}\left\Vert \tilde{\Phi}_{it}^{(1)}\right%
\Vert_2\leq c_{14}\xi_{N}%
\end{pmatrix}%
.
\end{align*}
Then, we have $\mathbb{P}(\mathscr{A}_{3,N}^c) \rightarrow 0$ and 
\begin{align*}
&\mathbb{P}\left(\max_{i\in I_{2},m \in [n_{\mathbb{S}}]}\left\vert
A_{3,i}^{II}(m)\right\vert >c_{12}\eta_{N}\right) \\
& \leq \mathbb{P}\left(\max_{i\in I_{2},m \in [n_{\mathbb{S}}]}\left\vert
A_{3,i}^{II}(m)\right\vert >c_{12}\eta_{N}, \mathscr{A}_{3,N}\right) + 
\mathbb{P}(\mathscr{A}_{3,N}^c) \\
& \leq \sum_{i \in [I_2]}\mathbb{P}\left(\max_{m \in [n_{\mathbb{S}%
}]}\left\vert A_{3,i}^{II}(m)\right\vert >c_{12}\eta_{N}, \mathscr{A}%
_{3,N}\right) + \mathbb{P}(\mathscr{A}_{3,N}^c) \\
& \leq \sum_{i \in [I_2]}\mathbb{P}\left(\max_{m \in [n_{\mathbb{S}%
}]}\left\vert A_{3,i}^{II}(m)\right\vert >c_{12}\eta_{N}, \mathscr{A}%
_{3,N,i}\right) + \mathbb{P}(\mathscr{A}_{3,N}^c) \\
& = \sum_{i \in [I_2]}\mathbb{E}\mathbb{P}\left(\max_{m \in [n_{\mathbb{S}%
}]}\left\vert A_{3,i}^{II}(m)\right\vert >c_{12}\eta_{N}\bigg| \mathscr{D}%
^{I_{1}}_{e_{i}}\right)\mathbf{1}\{\mathscr{A}_{3,N,i}\} + \mathbb{P}(%
\mathscr{A}_{3,N}^c) \\
&\leq\sum_{i\in I_{2},m \in [n_{\mathbb{S}}]}\exp\left(-\frac{%
c_{9}c_{12}^{2}T^{2}\eta_{N}^{2}}{c_{13}T\xi_{N}^{2}\left(\frac{\varepsilon}{%
T}\right)^{2/3}+c_{14}^{2}\xi_{N}^{2}+c_{12}c_{14}T\eta_{N}\xi_{N,1}\left(%
\log T\right)^{2}} \right)+o(1)=o(1),
\end{align*}
where second inequality is by the union bound, the first inequality is by $%
\mathscr{A}_{3,N} \subset \mathscr{A}_{3,N,i}$, the equality is by the fact
that $\mathscr{A}_{3,N,i}$ is $\mathscr{D}^{I_{1}}_{e_{i}}$ measurable, and
the last inequality is by Lemma \ref{Lem:Bern}(ii), the definition of $%
\mathscr{A}_{3,N,i}$, and the fact that $\{\epsilon_{it}\}_{t \in [T]}$ is
strong mixing given $\mathscr{D}^{I_{1}}_{e_{i}}$. This implies 
\begin{align*}
\max_{i\in I_{2},m \in [n_{\mathbb{S}}]}\sup_{s \in \mathbb{S}%
_{m}}\left\Vert \frac{1}{T}\sum_{t\in[T]} \left[A_{3,it}^{I}\left(s%
\right)-A_{3,it}^{I}\left(s_{m}\right) \right]\right\Vert_{2} =
O_p(\eta_{N}).
\end{align*}

Last, we turn to the first term of \eqref{B.11}. Denote $A_{3,i}^{I,k}(s_m)$
as the $k^{th}$ element of $A_{3,i}^{I,k}(s_m)$ and the event set $%
\mathscr{A}_{4,N,i} = \{\max_{t\in[T]}\left\Vert \tilde{\Phi}%
_{it}^{(1)}\right\Vert_2\leq c_{14}\xi_{N}\}$. Similarly, we have $\mathbb{P}%
(\cap_{i \in I_2}\mathscr{A}_{4,N,i}^c) = \mathbb{P}(\max_{i \in I_2, t\in[T]%
}\left\Vert \tilde{\Phi}_{it}^{(1)}\right\Vert_2 > c_{14}\xi_{N} )= o(1)$.
Following the same argument as above, we have 
\begin{align}  \label{B.12}
&\mathbb{P}\left(\max_{i\in I_{2},m \in [n_{\mathbb{S}}]} \left\vert
A_{3,i}^{I,k}(s_m)\right\vert>c_{12}\eta_{N}\right)  \notag \\
&\leq\sum_{m\in [n_{\mathbb{S}}],i\in I_2} \mathbb{E}\mathbb{P}%
\left(\left\vert A_{3,i}^{I,k}(s_{m})\right\vert>\frac{c_{12}}{2}\eta_{N}%
\bigg | \mathscr{D}^{I_{1}}_{e_{i}}\right)1\{\mathscr{A}_{4,N,i}\}+o(1) 
\notag \\
&\leq\sum_{m\in [n_{\mathbb{S}}],i\in I_2}\exp\left(-\frac{%
c_{9}c_{12}^{2}T^{2}\eta_{N}^{2}/4}{c_{14}^{2}T\xi_{N}^{2}+c_{12}c_{14}T%
\eta_{N}\xi_{N}\log T\left(\log\log T\right)/2 }\right)+o(1)=o(1),
\end{align}
where the second inequality combines Lemma \ref{Lem:Bern}(i). Combining (\ref%
{A3}), (\ref{B.10}) and (\ref{B.12}), we have 
\begin{align*}
\max_{i\in I_{2}} \frac{\left\vert A_{3,i}\right\vert}{\left\Vert\dot{\Delta}%
_{i,u} \right\Vert_{2}}=O_{p}(\eta_{N}).
\end{align*}

We now turn to $A_{2,i}$. Let $A_{2,it}^{I}=\left(\tau -\mathbf{1}%
\left\{w_{1,it}\leq 0\right\} \right) \tilde{\Phi}_{it}^{(1)}$ $-\mathbb{E}%
\left(\left(\tau -\mathbf{1}\left\{ w_{1,it}\leq 0\right\} \right) \tilde{\Phi}_{it}^{(1)}\bigg |\mathscr{D}_{e_{i}}^{I_{1}}\right) .$ Then $A_{2,i}=%
\frac{1}{T}\sum_{t=1}^{T}\left(A_{2,it}^{I}\right) ^{\prime}\dot{\Delta}%
_{i,u}.$ By conditional Bernstein's inequality and similarly as (\ref{B.12}%
), we can show that 
\begin{equation*}
\mathbb{P}\left\{ \max_{i\in I_{2}}\left\Vert \frac{1}{T}%
\sum_{t=1}^{T}A_{2,it}^{I}\right\Vert_{2}\geq c_{12}\eta_{N}\right\}=o(1),
\end{equation*}%
which implies $\max_{i\in I_{2}}\frac{\left\vert A_{2,i}\right\vert }{%
\left\Vert \dot{\Delta}_{i,u}\right\Vert_{2}}=O_{p}\left(\eta_{N}\right) .$
By similar arguments for $A_{3,i}$, we can also show that $\max_{i\in I_{2}}%
\frac{\left\vert A_{5,i}\right\vert }{\left\Vert \dot{\Delta}%
_{i,u}\right\Vert_{2}}=O_{p}\left(\eta_{N}\right) .$ For $A_{6,i}$, we note
that 
\begin{align*}
\max_{i\in I_{2}}\frac{\left\vert A_{6,i}\right\vert }{\left\Vert \dot{\Delta%
}_{i,u}\right\Vert_{2}}& =\max_{i\in I_{2}}\frac{\left\vert \frac{1}{T}%
\sum_{t=1}^{T}\int_{0}^{\tilde{\Phi}_{it}^{(1)\prime}\dot{\Delta}_{i,u}}%
\mathbb{E}\left(\mathbf{1}\left\{ w_{1,it}\leq s\right\} -\mathbf{1}%
\left\{\epsilon_{it}\leq s\right\} \bigg |\mathscr{D}_{e_{i}}^{I_{1}}%
\right)ds\right\vert }{\left\Vert \dot{\Delta}_{i,u}\right\Vert_{2}} \\
& =\max_{i\in I_{2}}\frac{\left\vert \frac{1}{T}\sum_{t=1}^{T}\int_{0}^{%
\tilde{\Phi}_{it}^{(1)\prime}\dot{\Delta}_{i,u}}\left[ \mathfrak{F}%
_{it}\left(u_{i}^{0\prime}\left(\tilde{\Phi}_{it}^{(1)}-\Phi_{it}^{0}\right)
+s\right) -\mathfrak{F}_{it}(s)\right] ds\right\vert }{\left\Vert \dot{\Delta%
}_{i,u}\right\Vert_{2}} \\
& \leq \max_{i\in I_{2}}\frac{\left\vert \frac{1}{T}\sum_{t=1}^{T}\int_{0}^{%
\tilde{\Phi}_{it}^{(1)\prime}\dot{\Delta}_{i,u}}u_{i}^{0\prime}\left(\tilde{\Phi}_{it}^{(1)}-\Phi_{it}^{0}\right) \mathfrak{f}_{it}(s)ds\right\vert }{%
\left\Vert \dot{\Delta}_{i,u}\right\Vert_{2}} \\
& \leq \max_{i\in I_{2}}\frac{\frac{\bar{\mathfrak{f}}}{T}%
\sum_{t=1}^{T}\left\vert \tilde{\Phi}_{it}^{(1)\prime}\dot{\Delta}%
_{i,u}\right\vert \left\vert u_{i}^{0\prime}\left(\tilde{\Phi}_{it}^{(1)}-\Phi_{it}^{0}\right) \right\vert }{\left\Vert \dot{\Delta}%
_{i,u}\right\Vert_{2}} \\
& \leq \bar{\mathfrak{f}}\max_{i\in I_{2}}\sqrt{\frac{1}{T}\sum_{t\in[T]%
}\left\Vert \tilde{\Phi}_{it}^{(1)}\right\Vert ^{2}}\sqrt{\frac{1}{T}\sum_{t%
\in[T]}\left\Vert \tilde{\Phi}_{it}^{(1)}-\Phi_{it}^{0}\right\Vert ^{2}}%
\left\Vert u_{i}^{0}\right\Vert_{2} \\
& =O_{p}\left(\eta_{N}\right) ,
\end{align*}%
where the first inequality is by mean-value theorem and the other
inequalities holds by similar reasons as those in (\ref{Lem:A:A1}).

Last, for $A_{7,i}$, we have 
\begin{align*}
&\frac{1}{T}\sum_{t\in[T]}\left[\tilde{\Phi}_{it}^{(1)\prime}\dot{\Delta}%
_{i,u}\left(\mathbf{1}\left\{\epsilon_{it}\leq 0 \right\}-\mathbf{1}%
\left\{w_{1,it}\leq 0 \right\} \right) \right] \\
&=\frac{1}{T}\sum_{t\in[T]}\mathbb{E}\left[\tilde{\Phi}_{it}^{(1)\prime}\dot{%
\Delta}_{i,u}\left(\mathbf{1}\left\{\epsilon_{it}\leq 0 \right\}-\mathbf{1}%
\left\{w_{1,it}\leq 0 \right\} \right)\bigg|\mathscr{D}^{I_{1}}_{e_{i}} %
\right] \\
&+\frac{1}{T}\sum_{t\in[T]}\left\{\left[\tilde{\Phi}_{it}^{(1)\prime}\dot{%
\Delta}_{i,u}\left(\mathbf{1}\left\{\epsilon_{it}\leq 0 \right\}-\mathbf{1}%
\left\{w_{1,it}\leq 0 \right\} \right) \right]-\mathbb{E}\left[\tilde{\Phi}%
_{it}^{(1)\prime}\dot{\Delta}_{i,u}\left(\mathbf{1}\left\{\epsilon_{it}\leq
0\right\}-\mathbf{1}\left\{w_{1,it}\leq 0 \right\} \right)\bigg|\mathscr{D}%
^{I_{1}}_{e_{i}} \right]\right\} \\
&:=A_{7,i}^{I}+A_{7,i}^{II},
\end{align*}
Then, following the same analyses of $A_{1,i}$ and $A_{3,i}$, we have 
\begin{align*}
\max_{i\in I_{2}}\frac{\left\vert A_{7,i}^{I}\right\vert}{ \left\Vert \dot{%
\Delta}_{i,u}\right\Vert_2 } =O_{p}(\eta_{N}) \quad \text{and}\quad
\max_{i\in I_{2}}\frac{\left\vert A_{7,i}^{II} \right\vert}{\left\Vert \dot{%
\Delta}_{i,u}\right\Vert_2 }=O_{p}(\eta_{N}),
\end{align*}
which implies $\max_{i\in I_{2}} \frac{\left\vert A_{7,i} \right\vert}{%
\left\Vert \dot{\Delta}_{i,u}\right\Vert_2 }=O_{p}(\eta_{N})$.
\end{proof}

Recall that event $\mathscr{A}_{4}=\left\{\max_{i\in I_{2}}\left\Vert
O_{j}^{(1)\prime}\dot{u}_{i,j}^{(1)}-u_{i,j}^{0}\right\Vert_{2}=O\left(%
\eta_{N}\right),\forall j\in[p]\cup\{0\}\right\}$ where we define $%
q_{i}^{III}=\inf_{\Delta}\frac{\left[\frac{1}{N_{2}}\sum_{i\in I_{2}}\left(%
\tilde{\Psi}_{it}^{(1)\prime}\Delta\right)^{2}\right]^{\frac{3}{2}}}{\frac{1%
}{N_{2}}\sum_{i\in I_{2}}\left\vert\tilde{\Psi}_{it}^{(1)\prime}\Delta
\right\vert^{3}}$.

\begin{lemma}
\label{Lem:B}. Suppose Assumptions \ref{ass:1}--\ref{ass:5} hold. Then for $%
\{B_{1,t},\cdots ,B_{6,t}\}_{t\in[T]}$ defined in \eqref{B:obj diff}, for
any constant $0<c_{11}<\min (\frac{3\underline{\mathfrak{f}}}{\bar{\mathfrak{%
f}}^{\prime }},1)$, we have 
\begin{align*}
& \max_{t\in[T]}\frac{|B_{m,t}|}{\left\Vert \dot{\Delta}_{t,v}\right\Vert_{2}%
}=O_{p}\left(\eta_{N}\right) \quad \forall m\in\left\{ 1,2,3,5,6\right\} , \\
& |B_{4,t}|\geq \min \left(\frac{\left(3c_{11}^{2}\underline{\mathfrak{f}}%
-c_{11}^{3}\bar{\mathfrak{f}}^{\prime }\right) c_{\psi }\left\Vert \dot{%
\Delta}_{t,v}\right\Vert_{2}^{2}}{12},\frac{\left(3c_{11}^{2}\underline{%
\mathfrak{f}}-c_{11}^{3}\bar{\mathfrak{f}}^{\prime }\right) \sqrt{c_{\psi }}%
q_{i}^{III}\left\Vert \dot{\Delta}_{t,v}\right\Vert_{2}}{6\sqrt{2}}%
\right),\quad \forall t\in[T].
\end{align*}
\end{lemma}

\begin{proof}
We first deal with $B_{4,t}$. Let $w_{4,it}=\dot{\Psi}_{it}^{(1)\prime}\dot{%
\Delta}_{t,v}$. Following the same arguments as used to derive the lower
bound for $A_{4,i}$ in the proof of Lemma \ref{Lem:A} by replacing $w_{2,it}$
with $w_{4,it}$, we can show that, for $t\in[T]$ and any constants $%
c_{11}<\min (\frac{3\underline{\mathfrak{f}}}{\bar{\mathfrak{f}}^{\prime }}%
,1)$, 
\begin{equation*}
\left\vert B_{4,t}\right\vert \geq \min \left(\frac{\left(3c_{11}^{2}%
\underline{\mathfrak{f}}-c_{11}^{3}\bar{\mathfrak{f}}^{\prime
}\right)c_{\psi }\left\Vert \dot{\Delta}_{t,v}\right\Vert_{2}^{2}}{12},\frac{%
\left(3c_{11}^{2}\underline{\mathfrak{f}}-c_{11}^{3}\bar{\mathfrak{f}}%
^{\prime}\right) \sqrt{c_{\psi }}q_{i}^{III}\left\Vert \dot{\Delta}%
_{t,v}\right\Vert_{2}}{6\sqrt{2}}\right) .
\end{equation*}

For $B_{1,t}$, we have 
\begin{equation*}
B_{1,t}=\left(\frac{1}{N_{2}}\sum_{i\in I_{2}}\Psi_{it}^{0}\left(\tau -%
\mathbf{1}\left\{ \epsilon_{it}\leq 0\right\} \right) \right) ^{\prime}\dot{%
\Delta}_{t,v}:=\left(\frac{1}{N_{2}}\sum_{i\in
I_{2}}B_{1,it}^{I}\right)^{\prime}\dot{\Delta}_{t,v}.
\end{equation*}%
Conditional on the fixed effects $\left\{ V_{j}^{0}\right\}_{j\in [p]\cup
\{0\}}$ and $\left\{ W_{j}^{0}\right\}_{j\in [p]}$, the randomness in $%
B_{1,t}$ is from $\left\{ \epsilon_{it}\right\}_{i\in I_{2},t\in[T]}$ and $%
\left\{ e_{j,it}\right\}_{j\in [p],i\in I_{2},t\in[T]}$, which are
independent across $i$. Owing to this, by conditional Hoeffding's
inequality, we can show that 
\begin{align}
& \mathbb{P}\left\{ \max_{t\in[T]}\frac{1}{N_{2}}\left\Vert \sum_{i\in
I_{2}}B_{1,it}^{I,k}\right\Vert_{2}>c_{15}\eta_{N}\bigg|\mathscr{D}\right\}
\leq \sum_{t\in[T]}\mathbb{P}\left\{ \left\Vert\sum_{i\in
I_{2}}B_{1,it}^{I,k}\right\Vert_{2}>c_{15}N_{2}\eta_{N}\bigg|\mathscr{D}%
\right\}  \notag  \label{Lem:B:B1} \\
& \leq 2\sum_{t\in[T]}\exp \left(-\frac{c_{15}^{2}N_{2}^{2}\eta_{N}^{2}}{%
4M^{2}\xi_{N}^{2}N_{2}}\right) =o(1),
\end{align}%
with $B_{1,it}^{I,k}$ being the $k^{th}$ element in $B_{1,it}^{I}$, $c_{15}$
is a positive constant, where the second inequality is by Hoeffding's
inequality with the fact that $\max_{i\in I_{2},t\in[T]}\left\vert
B_{1,it}^{I,k}\right\vert \leq M\xi_{N}$ a.s. by Assumption \ref{ass:1}(v)
and Lemma \ref{Lem:bounded u&v_tilde}(i). It follows that $\max_{t\in[T]}%
\frac{\left\vert B_{1,t}\right\vert }{\left\Vert \dot{\Delta}%
_{t,v}\right\Vert_{2}}=O_{p}\left(\eta_{N}\right).$ If we use the
conditional Bernstein's inequality for the independent sequence rather than
the Hoeffding's inequality above, we can show that $\max_{t\in[T]}\frac{1}{%
N_{2}}\left\Vert \sum_{i\in I_{2}}B_{1,it}^{I,k}\right\Vert_{2}=O_{p}\left(%
\sqrt{\frac{\log N\vee T}{N}} \right)$, but here we only need to show the
uniform convergence rate to be $\eta_{N}$.

Let $X_{0,it}=1.$ As for $B_{2,t}$, note that 
\begin{align}
\max_{t\in[T]}\frac{\left\vert B_{2,t}\right\vert }{\left\Vert \dot{\Delta}%
_{t,v}\right\Vert_{2}}& =\max_{t\in[T]}\frac{\left\vert\left[ \frac{1}{N_{2}}%
\sum_{i\in I_{2}}\left(\dot{\Psi}_{it}^{(1)}-\Psi_{it}^{0}\right) \left(\tau
-\mathbf{1}\left\{ \epsilon_{it}\leq 0\right\}\right) \right] ^{\prime}\dot{%
\Delta}_{t,v}\right\vert }{\left\Vert \dot{\Delta}_{t,v}\right\Vert_{2}} 
\notag \\
& \leq \max_{t\in[T]}\left\Vert \frac{1}{N_{2}}\sum_{i\in I_{2}}\left(\dot{%
\psi}_{it}^{(1)}-\Psi_{it}^{0}\right) \left(\tau -\mathbf{1}\left\{
\epsilon_{it}\leq 0\right\} \right) \right\Vert_{2}  \notag \\
& \leq \max_{t\in[T]}\frac{1}{N_{2}}\sum_{i\in I_{2}}\left\Vert \dot{\Psi}%
_{it}^{(1)}-\Psi_{it}^{0}\right\Vert_{2}=\max_{t\in[T]}\frac{1}{N_{2}}%
\sum_{i\in I_{2}}\sqrt{\sum_{j=0}^{p}\left\Vert O_{0}^{(1)\prime}\dot{u}%
_{i,j}^{(1)}-u_{i,j}^{0}\right\Vert_{2}^{2}X_{j,it}^{2}}  \notag \\
& \leq \max_{i\in I_{2}}\max_{j\in [p]\cup \{0\}}\left(||O_{0}^{(1)\prime}%
\dot{u}_{i,j}^{(1)}-u_{i,j}^{0}||_{2}\right) \left[ \frac{1}{N_{2}}%
\sum_{i\in I_{2}}\sum_{j=0}^{p}X_{j,it}^{2}\right] ^{1/2}=O_{p}(\eta_{N}),
\label{Lem:B:B2}
\end{align}%
where the first inequality is by Cauchy's inequality, the second inequality
is by Jensen's inequality, and the last equality is by Theorem \ref{Thm2}(i)
and Assumption \ref{ass:1}(iv).

Next, we deal with $B_{3,t}$. 
Following similar arguments as used in (\ref{Lem:B:B2}), we obtain that 
\begin{equation*}
\max_{t\in[T]}\frac{\left\vert \frac{1}{N_{2}}\sum_{i\in I_{2}}\left(\dot{\Psi}_{it}^{(1)}-\Psi_{it}^{0}\right) ^{\prime}\dot{\Delta}_{t,v}\left[ 
\mathbf{1}\left\{ \epsilon_{it}\leq 0\right\} -\mathbf{1}\left\{
w_{3,it}\leq 0\right\} \right] \right\vert }{\left\Vert \dot{\Delta}%
_{t,v}\right\Vert_{2}}=O_{p}\left(\eta_{N}\right) ,
\end{equation*}%
which implies 
\begin{align*}
\max_{t\in[T]}\frac{\left\vert B_{3,t}\right\vert }{2\left\Vert \dot{\Delta}%
_{t,v}\right\Vert_{2}}& =\max_{t\in[T]}\frac{\left\vert\left\{ \frac{1}{N_{2}%
}\sum_{i\in I_{2}}\Psi_{it}^{0}\left[ \mathbf{1}\left\{ \epsilon_{it}\leq
0\right\} -\mathbf{1}\left\{ w_{3,it}\leq 0\right\} \right] \right\}
^{\prime}\dot{\Delta}_{t,v}\right\vert }{\left\Vert \dot{\Delta}%
_{t,v}\right\Vert_{2}}+O_{p}\left(\eta_{N}\right) \\
& \leq \max_{t\in[T]}\left\Vert \frac{1}{N_{2}}\sum_{i\in I_{2}}\Psi_{it}^{0}%
\left[ \mathbf{1}\left\{ \epsilon_{it}\leq 0\right\} -\mathbf{1}\left\{
w_{3,it}\leq 0\right\} \right] \right\Vert_{2}+O_{p}\left(\eta_{N}\right) \\
& \leq \max_{t\in[T]}\frac{1}{N_{2}}\sum_{i\in I_{2}}\left\Vert
\Psi_{it}^{0}\right\Vert_{2}\mathbf{1}\left\{ |\epsilon_{it}|\leq \sum_{j\in[%
p]\cup \{0\}}\left\Vert v_{t,j}^{0}\right\Vert_{2}\left\Vert \dot{u}%
_{i,j}^{(1)}-O_{j}^{(1)}u_{i,j}^{0}\right\Vert_{2}\right\}
+O_{p}\left(\eta_{N}\right) \\
& \equiv \Xi_{NT}+O_{p}\left(\eta_{N}\right) .
\end{align*}%
Now define the event set $\mathscr{B}_{N,1}(M)=\{\max_{i\in I_{2},j\in[p]%
\cup \{0\}}||\dot{u}_{i,j}^{(1)}-O_{j}^{(1)}u_{i,j}^{0}||\leq M\eta_{N}\}$.
Then, Theorem \ref{Thm2}(i) implies, for any $e>0$, there is a sufficiently
large $M$ such that $\mathbb{P}(\mathscr{B}_{N,1}^{c}(M))\leq e$. Recall
that $\mathscr{D}_{e_{it}}$ is the $\sigma$-field generated by $e_{it}\cup
\left\{ V_{j}^{0}\right\}_{j\in [p]\cup \{0\}}\cup \left\{
W_{j}^{0}\right\}_{j\in [p]}$. Then, we have 
\begin{align}
\mathbb{P}\left(\Xi_{NT}\geq C\eta_{N}\right) & \leq \mathbb{P}%
\left(\Xi_{NT}\geq C\eta_{N},\mathscr{B}_{N,1}\right) +e  \notag \\
& \leq \mathbb{P}\left(\frac{1}{N_{2}}\sum_{i\in I_{2}}\left\Vert
\Psi_{it}^{0}\right\Vert_{2}\mathbf{1}\left\{ |\epsilon_{it}|\leq
M\eta_{N}\sum_{j\in [p]\cup \{0\}}\left\Vert
v_{t,j}^{0}\right\Vert_{2}\right\} \geq C\eta_{N}\right) +e  \notag \\
& \leq \mathbb{P}\left(\max_{t\in[T]}B_{3,t}^{I}\geq C\eta_{N}+\max_{t\in
T}B_{3,t}^{II}\right) +e,  \label{Lem:B:B3_2}
\end{align}%
where 
\begin{align*}
B_{3,t}^{I}& =\frac{1}{N_{2}}\sum_{i\in I_{2}}\left[ \left\Vert
\Psi_{it}^{0}\right\Vert_{2}\mathbf{1}\left\{ |\epsilon_{it}|\leq
M\eta_{N}\sum_{j\in [p]\cup \{0\}}\left\Vert
v_{t,j}^{0}\right\Vert_{2}\right\} -B_{3,t}^{II}\right] \quad \text{and} \\
B_{3,t}^{II}& =\frac{1}{N_{2}}\sum_{i\in I_{2}}\mathbb{E}\left(\left\Vert%
\Psi_{it}^{0}\right\Vert_{2}\mathbf{1}\left\{ |\epsilon_{it}|\leq
M\eta_{N}\sum_{j\in [p]\cup \{0\}}\left\Vert
v_{t,j}^{0}\right\Vert_{2}\right\} \bigg|\{\mathscr{D}_{e_{it}}\}_{i\in
I_{2}}\right) \\
& =\frac{1}{N_{2}}\sum_{i\in I_{2}}\mathbb{E}\left(\left\Vert
\Psi_{it}^{0}\right\Vert_{2}\mathbf{1}\left\{ |\epsilon_{it}|\leq
M\eta_{N}\sum_{j\in [p]\cup \{0\}}\left\Vert
v_{t,j}^{0}\right\Vert_{2}\right\} \bigg|\mathscr{D}_{e_{it}}\right) ,
\end{align*}%
and the second equality for $B_{3,t}^{II}$ holds by Assumption \ref{ass:1}%
(i). Following this, we show that 
\begin{align}
\max_{t\in[T]}B_{3,t}^{II}& =\max_{t\in[T]}\frac{1}{N_{2}}\sum_{i\in
I_{2}}\left\Vert \Psi_{it}^{0}\right\Vert_{2}\left[ \mathfrak{F}%
_{it}\left(M\eta_{N}\sum_{j\in [p]\cup \{0\}}\left\Vert
v_{t,j}^{0}\right\Vert_{2}\right) -\mathfrak{F}_{it}\left(-M\eta_{N}\sum_{j%
\in [p]\cup \{0\}}\left\Vert v_{t,j}^{0}\right\Vert_{2}\right) \right] 
\notag  \label{Lem:B:B3_3} \\
& \leq \max_{t\in[T]}\frac{1}{N_{2}}\sum_{i\in I_{2}}\sqrt{1+p\max_{j\in
[p]}\left\vert X_{j,it}\right\vert ^{2}}\frac{2(1+p)M^{2}}{c_{\sigma }}%
\eta_{N}  \notag \\
& \leq \frac{2(1+p)^{2}M^{2}C}{c_{\sigma }}\eta_{N}\quad \text{a.s.},
\end{align}%
where $\mathfrak{F}_{it}(\cdot)$ is the conditional CDF of $\epsilon_{it}$
given $\mathscr{D}_{e_{i}}$ who also has bounded PDF by Assumption \ref%
{ass:1}(vii), the first inequality is by mean-value theorem and facts that $%
\left\Vert \Psi_{it}^{0}\right\Vert_{2}^{2}\leq M^{2}\left(1+p\max_{j\in[p]%
}\left\vert X_{j,it}\right\vert ^{2}\right) $ together with Lemma \ref%
{Lem:bounded u&v_tilde}(i), and the second inequality is due to Assumption %
\ref{ass:1}(iv). In addition, given $\{\mathscr{D}_{e_{it}}\}_{i\in I_{2}}$, 
$\{\epsilon_{it}\}_{i\in I_{2}}$ are still independent across $i$. 
Therefore, by Hoeffding's inequality and similar arguments for term $B_{1,t}$
in (\ref{Lem:B:B1}), we can show that 
\begin{align}
& \mathbb{P}\left(\max_{t\in[T]}\left\Vert
B_{3,t}^{I}\right\Vert_{2}>c_{12}\eta_{N}\right) \leq \sum_{t\in[T]}\mathbb{P%
}\left(\left\Vert B_{3,t}^{I}\right\Vert_{2}>c_{12}\eta_{N}\right)  \notag
\label{Lem:B:B3_4} \\
& =\sum_{t\in[T]}\mathbb{E}\mathbb{P}\left(\left\Vert
B_{3,t}^{I}\right\Vert_{2}>c_{12}\eta_{N}\bigg|\{\mathscr{D}%
_{e_{it}}\}_{i\in I_{2}}\right) =o(1).
\end{align}%
Combining (\ref{Lem:B:B3_2})-(\ref{Lem:B:B3_4}), we obtain that $\max_{t\in[T%
]}\frac{\left\vert B_{3,t}\right\vert }{\left\Vert \dot{\Delta}%
_{t,v}\right\Vert_{2}}=O_{p}(\eta_{N}).$

For $B_{5,t}$, we observe that 
\begin{align*}
B_{5,t}& =\frac{1}{N_{2}}\sum_{i\in I_{2}}\left\{
\int_{0}^{\Psi_{it}^{0\prime}\dot{\Delta}_{t,v}}\left[ \left(\mathbf{1}%
\left\{ \epsilon_{it}\leq s\right\} -\mathbf{1}\left\{ \epsilon_{it}\leq
0\right\} \right) -\mathbb{E}\left(\left(\mathbf{1}\left\{ \epsilon_{it}\leq
s\right\} -\mathbf{1}\left\{ \epsilon_{it}\leq 0\right\} \right) \bigg |%
\mathscr{D}_{e_{it}}\right) \right] ds\right\} \\
& +\frac{1}{N_{2}}\sum_{i\in I_{2}}\left\{ \int_{\Psi_{it}^{0\prime}\dot{%
\Delta}_{t,v}}^{\dot{\Psi}_{it}^{(1)\prime}\dot{\Delta}_{t,v}}\left[ \left(%
\mathbf{1}\left\{ \epsilon_{it}\leq s\right\} -\mathbf{1}\left\{
\epsilon_{it}\leq 0\right\} \right) -\mathbb{E}\left(\left(\mathbf{1}%
\left\{\epsilon_{it}\leq s\right\} -\mathbf{1}\left\{ \epsilon_{it}\leq
0\right\}\right) \bigg |\mathscr{D}_{e_{it}}\right) \right] ds\right\} \\
& :=B_{5,t}^{I}+B_{5,t}^{II}.
\end{align*}%
Following similar arguments for $A_{3,i}$ with Bernstein's inequality
replacing by Hoeffding's inequality, we can show that $\max_{t\in[T]}\frac{%
\left\vert B_{5,t}^{I}\right\vert }{\left\Vert \dot{\Delta}%
_{t,v}\right\Vert_{2}}=O_{p}(\eta_{N})$. Besides, we obtain that 
\begin{equation*}
\max_{t\in[T]}\frac{\left\vert B_{5,t}^{II}\right\vert }{\left\Vert \dot{%
\Delta}_{t,v}\right\Vert_{2}}\leq \max_{t\in[T]}\frac{4}{N_{2}}\sum_{i\in
I_{2}}\left\Vert \dot{\Psi}_{it}^{(1)}-\Psi_{it}^{0}\right\Vert_{2}=O_{p}(%
\eta_{N}),
\end{equation*}
which implies $\max_{t\in[T]}\frac{\left\vert B_{5,t}\right\vert }{%
\left\Vert \dot{\Delta}_{t,v}\right\Vert_{2}}=O_{p}\left(\eta_{N}\right).$

For $B_{6,t}$, we first note that 
\begin{equation*}
\max_{t\in[T]}\frac{\left\vert \frac{1}{N_{2}}\sum_{i\in
I_{2}}\int_{\Psi_{it}^{0\prime}\dot{\Delta}_{t,v}}^{\dot{\Psi}%
_{it}^{(1)\prime}\dot{\Delta}_{t,v}}\left[ \mathbf{1}\left\{
\epsilon_{it}\leq s\right\} -\mathbf{1}\left\{ w_{3,it}\leq s\right\} \right]
\right\vert }{\left\Vert\dot{\Delta}_{t,v}\right\Vert_{2}}\leq \max_{t\in[T]}%
\frac{2}{N_{2}}\sum_{i\in I_{2}}\left\Vert \dot{\Psi}_{it}^{(1)}-%
\Psi_{it}^{0}\right\Vert_{2}=O_{p}(\eta_{N}),
\end{equation*}%
and this implies 
\begin{align*}
\max_{t\in[T]}\frac{\left\vert B_{6,t}\right\vert }{\left\Vert \dot{\Delta}%
_{t,v}\right\Vert_{2}}& =\max_{t\in[T]}\frac{\left\vert \frac{1}{N_{2}}%
\sum_{i\in I_{2}}\int_{0}^{\dot{\Psi}_{it}^{(1)\prime}\dot{\Delta}%
_{t,v}}\left(\mathbf{1}\left\{ w_{3,it}\leq s\right\} -\mathbf{1}\left\{
\epsilon_{it}\leq s\right\} \right) ds\right\vert }{\left\Vert \dot{\Delta}%
_{t,v}\right\Vert_{2}} \\
& =\max_{t\in[T]}\frac{\left\vert \frac{1}{N_{2}}\sum_{i\in
I_{2}}\int_{0}^{\Psi_{it}^{0\prime}\dot{\Delta}_{t,v}}\left(\mathbf{1}%
\left\{ w_{3,it}\leq s\right\} -\mathbf{1}\left\{ \epsilon_{it}\leq
s\right\} \right) ds\right\vert }{\left\Vert \dot{\Delta}_{t,v}\right%
\Vert_{2}}+O_{p}\left(\eta_{N}\right) .
\end{align*}%
In addition, for the first term on the RHS of the last equality, we have 
\begin{align*}
& \max_{t\in[T]}\frac{\left\vert \frac{1}{N_{2}}\sum_{i\in
I_{2}}\int_{0}^{\Psi_{it}^{0\prime}\dot{\Delta}_{t,v}}\left(\mathbf{1}%
\left\{ w_{3,it}\leq s\right\} -\mathbf{1}\left\{ \epsilon_{it}\leq
s\right\} \right) ds\right\vert }{\left\Vert \dot{\Delta}_{t,v}\right%
\Vert_{2}} \\
& \leq \max_{t\in[T]}\frac{\frac{1}{N_{2}}\sum_{i\in
I_{2}}\int_{0}^{\left\Vert \Psi_{it}^{0}\right\Vert_{2}\left\Vert \dot{\Delta%
}_{t,v}\right\Vert_{2}}\left(\mathbf{1}\left\{ |\epsilon_{it}-s|\leq
\sum_{j\in [p]\cup \{0\}}\left\Vert v_{t,j}^{0}\right\Vert_{2}\left\Vert 
\dot{u}_{i,j}^{(1)}-O_{j}^{(1)}u_{i,j}^{0}\right\Vert_{2}\right\} \right) ds%
}{\left\Vert \dot{\Delta}_{t,v}\right\Vert_{2}} \\
& =\max_{t\in[T]}\frac{1}{N_{2}}\sum_{i\in I_{2}}\int_{0}^{1}\left\Vert
\Psi_{it}^{0}\right\Vert_{2}\mathbf{1}\left\{\left\vert
\epsilon_{it}-\left\Vert \Psi_{it}^{0}\right\Vert_{2}\left\Vert \dot{\Delta}%
_{t,v}\right\Vert_{2}s\right\vert \leq\sum_{j\in [p]\cup \{0\}}\left\Vert
v_{t,j}^{0}\right\Vert_{2}\left\Vert \dot{u}%
_{i,j}^{(1)}-O_{j}^{(1)}u_{i,j}^{0}\right\Vert_{2}\right\} ds \\
& \leq \sup_{s\geq 0,t\in[T]}\frac{1}{N_{2}}\sum_{i\in I_{2}}\left\Vert
\Psi_{it}^{0}\right\Vert_{2}\mathbf{1}\left\{
\left\vert\epsilon_{it}-s\right\vert \leq \sum_{j\in [p]\cup
\{0\}}\left\Vert v_{t,j}^{0}\right\Vert_{2}\left\Vert \dot{u}%
_{i,j}^{(1)}-O_{j}^{(1)}u_{i,j}^{0}\right\Vert_{2}\right\} .
\end{align*}%
Following the same argument for $B_{3,t}$, we only need to upper bound the
RHS of the last display by $\sup_{s\geq 0,t\in[T]}B_{6,t}^{I}(s)$ on $%
\mathscr{B}_{N,1}(M)$ for some sufficiently large but fixed constant $M$,
where 
\begin{equation*}
B_{6,t}^{I}(s)=\frac{1}{N_{2}}\sum_{i\in I_{2}}\left\Vert
\Psi_{it}^{0}\right\Vert_{2}\mathbf{1}\left\{ \left\vert
\epsilon_{it}-s\right\vert \leq M\eta_{N}\sum_{j\in [p]\cup \{0\}}\left\Vert
v_{t,j}^{0}\right\Vert_{2}\right\} .
\end{equation*}%
Let 
\begin{equation*}
B_{6,t}^{II}(s)=\mathbb{E}(B_{6,t}^{I}(s)|\left\{ \mathscr{D}%
_{e_{it}}\right\}_{i\in I_{2}})=\frac{1}{N_{2}}\sum_{i\in I_{2}}\mathbb{E}%
\left(\left\Vert \Psi_{it}^{0}\right\Vert_{2}\mathbf{1}\left\{
\left\vert\epsilon_{it}-s\right\vert \leq M\eta_{N}\sum_{j\in
[p]\cup\{0\}}\left\Vert v_{t,j}^{0}\right\Vert_{2}\right\} \bigg|\mathscr{D}%
_{e_{it}}\right) .
\end{equation*}
We note that $\sup_{s\geq 0,t\in[T]}B_{6,t}^{II}(s)=O_{p}(\eta_{N}).$
Similar to the arguments in (\ref{Lem:B:B3_3}), to show $\max_{t\in[T]}\frac{%
\left\vert B_{6,t}\right\vert }{\left\Vert \dot{\Delta}_{t,v}\right\Vert_{2}}%
=O_{p}(\eta_{N})$, it suffices to show 
\begin{equation}
\mathbb{P}\left(\sup_{s\geq 0,t\in[T]}\left\vert
B_{6,t}^{I}(s)-B_{6,t}^{II}(s)\right\vert >c_{12}\eta_{N}\right) =o(1).
\label{Lem:B:B6_1}
\end{equation}

Further denote 
\begin{align*}
& B_{6,t}^{III}(s)=\frac{1}{N_{2}}\sum_{i\in I_{2}}\left\Vert
\Psi_{it}^{0}\right\Vert_{2}\mathbf{1}\left\{
\epsilon_{it}>s-M\eta_{N}\sum_{j\in [p]\cup \{0\}}\left\Vert
v_{t,j}^{0}\right\Vert_{2}\right\} , \\
& B_{6,t}^{IV}(s)=\frac{1}{N_{2}}\sum_{i\in I_{2}}\left\Vert
\Psi_{it}^{0}\right\Vert_{2}\mathbf{1}\left\{ \epsilon_{it}\geq
s+M\eta_{N}\sum_{j\in [p]\cup \{0\}}\left\Vert
v_{t,j}^{0}\right\Vert_{2}\right\} .
\end{align*}%
Then, we have $B_{6,t}^{I}(s)=B_{6,t}^{III}(s)-B_{6,t}^{IV}(s)$ and thus, 
\begin{align*}
& \sup_{s>T^{1/4},t\in[T]}\left\vert
B_{6,t}^{I}(s)-B_{6,t}^{II}(s)\right\vert \\
& \leq \sup_{s>T^{1/4},t\in[T]}\left\vert B_{6,t}^{III}(s)-\mathbb{E}%
\left(B_{6,t}^{III}(s)\big|\{\mathscr{D}_{e_{it}}\}_{i\in
I_{2}}\right)\right\vert +\sup_{s>T^{1/4},t\in[T]}\left\vert B_{6,t}^{IV}(s)-%
\mathbb{E}\left(B_{6,t}^{IV}(s)\big|\{\mathscr{D}_{e_{it}}\}_{i\in
I_{2}}\right) \right\vert \\
& \leq \max_{t\in[T]}B_{6,t}^{III}(T^{1/4})+\max_{t\in[T]}\mathbb{E}%
\left(B_{6,t}^{III}(T^{1/4})\big|\{\mathscr{D}_{e_{it}}\}_{i\in
I_{2}}\right) +\max_{t\in[T]}B_{6,t}^{IV}(T^{1/4})+\max_{t\in[T]}\mathbb{E}%
\left(B_{6,t}^{IV}(T^{1/4})\big|\{\mathscr{D}_{e_{it}}\}_{i\in I_{2}}\right)
\\
& \leq \max_{t\in[T]}\left\vert B_{6,t}^{III}(T^{1/4})-\mathbb{E}%
\left(B_{6,t}^{III}(T^{1/4})\big|\{\mathscr{D}_{e_{it}}\}_{i\in
I_{2}}\right) \right\vert +2\max_{t\in[T]}\mathbb{E}\left(
B_{6,t}^{III}(T^{1/4})\big|\{\mathscr{D}_{e_{it}}\}_{i\in I_{2}}\right) \\
& +\max_{t\in[T]}\left\vert B_{6,t}^{IV}(T^{1/4})-\mathbb{E}%
\left(B_{6,t}^{IV}(T^{1/4})\big|\{\mathscr{D}_{e_{it}}\}_{i\in
I_{2}}\right)\right\vert +2\max_{t\in[T]}\mathbb{E}%
\left(B_{6,t}^{IV}(T^{1/4})\big|\{\mathscr{D}_{e_{it}}\}_{i\in I_{2}}\right)
,
\end{align*}%
where the second inequality holds because both $B_{6,t}^{III}(s)$ and $%
B_{6,t}^{IV}(s)$ are non-decreasing in $s$. Further note that 
\begin{equation*}
\left\Vert \Psi_{it}^{0}\right\Vert_{2}\mathbf{1}\left\{
\epsilon_{it}>s-M\eta_{N}\sum_{j\in [p]\cup \{0\}}\left\Vert
v_{t,j}^{0}\right\Vert_{2}\right\} \quad \text{and}\quad \left\Vert
\Psi_{it}^{0}\right\Vert_{2}\mathbf{1}\left\{
\epsilon_{it}>s+M\eta_{N}\sum_{j\in [p]\cup \{0\}}\left\Vert
v_{t,j}^{0}\right\Vert_{2}\right\}
\end{equation*}%
are independent across $i\in I_{2}$ given $\{\mathscr{D}_{e_{it}}\}_{i\in
I_{2}}$. Therefore, following the same argument in the analysis of $A_{6,t}$
with the Bernstein's inequality replaced by the Hoeffding's inequality, we
have 
\begin{align*}
& \max_{t\in[T]}\left\vert B_{6,t}^{III}(T^{1/4})-\mathbb{E}%
\left(B_{6,t}^{III}(T^{1/4})\big|\{\mathscr{D}_{e_{it}}\}_{i\in
I_{2}}\right)\right\vert =O_{p}(\eta_{N}),\quad \max_{t\in[T]}\mathbb{E}%
\left(B_{6,t}^{III}(T^{1/4})\big|\{\mathscr{D}_{e_{it}}\}_{i\in
I_{2}}\right)=O_{p}(\eta_{N}) \\
& \max_{t\in[T]}\left\vert B_{6,t}^{IV}(T^{1/4})-\mathbb{E}%
\left(B_{6,t}^{IV}(T^{1/4})\big|\{\mathscr{D}_{e_{it}}\}_{i\in
I_{2}}\right)\right\vert =O_{p}(\eta_{N}),\quad \max_{t\in[T]}\mathbb{E}%
\left(B_{6,t}^{IV}(T^{1/4})\big|\{\mathscr{D}_{e_{it}}\}_{i\in
I_{2}}\right)=O_{p}(\eta_{N}),
\end{align*}
which implies 
\begin{equation*}
\sup_{s>T^{1/4},t\in[T]}\left\vert B_{6,t}^{I}(s)-B_{6,t}^{II}(s)\right\vert
=O_{p}(\eta_{N}).
\end{equation*}

In addition, following the same analysis in (\ref{B.11}) and (\ref%
{Lem:A:A3:diff}), we have 
\begin{equation*}
\sup_{s\in [0,T^{1/4}],t\in[T]}\left\vert
B_{6,t}^{I}(s)-B_{6,t}^{II}(s)\right\vert =O_{p}(\eta_{N}),
\end{equation*}
which leads to the desired result that $\max_{t\in[T]}\frac{\left\vert
B_{6,t}\right\vert }{\left\Vert \dot{\Delta}_{t,v}\right\Vert_{2}}%
=O_{p}\left(\eta_{N}\right) .$
\end{proof}

Recall for $i\in I_{3}$, $\dot{\mathcal{H}}_{i}\left(\left\{u_{i,j}\right%
\}_{j\in [p]\cup \{0\}}\right) =\frac{1}{T}\sum_{t=1}^{T}\left\{ \left[ \tau
-\mathfrak{F}_{it}\left(g_{it}(\{u_{i,j}\}_{j\in [p]\cup \{0\}})\right) %
\right] \dot{\varpi}_{it}\right\} ,$ where 
\begin{equation*}
g_{it}(\{u_{i,j}\}_{j\in [p]\cup \{0\}})=u_{i,0}^{\prime}\dot{v}%
_{t,0}^{(1)}+\sum_{j\in [p]}u_{i,j}^{\prime}\dot{v}%
_{t,j}^{(1)}X_{j,it}-u_{i,0}^{0\prime}v_{t,0}^{0}-\sum_{j\in
[p]}u_{i,j}^{0\prime}v_{t,j}^{0}X_{j,it},
\end{equation*}%
and $\dot{\varpi}_{it}=\left(\dot{v}_{t,0}^{(1)\prime},X_{1,it}\dot{v}%
_{t,1}^{(1)\prime},\cdots ,X_{p,it}\dot{v}_{t,p}^{(1)\prime}\right)
^{\prime} $.

\begin{lemma}
{\small \label{Lem16} } Under Assumptions \ref{ass:1}--\ref{ass:5}, the
second-order derivative of $\dot{\mathcal{H}}_{i}\left(\left\{u_{i,j}\right%
\}_{j=0}^{p}\right)$ is bounded in probability.
\end{lemma}

\begin{proof}
Noted that 
\begin{align*}
\frac{\dot{\mathcal{H}}_{i}\left(\left\{u_{i,j}\right\}_{j=0}^{p}\right)}{%
\partial u_{i}^{\prime}}=\frac{1}{T}\sum_{t=1}^{T}-\mathfrak{f}%
_{it}\left(u_{i,0}^{\prime}\dot{v}_{t,0}^{(1)}-u_{i,0}^{0%
\prime}v_{t,0}^{0}+X_{1,it}u_{i,1}^{\prime}\dot{v}%
_{t,1}^{(1)}-X_{1,it}u_{i,1}^{0\prime}v_{t,1}^{0}\right)\dot{\varpi}_{it} 
\dot{\varpi}_{it}^{\prime}.
\end{align*}
For notation simplicity, we focus on the case with $p=1$ and denote $%
u_{i}=\left(u_{i,0}^{\prime},u_{i,1}^{\prime}\right)^{\prime}$. Further
denote $\dot{\mathcal{H}}_{i,k}\left(\left\{u_{i,j}\right\}_{j=0}^{p}\right)$
as the $k^{th}$ element in $\dot{\mathcal{H}}_{i}\left(\left\{u_{i,j}\right%
\}_{j=0}^{p}\right)$ and $\dot{v}_{t,0,k}^{(1)}$ as the $k^{th}$ element in $%
\dot{v}_{t,0}^{(1)}$. For $k\in \left[K_{0} \right]$, we have 
\begin{align*}
&\frac{\partial \dot{\mathcal{H}}_{i,k}\left(\left\{u_{i,j}\right%
\}_{j=0}^{p}\right)}{\partial u_{i}}=\frac{1}{T}\sum_{t=1}^{T}-\mathfrak{f}%
_{it}\left(u_{i,0}^{\prime}\dot{v}_{t,0}^{(1)}-u_{i,0}^{0%
\prime}v_{t,0}^{0}+X_{1,it}u_{i,1}^{\prime}\dot{v}%
_{t,1}^{(1)}-X_{1,it}u_{i,1}^{0\prime}v_{t,1}^{0}\right)\dot{v}_{t,0,k}^{(1)}%
\dot{\varpi}_{it}\quad \text{and} \\
& \frac{\partial^{2} \dot{\mathcal{H}}_{i,k}\left(\left\{u_{i,j}\right%
\}_{j=0}^{p}\right)}{\partial u_{i}\partial u_{i}^{\prime}}=\frac{1}{T}%
\sum_{t=1}^{T}-\mathfrak{f}_{it}^{\prime}\left(u_{i,0}^{\prime}\dot{v}%
_{t,0}^{(1)}-u_{i,0}^{0\prime}v_{t,0}^{0}+X_{1,it}u_{i,1}^{\prime}\dot{v}%
_{t,1}^{(1)}-X_{1,it}u_{i,1}^{0\prime}v_{t,1}^{0}\right)\dot{v}_{t,0,k}^{(1)}%
\dot{\varpi}_{it} \dot{\varpi}_{it}^{\prime}.
\end{align*}

Therefore, we have 
\begin{align*}
\left\Vert \frac{\partial ^{2}\dot{\mathcal{H}}_{i,k}\left(\left\{u_{i,j}%
\right\}_{j=0}^{p}\right) }{\partial u_{i}\partial u_{i}^{\prime}}%
\right\Vert_{F}& \leq \frac{\bar{\mathfrak{f}}^{\prime }}{T}%
\sum_{t=1}^{T}\left\Vert \dot{v}_{t,0,k}^{(1)}\dot{\varpi}_{it}\dot{\varpi}%
_{it}^{\prime}\right\Vert_{F} \\
& \leq c\bar{\mathfrak{f}}^{\prime }\left[ \max_{t\in[T]}\left\Vert\dot{v}%
_{t,0}^{(1)}\right\Vert_{2}\right] \frac{1}{T}\sum_{t=1}^{T}\left\Vert \dot{%
\varpi}_{it}\right\Vert_{2}^{2} \\
& \leq c\bar{\mathfrak{f}}^{\prime }\left[ \max_{t\in[T]}\left\Vert\dot{v}%
_{t,0}^{(1)}\right\Vert_{2}\right] \left[ \max_{t\in [T]}(\left\Vert \dot{v}%
_{t,0}^{(1)}\right\Vert_{2}^{2}+\left\Vert \dot{v}_{t,1}^{(1)}\right%
\Vert_{2}^{2})\right] \left(1+\frac{1}{T}\sum_{t\in[T]}X_{1,it}^{2}\right)
=O_{p}(1),
\end{align*}%
where we use the fact that $\max_{t\in[T]}\left\Vert \dot{v}%
_{t,0}^{(1)}\right\Vert_{2}=O_{p}(1)$ by Theorem \ref{Thm2}(ii) and Lemma %
\ref{Lem:bounded u&v_tilde}(i). 

For $k\in \left[K_{0}+1,\cdots,K_{0}+K_{1} \right]$, we have 
\begin{align*}
& \frac{\partial \dot{\mathcal{H}}_{i,k}\left(\left\{u_{i,j}\right%
\}_{j=0}^{p}\right)}{\partial u_{i}}=\frac{1}{T}\sum_{t=1}^{T}-\mathfrak{f}%
_{it}\left(u_{i,0}^{\prime}\dot{v}_{t,0}^{(1)}-u_{i,0}^{0%
\prime}v_{t,0}^{0}+X_{1,it}u_{i,1}^{\prime}\dot{v}%
_{t,1}^{(1)}-X_{1,it}u_{i,1}^{0\prime}v_{t,1}^{0}\right)\dot{v}%
_{t,1,k-K_{0}}^{(1)}X_{1,it}\dot{\varpi}_{it} \quad \text{and} \\
& \frac{\partial^{2} \dot{\mathcal{H}}_{i,k}\left(\left\{u_{i,j}\right%
\}_{j=0}^{p}\right)}{\partial u_{i}\partial u_{i}^{\prime}}=\frac{1}{T}%
\sum_{t=1}^{T}-\mathfrak{f}_{it}^{\prime}\left(u_{i,0}^{\prime}\dot{v}%
_{t,0}^{(1)}-u_{i,0}^{0\prime}v_{t,0}^{0}+X_{1,it}u_{i,1}^{\prime}\dot{v}%
_{t,1}^{(1)}-X_{1,it}u_{i,1}^{0\prime}v_{t,1}^{0}\right)\dot{v}%
_{t,1,k-K_{0}}^{(1)}X_{1,it}\dot{\varpi}_{it} \dot{\varpi}_{it}^{\prime}.
\end{align*}
Therefore, we have 
\begin{align*}
\left\Vert \frac{\partial^{2} \dot{\mathcal{H}}_{i,k}\left(\left\{u_{i,j}%
\right\}_{j=0}^{p}\right)}{\partial u_{i}\partial u_{i}^{\prime}}%
\right\Vert_F & \leq c\bar{\mathfrak{f}}^{\prime }\left[\max_{t \in
[T]}\left\Vert\dot{v}_{t,0}^{(1)}\right\Vert_{2}\right] \frac{1}{T}%
\sum_{t=1}^T|X_{1,it}| \left\Vert\dot{\varpi}_{it}\right\Vert_{2}^{2} \\
& \leq c\bar{\mathfrak{f}}^{\prime }\left[\max_{t \in [T]} \left\Vert\dot{v}%
_{t,0}^{(1)}\right\Vert_{2}\right]\left[\max_{t \in [T]} (\left\Vert\dot{v}%
_{t,0}^{(1)}\right\Vert_{2}^{2} + \left\Vert\dot{v}_{t,1}^{(1)}\right%
\Vert_{2}^{2})\right]\left(1 + \frac{1}{T}\sum_{t\in T} X_{1,it}^3\right)
=O_p(1).
\end{align*}
\end{proof}

Recall that 
\begin{align*}
& \dot{\varpi}_{it}=\left(\dot{v}_{t,0}^{(1)\prime},\dot{v}%
_{t,1}^{(1)\prime}X_{1,it},\cdots ,\dot{v}_{t,p}^{(1)\prime}X_{p,it}\right)
^{\prime}, \\
& \varpi_{it}^{0}=\left(\left(O_{0}^{(1)}v_{t,0}^{0}\right)
^{\prime},\left(O_{1}^{(1)}v_{t,1}^{0}\right) ^{\prime}X_{1,it},\cdots
,\left(O_{p}^{(1)}v_{t,p}^{0}\right) ^{\prime}X_{p,it}\right) ^{\prime}, \\
& \dot{D}_{i}^{I}:=\frac{1}{T}\sum_{t=1}^{T}\mathfrak{f}_{it}\left(\dot{%
\Delta}_{t,v}^{\prime}\Psi_{it}^{0}\right) \dot{\varpi}_{it}\dot{\varpi}%
_{it}^{\prime},\quad \dot{D}_{i}^{II}:=\frac{1}{T}\sum_{t=1}^{T}\left[ \tau -%
\mathbf{1}\left\{ \epsilon_{it}\leq \dot{\Delta}_{t,v}^{\prime}\Psi_{it}^{0}%
\right\} \right] \dot{\varpi}_{it}, \\
& D_{i}^{I}=\frac{1}{T}\sum_{t=1}^{T}\mathfrak{f}_{it}(0)\varpi_{it}^{0}%
\varpi_{it}^{0\prime},\quad D_{i}^{II}=\frac{1}{T}\sum_{t=1}^{T}\left[ \tau -%
\mathbf{1}\left\{ \epsilon_{it}\leq 0\right\} \right] \varpi_{it}^{0}.
\end{align*}

\begin{lemma}
{\small \label{Lem17} } Under Assumptions \ref{ass:1}--\ref{ass:5}, we have

\begin{itemize}
\item[(i)] $\max_{i\in I_{3}}\left\Vert D_{i}^{II}\right\Vert_{F}=O_{p}\left(%
\sqrt{\frac{\log\left(N\vee T\right)}{T}}\xi_{N}\right)$,

\item[(ii)] $\max_{i\in \in I_{3}}\left\Vert\dot{D}_{i}^{I}-D_{i}^{I}\right%
\Vert_{F}=O_{p}(\eta_{N})$,

\item[(iii)] $\max_{i\in \in I_{3}}\left\Vert\dot{D}_{i}^{II}-D_{i}^{II}-%
\frac{1}{T}\sum_{t=1}^{T}\left[\mathbf{1}\left\{\epsilon_{it}\leq 0\right\}-%
\mathbf{1}\left\{\epsilon_{it}\leq \dot{\Delta}_{t,v}^{\prime}\Psi_{it}^{0}%
\right\}\right]\varpi_{it}^{0}\right\Vert_{2} = o_{p}\left(\left(N\vee
T\right)^{-\frac{1}{2}}\right)$.
\end{itemize}
\end{lemma}

\begin{proof}
Throughout the proof, we assume there is only one regressor $p=1$ for
notation simplicity.

(i) We notice that $\mathbb{E}\left(D_{i}^{II}\bigg|\mathscr{D}\right) =0$.
By conditional Bernstein's inequality, for a positive constant $c_{16}$, we
have 
\begin{align*}
& \mathbb{P}\left(\max_{i\in I_{3}}\left\Vert \sum_{t\in[T]}\left[\tau -%
\mathbf{1}\left\{ \epsilon_{it}\leq 0\right\} \right]%
O_{1}^{(1)}v_{t,1}^{0}X_{1,it}\right\Vert_{2}>c_{16}\sqrt{T\log (N\vee T)}%
\xi_{N}\bigg|\mathscr{D}\right) \\
& \leq \sum_{i\in I_{3}}\exp \left(-\frac{c_{9}c_{16}^{2}T\xi_{N}^{2}\log(N%
\vee T)}{\frac{4M^{2}}{c_{\sigma }^{2}}T\xi_{N}^{2}+\frac{2Mc_{16}}{%
c_{\sigma }}\xi_{N}^{2}\sqrt{T\log (N\vee T)}\log T\log \log T}\right)=o(1),
\end{align*}%
where the inequality follows from Lemma \ref{Lem:Bern}(i), Assumption \ref%
{ass:1}(ii), Assumption \ref{ass:1}(v), and the fact that $\max_{i\in I_{3},t%
\in[T]}\left\Vert \left[ \tau -\mathbf{1}\left\{ \epsilon_{it}\leq 0\right\} %
\right] O_{1}^{(1)}v_{t,1}^{0}X_{1,it}\right\Vert_{2}\leq \frac{2M}{%
c_{\sigma }}\xi_{N}$ a.s. Similar arguments hold for the upper block of $%
D_{i}^{II}$. This concludes the proof of (i).

(ii) Notice that 
\begin{equation*}
\dot{D}_{i}^{I}-D_{i}^{I}=\frac{1}{T}\sum_{t\in[T]}%
\begin{bmatrix}
\mathfrak{f}_{it}\left(\dot{\Delta}_{t,v}^{\prime}\Psi_{it}^{0}\right) \dot{v%
}_{t,0}^{(1)}\dot{v}_{t,0}^{(1)}-\mathfrak{f}_{it}(0)v_{t,0}^{0}v_{t,0}^{0%
\prime} & \mathfrak{f}_{it}\left(\dot{\Delta}_{t,v}^{\prime}\Psi_{it}^{0}%
\right) \dot{v}_{t,0}^{(1)}\dot{v}_{t,1}^{(1)}X_{1,it}-\mathfrak{f}%
_{it}(0)v_{t,0}^{0}v_{t,1}^{0\prime}X_{1,it} \\ 
\mathfrak{f}_{it}\left(\dot{\Delta}_{t,v}^{\prime}\Psi_{it}^{0}\right) \dot{v%
}_{t,1}^{(1)}\dot{v}_{t,0}^{(1)}X_{1,it}-\mathfrak{f}%
_{it}(0)v_{t,1}^{0}v_{t,0}^{0\prime}X_{1,it} & \mathfrak{f}_{it}\left(\dot{%
\Delta}_{t,v}^{\prime}\Psi_{it}^{0}\right) \dot{v}_{t,1}^{(1)}\dot{v}%
_{t,1}^{(1)}X_{1,it}^{2}-\mathfrak{f}_{it}(0)v_{t,1}^{0}v_{t,1}^{0%
\prime}X_{1,it}^{2}%
\end{bmatrix}%
.
\end{equation*}
To show the upper bound of $\dot{D}_{i}^{I}-D_{i}^{I}$, we take the lower
block for instance and all other three blocks follow the same pattern. Noted
that 
\begin{align}
& \max_{i\in I_{3}}\left\Vert \frac{1}{T}\sum_{t\in[T]}\left[\mathfrak{f}%
_{it}\left(\dot{\Delta}_{t,v}^{\prime}\Psi_{it}^{0}\right) \dot{v}%
_{t,1}^{(1)}\dot{v}_{t,1}^{(1)\prime}X_{1,it}^{2}-\mathfrak{f}%
_{it}(0)v_{t,1}^{0}v_{t,1}^{0\prime}X_{1,it}^{2}\right] \right\Vert_{F} 
\notag  \label{Lem17.1} \\
& \leq \max_{i\in I_{3}}\left\Vert \frac{1}{T}\sum_{t\in[T]}\left[\mathfrak{f%
}_{it}\left(\dot{\Delta}_{t,v}^{\prime}\Psi_{it}^{0}\right) -\mathfrak{f}%
_{it}(0)\right] \left[ \dot{v}_{t,1}^{(1)}\dot{v}_{t,1}^{(1)%
\prime}-v_{t,1}^{0}v_{t,1}^{0\prime}\right] X_{1,it}^{2}\right\Vert_{F} 
\notag \\
& +\max_{i\in I_{3}}\left\Vert \frac{1}{T}\sum_{t\in[T]}\left[\mathfrak{f}%
_{it}\left(\dot{\Delta}_{t,v}^{\prime}\Psi_{it}^{0}\right) -\mathfrak{f}%
_{it}(0)\right] v_{t,1}^{0}v_{t,1}^{0\prime}X_{1,it}^{2}\right\Vert_{F} 
\notag \\
& +\max_{i\in I_{3}}\left\Vert \frac{1}{T}\sum_{t\in[T]}\mathfrak{f}_{it}(0)%
\left[ \dot{v}_{t,1}^{(1)}\dot{v}_{t,1}^{(1)\prime}-v_{t,1}^{0}v_{t,1}^{0%
\prime}\right] X_{1,it}^{2}\right\Vert_{F}  \notag \\
& =O_{p}(\eta_{N}),
\end{align}%
where the equality is by the fact that 
\begin{align*}
& \max_{i\in I_{3},t\in[T]}\left\vert \mathfrak{f}_{it}\left(\dot{\Delta}%
_{t,v}^{\prime}\Psi_{it}^{0}\right) -\mathfrak{f}_{it}(0)\right\vert\lesssim
\max_{i\in I_{3},t\in[T]}\left\Vert
\Psi_{it}^{0}\right\Vert_{2}^{2}\left\Vert \dot{\Delta}_{t,v}\right%
\Vert_{2},\quad \max_{i\in I_{3}}\left\Vert \dot{v}_{t,1}^{(1)}\right%
\Vert_{2}+\left\Vert v_{t,1}^{0}\right\Vert_{2}=O_{p}(1), \\
& \max_{i\in I_{3}}\frac{1}{T}\sum_{t\in [T]}X_{1,it}^{2}\leq C,\quad \max_{t%
\in[T]}\left\Vert \dot{v}_{t,1}^{(1)}\dot{v}_{t,1}^{(1)%
\prime}-v_{t,1}^{0}v_{t,1}^{0\prime}\right\Vert_{F}\leq \max_{i\in
I_{3}}\left(\left\Vert \dot{v}_{t,1}^{(1)}\right\Vert_{2}+\left\Vert
v_{t,1}^{0}\right\Vert_{2}\right) \left\Vert \dot{v}_{t,1}^{(1)}-v_{t,1}^{0}%
\right\Vert_{2},
\end{align*}%
and 
\begin{align*}
& \max_{i\in I_{3}}\left\Vert \frac{1}{T}\sum_{t\in[T]}\left[\mathfrak{f}%
_{it}\left(\dot{\Delta}_{t,v}^{\prime}\Psi_{it}^{0}\right) -\mathfrak{f}%
_{it}(0)\right] \left[ \dot{v}_{t,1}^{(1)}\dot{v}_{t,1}^{(1)%
\prime}-v_{t,1}^{0}v_{t,1}^{0\prime}\right] X_{1,it}^{2}\right\Vert_{F} \\
& \leq \max_{i\in I_{3},t\in[T]}\bar{\mathfrak{f}}^{\prime}\left\Vert \dot{%
\Delta}_{t,v}\right\Vert_{2}\left\Vert \dot{v}_{t,1}^{(1)}\dot{v}%
_{t,1}^{(1)\prime}-v_{t,1}^{0}v_{t,1}^{0\prime}\right\Vert_{F}\max_{i\in
I_{3}}\frac{1}{T}\sum_{t\in[T]}\left\Vert
\Psi_{it}^{0}\right\Vert_{2}\left\vert X_{1,it}\right\vert
^{2}=O_{p}\left(\eta_{N}^{2}\right) .
\end{align*}

(iii) Note that 
\begin{align*}
& \dot{D}_{i}^{II}-D_{i}^{II}-\frac{1}{T}\sum_{t=1}^{T}\left[ \mathbf{1}%
\left\{ \epsilon_{it}\leq 0\right\} -\mathbf{1}\left\{ \epsilon_{it}\leq\dot{%
\Delta}_{t,v}^{\prime}\Psi_{it}^{0}\right\} \right] \varpi_{it}^{0} \\
& =\frac{1}{T}\sum_{t=1}^{T}\left[ \mathbf{1}\left\{ \epsilon_{it}\leq
0\right\} -\mathbf{1}\left\{ \epsilon_{it}\leq \dot{\Delta}%
_{t,v}^{\prime}\Psi_{it}^{0}\right\} \right] \left(\dot{\varpi}%
_{it}-\varpi_{it}^{0}\right)+\frac{1}{T}\sum_{t=1}^{T}\left[ \tau -\mathbf{1}%
\left\{ \epsilon_{it}\leq 0\right\} \right] \left(\dot{\varpi}%
_{it}-\varpi_{it}^{0}\right) \\
& =\frac{1}{T}\sum_{t=1}^{T}\mathbb{E}\left\{ \left[ \mathbf{1}%
\left\{\epsilon_{it}\leq 0\right\} -\mathbf{1}\left\{ \epsilon_{it}\leq \dot{%
\Delta}_{t,v}^{\prime}\Psi_{it}^{0}\right\} \right] \left(\dot{\varpi}%
_{it}-\varpi_{it}^{0}\right) \bigg|\mathscr{D}_{e_{i}}^{I_{1}\cup
I_{2}}\right\} \\
& +\frac{1}{T}\sum_{t=1}^{T}\left\{ \left[ \mathbf{1}\left\{
\epsilon_{it}\leq 0\right\} -\mathbf{1}\left\{ \epsilon_{it}\leq \dot{\Delta}%
_{t,v}^{\prime}\Psi_{it}^{0}\right\} \right] \left(\dot{\varpi}%
_{it}-\varpi_{it}^{0}\right) -\mathbb{E}\left[ \left[ \mathbf{1}\left\{
\epsilon_{it}\leq 0\right\} -\mathbf{1}\left\{ \epsilon_{it}\leq \dot{\Delta}%
_{t,v}^{\prime}\Psi_{it}^{0}\right\} \right] \left(\dot{\varpi}%
_{it}-\varpi_{it}^{0}\right) \bigg|\mathscr{D}_{e_{i}}^{I_{1}\cup I_{2}}%
\right] \right\} \\
& +\frac{1}{T}\sum_{t=1}^{T}\left[ \tau -\mathbf{1}\left\{ \epsilon_{it}\leq
0\right\} \right] \left(\dot{\varpi}_{it}-\varpi_{it}^{0}\right) \\
& :=S_{1,i}+S_{2,i}+S_{3,i},
\end{align*}
where 
\begin{equation*}
\max_{i\in I_{3}}\left\Vert S_{1,i}\right\Vert_{2}\leq \max_{i\in I_{3}}%
\frac{1}{T}\sum_{t\in[T]}\left\Vert \dot{\varpi}_{it}-\varpi_{it}^{0}\right%
\Vert_{2}\left\vert \mathfrak{F}_{it}\left(\dot{\Delta}_{t,v}^{\prime}%
\Psi_{it}^{0}\right) -\mathfrak{F}_{it}(0)\right\vert \lesssim\max_{i\in
I_{3}}\frac{1}{T}\sum_{t\in[T],j\in [p]}\left\vert X_{j,it}^{2}\right\vert
\max_{t\in[T]}\left\Vert \dot{\Delta}_{t,v}\right\Vert_{2}^{2}=O_{p}\left(%
\eta_{N}^{2}\right) .
\end{equation*}
As for $S_{2,i}$ and $S_{3,i}$, we first recall, for any $e>0$, there exists
a sufficiently large constant $M$ such that for 
\begin{equation*}
\mathscr{A}_{7}(M)=\left\{ \max_{i\in I_{3}}\left\Vert O_{j}^{(1)\prime}\dot{%
u}_{i,j}^{(1)}-u_{i,j}^{0}\right\Vert_{2}\leq M\eta_{N},\max_{t\in[T]%
}\left\Vert O_{j}^{(1)\prime}\dot{v}_{t,j}^{(1)}-v_{t,j}^{0}\right\Vert_{2}%
\leq M\eta_{N},\forall j\in [p]\cup \{0\}\right\}
\end{equation*}%
we have $\mathbb{P}(\mathscr{A}_{7}^{c}(M))\leq e$. In addition, let 
\begin{equation*}
\mathscr{A}_{7,i}(M)=\left\{ \left\Vert O_{j}^{(1)\prime}\dot{u}%
_{i,j}^{(1)}-u_{i,j}^{0}\right\Vert_{2}\leq M\eta_{N},\max_{t\in
[T]}\left\Vert O_{j}^{(1)\prime}\dot{v}_{t,j}^{(1)}-v_{t,j}^{0}\right%
\Vert_{2}\leq M\eta_{N},\forall j\in [p]\cup \{0\}\right\} .
\end{equation*}%
Then, we have 
\begin{align*}
& \mathbb{P}\left(\max_{i\in I_{3}}\left\Vert
S_{2,i}\right\Vert_{2}>c_{17}\eta_{N}^{2}\right) \leq \mathbb{P}%
\left(\max_{i\in I_{3}}\left\Vert S_{2,i}\right\Vert_{2}>c_{17}\eta_{N}^{2},%
\mathscr{A}_{7}(M)\right) +e \\
& \leq \sum_{i\in I_{3}}\mathbb{P}\left(\left\Vert
S_{2,i}\right\Vert_{2}>c_{17}\eta_{N}^{2},\mathscr{A}_{7}(M)\right) +e \\
& \leq \sum_{i\in I_{3}}\mathbb{P}\left(\left\Vert
S_{2,i}\right\Vert_{2}>c_{17}\eta_{N}^{2},\mathscr{A}_{7,i}(M)\right) +e \\
& =\sum_{i\in I_{3}}\mathbb{E}\mathbb{P}\left(\left\Vert
S_{2,i}\right\Vert_{2}>c_{17}\eta_{N}^{2}\bigg|\mathscr{D}%
_{e_{i}}^{I_{1}\cup I_{2}}\right)1\{\mathscr{A}_{7,i}(m)\}+e \\
& \leq \sum_{i\in I_{2}}\exp\left(-\frac{c_{9}c_{17}^{2}T^{2}\eta_{N}^{4}}{%
c_{18}T\eta_{N}^{2}\xi_{N}^{2}+c_{17}c_{18}^{1/2}T\eta_{N}^{3}\xi_{N}\log
T\log \log T}\right) +e=o(1)+e
\end{align*}
with a positive constant $c_{17}$ and the inequality above is by Lemma \ref%
{Lem:Bern}(i) with the fact that, under $\mathscr{A}_{6,i}(M)$, 
\begin{align*}
& \max_{i\in I_{3},t\in[T]}\left\Vert \left[ \mathbf{1}\left\{\epsilon_{it}%
\leq 0\right\} -\mathbf{1}\left\{ \epsilon_{it}\leq \dot{\Delta}%
_{t,v}^{\prime}\Psi_{it}^{0}\right\} \right] \left(\dot{\varpi}%
_{it}-\varpi_{it}^{0}\right) -\mathbb{E}\left[ \left[ \mathbf{1}%
\left\{\epsilon_{it}\leq 0\right\} -\mathbf{1}\left\{ \epsilon_{it}\leq \dot{%
\Delta}_{t,v}^{\prime}\Psi_{it}^{0}\right\} \right] \left(\dot{\varpi}%
_{it}-\varpi_{it}^{0}\right) \bigg|\mathscr{D}_{e_{i}}^{I_{1}\cup I_{2}}%
\right] \right\Vert_{2}^{2} \\
& \leq c_{18}\eta_{N}\xi_{N}.
\end{align*}%
As $e$ is arbitrary, we have $\max_{i\in I_{3}}\left\Vert
S_{2,i}\right\Vert_{2}=O_{p}(\eta_{N}^{2})$. Following a similar argument,
we have $\max_{i\in I_{3}}\left\Vert
S_{3,i}\right\Vert_{2}=O_{p}(\eta_{N}^{2})$. By Assumption \ref{ass:1}(ix),
we note that $O_{p}(\eta_{N}^{2})=o_{p}\left(\left(N\vee T\right)
^{-1/2}\right) $. 
\end{proof}

\subsection{Lemmas for the Proof Theorem \protect\ref{Thm3}}

\begin{lemma}
{\small \label{Lem18} } Define $H_{x,j}^{l}=\left(\frac{\hat{L}%
_{j}^{\prime}L_{j}^{0}}{N}\right)^{-1}$ and $H_{x,j}^{w}=W_{j}^{0\prime}\hat{%
W}_{j}\left(\hat{W}_{j}^{\prime}\hat{W}_{j}\right)^{-1}$. Under Assumptions %
\ref{ass:7}-\ref{ass:9}, we have

\begin{itemize}
\item[(i)] $H_{x,j}^{l}=H_{x,j}^{w}+O_{p}\left(\frac{1}{N\wedge T}\right)$,

\item[(ii)] $\frac{1}{N}\left\Vert \hat{L}_{j}-L_{j}^{0}H_{x,j}^{w}\right%
\Vert_{F}^{2}=O_{p}\left(\frac{1}{N\wedge T}\right)$,

\item[(iii)] $\frac{1}{T}\left\Vert\hat{W}_{j}-W_{j}\left(H_{x,j}^{l\prime
}\right)^{-1}\right\Vert_{F}^{2}=O_{p}\left(\frac{1}{N\wedge T}\right)$.
\end{itemize}
\end{lemma}

\begin{proof}
The proof can be found in \citet[Lemma 3 and Proposition
1]{bai2020simpler}.
\end{proof}

\begin{lemma}
{\small \label{Lem19} } Under Assumptions \ref{ass:7}-\ref{ass:9}, we have

\begin{itemize}
\item[(i)] $\frac{1}{N}L_{j}^{0\prime}\left(\hat{L}_{j}-L_{j}^{0}H_{x,j}^{w}%
\right)=O_{p}\left(\frac{1}{N\wedge T}\right)$,

\item[(ii)] $\frac{1}{T}W_{j}^{0\prime}\left(\hat{W}_{j}-W_{j}^{0}%
\left(H_{x,j}^{l \prime }\right)^{-1}\right)=O_{p}\left(\frac{1}{N\wedge T}%
\right)$,

\item[(iii)] $\max_{t\in[T]}\frac{1}{N}\left(\hat{L}%
_{j}-L_{j}^{0}H_{x,j}^{w}\right) ^{\prime}e_{j,t}=O_{p}\left(\frac{\log
(N\vee T)}{N\wedge T}\right) $,

\item[(iv)] $\max_{i\in [N]}\frac{1}{T}\left(\hat{W}_{j}-W_{j}^{0}%
\left(H_{x,j}^{l\prime }\right) ^{-1}\right) ^{\prime }e_{j,i}=O_{p}\left(%
\frac{\log (N\vee T)}{N\wedge T}\right) $.
\end{itemize}
\end{lemma}

\begin{proof}
Statements (i) and (ii) are the same as Lemma 4(i) and (ii) in \cite%
{bai2020simpler}. Statements (iii) and (iv) are the uniform version of Lemma
4(iii) and (iv) in \cite{bai2020simpler}. Below we focus on part (iv) as the
proof of part (iii) follows analogously.

Noting that $\hat{W}_{j}^{\prime}-\left(H_{x,j}^{l}\right)
^{-1}W_{j}^{0\prime}=\frac{1}{N}\hat{L}_{j}^{\prime}X_{j}-\frac{\hat{L}%
_{j}^{\prime}L_{j}^{0}}{N}W_{j}^{0\prime}=\frac{1}{N}\hat{L}%
_{j}^{\prime}E_{j},$ we have 
\begin{equation*}
\frac{1}{T}\left(\hat{W}_{j}-W_{j}^{0}\left(H_{x,j}^{l\prime
}\right)^{-1}\right) ^{\prime}e_{j,i}=\frac{1}{NT}e_{j,i}^{\prime}E_{j}^{%
\prime}\hat{L}_{j}=\frac{1}{NT}e_{j,i}^{\prime}E_{j}^{%
\prime}L_{j}^{0}H_{x,j}^{w}+\frac{1}{NT}e_{j,i}^{\prime}E_{j}^{\prime}\left(%
\hat{L}_{j}-L_{j}^{0}H_{x,j}^{w}\right) .
\end{equation*}%
For the first term on the right side, 
\begin{equation*}
\max_{i\in [N]}\frac{1}{NT}e_{j,i}^{\prime}E_{j}^{%
\prime}L_{j}^{0}H_{x,j}^{w}\lesssim \max_{i\in [N]}\frac{1}{NT}\left\Vert
e_{j,i}^{\prime}E_{j}^{\prime}L_{j}^{0}\right\Vert_{2}=O_{p}\left(\frac{%
\log(N\vee T)}{N\wedge T}\right) ,
\end{equation*}%
by Assumption \ref{ass:9}(i). For the second term on the right side, we have 
\begin{align*}
\max_{i\in [N]}\frac{1}{NT}\left\Vert e_{j,i}^{\prime}E_{j}^{\prime}\left(%
\hat{L}_{j}-L_{j}^{0}H_{x,j}^{w}\right) \right\Vert_{2}& \leq\max_{i\in [N]}%
\frac{1}{\sqrt{N}T}\left\Vert e_{j,i}^{\prime}E_{j}^{\prime}\right\Vert_{2}%
\frac{\left\Vert \hat{L}_{j}-L_{j}^{0}H_{x,j}^{w}\right\Vert_{F}}{\sqrt{N}}
\\
& =O_{p}\left(\frac{\log (N\vee T)}{\sqrt{N\wedge T}}\right) O_{p}\left(%
\frac{1}{\sqrt{N\wedge T}}\right) ,
\end{align*}
by Assumption \ref{ass:7}(iv) and Lemma \ref{Lem18}(ii). Combining the above
results completes the proof of part (iv).
\end{proof}

\begin{lemma}
{\small \label{Lem20} } Under Assumption \ref{ass:1} and Assumptions \ref%
{ass:7}-\ref{ass:9}, we have $\forall j\in [p]$,

\begin{itemize}
\item[(i)] $\hat{w}_{j,t}-\left(H_{x,j}^{l}%
\right)^{-1}w_{j,t}^{0}=H_{x,j}^{l\prime}\frac{1}{N}\sum_{i\in
[N]}l_{j,i}^{0}e_{j,it}+\mathcal{R}_{w,t}$,

\item[(ii)] $\hat{l}_{j,i}-H_{x,j}^{w\prime}l_{j,i}^{0}=\left(\frac{%
W_{j}^{0\prime}W_{j}^{0}}{T}\right)^{-1}\left(H_{x,j}^{l}\right)^{-1}\frac{1%
}{T}\sum_{t=1}^{T}w_{j,t}^{0}e_{j,it}+\mathcal{R}_{l,i}$,

\item[(iii)] $\hat{\mu}_{j,it}-\mu_{j,it}=e_{j,it}-\hat{e}%
_{j,it}=w_{j,t}^{0\prime}\left(\frac{W_{j}^{0\prime}W_{j}^{0}}{T}\right)^{-1}%
\frac{1}{T}\sum_{t=1}^{T}w_{j,t}^{0}e_{j,it}+l_{j,i}^{0\prime}\frac{1}{N}%
\sum_{i\in [N]}l_{j,i}^{0}e_{j,it}+\mathcal{R}_{j,it}$,
\end{itemize}

such that 
\begin{align*}
& \max_{t\in[T]}\left\vert \mathcal{R}_{w,t}\right\vert =O_{p}\left(\frac{%
\log (N\vee T)}{N\wedge T}\right) ,\quad \max_{i\in [N]}\left\vert \mathcal{R%
}_{l,i}\right\vert =O_{p}\left(\frac{\log (N\vee T)}{N\wedge T}\right)
,\quad \max_{i\in [N],t\in[T],j\in[p]}\left\vert \mathcal{R}%
_{j,it}\right\vert =O_{p}\left(\frac{\log(N\vee T)}{N\wedge T}\right) , \\
& \max_{t\in[T]}\left\Vert \hat{w}_{j,t}-\left(H_{x,j}^{l}%
\right)^{-1}w_{j,t}^{0}\right\Vert_{F}=O_{p}(\eta_{N}) ,\quad \max_{i\in
[N]}\left\Vert \hat{l}_{j,i}-H_{x,j}^{w\prime
}l_{j,i}^{0}\right\Vert_{F}=O_{p}(\eta_{N}) , \\
& \max_{j\in [p],i\in [N],t\in[T]}\left\vert e_{j,it}-\hat{e}%
_{j,it}\right\vert =\max_{j\in [p],i\in [N],t\in[T]}\left\vert \hat{\mu}%
_{j,it}-\mu_{j,it}\right\vert =O_{p}(\eta_{N}),
\end{align*}%
with $\eta_{N}=\frac{\sqrt{\log \left(N\vee T\right) }}{\sqrt{N\wedge T}}%
\xi_{N}^{2}$.
\end{lemma}

\begin{proof}
Recall that $X_{j,it}=\mu_{j,it}+e_{j,it}=l_{j,i}^{0%
\prime}w_{j,t}^{0}+e_{j,it}$ and $X_{j}=L_{j}^{0}W_{j}^{0\prime}+E_{j}$ in
matrix form, where $L_{j}^{0}\in \mathbb{R}^{N\times r_{j}}$ is the factor
loading and $W_{j}^{0}\in \mathbb{R}^{T\times r_{j}}$ is the factor matrix.
Following \cite{bai2002determining}, \cite{bai2003inferential} and \cite%
{bai2020simpler}, if we impose the normalization restrictions that 
\begin{equation*}
\frac{L_{j}^{0\prime}L_{j}^{0}}{N}=I_{r_{j}}\quad\text{and}\quad\frac{%
W_{j}^{0\prime}W_{j}^{0}}{T}\quad\text{is a diagonal matrix with descending
diagonal elements,}
\end{equation*}%
we have the principal components estimators: 
\begin{equation}
\hat{L}_{j}=X_{j}\hat{W}_{j}\left(\hat{W}_{j}^{\prime}\hat{W}_{j}\right)^{-1}%
\text{ and }\hat{W}_{j}^{\prime}=\left(\hat{L}_{j}^{\prime}\hat{L}%
_{j}\right) ^{-1}\hat{L}_{j}^{\prime}X_{j}=\frac{1}{N}\hat{L}_{j}^{\prime
}X_{j}.  \label{Lem20.1}
\end{equation}%
Let $H_{x,j}^{l}=\left(\frac{1}{N}\hat{L}_{j}^{\prime}L_{j}^{0}\right)
^{-1}. $ Premultiplying $\frac{1}{N}\hat{L}_{j}^{\prime}$ on both sides of $%
X_{j}=L_{j}^{0}W_{j}^{0\prime}+E_{j}$ yields 
\begin{equation*}
\frac{1}{N}\hat{L}_{j}^{\prime}X_{j}=\frac{1}{N}\hat{L}_{j}^{%
\prime}L_{j}^{0}W_{j}^{0\prime}+\frac{1}{N}\hat{L}_{j}^{\prime}E_{j}.
\end{equation*}
It follows that 
\begin{align*}
\hat{W}_{j}^{\prime}& =\left(H_{x,j}^{l}\right) ^{-1}W_{j}^{0\prime}+\frac{1%
}{N}\hat{L}_{j}^{0\prime}E_{j} \\
& =\left(H_{x,j}^{l}\right) ^{-1}W_{j}^{0\prime}+\frac{1}{N}%
\left(H_{x,j}^{l}\right) ^{\prime}L_{j}^{0\prime}E_{j}+\frac{1}{N}\left[ 
\hat{L}-L_{j}^{0}H_{x,j}^{l}\right] ^{\prime}E_{j}.
\end{align*}%
We then show the expansion for each factor, i.e., 
\begin{equation}
\hat{w}_{j,t}-\left(H_{x,j}^{l}\right) ^{-1}w_{j,t}^{0}=H_{x,j}^{l\prime }%
\frac{1}{N}\sum_{i\in [N]}l_{j,i}^{0}e_{j,it}+\frac{1}{N}\left[ \hat{L}%
_{j}-L_{j}^{0}H_{x,j}^{l}\right] ^{\prime}e_{j,t}.  \label{Lem20.2}
\end{equation}

For equation (\ref{Lem20.2}), we have the uniform bound for the second term,
i.e., 
\begin{equation*}
\max_{t\in[T]}\frac{1}{N}\left[ \hat{L}_{j}-L_{j}^{0}H_{x,j}^{l}\right]
^{\prime}e_{j,t}=O_{p}\left(\frac{\log (N\vee T)}{N\wedge T}\right)
\end{equation*}%
by Lemma \ref{Lem19}(iii). With $c_{19}$ being a positive constant and $%
\max_{i\in [N],t\in[T]}\left\vert l_{j,i}^{0}e_{j,it}\right\Vert_{2}\leq
c_{19}\xi_{N}$ a.s., we show that 
\begin{equation}
\mathbb{P}\left(\max_{t\in[T]}\left\Vert \frac{1}{N}\sum_{i\in[N]%
}l_{j,i}^{0}e_{j,it}\right\Vert_{2}>c_{20}\sqrt{\frac{\log (N\vee T)}{N}}%
\xi_{N}\right) \leq \max_{t\in[T]}2\exp \left(-\frac{c_{20}^{2}N\xi_{N}^{2}%
\log (N\vee T)}{N4c_{19}^{2}\xi_{N}^{2}}\right) =o(1)  \label{Lem:le}
\end{equation}%
for a positive constant $c_{20}$ by Hoeffding's inequality. It follows that $%
\hat{w}_{j,t}-\left(H_{x,j}^{l}\right) ^{-1}w_{j,t}^{0}=H_{x,j}^{l\prime }%
\frac{1}{N}\sum_{i\in [N]}l_{j,i}^{0}e_{j,it}+\mathcal{R}_{w,t},$ such that 
\begin{equation}
\max_{t\in[T]}\left\Vert \hat{w}_{j,t}-\left(H_{x,j}^{l}%
\right)^{-1}w_{j,t}^{0}\right\Vert_{F}=O_{p}\left(\sqrt{\frac{\log (N\vee T)%
}{T}}\xi_{N}\right) ,  \label{Lem20.3}
\end{equation}%
and $\max_{t\in[T]}\left\vert \mathcal{R}_{w,t}\right\vert =O_{p}\left(\frac{%
\log (N\vee T)}{N\wedge T}\right) $.

Similarly, if we premultiply $\hat{W}_{j}\left(\hat{W}_{j}^{\prime}\hat{W}%
_{j}\right) ^{-1}$ to both sides of $X_{j}=L_{j}^{0}W_{j}^{0\prime}+E_{j}$,
it yields 
\begin{equation*}
X_{j}\hat{W}_{j}\left(\hat{W}_{j}^{\prime}\hat{W}_{j}%
\right)^{-1}=L_{j}^{0}W_{j}^{0\prime}\hat{W}_{j}\left(\hat{W}_{j}^{\prime}%
\hat{W}_{j}\right) ^{-1}+E_{j}\hat{W}_{j}\left(\hat{W}_{j}^{\prime}\hat{W}%
_{j}\right) ^{-1}.
\end{equation*}%
It follows that 
\begin{align*}
\hat{L}_{j}& =L_{j}^{0}H_{x,j}^{w}+E_{j}\hat{W}_{j}\left(\hat{W}_{j}^{\prime}%
\hat{W}_{j}\right) ^{-1} \\
& =L_{j}^{0}H_{x,j}^{w}+E_{j}W_{j}^{0}\left(H_{x,j}^{l\prime
}\right)^{-1}\left(\hat{W}_{j}^{\prime}\hat{W}_{j}\right) ^{-1}+E_{j}\left(%
\hat{W}_{j}-W_{j}^{0}\left(H_{x,j}^{l\prime }\right) ^{-1}\right) \left(\hat{%
W}_{j}^{\prime}\hat{W}_{j}\right) ^{-1},
\end{align*}%
where the fist line is due to (\ref{Lem20.1}) and the definition that $%
H_{x,l}=W_{j}^{0\prime}\hat{W}_{j}\left(\hat{W}_{j}^{\prime}\hat{W}%
_{j}\right) ^{-1}$. Then we obtain the expansion for the factor loading 
\begin{equation*}
\hat{l}_{j,i}-H_{x,j}^{w\prime }l_{j,i}^{0}=\left(\frac{\hat{W}_{j}^{\prime}%
\hat{W}_{j}}{T}\right) ^{-1}\left(H_{x,j}^{l}\right) ^{-1}\frac{1}{T}%
\sum_{t=1}^{T}w_{j,t}^{0}e_{j,it}+\left(\frac{\hat{W}_{j}^{\prime}\hat{W}_{j}%
}{T}\right) ^{-1}\frac{1}{T}\left[ \left(\hat{W}-W\left(H_{x,j}^{l\prime}%
\right) ^{-1}\right) ^{\prime}e_{j,i}\right] .
\end{equation*}%
Note that 
\begin{align}
\frac{\hat{W}_{j}^{\prime}\hat{W}_{j}}{T}& =\frac{\left(\hat{W}%
_{j}-W_{j}^{0}\left(H_{x,j}^{l\prime }\right) ^{-1}\right) ^{\prime}\hat{W}%
_{j}}{T}+\frac{\left(H_{x,j}^{l}\right) ^{-1}W_{j}^{0\prime}\hat{W}_{j}}{T} 
\notag  \label{Lem20.5} \\
& =\frac{\left(\hat{W}_{j}-W_{j}^{0}\left(H_{x,j}^{l\prime
}\right)^{-1}\right) ^{\prime}\left(\hat{W}_{j}-W_{j}^{0}\left(H_{x,j}^{l%
\prime}\right) ^{-1}\right) }{T}+\frac{\left(\hat{W}_{j}-W_{j}^{0}%
\left(H_{x,j}^{l\prime }\right) ^{-1}\right)
^{\prime}W_{j}^{0}\left(H_{x,j}^{l\prime}\right) ^{-1}}{T}  \notag \\
& +\frac{\left(H_{x,j}^{l}\right) ^{-1}W_{j}^{0\prime}\left(\hat{W}%
_{j}-W_{j}\left(H_{x,j}^{l\prime }\right) ^{-1}\right) }{T}+\frac{%
\left(H_{x,j}^{l}\right)
^{\prime}W_{j}^{0\prime}W_{j}^{0}\left(H_{x,j}^{l\prime}\right) ^{-1}}{T} 
\notag \\
& =\frac{\left(H_{x,j}^{l}\right)
^{\prime}W_{j}^{0\prime}W_{j}\left(H_{x,j}^{l\prime }\right) ^{-1}}{T}%
+O_{p}\left(\frac{1}{N\wedge T}\right) ,
\end{align}%
where the last equality holds by Lemma \ref{Lem18}(iii). Note that $%
\max_{i\in [N]}\frac{1}{T}\left(\hat{W}-W\left(H_{x,j}^{l\prime}\right)
^{-1}\right) ^{\prime}e_{i}=O_{p}\left(\frac{\log (N\vee T)}{N\wedge T}%
\right) $ by Lemma \ref{Lem19}(iv), and we can show that 
\begin{equation}
\max_{i\in [N]}\left\Vert \frac{1}{T}\sum_{t=1}^{T}w_{j,t}^{0}e_{j,it}\right%
\Vert_{2}=O_{p}\left(\sqrt{\frac{\log (N\vee T)}{T}}\xi_{N}\right)
\label{Lem:we}
\end{equation}
as in (\ref{Lem:le}) by conditional Bernstein's inequality in Lemma \ref%
{Lem:Bern}(i) given $\left\{ W_{j}^{0}\right\}_{j\in [p]}$. It follows that $%
\hat{l}_{j,i}-H_{x,j}^{w\prime }l_{j,i}^{0}=\left(\frac{W_{j}^{0%
\prime}W_{j}^{0}}{T}\right) ^{-1}\left(H_{x,j}^{l}\right) ^{-1}\frac{1}{T}%
\sum_{t=1}^{T}w_{j,t}^{0}e_{j,it}+\mathcal{R}_{l,i}\ $such that 
\begin{equation*}
\max_{i\in [N]}\left\Vert \hat{l}_{j,i}-H_{x,j}^{w\prime}l_{j,i}^{0}\right%
\Vert_{F}=O_{p}(\eta_{N})\ \text{and }\max_{i\in[N]}\left\vert \mathcal{R}%
_{l,i}\right\vert =O_{p}\left(\frac{\log(N\vee T)}{N\wedge T}\right) .
\end{equation*}
Then, it's natural to obtain that 
\begin{align*}
\hat{\mu}_{j,it}-\mu_{j,it}& =\hat{l}_{j,i}^{\prime}\hat{w}%
_{t}-l_{j,i}^{0\prime}w_{j,t}^{0} \\
& =\left(\hat{l}_{j,i}-\left(H_{x,j}^{l}\right)
^{\prime}l_{j,i}^{0}\right)^{\prime}\left(\hat{w}_{j,t}-\left(H_{x,j}^{l}%
\right)^{-1}w_{j,t}^{0}\right) +\left(\hat{l}_{j,i}-\left(H_{x,j}^{l}%
\right)^{\prime}l_{j,i}^{0}\right) ^{\prime}\left(H_{x,j}^{l}\right)
^{-1}w_{j,t}^{0} \\
& +l_{j,i}^{0\prime}H_{x,j}^{l}\left(\hat{w}_{j,t}-\left(H_{x,j}^{l}%
\right)^{-1}w_{j,t}^{0}\right) \\
& =\left(\hat{l}_{j,i}-\left(H_{x,j}^{l}\right)
^{\prime}l_{j,i}^{0}\right)^{\prime}\left(H_{x,j}^{l}\right)
^{-1}w_{j,t}^{0}+l_{j,i}^{0\prime}H_{x,j}^{l}\left(\hat{w}%
_{j,t}-\left(H_{x,j}^{l}\right)^{-1}w_{j,t}^{0}\right) +O_{p}\left(\frac{%
\log (N\vee T)}{N\wedge T}\right) \\
& =w_{j,t}^{0}\left(\frac{W_{j}^{0\prime}W_{j}^{0}}{T}\right) ^{-1}\frac{1}{T%
}\sum_{t=1}^{T}w_{j,t}^{0}e_{j,it}+l_{j,i}^{0\prime}\frac{1}{N}\sum_{i\in[N]%
}l_{j,i}^{0}e_{j,it}+\mathcal{R}_{j,it}
\end{align*}%
where the second equality holds by statements (i) and (ii), and the last
equality holds with $\max_{i\in [N],t\in[T]}\left\vert \mathcal{R}%
_{j,it}\right\vert =O_{p}\left(\frac{\log (N\vee T)}{N\wedge T}\right) $. By
Assumption \ref{ass:1}(iv), $\max_{t\in[T]}\left\Vert
w_{j,t}^{0}\right\Vert_{2}\leq M$ a.s. and $\max_{i\in [N]}\left\Vert
l_{j,i}^{0}\right\Vert_{2}\leq M$ a.s., which leads to $\max_{j\in [p],i\in
[N],t\in[T]}\left\vert\hat{\mu}_{j,it}-\mu_{j,it}\right\vert =O_{p}\left(%
\sqrt{\frac{\log \left(N\vee T\right) }{N\wedge T}}\xi_{N}\right)
=O_{p}(\eta_{N})$.
\end{proof}

\begin{lemma}
{\small \label{Lem21} } Under Assumptions \ref{ass:1}-\ref{ass:10}, for
matrices $\dot{D}_{i}^{F}$ and $D_{i}^{F}$ defined in the proof of Theorem %
\ref{Thm3}, we have $\max_{i\in I_{3}}\left\Vert \dot{D}_{i}^{F}-D_{i}^{F}%
\right\Vert_{F}=O_{p}(\eta_{N})$ with $\eta_{N}=\frac{\sqrt{\log\left(N\vee
T\right) }\xi_{N}^{2}}{\sqrt{N\wedge T}}$.
\end{lemma}

\begin{proof}
Recall that 
\begin{eqnarray*}
\dot{D}_{i}^{F} &=&\frac{1}{T}\sum_{t=1}^{T}f_{it}\left[ \iota_{it}%
\left(O_{0}^{(1)}u_{i,0}^{0},O_{1}^{(1)}u_{i,1}^{0},\dot{u}%
_{i,1}^{(1)}\right) \right] 
\begin{bmatrix}
\dot{v}_{t,0}^{(1)}\dot{v}_{t,0}^{(1)\prime} & \hat{e}_{1,it}\dot{v}%
_{t,0}^{(1)}\dot{v}_{t,1}^{(1)\prime} \\ 
\hat{e}_{1,it}\dot{v}_{t,1}^{(1)}\dot{v}_{t,0}^{(1)\prime} & \hat{e}%
_{1,it}^{2}\dot{v}_{t,1}^{(1)}\dot{v}_{t,1}^{(1)\prime}%
\end{bmatrix}%
~ \text{and} \\
D_{i}^{F} &=&\frac{1}{T}\sum_{t=1}^{T}f_{it}(0)%
\begin{bmatrix}
O_{0}^{(1)}v_{t,0}^{0}v_{t,0}^{0\prime}O_{0}^{(1)\prime} & 0 \\ 
0 & e_{1,it}^{2}O_{1}^{(1)}v_{t,1}^{0}v_{t,1}^{0\prime}O_{1}^{(1)\prime}%
\end{bmatrix}%
,
\end{eqnarray*}
with $\max_{i\in I_{2}}\left\Vert D_{i}^{F}\right\Vert_{F}=O(1)$ a.s..

Let $\dot{\iota}_{it}$ denote $\iota_{it}%
\left(O_{0}^{(1)}u_{i,0}^{0},O_{1}^{(1)}u_{i,1}^{0},\dot{u}%
_{i,1}^{(1)}\right) $ for short. We have 
\begin{align*}
& \dot{D}_{i}^{F}-D_{i}^{F} \\
& =\frac{1}{T}\sum_{t=1}^{T}f_{it}(0)%
\begin{bmatrix}
\dot{v}_{t,0}^{(1)}\dot{v}_{t,0}^{(1)%
\prime}-O_{0}^{(1)}v_{t,0}^{0}v_{t,0}^{0\prime}O_{0}^{(1)\prime} & \dot{v}%
_{t,0}^{(1)}\dot{v}_{t,1}^{(1)\prime}\hat{e}%
_{1,it}-O_{0}^{(1)}v_{t,0}^{0}v_{t,1}^{0\prime}O_{1}^{(1)\prime}e_{1,it} \\ 
\dot{v}_{t,1}^{(1)}\dot{v}_{t,0}^{(1)\prime}\hat{e}%
_{1,it}-O_{1}^{(1)}v_{t,1}^{0}v_{t,0}^{0\prime}O_{0}^{(1)\prime}e_{1,it} & 
\hat{e}_{1,it}^{2}\dot{v}_{t,1}^{(1)}\dot{v}_{t,1}^{(1)%
\prime}-e_{1,it}^{2}O_{1}^{(1)}v_{t,1}^{0}v_{t,1}^{0\prime}O_{1}^{(1)\prime}%
\end{bmatrix}
\\
& +\frac{1}{T}\sum_{t=1}^{T}f_{it}(0)%
\begin{bmatrix}
0 & O_{0}^{(1)}v_{t,0}^{0}v_{t,1}^{0\prime}O_{1}^{(1)\prime}e_{1,it} \\ 
O_{1}^{(1)}v_{t,1}^{0}v_{t,0}^{0\prime}O_{0}^{(1)\prime}e_{1,it} & 0%
\end{bmatrix}
\\
& +\frac{1}{T}\sum_{t=1}^{T}\left(f_{it}\left(\dot{\iota}_{it}%
\right)-f_{it}(0)\right) 
\begin{bmatrix}
\dot{v}_{t,0}^{(1)}\dot{v}_{t,0}^{(1)\prime} & \dot{v}_{t,0}^{(1)}\dot{v}%
_{t,1}^{(1)\prime}\hat{e}_{1,it} \\ 
\dot{v}_{t,1}^{(1)}\dot{v}_{t,0}^{(1)\prime}\hat{e}_{1,it} & \hat{e}%
_{1,it}^{2}\dot{v}_{t,1}^{(1)}\dot{v}_{t,1}^{(1)\prime}%
\end{bmatrix}
\\
:& =F_{1,i}+F_{2,i}+F_{3,i}.
\end{align*}%
We define $F_{1,i}^{m}$ for $m=\left\{ 1,2,3,4\right\} $ as fourth clockwise
blocks in $F_{1,i}:=%
\begin{bmatrix}
F_{1,i}^{1} & F_{1,i}^{2} \\ 
F_{1,i}^{4} & F_{1,i}^{3}%
\end{bmatrix}%
.$ Define $F_{2,i}^{m}$ and $F_{3,i}^{m}$ similarly. We aim to show the
uniform bound for each block.

First, we observe that 
\begin{align}
& F_{1,i}^{1}=\frac{1}{T}\sum_{t=1}^{T}f_{it}(0)\left(\dot{v}_{t,0}^{(1)}%
\dot{v}_{t,0}^{(1)\prime}-O_{0}^{(1)}v_{t,0}^{0}v_{t,0}^{0\prime}O_{0}^{(1)%
\prime}\right)  \notag  \label{Lem21.1} \\
& =\frac{1}{T}\sum_{t=1}^{T}f_{it}(0)\left[ \left(\dot{v}%
_{t,0}^{(1)}-O_{0}^{(1)}v_{t,0}^{0}\right) \left(\dot{v}%
_{t,0}^{(1)}-O_{0}^{(1)}v_{t,0}^{0}\right)
^{\prime}+O_{0}^{(1)}v_{t,0}^{0}\left(\dot{v}%
_{t,0}^{(1)}-O_{0}^{(1)}v_{t,0}^{0}\right) ^{\prime}+\left(\dot{v}%
_{t,0}^{(1)}-O_{0}^{(1)}v_{t,0}^{0}\right)\left(O_{0}^{(1)}v_{t,0}^{0}%
\right) ^{\prime}\right]  \notag \\
& =O_{p}\left(\eta_{N}^{2}\right) +O_{p}(\eta_{N})=O_{p}(\eta_{N})
\end{align}%
uniformly over $i\in I_{3}$. For $F_{1,i}^{2}$, we have by Theorem \ref{Thm2}
and Lemma \ref{Lem20}, 
\begin{align}
F_{1,i}^{2}& =\frac{1}{T}\sum_{t=1}^{T}f_{it}(0)\left(\dot{v}_{t,0}^{(1)}%
\dot{v}_{t,1}^{(1)\prime}\hat{e}_{1,it}-O_{0}^{(1)}v_{t,0}^{0}v_{t,1}^{0%
\prime}O_{1}^{(1)\prime}e_{1,it}\right)  \notag  \label{Lem21.2} \\
& =\frac{1}{T}\sum_{t=1}^{T}f_{it}(0)\Bigg\{\dot{v}_{t,0}^{(1)}\dot{v}%
_{t,1}^{(1)\prime}\left(\hat{e}_{1,it}-e_{1,it}\right) +\left(\dot{v}%
_{t,0}^{(1)}\dot{v}_{t,1}^{(1)\prime}-O_{0}^{(1)}v_{t,0}^{0}v_{t,1}^{0%
\prime}O_{1}^{(1)\prime}\right) e_{1,it}\Bigg \}  \notag \\
& =O_{p}(\eta_{N})\text{ uniformly in }i\in I_{3}.
\end{align}%
The same order holds for $\max_{i\in I_{3}}\left\Vert
F_{1,i}^{3}\right\Vert_{F}$ as $F_{1,i}^{3}=F_{1,i}^{2\prime }$. Next, we
study $F_{1,i}^{4}$. Noting that 
\begin{equation*}
|\hat{e}_{1,it}^{2}-e_{1,it}^{2}|=\left\vert \left(\hat{e}%
_{1,it}+e_{1,it}\right) \left(\hat{e}_{1,it}-e_{1,it}\right)\right\vert \leq
\left(2|e_{1,it}|+\max_{i\in I_{3},t\in[T]}|\hat{e}_{1,it}-e_{1,it}|\right)
\left(\max_{i\in I_{3},t\in[T]}|\hat{e}_{1,it}-e_{1,it}|\right) ,
\end{equation*}%
we have 
\begin{align}
\max_{i\in I_{3}}||F_{1,i}^{4}||_{F}& =\left\Vert \frac{1}{T}%
\sum_{t=1}^{T}f_{it}(0)\left(\dot{v}_{t,1}^{(1)}\dot{v}_{t,1}^{(1)\prime}%
\hat{e}_{1,it}^{2}-O_{1}^{(1)}v_{t,1}^{0}v_{t,1}^{0\prime}O_{1}^{(1)%
\prime}e_{1,it}^{2}\right) \right\Vert_{F}  \notag  \label{Lem21.3} \\
& \leq \max_{i\in I_{3}}\frac{1}{T}\sum_{t=1}^{T}f_{it}(0)\left\Vert \dot{v}%
_{t,1}^{(1)}\right\Vert_{2}^{2}\left(|\hat{e}_{1,it}^{2}-e_{1,it}^{2}|%
\right) +\max_{i\in I_{3}}\frac{1}{T}\sum_{t=1}^{T}f_{it}(0)\left\Vert \dot{v%
}_{t,1}^{(1)}\dot{v}_{t,1}^{(1)\prime}-O_{1}^{(1)}v_{t,1}^{0}v_{t,1}^{0%
\prime}O_{1}^{(1)\prime}\right\Vert_{F}e_{1,it}^{2}  \notag \\
& =\max_{t\in[T]}\left\Vert \dot{v}_{t,1}^{(1)}\right\Vert_{2}^{2}\left(%
\max_{i\in I_{3},t\in[T]}|\hat{e}_{1,it}-e_{1,it}|\right) \max_{i\in
I_{3}}\left(\frac{2}{T}\sum_{t=1}^{T}f_{it}(0)|e_{1,it}|\right)  \notag \\
& +\max_{t\in[T]}\left\Vert \dot{v}_{t,1}^{(1)}\right\Vert_{2}^{2}\left(%
\max_{i\in I_{3},t\in[T]}|\hat{e}_{1,it}-e_{1,it}|^{2}\right) \max_{i\in
I_{3}}\left(\frac{2}{T}\sum_{t=1}^{T}f_{it}(0)\right)  \notag \\
& +\max_{t\in[T]}\left\Vert \dot{v}_{t,1}^{(1)}\dot{v}_{t,1}^{(1)%
\prime}-O_{1}^{(1)}v_{t,1}^{0}v_{t,1}^{0\prime}O_{1}^{(1)\prime}\right%
\Vert_{F}\max_{i\in I_{3}}\left(\frac{1}{T}%
\sum_{t=1}^{T}f_{it}(0)e_{1,it}^{2}\right) =O_{p}(\eta_{N}).
\end{align}%
Combining (\ref{Lem21.1})-(\ref{Lem21.3}), we conclude $\max_{i\in
I_{3}}\left\Vert F_{1,i}\right\Vert_{2}=O_{p}(\eta_{N}).$

For $F_{2,i}^{2}$, we note that $\mathbb{E}%
\left(f_{it}(0)v_{t,0}^{0}v_{t,1}^{0\prime}e_{1,it}\bigg|\mathscr{D}\right)
=0$ by Assumption \ref{ass:10} and $\max_{i\in I_{3},t\in[T]}\left\Vert
f_{it}(0)v_{t,0}^{0}v_{t,1}^{0\prime}e_{1,it}\right\Vert_{F}\leq
c_{22}\xi_{N}~a.s.$ by Assumption \ref{ass:1}(iv) and Lemma \ref{Lem:bounded
u&v_tilde}(i). Then, by conditional Bernstein's inequality in Lemma \ref%
{Lem:Bern}(i) and Assumption \ref{ass:10}(iii), we can show that, with
positive constants $c_{21}$ and $c_{22}$, 
\begin{align}
& \mathbb{P}\left(\max_{i\in I_{3}}\left\Vert \frac{1}{T}\sum_{t\in
[T]}f_{it}(0)v_{t,0}^{0}v_{t,1}^{0\prime}e_{1,it}\right\Vert_{F}>c_{21}\sqrt{%
\frac{\log (N\vee T)}{T}}\xi_{N}\bigg|\mathscr{D}\right)  \notag
\label{Lem:F2} \\
& \leq \sum_{i\in I_{3}}\exp \left(-\frac{c_{12}c_{21}^{2}T\xi_{N}^{2}\log(N%
\vee T)}{c_{22}^{2}T\xi_{N}^{2}+c_{21}c_{22}\sqrt{T\log (N\vee T)}%
\xi_{N}^{2}\log T\log \log T}\right) =o(1),
\end{align}%
which yields $\max_{i\in I_{3}}\left\Vert
F_{2,i}\right\Vert_{F}=O_{p}(\eta_{N})$.

As for $F_{3,i}$, we show the bound for the first block and all other three
blocks follow the same argument. Recall that 
\begin{align}
\dot{\iota}_{it}& =u_{i,0}^{0\prime}O_{0}^{(1)}\dot{v}%
_{t,0}^{(1)}-u_{i,0}^{0\prime}v_{t,0}^{0}+\hat{\mu}_{1,it}\dot{u}%
_{i,1}^{(1)\prime}\dot{v}_{t,1}^{(1)}-\mu_{1,it}u_{i,1}^{0\prime}v_{t,1}^{0}+%
\hat{e}_{1,it}u_{i,1}^{0\prime}O_{1}^{(1)}\dot{v}%
_{t,1}^{(1)}-e_{1,it}u_{i,1}^{0\prime}v_{t,1}^{0}  \notag  \label{Lem21.4} \\
& =\left[ \left(O_{0}^{(1)}u_{i,0}^{0}\right) ^{\prime}\left(\dot{v}%
_{t,0}^{(1)}-O_{0}^{(1)}v_{t,0}^{0}\right) \right] +\left(\hat{\mu}_{1,it}%
\dot{u}_{i,1}^{(1)\prime}\dot{v}_{t,1}^{(1)}-\mu_{1,it}u_{i,1}^{0%
\prime}v_{t,1}^{0}\right) +\left(\hat{e}_{1,it}u_{i,1}^{0\prime}O_{1}^{(1)}%
\dot{v}_{t,1}^{(1)}-e_{1,it}u_{i,1}^{0\prime}v_{t,1}^{0}\right)  \notag \\
& :=\dot{\iota}_{it}^{I}+\dot{\iota}_{it}^{II}+\dot{\iota}_{it}^{III},
\end{align}
with the fact that $\left\vert \dot{\iota}_{it}^{I}\right\vert
\leq\left\Vert O_{0}^{(1)}u_{i,0}^{0}\right\Vert_{2}\left\Vert \dot{v}%
_{t,0}^{(1)}-O_{0}^{(1)}v_{t,0}^{0}\right\Vert_{2}:=R_{1,it}^{I}$ such that $%
\max_{i\in I_{3},t\in[T]}R_{1,it}^{I}=O_{p}(\eta_{N})$. In addition, we have 
\begin{align}
\max_{i\in I_{3},t\in[T]}\left\vert \dot{\iota}_{it}^{II}\right\vert &
=\max_{i\in I_{3},t\in[T]}\left\vert \left(\hat{\mu}_{1,it}-\mu_{1,it}%
\right) \dot{u}_{i,1}^{(1)\prime}\dot{v}_{t,1}^{(1)}+\mu_{1,it}\left(\dot{u}%
_{i,1}^{(1)\prime}\dot{v}_{t,1}^{(1)}-u_{i,1}^{0\prime}v_{t,1}^{0}\right)
\right\vert  \notag  \label{Lem21.5} \\
& \leq R_{1,it}^{II}+R_{2,it}^{II}|\mu_{1,it}|,
\end{align}%
where $\max_{i\in I_{3},t\in[T]}|R_{1,it}^{II}|=O_{p}(\eta_{N})$ and $%
\max_{i\in I_{3},t\in[T]}|R_{2,it}^{II}|=O_{p}(\eta_{N})$. 
Similarly, we have 
\begin{align}
\left\vert \dot{\iota}_{it}^{III}\right\vert & =\left\vert \hat{e}%
_{1,it}u_{i,1}^{0\prime}O_{1}^{(1)}\dot{v}_{t,1}^{(1)}-e_{1,it}u_{i,1}^{0%
\prime}v_{t,1}^{0}\right\vert  \notag  \label{Lem21.6} \\
& \leq \left\vert (\hat{e}_{1,it}-e_{1,it})\left(O_{1}^{(1)}u_{i,1}^{0}%
\right) ^{\prime}\left(\dot{v}_{t,1}^{(1)}-O_{1}^{(1)}v_{t,1}^{0}\right) +(%
\hat{e}_{1,it}-e_{1,it})\left(O_{1}^{(1)}u_{i,1}^{0}\right)
^{\prime}O_{1}^{(1)}v_{t,1}^{0}\right\vert  \notag \\
& +\left\vert e_{1,it}\left(O_{1}^{(1)}u_{i,1}^{0}\right) ^{\prime}\left(%
\dot{v}_{t,1}^{(1)}-O_{1}^{(1)}v_{t,1}^{0}\right) \right\vert  \notag \\
& =R_{1,it}^{III}+R_{2,it}^{III}|e_{1,it}|,
\end{align}%
where $\max_{i\in I_{3},t\in[T]}|R_{1,it}^{III}|=O_{p}(\eta_{N})$ and $%
\max_{i\in I_{3},t\in[T]}|R_{2,it}^{III}|=O_{p}(\eta_{N})$. 
Therefore, we have 
\begin{align}
\max_{i\in I_{3}}||F_{3,i}^{1}||_{F}& \leq \frac{1}{T}\sum_{t=1}^{T}\left%
\vert f_{it}\left(\dot{\iota}_{it}\right) -f_{it}(0)\right\vert \left\Vert 
\dot{v}_{t,0}^{(1)}\right\Vert_{2}^{2}  \notag  \label{Lem21.8} \\
& \lesssim \max_{t\in[T]}\left\Vert \dot{v}_{t,0}^{(1)}\right\Vert_{2}%
\max_{i\in I_{3},t\in
[T]}(R_{1,it}^{I}+R_{1,it}^{II}+R_{2,it}^{II}+R_{1,it}^{III}+R_{2,it}^{III})%
\frac{1}{T}\sum_{t\in[T]}(1+|e_{1,it}|+|\mu_{1,it}|)  \notag \\
& =O_{p}(\eta_{N}),
\end{align}%
%
%
%
%
Combining all results above, we obtain that $\max_{i\in I_{3}}||\hat{D}%
_{i}^{F}-D_{i}^{F}||_{F}=O_{p}(\eta_{N})$.
\end{proof}

\begin{lemma}
{\small \label{Lem22} } Under Assumptions \ref{ass:1}-\ref{ass:10}, for
matrices $\dot{D}_{i}^{J}$ and $D_{i}^{J}$ defined in the proof of Theorem %
\ref{Thm3}, we have 
\begin{align*}
\max_{i \in I_3}\left\Vert\dot{D}_{i}^{J}-D_{i}^{J}\right\Vert_{F}=\left%
\Vert 
\begin{bmatrix}
O_{p}\left(\eta_{N}^{2}\right) & O_{p}(\eta_{N}) \\ 
O_{p}\left(\eta_{N}^{2}\right) & O_{p}\left(\eta_{N}^{2}\right)%
\end{bmatrix}
\right\Vert_{F}.
\end{align*}
\end{lemma}

\begin{proof}
Recall that 
\begin{equation*}
D_{i}^{J}=\frac{1}{T}\sum_{t=1}^{T}f_{it}(0)%
\begin{bmatrix}
\dot{v}_{t,0}^{(1)}\left(O_{0}^{(1)}v_{t,0}^{0}-\dot{v}_{t,0}^{(1)}\right)^{%
\prime} & 0 \\ 
0 & e_{1,it}^{2}O_{1}^{(1)}v_{t,1}^{0}\left(O_{1}^{(1)}v_{t,1}^{0}-\dot{v}%
_{t,1}^{(1)}\right) ^{\prime}%
\end{bmatrix}%
.
\end{equation*}%
%
%
%
%
We have 
\begin{align*}
& \hat{D}_{i}^{J}-D_{i}^{J} \\
& =\frac{1}{T}\sum_{t=1}^{T}f_{it}(0)%
\begin{bmatrix}
0 & \dot{v}_{t,0}^{(1)}v_{t,1}^{0\prime}O_{1}^{(1)\prime}\left(e_{1,it}-\hat{%
e}_{1,it}\right) +\dot{v}_{t,0}^{(1)}\left(O_{1}^{(1)}v_{t,1}^{0}-\dot{v}%
_{t,1}^{(1)}\right) ^{\prime}\left(\hat{e}_{1,it}-e_{1,it}\right) \\ 
\left(\hat{e}_{1,it}-e_{1,it}\right) \dot{v}_{t,1}^{(1)}%
\left(O_{0}^{(1)}v_{t,0}^{0}-\dot{v}_{t,0}^{(1)}\right) ^{\prime} & \dot{v}%
_{t,1}^{(1)}v_{t,1}^{0\prime}O_{1}^{(1)\prime}e_{1,it}\left(\hat{e}%
_{1,it}-e_{1,it}\right) -\dot{v}_{t,1}^{(1)}\dot{v}_{t,1}^{(1)\prime}\left(%
\hat{e}_{1,it}^{2}-e_{1,it}^{2}\right)%
\end{bmatrix}
\\
& +\frac{1}{T}\sum_{t=1}^{T}f_{it}(0)%
\begin{bmatrix}
0 & \dot{v}_{t,0}^{(1)}\left(O_{1}^{(1)}v_{t,1}^{0}-\dot{v}%
_{t,1}^{(1)}\right) ^{\prime}e_{1,it} \\ 
\dot{v}_{t,1}^{(1)}\left(O_{0}^{(1)}v_{t,0}^{0}-\dot{v}_{t,0}^{(1)}\right)^{%
\prime}e_{1,it} & \left(\dot{v}_{t,1}^{(1)}-O_{1}^{(1)}v_{t,1}^{0}\right)%
\left(O_{1}^{(1)}v_{t,1}^{0}-\dot{v}_{t,1}^{(1)}\right) ^{\prime}e_{1,it}^{2}%
\end{bmatrix}
\\
& +\frac{1}{T}\sum_{t=1}^{T}\left[ f_{it}(\tilde{\iota}_{it})-f_{it}(0)%
\right] 
\begin{bmatrix}
\dot{v}_{t,0}^{(1)}\left(O_{0}^{(1)}v_{t,0}^{0}-\dot{v}_{t,0}^{(1)}\right)^{%
\prime} & \dot{v}_{t,0}^{(1)}\left(e_{1,it}O_{1}^{(1)}v_{t,1}^{0}-\hat{e}%
_{1,it}\dot{v}_{t,1}^{(1)}\right) ^{\prime} \\ 
\hat{e}_{1,it}\dot{v}_{t,1}^{(1)}\left(O_{0}^{(1)}v_{t,0}^{0}-\dot{v}%
_{t,0}^{(1)}\right) ^{\prime} & \hat{e}_{1,it}\dot{v}_{t,1}^{(1)}%
\left(e_{1,it}O_{1}^{(1)}v_{t,1}^{0}-\hat{e}_{1,it}\dot{v}%
_{t,1}^{(1)}\right) ^{\prime}%
\end{bmatrix}
\\
:& =J_{1,i}+J_{2,i}+J_{3,i}.
\end{align*}%
As in the proof of the last lemma, we define $J_{1,i}^{m}$, $J_{2,i}^{m}$
and $J_{3,i}^{m}$ for $m=1,2,3,4$ as four clockwise blocks in $J_{1,i}$, $%
J_{2,i}$ and $J_{3,i}$, respectively.

First, we study $J_{1,i}.$ For $J_{1,i}^{2}$, we notice that 
\begin{align}
J_{1,i}^{2}& =\frac{1}{T}\sum_{t=1}^{T}f_{it}(0)\dot{v}%
_{t,0}^{(1)}v_{t,1}^{0\prime}O_{1}^{(1)\prime}\left(e_{1,it}-\hat{e}%
_{1,it}\right) +\frac{1}{T}\sum_{t=1}^{T}f_{it}(0)\dot{v}_{t,0}^{(1)}%
\left(O_{1}^{(1)}v_{t,1}^{0}-\dot{v}_{t,1}^{(1)}\right) ^{\prime}\left(\hat{e%
}_{1,it}-e_{1,it}\right)  \notag  \label{Lem:J_eta} \\
& =\frac{1}{T}\sum_{t=1}^{T}f_{it}(0)O_{0}^{(1)}v_{t,0}^{0}v_{t,1}^{0%
\prime}O_{1}^{(1)\prime}\left(e_{1,it}-\hat{e}_{1,it}\right)  \notag \\
& +\frac{1}{T}\sum_{t=1}^{T}f_{it}(0)\left(\dot{v}%
_{t,0}^{(1)}-O_{0}^{(1)}v_{t,0}^{0}\right)
v_{t,1}^{0\prime}O_{1}^{(1)\prime}\left(e_{1,it}-\hat{e}_{1,it}\right)
+O_{p}\left(\eta_{N}^{2}\right)  \notag \\
& =\frac{1}{T}\sum_{t=1}^{T}f_{it}(0)O_{0}^{(1)}v_{t,0}^{0}v_{t,1}^{0%
\prime}O_{1}^{(1)\prime}\left(e_{1,it}-\hat{e}_{1,it}\right)
+O_{p}\left(\eta_{N}^{2}\right) +O_{p}\left(\eta_{N}^{2}\right)  \notag \\
& =O_{p}(\eta_{N})\quad \text{uniformly over }i\in I_{3}.
\end{align}%
Noted that the leading term in $J_{1,i}^{2}$ is $\frac{1}{T}%
\sum_{t=1}^{T}f_{it}(0)O_{0}^{(1)}v_{t,0}^{0}v_{t,1}^{0\prime}O_{1}^{(1)%
\prime}\left(e_{1,it}-\hat{e}_{1,it}\right)$, which will remain as the bias
term of $\hat{u}_{i,0}^{(3,1)}$.

Furthermore, it is clear that 
\begin{equation*}
J_{1,i}^{3}=\frac{1}{T}\sum_{t=1}^{T}f_{it}(0)\left(\hat{e}%
_{1,it}-e_{1,it}\right) \dot{v}_{t,1}^{(1)}\left(O_{0}^{(1)}v_{t,0}^{0}-\dot{%
v}_{t,0}^{(1)}\right) ^{\prime}=O_{p}\left(\eta_{N}^{2}\right) \quad \text{%
uniformly over }i\in I_{3}.
\end{equation*}
Next, we obtain that 
\begin{align}
J_{1,i}^{4}& =\frac{1}{T}\sum_{t=1}^{T}f_{it}(0)\dot{v}%
_{t,1}^{(1)}v_{t,1}^{0\prime}O_{1}^{(1)\prime}e_{1,it}\left(\hat{e}%
_{1,it}-e_{1,it}\right) -\frac{1}{T}\sum_{t=1}^{T}f_{it}(0)\dot{v}%
_{t,1}^{(1)}\dot{v}_{t,1}^{(1)\prime}\left(\hat{e}_{1,it}^{2}-e_{1,it}^{2}%
\right)  \notag  \label{Lem22.4} \\
& =O_{1}^{(1)}\left[ \frac{1}{T}\sum_{t=1}^{T}f_{it}(0)v_{t,1}^{0}v_{t,1}^{0%
\prime}e_{1,it}\left(\hat{e}_{1,it}-e_{1,it}\right) \right]
O_{1}^{(1)\prime}-\frac{1}{T}\sum_{t=1}^{T}f_{it}(0)\dot{v}_{t,1}^{(1)}\dot{v%
}_{t,1}^{(1)\prime}\left(\hat{e}_{1,it}-e_{1,it}\right) ^{2}  \notag \\
& -\frac{2}{T}\sum_{t=1}^{T}f_{it}(0)\dot{v}_{t,1}^{(1)}\dot{v}%
_{t,1}^{(1)\prime}e_{1,it}\left(\hat{e}_{1,it}-e_{1,it}\right)
+O_{p}(\eta_{N}^{2})  \notag \\
& =-O_{1}^{(1)}\left[ \frac{1}{T}%
\sum_{t=1}^{T}f_{it}(0)v_{t,1}^{0}v_{t,1}^{0\prime}e_{1,it}\left(\hat{e}%
_{1,it}-e_{1,it}\right) \right] O_{1}^{(1)\prime}+O_{p}(\eta_{N}^{2}),
\end{align}
where the first and second equalities hold by Theorem \ref{Thm2}(ii) and
Lemma \ref{Lem20}. We deal with the first term in the second equality of (%
\ref{Lem22.4}) by inserting the linear expansion of $\hat{e}_{1,it}-e_{1,it}$
in Lemma \ref{Lem20}(iii), i.e., 
\begin{align}
& \frac{1}{T}\sum_{t=1}^{T}f_{it}(0)v_{t,1}^{0}v_{t,1}^{0\prime}e_{1,it}%
\left(\hat{e}_{1,it}-e_{1,it}\right)  \notag  \label{Lem22.5} \\
& =-\frac{1}{T}\sum_{t=1}^{T}\left\{
f_{it}(0)v_{t,1}^{0}v_{t,1}^{0\prime}e_{1,it}\left[ w_{1,t}^{0\prime}\left(%
\frac{W_{1}^{0\prime}W_{1}^{0}}{T}\right) ^{-1}\frac{1}{T}%
\sum_{t=1}^{T}w_{1,t}^{0}e_{1,it}\right] \right\}  \notag \\
& -\frac{1}{T}\sum_{t=1}^{T}\left\{ f_{it}(0)v_{t,1}^{0}v_{t,1}^{0\prime}%
\left[ e_{1,it}l_{1,i}^{0\prime}\frac{1}{N}\sum_{i^{*}\in[N]%
}l_{1,i^{*}}^{0}e_{1,i^{*}t}\right] \right\} +O_{p}\left(\frac{\log (N\vee T)%
}{N\wedge T}\right)
\end{align}
uniformly over $i\in I_{3}$. For the first term in the right side of (\ref%
{Lem22.5}), we notice that 
\begin{align*}
& \max_{i\in I_{3}}\left\Vert \frac{1}{T}\sum_{t=1}^{T}\left%
\{f_{it}(0)v_{t,1}^{0}v_{t,1}^{0\prime}e_{1,it}\left[ w_{1,t}^{0\prime}\left(%
\frac{W_{1}^{0\prime}W_{1}^{0}}{T}\right) ^{-1}\frac{1}{T}%
\sum_{t=1}^{T}w_{1,t}^{0}e_{1,it}\right] \right\} \right\Vert_{F} \\
& =\max_{i\in I_{3}}\left\Vert \frac{1}{T}\sum_{t=1}^{T}\left%
\{f_{it}(0)v_{t,1}^{0}v_{t,1}^{0\prime}e_{1,it}\left[ \sum_{k\in
[r_{1}]}w_{1,t,k}^{0}\left[ \left(\frac{W_{1}^{0\prime}W_{1}^{0}}{T}%
\right)^{-1}\frac{1}{T}\sum_{t=1}^{T}w_{1,t}^{0}e_{1,it}\right]_{k}\right]%
\right\} \right\Vert_{F} \\
& \leq \sum_{k\in [r_{1}]}\max_{i\in I_{3}}\left\Vert \left[ \frac{1}{T}%
\sum_{t=1}^{T}f_{it}(0)v_{t,1}^{0}v_{t,1}^{0\prime}w_{1,t,k}^{0}e_{1,it}%
\right] \right\Vert_{F}\max_{k\in [r_{1}],i\in I_{3}}\left\vert \left[ \left(%
\frac{W_{1}^{0\prime}W_{1}^{0}}{T}\right) ^{-1}\frac{1}{T}%
\sum_{t=1}^{T}w_{1,t}^{0}e_{1,it}\right]_{k}\right\vert \\
& =\sum_{k\in [r_{1}]}\max_{i\in I_{3}}\left\Vert \left[ \frac{1}{T}%
\sum_{t=1}^{T}f_{it}(0)v_{t,1}^{0}v_{t,1}^{0\prime}w_{1,t,k}^{0}e_{1,it}%
\right] \right\Vert_{F}O_{p}(\eta_{N}) \\
& =O_{p}(\eta_{N}^{2}),
\end{align*}%
where $w_{1,t}$ is the fact of $X_{1,it}$, $r_{1}$ is the dimension of $%
w_{1,t}$, and $w_{1,t,k}$ is the $k$th element of $w_{1,it}$, the second
equality holds by the results in (\ref{Lem20.5}) and (\ref{Lem:we}), and the
last equality is by the fact that 
\begin{equation*}
\mathbb{P}\left(\max_{i\in I_{3}}\left\Vert \left[ \frac{1}{T}%
\sum_{t=1}^{T}f_{it}(0)v_{t,1}^{0}v_{t,1}^{0\prime}w_{1,t,k}^{0}e_{1,it}%
\right] \right\Vert_{F}>c_{23}\sqrt{\frac{\log (N\vee T)}{T}}\xi_{N}\bigg|%
\mathscr{D}\right) =o(1)
\end{equation*}%
following analogous analysis as (\ref{Lem:F2}) with a positive constant $%
c_{23}$.

For the second term on the RHS of equation (\ref{Lem22.5}), we notice that 
\begin{align}
& \frac{1}{T}\sum_{t=1}^{T}\left\{ f_{it}(0)v_{t,1}^{0}v_{t,1}^{0\prime}%
\left[e_{1,it}l_{1,i}^{0\prime}\frac{1}{N}\sum_{i^{*}\in
N}l_{1,i^{*}}^{0}e_{1,i^{*}t}\right] \right\}  \notag  \label{Lem:J1.4} \\
& =\frac{1}{NT}\sum_{i^{*}\in [N]}\sum_{t\in
[T]}f_{it}(0)v_{t,1}^{0}v_{t,1}^{0\prime}l_{1,i}^{0%
\prime}l_{1,i^{*}}^{0}e_{1,it}e_{1,i^{*}t}  \notag \\
& =\frac{1}{NT}\sum_{t\in[T]}f_{it}(0)v_{t,1}^{0}v_{t,1}^{0\prime}l_{1,i}^{0%
\prime}l_{1,i}^{0}e_{1,it}^{2}+\frac{1}{NT}\sum_{i^{*}\neq i}\sum_{t\in[T]%
}f_{it}(0)v_{t,1}^{0}v_{t,1}^{0\prime}l_{1,i}^{0%
\prime}l_{1,i^{*}}^{0}e_{1,it}e_{1,i^{*}t}
\end{align}%
with 
\begin{align*}
\max_{i\in I_{3}}\left\Vert \frac{1}{NT}\sum_{t\in
[T]}f_{it}(0)v_{t,1}^{0}v_{t,1}^{0\prime}l_{1,i}^{0%
\prime}l_{1,i}^{0}e_{1,it}^{2}\right\Vert & \leq \frac{1}{N}\max_{i\in
I_{3}}\left\Vert l_{1,i}^{0}\right\Vert_{2}^{2}\max_{t\in[T]%
}\left(\left\Vert v_{t,1}^{0}\right\Vert_{2}\left\Vert
v_{t,1}^{0}\right\Vert_{2}\right) \frac{1}{T}\sum_{t\in[T]}e_{1,it}^{2} \\
& =O_{p}\left(\frac{\xi_{N}^{2}}{N}\right) ,
\end{align*}%
where the equality holds by Assumption \ref{ass:1}(iv), Lemma \ref%
{Lem:bounded u&v_tilde}(i) and Theorem \ref{Thm2}(ii). For the second term
in (\ref{Lem:J1.4}), by Assumption \ref{ass:1}(iii), $e_{j,it}$ is strong
mixing across $t$ and independent given fixed effect. We define $%
E_{i}^{vec}=(\tilde{e}_{1,1}^{\prime},\cdots ,\tilde{e}_{1,i-1}^{\prime},%
\tilde{e}_{1,i+1}^{\prime},\cdots ,\tilde{e}_{1,N}^{\prime})^{\prime}$ with $%
\tilde{e}_{1,i^{*}}=\left(e_{1,i^{*}1}e_{1,i1},\cdots,e_{1,i^{*}T}e_{1,iT}%
\right) $ for $i^{*}\neq i$. We can see $E_{i}^{vec}$ will also be a strong
mixing sequence, conditional on $\mathscr{D}$, which implies 
\begin{align*}
& \mathbb{P}\left(\max_{i\in I_{3}}\left\Vert \frac{1}{NT}\sum_{i^{*}\neq
i}\sum_{t\in[T]}f_{it}(0)\dot{v}_{t,1}^{(1)}v_{t,1}^{0\prime}O_{1}^{(1)%
\prime}l_{1,i}^{0\prime}l_{1,i^{*}}^{0}e_{1,it}e_{1,i^{*}t}\right%
\Vert_{F}>c_{23}\eta_{N}^{2}\right) \\
& \leq \sum_{i\in I_{3}}\mathbb{E}\mathbb{P}\left(\left\Vert \frac{1}{NT}%
\sum_{i^{*}\neq i}\sum_{t\in[T]}f_{it}(0)\dot{v}_{t,1}^{(1)}v_{t,1}^{0%
\prime}O_{1}^{(1)\prime}l_{1,i}^{0\prime}l_{1,i^{*}}^{0}e_{1,it}e_{1,i^{*}t}%
\right\Vert_{F}>c_{23}\eta_{N}^{2}\bigg|\mathscr{D}\right) =o(1),
\end{align*}%
where the equality holds by Lemma \ref{Lem:Bern}(i). 
This implies 
\begin{align*}
\max_{i\in I_{3}}\left\Vert O_{1}^{(1)}\left[ \frac{1}{T}%
\sum_{t=1}^{T}f_{it}(0)v_{t,1}^{0}v_{t,1}^{0\prime}e_{1,it}\left(\hat{e}%
_{1,it}-e_{1,it}\right) \right] O_{1}^{(1)\prime}\right\Vert_{F}&=\max_{i\in
I_{3}}\left\Vert \frac{1}{T}\sum_{t=1}^{T}f_{it}(0)v_{t,1}^{0}v_{t,1}^{0%
\prime}e_{1,it}\left(\hat{e}_{1,it}-e_{1,it}\right) \right\Vert_{F} \\
& =O_{p}(\eta_{N}^{2}).
\end{align*}
Therefore, we conclude that $\max_{i\in I_{3}}\left\Vert
J_{1,i}^{4}\right\Vert_{F}=O_{p}\left(\eta_{N}^{2}\right)$. Then 
\begin{equation*}
\max_{i\in I_{3}}\left\Vert J_{1,i}\right\Vert_{F}=\left\Vert 
\begin{bmatrix}
0 & O_{p}(\eta_{N}) \\ 
O_{p}\left(\eta_{N}^{2}\right) & O_{p}\left(\eta_{N}^{2}\right)%
\end{bmatrix}%
\right\Vert_{F}.
\end{equation*}

Next, for $J_{2,i}^{2}$ and $J_{2,i}^{3}$, conditioning on $\mathscr{D}%
^{I_{1}\cup I_{2}}$ and following Bernstein's inequality in Lemma \ref%
{Lem:Bern}(i), we have 
\begin{align*}
\max_{i\in I_{3}}\left\Vert J_{2,i}^{2}\right\Vert_{F}& =\max_{i\in
I_{3}}\left\Vert \frac{1}{T}\sum_{t=1}^{T}f_{it}(0)\dot{v}%
_{t,0}^{(1)}\left(O_{1}^{(1)}v_{t,1}^{0}-\dot{v}_{t,1}^{(1)}\right)
^{\prime}e_{1,it}\right\Vert_{F}=O_{p}\left(\eta_{N}^{2}\right) , \\
\max_{i\in I_{3}}\left\Vert J_{2,i}^{3}\right\Vert_{F}& =\max_{i\in
I_{3}}\left\Vert \frac{1}{T}\sum_{t=1}^{T}f_{it}(0)\dot{v}%
_{t,1}^{(1)}\left(O_{0}^{(1)}v_{t,0}^{0}-\dot{v}_{t,0}^{(1)}\right)
^{\prime}e_{1,it}\right\Vert_{F}=O_{p}\left(\eta_{N}^{2}\right) ,
\end{align*}%
which can be obtained by the similar arguments as Lemma \ref{Lem23}(i).
These results, in conjunction with the fact that 
\begin{align*}
\max_{i\in I_{3}}\left\Vert J_{2,i}^{4}\right\Vert_{F}& =\max_{i\in
I_{3}}\left\Vert \frac{1}{T}\sum_{t=1}^{T}f_{it}(0)\left(\dot{v}%
_{t,1}^{(1)}-O_{1}^{(1)}v_{t,1}^{0}\right) \left(O_{1}^{(1)}v_{t,1}^{0}-\dot{%
v}_{t,1}^{(1)}\right) ^{\prime}e_{1,it}^{2}\right\Vert_{F} \\
& \leq \max_{i\in I_{3},t\in[T]}\left(\frac{1}{T}%
\sum_{t=1}^{T}f_{it}(0)e_{1,it}^{2}\right) \max_{t\in[T]}\left\Vert \dot{v}%
_{t,1}^{(1)}-O_{1}^{(1)}v_{t,1}^{0}\right\Vert_{2}^{2}=O_{p}\left(
\eta_{N}^{2}\right) ,
\end{align*}
imply that $\max_{i\in I_{3}}\left\Vert
J_{2,i}\right\Vert_{F}=O_{p}(\eta_{N}^{2}).$

For $J_{3,i}$, we have 
\begin{align}
\max_{i\in I_{3}}\left\Vert J_{3,i}^{1}\right\Vert_{F}& =\max_{i\in
I_{3}}\left\Vert \frac{1}{T}\sum_{t=1}^{T}\left[ f_{it}\left(\tilde{\iota}%
_{it}\right) -f_{it}(0)\right] \dot{v}_{t,0}^{(1)}\left(O_{0}^{%
\prime}v_{t,0}^{0}-\dot{v}_{t,0}^{(1)}\right) ^{\prime}\right\Vert_{F} 
\notag  \label{Lem22.8} \\
& \leq \max_{t\in[T]}\left\Vert \dot{v}_{t,0}^{(1)}\right\Vert_{2}\left\Vert 
\dot{v}_{t,0}^{(1)}\left(O_{0}^{\prime}v_{t,0}^{0}-\dot{v}%
_{t,0}^{(1)}\right) ^{\prime}\right\Vert_{2}\max_{i\in I_{3}}\frac{1}{T}%
\sum_{t=1}^{T}\left\vert f_{it}\left(\tilde{\iota}_{it}\right)-f_{it}(0)%
\right\vert  \notag \\
& \lesssim O_{p}(\eta_{N})\max_{i\in I_{3}}\frac{1}{T}\sum_{t\in
[T]}\left\vert \tilde{\iota}_{it}\right\vert  \notag \\
& =O_{p}\left(\eta_{N}^{2}\right) ,
\end{align}%
where the third line is by Lipschitz continuity of the density function,
Theorem \ref{Thm2}(ii), Assumption \ref{ass:2}, and the fact that $%
\left\vert \tilde{\iota}_{it}\right\vert $ lies between $0$ and $\left\vert%
\dot{\iota}_{it}\right\vert $, and the last line is by the fact that $%
\max_{i\in I_{3}}\frac{1}{T}\sum_{t\in[T]}\left\vert \tilde{\iota}%
_{it}\right\vert =O_{p}\left(\eta_{N}\right) $ by (\ref{Lem21.5}) and (\ref%
{Lem21.6}). The bounds for other three blocks in $J_{3,i}$ can be
established in the same manner. Hence, we have $\max_{i\in
I_{3}}||J_{3,i}||_{F}=O_{p}(\eta_{N}^{2}).$

Combining all results above yields the desired result.
\end{proof}

\begin{lemma}
{\small \label{Lem23} } Under Assumptions \ref{ass:1}-\ref{ass:10}, we have

\begin{itemize}
\item[(i)] $\max_{i\in I_{3}}\left\Vert\frac{1}{T}\sum_{t=1}^{T}\left(\dot{v}%
_{t,0}^{(1)}-O_{0}^{(1)}v_{t,0}^{0}\right)\left(\tau-\mathbf{1}%
\left\{\epsilon_{it}\leq
0\right\}\right)\right\Vert_{F}=O_{p}\left(\eta_{N}^{2}\right)$,

\item[(ii)] $\max_{i\in I_{3}}\left\Vert\frac{1}{T}\sum_{t=1}^{T}\left(\dot{v%
}_{t,1}^{(1)}-O_{0}^{(1)}v_{t,1}^{0}\right)\left(\hat{e}_{1,it}-e_{1,it}%
\right)\left(\tau-\mathbf{1}\left\{\epsilon_{it}\leq
0\right\}\right)\right\Vert_{F}=O_{p}\left(\eta_{N}^{2}\right)$,

\item[(iii)] $\max_{i\in I_{3}}\left\Vert\frac{1}{T}\sum_{t=1}^{T}e_{1,it}%
\left(\dot{v}_{t,1}^{(1)}-O_{0}^{(1)}v_{t,1}^{0}\right)\left(\tau-\mathbf{1}%
\left\{\epsilon_{it}\leq
0\right\}\right)\right\Vert_{F}=O_{p}\left(\eta_{N}^{2}\right)$,

\item[(iv)] $\max_{i\in I_{3}}\left\Vert\frac{1}{T}%
\sum_{t=1}^{T}O_{1}^{(1)}v_{t,1}^{0}\left(\hat{e}_{1,it}-e_{1,it}\right)%
\left(\tau-\mathbf{1}\left\{\epsilon_{it}\leq
0\right\}\right)\right\Vert_{F}=O_{p}\left(\eta_{N}^{2}\right)$.
\end{itemize}
\end{lemma}

\begin{proof}
(i) We notice that $\mathbb{E}\left[ \left(\dot{v}%
_{t,0}^{(1)}-O_{0}^{(1)}v_{t,0}^{0}\right) \left(\tau -\mathbf{1}%
\left\{\epsilon_{it}\leq 0\right\} \right) \bigg|\mathscr{D}^{I_{1}\cup
I_{2}}\right] =0$ by Assumption \ref{ass:1}(ii). Define event $\mathscr{A}%
_{10}(M)=\left\{ \max_{t\in[T]}\left\Vert \dot{v}%
_{t,0}^{(1)}-O_{0}^{(1)}v_{t,0}^{0}\right\Vert_{2}\leq M\xi_{N}\right\} $
with $\mathbb{P}\left(\mathscr{A}_{10}(M)^{c}\right) \leq e$ for any $e>0$
by Theorem \ref{Thm2}(ii). With some positive constant $c_{24}$, it follows
that 
\begin{align}
& \mathbb{P}\left(\max_{i\in I_{3}}\left\Vert \frac{1}{T}\sum_{t=1}^{T}\left(%
\dot{v}_{t,0}^{(1)}-O_{0}^{(1)}v_{t,0}^{0}\right)\left(\tau -\mathbf{1}%
\left\{ \epsilon_{it}\leq 0\right\}
\right)\right\Vert_{F}>c_{24}\eta_{N}^{2}\right)  \notag  \label{Lem23.1} \\
& \leq \mathbb{P}\left(\max_{i\in I_{3}}\left\Vert \frac{1}{T}%
\sum_{t=1}^{T}\left(\dot{v}_{t,0}^{(1)}-O_{0}^{(1)}v_{t,0}^{0}\right)\left(%
\tau -\mathbf{1}\left\{ \epsilon_{it}\leq 0\right\}
\right)\right\Vert_{F}>c_{24}\eta_{N}^{2},\mathscr{A}_{10}(M)\right) +e 
\notag \\
& \leq \sum_{i\in I_{3}}\mathbb{P}\left(\left\Vert \frac{1}{T}%
\sum_{t=1}^{T}\left(\dot{v}_{t,0}^{(1)}-O_{0}^{(1)}v_{t,0}^{0}\right)\left(%
\tau -\mathbf{1}\left\{ \epsilon_{it}\leq 0\right\}
\right)\right\Vert_{F}>c_{24}\eta_{N}^{2},\mathscr{A}_{10}(M)\right) +e 
\notag \\
& \leq \sum_{i\in I_{3}}\mathbb{E}\mathbb{P}\left(\left\Vert \frac{1}{T}%
\sum_{t=1}^{T}\left(\dot{v}_{t,0}^{(1)}-O_{0}^{(1)}v_{t,0}^{0}\right)\left(%
\tau -\mathbf{1}\left\{ \epsilon_{it}\leq 0\right\}
\right)\right\Vert_{F}>c_{24}\eta_{N}^{2}\bigg|\mathscr{D}^{I_{1}\cup
I_{2}}\right) \mathbf{1}\left\{ \mathscr{A}_{10}(M)\right\} +e  \notag \\
& \leq \sum_{i\in I_{3}}\exp \left\{ -\frac{c_{12}c_{24}^{2}T^{2}\eta_{N}^{4}%
}{M^{2}T\eta_{N}^{2}+c_{24}MT\eta_{N}^{2}\log T\log \log T}\right\} +e 
\notag \\
& =o(1)+e.
\end{align}%
Since $e$ can be made arbitrarily small, this completes the proof of
statement (i).

(ii) It's clear that 
\begin{align*}
& \max_{i\in I_{3}}\left\Vert \frac{1}{T}\sum_{t=1}^{T}\left(\dot{v}%
_{t,1}^{(1)}-O_{0}^{(1)}v_{t,1}^{0}\right) \left(\hat{e}_{1,it}-e_{1,it}%
\right) \left(\tau -\mathbf{1}\left\{ \epsilon_{it}\leq 0\right\} \right)
\right\Vert_{F} \\
& \leq \max_{i\in I_{3},t\in[T]}\left\vert \hat{e}_{1,it}-e_{1,it}\right%
\vert \max_{t\in[T]}\left\Vert \dot{v}_{t,1}^{(1)}-O_{0}^{(1)}v_{t,1}^{0}%
\right\Vert_{2}=O_{p}(\eta_{N}^{2}),
\end{align*}%
where the equality holds by Lemma \ref{Lem20}(iii) and Theorem \ref{Thm2}%
(ii).

(iii) Noting that 
\begin{equation*}
\mathbb{E}\left[ e_{1,it}\left(\dot{v}_{t,1}^{(1)}-O_{0}^{(1)}v_{t,1}^{0}%
\right) \left(\tau -\mathbf{1}\left\{ \epsilon_{it}\leq 0\right\} \right) %
\bigg|\mathscr{D}^{I_{1}\cup I_{2}}\right] =0
\end{equation*}%
by law of iterated expectation and Assumption \ref{ass:1}(ii), we can obtain
the desired result as in (\ref{Lem23.1}).

(iv) Note that 
\begin{align}
& \frac{1}{T}\sum_{t\in[T]}v_{t,1}^{0}\left(\hat{e}_{1,it}-e_{1,it}\right)
\left(\tau -\mathbf{1}\left\{ \epsilon_{it}\leq 0\right\} \right)  \notag \\
& =\frac{1}{T^{2}}\sum_{s\in [T]}\sum_{t\in [T]}v_{t,1}^{0}\left(\tau -%
\mathbf{1}\left\{ \epsilon_{it}\leq 0\right\}\right) w_{j,t}^{0\prime}\left(%
\frac{W_{j}^{0\prime}W_{j}^{0}}{T}\right)^{-1}w_{j,s}^{0}e_{j,is}  \notag \\
& +\frac{1}{NT}\sum_{m\in [N]}\sum_{t\in[T]}v_{t,1}^{0}\left(\tau -\mathbf{1}%
\left\{ \epsilon_{it}\leq 0\right\} \right)
l_{j,i}^{0\prime}l_{j,m}^{0}e_{j,ms}+\mathcal{R}_{j,it}.  \label{eq:vehata}
\end{align}%
For the first term on the RHS of \eqref{eq:vehata}, we have 
\begin{align*}
& \max_{i\in I_{3}}\left\Vert \frac{1}{T^{2}}\sum_{s\in [T]}\sum_{t\in[T]%
}v_{t,1}^{0}\left(\tau -\mathbf{1}\left\{ \epsilon_{it}\leq 0\right\}
\right) w_{j,t}^{0\prime}\left(\frac{W_{j}^{0\prime}W_{j}^{0}}{T}\right)
^{-1}w_{j,s}^{0}e_{j,is}\right\Vert_{2} \\
& \leq \max_{i\in I_{3}}\left\Vert \frac{1}{T}\sum_{t\in
[T]}v_{t,1}^{0}\left(\tau -\mathbf{1}\left\{ \epsilon_{it}\leq
0\right\}\right) w_{j,t}^{0\prime}\right\Vert_{F}\left\Vert \frac{1}{T}%
\sum_{s\in[T]}\left(\frac{W_{j}^{0\prime}W_{j}^{0}}{T}%
\right)^{-1}w_{j,s}^{0}e_{j,is}\right\Vert_{2} \\
& =O_{p}\left(\eta_{N}^{2}\right) .
\end{align*}%
For the second term on RHS of \eqref{eq:vehata}, by Assumption \ref{ass:10}%
(ii), Bernstein's inequality in Lemma \ref{Lem:Bern}(i) conditional on
factors, we have 
\begin{equation*}
\max_{i\in I_{3}}\left\Vert \frac{1}{NT}\sum_{m\in [N]}\sum_{t\in[T]%
}v_{t,1}^{0}\left(\tau -\mathbf{1}\left\{ \epsilon_{it}\leq 0\right\}
\right)
l_{j,i}^{0\prime}l_{j,m}^{0}e_{j,ms}\right\Vert_{2}=O_{p}\left(\eta_{N}^{2}%
\right) .
\end{equation*}%
Then statement (iv) follows.
\end{proof}

\begin{lemma}
{\small \label{Lem24} } Under Assumptions \ref{ass:1}-\ref{ass:10}, we have

\begin{itemize}
\item[(i)] $\max_{i\in I_{3}}\left\Vert \frac{1}{T}%
\sum_{t=1}^{T}e_{1,it}O_{1}^{(1)}v_{t,1}^{0}f_{it}(0)\left(%
\mu_{1,it}u_{i,1}^{0\prime}v_{t,1}^{0}-\hat{\mu}_{1,it}\dot{u}%
_{i,1}^{(1)\prime}\dot{v}_{t,1}^{(1)}\right)
\right\Vert_{2}=O_{p}\left(\eta_{N}^{2}\right)$,

\item[(ii)] $\max_{i\in I_{3}}\left\Vert \frac{1}{T}%
\sum_{t=1}^{T}O_{0}^{(1)}v_{t,0}^{0}f_{it}(0)\left(\mu_{1,it}u_{i,1}^{0%
\prime}v_{t,1}^{0}-\hat{\mu}_{1,it}\dot{u}_{i,1}^{(1)\prime}\dot{v}%
_{t,1}^{(1)}\right) \right\Vert_{2}=O_{p}(\eta_{N})$.
\end{itemize}
\end{lemma}

\begin{proof}
(i) Recall from (\ref{Lem21.5}), we have 
\begin{align*}
\hat{\mu}_{1,it}\dot{u}_{i,1}^{(1)\prime}\dot{v}_{t,1}^{(1)}-%
\mu_{1,it}u_{i,1}^{0\prime}v_{t,1}^{0}& =\left(\hat{\mu}_{1,it}-\mu_{1,it}%
\right) u_{i,j}^{0\prime}v_{t,j}^{0}+\mu_{1,it}\left(\dot{u}%
_{i,1}^{(1)}-O_{1}^{(1)}u_{i,1}^{0}\right) O_{1}^{(1)}v_{t,1}^{0} \\
& +\mu_{1,it}\left(O_{1}^{(1)}u_{i,1}^{0}\right) ^{\prime}\left(\dot{v}%
_{t,1}^{(1)}-O_{1}^{(1)}v_{t,1}^{0}\right) +O_{p}\left(\eta_{N}^{2}\right)
\quad \text{uniformly over $i\in I_{3},t\in[T]$}.
\end{align*}%
Then 
\begin{align}
& \frac{1}{T}\sum_{t=1}^{T}e_{1,it}O_{1}^{(1)}v_{t,1}^{0}f_{it}(0)\left(%
\mu_{1,it}u_{i,1}^{0\prime}v_{t,1}^{0}-\hat{\mu}_{1,it}\dot{u}%
_{i,1}^{(1)\prime}\dot{v}_{t,1}^{(1)}\right)  \notag  \label{Lem24.1} \\
& =\frac{1}{T}%
\sum_{t=1}^{T}e_{1,it}O_{1}^{(1)}v_{t,1}^{0}f_{it}(0)u_{i,j}^{0%
\prime}v_{t,j}^{0}\left(\hat{\mu}_{1,it}-\mu_{1,it}\right)  \notag \\
& +\frac{1}{T}\sum_{t=1}^{T}e_{1,it}O_{1}^{(1)}v_{t,1}^{0}f_{it}(0)%
\mu_{1,it}\left(\dot{u}_{i,1}^{(1)}-O_{1}^{(1)}u_{i,1}^{0}\right)
^{\prime}O_{1}^{(1)}v_{t,1}^{0}  \notag \\
& +\frac{1}{T}\sum_{t=1}^{T}e_{1,it}O_{1}^{(1)}v_{t,1}^{0}f_{it}(0)%
\mu_{1,it}\left(O_{1}^{(1)}u_{i,1}^{0}\right) ^{\prime}\left(\dot{v}%
_{t,1}^{(1)}-O_{1}^{(1)}v_{t,1}^{0}\right) +O_{p}\left(\eta_{N}^{2}\right) .
\end{align}%
First, note that $\mathbb{E}\left(e_{1,it}v_{t,1}^{0}f_{it}(0)\mu_{1,it}%
\left(O_{1}^{(1)}u_{i,1}^{0}\right) ^{\prime}\left(\dot{v}%
_{t,1}^{(1)}-O_{1}^{(1)}v_{t,1}^{0}\right) \bigg|\mathscr{D}^{I_{1}\cup
I_{2}}\right) =0$ by Assumption \ref{ass:10}(i). Following similar arguments
used in the proof of Lemma \ref{Lem20}(i), we can show that the third term
on the RHS of (\ref{Lem24.1}) is $O_{p}\left(\eta_{N}^{2}\right) $. By
analogous arguments in (\ref{Lem22.5}) with the fact that $\hat{\mu}%
_{1,it}-\mu_{1,it}=e_{1,it}-\hat{e}_{1,it}$, we obtain that the first term
on the RHS of (\ref{Lem24.1}) is $O_{p}\left(\eta_{N}^{2}\right) $. In
addition 
\begin{align}
& \max_{i\in I_{3}}\left\Vert \frac{1}{T}%
\sum_{t=1}^{T}e_{1,it}v_{t,1}^{0}f_{it}(0)\mu_{1,it}\left(\dot{u}%
_{i,1}^{(1)}-O_{1}^{(1)}u_{i,1}^{0}\right)
^{\prime}O_{1}^{(1)}v_{t,1}^{0}\right\Vert_{2}  \notag  \label{Lem24.2} \\
& \leq \max_{i\in I_{3}}||\dot{u}_{i,1}^{(1)}-O_{1}^{(1)}u_{i,1}^{0}||_{2}%
\max_{i\in I_{3}}\left\Vert \frac{1}{T}\sum_{t=1}^{T}f_{it}(0)e_{1,it}%
\mu_{1,it}v_{t,1}^{0}v_{t,1}^{0\prime}\right\Vert_{F}=O_{p}\left(%
\eta_{N}^{2}\right) ,
\end{align}%
where the equality holds by the fact that 
\begin{equation*}
\mathbb{P}\left(\max_{i\in I_{3}}\left\Vert \frac{1}{T}%
\sum_{t=1}^{T}f_{it}(0)e_{1,it}\mu_{1,it}v_{t,1}^{0}v_{t,1}^{0\prime}\right%
\Vert_{F}>c_{25}\sqrt{\frac{\log (N\vee T)}{T}}\xi_{N}\bigg|\mathscr{D}%
\right) =o(1)
\end{equation*}
as in (\ref{Lem:F2}). Noted that Assumption \ref{ass:1}(iv) implies factor
and factor loading of $X_{j,it}$ is uniformly bounded for $\forall j\in[p]$,
which indicates that $\mu_{j,it}$ is also uniformly bounded a.s.. This
completes the proof statement (i).

(ii). Note that uniformly over $i\in I_{3}$, we have%
\begin{align}
& \frac{1}{T}\sum_{t=1}^{T}O_{0}^{(1)}v_{t,0}^{0}f_{it}(0)\left(%
\mu_{1,it}u_{i,1}^{0\prime}v_{t,1}^{0}-\hat{\mu}_{1,it}\dot{u}%
_{i,1}^{(1)\prime}\dot{v}_{t,1}^{(1)}\right)  \notag  \label{Lem24.3} \\
& =\frac{1}{T}\sum_{t=1}^{T}O_{0}^{(1)}v_{t,0}^{0}f_{it}(0)\left(\hat{\mu}%
_{1,it}-\mu_{1,it}\right) u_{i,j}^{0\prime}v_{t,j}^{0}+\frac{1}{T}%
\sum_{t=1}^{T}O_{0}^{(1)}v_{t,0}^{0}f_{it}\mu_{1,it}\left(\dot{u}%
_{i,1}^{(1)}-O_{1}^{(1)}u_{i,1}^{0}\right) O_{1}^{(1)}v_{t,1}^{0}  \notag \\
& +\frac{1}{T}\sum_{t=1}^{T}O_{0}^{(1)}v_{t,0}^{0}f_{it}\mu_{1,it}%
\left(O_{1}^{(1)}u_{i,1}^{0}\right) ^{\prime}\left(\dot{v}%
_{t,1}^{(1)}-O_{1}^{(1)}v_{t,1}^{0}\right) +O_{p}\left(\eta_{N}^{2}\right) 
\notag \\
& =O_{p}(\eta_{N}).
\end{align}
This term will remain as the bias term in the linear expansion of $\hat{u}%
_{i,0}^{(3,1)}$.
\end{proof}

\begin{lemma}
{\small \label{Lem25} } Under Assumptions \ref{ass:1}-\ref{ass:10}, we have

\begin{itemize}
\item[(i)] $\max_{i\in I_{3}}\Bigg\Vert\frac{1}{T}%
\sum_{t=1}^{T}e_{1,it}O_{1}^{(1)}v_{t,1}^{0}\Bigg\{\left[\mathbf{1}%
\left\{\epsilon_{it}\leq 0\right\}-\mathbf{1}\left\{\epsilon_{it}\leq%
\iota_{it}\left(O_{0}^{(1)}u_{i,0}^{0},O_{1}^{(1)}u_{i,1}^{0},\dot{u}%
_{i,1}^{(1)}\right)\right\}\right]\newline
-\left(F_{it}(0)-F_{it}\left[ \iota_{it}%
\left(O_{0}^{(1)}u_{i,0}^{0},O_{1}^{(1)}u_{i,1}^{0},\dot{u}%
_{i,1}^{(1)}\right) \right]\right)\Bigg\}\Bigg\Vert_{2}=o_{p}\left(\left(N%
\vee T\right)^{-\frac{1}{2}}\right)$,

\item[(ii)] $\max_{i\in I_{3}}\Bigg\Vert\frac{1}{T}%
\sum_{t=1}^{T}O_{0}^{(1)}v_{t,0}^{0}\Bigg\{\left[\mathbf{1}%
\left\{\epsilon_{it}\leq 0\right\}-\mathbf{1}\left\{\epsilon_{it}\leq%
\iota_{it}\left(O_{0}^{(1)}u_{i,0}^{0},O_{1}^{(1)}u_{i,1}^{0},\dot{u}%
_{i,1}^{(1)}\right)\right\}\right]\newline
-\left(F_{it}(0)-F_{it}\left[ \iota_{it}%
\left(O_{0}^{(1)}u_{i,0}^{0},O_{1}^{(1)}u_{i,1}^{0},\dot{u}%
_{i,1}^{(1)}\right) \right]\right)\Bigg\}\Bigg\Vert_{2}=o_{p}\left(\left(N%
\vee T\right)^{-\frac{1}{2}}\right)$.
\end{itemize}
\end{lemma}

\begin{proof}
(i) We still assume $K_{1}=1$ for notation simplicity. Let $%
O^{(1)}=diag\left(O_{0}^{(1)},O_{1}^{(1)}\right) $. Recall from Theorem \ref%
{Thm2}(iii) that with 
\begin{align*}
& D_{i}^{I}=O^{(1)}\frac{1}{T}\sum_{t=1}^{T}f_{it}(0) 
\begin{bmatrix}
v_{t,0}^{0}v_{t,0}^{\prime} &  & v_{t,0}^{0}v_{t,1}^{0\prime}X_{1,it} \\ 
v_{t,1}^{0}v_{t,0}^{0\prime}X_{1,it} &  & v_{t,1}^{0}v_{t,1}^{0%
\prime}X_{1,it}^{2}%
\end{bmatrix}%
O^{(1)\prime},\quad D_{i}^{II}=O^{(1)}\frac{1}{T}\sum_{t=1}^{T}\left[ \tau -%
\mathbf{1}\left\{ \epsilon_{it}\leq 0\right\} \right] 
\begin{bmatrix}
v_{t,0}^{0} \\ 
v_{t,1}^{0}X_{1,it}%
\end{bmatrix}
\\
& \mathbb{J}_{i}\left(\left\{ \dot{\Delta}_{t,v}\right\}_{t\in [T]}\right)
:=O^{(1)}\frac{1}{T}\sum_{t=1}^{T}\left[ \mathbf{1}\left\{\epsilon_{it}\leq
0\right\} -\mathbf{1}\left\{ \epsilon_{it}\leq \dot{\Delta}%
_{t,v}^{\prime}\Psi_{it}^{0}\right\} \right] 
\begin{bmatrix}
v_{t,0}^{0} \\ 
v_{t,1}^{0}X_{1,it}%
\end{bmatrix}%
,
\end{align*}%
uniformly over $i\in I_{3}$, we have 
\begin{align*}
\dot{\Delta}_{i,u}& =\left[ D_{i}^{I}\right] ^{-1}\left[ D_{i}^{II}+\mathbb{J%
}_{i}\left(\left\{ \dot{\Delta}_{t,v}\right\}_{t\in[T]}\right) \right]
+o_{p}\left(\left(N\vee T\right) ^{-\frac{1}{2}}\right) \\
& =\left(O^{(1)\prime}\right) ^{-1}\left(\frac{1}{T}\sum_{t=1}^{T}f_{it}(0)%
\begin{bmatrix}
v_{t,0}^{0}v_{t,0}^{\prime} &  & v_{t,0}^{0}v_{t,1}^{0\prime}X_{1,it} \\ 
v_{t,1}^{0}v_{t,0}^{0\prime}X_{1,it} &  & v_{t,1}^{0}v_{t,1}^{0%
\prime}X_{1,it}^{2}%
\end{bmatrix}%
\right) ^{-1}\Bigg[\frac{1}{T}\sum_{t=1}^{T}\left[ \tau -\mathbf{1}%
\left\{\epsilon_{it}\leq 0\right\} \right] 
\begin{bmatrix}
v_{t,0}^{0} \\ 
v_{t,1}^{0}X_{1,it}%
\end{bmatrix}
\\
& +\frac{1}{T}\sum_{t=1}^{T}\left[ \mathbf{1}\left\{ \epsilon_{it}\leq
0\right\} -\mathbf{1}\left\{ \epsilon_{it}\leq \dot{\Delta}%
_{t,v}^{\prime}\Psi_{it}^{0}\right\} \right] 
\begin{bmatrix}
v_{t,0}^{0} \\ 
v_{t,1}^{0}X_{1,it}%
\end{bmatrix}%
\Bigg]+o_{p}\left(\left(N\vee T\right) ^{-\frac{1}{2}}\right) \\
& :=h_{i}+o_{p}\left(\left(N\vee T\right) ^{-\frac{1}{2}}\right) .
\end{align*}%
Let $\mathfrak{l}=\left(\mathbf{0}_{K_{0}}^{\prime},\mathbf{1}%
_{K_{1}}^{\prime}\right) ^{\prime}$ with $\mathbf{0}_{K_{0}}$ being a $%
K_{0}\times 1$ vector of zeros and $\mathbf{1}_{K_{1}}$ a $K_{1}\times 1$
vector of ones. Let 
\begin{equation}
h_{i}^{I}=\mathfrak{l}^{\prime}h_{i}.  \label{h_I}
\end{equation}%
Then we have 
\begin{equation*}
O_{1}^{(1)\prime}\dot{u}_{i,1}^{(1)}-u_{i,1}^{0}=h_{i}^{I}+o_{p}\left(%
\left(N\vee T\right) ^{-\frac{1}{2}}\right) \quad \text{uniformly over $i\in
I_{3}$,}
\end{equation*}%
and $\max_{i\in I_{3}}||h_{i}^{I}||_{2}=O_{p}(\eta_{N}).$ 

Combining (\ref{Lem21.4})-(\ref{Lem21.6}), uniformly in $i\in I_{3}$ and $%
t\in \left[T\right]$ we have 
\begin{align}
& \iota_{it}\left(O_{0}^{(1)}u_{i,0}^{0},O_{1}^{(1)}u_{i,1}^{0},\dot{u}%
_{i,1}^{(1)}\right)  \notag  \label{Lem25.3} \\
& =\left[ \left(O_{0}^{(1)}u_{i,0}^{0}\right) ^{\prime}\left(\dot{v}%
_{t,0}^{(1)}-O_{0}^{(1)}v_{t,0}^{0}\right) \right] +\left(\hat{\mu}_{1,it}%
\dot{u}_{i,1}^{(1)\prime}\dot{v}_{t,1}^{(1)}-\mu_{1,it}u_{i,1}^{0%
\prime}v_{t,1}^{0}\right) +\left(\hat{e}_{1,it}u_{i,1}^{0\prime}O_{1}^{(1)}%
\dot{v}_{t,1}^{(1)}-e_{1,it}u_{i,1}^{0\prime}v_{t,1}^{0}\right)  \notag \\
& =\left(O_{0}^{(1)}u_{i,0}^{0}\right) ^{\prime}\left(\dot{v}%
_{t,0}^{(1)}-O_{0}^{(1)}v_{t,0}^{0}\right) +\left(\hat{\mu}%
_{1,it}^{(1)}-\mu_{1,it}\right) \left(O_{1}^{(1)}u_{i,1}^{0}\right)
^{\prime}O_{1}^{(1)}v_{t,1}^{0}  \notag \\
& +\mu_{1,it}\left(\dot{u}_{i,1}^{(1)}-O_{1}^{(1)}u_{i,1}^{0}\right)^{%
\prime}O_{1}^{(1)}v_{t,1}^{0}+\mu_{1,it}\left(O_{1}^{(1)}u_{i,1}^{0}\right)
^{\prime}\left(\dot{v}_{t,1}^{(1)}-O_{1}^{(1)}v_{t,1}^{0}\right) +(\hat{e}%
_{1,it}-e_{1,it})\left(O_{1}^{(1)}u_{i,1}^{0}\right)
^{\prime}O_{1}^{(1)}v_{t,1}^{0}  \notag \\
& +e_{1,it}\left(O_{1}^{(1)}u_{i,1}^{0}\right) ^{\prime}\left(\dot{v}%
_{t,1}^{(1)}-O_{1}^{(1)}v_{t,1}^{0}\right) +o_{p}\left(\left(N\vee T\right)
^{-\frac{1}{2}}\right)  \notag \\
& =\mu_{1,it}\left(\dot{u}_{i,1}^{(1)}-O_{1}^{(1)}u_{i,1}^{0}\right)^{%
\prime}O_{1}^{(1)}v_{t,1}^{0}+\left(O_{0}^{(1)}u_{i,0}^{0}\right)
^{\prime}\left(\dot{v}_{t,0}^{(1)}-O_{0}^{(1)}v_{t,0}^{0}\right)
+\left(O_{1}^{(1)}u_{i,1}^{0}\right) ^{\prime}\left(\dot{v}%
_{t,1}^{(1)}-O_{1}^{(1)}v_{t,1}^{0}\right) X_{1,it}  \notag \\
& +o_{p}\left(\left(N\vee T\right) ^{-\frac{1}{2}}\right)  \notag \\
& =\mu_{1,it}v_{t,1}^{0\prime}h_{i}^{I}+\left\{
\left(O_{0}^{(1)}u_{i,0}^{0}\right) ^{\prime}\left(\dot{v}%
_{t,0}^{(1)}-O_{0}^{(1)}v_{t,0}^{0}\right)
+\left(O_{1}^{(1)}u_{i,1}^{0}\right) ^{\prime}\left(\dot{v}%
_{t,1}^{(1)}-O_{1}^{(1)}v_{t,1}^{0}\right) X_{1,it}\right\} +\mathcal{R}%
_{\iota,it}  \notag \\
& :=\mu_{1,it}v_{t,1}^{0\prime}h_{i}^{I}+h_{it}^{II}\left(\dot{\Delta}%
_{t,v}\right) +\mathcal{R}_{\iota,it} \\
& :=R_{\iota ,it}^{1}\left(\left\vert \mu_{1,it}\right\vert +\left\vert
e_{1,it}\right\vert \right) +R_{\iota ,it}^{2}  \notag
\end{align}%
such that $\max_{i\in I_{3},t\in[T]}\left\vert \mathcal{R}%
_{\iota,it}\right\vert =o_{p}\left(\left(N\vee T\right) ^{-\frac{1}{2}%
}\right)$, $\max_{i\in I_{3},t\in[T]}\left\vert R_{\iota ,it}^{1}\right\vert
=O_{p}(\eta_{N})$ and $\max_{i\in I_{3},t\in[T]}\left\vert
R_{\iota,it}^{2}\right\vert =O_{p}(\eta_{N})$, where we use the fact that $%
\hat{\mu}_{1,it}-\mu_{1,it}+\hat{e}_{1,it}-e_{1,it}=0$, and $%
h_{it}^{II}\left(\dot{\Delta}_{t,v}\right) =\Psi_{it}^{0\prime}\dot{\Delta}%
_{t,v}$ with $\Psi_{it}=\left(O_{0}^{(1)}u_{i,0}^{0})^{%
\prime},(O_{1}^{(1)}u_{i,1}^{0}X_{1,it})^{\prime}\right) ^{\prime}$. Let 
\begin{align*}
\hat{\mathbb{I}}_{3,it}^{I}&=e_{1,it}v_{t,1}^{0}\Bigg\{\left[ \mathbf{1}%
\left\{ \epsilon_{it}\leq 0\right\} -\mathbf{1}\left\{
\epsilon_{it}\leq\mu_{1,it}v_{t,1}^{0\prime}h_{i}^{I}+h_{it}^{II}\left(\dot{%
\Delta}_{t,v}\right) \right\} \right] -\left(F_{it}(0)-F_{it}\left[
\mu_{1,it}v_{t,1}^{0\prime}h_{i}^{I}+h_{it}^{II}\left(\dot{\Delta}%
_{t,v}\right) \right] \right) \Bigg\}, \\
\hat{\mathbb{I}}_{3,it}^{II}& =e_{1,it}v_{t,1}^{0}\Bigg\{\left[ \mathbf{1}%
\left\{ \epsilon_{it}\leq
\mu_{1,it}v_{t,1}^{0\prime}h_{i}^{I}+h_{it}^{II}\left(\dot{\Delta}%
_{t,v}\right) +\mathcal{R}_{\iota,it}\right\} -\mathbf{1}\left\{
\epsilon_{it}\leq \mu_{1,it}v_{t,1}^{0\prime}h_{i}^{I}+h_{it}^{II}\left(\dot{%
\Delta}_{t,v}\right) \right\} \right] \\
& -\left(F_{it}\left[ \mu_{1,it}v_{t,1}^{0\prime}h_{i}^{I}+h_{it}^{II}\left(%
\dot{\Delta}_{t,v}\right) +\mathcal{R}_{\iota,it}\right] -F_{it}\left[
\mu_{1,it}v_{t,1}^{0\prime}h_{i}^{I}+h_{it}^{II}\left(\dot{\Delta}%
_{t,v}\right) \right] \right) \Bigg\}.
\end{align*}%
We first show that $\max_{i\in I_{3}}\Bigg\Vert\frac{1}{T}\sum_{t=1}^{T}\hat{%
\mathbb{I}}_{3,it}^{I}\Bigg\Vert_{2}=O_{p}\left(\eta_{N}^{2}\right) $.

For some sufficiently large constant $M$, define event $\mathscr{A}%
_{11}(M)=\left\{ \max_{i\in I_{3}}\left\Vert h_{i}^{I}\right\Vert_{2}\leq
M\eta_{N}\right\} $. We have $\mathbb{P}\left(\mathscr{A}_{11}^{c}(M)\right)
\leq e$ for any $e>0$. For some positive large enough constant $c_{26}$, we
have 
\begin{align}
& \mathbb{P}\left(\max_{i\in I_{3}}\Bigg\Vert\frac{1}{T}\sum_{t=1}^{T}\hat{%
\mathbb{I}}_{3,it}^{I}\Bigg\Vert_{2}>c_{26}\xi_{N}^{\frac{5+\vartheta }{%
2+\vartheta }}\left(\frac{\log (N\vee T)}{N\wedge T}\right) ^{\frac{1}{%
4+2\vartheta }}\sqrt{\frac{\log (N\vee T)}{T}}\right)  \notag  \label{I3_1}
\\
& \leq \mathbb{P}\left(\max_{i\in I_{3}}\Bigg\Vert\frac{1}{T}\sum_{t=1}^{T}%
\hat{\mathbb{I}}_{3,it}^{I}\Bigg\Vert_{2}>c_{26}\xi_{N}^{\frac{5+\vartheta}{%
2+\vartheta }}\left(\frac{\log (N\vee T)}{N\wedge T}\right) ^{\frac{1}{%
4+2\vartheta }}\sqrt{\frac{\log (N\vee T)}{T}},\mathscr{A}_{11}(M)\right) +e
\notag \\
& \leq \mathbb{P}\left(\max_{i\in I_{3}}\sup_{\xi \in \Xi }\Bigg\Vert\frac{1%
}{T}\sum_{t=1}^{T}\hat{\mathbb{I}}_{3,it}^{I}(\xi )\Bigg\Vert%
_{2}>c_{26}\xi_{N}^{\frac{5+\vartheta }{2+\vartheta }}\left(\frac{\log
(N\vee T)}{N\wedge T}\right) ^{\frac{1}{4+2\vartheta }}\sqrt{\frac{\log
(N\vee T)}{T}}\right)+e.
\end{align}%
with $\Xi ^{1}:=\left\{ \xi \in \mathbb{R}^{\left(K_{0}+K_{1}\right) \times
1}:\left\Vert \xi \right\Vert_{2}\leq M\eta_{N}\right\} $ and 
\begin{equation*}
\hat{\mathbb{I}}_{3,it}^{I}(\xi )=e_{1,it}v_{t,1}^{0}\Bigg\{\left[ \mathbf{1}%
\left\{ \epsilon_{it}\leq 0\right\} -\mathbf{1}\left\{
\epsilon_{it}\leq\mu_{1,it}v_{t,1}^{0\prime}\xi +h_{it}^{II}\left(\dot{\Delta%
}_{t,v}\right)\right\} \right] -\left[ F_{it}(0)-F_{it}\left(%
\mu_{1,it}v_{t,1}^{0\prime}\xi +h_{it}^{II}\left(\dot{\Delta}_{t,v}\right)
\right) \right] \Bigg\}.
\end{equation*}%
Divide $\Xi ^{1}$ into sub classes $\Xi_{s}^{1}$ for $s=1,...,n_{\Xi ^{1}}$
such that $\left\Vert \xi -\tilde{\xi}\right\Vert_{2}<\frac{\varepsilon }{T}$
for $\forall \xi ,\,\,\bar{\xi}\in \Xi_{s}^{1}$ and $n_{\Xi ^{1}}\asymp
T^{K_{0}+K_{1}}$. With analogous arguments from (\ref{Wdot_1})-(\ref{Wdot_4}%
), for $\forall \xi_{s}\in \Xi_{s}^{1}$, we have 
\begin{equation*}
\left\Vert \frac{1}{T}\sum_{t\in[T]}\hat{\mathbb{I}}_{3,it}^{I}(\xi)\right%
\Vert_{2}\leq \left\Vert \frac{1}{T}\sum_{t\in[T]}\hat{\mathbb{I}}%
_{3,it}^{I}(\xi_{s})\right\Vert_{2}+\left\Vert \frac{1}{T}\sum_{t\in[T]}%
\left[ \hat{\mathbb{I}}_{3,it}^{I}(\xi )-\hat{\mathbb{I}}_{3,it}^{I}(\xi_{s})%
\right] \right\Vert_{2}
\end{equation*}%
with 
\begin{align}
& \max_{i\in I_{3},s\in [n_{\Xi ^{1}}]}\sup_{\xi \in \Xi_{s}}\left\Vert 
\frac{1}{T}\sum_{t\in[T]}\left[ \hat{\mathbb{I}}_{3,it}^{I}(\xi )-\hat{%
\mathbb{I}}_{3,it}^{I}(\xi_{s})\right] \right\Vert_{2} & &  \notag
\label{I3_3} \\
& \leq \max_{i\in I_{3}}\frac{1}{T}\sum_{t\in[T]}\mathbb{E}\left[\left\Vert
e_{1,it}v_{t,1}^{0}\right\Vert_{2}\mathbf{1}\left\{
\left\vert\epsilon_{it}-h_{it}^{II}\left(\dot{\Delta}_{t,v}\right)
\right\vert \leq\left\vert \mu_{1,it}\right\vert \left\Vert
v_{t,1}^{0\prime}\right\Vert_{2}\frac{\varepsilon }{T}\right\} \bigg|%
\mathscr{D}_{e}^{I_{1}\cup I_{2}}\right] & &  \notag \\
& +\max_{i\in I_{3}}\left\Vert \hat{\mathbb{I}}_{3,i}^{III}\right\Vert_{2}+%
\max_{i\in I_{3}}\frac{1}{T}\sum_{t\in[T]}\left\Vert
e_{1,it}v_{t,1}^{0}\right\Vert_{2}\left\vert \mu_{1,it}\right\vert\left\Vert
v_{t,1}^{0\prime}\right\Vert_{2}\frac{\varepsilon }{T} & &  \notag \\
& \leq \max_{i\in I_{3}}\left\Vert \hat{\mathbb{I}}_{3,i}^{III}\right%
\Vert_{2}+O\left(\frac{\varepsilon }{T}\right) ~a.s. & &
\end{align}%
where $\hat{\mathbb{I}}_{3,i}^{III}=\frac{1}{T}\sum_{t\in[T]}\hat{\mathbb{I}}%
_{3,it}^{III}$ and 
\begin{align*}
\hat{\mathbb{I}}_{3,it}^{III}& =\left\Vert e_{1,it}v_{t,1}^{0}\right\Vert_{2}%
\Bigg\{\mathbf{1}\left\{ \left\vert \epsilon_{it}-h_{it}^{II}\left(\dot{%
\Delta}_{t,v}\right) \right\vert \leq \left\vert
\mu_{1,it}\right\vert\left\Vert v_{t,1}^{0\prime}\right\Vert_{2}\frac{%
\varepsilon }{T}\right\} \\
& -\mathbb{E}\left[ \mathbf{1}\left\{ \left\vert
\epsilon_{it}-h_{it}^{II}\left(\dot{\Delta}_{t,v}\right) \right\vert
\leq\left\vert \mu_{1,it}\right\vert \left\Vert
v_{t,1}^{0\prime}\right\Vert_{2}\frac{\varepsilon }{T}\right\} \bigg|%
\mathscr{D}_{e}^{I_{1}\cup I_{2}}\right]\Bigg\}.
\end{align*}%
Similarly to (\ref{Wdot_4}), we can show that 
\begin{align}
& \mathbb{P}\left(\max_{i\in I_{3}}\left\Vert \hat{\mathbb{I}}%
_{3,it}^{III}\right\Vert_{2}>c_{26}\xi_{N}^{\frac{5+\vartheta }{2+\vartheta }%
}\left(\frac{\log (N\vee T)}{N\wedge T}\right) ^{\frac{1}{4+2\vartheta }}%
\sqrt{\frac{\log (N\vee T)}{T}}\right) =o(1),\quad \text{and}  \label{I3_5}
\\
& \mathbb{P}\left(\max_{i\in I_{3}}\max_{s\in [n_{\Xi^{1}}]}\left\Vert \hat{%
\mathbb{I}}_{3,it}^{I}\left(\xi_{s}\right)\right\Vert_{2}>c_{26}\xi_{N}^{%
\frac{5+\vartheta }{2+\vartheta }}\left(\frac{\log (N\vee T)}{N\wedge T}%
\right) ^{\frac{1}{4+2\vartheta }}\sqrt{\frac{\log (N\vee T)}{T}}\right)
=o(1).
\end{align}%
Combining (\ref{I3_1})-(\ref{I3_5}), we have 
\begin{equation*}
\max_{i\in I_{3}}\Bigg\Vert\frac{1}{T}\sum_{t=1}^{T}\hat{\mathbb{I}}%
_{3,it}^{I}\Bigg\Vert_{2}=o_{p}\left(\left(N\vee T\right) ^{-\frac{1}{2}%
}\right) .
\end{equation*}

Next, we notice that 
\begin{align}  \label{I3_6}
&\max_{i\in I_{3}}\Bigg\Vert\frac{1}{T}\sum_{t=1}^{T}\hat{\mathbb{I}}%
_{3,it}^{II}\Bigg\Vert_{2}  \notag \\
&\leq\max_{i\in I_{3}}\frac{1}{T}\sum_{t=1}^{T}\left\Vert
e_{1,it}v_{t,1}^{0}\right\Vert_{2}\mathbf{1}\left\{\left\vert%
\mu_{1,it}v_{t,1}^{0\prime}h_{i}^{I}+h_{it}^{II}\left(\dot{\Delta}%
_{t,v}\right)
\right\vert\leq\left\vert\epsilon_{it}\right\vert\leq\left\vert%
\mu_{1,it}v_{t,1}^{0\prime}h_{i}^{I}+h_{it}^{II}\left(\dot{\Delta}%
_{t,v}\right)\right\vert+\left\vert\mathcal{R}_{\iota,it}\right\vert \right\}
\notag \\
&+\max_{i\in I_{3}}\frac{1}{T}\sum_{t=1}^{T}\left\Vert
e_{1,it}v_{t,1}^{0}\right\Vert_{2}\left\vert F_{it}\left[%
\mu_{1,it}v_{t,1}^{0\prime}h_{i}^{I}+h_{it}^{II}\left(\dot{\Delta}%
_{t,v}\right)+\mathcal{R}_{\iota,it} \right]-F_{it}\left[%
\mu_{1,it}v_{t,1}^{0\prime}h_{i}^{I}+h_{it}^{II}\left(\dot{\Delta}%
_{t,v}\right) \right]\right\vert  \notag \\
&\leq\max_{i\in I_{3}}\frac{1}{T}\sum_{t=1}^{T}\left\Vert
e_{1,it}v_{t,1}^{0}\right\Vert_{2}\mathbf{1}\left\{\left\vert%
\mu_{1,it}v_{t,1}^{0\prime}h_{i}^{I}+h_{it}^{II}\left(\dot{\Delta}%
_{t,v}\right)
\right\vert\leq\left\vert\epsilon_{it}\right\vert\leq\left\vert%
\mu_{1,it}v_{t,1}^{0\prime}h_{i}^{I}+h_{it}^{II}\left(\dot{\Delta}%
_{t,v}\right)\right\vert+\left\vert\mathcal{R}_{\iota,it}\right\vert \right\}
\notag \\
&+\max_{i\in I_{3}}\frac{1}{T}\sum_{t=1}^{T}\left\Vert
e_{1,it}v_{t,1}^{0}\right\Vert_{2}\left\vert\mathcal{R}_{\iota,it}\right\vert
\notag \\
&=\max_{i\in I_{3}}\frac{1}{T}\sum_{t=1}^{T}\hat{\mathbb{I}}%
_{3,it}^{IV}+o_{p}\left(\left(N\vee T\right)^{-\frac{1}{2}}\right),
\end{align}
where the first inequality is by triangle inequality, the second inequality
is by mean value theorem and the last line is by the uniform convergence
rate for $\mathcal{R}_{\iota,it}$ with 
\begin{equation*}
\hat{\mathbb{I}}_{3,it}^{IV}=\left\Vert e_{1,it}v_{t,1}^{0}\right\Vert_{2}%
\mathbf{1}\left\{\left\vert\mu_{1,it}v_{t,1}^{0\prime}h_{i}^{I}+h_{it}^{II}%
\left(\dot{\Delta}_{t,v}\right)
\right\vert\leq\left\vert\epsilon_{it}\right\vert\leq\left\vert%
\mu_{1,it}v_{t,1}^{0\prime}h_{i}^{I}+h_{it}^{II}\left(\dot{\Delta}%
_{t,v}\right)\right\vert+\left\vert\mathcal{R}_{\iota,it}\right\vert
\right\}.
\end{equation*}

Define the event $\mathscr{A}_{12}(M):=\left\{ \max_{i\in I_{3},t\in
[T]}\left\vert \mathcal{R}_{\iota,it}\right\vert \leq M\eta_{N}^{2}\right\} $
with $\mathbb{P}\left\{ \mathscr{A}_{12}^{c}(M)\right\}\leq e$ for any $e>0$%
. Then for a large enough constant $c_{26}$, similarly as (\ref{step4_4}),
we have 
\begin{align}
& \mathbb{P}\left(\max_{i\in I_{3}}\left\Vert \frac{1}{T}\sum_{t=1}^{T}\hat{%
\mathbb{I}}_{3,it}^{IV}\right\Vert_{2}>c_{26}\eta_{N}^{2}\right) \leq\mathbb{%
P}\left(\max_{i\in I_{3}}\left\Vert \frac{1}{T}\sum_{t=1}^{T}\hat{\mathbb{I}}%
_{3,it}^{IV}\right\Vert_{2}>c_{26}\eta_{N}^{2},\mathscr{A}_{12}(M)\right) +e
\notag  \label{I3_7} \\
& \leq \mathbb{P}\left(\max_{i\in I_{3}}\frac{1}{T}\sum_{t=1}^{T}\tilde{%
\mathbf{1}}_{it}(h_{i}^{I})>c_{26}\eta_{N}^{2}\right) +e  \notag \\
& \leq \mathbb{P}\left(\sup_{i\in I_{3},\xi \in \Xi ^{1}}\frac{1}{T}%
\sum_{t=1}^{T}\tilde{\mathbf{1}}_{it}(\xi )>c_{26}\eta_{N}^{2}\right) +2e 
\notag \\
& \leq \mathbb{P}\left(\max_{i\in I_{3}}\sup_{\xi \in \Xi ^{1}}\left\vert 
\frac{1}{T}\sum_{t=1}^{T}\tilde{\mathbf{1}}_{it}(\xi )-\bar{\tilde{\mathbf{1}%
}}_{it}(\xi )\right\vert >\frac{c_{26}\eta_{N}^{2}}{2}\right) +\mathbb{P}%
\left(\max_{i\in I_{3}}\sup_{\xi \in \Xi ^{1}}\left\vert \frac{1}{T}%
\sum_{t=1}^{T}\bar{\tilde{\mathbf{1}}}_{it}(\xi )\right\vert >\frac{%
c_{26}\eta_{N}^{2}}{2}\right) +2e  \notag \\
& =\mathbb{E}\left[ \mathbb{P}\left(\max_{i\in I_{3}}\sup_{\xi \in
\Xi^{1}}\left\vert \frac{1}{T}\sum_{t=1}^{T}\tilde{\mathbf{1}}_{it}(\xi )-%
\bar{\tilde{\mathbf{1}}}_{it}(\xi )\right\vert >\frac{c_{26}\eta_{N}^{2}}{2}%
\bigg|\mathscr{D}_{e}^{I_{1}\cup I_{2}}\right) \right] +\mathbb{P}%
\left(\max_{i\in I_{3}}\sup_{\xi \in \Xi ^{1}}\left\vert \frac{1}{T}%
\sum_{t=1}^{T}\bar{\tilde{\mathbf{1}}}_{it}(\xi )\right\vert >\frac{%
c_{26}\eta_{N}^{2}}{2}\right) +2e
\end{align}%
where $\tilde{\mathbf{1}}_{it}(h_{i}^{I}):=\left\Vert
e_{1,it}v_{t,1}^{0}\right\Vert_{2}\mathbf{1}\left\{ \left\vert
\mu_{1,it}v_{t,1}^{0\prime}h_{i}^{I}+h_{it}^{II}\left(\dot{\Delta}%
_{t,v}\right)\right\vert \leq \left\vert \epsilon_{it}\right\vert \leq
\left\vert \mu_{1,it}v_{t,1}^{0\prime}h_{i}^{I}+h_{it}^{II}\left(\dot{\Delta}%
_{t,v}\right)\right\vert +M\eta_{N}^{2}\right\} $ and $\bar{\tilde{\mathbf{1}%
}}_{it}(h_{i}^{I})=\mathbb{E}\left[ \tilde{\mathbf{1}}_{it}(h_{i}^{I})\bigg|%
\mathscr{D}_{e}^{I_{1}\cup I_{2}}\right]$. By analogous arguments for the
first term on the RHS of inequality (\ref{I3_1}), we can show that 
\begin{equation*}
\mathbb{E}\left[ \mathbb{P}\left(\max_{i\in I_{3}}\sup_{\xi \in
\Xi^{1}}\left\vert \frac{1}{T}\sum_{t=1}^{T}\tilde{\mathbf{1}}_{it}(\xi )-%
\bar{\tilde{\mathbf{1}}}_{it}(\xi )\right\vert >\frac{c_{26}\eta_{N}^{2}}{2}%
\bigg|\mathscr{D}_{e}^{I_{1}\cup I_{2}}\right) \right] =o(1).
\end{equation*}%
Besides, we observe that 
\begin{align*}
& \max_{i\in I_{3}}\sup_{\xi \in \Xi ^{1}}\left\vert \frac{1}{T}%
\sum_{t=1}^{T}\bar{\tilde{\mathbf{1}}}_{it}(\xi )\right\vert \\
& \leq \max_{i\in I_{3}}\sup_{\xi \in \Xi ^{1}}\frac{1}{T}%
\sum_{t=1}^{T}\left\Vert e_{1,it}v_{t,1}^{0}\right\Vert_{2}\left\vert F_{it}%
\left[ \left\vert \mu_{1,it}v_{t,1}^{0\prime}\xi +h_{it}^{II}\left(\dot{%
\Delta}_{t,v}\right) \right\vert +M\eta_{N}^{2}\right] -F_{it}\left[%
\left\vert \mu_{1,it}v_{t,1}^{0\prime}\xi +h_{it}^{II}\left(\dot{\Delta}%
_{t,v}\right) \right\vert \right] \right\vert \\
& \leq \max_{i\in I_{3}}\frac{1}{T}\sum_{t=1}^{T}\left\Vert
e_{1,it}v_{t,1}^{0}\right\Vert_{2}M\eta_{N}^{2},
\end{align*}%
which yields 
\begin{equation}
\mathbb{P}\left(\max_{i\in I_{3}}\sup_{\xi \in \Xi ^{1}}\left\vert \frac{1}{T%
}\sum_{t=1}^{T}\bar{\tilde{\mathbf{1}}}_{it}(\xi )\right\vert >\frac{%
c_{26}\eta_{N}^{2}}{2}\right) =o(1).  \label{I3_9}
\end{equation}

Combining (\ref{I3_6})-(\ref{I3_9}) yields $\max_{i\in I_{3}}\Bigg\Vert\frac{%
1}{T}\sum_{t=1}^{T}\hat{\mathbb{I}}_{3,it}^{II}\Bigg\Vert_{2}=O_{p}\left(%
\eta_{N}^{2}\right) =o_{p}\left(\left(N\vee T\right) ^{-\frac{1}{2}}\right) $
by Assumption \ref{ass:1}(ix). This concludes the proof of statement (i).

(ii) The proof is analogous to that of part (i) and thus omitted.
\end{proof}

\begin{lemma}
{\small \label{Lem26} } Under Assumptions \ref{ass:1}-\ref{ass:10},
uniformly in $i\in I_{3}$, we can show that

\begin{itemize}
\item[(i)] $\max_{i\in I_{3}}\Bigg\Vert\frac{1}{T}\operatornamewithlimits{%
\sum}\limits_{t=1}^{T}e_{1,it}O_{1}^{(1)}v_{t,1}^{0}\Bigg\{\mathbf{1}%
\left\{\epsilon_{it}\leq 0\right\}- \mathbf{1}\left\{\epsilon_{it}\leq%
\iota_{it}\left(\hat{u}_{i,0}^{(3,1)},\hat{u}_{i,1}^{(3,1)},\dot{u}%
_{i,1}^{(1)}\right)\right\}-\left(F_{it}(0)-F_{it}\left[\iota_{it}\left(\hat{%
u}_{i,0}^{(3,1)},\hat{u}_{i,1}^{(3,1)},\dot{u}_{i,1}^{(1)}\right)\right]%
\right) \Bigg\}\Bigg\Vert_{2}\newline
=o_{p}\left(\left(N\vee T\right)^{-\frac{1}{2}}\right)$,

\item[(ii)] $\max_{i\in I_{3}}\Bigg\Vert\frac{1}{T}\operatornamewithlimits{%
\sum}\limits_{t=1}^{T}v_{t,0}^{0}\Bigg\{\mathbf{1}\left\{\epsilon_{it}\leq
0\right\}- \mathbf{1}\left\{\epsilon_{it}\leq \iota_{it}\left(\hat{u}%
_{i,0}^{(3,1)},\hat{u}_{i,1}^{(3,1)},\dot{u}_{i,1}^{(1)}\right)\right\}-%
\left(F_{it}(0)-F_{it}\left[\iota_{it}\left(\hat{u}_{i,0}^{(3,1)},\hat{u}%
_{i,1}^{(3,1)},\dot{u}_{i,1}^{(1)}\right)\right]\right) \Bigg\}\Bigg\Vert_{2}%
\newline
=o_{p}\left(\left(N\vee T\right)^{-\frac{1}{2}}\right)$.
\end{itemize}
\end{lemma}

\begin{proof}
As in (\ref{Lem25.3}), we can show that 
\begin{align}
& \iota_{it}\left(\hat{u}_{i,0}^{(3,1)},\hat{u}_{i,1}^{(3,1)},\dot{u}%
_{i,1}^{(1)}\right)  \notag  \label{Lem26.1} \\
& =\mu_{1,it}v_{t,1}^{0\prime}h_{i}^{I}+\left[ \left(\hat{u}%
_{i,0}^{(3,1)}-O_{0}^{(1)}u_{i,0}^{0}\right)
^{\prime}O_{0}^{(1)}v_{t,0}^{0}+e_{1,it}\left(\hat{u}%
_{i,1}^{(3,1)}-O_{1}^{(1)}u_{i,1}^{0}\right) ^{\prime}O_{1}^{(1)}v_{t,1}^{0}%
\right]  \notag \\
& +\left[ \left(O_{0}^{(0)}u_{i,0}^{0}\right) ^{\prime}\left(\dot{v}%
_{t,0}^{(1)}-O_{0}^{(1)}v_{t,0}^{0}\right)
+X_{1,it}\left(O_{1}^{(1)}u_{i,1}^{0}\right) ^{\prime}\left(\dot{v}%
_{t,1}^{(1)}-O_{1}^{(1)}v_{t,1}^{0}\right) \right] +\mathcal{R}_{\hat{\iota}%
,it}  \notag \\
:& =\mu_{1,it}v_{t,1}^{0\prime}h_{i}^{I}+\left(\hat{u}%
_{i,0}^{(3,1)}-O_{0}^{(1)}u_{i,0}^{0}\right)
^{\prime}O_{0}^{(1)}v_{t,0}^{0}+e_{1,it}\left(\hat{u}%
_{i,1}^{(3,1)}-O_{1}^{(1)}u_{i,1}^{0}\right)
^{\prime}O_{1}^{(1)}v_{t,1}^{0}+h_{it}^{II}\left(\dot{\Delta}_{t,v}\right) +%
\mathcal{R}_{\hat{\iota},it},
\end{align}%
where $\max_{i\in I_{3},t\in[T]}\left\vert \mathcal{R}_{\hat{\iota}%
,it}\right\vert =o_{p}\left(\left(N\vee T\right) ^{-\frac{1}{2}}\right) $.
As in the proof of Theorem \ref{Thm2}, we can show that $\max_{i\in
I_{3}}\left\Vert \hat{u}_{i,j}^{(3,1)}-O_{j}^{(1)}u_{i,j}^{0}\right\Vert_{2}$
$=O_{p}\left(\eta_{N}\right) $ for $\forall j\in [p]\cup \{0\}$. Then by
changing the event set $\mathscr{A}_{11}(M)$ to 
\begin{equation*}
\left\{ \max_{i\in I_{3}}\left\Vert h_{i}^{I}\right\Vert_{2}\leq
M\eta_{N},\max_{i\in I_{3}}\left\Vert \hat{u}_{i,j}-O_{j}^{(1)}u_{i,j}^{0}%
\right\Vert_{2}\leq M\eta_{N}\right\} ,
\end{equation*}%
we can repeat the analysis in Lemma \ref{Lem25} and obtain the desired
results for statement (i). With some obvious adjustment, statement (ii) can
be proved.
\end{proof}

\begin{lemma}
{\small \label{Lem27} } Under Assumptions \ref{ass:1}-\ref{ass:10}, for
block matrices $\hat{D}_{t}^{F}$, $D_{t}^{F}$, $\hat{D}_{t}^{J}$ and $%
D_{t}^{J}$ defined in Appendix A, we have 
\begin{equation*}
\max_{t\in[T]}||\hat{D}_{t}^{F}-D_{t}^{F}||_{F}=O_{p}(\eta_{N}),\text{ and }%
\max_{t\in[T]}||\hat{D}_{t}^{J}-D_{t}^{J}||_{F}=\left\Vert 
\begin{bmatrix}
O_{p}\left(\eta_{N}^{2}\right) & O_{p}(\eta_{N}) \\ 
O_{p}\left(\eta_{N}^{2}\right) & O_{p}\left(\eta_{N}^{2}\right)%
\end{bmatrix}
\right\Vert_{F}.
\end{equation*}
\end{lemma}

\begin{proof}
By analogous arguments as used in the proofs of Lemmas \ref{Lem21} and \ref%
{Lem22}, we can prove the lemma.
\end{proof}

\begin{lemma}
{\small \label{Lem28} } Under Assumptions \ref{ass:1}-\ref{ass:10}, we have

\begin{itemize}
\item[(i)] $\max_{t\in[T]}\left\Vert \frac{1}{N_{3}}\sum_{i\in
I_{3}}O_{u,0}^{(1)}u_{i,0}^{0}f_{it}(0)\left(\mu_{1,it}u_{i,1}^{0%
\prime}v_{t,1}^{0}-\hat{\mu}_{1,it}\hat{u}_{i,1}^{(3,1)\prime}\dot{v}%
_{t,1}^{(1)}\right) \right\Vert_{2}=O_{p}(\eta_{N})$,

\item[(ii)] $\max_{t\in[T]}\left\Vert \frac{1}{N_{3}}\sum_{i\in
I_{3}}e_{1,it}O_{u,1}^{(1)}u_{i,1}^{0}f_{it}(0)\left(\mu_{1,it}u_{i,1}^{0%
\prime}v_{t,1}^{0}-\hat{\mu}_{1,it}\hat{u}_{i,1}^{(3,1)\prime}\dot{v}%
_{t,1}^{(1)}\right) \right\Vert_{2}=o_{p}\left(\left(N\vee T\right)^{-\frac{1%
}{2}}\right)$.
\end{itemize}
\end{lemma}

\begin{proof}
(i) Note that 
\begin{align}
& \hat{\mu}_{1,it}\hat{u}_{i,1}^{(3,1)\prime}\dot{v}_{t,1}^{(1)}-%
\mu_{1,it}u_{i,1}^{0\prime}v_{t,1}^{0}  \notag  \label{Lem28.1} \\
& =\left(\hat{\mu}_{1,it}-\mu_{1,it}\right) \left(\hat{u}%
_{i,1}^{(3,1)}-O_{u,1}^{(1)}u_{i,1}^{0}\right) ^{\prime}\left(\dot{v}%
_{t,1}^{(1)}-O_{1}^{(1)}v_{t,1}^{0}\right) +\left(\hat{\mu}%
_{1,it}-\mu_{1,it}\right) \left(O_{u,1}^{(1)}u_{i,1}^{0}\right)
^{\prime}\left(\dot{v}_{t,1}^{(1)}-O_{1}^{(1)}v_{t,1}^{0}\right)  \notag \\
& +\left(\hat{\mu}_{1,it}-\mu_{1,it}\right) \left(\hat{u}%
_{i,1}^{(3,1)}-O_{u,1}^{(1)}u_{i,1}^{0}\right)
^{\prime}O_{1}^{(1)}v_{t,1}^{0}+\left(\hat{\mu}_{1,it}-\mu_{1,it}\right)
\left(O_{u,1}^{(1)}u_{i,1}^{0}\right) ^{\prime}O_{1}^{(1)}v_{t,1}^{0}  \notag
\\
& +\mu_{1,it}\left(\hat{u}_{i,1}^{(3,1)}-O_{u,1}^{(1)}u_{i,1}^{0}\right)^{%
\prime}\left(\dot{v}_{t,1}^{(1)}-O_{1}^{(1)}v_{t,1}^{0}\right)
+\mu_{1,it}\left(\hat{u}_{i,1}^{(3,1)}-O_{u,1}^{(1)}u_{i,1}^{0}\right)
^{\prime}O_{1}^{(1)}v_{t,1}^{0}  \notag \\
& +\mu_{1,it}\left(O_{u,1}^{(1)}u_{i,1}^{0}\right) ^{\prime}\left(\dot{v}%
_{t,1}^{(1)}-O_{1}^{(1)}v_{t,1}^{0}\right)
+\mu_{1,it}\left(O_{u,1}^{(1)}u_{i,1}^{0}\right)
^{\prime}O_{1}^{(1)}v_{t,1}^{0}-\mu_{1,it}u_{i,1}^{0\prime}v_{t,1}^{0} 
\notag \\
& =\left(\hat{\mu}_{1,it}-\mu_{1,it}\right)
\left(O_{u,1}^{(1)}u_{i,1}^{0}\right)
^{\prime}O_{1}^{(1)}v_{t,1}^{0}+\mu_{1,it}\left(\hat{u}%
_{i,1}^{(3,1)}-O_{u,1}^{(1)}u_{i,1}^{0}\right)
^{\prime}O_{1}^{(1)}v_{t,1}^{0}+\mu_{1,it}\left(O_{u,1}^{(1)}u_{i,1}^{0}%
\right)^{\prime}\left(\dot{v}_{t,1}^{(1)}-O_{1}^{(1)}v_{t,1}^{0}\right) 
\notag \\
&
+\mu_{1,it}v_{t,1}^{0\prime}O_{1}^{(1)\prime}\left(O_{u,1}^{(1)}-O_{1}^{(1)}%
\right) u_{i,1}^{0}+O_{p}\left(\eta_{N}^{2}\right)  \notag \\
& =\left(\hat{\mu}_{1,it}-\mu_{1,it}\right)
u_{i,1}^{0\prime}v_{t,1}^{0}+\mu_{1,it}\left(\hat{u}%
_{i,1}^{(3,1)}-O_{u,1}^{(1)}u_{i,1}^{0}\right)
^{\prime}O_{1}^{(1)}v_{t,1}^{0}+\mu_{1,it}\left(O_{1}^{(1)}u_{i,1}^{0}%
\right) ^{\prime}\left(\dot{v}_{t,1}^{(1)}-O_{1}^{(1)}v_{t,1}^{0}\right) 
\notag \\
&
+\mu_{1,it}v_{t,1}^{0\prime}O_{1}^{(1)\prime}\left(O_{u,1}^{(1)}-O_{1}^{(1)}%
\right) u_{i,1}^{0}+O_{p}\left(\eta_{N}^{2}\right) =O_{p}(\eta_{N}),
\end{align}%
uniformly over $i\in I_{3}$ and $t\in[T]$, where the last equality holds by
the fact that $\left\Vert
O_{u,1}^{(1)}-O_{1}^{(1)}\right\Vert_{F}=O_{p}(\eta_{N})$. It follows that 
\begin{equation*}
\max_{t\in[T]}\left\Vert \frac{1}{N_{3}}\sum_{i\in
I_{3}}O_{u,0}^{(1)}u_{i,0}^{0}f_{it}(0)\left(\hat{\mu}_{1,it}\hat{u}%
_{i,1}^{(3,1)\prime}\dot{v}_{t,1}^{(1)}-\mu_{1,it}u_{i,1}^{0%
\prime}v_{t,1}^{0}\right) \right\Vert_{2}=O_{p}(\eta_{N}).
\end{equation*}

(ii) Observe that 
\begin{align}
& \frac{1}{N_{3}}\sum_{i\in
I_{3}}e_{1,it}O_{u,1}^{(1)}u_{i,1}^{0}f_{it}(0)\left(\hat{\mu}_{1,it}^{(1)}%
\hat{u}_{i,1}^{(3,1)\prime}\dot{v}_{t,1}^{(1)}-\mu_{1,it}u_{i,1}^{0%
\prime}v_{t,1}^{0}\right)  \notag  \label{Lem28.2} \\
& =O_{u,1}^{(1)}\frac{1}{N_{3}}\sum_{i\in
I_{3}}e_{1,it}u_{i,1}^{0}f_{it}(0)\left(\hat{\mu}_{1,it}^{(1)}-\mu_{1,it}%
\right) u_{i,1}^{0\prime}v_{t,1}^{0}  \notag \\
& +O_{u,1}^{(1)}\frac{1}{N_{3}}\sum_{i\in
I_{3}}e_{1,it}u_{i,1}^{0}f_{it}(0)\mu_{1,it}\left(\hat{u}%
_{i,1}^{(3,1)}-O_{u,1}^{(1)}u_{i,1}^{0}\right)
^{\prime}O_{1}^{(1)}v_{t,1}^{0}  \notag \\
& +O_{u,1}^{(1)}\frac{1}{N_{3}}\sum_{i\in
I_{3}}e_{1,it}u_{i,1}^{0}f_{it}(0)\mu_{1,it}\left(O_{1}^{(1)}u_{i,1}^{0}%
\right) ^{\prime}\left(\dot{v}_{t,1}^{(1)}-O_{1}^{(1)}v_{t,1}^{0}\right) 
\notag \\
& +O_{u,1}^{(1)}\frac{1}{N_{3}}\sum_{i\in
I_{3}}e_{1,it}u_{i,1}^{0}f_{it}(0)\mu_{1,it}v_{t,1}^{(1)\prime}O_{1}^{(1)%
\prime}\left(O_{u,1}^{(1)}-O_{1}^{(1)}\right)
u_{i,1}^{0}+O_{p}\left(\eta_{N}^{2}\right) .
\end{align}%
By similar arguments as used in (\ref{Lem22.5}), we can show that the first
term on the RHS of (\ref{Lem28.2}) is $o_{p}\left(\left(N\vee T\right)^{-%
\frac{1}{2}}\right)$ uniformly over $t$. For the second term, by inserting
the linear expansion for $\hat{u}_{i,1}^{(3,1)}-O_{u,1}^{(1)}u_{i,1}^{0}$ in
(\ref{step6_2}), we notice that 
\begin{align*}
& \frac{1}{N_{3}}\sum_{i\in
I_{3}}e_{1,it}u_{i,1}^{0}f_{it}(0)\mu_{1,it}\left(\hat{u}%
_{i,1}^{(3,1)}-O_{u,1}^{(1)}u_{i,1}^{0}\right)
^{\prime}O_{1}^{(1)}v_{t,1}^{0} \\
& =\frac{1}{N_{3}T}\sum_{i\in I_{3}}\sum_{t\in
[T]}e_{1,it}f_{it}\left(0\right) \mu_{1,it}u_{i,1}^{0}v_{t^{*},1}^{0\prime}%
\hat{V}_{u_{1}}^{-1}e_{1,it^{*}}v_{t^{*},1}^{0}\left(\tau -\mathbf{1}\left\{
\epsilon_{it^{*}}\leq 0\right\} \right) +o_{p}\left(\left(N\vee T\right) ^{-%
\frac{1}{2}}\right) .
\end{align*}%
By arguments as used in (\ref{Lem:J1.4}), we can show the leading term in
the last equality is $O_{p}\left(\sqrt{\frac{\log (N\vee T)}{NT}}%
\xi_{N}^{2}\right) $. Then the second term on the RHS of (\ref{Lem28.2}) is $%
o_{p}\left(\left(N\vee T\right) ^{-\frac{1}{2}}\right) $. For the third and
the fourth terms on the RHS of (\ref{Lem28.2}), by conditional on $%
\mathscr{D}^{I_{1}\cup I_{2}}$, we notice that $\frac{1}{N_{3}}\sum_{i\in
I_{3}}e_{1,it}u_{i,1}^{0}f_{it}(0)\mu_{1,it}\left(O_{1}^{(1)}u_{i,1}^{0}%
\right) ^{\prime}\left(\dot{v}_{t,1}^{(1)}-O_{1}^{(1)}v_{t,1}^{0}\right) $
and $\frac{1}{N_{3}}\sum_{i\in
I_{3}}e_{1,it}u_{i,1}^{0}f_{it}(0)\mu_{1,it}v_{t,1}^{0\prime}O_{1}^{(1)%
\prime}\left(O_{u,1}^{(1)}-O_{1}^{(1)}\right) u_{i,1}^{0}$ are both mean
zero and the randomness depends only on $\{e_{1,it}f_{it}(0)\}$. By
conditional Hoeffding's inequality, we can show that both these two terms
are $o_{p}\left(\left(N\vee T\right) ^{-\frac{1}{2}}\right) $. 
\end{proof}

\bigskip

Define 
\begin{align}
& \mathbb{J}_{it}\left(\dot{\Delta}_{t,v}\right) =\left[ \mathbf{1}%
\left\{\epsilon_{it}\leq 0\right\} -\mathbf{1}\left\{ \epsilon_{it}\leq \dot{%
\Delta}_{t,v}^{\prime}\Psi_{it}^{0}\right\} \right] 
\begin{bmatrix}
v_{t,0}^{0} \\ 
v_{t,1}^{0}X_{1,it}%
\end{bmatrix}%
,  \notag \\
& h_{i}^{I,1}=\frac{1}{T}\sum_{t=1}^{T}\mathbb{E}\left[ f_{it}(0)%
\begin{bmatrix}
v_{t,0}^{0}v_{t,0}^{\prime} &  & v_{t,0}^{0}v_{t,1}^{0\prime}X_{1,it} \\ 
v_{t,1}^{0}v_{t,0}^{0\prime}X_{1,it} &  & v_{t,1}^{0}v_{t,1}^{0%
\prime}X_{1,it}^{2}%
\end{bmatrix}%
\bigg|\mathscr{D}\right] ,  \notag \\
& h_{i}^{I,2}=\frac{1}{T}\sum_{t=1}^{T}\mathbb{E}\left(\mathbb{J}_{i}\left(%
\dot{\Delta}_{t,v}\right) \bigg|\mathscr{D}^{I_{1}\cup I_{2}}\right) , 
\notag \\
& h_{it}^{III}=v_{t,0}^{0\prime}\hat{V}_{u_{0},i}^{-1}\frac{1}{T}%
\sum_{t^{*}=1}^{T}v_{t^{*},0}^{0}\left(\tau -1\left\{ \epsilon_{it^{*}}\leq
0\right\} \right) +X_{1,it}v_{t,1}^{0\prime}\hat{V}_{u_{1}}^{-1}\frac{1}{T}%
\sum_{t^{*}=1}^{T}e_{1,it^{*}}v_{t^{*},1}^{0}\left(\tau -1\left\{
\epsilon_{it^{*}}\leq 0\right\} \right)  \notag \\
& -\left(v_{t,0}^{0\prime}\hat{V}_{u_{0},i}^{-1}\frac{1}{T}%
\sum_{t^{*}=1}^{T}v_{t^{*},0}^{0}f_{it^{*}}(0)\mu_{1,it^{*}}v_{t^{*},1}^{0%
\prime}\right) \mathfrak{l}^{\prime}\left(O_{1}^{(1)}\right)
^{-1}\left(h_{i}^{I,1}\right) ^{-1}\frac{1}{T}\sum_{t^{*}=1}^{T}\left[ \tau -%
\mathbf{1}\left\{ \epsilon_{it^{*}}\leq 0\right\} \right] 
\begin{bmatrix}
v_{t^{*},0}^{0} \\ 
v_{t^{*},1}^{0}X_{1,it^{*}}%
\end{bmatrix}
\notag \\
& -\left(v_{t,0}^{0\prime}\hat{V}_{u_{0},i}^{-1}\frac{1}{T}%
\sum_{t^{*}=1}^{T}v_{t^{*},0}^{0}\mathbb{E}\left[ f_{it^{*}}(0)\bigg|%
\mathscr{D}^{I_{1}\cup I_{2}}\right] \mu_{1,it^{*}}v_{t^{*},1}^{0\prime}%
\right) \mathfrak{l}^{\prime}\left(O_{1}^{(1)}\right)^{-1}\left(h_{i}^{I,1}%
\right) ^{-1}h_{i}^{I,2}  \notag \\
& =\mathcal{R}_{h,it}^{1}+X_{1,it}\mathcal{R}_{h,it}^{2},  \label{order_h}
\end{align}%
with $\max_{i\in I_{3},t\in[T]}\left\vert \mathcal{R}_{h,it}^{a}\right\vert
=O_{p}(\eta_{N})$ for $a\in \{1,2\}$. Note that $\mathbb{E}\left[ f_{it}(0)%
\bigg|\mathscr{D}^{I_{1}\cup I_{2}}\right] =\mathbb{E}\left[ f_{it}(0)\bigg|%
\mathscr{D}\right]$.

\begin{ass}
\label{ass:15} Let $\mathcal{F}_{it}(\cdot)$ and $f_{it|h_{it}^{III}}(\cdot)$
be the conditional CDF and PDF of $\epsilon_{it}$ given $\mathscr{D}_{e}$
and $h_{it}^{III}$.

\begin{itemize}
\item[(i)] The derivative of the density $f_{it|h_{it}^{III}}$ is uniformly
bounded in absolute value.

\item[(ii)] $\max_{i \in [N],t \in [T]}\left\vert
f_{it|h_{it}^{III}}(0)-f_{it|h_{it}^{III}=0}(0)\right\vert \leq C\left\vert
h_{it}^{III}\right\vert$ for some Lipschitz constant $C>0$.
\end{itemize}
\end{ass}

\begin{lemma}
{\small \label{Lem29} } Under Assumptions \ref{ass:1}-\ref{ass:10} and
Assumption \ref{ass:15}, we have

\begin{itemize}
\item[(i)] $\max_{t\in[T]}\Bigg\Vert\frac{1}{N_{3}}\sum_{i\in
I_{3}}e_{1,it}O_{u,1}^{(1)}u_{i,1}^{0}\Bigg\{\left[ \mathbf{1}%
\left\{\epsilon_{it}\leq 0\right\} -\mathbf{1}\left\{ \epsilon_{it}\leq
\varrho_{it}\left(\hat{u}_{i,0}^{(3,1)},\hat{u}_{i,1}^{(3,1)},%
\left(O_{u,0}^{(1)}\right)
^{\prime-1}v_{t,0}^{0},\left(O_{u,1}^{(1)}\right)^{\prime-1}v_{t,1}^{0}%
\right) \right\} \right] \newline
-\left(F_{it}(0)-F_{it}\left[ \varrho_{it}\left(\hat{u}_{i,0}^{(3,1)},\hat{u}%
_{i,1}^{(3,1)},\left(O_{u,0}^{(1)}\right)
^{\prime-1}v_{t,0}^{0},\left(O_{u,1}^{(1)}\right)
^{\prime-1}v_{t,1}^{0}\right) \right] \right) \Bigg\}\Bigg\Vert%
_{2}=o_{p}\left(\left(N\vee T\right) ^{-\frac{1}{2}}\right) $,

\item[(ii)] $\max_{t\in[T]}\Bigg\Vert\frac{1}{N_{3}}\sum_{i\in
I_{3}}O_{u,0}^{(1)}u_{i,0}^{0}\Bigg\{\left[ \mathbf{1}\left\{
\epsilon_{it}\leq 0\right\} -\mathbf{1}\left\{ \epsilon_{it}\leq
\varrho_{it}\left(\hat{u}_{i,0}^{(3,1)},\hat{u}_{i,1}^{(3,1)},%
\left(O_{u,0}^{(1)}\right)
^{\prime-1}v_{t,0}^{0},\left(O_{u,1}^{(1)}\right)^{\prime-1}v_{t,1}^{0}%
\right) \right\} \right] \newline
-\left(F_{it}(0)-F_{it}\left[ \varrho_{it}\left(\hat{u}_{i,0}^{(3,1)},\hat{u}%
_{i,1}^{(3,1)},\left(O_{u,0}^{(1)}\right)
^{\prime-1}v_{t,0}^{0},\left(O_{u,1}^{(1)}\right)
^{\prime-1}v_{t,1}^{0}\right) \right] \right) \Bigg\}\Bigg\Vert%
_{2}=O_{p}(\eta_{N})$.
\end{itemize}
\end{lemma}

\begin{proof}
(i) Recall from (\ref{h_I}) that 
\begin{align*}
h_{i}^{I}& =\mathfrak{l}^{\prime}\left(O_{1}^{(1)}\right) ^{-1}\left(\frac{1%
}{T}\sum_{t=1}^{T}f_{it}(0)%
\begin{bmatrix}
v_{t,0}^{0}v_{t,0}^{\prime} &  & v_{t,0}^{0}v_{t,1}^{0\prime}X_{1,it} \\ 
v_{t,1}^{0}v_{t,0}^{0\prime}X_{1,it} &  & v_{t,1}^{0}v_{t,1}^{0%
\prime}X_{1,it}^{2}%
\end{bmatrix}%
\right) ^{-1}\Bigg[\frac{1}{T}\sum_{t=1}^{T}\left[ \tau -\mathbf{1}%
\left\{\epsilon_{it}\leq 0\right\} \right] 
\begin{bmatrix}
v_{t,0}^{0} \\ 
v_{t,1}^{0}X_{1,it}%
\end{bmatrix}
\\
& +\frac{1}{T}\sum_{t=1}^{T}\left[ \mathbf{1}\left\{ \epsilon_{it}\leq
0\right\} -\mathbf{1}\left\{ \epsilon_{it}\leq \dot{\Delta}%
_{t,v}^{\prime}\Psi_{it}^{0}\right\} \right] 
\begin{bmatrix}
v_{t,0}^{0} \\ 
v_{t,1}^{0}X_{1,it}%
\end{bmatrix}%
\Bigg].
\end{align*}%
By Bernstein's inequality in Lemma \ref{Lem:Bern}, we can show that 
\begin{align*}
& \max_{i\in I_{3}}\left\Vert \frac{1}{T}\sum_{t=1}^{T}\left\{ f_{it}(0)%
\begin{bmatrix}
v_{t,0}^{0}v_{t,0}^{\prime} &  & v_{t,0}^{0}v_{t,1}^{0\prime}X_{1,it} \\ 
v_{t,1}^{0}v_{t,0}^{0\prime}X_{1,it} &  & v_{t,1}^{0}v_{t,1}^{0%
\prime}X_{1,it}^{2}%
\end{bmatrix}%
-\mathbb{E}\left(f_{it}(0)%
\begin{bmatrix}
v_{t,0}^{0}v_{t,0}^{\prime} &  & v_{t,0}^{0}v_{t,1}^{0\prime}X_{1,it} \\ 
v_{t,1}^{0}v_{t,0}^{0\prime}X_{1,it} &  & v_{t,1}^{0}v_{t,1}^{0%
\prime}X_{1,it}^{2}%
\end{bmatrix}%
\bigg|\mathscr{D}\right) \right\} \right\Vert_{F} \\
& =O_{p}\left(\sqrt{\frac{\log (N\vee T)}{T}}\xi_{N}^{2}\right) ,\text{ and}
\\
& \max_{i\in I_{3}}\left\Vert \frac{1}{T}\sum_{t=1}^{T}\left[ \tau -\mathbf{1%
}\left\{ \epsilon_{it}\leq 0\right\} \right] 
\begin{bmatrix}
v_{t,0}^{0} \\ 
v_{t,1}^{0}X_{1,it}%
\end{bmatrix}%
\right\Vert_{2}=O_{p}\left(\sqrt{\frac{\log (N\vee T)}{T}}\xi_{N}\right) .
\end{align*}%
Besides, we observe that 
\begin{align}
& \max_{i\in I_{3},t\in[T]}\left\Vert Var\left(\mathbb{J}_{it}\left(\dot{%
\Delta}_{t,v}\right) \bigg|\mathscr{D}^{I_{1}\cup I_{2}}\right)
\right\Vert_{F}  \notag  \label{I6_1} \\
& \leq \max_{i\in I_{3},t\in[T]}\left\Vert \mathbb{E}\left(\left[\mathbf{1}%
\left\{ \epsilon_{it}\leq 0\right\} -\mathbf{1}\left\{ \epsilon_{it}\leq 
\dot{\Delta}_{t,v}^{\prime}\Psi_{it}^{0}\right\} \right] 
\begin{bmatrix}
v_{t,0}^{0}v_{t,0}^{\prime} &  & v_{t,0}^{0}v_{t,1}^{0\prime}X_{1,it} \\ 
v_{t,1}^{0}v_{t,0}^{0\prime}X_{1,it} &  & v_{t,1}^{0}v_{t,1}^{0%
\prime}X_{1,it}^{2}%
\end{bmatrix}%
\bigg|\mathscr{D}^{I_{1}\cup I_{2}}\right) \right\Vert_{F}  \notag \\
& =\max_{i\in I_{3},t\in[T]}\left\Vert \mathbb{E}\left(\left[%
F_{it}(0)-F_{it}\left(\dot{\Delta}_{t,v}^{\prime}\Psi_{it}^{0}\right) \right]
\begin{bmatrix}
v_{t,0}^{0}v_{t,0}^{\prime} &  & v_{t,0}^{0}v_{t,1}^{0\prime}X_{1,it} \\ 
v_{t,1}^{0}v_{t,0}^{0\prime}X_{1,it} &  & v_{t,1}^{0}v_{t,1}^{0%
\prime}X_{1,it}^{2}%
\end{bmatrix}%
\bigg|\mathscr{D}^{I_{1}\cup I_{2}}\right) \right\Vert_{F}  \notag \\
& =\max_{t\in[T]}\left\Vert \dot{\Delta}_{t,v}\right\Vert_{2}\max_{i\in
I_{3},t\in[T]}\left\Vert \mathbb{E}\left(f_{it}(\tilde{s}_{it})\Psi_{it}^{0}%
\begin{bmatrix}
v_{t,0}^{0}v_{t,0}^{\prime} &  & v_{t,0}^{0}v_{t,1}^{0\prime}X_{1,it} \\ 
v_{t,1}^{0}v_{t,0}^{0\prime}X_{1,it} &  & v_{t,1}^{0}v_{t,1}^{0%
\prime}X_{1,it}^{2}%
\end{bmatrix}%
\bigg|\mathscr{D}^{I_{1}\cup I_{2}}\right) \right\Vert_{F}  \notag \\
& =O_{p}(\eta_{N}),
\end{align}%
where the first equality holds by law of iterated expectations, the second
equality is by mean-value theorem for $\left\vert \tilde{s}_{it}\right\vert $
lies between 0 and $\left\vert \dot{\Delta}_{t,v}^{\prime}\Psi_{it}^{0}\right%
\vert $ and the last line is by Theorem \ref{Thm2}(ii) and Assumption \ref%
{ass:1}(iv). Similarly, for any $\vartheta >0$, we also have 
\begin{equation*}
\max_{i\in I_{3},t\in[T]}\sum_{s=t+1}^{T}\left\Vert Cov\left(\mathbb{J}%
_{it}\left(\dot{\Delta}_{t,v}\right) ,\mathbb{J}_{is}\left(\dot{\Delta}%
_{t,v}\right) ^{\prime}\bigg|\mathscr{D}^{I_{1}\cup
I_{2}}\right)\right\Vert_{F}=O_{p}\left(\eta_{N}^{\frac{2}{2+\vartheta }%
}\right) .
\end{equation*}%
By similar arguments as used in (\ref{Wdot_exp1}) and (\ref{Wdot_exp2}), we
have 
\begin{equation}
\max_{i\in I_{3}}\left\Vert \frac{1}{T}\sum_{t=1}^{T}\left\{ \mathbb{J}%
_{it}\left(\dot{\Delta}_{t,v}\right) -\mathbb{E}\left(\mathbb{J}_{it}\left(%
\dot{\Delta}_{t,v}\right) \bigg|\mathscr{D}^{I_{1}\cup I_{2}}\right)
\right\} \right\Vert_{2}=o_{p}\left(\left(N\vee T\right)^{-1/2}\right) .
\label{pross_J}
\end{equation}%
Together with the fact that $\max_{i\in I_{3}}\left\Vert \mathbb{J}_{i}\left(%
\dot{\Delta}_{t,v}\right) \right\Vert_{2}=O_{p}(\eta_{N})$, uniformly over $t%
\in[T]$, we have 
\begin{equation*}
h_{i}^{I}=\mathfrak{l}^{\prime}\left(O_{1}^{(1)}\right)
^{-1}\left(h_{i}^{I,1}\right) ^{-1}\frac{1}{T}\sum_{t=1}^{T}\left[ \tau -%
\mathbf{1}\left\{ \epsilon_{it}\leq 0\right\} \right] 
\begin{bmatrix}
v_{t,0}^{0} \\ 
v_{t,1}^{0}X_{1,it}%
\end{bmatrix}%
+\mathfrak{l}^{\prime}\left(O_{1}^{(1)}\right) ^{-1}\left(h_{i}^{I,1}\right)
^{-1}h_{i}^{I,2}+o_{p}\left(\left(N\vee T\right)^{-1/2}\right) .
\end{equation*}%
By Assumption \ref{ass:1}(iv) and Lemma \ref{Lem:bounded u&v_tilde}, we have 
$\max_{i\in I_{3}}\left\Vert h_{i}^{I,1}\right\Vert_{F}=O(1)~a.s.$. Like (%
\ref{I6_1}), we can show that $\max_{i\in I_{3}}\left\Vert
h_{i}^{I,2}\right\Vert_{2}=O_{p}(\eta_{N})$.

Similarly as (\ref{Lem28.1}), uniformly in $i\in I_{3}$ and $t\in[T]$ we
have 
\begin{align}
& \varrho_{it}\left(\hat{u}_{i,0}^{(3,1)},\hat{u}_{i,1}^{(3,1)},%
\left(O_{u,0}^{(1)}\right)^{\prime-1}v_{t,0}^{0},\left(O_{u,1}^{(1)}%
\right)^{\prime-1}v_{t,1}^{0}\right)  \notag  \label{I6_2} \\
& =\hat{u}_{i,0}^{(3,1)\prime}\left(O_{u,0}^{(1)}\right)^{%
\prime-1}v_{t,0}^{0}-u_{i,0}^{0\prime}v_{t,0}^{0}+\hat{\mu}_{1,it}\hat{u}%
_{i,1}^{(3,1)\prime}\dot{v}_{t,1}^{(1)}-\mu_{1,it}u_{i,1}^{0%
\prime}v_{t,1}^{0}  \notag \\
& +\hat{e}_{1,it}\hat{u}_{i,1}^{(3,1)\prime}\left(O_{u,1}^{(1)}\right)
^{\prime-1}v_{t,1}^{0}-e_{1,it}u_{i,1}^{0\prime}v_{t,1}^{0}  \notag \\
& =v_{t,0}^{0\prime}O_{u,0}^{-1}\left(\hat{u}%
_{i,0}^{(3,1)}-O_{u,0}^{(1)}u_{i,0}^{0}\right)
+X_{1,it}v_{t,1}^{0\prime}O_{u,1}^{(1)-1}\left(\hat{u}%
_{i,1}^{(3,1)}-O_{u,1}^{(1)}u_{i,1}^{0}\right)+\mu_{1,it}%
\left(O_{1}^{(1)}u_{i,1}^{0}\right) ^{\prime}\left(\dot{v}%
_{t,1}^{(1)}-O_{1}^{(1)}v_{t,1}^{0}\right)  \notag \\
&
+\mu_{1,it}u_{i,1}^{0\prime}O_{u,1}^{(1)\prime}%
\left(O_{u,1}^{(1)}-O_{1}^{(1)}\right) v_{t,1}^{0}+o_{p}\left(\left(N\vee
T\right)^{-1/2}\right)  \notag \\
& =v_{t,0}^{0\prime}\left(O_{0}^{(1)}\right) ^{-1}\left(\hat{u}%
_{i,0}^{(3,1)}-O_{u,0}^{(1)}u_{i,0}^{0}\right)
+X_{1,it}v_{t,1}^{0\prime}\left(O_{1}^{(1)}\right) ^{-1}\left(\hat{u}%
_{i,1}^{(3,1)}-O_{u,1}^{(1)}u_{i,1}^{0}\right)
+\mu_{1,it}\left(O_{1}^{(1)}u_{i,1}^{0}\right) ^{\prime}\left(\dot{v}%
_{t,1}^{(1)}-O_{1}^{(1)}v_{t,1}^{0}\right)  \notag \\
&
+\mu_{1,it}u_{i,1}^{0\prime}O_{1}^{(1)\prime}\left(O_{u,1}^{(1)}-O_{1}^{(1)}%
\right) v_{t,1}^{0}+o_{p}\left(\left(N\vee T\right) ^{-1/2}\right) ,
\end{align}%
where the last line is by the fact that $\left\Vert
O_{u,0}^{(1)}-O_{0}^{(1)}\right\Vert_{F}=O_{p}(\eta_{N})$ and $\left\Vert
O_{u,1}^{(1)}-O_{1}^{(1)}\right\Vert_{F}=O_{p}(\eta_{N})$. Combining (\ref%
{step6_2}), (\ref{step6_3}) and Theorem \ref{Thm2}(iii), we have 
\begin{align}
& \varrho_{it}\left(\hat{u}_{i,0}^{(3,1)},\hat{u}_{i,1}^{(3,1)},%
\left(O_{u,0}^{(1)}\right)
^{\prime-1}v_{t,0}^{0},\left(O_{u,1}^{(1)}\right)^{\prime-1}v_{t,1}^{0}%
\right)  \notag  \label{I6_3} \\
& =v_{t,0}^{0\prime}\hat{V}_{u_{0},i}^{-1}\frac{1}{T}%
\sum_{t^{*}=1}^{T}v_{t^{*},0}^{0}\left(\tau -1\left\{ \epsilon_{it^{*}}\leq
0\right\} \right) +v_{t,0}^{0\prime}\hat{V}_{u_{0},i}^{-1}\frac{1}{T}%
\sum_{t^{*}=1}^{T}v_{t^{*},0}^{0}f_{it^{*}}(0)\left(\mu_{1,it^{*}}u_{i,1}^{0%
\prime}v_{t^{*},1}^{0}-\hat{\mu}_{1,it^{*}}^{(1)}\dot{u}_{i,1}^{(1)\prime}%
\dot{v}_{t^{*},1}^{(1)}\right)  \notag \\
& +v_{t,0}^{0\prime}\hat{V}_{u_{0},i}^{-1}\frac{1}{T}%
\sum_{t=1}^{T}f_{it^{*}}(0)v_{t^{*},0}^{0}v_{t^{*},1}^{0\prime}u_{i,1}^{0}%
\left(e_{1,it^{*}}-\hat{e}_{1,it^{*}}\right) +X_{1,it}v_{t,1}^{0\prime}\hat{V%
}_{u_{1}}^{-1}\frac{1}{T}\sum_{t^{*}=1}^{T}e_{1,it^{*}}v_{t^{*},1}^{0}\left(%
\tau -1\left\{ \epsilon_{it^{*}}\leq 0\right\} \right)  \notag \\
& +\mu_{1,it}\left(O_{1}^{(1)}u_{i,1}^{0}\right) ^{\prime}\left(\dot{v}%
_{t,1}^{(1)}-O_{1}^{(1)}v_{t,1}^{0}\right)
+\mu_{1,it}u_{i,1}^{0\prime}O_{1}^{(1)\prime}\left(O_{u,1}^{(1)}-O_{1}^{(1)}%
\right)v_{t,1}^{0}+o_{p}\left(\left(N\vee T\right) ^{-1/2}\right)  \notag \\
& =v_{t,0}^{0\prime}\hat{V}_{u_{0},i}^{-1}\frac{1}{T}%
\sum_{t^{*}=1}^{T}v_{t^{*},0}^{0}\left(\tau -1\left\{ \epsilon_{it^{*}}\leq
0\right\} \right) +X_{1,it}v_{t,1}^{0\prime}\hat{V}_{u_{1}}^{-1}\frac{1}{T}%
\sum_{t^{*}=1}^{T}e_{1,it^{*}}v_{t^{*},1}^{0}\left(\tau -1\left\{
\epsilon_{it^{*}}\leq 0\right\} \right)  \notag \\
& -v_{t,0}^{0\prime}\hat{V}_{u_{0},i}^{-1}\frac{1}{T}%
\sum_{t^{*}=1}^{T}v_{t^{*},0}^{0}f_{it^{*}}(0)\mu_{1,it^{*}}\left(\dot{u}%
_{i,1}^{(1)}-O_{1}^{(1)}u_{i,1}^{0}\right)
^{\prime}O_{1}^{(1)}v_{t^{*},1}^{0}  \notag \\
& -v_{t,0}^{0\prime}\hat{V}_{u_{0},i}^{-1}\frac{1}{T}%
\sum_{t^{*}=1}^{T}v_{t^{*},0}^{0}f_{it^{*}}(0)\left(\hat{\mu}%
_{1,it^{*}}-\mu_{1,it^{*}}\right)
u_{i,1}^{0\prime}v_{t^{*},1}^{0}-v_{t,0}^{0\prime}\hat{V}_{u_{0},i}^{-1}%
\frac{1}{T}\sum_{t^{*}=1}^{T}v_{t^{*},0}^{0}f_{it^{*}}(0)\mu_{1,it^{*}}%
\left(O_{1}^{(1)}u_{i,1}^{0}\right) ^{\prime}\left(\dot{v}%
_{t^{*},1}^{(1)}-O_{1}^{(1)}v_{t^{*},1}^{0}\right)  \notag \\
& +v_{t,0}^{0\prime}\hat{V}_{u_{0},i}^{-1}\frac{1}{T}%
\sum_{t^{*}=1}^{T}f_{it^{*}}(0)v_{t^{*},0}^{0}v_{t^{*},1}^{0%
\prime}u_{i,1}^{0}\left(e_{1,it^{*}}-\hat{e}_{1,it^{*}}\right)
+\mu_{1,it}\left(O_{1}^{(1)}u_{i,1}^{0}\right) ^{\prime}\left(\dot{v}%
_{t,1}^{(1)}-O_{1}^{(1)}v_{t,1}^{0}\right)  \notag \\
&
+\mu_{1,it}u_{i,1}^{0\prime}O_{1}^{(1)\prime}\left(O_{u,1}^{(1)}-O_{1}^{(1)}%
\right)v_{t,1}^{0}+o_{p}\left(\left(N\vee T\right) ^{-1/2}\right)  \notag \\
& =v_{t,0}^{0\prime}\hat{V}_{u_{0},i}^{-1}\frac{1}{T}%
\sum_{t^{*}=1}^{T}v_{t^{*},0}^{0}\left(\tau -1\left\{ \epsilon_{it^{*}}\leq
0\right\} \right) +X_{1,it}v_{t,1}^{0\prime}\hat{V}_{u_{1}}^{-1}\frac{1}{T}%
\sum_{t^{*}=1}^{T}e_{1,it^{*}}v_{t^{*},1}^{0}\left(\tau -1\left\{
\epsilon_{it^{*}}\leq 0\right\} \right)  \notag \\
& -v_{t,0}^{0\prime}\hat{V}_{u_{0},i}^{-1}\frac{1}{T}%
\sum_{t^{*}=1}^{T}v_{t^{*},0}^{0}\mathbb{E}\left[ f_{it^{*}}(0)\bigg|%
\mathscr{D}^{I_{1}\cup I_{2}}\right] \mu_{1,it^{*}}v_{t^{*},1}^{0%
\prime}h_{i}^{I}+h_{it}^{IV}+o_{p}\left(\left(N\vee T\right)^{-1/2}\right) 
\notag \\
& -v_{t,0}^{0\prime}\hat{V}_{u_{0},i}^{-1}\frac{1}{T}%
\sum_{t^{*}=1}^{T}v_{t^{*},0}^{0}\left(f_{it^{*}}(0)-\mathbb{E}\left[%
f_{it^{*}}(0)\bigg|\mathscr{D}^{I_{1}\cup I_{2}}\right] \right)
\mu_{1,it}\left(O_{1}^{(1)}u_{i,1}^{0}\right) ^{\prime}\left(\dot{v}%
_{t^{*},1}^{(1)}-O_{1}^{(1)}v_{t^{*},1}^{0}\right)  \notag \\
& -v_{t,0}^{0\prime}\hat{V}_{u_{0},i}^{-1}\frac{1}{T}%
\sum_{t^{*}=1}^{T}v_{t^{*},0}^{0}\left(f_{it^{*}}(0)-\mathbb{E}\left[%
f_{it^{*}}(0)\bigg|\mathscr{D}^{I_{1}\cup I_{2}}\right] \right)
\mu_{1,it^{*}}v_{t^{*},1}^{0\prime}h_{i}^{I}  \notag \\
& :=h_{it}^{III}+h_{it}^{IV}+\mathcal{R}_{\varrho ,it}
\end{align}%
such that $\max_{i\in I_{3},t\in[T]}\left\vert \mathcal{R}%
_{\varrho,it}\right\vert =O_{p}(\eta_{N}^{2})$, where the first equality is
by inserting the linear expansion for $\hat{u}%
_{i,0}^{(3,1)}-O_{u,0}^{(1)}u_{i,0}^{0}$ in (\ref{step6_2}) and $\hat{u}%
_{i,1}^{(3,1)}-O_{u,1}^{(1)}u_{i,1}^{0}$ in (\ref{step6_3}), the second
equality is by inserting the linear expansion for $\dot{u}%
_{i,0}^{(1)}-O_{u,0}^{(1)}u_{i,0}^{0}$, the third equality is by the fact
that $\hat{\mu}_{1,it}-\mu_{1,it}=e_{1,it}-\hat{e}_{1,it}$ which leads to
the cancelling of the fourth and sixth terms in the second equality, the
last equality is by the definition in (\ref{order_h}) and 
\begin{align*}
h_{it}^{IV}& =-v_{t,0}^{0\prime}\hat{V}_{u_{0},i}^{-1}\frac{1}{T}%
\sum_{t^{*}=1}^{T}v_{t^{*},0}^{0}\mathbb{E}\left[ f_{it^{*}}(0)\bigg|%
\mathscr{D}^{I_{1}\cup I_{2}}\right] \mu_{1,it^{*}}%
\left(O_{1}^{(1)}u_{i,1}^{0}\right) ^{\prime}\left(\dot{v}%
_{t^{*},1}^{(1)}-O_{1}^{(1)}v_{t^{*},1}^{0}\right) \\
& +\mu_{1,it}\left(O_{1}^{(1)}u_{i,1}^{0}\right) ^{\prime}\left(\dot{v}%
_{t,1}^{(1)}-O_{1}^{(1)}v_{t,1}^{0}\right)
+\mu_{1,it}u_{i,1}^{0\prime}O_{1}^{(1)\prime}\left(O_{u,1}^{(1)}-O_{1}^{(1)}%
\right) v_{t,1}^{0} \\
& =\mathcal{R}_{h,it}^{3}+\mu_{1,it}\mathcal{R}_{h,it}^{4}
\end{align*}
with $\max_{i\in I_{3},t\in[T]}\left\vert \mathcal{R}_{h,it}^{a}\right\vert
=O_{p}(\eta_{N})$ for $a\in \{3,4\}$, and the last equality holds by the
fact that 
\begin{equation*}
\max_{i\in I_{3}}\left\Vert \frac{1}{T}\sum_{t^{*}=1}^{T}v_{t^{*},0}^{0}%
\left(f_{it^{*}}(0)-\mathbb{E}\left[ f_{it^{*}}(0)\bigg|\mathscr{D}%
^{I_{1}\cup I_{2}}\right] \right)
\mu_{1,it^{*}}\left(O_{1}^{(1)}u_{i,1}^{0}\right) ^{\prime}\left(\dot{v}%
_{t^{*},1}^{(1)}-O_{1}^{(1)}v_{tv,1}^{0}\right)
\right\Vert_{2}=o_{p}\left(\left(N\vee T\right) ^{-1/2}\right)
\end{equation*}%
and 
\begin{align*}
& \max_{i\in I_{3}}\left\Vert \frac{1}{T}\sum_{t^{*}=1}^{T}v_{t,0}^{0}%
\left(f_{it^{*}}(0)-\mathbb{E}\left[ f_{it^{*}}(0)\bigg|\mathscr{D}%
^{I_{1}\cup I_{2}}\right] \right)
\mu_{1,it^{*}}v_{t^{*},1}^{0\prime}h_{i}^{I}\right\Vert_{2} \\
& \leq \max_{i\in I_{3}}\left\vert \frac{1}{T}%
\sum_{t^{*}=1}^{T}v_{t^{*},0}^{0}\left(f_{it^{*}}(0)-\mathbb{E}\left[%
f_{it^{*}}(0)\bigg|\mathscr{D}^{I_{1}\cup I_{2}}\right] \right)
\mu_{1,it^{*}}v_{t^{*},1}^{0\prime}\right\vert \max_{i\in
I_{3}}||h_{i}^{I}||_{2}=o_{p}\left(\left(N\vee T\right) ^{-1/2}\right)
\end{align*}%
by conditional Bernstein's inequality given $\mathscr{D}^{I_{1}\cup I_{2}}$
and similar arguments as used in (\ref{pross_J}).

We notice that 
\begin{align}  \label{I6_4}
&\frac{1}{N_{3}}\sum_{i\in I_{3}}e_{1,it}u_{i,1}^{0}\Bigg\{\left[\mathbf{1}%
\left\{\epsilon_{it}\leq 0\right\}-\mathbf{1}\left\{\epsilon_{it}\leq%
\varrho_{it}\left(\hat{u}_{i,0}^{(3,1)},\hat{u}_{i,1}^{(3,1)},%
\left(O_{u,0}^{(1)}\right)^{\prime-1}v_{t,0}^{0},\left(O_{u,1}^{(1)}%
\right)^{\prime-1}v_{t,1}^{0}\right)\right\}\right]  \notag \\
&-\left(F_{it}(0)-F_{it}\left[ \varrho_{it}\left(\hat{u}_{i,0}^{(3,1)},\hat{u%
}_{i,1}^{(3,1)},\left(O_{u,0}^{(1)}\right)^{\prime-1}v_{t,0}^{0},%
\left(O_{u,1}^{(1)}\right)^{\prime-1}v_{t,1}^{0}\right) \right]\right)\Bigg\}
\notag \\
&=\frac{1}{N_{3}}\sum_{i\in I_{3}}e_{1,it}u_{i,1}^{0}\left\{\left[\mathbf{1}%
\left\{\epsilon_{it}\leq 0\right\}-\mathbf{1}\left\{\epsilon_{it}\leq
h_{it}^{III}+h_{it}^{IV}\right\}\right]-\left(F_{it}(0)-F_{it}\left[%
h_{it}^{III}+h_{it}^{IV} \right]\right)\right\}  \notag \\
&+\frac{1}{N_{3}}\sum_{i\in I_{3}}e_{1,it}u_{i,1}^{0}\Bigg\{\left[\mathbf{1}%
\left\{\epsilon_{it}\leq h_{it}^{III}+h_{it}^{IV}\right\}-\mathbf{1}%
\left\{\epsilon_{it}\leq h_{it}^{III}+h_{it}^{IV}+\mathcal{R}%
_{\varrho,it}\right\}\right]  \notag \\
&-\left[F_{it}\left(h_{it}^{III}+h_{it}^{IV}\right)-F_{it}%
\left(h_{it}^{III}+h_{it}^{IV}+\mathcal{R}_{\varrho,it} \right)\right]\Bigg\}
\notag \\
&=\frac{1}{N_{3}}\sum_{i\in I_{3}}e_{1,it}u_{i,1}^{0}\left\{\left[\mathbf{1}%
\left\{\epsilon_{it}\leq 0\right\}-\mathbf{1}\left\{\epsilon_{it}\leq
h_{it}^{III}+h_{it}^{IV}\right\}\right]-\left[\mathcal{F}_{it}(0)-\mathcal{F}%
_{it}\left(h_{it}^{III}+h_{it}^{IV}\right)\right]\right\}  \notag \\
&+\frac{1}{N_{3}}\sum_{i\in I_{3}}e_{1,it}u_{i,1}^{0}\Bigg\{\left[\mathbf{1}%
\left\{\epsilon_{it}\leq h_{it}^{III}+h_{it}^{IV}\right\}-\mathbf{1}%
\left\{\epsilon_{it}\leq h_{it}^{III}+h_{it}^{IV}+\mathcal{R}%
_{\varrho,it}\right\}\right]  \notag \\
&-\left[\mathcal{F}_{it}\left(h_{it}^{III}+h_{it}^{IV}\right)-\mathcal{F}%
_{it}\left(h_{it}^{III}+h_{it}^{IV}+\mathcal{R}_{\varrho,it} \right)\right]%
\Bigg\}  \notag \\
&+\frac{1}{N_{3}}\sum_{i\in I_{3}}e_{1,it}u_{i,1}^{0}\left\{\left[\mathcal{F}%
_{it}(0)-\mathcal{F}_{it}\left(h_{it}^{III}+h_{it}^{IV}\right)\right]-\left[%
F_{it}(0)-F_{it}\left(h_{it}^{III}+h_{it}^{IV} \right)\right]\right\}  \notag
\\
&+\frac{1}{N_{3}}\sum_{i\in I_{3}}e_{1,it}u_{i,1}^{0}\left[\mathcal{F}%
_{it}\left(h_{it}^{III}+h_{it}^{IV}\right)-\mathcal{F}_{it}%
\left(h_{it}^{III}+h_{it}^{IV}+\mathcal{R}_{\varrho,it} \right)\right] 
\notag \\
&-\frac{1}{N_{3}}\sum_{i\in I_{3}}e_{1,it}u_{i,1}^{0}\left[%
F_{it}\left(h_{it}^{III}+h_{it}^{IV}\right)-F_{it}%
\left(h_{it}^{III}+h_{it}^{IV}+\mathcal{R}_{\varrho,it} \right)\right] 
\notag \\
&:=\mathbb{I}_{6,t}^{I}+\mathbb{I}_{6,t}^{II}-\mathbb{I}_{6,t}^{III}+\mathbb{%
I}_{6,t}^{IV}+o_{p}\left(\left(N\vee T\right)^{-\frac{1}{2}}\right)
\end{align}
such that $\mathbb{I}_{6,t}^{I}:=\frac{1}{N_{3}}\sum_{i\in I_{3}}\mathbb{I}%
_{6,it}^{I}$, $\mathbb{I}_{6,t}^{II}:=\frac{1}{N_{3}}\sum_{i\in I_{3}}%
\mathbb{I}_{6,it}^{II}$, $\mathbb{I}_{6,t}^{III}:=\frac{1}{N_{3}}\sum_{i\in
I_{3}}\mathbb{I}_{6,it}^{III}$, $\mathbb{I}_{6,t}^{IV}:=\frac{1}{N_{3}}%
\sum_{i\in I_{3}}\mathbb{I}_{6,it}^{IV}$, with 
\begin{align}  \label{Lem:I6_2}
\mathbb{I}_{6,it}^{I}&=e_{1,it}u_{i,1}^{0}\left\{\left[\mathbf{1}%
\left\{\epsilon_{it}\leq 0\right\}-\mathbf{1}\left\{\epsilon_{it}\leq
h_{it}^{III}+h_{it}^{IV}\right\}\right]-\left[\mathcal{F}_{it}(0)-\mathcal{F}%
_{it}\left(h_{it}^{III}+h_{it}^{IV}\right)\right]\right\},  \notag \\
\mathbb{I}_{6,it}^{II}&=e_{1,it}u_{i,1}^{0}\Bigg\{\left[\mathbf{1}%
\left\{\epsilon_{it}\leq h_{it}^{III}+h_{it}^{IV}\right\}-\mathbf{1}%
\left\{\epsilon_{it}\leq h_{it}^{III}+h_{it}^{IV}+\mathcal{R}%
_{\varrho,it}\right\}\right]  \notag \\
&-\left[\mathcal{F}_{it}\left(h_{it}^{III}+h_{it}^{IV}\right)-\mathcal{F}%
_{it}\left(h_{it}^{III}+h_{it}^{IV}+\mathcal{R}_{\varrho,it} \right)\right]%
\Bigg\}, \\
\mathbb{I}_{6,it}^{III}&=e_{1,it}u_{i,1}^{0}\left[F_{it}(0)-F_{it}%
\left(h_{it}^{III}+h_{it}^{IV} \right)\right],  \notag \\
\mathbb{I}_{6,it}^{IV}&=e_{1,it}u_{i,1}^{0}\left[\mathcal{F}_{it}(0)-%
\mathcal{F}_{it}\left(h_{it}^{III}+h_{it}^{IV}\right)\right].  \notag
\end{align}
The last line in (\ref{I6_4}) is due to the fact that the last two terms in
the second equality is $o_{p}\left(\left(N\vee T\right)^{-\frac{1}{2}%
}\right) $ by mean-value theorem, Assumption (\ref{ass:1})(iv), and
Assumption (\ref{ass:1})(viii) and the union bound for $\mathcal{R}%
_{\varrho,it}$.

For $\mathbb{I}_{6,t}^{I}$ and $\mathbb{I}_{6,t}^{II}$, conditional on $%
\mathscr{D}_{e}^{I_{1}\cup I_{2}}$ and $h_{it}^{III}$, the randomness is
from $\epsilon_{it}$, and we observe that $\mathbb{I}_{6,it}^{I}$ and $%
\mathbb{I}_{6,it}^{II}$ are independent over $i$ by conditioning on $%
\mathscr{D}_{e}^{I_{1}\cup I_{2}}$ and $h_{it}^{III}$. Therefore, we obtain
that $\max_{t\in[T]}\left\Vert \mathbb{I}_{6,t}^{I}\right\Vert_{2}=o_{p}%
\left(\left(N\vee T\right)^{-\frac{1}{2}}\right)$ and $\max_{t\in[T]%
}\left\Vert \mathbb{I}_{6,t}^{II}\right\Vert_{2}=o_{p}\left(\left(N\vee
T\right)^{-\frac{1}{2}}\right)$ by conditional Bernstein's inequality for
independent sequence given $\mathscr{D}_{e}^{I_{1}\cup I_{2}}$ and $%
h_{it}^{III}$.

For $\mathbb{I}_{6,t}^{III}$, we notice that 
\begin{align}
\mathbb{I}_{6,t}^{III}& =\frac{1}{N_{3}}\sum_{i\in I_{3}}e_{1,it}u_{i,1}^{0}%
\left[ F_{it}(0)-F_{it}\left(h_{it}^{III}+h_{it}^{IV}\right) \right]  \notag
\label{I6_5} \\
& =\frac{1}{N_{3}}\sum_{i\in I_{3}}e_{1,it}u_{i,1}^{0}f_{it}(\dot{s}%
_{it})\left(h_{it}^{III}+h_{it}^{IV}\right)  \notag \\
& =\frac{1}{N_{3}}\sum_{i\in
I_{3}}e_{1,it}u_{i,1}^{0}f_{it}(0)\left(h_{it}^{III}+h_{it}^{IV}\right) +%
\frac{1}{N_{3}}\sum_{i\in I_{3}}e_{1,it}u_{i,1}^{0}\left[ f_{it}(\dot{s}%
_{it})-f_{it}(0)\right] \left(h_{it}^{III}+h_{it}^{IV}\right)  \notag \\
& =\frac{1}{N_{3}}\sum_{i\in
I_{3}}e_{1,it}u_{i,1}^{0}f_{it}(0)\left(h_{it}^{III}+h_{it}^{IV}\right)
+o_{p}\left(\left(N\vee T\right) ^{-\frac{1}{2}}\right) \quad \text{%
uniformly,}
\end{align}%
where the second line is by mean-value theorem with $\left\vert \dot{s}%
_{it}\right\vert $ lies between 0 and $\left\vert
h_{it}^{III}+h_{it}^{IV}\right\vert $ and the last line is by Assumption \ref%
{ass:1}(viii). By inserting $h_{it}^{III}$ and $h_{it}^{IV}$, we have 
\begin{align}
& \frac{1}{N_{3}}\sum_{i\in
I_{3}}e_{1,it}u_{i,1}^{0}f_{it}(0)\left(h_{it}^{III}+h_{it}^{IV}\right) 
\notag  \label{I6_6} \\
& =\frac{1}{N_{3}}\sum_{i\in
I_{3}}e_{1,it}u_{i,1}^{0}f_{it}(0)v_{t,0}^{0\prime}\hat{V}_{u_{0},i}^{-1}%
\frac{1}{T}\sum_{t^{*}=1}^{T}v_{t^{*},0}^{0}\left(\tau
-1\left\{\epsilon_{it^{*}}\leq 0\right\} \right)  \notag \\
& +\frac{1}{N_{3}}\sum_{i\in
I_{3}}e_{1,it}u_{i,1}^{0}f_{it}(0)X_{1,it}v_{t,1}^{0\prime}\hat{V}%
_{u_{1}}^{-1}\frac{1}{T}\sum_{t^{*}=1}^{T}e_{1,it^{*}}v_{t^{*},1}^{0}\left(%
\tau -1\left\{ \epsilon_{it^{*}}\leq 0\right\} \right)  \notag \\
& -\frac{1}{N_{3}}\sum_{i\in
I_{3}}e_{1,it}u_{i,1}^{0}f_{it}(0)\left(v_{t,0}^{0\prime}\hat{V}%
_{u_{0},i}^{-1}\frac{1}{T}\sum_{t^{*}=1}^{T}v_{t^{*},0}^{0}f_{it^{*}}(0)%
\mu_{1,it^{*}}v_{t^{*},1}^{0\prime}\right) \mathfrak{l}^{\prime}%
\left(O_{1}^{(1)}\right) ^{-1}\left(h_{i}^{I,1}\right) ^{-1}\frac{1}{T}%
\sum_{t^{*}=1}^{T}\left[ \tau -\mathbf{1}\left\{ \epsilon_{it^{*}}\leq
0\right\} \right] 
\begin{bmatrix}
v_{t^{*},0}^{0} \\ 
v_{t^{*},1}^{0}X_{1,it^{*}}%
\end{bmatrix}
\notag \\
& -\frac{1}{N_{3}}\sum_{i\in
I_{3}}e_{1,it}u_{i,1}^{0}f_{it}(0)\left(v_{t,0}^{0\prime}\hat{V}%
_{u_{0},i}^{-1}\frac{1}{T}\sum_{t^{*}=1}^{T}v_{t^{*},0}^{0}\mathbb{E}\left[
f_{it^{*}}(0)\bigg|\mathscr{D}^{I_{1}\cup I_{2}}\right]
\mu_{1,it^{*}}v_{t^{*},1}^{0\prime}\right) \mathfrak{l}^{\prime}%
\left(O_{1}^{(1)}\right)^{-1}\left(h_{i}^{I,1}\right) ^{-1}h_{i}^{I,2} 
\notag \\
& -\frac{1}{N_{3}}\sum_{i\in
I_{3}}e_{1,it}u_{i,1}^{0}f_{it}(0)v_{t,0}^{0\prime}\hat{V}_{u_{0},i}^{-1}%
\frac{1}{T}\sum_{t^{*}=1}^{T}v_{t^{*},0}^{0}\mathbb{E}\left[ f_{it^{*}}(0)%
\bigg|\mathscr{D}^{I_{1}\cup I_{2}}\right] \mu_{1,it^{*}}%
\left(O_{1}^{(1)}u_{i,1}^{0}\right) ^{\prime}\left(\dot{v}%
_{t^{*},1}^{(1)}-O_{1}^{(1)}v_{t^{*},1}^{0}\right)  \notag \\
& +\frac{1}{N_{3}}\sum_{i\in
I_{3}}e_{1,it}u_{i,1}^{0}f_{it}(0)\mu_{1,it}\left(O_{1}^{(1)}u_{i,1}^{0}%
\right) ^{\prime}\left(\dot{v}_{t,1}^{(1)}-O_{1}^{(1)}v_{t,1}^{0}\right) 
\notag \\
& +\frac{1}{N_{3}}\sum_{i\in
I_{3}}e_{1,it}u_{i,1}^{0}f_{it}(0)\mu_{1,it}u_{i,1}^{0\prime}O_{1}^{(1)%
\prime}\left(O_{u,1}^{(1)}-O_{1}^{(1)}\right) v_{t,1}^{0}  \notag \\
& :=\sum_{m\in [7]}\mathbb{I}_{6,t}^{III,m}.
\end{align}%
$\mathbb{I}_{6,t}^{III,1}$, $\mathbb{I}_{6,t}^{III,2}$, $\mathbb{I}%
_{6,t}^{III,3}$ can be analyzed in the same manner, and we take $\mathbb{I}%
_{6,t}^{III,1}$ for instance. Noticed that 
\begin{equation*}
\mathbb{I}_{6,t}^{III,1}=\frac{1}{N_{3}T}\sum_{i\in
I_{3}}\sum_{t^{*}=1}^{T}e_{1,it}u_{i,1}^{0}f_{it}(0)v_{t,0}^{0\prime}\hat{V}%
_{u_{0},i}^{-1}v_{t^{*},0}^{0}\left(\tau -1\left\{ \epsilon_{it^{*}}\leq
0\right\} \right) ,
\end{equation*}%
it is clear that conditioning on $\mathscr{D}_{e}$, $\mathbb{I}%
_{6,t}^{III,1} $ is mean zero by Assumption \ref{ass:1}(ii) and the
randomness is from $\epsilon_{it}$ which is strong mixing. With the similar
arguments as the second term in (\ref{Lem:J1.4}), we obtain that $\max_{t\in
[T]}\left\Vert \mathbb{I}_{6,t}^{III,1}\right\Vert_{2}=O_{p}\left(\sqrt{%
\frac{\log (N\vee T)}{NT}}\xi_{N}\right) $, and by Assumption \ref{ass:1}%
(ix), it follows that 
\begin{equation*}
\max_{t\in[T]}\left\Vert \mathbb{I}_{6,t}^{III,1}\right\Vert
_{2}=o_{p}\left(\left(N\vee T\right) ^{-1/2}\right) .
\end{equation*}

We can also analyze $\mathbb{I}_{6,t}^{III,4}$ and $\mathbb{I}_{6,t}^{III,5}$
in the same manner. Take $\mathbb{I}_{6,t}^{III,4}$ for instance. Note that 
\begin{equation*}
-\mathbb{I}_{6,t}^{III,4}=\frac{1}{N_{3}T}\sum_{i\in
I_{3}}\sum_{t^{*}=1}^{T}e_{1,it}u_{i,1}^{0}f_{it}(0)\left(v_{t,0}^{0\prime}%
\hat{V}_{u_{0},i}^{-1}v_{t,0}^{0}\mathbb{E}\left[ f_{it^{*}}(0)\bigg|%
\mathscr{D}^{I_{1}\cup I_{2}}\right] \mu_{1,it^{*}}v_{t^{*},1}^{0\prime}%
\right) \mathfrak{l}^{\prime}\left(O_{1}^{(1)}\right)
^{-1}\left(h_{i}^{I,1}\right) ^{-1}h_{i}^{I,2}
\end{equation*}%
which is mean zero by conditioning on $\mathscr{D}^{I_{1}\cup I_{2}}$ owing
to Assumption \ref{ass:10}(i) and the fact that $h_{i}^{I,1}$ and $%
h_{i}^{I,2}$ are fixed given $\mathscr{D}^{I_{1}\cup I_{2}}$. Then by
similar arguments for $\mathbb{I}_{6,t}^{III,1}$ above, we have 
\begin{equation}
\max_{t\in[T]}\left\Vert \mathbb{I}_{6,t}^{III,4}\right\Vert=o_{p}\left(%
\left(N\vee T\right) ^{-1/2}\right) .  \label{I6_8}
\end{equation}

For $\mathbb{I}_{6,t}^{III,6}$ and $\mathbb{I}_{6,t}^{III,7}$, we have 
\begin{align}
\max_{t\in[T]}\left\Vert \mathbb{I}_{6,t}^{III,6}\right\Vert_{2}&\leq \max_{t%
\in[T]}\left\Vert \frac{1}{N_{3}}\sum_{i\in
I_{3}}e_{1,it}u_{i,1}^{0}f_{it}(0)\mu_{1,it}u_{i,1}^{0\prime}\right\Vert_{2}%
\max_{t\in[T]}\left\Vert \dot{v}_{t,1}^{(1)}-O_{1}^{(1)}v_{t,1}^{0}\right%
\Vert_{2}  \notag  \label{I6_10} \\
& \leq \max_{t\in[T]}\left\Vert \frac{1}{N_{3}}\sum_{i\in
I_{3}}e_{1,it}u_{i,1}^{0}f_{it}(0)\mu_{1,it}u_{i,1}^{0\prime}\right%
\Vert_{2}O_{p}(\eta_{N})  \notag \\
& \leq O_{p}\left(\sqrt{\frac{\log (N\vee T)}{N}}\xi_{N}\right)O_{p}(%
\eta_{N})=o_{p}\left(\left(N\vee T\right) ^{-1/2}\right) , \\
\max_{t\in[T]}\left\Vert \mathbb{I}_{6,t}^{III,7}\right\Vert_{2}&\leq \max_{t%
\in[T]}\left\Vert \frac{1}{N_{3}}\sum_{i\in
I_{3}}e_{1,it}u_{i,1}^{0}f_{it}(0)\mu_{1,it}u_{i,1}^{0\prime}\right\Vert_{2}%
\max_{t\in[T]}\left\Vert v_{t,1}^{0}\right\Vert_{2}\left\Vert
O_{u,1}^{(1)}-O_{1}^{(1)}\right\Vert_{F}  \notag \\
& \leq O_{p}\left(\sqrt{\frac{\log (N\vee T)}{N}}\xi_{N}\right)O_{p}(%
\eta_{N})=o_{p}\left(\left(N\vee T\right) ^{-1/2}\right) ,
\end{align}%
where the second inequality is by Theorem \ref{Thm2}(ii) and the third
inequality is by Hoeffding's inequality conditional on $\mathscr{D}$ and the
last line combines Lemma \ref{Lem:bounded u&v_tilde}(i) and the fact that $%
\left\Vert O_{u,1}^{(1)}-O_{1}^{(1)}\right\Vert_{F}=O_{p}(\eta_{N})$.
Combining (\ref{I6_5})-(\ref{I6_10}), we have $\max_{t\in [T]}\left\Vert 
\mathbb{I}_{6,it}^{III}\right\Vert_{2}=o_{p}\left(\left(N\vee T\right)
^{-1/2}\right) $.

Last, we analyze $\mathbb{I}_{6,t}^{IV}$. Like (\ref{I6_5}), we have 
\begin{align}
\mathbb{I}_{6,t}^{IV}& =\frac{1}{N_{3}}\sum_{i\in I_{3}}e_{1,it}u_{i,1}^{0}%
\left[ \mathcal{F}_{it}(0)-\mathcal{F}_{it}\left(h_{it}^{III}+h_{it}^{IV}%
\right) \right]  \notag  \label{I6_11} \\
& =\frac{1}{N_{3}}\sum_{i\in I_{3}}e_{1,it}u_{i,1}^{0}f_{it|h_{it}^{III}}(%
\dot{s}_{it})\left(h_{it}^{III}+h_{it}^{IV}\right)  \notag \\
& =\frac{1}{N_{3}}\sum_{i\in
I_{3}}e_{1,it}u_{i,1}^{0}f_{it|h_{it}^{III}}(0)%
\left(h_{it}^{III}+h_{it}^{IV}\right) +\frac{1}{N_{3}}\sum_{i\in
I_{3}}e_{1,it}u_{i,1}^{0}\left[ f_{it|h_{it}^{III}}(\dot{s}%
_{it})-f_{it|h_{it}^{III}}(0)\right] \left(h_{it}^{III}+h_{it}^{IV}\right) 
\notag \\
& =\frac{1}{N_{3}}\sum_{i\in
I_{3}}e_{1,it}u_{i,1}^{0}f_{it|h_{it}^{III}}(0)h_{it}^{III}+\frac{1}{N_{3}}%
\sum_{i\in
I_{3}}e_{1,it}u_{i,1}^{0}f_{it|h_{it}^{III}}(0)h_{it}^{IV}+o_{p}\left(%
\left(N\vee T\right) ^{-\frac{1}{2}}\right)  \notag \\
& =\frac{1}{N_{3}}\sum_{i\in I_{3}}e_{1,it}u_{i,1}^{0}f_{it|h_{it}^{III}}(0)%
\left[ h_{it}^{III}-\mathbb{E}\left(h_{it}^{III}\bigg|\mathscr{D}%
_{e}^{I_{1}\cup I_{2}}\right) \right] +\frac{1}{N_{3}}\sum_{i\in
I_{3}}e_{1,it}u_{i,1}^{0}f_{it|h_{it}^{III}}(0)\mathbb{E}\left(h_{it}^{III}%
\bigg|\mathscr{D}_{e}^{I_{1}\cup I_{2}}\right)  \notag \\
& +\frac{1}{N_{3}}\sum_{i\in I_{3}}e_{1,it}u_{i,1}^{0}h_{it}^{IV}\left[%
f_{it|h_{it}^{III}}(0)-\mathbb{E}\left(f_{it|h_{it}^{III}}(0)|\mathscr{D}%
_{e}^{I_{1}\cup I_{2}}\right) \right] +\frac{1}{N_{3}}\sum_{i\in
I_{3}}e_{1,it}u_{i,1}^{0}h_{it}^{IV}\mathbb{E}\left(f_{it|h_{it}^{III}}(0)|%
\mathscr{D}_{e}^{I_{1}\cup I_{2}}\right)  \notag \\
& +o_{p}\left(\left(N\vee T\right) ^{-\frac{1}{2}}\right)  \notag \\
& =\frac{1}{N_{3}}\sum_{i\in I_{3}}e_{1,it}u_{i,1}^{0}\left[
f_{it|h_{it}^{III}}(0)-f_{it|h_{it}^{III}=0}(0)\right] \left[ h_{it}^{III}-%
\mathbb{E}\left(h_{it}^{III}\bigg|\mathscr{D}_{e}^{I_{1}\cup I_{2}}\right) %
\right]  \notag \\
& +\frac{1}{N_{3}}\sum_{i\in
I_{3}}e_{1,it}u_{i,1}^{0}f_{it|h_{it}^{III}=0}(0)\left[ h_{it}^{III}-\mathbb{%
E}\left(h_{it}^{III}\bigg|\mathscr{D}_{e}^{I_{1}\cup I_{2}}\right) \right] 
\notag \\
& +\frac{1}{N_{3}}\sum_{i\in I_{3}}e_{1,it}u_{i,1}^{0}\mathbb{E}\left(
h_{it}^{III}\bigg|\mathscr{D}_{e}^{I_{1}\cup I_{2}}\right) \left[%
f_{it|h_{it}^{III}}(0)-\mathbb{E}\left(f_{it|h_{it}^{III}}(0)|\mathscr{D}%
_{e}^{I_{1}\cup I_{2}}\right) \right]  \notag \\
& +\frac{1}{N_{3}}\sum_{i\in I_{3}}e_{1,it}u_{i,1}^{0}\mathbb{E}%
\left(h_{it}^{III}\bigg|\mathscr{D}_{e}^{I_{1}\cup I_{2}}\right) \mathbb{E}%
\left(f_{it|h_{it}^{III}}(0)|\mathscr{D}_{e}^{I_{1}\cup I_{2}}\right)  \notag
\\
& +\frac{1}{N_{3}}\sum_{i\in I_{3}}e_{1,it}u_{i,1}^{0}h_{it}^{IV}\left[%
f_{it|h_{it}^{III}}(0)-\mathbb{E}\left(f_{it|h_{it}^{III}}(0)|\mathscr{D}%
_{e}^{I_{1}\cup I_{2}}\right) \right] +\frac{1}{N_{3}}\sum_{i\in
I_{3}}e_{1,it}u_{i,1}^{0}h_{it}^{IV}\mathbb{E}\left(f_{it|h_{it}^{III}}(0)|%
\mathscr{D}_{e}^{I_{1}\cup I_{2}}\right)  \notag \\
& +o_{p}\left(\left(N\vee T\right) ^{-\frac{1}{2}}\right)  \notag \\
& :=\sum_{m\in [6]}\mathbb{I}_{6,t}^{IV,m}+o_{p}\left(\left(N\vee T\right)
^{-\frac{1}{2}}\right) ,
\end{align}%
where all the $o_{p}\left(\left(N\vee T\right) ^{-\frac{1}{2}}\right) $
terms hold uniformly over $t\in[T]$ and $o_{p}\left(\left(N\vee T\right) ^{-%
\frac{1}{2}}\right) $ term in the fourth equality is by mean-value theorem,
Assumption \ref{ass:15}(i) and the fact that $\left\vert \dot{s}%
_{it}\right\vert $ lies between 0 and $\left\vert
h_{it}^{III}+h_{it}^{IV}\right\vert $.

For $\mathbb{I}_{6,t}^{IV,1}$, with Assumption \ref{ass:15}(ii), (\ref%
{order_h}) and the fact that 
\begin{align*}
&\max_{i\in I_{3},t\in[T]}\left\vert \mathbb{E}\left(h_{it}^{III}\bigg|%
\mathscr{D}_{e}^{I_{1}\cup I_{2}}\right)\right\vert \\
&=\max_{i\in I_{3},t\in[T]}\left\vert\mathbb{E}\left[\left(v_{t,0}^{0\prime}%
\hat{V}_{u_{0},i}^{-1}\frac{1}{T}\sum_{t=1}^{T}v_{t,0}^{0}\mathbb{E}\left[%
f_{it}(0)\bigg|\mathscr{D}^{I_{1}\cup I_{2}}\right]\mu_{1,it}v_{t,1}^{0%
\prime}\right)\mathfrak{l}^{\prime}\left(O_{1}^{(1)}\right)^{-1}%
\left(h_{i}^{I,1}\right)^{-1}h_{i}^{I,2} \bigg|\mathscr{D}_{e}^{I_{1}\cup
I_{2}}\right]\right\vert \\
&=\max_{i\in I_{3},t\in[T]}\left\vert\left(v_{t,0}^{0\prime}\hat{V}%
_{u_{0},i}^{-1}\frac{1}{T}\sum_{t=1}^{T}v_{t,0}^{0}\mathbb{E}\left[f_{it}(0)%
\bigg|\mathscr{D}^{I_{1}\cup I_{2}}\right]\mu_{1,it}v_{t,1}^{0\prime}\right)%
\mathfrak{l}^{\prime}\left(O_{1}^{(1)}\right)^{-1}\left(h_{i}^{I,1}%
\right)^{-1}h_{i}^{I,2}\right\vert \\
&=O_{p}(\eta_{N}),
\end{align*}
we have 
\begin{align}  \label{I6_12}
\max_{t\in[T]}\left\Vert \mathbb{I}_{6,t}^{IV,1}\right\Vert_{2}\lesssim%
\max_{t\in[T]}\frac{1}{N_{3}}\sum_{i\in I_{3}}\left\Vert
e_{1,it}u_{i,1}^{0}\right\Vert_{2} \left\vert
h_{it}^{III}\right\vert\left(\left\vert h_{it}^{III}\right\vert+\left\vert%
\mathbb{E}\left(h_{it}^{III}\bigg|\mathscr{D}_{e}^{I_{1}\cup
I_{2}}\right)\right\vert\right)=O_{p}(\eta_{N}^{2}).
\end{align}

For $\mathbb{I}_{6,t}^{IV,2}$, $\mathbb{I}_{6,t}^{IV,3}$ and $\mathbb{I}%
_{6,t}^{IV,5}$, conditioning on $\mathscr{D}_{e}^{I_{1}\cup I_{2}}$, the
randomness is only from $h_{it}^{III}$, which is independent across $i$, and 
$\mathbb{I}_{6,t}^{IV,2}$, $\mathbb{I}_{6,t}^{IV,3}$ and $\mathbb{I}%
_{6,t}^{IV,5}$ are zero mean by conditioning on $\mathscr{D}_{e}^{I_{1}\cup
I_{2}}$. Similar to the arguments for $\mathbb{I}_{6,t}^{I}$ and $\mathbb{I}%
_{6,t}^{II}$ in (\ref{Lem:I6_2}), we have 
\begin{align}  \label{I6_13}
\max_{t\in[T]}\left\Vert \mathbb{I}_{6,t}^{IV,m}\right\Vert_{2}=o_{p}\left(%
\left(N\vee T\right)^{-\frac{1}{2}}\right), m\in\{2,3,5\}.
\end{align}

For $\mathbb{I}_{6,t}^{IV,4}$, by inserting $\mathbb{E}\left(h_{it}^{III}%
\bigg|\mathscr{D}_{e}^{I_{1}\cup I_{2}}\right) $ and the fact that $\mathbb{E%
}\left(f_{it|h_{it}^{III}}(0)|\mathscr{D}_{e}^{I_{1}\cup
I_{2}}\right)=f_{it}(0)$, it yields 
\begin{align}
\mathbb{I}_{6,t}^{IV,4}& =-\frac{1}{N_{3}}\sum_{i\in
I_{3}}e_{1,it}f_{it}(0)u_{i,1}^{0}\left(v_{t,0}^{0\prime}\hat{V}%
_{u_{0},i}^{-1}\frac{1}{T}\sum_{t=1}^{T}v_{t,0}^{0}\mathbb{E}\left[ f_{it}(0)%
\bigg|\mathscr{D}^{I_{1}\cup I_{2}}\right] \mu_{1,it}v_{t,1}^{0\prime}\right)%
\mathfrak{l}^{\prime}\left(O_{1}^{(1)}\right)
^{-1}\left(h_{i}^{I,1}\right)^{-1}h_{i}^{I,2}  \notag  \label{I6_14} \\
& =\mathbb{I}_{6,t}^{III,4}=o_{p}\left(\left(N\vee T\right)
^{-1/2}\right)\quad \text{uniformly},
\end{align}
where the last equality is by (\ref{I6_8}).

For $\mathbb{I}_{6,t}^{IV,6}$, we notice that 
\begin{align}
\mathbb{I}_{6,t}^{IV,6}& =\frac{1}{N_{3}}\sum_{i\in
I_{3}}e_{1,it}u_{i,1}^{0}h_{it}^{IV}\mathbb{E}\left(f_{it|h_{it}^{III}}(0)|%
\mathscr{D}_{e}^{I_{1}\cup I_{2}}\right)  \notag  \label{I6_15} \\
& =\frac{1}{N_{3}}\sum_{i\in
I_{3}}e_{1,it}f_{it}(0)u_{i,1}^{0}h_{it}^{IV}=\sum_{m\in \{5,6,7\}}\mathbb{I}%
_{6,t}^{III,m}=o_{p}\left(\left(N\vee T\right) ^{-1/2}\right) \quad \text{%
uniformly}.
\end{align}%
Combining (\ref{I6_11})-(\ref{I6_15}) yields $\max_{t\in [T]}\left\Vert 
\mathbb{I}_{6,it}^{III}\right\Vert_{2}=o_{p}\left(\left(N\vee T\right)
^{-1/2}\right) ,$ which leads to the desired result in statement (i).

(ii) As in (\ref{I6_4}), we have 
\begin{align*}
& \frac{1}{N_{3}}\sum_{i\in I_{3}}O_{u,0}^{(1)}u_{i,0}^{0}\Bigg\{\left[ 
\mathbf{1}\left\{ \epsilon_{it}\leq 0\right\} -\mathbf{1}\left\{
\epsilon_{it}\leq \varrho_{it}\left(\hat{u}_{i,0}^{(3,1)},\hat{u}%
_{i,1}^{(3,1)},\left(O_{u,0}^{(1)}\right)
^{\prime-1}v_{t,0}^{0},\left(O_{u,1}^{(1)}\right)^{\prime-1}v_{t,1}^{0}%
\right) \right\} \right] \\
& -\left(F_{it}(0)-F_{it}\left[ \varrho_{it}\left(\hat{u}_{i,0}^{(3,1)},\hat{%
u}_{i,1}^{(3,1)},\left(O_{u,0}^{(1)}\right)
^{\prime-1}v_{t,0}^{0},\left(O_{u,1}^{(1)}\right)
^{\prime-1}v_{t,1}^{0}\right) \right] \right) \Bigg\} \\
& =\frac{1}{N_{3}}\sum_{i\in I_{3}}u_{i,1}^{0}\left\{ \left[ \mathbf{1}%
\left\{ \epsilon_{it}\leq 0\right\} -\mathbf{1}\left\{ \epsilon_{it}\leq
h_{it}^{III}+h_{it}^{IV}\right\} \right] -\left[ \mathcal{F}_{it}(0)-%
\mathcal{F}_{it}\left(h_{it}^{III}+h_{it}^{IV}\right) \right] \right\} \\
& +\frac{1}{N_{3}}\sum_{i\in I_{3}}u_{i,1}^{0}\Bigg\{\left[ \mathbf{1}%
\left\{ \epsilon_{it}\leq h_{it}^{III}+h_{it}^{IV}\right\} -\mathbf{1}%
\left\{ \epsilon_{it}\leq h_{it}^{III}+h_{it}^{IV}+\mathcal{R}%
_{\varrho,it}\right\} \right] \\
& -\left[ \mathcal{F}_{it}\left(h_{it}^{III}+h_{it}^{IV}\right) -\mathcal{F}%
_{it}\left(h_{it}^{III}+h_{it}^{IV}+\mathcal{R}_{\varrho ,it}\right) \right]%
\Bigg\} \\
& +\frac{1}{N_{3}}\sum_{i\in I_{3}}u_{i,1}^{0}\left\{ \left[ \mathcal{F}%
_{it}(0)-\mathcal{F}_{it}\left(h_{it}^{III}+h_{it}^{IV}\right) \right] -%
\left[ F_{it}(0)-F_{it}\left(h_{it}^{III}+h_{it}^{IV}\right) \right]\right\}
\\
& +\frac{1}{N_{3}}\sum_{i\in I_{3}}u_{i,1}^{0}\left[ \mathcal{F}%
_{it}\left(h_{it}^{III}+h_{it}^{IV}\right) -\mathcal{F}_{it}%
\left(h_{it}^{III}+h_{it}^{IV}+\mathcal{R}_{\varrho ,it}\right) \right] \\
& -\frac{1}{N_{3}}\sum_{i\in I_{3}}u_{i,1}^{0}\left[ F_{it}%
\left(h_{it}^{III}+h_{it}^{IV}\right) -F_{it}\left(h_{it}^{III}+h_{it}^{IV}+%
\mathcal{R}_{\varrho ,it}\right) \right] \\
& =\frac{1}{N_{3}}\sum_{i\in I_{3}}u_{i,1}^{0}\left\{ \left[ \mathbf{1}%
\left\{ \epsilon_{it}\leq 0\right\} -\mathbf{1}\left\{ \epsilon_{it}\leq
h_{it}^{III}+h_{it}^{IV}\right\} \right] -\left[ \mathcal{F}_{it}(0)-%
\mathcal{F}_{it}\left(h_{it}^{III}+h_{it}^{IV}\right) \right] \right\} \\
& +\frac{1}{N_{3}}\sum_{i\in I_{3}}u_{i,1}^{0}\Bigg\{\left[ \mathbf{1}%
\left\{ \epsilon_{it}\leq h_{it}^{III}+h_{it}^{IV}\right\} -\mathbf{1}%
\left\{ \epsilon_{it}\leq h_{it}^{III}+h_{it}^{IV}+\mathcal{R}%
_{\varrho,it}\right\} \right] \\
& -\left[ \mathcal{F}_{it}\left(h_{it}^{III}+h_{it}^{IV}\right) -\mathcal{F}%
_{it}\left(h_{it}^{III}+h_{it}^{IV}+\mathcal{R}_{\varrho ,it}\right) \right]%
\Bigg\} \\
& +\frac{1}{N_{3}}\sum_{i\in I_{3}}u_{i,1}^{0}\left\{ \left[ \mathcal{F}%
_{it}(0)-\mathcal{F}_{it}\left(h_{it}^{III}+h_{it}^{IV}\right) \right] -%
\left[ F_{it}(0)-F_{it}\left(h_{it}^{III}+h_{it}^{IV}\right) \right]\right\}
+o_{p}\left(\left(N\vee T\right) ^{-1/2}\right) ,
\end{align*}%
where the last equality holds by similar arguments as used in the last line
in (\ref{I6_4}). We can show the first and second terms are $%
o_{p}\left(\left(N\vee T\right) ^{-1/2}\right) $ by similar arguments for $%
\mathbb{I}_{6,t}^{I}$ and $\mathbb{I}_{6,t}^{II}$. The third term is $%
O_{p}(\eta_{N})$ by mean-value theorem and Assumption \ref{ass:1}(viii).
Compared with $\mathbb{I}_{6,t}^{III}$ and $\mathbb{I}_{6,t}^{IV}$ in the
proof of statement (i), the third term here is not mean zero, and converges
to zero at the rate $\eta_{N}$.
\end{proof}

\begin{lemma}
{\small \label{Lem30} } Under Assumptions \ref{ass:1}-\ref{ass:10} and
Assumption \ref{ass:15}, we have

\begin{itemize}
\item[(i)] $\max_{t\in[T]}\Bigg\Vert\frac{1}{N_{3}}\sum_{i\in
I_{3}}e_{1,it}O_{u,1}^{(1)}u_{i,1}^{0}\Bigg\{\mathbf{1}\left\{
\epsilon_{it}\leq 0\right\} -\mathbf{1}\left\{ \epsilon_{it}\leq
\varrho_{it}\left(\hat{u}_{i,0}^{(3,1)},\hat{u}_{i,1}^{(3,1)},\hat{v}%
_{t,0}^{(3,1)},\hat{v}_{t,1}^{(3,1)}\right) \right\} \newline
-\left(F_{it}(0)-F_{it}\left[ \varrho_{it}\left(\hat{u}_{i,0}^{(3,1)},\hat{u}%
_{i,1}^{(3,1)},\hat{v}_{t,0}^{(3,1)},\hat{v}_{t,1}^{(3,1)}\right) \right]
\right) \Bigg\}\Bigg\Vert_{2}=o_{p}\left(\left(N\vee T\right) ^{-\frac{1}{2}%
}\right) $,

\item[(ii)] $\max_{t\in[T]}\Bigg\Vert\frac{1}{N_{3}}\sum_{i\in
I_{3}}O_{u,0}^{(1)}u_{i,0}^{0}\Bigg\{\mathbf{1}\left\{ \epsilon_{it}\leq
0\right\} -\mathbf{1}\left\{ \epsilon_{it}\leq \varrho_{it}\left(\hat{u}%
_{i,0}^{(3,1)},\hat{u}_{i,1}^{(3,1)},\hat{v}_{t,0}^{(3,1)},\hat{v}%
_{t,1}^{(3,1)}\right) \right\} \newline
-\left(F_{it}(0)-F_{it}\left[ \varrho_{it}\left(\hat{u}_{i,0}^{(3,1)},\hat{u}%
_{i,1}^{(3,1)},\hat{v}_{t,0}^{(3,1)},\hat{v}_{t,1}^{(3,1)}\right) \right]
\right) \Bigg\}\Bigg\Vert_{2}=O_{p}(\eta_{N})$.
\end{itemize}
\end{lemma}

\begin{proof}
To handle the correlation between $\{\epsilon_{it},e_{it}\}$ and $\left\{%
\hat{u}_{i,0}^{(3,1)},\hat{u}_{i,1}^{(3,1)}\right\} $, we follow similar
arguments as used in the proof of Lemma \ref{Lem26} by putting $\left\{ \hat{%
u}_{i,0}^{(3,1)},\hat{u}_{i,1}^{(3,1)}\right\} $ in a parameter set. Then by
similar arguments as used in the proof of Lemma \ref{Lem29}, we can obtain
the desired results.
\end{proof}

\begin{lemma}
{\small \label{Lem31} } Under Assumptions \ref{ass:1}-\ref{ass:10} and
Assumption \ref{ass:15}, we have 
\begin{align*}
\max_{t\in[T]}\left\Vert
O_{v_{1},t}^{(1)}-\left(O_{u,1}^{(1)\prime}\right)^{-1}\right\Vert_{F}=o_{p}%
\left(\left(N\vee T\right)^{-\frac{1}{2}}\right).
\end{align*}
\end{lemma}

\begin{proof}
Recall that 
\begin{equation*}
O_{v_{1},t}^{(1)}=\left\{ I_{K_{1}}+\left(O_{u,1}^{(1)\prime}\right) ^{-1} 
\left[ \hat{V}_{v_{1},t}^{I}\right] ^{-1}\left[ \frac{1}{N_{3}}\sum_{i\in
I_{3}}f_{it}(0)e_{1,it}^{2}u_{i,1}^{0}\left(O_{u,1}^{(1)}u_{i,1}^{0}-\hat{u}%
_{i,1}^{(3,1)}\right) ^{\prime}\right] \right\}
\left(O_{u,1}^{(1)\prime}\right) ^{-1}.
\end{equation*}%
Then 
\begin{equation*}
O_{v_{1},t}^{(1)}-\left(O_{u,1}^{(1)\prime}\right)
^{-1}=\left(O_{u,1}^{(1)\prime}\right) ^{-1}\left[ \hat{V}_{v_{1},t}^{3}%
\right] ^{-1}\left[ \frac{1}{N_{3}}\sum_{i\in
I_{3}}f_{it}(0)e_{1,it}^{2}u_{i,1}^{0}\left(O_{u,1}^{(1)}u_{i,1}^{0}-\hat{u}%
_{i,1}^{(3,1)}\right) ^{\prime}\right] .
\end{equation*}%
Note that 
\begin{align*}
& \frac{1}{N_{3}}\sum_{i\in
I_{3}}f_{it}(0)e_{1,it}^{2}u_{i,1}^{0}\left(O_{u,1}^{(1)}u_{i,1}^{0}-\hat{u}%
_{i,1}^{(3,1)}\right) ^{\prime} \\
& =\frac{1}{N_{3}}\sum_{i\in I_{3}}f_{it}(0)e_{1,it}^{2}\left\{ O_{1}^{(1)}%
\hat{V}_{u_{1}}^{-1}\frac{1}{T}\sum_{t=1}^{T}e_{1,it}v_{t,1}^{0}\left(\tau -%
\mathbf{1}\left\{ \epsilon_{it}\leq 0\right\} \right) +\mathcal{R}%
_{i,u}^{1}\right\} u_{i,1}^{0\prime} \\
& =\frac{1}{N_{3}}\sum_{i\in I_{3}}f_{it}(0)e_{1,it}^{2}O_{1}^{(1)}\hat{V}%
_{u_{1}}^{-1}\frac{1}{T}\sum_{t=1}^{T}e_{1,it}v_{t,1}^{0}\left(\tau -\mathbf{%
1}\left\{ \epsilon_{it}\leq 0\right\} \right)
u_{i,1}^{0\prime}+o_{p}\left(\left(N\vee T\right) ^{-\frac{1}{2}}\right) \\
& =\frac{1}{N_{3}T}\sum_{i\in
I_{3}}\sum_{t=1^{\prime}}^{T}f_{it}(0)e_{1,it}^{2}O_{1}^{(1)}\hat{V}%
_{u_{1}}^{-1}e_{1,it^{*}}v_{t^{*},1}^{0}\left(\tau -1\left\{
\epsilon_{it^{*}}\leq 0\right\} \right)
u_{i,1}^{0\prime}+o_{p}\left(\left(N\vee T\right) ^{-\frac{1}{2}}\right) \\
& =o_{p}\left(\left(N\vee T\right) ^{-\frac{1}{2}}\right) \quad \text{%
uniformly over }t\in[T]\text{,}
\end{align*}%
where the second equality is by uniform convergence rate of $\mathcal{R}%
_{i,u}^{1}$ and the last line follows by similar arguments as in (\ref%
{Lem:J1.4}) by Bernstein's inequality conditional on $\mathscr{D}_{e}$. Then
the result follows by noting that $O_{u,1}$ is bounded and $\hat{V}%
_{v_{1},t}^{I}$ is bounded uniformly over $t\in[T]$.
\end{proof}

\subsection{Lemmas for the Consistent Estimation of the Asymptotic Variances}

Recall that 
\begin{align*}
\hat{\mathbb{V}}_{u_{j}}& =\frac{1}{NT}\sum_{i\in [N]}\sum_{t\in[T]%
}k_{h_{N}}(\hat{\epsilon}_{it})\hat{e}_{j,it}^{2}\hat{\mathrm{v}}%
_{t,t,j},\quad \hat{\mathbb{V}}_{v_{j}}=\frac{1}{NT}\sum_{i\in [N]}\sum_{t\in%
[T]}k_{h_{N}}(\hat{\epsilon}_{it})\hat{e}_{j,it}^{2}\hat{\mathrm{u}}_{i,i,j},
\\
\hat{\Omega}_{u_{j}}& =\frac{1}{NT}\sum_{i\in [N]}\sum_{t\in[T]}\tau (1-\tau
)\hat{e}_{j,it}^{2}\hat{\mathrm{v}}_{t,t,j} \\
& +\frac{1}{NT}\sum_{i\in [N]}\sum_{t=1}^{T-T_{1}}\sum_{s=t+1}^{t+T_{1}}\hat{%
e}_{j,it}\hat{e}_{j,is}\hat{\mathrm{v}}_{t,s,j}\left[ \tau -K\left(\frac{%
\hat{\epsilon}_{it}}{h_{N}}\right) \right] \left[ \tau -K\left(\frac{\hat{%
\epsilon}_{is}}{h_{N}}\right) \right] \\
& +\frac{1}{NT}\sum_{i\in [N]}\sum_{t=1+T_{1}}^{T}\sum_{s=t-T_{1}}^{t-1}\hat{%
e}_{j,it}\hat{e}_{j,is}\hat{\mathrm{v}}_{t,s,j}\left[ \tau -K\left(\frac{%
\hat{\epsilon}_{it}}{h_{N}}\right) \right] \left[ \tau -K\left(\frac{\hat{%
\epsilon}_{is}}{h_{N}}\right) \right] , \\
\hat{\Omega}_{v_{j}}& =\tau (1-\tau )\frac{1}{NT}\sum_{i\in [N]}\sum_{t\in[T]%
}\hat{e}_{j,it}^{2}\hat{\mathrm{u}}_{i,i,j}, \\
\hat{\Sigma}_{u_{j}}& =\left(\hat{\mathbb{V}}_{u_{j}}\right)^{-1}\hat{\Omega}%
_{u_{j}}\left(\hat{\mathbb{V}}_{u_{j}}\right)^{-1},\quad \hat{\Sigma}%
_{v_{j}}=\left(\hat{\mathbb{V}}_{v_{j}}\right) ^{-1}\hat{\Omega}%
_{v_{j}}\left(\hat{\mathbb{V}}_{v_{j}}\right) ^{-1}, \\
V_{u_{j}}& =\frac{1}{T}\sum_{t=1}^{T}\mathbb{E}\left[%
f_{it}(0)e_{j,it}^{2}v_{t,j}^{0}v_{t,j}^{0\prime}\right] ,\quad
V_{v_{j}}^{(a)}=\frac{1}{N_{a}}\sum_{i\in I_{a}}\mathbb{E}\left[
f_{it}(0)e_{j,it}^{2}\right]u_{i,j}^{0}u_{i,j}^{0\prime}, \\
\Omega_{u_{j}}& =Var\left[ \frac{1}{\sqrt{T}}%
\sum_{t=1}^{T}e_{j,it}v_{t,j}^{0}(\tau -\mathbf{1}\left\{ \epsilon_{it}\leq
0\right\} )\right] ,\quad \Omega_{v_{j}}=\tau \left(1-\tau \right) \frac{1}{N%
}\sum_{i\in I_{3}}\mathbb{E}\left(e_{j,it}^{2}u_{i,j}^{0}u_{i,j}^{0\prime}%
\right) , \\
\Sigma_{u_{j}}&
=O_{j}^{(1)}V_{u_{j}}^{-1}\Omega_{u_{j}}V_{u_{j}}^{-1}O_{j}^{(1)\prime},%
\quad
\Sigma_{v_{j}}=O_{j}^{(1)}V_{v_{j}}^{-1}%
\Omega_{v_{j}}V_{v_{j}}^{-1}O_{j}^{(1)\prime}.
\end{align*}

\begin{lemma}
{\small \label{Lem:covhat}} Under Assumptions \ref{ass:1}-\ref{ass:11} and
Assumption \ref{ass:15}, $\hat{\Sigma}_{u_{j}}=\Sigma_{u_{j}}+o_{p}(1)$ and $%
\hat{\Sigma}_{v_{j}}=\Sigma_{v_{j}}+o_{p}(1)$.
\end{lemma}

\begin{proof}
First, we show that $\hat{\mathbb{V}}_{u_{j}}=O_{j}^{(1)}V_{u_{j}}O_{j}^{(1)%
\prime}+o_{p}(1)$. Note that 
\begin{align}
\max_{i\in I_{3},t\in[T]}\left\vert \hat{\epsilon}_{it}-\epsilon_{it}\right%
\vert & \leq \max_{i\in I_{3},t\in[T]}\left\vert \hat{\Theta}%
_{0,it}-\Theta_{0,it}^{0}\right\vert +\max_{i\in I_{3},t\in[T]}\sum_{j\in
[p]}\left\vert X_{j,it}\right\vert \left\vert \hat{\Theta}%
_{j,it}-\Theta_{j,it}^{0}\right\vert  \notag  \label{kerneldiff} \\
& =R_{\epsilon ,it}^{1}+\max_{i\in I_{3},t\in[T],j\in [p]}\left\vert
X_{j,it}\right\vert R_{\epsilon ,it}^{2},\text{ and}  \notag \\
\max_{i\in I_{3},t\in[T]}\left\vert k_{h_{N}}(\hat{\epsilon}%
_{it})-k_{h_{N}}(\epsilon_{it})\right\vert & =\frac{1}{h_{N}}\max_{i\in
I_{3},t\in[T]}\left\vert k(\frac{\hat{\epsilon}_{it}}{h_{N}})-k(\frac{%
\epsilon_{it}}{h_{N}})\right\vert \lesssim \frac{1}{h_{N}^{2}}\max_{i\in
I_{3},t\in[T]}\left\vert \hat{\epsilon}_{it}-\epsilon_{it}\right\vert  \notag
\\
& =R_{k,it}^{1}+\max_{i\in I_{3},t\in[T],j\in [p]}\left\vert
X_{j,it}\right\vert R_{k,it}^{2},
\end{align}
where $\max_{i\in I_{3},t\in[T]}\left\vert R_{\epsilon,it}^{1}\right\vert
=O_{p}(\eta_{N})$, $\max_{i\in I_{3},t\in[T]}\left\vert R_{\epsilon
,it}^{2}\right\vert =O_{p}\left(\frac{\log N\vee T}{N\wedge T}\right)$, $%
\max_{i\in I_{3},t\in[T]}\left\vert
R_{k,it}^{1}\right\vert=O_{p}(\eta_{N}h_{N}^{-2})$ and $\max_{i\in I_{3},t\in%
[T]}\left\vert R_{k,it}^{2}\right\vert =O_{p}(\frac{\log N\vee T}{N\wedge T}%
h_{N}^{-2})$ by (\ref{max:theta}) and Assumption \ref{ass:1}(iv). Let 
\begin{equation*}
\mathrm{v}_{t,t,j}^{0}=\frac{1}{6}\sum_{a\in[3]}\sum_{b\in[3]%
\setminus\{a\}}O_{j}^{(b)}v_{t,j}^{0}v_{t,j}^{0\prime}O_{j}^{(b)\prime}
\end{equation*}
and recall that $\hat{\mathrm{v}}_{t,t,j}=\frac{1}{6}\sum_{a\in[3]}\sum_{b\in%
[3]\setminus\{a\}}\hat{v}_{t,j}^{(a,b)}\hat{v}_{t,j}^{(a,b)\prime}$. With
Theorem \ref{Thm3}, it is clear that 
\begin{align*}
\max_{t\in[T]}\left\Vert\hat{\mathrm{v}}_{t,t,j}-\mathrm{v}%
_{t,t,j}^{0}\right\Vert_{F}&=\frac{1}{6}\sum_{a\in[3]}\sum_{b\in[3]%
\setminus\{a\}}\left(\hat{v}_{t,j}^{(a,b)}\hat{v}_{t,j}^{(a,b)%
\prime}-O_{j}^{(b)}v_{t,j}^{0}v_{t,j}^{0\prime}O_{j}^{(b)\prime}\right) \\
&=\frac{1}{6}\sum_{a\in[3]}\sum_{b\in[3]\setminus\{a\}}\bigg[\left(\hat{v}%
_{t,j}^{(a,b)}-O_{j}^{(b)}v_{t,j}^{0} \right)\left(\hat{v}%
_{t,j}^{(a,b)}-O_{j}^{(b)}v_{t,j}^{0}
\right)^{\prime}+O_{j}^{(b)}v_{t,j}^{0}\left(\hat{v}%
_{t,j}^{(a,b)}-O_{j}^{(b)}v_{t,j}^{0} \right)^{\prime} \\
&+\left(\hat{v}_{t,j}^{(a,b)}-O_{j}^{(b)}v_{t,j}^{0}
\right)\left(O_{j}^{(b)}v_{t,j}^{0} \right)^{\prime}\bigg] \\
&=O_{p}\left(\sqrt{\frac{\log N\vee T}{N}}\right).
\end{align*}
Let $\mathrm{v}_{t,s,j}^{0}=\frac{1}{6}\sum_{a\in[3]}\sum_{b\in[3]%
\setminus\{a\}}O_{j}^{(b)}v_{t,j}^{0}v_{s,j}^{0\prime}O_{j}^{(b)\prime}$ and
recall that $\hat{\mathrm{v}}_{t,s,j}=\frac{1}{6}\sum_{a\in[3]}\sum_{b\in[3]%
\setminus\{a\}}\hat{v}_{t,j}^{(a,b)}\hat{v}_{s,j}^{(a,b)\prime}$, similarly
as above, we have $\max_{t\in[T],s\in[T]}\left\Vert\hat{\mathrm{v}}_{t,s,j}-%
\mathrm{v}_{t,s,j}^{0}\right\Vert_{F}=O_{p}\left(\sqrt{\frac{\log N\vee T}{N}%
}\right)$. It follows that 
\begin{align}
\hat{\mathbb{V}}_{u_{j}}& =\frac{1}{NT}\sum_{i\in [N]}\sum_{t\in[T]%
}k_{h_{N}}(\hat{\epsilon}_{it})\hat{e}_{j,it}^{2}\hat{\mathrm{v}}_{t,t,j} 
\notag  \label{Vhat_1} \\
& =\frac{1}{NT}\sum_{i\in [N]}\sum_{t\in[T]}k_{h_{N}}(%
\epsilon_{it})e_{j,it}^{2}\mathrm{v}_{t,t,j}^{0}+\frac{1}{NT}\sum_{i\in
[N]}\sum_{t\in[T]}k_{h_{N}}(\epsilon_{it})e_{j,it}^{2}\left(\hat{\mathrm{v}}%
_{t,t,j}-\mathrm{v}_{t,t,j}^{0}\right)  \notag \\
& +\frac{1}{NT}\sum_{i\in [N]}\sum_{t\in[T]}k_{h_{N}}(\epsilon_{it})\left(%
\hat{e}_{j,it}^{2}-e_{j,it}^{2}\right)\mathrm{v}_{t,t,j}^{0} +\frac{1}{NT}%
\sum_{i\in [N]}\sum_{t\in[T]}k_{h_{N}}(\epsilon_{it})\left(\hat{e}%
_{j,it}^{2}-e_{j,it}^{2}\right)\left(\hat{\mathrm{v}}_{t,t,j}-\mathrm{v}%
_{t,t,j}^{0}\right)  \notag \\
& +\frac{1}{NT}\sum_{i\in [N]}\sum_{t\in[T]}\left[k_{h_{N}}(\hat{\epsilon}%
_{it})-k_{h_{N}}(\epsilon_{it})\right]e_{j,it}^{2}\mathrm{v}_{t,t,j}^{0}+%
\frac{1}{NT}\sum_{i\in [N]}\sum_{t\in[T]}\left[k_{h_{N}}(\hat{\epsilon}%
_{it})-k_{h_{N}}(\epsilon_{it})\right]e_{j,it}^{2}\left(\hat{\mathrm{v}}%
_{t,t,j}-\mathrm{v}_{t,t,j}^{0}\right)  \notag \\
& +\frac{1}{NT}\sum_{i\in [N]}\sum_{t\in[T]}\left[k_{h_{N}}(\hat{\epsilon}%
_{it})-k_{h_{N}}(\epsilon_{it})\right] \left(\hat{e}_{j,it}^{2}-e_{j,it}^{2}%
\right) \mathrm{v}_{t,t,j}^{0}  \notag \\
&+\frac{1}{NT}\sum_{i\in [N]}\sum_{t\in[T]}\left[k_{h_{N}}(\hat{\epsilon}%
_{it})-k_{h_{N}}(\epsilon_{it})\right] \left(\hat{e}_{j,it}^{2}-e_{j,it}^{2}%
\right)\left(\hat{\mathrm{v}}_{t,t,j}-\mathrm{v}_{t,t,j}^{0}\right)  \notag
\\
& =\frac{1}{NT}\sum_{i\in [N]}\sum_{t\in[T]}k_{h_{N}}(%
\epsilon_{it})e_{j,it}^{2}\mathrm{v}_{t,t,j}^{0}+O_{p}\left(%
\eta_{N}h_{N}^{-2}\right)  \notag \\
& =\frac{1}{NT}\sum_{i\in [N]}\sum_{t\in[T]}k_{h_{N}}(%
\epsilon_{it})e_{j,it}^{2}\mathrm{v}_{t,t,j}^{0}+o_{p}(1),
\end{align}%
where the last two lines combines (\ref{max:v}), (\ref{kerneldiff}),
Assumption \ref{ass:11}(ii) and facts that $\hat{e}_{j,it}^{2}-e_{j,it}^{2}=%
\left(\hat{e}_{j,it}-e_{j,it}\right) ^{2}+e_{j,it}\left(\hat{e}%
_{j,it}-e_{j,it}\right) =O_{p}(\eta_{N}^{2})+e_{j,it}O_{p}(\eta_{N})$
uniformly by Lemma \ref{Lem20} and $\max_{i\in I_{3},t\in[T]}\left\vert
k_{h_{N}}(\epsilon_{it})\right\vert =O(h_{N}^{-1})$. By Bernstein's
inequality, we obtain that 
\begin{align}
& \left\Vert \frac{1}{NT}\sum_{i\in [N]}\sum_{t\in[T]}\left[%
k_{h_{N}}(\epsilon_{it})e_{j,it}^{2}-\mathbb{E}\left(k_{h_{N}}(%
\epsilon_{it})e_{j,it}^{2}\big|\mathscr{D}\right) \right]
v_{t,j}^{0}v_{t,j}^{0\prime}\right\Vert_{F} =O_{p}\left(\sqrt{\frac{\log N}{%
NT}}\frac{\xi_{N}^{2}}{h_{N}}\right) ,  \notag \\
& \left\Vert \frac{1}{T}\sum_{t\in[T]}\mathbb{E}\left[f_{it}(0)e_{j,it}^{2}%
\big|\mathscr{D}\right] v_{t,j}^{0}v_{t,j}^{0\prime}-\mathbb{E}\left[
f_{it}(0)e_{j,it}^{2}v_{t,j}^{0}v_{t,j}^{0\prime}\right]\right%
\Vert_{F}=O_{p}\left(\sqrt{\frac{\log T}{T}}\xi_{N}^{2}\right) .
\label{Vhat_2}
\end{align}%
Besides, by Assumption \ref{ass:11}(i), we observe that 
\begin{equation}
\mathbb{E}\left[ k_{h_{N}}(\epsilon_{it})\big|\mathscr{D}_{e}\right]%
=f_{it}(0)+O(h_{N}^{m}),  \label{Vhat_3}
\end{equation}%
together with Assumption \ref{ass:11}(v), and it gives 
\begin{equation}
\frac{1}{NT}\sum_{i\in [N]}\sum_{t\in[T]}\mathbb{E}\left(k_{h_{N}}(%
\epsilon_{it})e_{j,it}^{2}\big|\mathscr{D}\right)v_{t,j}^{0}v_{t,j}^{0%
\prime}=\frac{1}{T}\sum_{t\in[T]}\mathbb{E}\left[f_{it}(0)e_{j,it}^{2}\big|%
\mathscr{D}\right] v_{t,j}^{0}v_{t,j}^{0\prime}+O(h_{N}^{m}).  \label{Vhat_4}
\end{equation}%
Combining (\ref{Vhat_1})-(\ref{Vhat_4}) and Assumption \ref{ass:11}(ii), we
obtain that $\hat{\mathbb{V}}_{u_{j}}=O_{j}^{(1)}V_{u_{j}}O_{j}^{(1)%
\prime}+o_{p}(1)$. By analogous analysis and Assumption \ref{ass:11}(v), we
can also show that $\hat{\mathbb{V}}%
_{v_{j}}^{(1)}=O_{j}^{(1)}V_{v_{j}}O_{j}^{(1)\prime}+o_{p}(1)$.

Next, we show the consistency of $\hat{\Omega}_{u_{j}}$. With the
restriction for $T_{1}$ in Assumption \ref{ass:11}(iii), we first note that 
\begin{align*}
\Omega_{u_{j}}& =Var\left[ \frac{1}{\sqrt{T}}%
\sum_{t=1}^{T}e_{j,it}v_{t,j}^{0}(\tau -\mathbf{1}\left\{ \epsilon_{it}\leq
0\right\} )\right] \\
& =\tau (1-\tau)\frac{1}{T}\sum_{t=1}^{T}\mathbb{E}%
\left(e_{j,it}^{2}v_{t,j}^{0}v_{t,j}^{0\prime}\right) +\frac{1}{T}%
\sum_{t=1}^{T-T_{1}}\sum_{s=t+1}^{t+T_{1}}\mathbb{E}\left[%
e_{j,it}e_{j,is}v_{t,j}^{0}v_{s,j}^{0\prime}(\tau -\mathbf{1}\left\{
\epsilon_{it}\leq 0\right\} )(\tau -\mathbf{1}\left\{ \epsilon_{is}\leq
0\right\} )\right] \\
& +\frac{1}{T}\sum_{t=1+T_{1}}^{T}\sum_{s=t-T_{1}}^{t-1}\mathbb{E}\left[%
e_{j,it}e_{j,is}v_{t,j}^{0}v_{s,j}^{0\prime}(\tau -\mathbf{1}\left\{
\epsilon_{it}\leq 0\right\} )(\tau -\mathbf{1}\left\{ \epsilon_{is}\leq
0\right\} )\right] +O(\alpha ^{T_{1}}) \\
& =\tau (1-\tau )\frac{1}{T}\sum_{t=1}^{T}\mathbb{E}%
\left(e_{j,it}^{2}v_{t,j}^{0}v_{t,j}^{0\prime}\right) +\frac{1}{T}%
\sum_{t=1}^{T-T_{1}}\sum_{s=t+1}^{t+T_{1}}\mathbb{E}\left[%
e_{j,it}e_{j,is}v_{t,j}^{0}v_{s,j}^{0\prime}\left(F_{i,ts}(0,0)-\tau^{2}%
\right) \right] \\
& +\frac{1}{T}\sum_{t=1+T_{1}}^{T}\sum_{s=t-T_{1}}^{t-1}\mathbb{E}\left[%
e_{j,it}e_{j,is}v_{t,j}^{0}v_{s,j}^{0\prime}\left(F_{i,ts}(0,0)-\tau^{2}%
\right) \right] +o(1),
\end{align*}
where the second equality is by Assumption \ref{ass:1}(iii), the third
equality is by Assumption \ref{ass:1}(vii) and Assumption \ref{ass:11}(iii).
Compared to $\hat{\Omega}_{u_{j}}$, what remains to show are 
\begin{align}
& \frac{1}{NT}\sum_{i\in [N]}\sum_{t\in[T]}\hat{e}_{j,it}^{2}\hat{\mathrm{v}}%
_{t,t,j}=\frac{1}{T}\sum_{t=1}^{T}\mathbb{E}\left(e_{j,it}^{2}\mathrm{v}%
_{t,t,j}^{0}\right) +o_{p}(1),  \label{omegahat_1} \\
& \frac{1}{NT}\sum_{i\in [N]}\sum_{t=1}^{T-T_{1}}\sum_{s=t+1}^{t+T_{1}}\hat{e%
}_{j,it}\hat{e}_{j,is}\hat{\mathrm{v}}_{t,s,j}\left[\tau -K\left(\frac{\hat{%
\epsilon}_{it}}{h_{N}}\right) \right] \left[ \tau-K\left(\frac{\hat{\epsilon}%
_{is}}{h_{N}}\right) \right]  \notag \\
& =\frac{1}{T}\sum_{t=1}^{T-T_{1}}\sum_{s=t+1}^{t+T_{1}}\mathbb{E}\left[%
e_{j,it}e_{j,is}\mathrm{v}_{t,s,j}^{0}\left(F_{i,ts}(0,0)-\tau^{2}\right) %
\right] +o_{p}(1),  \label{omegahat_2} \\
& \frac{1}{NT}\sum_{i\in [N]}\sum_{t=1+T_{1}}^{T}\sum_{s=t-T_{1}}^{t-1}\hat{e%
}_{j,it}\hat{e}_{j,is}\hat{\mathrm{v}}_{t,s,j}\left[\tau -K\left(\frac{\hat{%
\epsilon}_{it}}{h_{N}}\right) \right] \left[ \tau-K\left(\frac{\hat{\epsilon}%
_{is}}{h_{N}}\right) \right]  \notag \\
& =\frac{1}{T}\sum_{t=1+T_{1}}^{T}\sum_{s=t-T_{1}}^{t-1}\mathbb{E}\left[%
e_{j,it}e_{j,is}\mathrm{v}_{t,s,j}^{0}\left(F_{i,ts}(0,0)-\tau^{2}\right) %
\right] +o_{p}(1).  \label{omegahat_3}
\end{align}%
For (\ref{omegahat_1}), like (\ref{Vhat_1}), we notice that 
\begin{align}
& \frac{1}{NT}\sum_{i\in [N]}\sum_{t\in[T]}\hat{e}_{j,it}^{2}\hat{\mathrm{v}}%
_{t,t,j}=\frac{1}{NT}\sum_{i\in [N]}\sum_{t\in[T]}e_{j,it}^{2}\mathrm{v}%
_{t,t,j}^{0}+o_{p}(1)  \notag  \label{omegahat_4} \\
& =\frac{1}{T}\sum_{t=1}^{T}\mathbb{E}\left(e_{j,it}^{2}\mathrm{v}%
_{t,t,j}^{0}\right) +\frac{1}{NT}\sum_{i\in [N]}\sum_{t\in[T]}\left[
e_{j,it}^{2}\mathrm{v}_{t,t,j}^{0}-\mathbb{E}\left(e_{j,it}^{2}\mathrm{v}%
_{t,t,j}^{0}\right) \right]+o_{p}(1)  \notag \\
& =\frac{1}{T}\sum_{t=1}^{T}\mathbb{E}\left(e_{j,it}^{2}\mathrm{v}%
_{t,t,j}^{0}\right) +o_{p}(1),
\end{align}%
where the first equality is by Lemma \ref{Lem20} and (\ref{max:v}), and the
second equality is by Bernstein's inequality and Assumption \ref{ass:11}(v).

For (\ref{omegahat_2}), we observe that 
\begin{align}
& \frac{1}{NT}\sum_{i\in [N]}\sum_{t=1}^{T-T_{1}}\sum_{s=t+1}^{t+T_{1}}\hat{e%
}_{j,it}\hat{e}_{j,is}\hat{\mathrm{v}}_{t,s,j}\left[\tau -K\left(\frac{\hat{%
\epsilon}_{it}}{h_{N}}\right) \right] \left[ \tau-K\left(\frac{\hat{\epsilon}%
_{is}}{h_{N}}\right) \right]  \notag \\
& =\tau ^{2}\frac{1}{NT}\sum_{i\in[N]}\sum_{t=1}^{T-T_{1}}%
\sum_{s=t+1}^{t+T_{1}}\hat{e}_{j,it}\hat{e}_{j,is}\hat{\mathrm{v}}%
_{t,s,j}-\tau \frac{1}{NT}\sum_{i\in[N]}\sum_{t=1}^{T-T_{1}}%
\sum_{s=t+1}^{t+T_{1}}\hat{e}_{j,it}\hat{e}_{j,is}\hat{\mathrm{v}}%
_{t,s,j}K\left(\frac{\hat{\epsilon}_{it}}{h_{N}}\right)  \notag \\
& -\tau \frac{1}{NT}\sum_{i\in[N]}\sum_{t=1}^{T-T_{1}}\sum_{s=t+1}^{t+T_{1}}%
\hat{e}_{j,it}\hat{e}_{j,is}\hat{\mathrm{v}}_{t,s,j}K\left(\frac{\hat{%
\epsilon}_{is}}{h_{N}}\right) +\frac{1}{NT}\sum_{i\in[N]%
}\sum_{t=1}^{T-T_{1}}\sum_{s=t+1}^{t+T_{1}}\hat{e}_{j,it}\hat{e}_{j,is}\hat{%
\mathrm{v}}_{t,s,j}K\left(\frac{\hat{\epsilon}_{it}}{h_{N}}\right) K\left(%
\frac{\hat{\epsilon}_{is}}{h_{N}}\right)  \label{omegahat_5}
\end{align}%
such that 
\begin{align}
& \frac{1}{NT}\sum_{i\in [N]}\sum_{t=1}^{T-T_{1}}\sum_{s=t+1}^{t+T_{1}}\hat{e%
}_{j,it}\hat{e}_{j,is}\hat{\mathrm{v}}_{t,s,j}=\frac{1}{NT}\sum_{i\in[N]%
}\sum_{t=1}^{T-T_{1}}\sum_{s=t+1}^{t+T_{1}}e_{j,it}e_{j,is}\mathrm{v}%
_{t,s,j}^{0}+O_{p}\left(T_{1}\eta_{N}\right)  \notag \\
& =\frac{1}{NT}\sum_{i\in[N]}\sum_{t=1}^{T-T_{1}}\sum_{s=t+1}^{t+T_{1}}%
\mathbb{E}\left(e_{j,it}e_{j,is}\mathrm{v}_{t,s,j}^{0}\right) +\frac{1}{NT}%
\sum_{i\in [N]}\sum_{t=1}^{T-T_{1}}\sum_{s=t+1}^{t+T_{1}}\left[%
e_{j,it}e_{j,is}\mathrm{v}_{t,s,j}^{0}-\mathbb{E}\left(e_{j,it}e_{j,is}%
\mathrm{v}_{t,s,j}^{0}\right) \right]  \notag \\
& +O_{p}\left(T_{1}\eta_{N}\right)  \notag \\
& =\frac{1}{NT}\sum_{i\in [N]}\sum_{t=1}^{T-T_{1}}\sum_{s=t+1}^{t+T_{1}}%
\mathbb{E}\left(e_{j,it}e_{j,is}\mathrm{v}_{t,s,j}^{0}\right) +o_{p}(1),
\label{omegahat_6}
\end{align}%
where the first equality is by the similar arguments as (\ref{Vhat_1}) and
the last line combines Assumption \ref{ass:11}(iii) and the fact that the
second term in the second equality can be shown to be $o_{p}(1)$ by
Bernstein's inequality. Furthermore, with the fact that $\max_{i\in I_{3},t%
\in[T]}\left\vert K\left(\frac{\hat{\epsilon}_{it}}{h_{N}}\right) -K\left(%
\frac{\epsilon_{it}}{h_{N}}\right) \right\vert \lesssim \frac{1}{h_{N}}%
\max_{i\in I_{3},t\in[T]}\left\vert \hat{\epsilon}_{it}-\epsilon_{it}\right%
\vert =O_{p}\left(\eta_{N}h_{N}^{-1}\right) $ and by the analogous arguments
as above, we can show that 
\begin{align}
& \frac{1}{NT}\sum_{i\in [N]}\sum_{t=1}^{T-T_{1}}\sum_{s=t+1}^{t+T_{1}}\hat{e%
}_{j,it}\hat{e}_{j,is}\hat{\mathrm{v}}_{t,s,j}K\left(\frac{\hat{\epsilon}%
_{it}}{h_{N}}\right)  \notag \\
& =\frac{1}{NT}\sum_{i\in
[N]}\sum_{t=1}^{T-T_{1}}\sum_{s=t+1}^{t+T_{1}}e_{j,it}e_{j,is}\mathrm{v}%
_{t,s,j}^{0}K\left(\frac{\epsilon_{it}}{h_{N}}\right)
+O_{p}\left(T_{1}\eta_{N}h_{N}^{-1}\right)  \notag \\
& =\frac{1}{NT}\sum_{i\in[N]}\sum_{t=1}^{T-T_{1}}%
\sum_{s=t+1}^{t+T_{1}}e_{j,it}e_{j,is}\mathrm{v}_{t,s,j}^{0}\mathbb{E}\left[
K\left(\frac{\epsilon_{it}}{h_{N}}\right) \bigg|\mathscr{D}_{e}\right]
+o_{p}(1)  \notag \\
& =\frac{\tau }{NT}\sum_{i\in[N]}\sum_{t=1}^{T-T_{1}}%
\sum_{s=t+1}^{t+T_{1}}e_{j,it}e_{j,is}\mathrm{v}%
_{t,s,j}^{0}+O(h_{N}^{m})+o_{p}(1)  \notag \\
& =\frac{\tau }{NT}\sum_{i\in[N]}\sum_{t=1}^{T-T_{1}}\sum_{s=t+1}^{t+T_{1}}%
\mathbb{E}\left[e_{j,it}e_{j,is}\mathrm{v}_{t,s,j}^{0}\right]
+O(h_{N}^{m})+o_{p}(1)  \notag \\
& =\frac{\tau }{NT}\sum_{i\in[N]}\sum_{t=1}^{T-T_{1}}\sum_{s=t+1}^{t+T_{1}}%
\mathbb{E}\left[e_{j,it}e_{j,is}\mathrm{v}_{t,s,j}^{0}\right] +o_{p}(1),
\label{omegahat_7}
\end{align}%
where the first equality is similar as (\ref{Vhat_1}), the second equality
is by addition and subtracting and Assumption \ref{ass:11}(iii), the third
equality is by the fact that $\mathbb{E}\left[ K\left(\frac{\epsilon_{it}}{%
h_{N}}\right) \big|\mathscr{D}_{e}\right] =\tau +O(h_{N}^{m})$ by the
calculation of nonparametric kernel estimator which can be found in \cite%
{galvao2016smoothed}. The last equality is similar as the second equality
and combines Assumption \ref{ass:11}(i) and Assumption \ref{ass:11}(ii).

Moreover, similarly as (\ref{omegahat_6}), with the fact that $\mathbb{E}%
\left[ K\left(\frac{\epsilon_{it}}{h_{N}}\right) K\left(\frac{\epsilon_{is}}{%
h_{N}}\right) \big|\mathscr{D}_{e}\right] =F_{i,ts}(0,0)+O(h_{N}^{m})$, we
can show that 
\begin{align}
& \frac{1}{NT}\sum_{i\in [N]}\sum_{t=1}^{T-T_{1}}\sum_{s=t+1}^{t+T_{1}}\hat{e%
}_{j,it}\hat{e}_{j,is}\hat{\mathrm{v}}_{t,s,j}K\left(\frac{\hat{\epsilon}%
_{it}}{h_{N}}\right) K\left(\frac{\hat{\epsilon}_{is}}{h_{N}}\right)  \notag
\\
& =\frac{1}{NT}\sum_{i\in[N]}\sum_{t=1}^{T-T_{1}}%
\sum_{s=t+1}^{t+T_{1}}e_{j,it}e_{j,is}\mathrm{v}_{t,s,j}^{0}K\left(\frac{%
\epsilon_{it}}{h_{N}}\right) K\left(\frac{\epsilon_{is}}{h_{N}}\right)
+O_{p}\left(\sqrt{\frac{\log (N\vee T)}{N\wedge T}}\frac{T_{1}\xi_{N}^{2}}{%
h_{N}^{2}}\right)  \notag \\
& =\frac{1}{NT}\sum_{i\in[N]}\sum_{t=1}^{T-T_{1}}\sum_{s=t+1}^{t+T_{1}}%
\mathbb{E}\left[e_{j,it}e_{j,is}\mathrm{v}_{t,s,j}^{0}F_{i,ts}(0,0)\right]
+o_{p}(1).  \label{omegahat_8}
\end{align}

Combining (\ref{omegahat_5})-(\ref{omegahat_8}), we complete the proof for (%
\ref{omegahat_2}). By the analogous arguments, we can show the proof for (%
\ref{omegahat_3}), which yields $\hat{\Omega}_{u_{j}}=%
\Omega_{u_{j}}+o_{p}(1) $.
\end{proof}

\section{Algorithm for low-rank Estimation}

In this section, we provide the algorithm for the case of low-rank
estimation with two regressors, the case of more than two regressors is
self-evident. To solve the regularized quantile regression, let the
optimization problem with two regressors be as 
\begin{equation*}
\operatornamewithlimits{min}\limits_{\Theta_{0},\Theta_{1},\Theta_{2}}\quad%
\frac{1}{NT}\sum_{i=1}^{N}\sum_{t=1}^{T}\rho_{\tau
}\left(y_{it}-\Theta_{0,it}-x_{1,it}\Theta_{1,it}-x_{2,it}\Theta_{2,it}%
\right) +\nu_{0}\left\Vert \Theta_{0}\right\Vert_{\ast }+\nu_{1}\left\Vert
\Theta_{1}\right\Vert_{\ast }+\nu_{2}\left\Vert \Theta_{2}\right\Vert_{\ast
}.
\end{equation*}%
As in \cite{belloni2019high}, the above minimization problem is equivalent
to the following one: 
\begin{align*}
& \operatornamewithlimits{min}\limits_{\Theta_{0},\Theta_{1},%
\Theta_{2},V,W,Z_{\Theta_{0}},Z_{\Theta_{1}},Z_{\Theta_{2}}}\frac{1}{NT}%
\sum_{i=1}^{N}\sum_{t=1}^{T}\rho_{\tau }\left(V_{it}\right)
+\nu_{0}\left\Vert \Theta_{0}\right\Vert_{\ast }+\nu_{1}\left\Vert
Z_{\Theta_{1}}\right\Vert_{\ast }+\nu_{2}\left\Vert
Z_{\Theta_{2}}\right\Vert_{\ast } \\
& s.t.\,\,V=W,\,\,\,W=Y-X_{1}\odot \Theta_{1}-X_{2}\odot
\Theta_{2}-Z_{\Theta_{0}}, \\
&
Z_{\Theta_{0}}-\Theta_{0}=0,\,\,\,Z_{\Theta_{1}}-\Theta_{1}=0,\,\,\,Z_{%
\Theta_{2}}-\Theta_{2}=0.
\end{align*}%
As our theoretical results show, $\nu_{0}$, $\nu_{1}$ and $\nu_{2}$ converge
to zero at rate $\frac{\sqrt{N}\vee \sqrt{T}}{NT}$. The augmented Lagrangian
is 
\begin{align*}
& \mathscr{L}\left(V,W,\Theta_{0},Z_{\Theta_{0}},\Theta_{1},Z_{\Theta_{1}},%
\Theta_{2},Z_{\Theta_{2}},U_{v},U_{w},U_{\Theta_{0}},U_{\Theta_{1}},U_{%
\Theta_{2}}\right) \\
& =\frac{1}{NT}\sum_{i=1}^{N}\sum_{t=1}^{T}\rho_{\tau
}(V_{it})+\nu_{0}\left\Vert \Theta_{0}\right\Vert_{\ast }+\nu_{1}\left\Vert
Z_{\Theta_{1}}\right\Vert_{\ast }+\nu_{2}\left\Vert
Z_{\Theta_{2}}\right\Vert_{\ast }+\frac{\rho }{2NT}\left\Vert
V-W+U_{v}\right\Vert_{F}^{2} \\
& +\frac{\rho }{2NT}\left\Vert W-Y+X_{1}\odot \Theta_{1}+X_{2}\odot
\Theta_{2}+Z_{\Theta_{0}}+U_{W}\right\Vert_{F}^{2}+\frac{\rho }{2NT}%
\left\Vert Z_{\Theta_{0}}-\Theta_{0}+U_{\Theta_{0}}\right\Vert_{F}^{2} \\
& +\frac{\rho }{2NT}\left\Vert
Z_{\Theta_{1}}-\Theta_{1}+U_{\Theta_{1}}\right\Vert_{F}^{2}+\frac{\rho }{2NT}%
\left\Vert Z_{\Theta_{2}}-\Theta_{2}+U_{\Theta_{2}}\right\Vert_{F}^{2},
\end{align*}%
where $\rho >0$ is the penalty parameter.

By ADMM algorithm, similarly as Belloni et al. (2019), updates are as
follows: 
\begin{align}  \label{G.1}
&V^{k+1}\gets \argmin\limits_{V}\left\{ \frac{1}{NT}\sum_{i}\sum_{t}\rho_{%
\tau}\left(V_{it}\right)+\frac{\rho}{2NT}\left\Vert
V-W^{k}+U_{v}^{k}\right\Vert_{F}^{2} \right\} \\
\label{G.2}
&(\Theta_{1}^{k+1}, \Theta_{2}^{k+1)}\gets\argmin\limits_{\Theta_{1},%
\Theta_{2}}\bigg\{\left\Vert
W^{k}-Y+X_{1}\odot\Theta_{1}+X_{2}\odot\Theta_{2}+Z_{%
\Theta_{0}}^{k}+U_{W}^{k}\right\Vert_{F}^{2}+\left\Vert
Z_{\Theta_{1}}^{k}-\Theta_{1}+U_{\Theta_{1}}^{k}\right\Vert_{F}^{2} \\
&+\left\Vert
Z_{\Theta_{2}}^{k}-\Theta_{2}+U_{\Theta_{2}}^{k}\right\Vert_{F}^{2} \bigg\} 
\notag \\
\label{G.3}
&\Theta_{0}^{k+1}\gets\argmin\limits_{\Theta_{0}}\left\{\frac{1}{2}%
\left\Vert
Z_{\Theta_{0}}^{k}-\Theta_{0}+U_{\Theta_{0}}^{k}\right\Vert_{F}^{2}+\frac{%
\nu_{0}NT}{\rho}\left\Vert\Theta_{0}\right\Vert_{\ast} \right\} \\
&Z_{\Theta_{1}}^{k+1}\gets\argmin\limits_{Z_{\Theta_{1}}}\left\{\frac{1}{2}%
\left\Vert\Theta_{1}^{k+1}-U_{\Theta_{1}}^{k}-Z_{\Theta_{1}}\right%
\Vert_{F}^{2}+\frac{\nu_{1}NT}{\rho}\left\Vert
Z_{\Theta_{1}}\right\Vert_{\ast} \right\}  \notag \\
&Z_{\Theta_{2}}^{k+1}\gets\argmin\limits_{Z_{\Theta_{2}}}\left\{\frac{1}{2}%
\left\Vert\Theta_{2}^{k+1}-U_{\Theta_{2}}^{k}-Z_{\Theta_{2}}\right%
\Vert_{F}^{2}+\frac{\nu_{2}NT}{\rho}\left\Vert
Z_{\Theta_{2}}\right\Vert_{\ast} \right\}  \notag \\
&(Z_{\Theta_{0}}^{k+1},W^{k+1})\gets \argmin\limits_{Z_{\Theta_{0}},W}\bigg\{%
\left\Vert V^{k+1}-W+U_{v}^{k}\right\Vert_{F}^{2}+\left\Vert
W-Y+X_{1}\odot\Theta_{1}^{k+1}+X_{2}\odot\Theta_{2}^{k+1}+Z_{%
\Theta_{0}}+U_{W}^{k}\right\Vert_{F}^{2}  \notag \\
&+\left\Vert
Z_{\Theta_{0}}-\Theta_{0}^{k+1}+U_{\Theta_{0}}^{k}\right\Vert_{F}^{2}\bigg\}
\notag \\
&U_{v}^{k+1}\gets V^{k+1}-W^{k+1}+U_{v}^{k}  \notag \\
&U_{W}^{k+1}\gets
W^{k+1}-Y+X_{1}\odot\Theta_{1}^{k+1}+X_{2}\odot\Theta_{2}^{k+1}+Z_{%
\Theta_{0}}^{k+1}+U_{W}^{k}  \notag \\
&U_{\Theta_{0}}^{k+1}\gets Z_{\Theta_{0}}^{k+1}-\tilde{\Theta_{0}}%
^{k+1}+U_{\Theta_{0}}^{k}  \notag \\
&U_{\Theta_{1}}^{k+1}\gets
Z_{\Theta_{1}}^{k+1}-\Theta_{1}^{k+1}+U_{\Theta_{1}}^{k}  \notag \\
&U_{\Theta_{2}}^{k+1}\gets
Z_{\Theta_{2}}^{k+1}-\Theta_{2}^{k+1}+U_{\Theta_{2}}^{k}  \notag
\end{align}

For (\ref{G.1}), by Ali et al. (2016), 
\begin{align*}
V^{k+1}\gets P_{+}\left(W^{k}-U_{v}^{k}-\frac{\tau}{\rho}\iota_{N}%
\iota_{T}^{\prime}\right)+P_{-}\left(W^{k}-U_{v}^{k}-\frac{%
\left(1-\tau\right)}{\rho}\iota_{N}\iota_{T}^{\prime}\right).
\end{align*}
where $\iota_{N}$ is the $N\times 1$ all-ones vector, and same for $%
\iota_{T} $. For (\ref{G.2}), first order condition gives 
\begin{align*}
\Theta_{1,it}^{k+1}=\frac{\left(1+x_{2,it}^{2}\right)\left(Z_{%
\Theta_{1},it}^{k}+U_{\Theta_{1},it}^{k}-A_{it}x_{1,it}%
\right)-x_{1,it}x_{2,it}\left(Z_{\Theta_{2},it}^{k}+U_{%
\Theta_{2},it}^{k}-A_{it}x_{2,it}\right)}{1+x_{1,it}^{2}+x_{2,it}^{2}}, \\
\Theta_{2,it}^{k+1}=\frac{\left(1+x_{1,it}^{2}\right)\left(Z_{%
\Theta_{2},it}^{k}+U_{\Theta_{2},it}^{k}-A_{it}x_{2,it}%
\right)-x_{1,it}x_{2,it}\left(Z_{\Theta_{1},it}^{k}+U_{%
\Theta_{1},it}^{k}-A_{it}x_{1,it}\right)}{1+x_{1,it}^{2}+x_{2,it}^{2}},
\end{align*}
where 
\begin{align*}
A:=W^{k}+Z_{\Theta_{0}}^{k}+U_{W}^{k}-Y.
\end{align*}
To solve (\ref{G.3}), by singular value thresholding estimations, the update
for $\Theta_{0}^{k+1}$ is 
\begin{align*}
\Theta_{0}^{k+1}\gets P_{0}D_{0,\frac{\nu_{1}NT}{\rho}}Q_{0}^{\prime},
\end{align*}
where $Z_{\Theta_{0}}^{k}+U_{\Theta_{0}}^{k}=P_{0}D_{0}Q_{0}^{\prime}$, and $%
D_{0,\frac{\nu_{0}}{\rho},ii}=\max(D_{0,ii}-\frac{\nu_{0}}{\rho},0)$.
Similarly for $Z_{\Theta_{1}}^{k+1}$ and $Z_{\Theta_{2}}^{k+1}$, 
\begin{align*}
Z_{\Theta_{1}}^{k+1}\gets P_{1}D_{1,\frac{\nu_{1}NT}{\rho}}Q_{1}^{\prime}, \\
Z_{\Theta_{2}}^{k+1}\gets P_{2}D_{2,\frac{\nu_{2}NT}{\rho}}Q_{2}^{\prime},
\end{align*}
where $\Theta_{1}^{k+1}-U_{\Theta_{1}}^{k}=P_{1}D_{1}Q_{1}^{\prime}$, $%
\Theta_{2}^{k+1}-U_{\Theta_{2}}^{k}=P_{2}D_{2}Q_{2}^{\prime}$, $D_{1,\frac{%
\nu_{1}}{\rho},ii}=\max\left(D_{1,ii}-\frac{\nu_{1}}{\rho},0\right)$, and $%
D_{2,\frac{\nu_{2}}{\rho},ii}=\max\left(D_{2,ii}-\frac{\nu_{2}}{\rho}%
,0\right)$.

Finally, let $\tilde{A}:=-Y+X_{1}\odot\Theta_{1}^{k+1}+X_{2}\odot%
\Theta_{2}^{k+1}+U_{W}^{k+1},\,\,\tilde{B}:=-V^{k+1}-U_{v}^{k},\,\,\tilde{C}%
:=-\tilde{\Theta_{0}}^{k+1}+U_{\Theta_{0}}^{k}$, then 
\begin{align*}
Z_{\Theta_{0}}^{k+1}\gets\frac{-\tilde{A}-2\tilde{C}+\tilde{B}}{3}, \\
W^{k+1}\gets -\tilde{A}-\tilde{C}-2Z_{\Theta_{0}}^{k+1}.
\end{align*}

\end{document}